\documentclass[twoside]{article}

\usepackage[a4paper]{geometry}
\usepackage[utf8]{inputenc}
\usepackage[OT1]{fontenc}
\usepackage{RR}
\usepackage{relsize}

\newif\ifdraft\draftfalse

\usepackage{amsmath}
\usepackage{amssymb}
\usepackage{cite}
\usepackage{url}
\usepackage{theorem}
\usepackage{hyperref}
\usepackage{graphicx}
\usepackage{tikz}
\usetikzlibrary{positioning}

\makeatletter
\newcommand{\pname}[1]{\item[{#1}.]\def\@currentlabel{#1}}
\newcommand{\prname}[1]{\textbf{#1}\def\@currentlabel{#1}}
\makeatother

\newcommand{\ab}{\allowbreak}

\newcommand{\bb}[1]{\mcomment{{\bf Bruno:} #1}}
\newcommand{\bbnote}[1]{\mcomment{{\it BB: {#1}}}}
\ifdraft
\newlength{\shiftmargin} 
\newcommand{\mcomment}[1]{\marginpar{\tiny{#1}}}
\else
\newcommand{\mcomment}[1]{}
\fi

\newtheorem{lemma}{Lemma}

\newtheorem{corollary}{Corollary}
\newtheorem{proposition}{Proposition}
\theorembodyfont{\rmfamily}
\newtheorem{example}{Example}
\newtheorem{invariant}{Invariant}
\newtheorem{definition}{Definition}
\newtheorem{remark}{Remark}
\newtheorem{property}{Property}

\newenvironment{defn}{\begin{tabbing}
  \hspace{1.5em} \= \hspace{.5\linewidth - 1.5em} \= \hspace{1.5em} \= \kill
  }{
  \end{tabbing}}

\newcommand{\entry}[2]{\>$#1$\>\>#2}
\newcommand{\clause}[2]{$#1$\>\>#2}
\newcommand{\categ}[2]{\clause{#1::=}{#2}}

\newcommand{\parpop}{\mid}
\newcommand{\repl}[2]{{}!^{{#1}\leq{#2}}}
\newcommand{\cinput}[2]{#1(#2)}
\newcommand{\coutput}[2]{\overline{#1}\langle #2 \rangle}

\newcommand{\kw}[1]{\mathsf{#1}}
\newcommand{\kwf}[1]{\mathrm{#1}}
\newcommand{\rn}[1]{\textbf{#1}}
\newcommand{\kvar}[1]{\mathit{#1}}
\newcommand{\kwp}[1]{\mathit{#1}}
\newcommand{\kwt}[1]{\mathit{#1}}

\newcommand{\IF}{\kw{if}\ }
\newcommand{\THEN}{\ \kw{then}\ }
\newcommand{\ELSE}{\kw{else}\ }
\newcommand{\FIND}{\kw{find}}
\newcommand{\SUCHTHAT}{\kw{suchthat}}
\newcommand{\NEW}{\kw{new}}
\newcommand{\NEWCHANNEL}{\kw{newChannel}}
\newcommand{\LET}{\kw{let}}
\newcommand{\IN}{\kw{in}}
\newcommand{\assign}[2]{\LET\ #1 = #2\ \IN\ }
\newcommand{\bassign}[2]{\LET\ #1 = #2}
\newcommand{\Res}[2]{\NEW \ {#1}:{#2}}
\newcommand{\Reschan}[1]{\NEWCHANNEL\ {#1}}
\newcommand{\baguard}[1]{\IF{#1}\THEN}
\newcommand{\bguard}[3]{\baguard{#1}{#2}\ \ELSE{#3}}
\newcommand{\kevent}[1]{\kw{event}\ #1}
\newcommand{\keventabort}[1]{\kw{event\string_abort}\ #1}

\newcommand{\usedevents}{\mathit{EvUsed}}
\newcommand{\sevent}{\mathsf{S}}
\newcommand{\sbarevent}{\mathsf{\overline{S}}}
\newcommand{\advloses}{\kwf{adv\_loses}}

\newcommand{\bigorfind}{\mathop\bigoplus\nolimits}
\newcommand{\defined}{\kw{defined}}

\newcommand{\bitstring}{\kwt{bitstring}}
\newcommand{\bitstringbot}{\kwt{bitstring}_{\bot}}
\newcommand{\injbot}{\kwf{i}_{\bot}}
\newcommand{\uniqueopt}{[\mathit{unique}?]}
\newcommand{\unique}[1]{[\kw{unique}_{#1}]}

\newcommand{\nonunique}[1]{\mathsf{NonUnique}_{#1}}
\newcommand{\pat}{p}

\newcommand{\INSERT}{\kw{insert}}
\newcommand{\GET}{\kw{get}}
\newcommand{\tbl}{\mathit{Tbl}}

\newcommand{\sem}[1]{[\![#1]\!]}

\newcommand{\bool}{\mathit{bool}}
\newcommand{\true}{\kwf{true}}
\newcommand{\false}{\kwf{false}}

\newcommand{\Adv}{{\cal A}} %Adversary
\newcommand{\Advt}{\mathsf{Adv}} %Advantage
\newcommand{\Advtev}[4]{\mathsf{Adv}_{#1}(#3, #2, #4)} %Advantage
\newcommand{\boundfun}{\mathsf{Bound}}
\newcommand{\bound}[5]{\boundfun_{#1}(#2, #3, #4, #5)}
\newcommand{\secp}{\eta} 

\newcommand{\fand}{\wedge}
\newcommand{\for}{\vee}
\newcommand{\fnot}{\neg}
\newcommand{\iffun}{\kwf{if\_fun}}

\newcommand{\trace}{\mathit{Tr}}
\newcommand{\traceset}{\mathit{{\cal T}\!r}}
\newcommand{\tracesetfull}{\traceset_{\mathrm{full}}}
\newcommand{\conf}{\mathit{Conf}}
\newcommand{\before}{\preceq}
\newcommand{\beforetr}[1]{\before_{#1}}
\newcommand{\Prss}[2]{\Pr[#1 \preceq #2 ]}

\newcommand{\multiunion}{\uplus}
\newcommand{\bigmultiunion}{\biguplus}

\newcommand{\tup}[1]{\widetilde{#1}}

\newcommand{\mset}{{\cal M}}

\newcommand{\im}{\mathrm{im}\ }

\newcommand{\xor}{\kwf{xor}}

\newcommand{\randomchoice}{\mathop\leftarrow\limits^R}

\newcommand{\ltr}[1]{\mathrel{\xrightarrow{#1}}}
\newcommand{\red}[2]{\ltr{#1}_{#2}}
\newcommand{\redq}{\rightsquigarrow}
\newcommand{\initconfig}{\mathrm{initConfig}}
\newcommand{\startch}{\mathit{start}}

\newcommand{\ix}{t}
\newcommand{\dom}{\mathrm{Dom}}
\newcommand{\image}{\mathrm{Im}}
\newcommand{\reduce}{\mathrm{reduce}}
\newcommand{\fc}{\mathrm{fc}}
\newcommand{\fvar}{\mathrm{var}}
\newcommand{\vardef}{\mathrm{vardef}}

\newcommand{\defset}{\mathit{Defined}}
\newcommand{\defsetfut}{\defset^{\mathrm{Fut}}}
\newcommand{\cevent}{\mathrm{event}}

\newcommand{\FB}{\mathit{FB}}
\newcommand{\targetFB}{\FB_{\mathrm{T}}}

\newcommand{\enc}{\kwf{enc}}
\newcommand{\dec}{\kwf{dec}}

\newcommand{\mac}{\kwf{mac}}
\newcommand{\mverify}{\kwf{verify}}

\newcommand{\pkgen}{\kwf{pkgen}}
\newcommand{\skgen}{\kwf{skgen}}

\newcommand{\sign}{\kwf{sign}}

\newcommand{\ktob}{\kwf{k2b}}

\newcommand{\Z}{\kwf{Z}}

\newcommand{\vn}[1]{x_{#1}}

\newcommand{\tyenv}{{\cal E}}
\newcommand{\evset}{{\cal E}}

\newcommand{\rewrite}{\rightarrow}
\newcommand{\indep}{\mathrm{only\_dep}}

\newcommand{\fset}{{\cal F}}
\newcommand{\fsetmod}{\fset^{\mathrm{mod}}}
\newcommand{\lset}{{\cal L}}

\newcommand{\facts}{{\cal F}}

\newcommand{\adeftest}[1]{\baguard{\defined ({#1})}}

\newcommand{\vf}{u} %find variable

\newcommand{\pset}{{\cal Q}}
\newcommand{\cset}{{\cal C}h}
\newcommand{\tblcts}{{\cal T}}
\newcommand{\evseq}{\pp\mathit{{\cal E}\!v}}
\newcommand{\evseqnopp}{\mathit{{\cal E}\!v}}
\newcommand{\removepp}{\mathrm{removepp}}
\newcommand{\restconfig}{,\ab \pset,\ab \cset,\ab \tblcts,\ab \evseq}

\newcommand{\futfset}{{\cal F}^{\mathrm{Fut}}}

\newcommand{\ef}{{\cal F}^{\mathrm{ElseFind}}}
\newcommand{\elsefind}{\mathit{elsefind}}
\newcommand{\elselet}{\mathit{elselet}}
\newcommand{\ppf}{\mathit{programpoint}}
\newcommand{\lppf}{\mathit{lastdefprogrampoint}}
\newcommand{\convertelsefind}{\mathrm{convert\_elsefind}}
\newcommand{\Sdef}{S_{\mathrm{def}}}
\newcommand{\Sdep}{S_{\mathrm{dep}}}

\newcommand{\dset}{{\cal D}}
\newcommand{\dsetsnu}{{\cal D}_{\mathsf{SNU}}}
\newcommand{\indistev}[2]{\xrightarrow{#1}_{#2}}

\def\Succ{{\normalfont{\sf Succ}}}

\newcommand{\MAC}{\mathsf{MAC}}
\newcommand{\SE}{\mathsf{SE}}

\newcommand{\succufcma}{\Succ_\MAC^{\mathsf{uf-cma}}}
\newcommand{\succindcpa}{\Succ_\SE^{\mathsf{ind-cpa}}}

\newcommand{\pid}{\mathit{pid}}
\newcommand{\qid}{\mathit{qid}}

\newcommand{\secrone}{\mathsf{1\text{-}ses. secr.}}
\newcommand{\secr}{\mathsf{Secrecy}}
\newcommand{\secrreachone}{\mathsf{1\text{-}ses. reach. secr.}}
\newcommand{\secrreach}{\mathsf{Reach. secr.}}
\newcommand{\secrbit}{\mathsf{bit~secr.}}
\newcommand{\prop}{\mathit{sp}}
\newcommand{\tproc}[1]{Q_{#1}}               %test process for (one-session) secrecy
\newcommand{\tproco}[1]{\tproc{\secrone(#1)}}%test process for one-session secrecy
\newcommand{\tprocs}[1]{\tproc{\secr(#1)}}   %test process for secrecy
\newcommand{\tprocb}[1]{\tproc{\secrbit(#1)}}%test process for bit secrecy
\newcommand{\Stestpp}{\mathrm{Tpp}} %program points of successful test queries 
\newcommand{\Stestidx}{\mathrm{Tidx}}           %indices of successful test queries 
\newcommand{\defrestr}{\mathrm{defRand}}
\newcommand{\Dfalse}{D_{\false}}
\newcommand{\Tall}{T_{\mathsf{all}}}

\newcommand{\nth}[2]{\mathsf{nth}({#1},{#2})}

\newcommand{\iS}{i_s}
\newcommand{\nS}{n_s}
\newcommand{\cS}{c_s}
\newcommand{\cSz}{c_{s0}}
\newcommand{\vfS}{\vf_s}

\newcommand{\iT}{i_t}
\newcommand{\nT}{n_t}
\newcommand{\iR}{i_r}
\newcommand{\nR}{n_r}
\newcommand{\cR}{c_r}
\newcommand{\reveal}{\kvar{reveal}}

\newcommand{\Gru}{G_{\mathrm{RU}}}

\newtheorem{proofsk}{Proof sketch}

\newcommand{\proofcomplete}{\hspace*{\fill}$\Box$}
\newenvironment{proof}%
  {\begin{trivlist}\item[]{\normalsize\bf Proof}\hspace*{4mm}}%
  {\end{trivlist}}
\newenvironment{proofof}[1]%
  {\begin{trivlist}\item[]{\normalsize\bf Proof of #1}\hspace*{4mm}}%
  {\end{trivlist}}
  {\begin{trivlist}\item[]{\normalsize\bf Proof sketch of #1}\hspace*{4mm}}%
  {\end{trivlist}}

\newcommand{\corresp}{\varphi}
\newcommand{\nicorresp}{\corresp_{\mathrm{ni}}}
\newcommand{\prove}[1]{\mathrm{prove}^{#1}}
\newcommand{\noleak}{\mathrm{noleak}}
\newcommand{\distinct}{\mathsf{distinct}}

\newcommand{\Dnu}{D_{\mathsf{NU}}}
\newcommand{\Dnuarg}[1]{D_{\mathsf{NU}#1}}

\newcommand{\pp}{\mu}
\newcommand{\pptag}{{}^{\pp}}
\newcommand{\case}{c}

\newcommand{\shows}{\mathop{\models\!\!\!\!\!\Rightarrow}\nolimits}

\newcommand{\replidx}[1]{I_{#1}}
\newcommand{\sri}{{\cal I}} %sequence of replication indices
\newcommand{\vset}{{\cal V}}
\newcommand{\coll}{{\cal C}}
\newcommand{\scoll}{{\cal S}}
\newcommand{\sset}{{\cal S}}
\newcommand{\formula}{\mathrm{formula}}
\newcommand{\nonce}{\mathit{nonce}}
\newcommand{\keyseed}{\mathit{keyseed}}
\newcommand{\seed}{\mathit{seed}}
\newcommand{\host}{\mathit{host}}
\newcommand{\pkey}{\mathit{pkey}}
\newcommand{\concat}{\kwf{concat}}
\newcommand{\signature}{\mathit{signature}}
\newcommand{\verify}{\kwf{verify}}

\newcommand{\venv}{\rho}
\newcommand{\evalterm}{\Downarrow}

\newcommand{\noninj}{\mathrm{noninj}}

\newcommand{\stdtransfproba}[2]{\kwf{pstd}^{#1}_{#2}}

\newcommand{\algyc}{\varphi}
\newcommand{\mathf}{\psi}
\newcommand{\algtest}[3]{\mathrm{if}\ #1\ \mathrm{then}\ #2\ \mathrm{else}\ #3}
\newcommand{\yctran}[1]{\{\![#1]\!\}}

\newcounter{linenumber}
\makeatletter
\newcommand{\firstline}{\nextline}
\newcommand{\nextline}{\addtocounter{linenumber}{1}\arabic{linenumber}\gdef\@currentlabel{\arabic{linenumber}}&&}
\newcommand*{\linelabel}[1]{\ltx@label{#1}}% label does not work because amsmath redefines it.
\makeatother

\RRNo{RR-9525}
\RRdate{October 2023}
\RRauthor{Bruno Blanchet\thanks{Inria}}
\RRtitle{CryptoVerif:
  un vérificateur de protocoles cryptographiques sûr dans le modèle calculatoire\\
  {\smaller (Version initiale avec communication sur des canaux)}}
\RRetitle{CryptoVerif:
  a Computationally-Sound Security Protocol Verifier\\
  {\smaller (Initial Version with Communications on Channels)}}
\titlehead{CryptoVerif: A Computationally-Sound Security Protocol Verifier}
\RRresume{Ce document présente le vérificateur de protocoles cryptographiques CryptoVerif.
CryptoVerif ne s'appuie pas sur
le modèle symbolique de Dolev-Yao, mais sur le modèle calculatoire.
Il peut vérifier le secret, les correspondances (qui comprennent
l'authentification) et les propriétés d'indistinguabilité. Il produit
des preuves par suites de jeux, comme celles écrites manuellement
par les cryptographes~; ces jeux sont formalisés dans un
calcul de processus probabiliste.
CryptoVerif fournit une méthode générique pour spécifier les propriétés de sécurité
des primitives cryptographiques.
Il produit des preuves
valables pour un nombre quelconque de sessions du protocole, et fournit
une borne supérieure sur la probabilité de succès d'une attaque contre le
protocole en fonction de la probabilité de casser chaque primitive
et du nombre de sessions.
Il peut fonctionner automatiquement, ou l'utilisateur peut guider la preuve manuellement.}
\RRabstract{This document presents the security protocol verifier CryptoVerif.
CryptoVerif does not rely on
the symbolic, Dolev-Yao model, but on the computational model.  
It can verify secrecy, correspondence (which include 
authentication), and indistinguishability properties. It produces
proofs presented as sequences of games, like those manually written
by cryptographers; these games are formalized in a
probabilistic process calculus.
CryptoVerif provides a generic method for specifying security properties
of the cryptographic primitives.
It produces proofs 
valid for any number of sessions of the protocol, and provides
an upper bound on the probability of success of an attack against the
protocol as a function of the probability of breaking each primitive
and of the number of sessions.
It can work automatically, or the user can guide it with manual
proof indications.}
\RRmotcle{protocoles cryptographiques, vérification, modèle calculatoire}
\RRkeyword{security protocols, verification, computational model}
\RRprojet{Prosecco}
\RCParis

\begin{document}

\makeRR

\tableofcontents
\newpage

\section{Introduction}

There exist two main approaches for analyzing security 
protocols.
In the computational model, messages are bitstrings, and the
adversary is a probabilistic polynomial-time Turing machine.
This model is close to the real execution of protocols, but
the proofs are usually manual and informal.
In contrast, in the symbolic, Dolev-Yao model, cryptographic primitives
are considered as perfect blackboxes, modeled by function symbols
in an algebra of terms, possibly with equations. The adversary can 
compute using only these blackboxes. This abstract model makes it easier
to build automatic verification tools, but the security proofs are 
in general not sound with respect to the computational model.

In contrast to most previous protocol verifiers, CryptoVerif 
works directly in the
computational model, without considering the Dolev-Yao model.
It produces proofs valid for any number of sessions of the protocol,
in the presence of an active adversary.
These proofs are presented as sequences of games,
as used by cryptographers~\cite{Shoup01b,Shoup02,Bellare06}:
the initial game represents the protocol to prove;
the goal is to bound the probability of breaking a certain security
property in this game;
intermediate games are obtained each from the previous one
by transformations such that the difference of probability
between consecutive games can easily be bounded;
the final game is such that the desired probability is obviously bounded
from the form of the game. (In general, it is simply 0 in that game.) 
The desired probability can then be easily bounded
in the initial game. 

We represent games in a process calculus. This calculus
is inspired by the pi-calculus and by the calculi 
of~\cite{Mitchell04} and of~\cite{Laud05}.
In this calculus, messages are bitstrings, and cryptographic
primitives are functions from bitstrings to bitstrings.
The calculus has a probabilistic semantics.
The main tool for specifying security properties is indistinguishability:
$Q$ is indistinguishable from $Q'$ up to probability $p$, $Q \approx_p Q'$, when the adversary
has probability at most $p$ of distinguishing $Q$ from $Q'$.
With respect to previous calculi mentioned above, our calculus introduces
an important novelty which is key for the automatic proof of security
protocols: the values of all variables during the execution of
a process are stored in arrays. For instance, $x[i]$ is the value of 
$x$ in the $i$-th copy of the process that defines $x$.
Arrays replace lists often used by cryptographers in their
manual proofs of protocols. For example, consider the standard security assumption on a message authentication code (MAC). 
Informally, this definition says that the adversary has a negligible
probability of forging a MAC, that is, that all correct MACs have
been computed by calling the MAC oracle (\emph{i.e.}, function). So, in cryptographic proofs,
one defines a list containing the arguments of calls to the MAC oracle,
and when checking a MAC of a message $m$, one can additionally check that 
$m$ is in this list, with a negligible change in probability.
In our calculus, the arguments of the MAC oracle are stored in
arrays, and we perform a lookup in these arrays in order to find
the message $m$. Arrays make it easier to automate proofs since they
are always present in the calculus: one does not need to add explicit
instructions to insert values in them, in contrast to the lists used in
manual proofs. Therefore, many trivially sound 
but difficult to automate syntactic transformations disappear.
Furthermore, relations between elements of arrays 
can easily be expressed by equalities, possibly involving
computations on array indices. 

CryptoVerif relies on a collection of game transformations, in order to
transform the initial protocol into a game on which
the desired security property is obvious.
The most important kind of transformations exploits
the security assumptions on cryptographic primitives in order to
obtain a simpler game.
As described in Section~\ref{sec:primdef}, these transformations
can be specified in a generic way: we represent the security assumption of each cryptographic primitive by an observational
equivalence $L \approx_p R$, where the processes $L$ and $R$ encode oracles:
they input the arguments of the oracle and send its result back.
Then, the prover can automatically
transform a process $Q$ that calls the oracles of $L$
(more precisely, contains as subterms terms that perform the same
computations as oracles of $L$) into a process $Q'$ that
calls the oracles of $R$ instead.
We have used this technique to specify  
several variants of shared-key and public-key encryption, signature,
message authentication codes, hash functions, Diffie-Hellman key agreement, 
simply by giving
the appropriate equivalence $L \approx_p R$ to the prover.
Other game transformations are syntactic
transformations, used in order to be able to apply an 
assumption on a cryptographic primitive, or to simplify the game
obtained after applying such an assumption.

In order to prove protocols, these game transformations
are organized using a proof strategy based on advice: when a transformation
fails, it suggests other transformations that should be applied before,
in order to enable the desired transformation. Thanks to this strategy,
simple protocols can often be proved in a fully automatic way. 
For delicate cases, CryptoVerif has an interactive mode, in which
the user can manually specify the transformations to apply.
It is often sufficient to specify a few well-chosen case distinctions and
transformations coming from
the security assumptions of primitives, by indicating the
concerned cryptographic primitive and the concerned secret key if any; 
the prover infers 
the intermediate syntactic transformations by the advice strategy.
This mode is helpful for instance for proving some public-key protocols, 
in which several security assumptions on primitives can be applied,
but only one leads to a proof of the protocol.
Importantly, CryptoVerif is always sound: whatever indications the user
gives, when the prover shows a security property of
the protocol, the property indeed holds assuming the given assumptions
on the cryptographic primitives.

CryptoVerif has been implemented in OCaml (more than 60000 lines of code) and 
is available at \url{http://cryptoverif.inria.fr/}.

\paragraph{Related Work}
Various methods have been proposed for verifying security protocols
in the computational model.
Following the seminal paper by Abadi and Rogaway~\cite{Abadi2002b}, 
many results show the %Micciancio04b,Janvier06,Cortier06,Backes08b,
soundness of the Dolev-Yao model with respect to the computational
model, which makes it possible to use Dolev-Yao provers in order to
prove protocols in the computational model (see, e.g.,~\cite{Cortier05,Janvier05,Comon08,Backes09,CortierCCS11} 
and the survey~\cite{CortierJAR2011}). However, these results
have limitations, in particular in terms of allowed cryptographic
primitives (they must satisfy strong security properties so that they
correspond to Dolev-Yao style primitives), and they require some
restrictions on protocols (such as the absence of key cycles).
A tool~\cite{Cortier06b} was developed based on~\cite{Cortier05}
to obtain computational proofs using the formal verifier AVISPA,
for protocols that rely on public-key encryption and signatures.

Several frameworks exist for formalizing proofs of protocols in the
computational model. 
%Backes and Pfitzmann~\cite{Backes02,Backes03b} have developed a
%  semantics for information flow in the presence of cryptography,
%  under arbitrary attacks, in a computational model. 
Backes, Pfitzmann, and Waidner~\cite{Backes03,Backes04} %Backes03d,Backes05
designed an abstract cryptographic library and showed its soundness with 
respect to computational
primitives, under arbitrary active attacks. 
This framework has been used for a computationally-sound
machine-checked proof of the Needham-Schroeder-Lowe 
protocol~\cite{Sprenger06,Sprenger08}.
Canetti~\cite{Canetti01} introduced the notion of universal composability.
With Herzog~\cite{Canetti04}, they show how a Dolev-Yao-style
symbolic analysis can be used to prove security properties
of protocols within the
framework of universal composability, for a restricted class of 
protocols using public-key encryption as only cryptographic primitive.
Then, they use the automatic Dolev-Yao verification tool 
ProVerif~\cite{Blanchet04} for verifying protocols in this framework.
Process calculi have been designed for representing cryptographic
games, such as the probabilistic polynomial-time calculus of~\cite{Mitchell04}
and the cryptographic lambda-calculus of~\cite{Nowak10}.
% Lincoln, Mateus, Mitchell, Mitchell, Ramanathan, Scedrov, and
% Teague~\cite{Lincoln98,Lincoln99,Mateus03,Ramanathan04,Mitchell04}
% developed a probabilistic polynomial-time calculus for the analysis of
% cryptographic protocols. 
%
%\cite{Nowak10} = the calculus CSLR (a cryptographic lambda calculus for representing games)
%
Logics have also been designed for proving security protocols in the
computational model, such as the computational variant of PCL
(Protocol Composition Logic)~\cite{Turuani05,Datta06} and CIL
(Computational Indistinguishability Logic)~\cite{Barthe10}.
Canetti {\it et al.}~\cite{Canetti06b} use the framework of time-bounded
task-PIOAs (Probabilistic Input/Output Automata) to prove security
protocols in the computational model. This framework makes it possible
to combine probabilistic and non-deterministic behaviors.
These frameworks can be used to prove security properties of protocols
in the computational sense, but except for~\cite{Canetti04} which
relies on a Dolev-Yao prover, they have not been automated up to now,
as far as we know.

Several techniques have been used for directly mechanizing proofs
in the computational model.
Type systems~\cite{Laud05,Laud05b,Smith06,Courant07} provide 
computational security guarantees. For instance,
\cite{Laud05} handles shared-key and public-key encryption, with an unbounded
number of sessions, by relying on the Backes-Pfitzmann-Waidner
library. A type inference algorithm is given in~\cite{Backes06b}.
The recent tool OWL~\cite{Gancher23} also relies on a type system 
that provides computational security guarantees.
It supports MACs, public-key signatures, authenticated symmetric and public key encryption, random oracles, and the gap Diffie-Hellman assumption~\cite{Okamoto01}. It can prove secrecy and integrity properties.
In another line of research, 
a specialized Hoare logic was designed for proving asymmetric
encryption schemes in the random oracle
model~\cite{Courant08,Courant10}.

The tool CertiCrypt~\cite{Barthe09,Barthe08,Zanella09,Zanella10,Barthe11a} 
%has done significant progress towards Halevi's goal: this is a framework that 
enables the
machine-checked construction and verification of 
cryptographic proofs by sequences of games~\cite{Shoup04,Bellare06}. 
It relies on the
general-purpose proof assistant Coq, which is widely believed to be correct.
Nowak {\it et al.}~\cite{Nowak07,Nowak08,Affeldt09} 
follow a similar idea by providing Coq proofs for several cryptographic
primitives. 
% \cite{Nowak07} = proofs in Coq of El Gamal and hashed El Gamal
% \cite{Nowak08} = proofs in Coq for the Blum-Blum-Shoup pseudo random generator and the Goldwasser-Micali 
% \cite{Affeldt09} = game-based proofs in Coq for the Blum-Blum-Shoup pseudo random generator in assembly code 
More recently, frameworks for cryptographic proofs in Coq, FCF~\cite{Petcher15}, and in Isabelle, CryptHOL~\cite{Basin20}, have been designed and used for proving cryptographic schemes.
EasyCrypt~\cite{Barthe11b}, the successor of CertiCrypt, no longer generates Coq proofs, but provides a higher automation level by relying on SMT solvers, which makes the tool easier to use.
Even if it focuses more on cryptographic primitives and schemes than on protocols, it has been used for proving some protocols such as one-round key exchange~\cite{Barthe15}, e-voting~\cite{Cortier17}, AWS key management~\cite{Almeida19}, and distance bounding~\cite{Boureanu21}.
These frameworks and tools can perform more subtle reasoning than CryptoVerif, at the cost of more user effort: the user has to give all games and guide the proof that the games are indistinguishable. That becomes tedious for large protocols, which require many large games.

The tool Squirrel~\cite{Baelde21}
%\cite{Baelde22}
relies on a computationally sound logic that allows to write interactive proofs of, e.g., stateful protocols. Still, it currently proves a security notion weaker than the standard one: the number of sessions of the protocol must be bounded independently of the security parameter (instead of being polynomial in the security parameter).

Independently, we have built the tool CryptoVerif~\cite{Blanchet07c} to help
cryptographers, not only for the verification, but also by generating the proofs
by sequences of games~\cite{Shoup04,Bellare06}, automatically or with little 
user interaction.  In particular, CryptoVerif generates the games, possibly
using the indications of which transformations to perform.
This tool extends considerably early
work by Laud~\cite{Laud03,Laud04} which was limited either to
passive adversaries or to a single session of the protocol. More
recently, T{\v s}ahhirov and Laud~\cite{Tsahhirov07,Laud09} developed
a tool similar to CryptoVerif but that represents games by dependency
graphs.  It handles public-key and shared-key encryption and
proves secrecy properties; it does not provide bounds on the
probability of success of an attack.

\paragraph{Outline}
The next section presents our process calculus for representing games,
with its syntax, type system, formal semantics, as well as the definition
of security properties.
Section~\ref{sec:truefacts} collects information about games and reasons
using it.
Section~\ref{sec:success} gives
criteria for proving security properties of protocols. 
Section~\ref{sec:gametransf} describes the game transformations that
we use for proving protocols. Section~\ref{sec:strategy} explains how
the prover chooses which transformation to apply at each
point. 
% The appendices contain additional formal details, proof sketches, 
% and details on the modeling of some cryptographic primitives.

\paragraph{Notations}
We recall the following notations.
We denote by $\{ M_1/x_1, \ldots, \ab M_m/x_m\}$ the substitution
that replaces $x_j$ with $M_j$ for each $j \in \{1, \dots, m\}$.
The cardinal of a set or multiset $S$ is denoted by $|S|$.
Multisets $S$ are represented by functions that map each element $x$ of $S$
to the number of occurrences of $x$ in $S$, that is,
when $S$ is a multiset, $S(x)$ is the number of elements of $S$ equal to $x$.
We use $\multiunion$ for multiset union, defined by
$(S_1 \multiunion S_2)(x) = S_1(x) + S_2(x)$.
When $S$ and $S'$ are multisets, $\max(S, S')$ is the multiset 
such that $\max(S,S')(x) = \max(S(x), \ab S'(x))$.
The notation $\{ x_1 \mapsto a_1, \dots, x_m \mapsto a_m \}$
designates the function that maps $x_j$ for $a_j$ for each
$j \in \{1, \dots, m\}$ and is undefined for other inputs.
When $f$ is a function, $f[x \mapsto a]$ is the function that maps $x$ to
$a$ and all other elements as $f$.
If $S$ is a finite set, $x \randomchoice S$ chooses a random element
uniformly in $S$ and assigns it to $x$.
If $\Adv$ is a probabilistic algorithm, $x \leftarrow 
\Adv(x_1, \ldots, x_m)$ denotes the experiment of choosing random coins
$r$ and assigning to $x$ the result of running $\Adv(x_1, \ldots,
x_m)$ with coins $r$.
Otherwise, $x \leftarrow M$ is a simple assignment statement.
If $D$ is a discrete probability distribution, we denote by $D(a)$ 
the probability that $X = a$, $\Pr[X = a]$, where $X$ is a random variable 
with probability distribution~$D$.

\section{A Calculus for Cryptographic Games}\label{sec:process}

\subsection{Syntax and Informal Semantics}\label{sec:syntax}

\begin{figure}[tp]
\begin{defn}
\categ{M, N}{terms}\\ 
\entry{i}{replication index}\\
\entry{x[M_1, \ldots, M_m]}{variable access}\\ 
\entry{f(M_1, \ldots, M_m)}{function application}\\
\entry{\Res{x[\tup{i}]}{T}; N}{random number}\\
\entry{\assign{\pat}{M}{N}\ \ELSE N'}{assignment (pattern-matching)}\\
\entry{\assign{x[\tup{i}]:T}{M}{N}}{assignment}\\
\entry{\bguard{M}{N}{N'}}{conditional}\\
%\entry{\cfind{j=1}{m}{
%\vf_{j1}[\tup{i}] \leq n_{j1}, \ldots, \vf_{jm_j}[\tup{i}] \leq n_{jm_j}}{
%M_{j1}, \ldots, M_{jl_j}}{M_j}{P_j}{P}}{}\\
\entry{\kw{find}\uniqueopt\ (\mathop\bigoplus\nolimits_{j=1}^m
\vf_{j1}[\tup{i}] = i_{j1} \leq n_{j1}, \ldots, \vf_{jm_j}[\tup{i}] = i_{jm_j} \leq n_{jm_j}\ \kw{suchthat}}{}\\ 
\entry{\quad\kw{defined}(M_{j1}, \ldots, M_{jl_j}) \fand M'_j\ \kw{then}\ N_j)\ \ELSE N'}{array lookup}\\
\entry{\INSERT\ \tbl(M_1,\ldots,M_l);N}{insert in table}\\
\entry{\GET\uniqueopt\ \tbl(\pat_1,\ldots,\pat_l)\ \SUCHTHAT\ M\ \IN\ N\ \ELSE N'}{get from table}\\
\entry{\kevent{e(M_1, \ldots, M_l)}; N}{event}\\
\entry{\keventabort{e}}{event $e$ and abort}
\end{defn}
\begin{defn}
\categ{\pat}{pattern}\\
\entry{x[\tup{i}]:T}{variable}\\
\entry{f(\pat_1, \ldots, \pat_m)}{function application}\\
\entry{{=}M}{comparison with a term}
\end{defn}
\begin{defn}
\categ{Q}{input process}\\
\entry{0}{nil}\\
\entry{Q \parpop Q'}{parallel composition}\\
\entry{\repl{i}{n}{Q}}{replication $n$ times}\\
\entry{\Reschan{c}; Q}{channel restriction}\\
\entry{\cinput{c[M_1, \ldots, M_l]}{\pat}; P}{input}
\end{defn}
\begin{defn}
\categ{P}{output process}\\
\entry{\coutput{c[M_1, \ldots, M_l]}{N}; Q}{output}\\
\entry{\Res{x[\tup{i}]}{T}; P}{random number}\\
\entry{\assign{\pat}{M}{P}\ \ELSE P'}{assignment}\\
\entry{\bguard{M}{P}{P'}}{conditional}\\
\entry{\FIND\uniqueopt \ (\mathop\bigoplus\nolimits_{j=1}^m \vf_{j1}[\tup{i}] = i_{j1} \leq n_{j1}, \ldots, \vf_{jm_j}[\tup{i}] = i_{jm_j} \leq n_{jm_j}\ \SUCHTHAT}{}\\
\entry{\quad \defined (M_{j1}, \ldots, M_{jl_j}) \fand M_j \THEN P_j)\ \ELSE P}{array lookup}\\
\entry{\INSERT\ \tbl(M_1,\ldots,M_l);P}{insert in table}\\
\entry{\GET\uniqueopt\ \tbl(\pat_1,\ldots,\pat_l)\ \SUCHTHAT\ M\ \IN\ P\ \ELSE P'}{get from table}\\
\entry{\kevent{e(M_1, \ldots, M_l)}; P}{event}\\
\entry{\keventabort{e}}{event $e$ and abort}\\
\entry{\kw{yield}}{end}
\end{defn}
\caption{Syntax of the process calculus}\label{fig:syntax}
\end{figure}

CryptoVerif represents games in the syntax of Figure~\ref{fig:syntax}.
This calculus assumes a countable set of channel names, 
denoted by $c$.
It uses parameters, denoted by $n$, which are integers that bound the
number of executions of processes.

It also uses types, denoted by $T$, which are non-empty, countable
sets of values. We assume that there exists an efficient injection
from each type to the set of bitstrings, and that its inverse is
also efficiently computable.
A type is \emph{fixed} when it is the
set of all bitstrings of a certain length; a type is \emph{bounded}
when it is a finite set.
Particular types are predefined: $\bool
= \{ \true,\false\}$, where $\false$ is 0 and $\true$ is 1; 
$\bitstring$ is the set of all bitstrings; 
$\bitstringbot = \bitstring \cup \{ \bot \}$ where
$\bot$ is a special symbol; $[1,n]$ is the set of integers $\{1, \dots, n\}$,
where $n$ is a parameter.
(We consider integers as bitstrings without leading zeroes.)

The calculus also uses function symbols $f$. Each function
symbol comes with a type declaration $f : T_1 \times \ldots \times T_m
\rightarrow T$, and represents an efficiently computable, 
deterministic function that maps each tuple in 
$T_1 \times \ldots \times T_m$ to an element of $T$.  
Particular functions are predefined, and some of them use the infix
notation: $M = N$ for the equality test, $M \neq N$ for the inequality test 
(both taking two values of the same type 
$T$ and returning a value of type $\bool$), $M\for N$
for the boolean or, $M\fand N$ for the boolean and, $\fnot M$ for the
boolean negation (taking and returning values of type $\bool$),
tuples $(M_1, \ldots, M_m)$ (taking values of any types and returning
values of type $\bitstring$; tuples are assumed to provide 
unambiguous concatenation, with tags for the types of $M_1, \ldots, M_m$
so that tuples of different types are always different);
test $\iffun(M_1, M_2, M_3)$ (with a first argument of type $\bool$ and the
last two arguments of the same type $T$; it returns a value of type $T$:
$M_2$ when $M_1$ is $\true$ and $M_3$ when $M_1$ is $\false$).

In this calculus, terms represent computations on bitstrings.
The replication index $i$ is an integer which serves in distinguishing
different copies of a replicated process $\repl{i}{n}$.
(Replication indices are typically used as array indices.)
The variable access $x[M_1, \ldots, M_m]$ returns the content
of the cell of indices $M_1, \ldots, M_m$ of the $m$-dimensional 
array variable $x$.
We use $x, y, z, \vf$ as variable names.
The function application $f(M_1, \ldots, M_m)$ returns the
result of applying function $f$ to $M_1, \ldots, M_m$
Terms contain additional constructs which are very similar to
those also included in output processes and explained below.
These constructs conclude by evaluating a term, instead of
executing a process. The construct $\keventabort{e}$ executes event $e$
(without argument) and aborts the game.

The calculus distinguishes two kinds of processes: input processes
$Q$ are ready to receive a message on a channel; output processes $P$ 
output a message on a channel after executing some internal computations.
The input process 0 does nothing; $Q \parpop Q'$
is the parallel composition of $Q$ and $Q'$; 
$\repl{i}{n}Q$ represents $n$ copies of $Q$ in parallel,
each with a different value of $i \in [1,n]$;
$\Reschan{c};Q$ creates a new private channel $c$ and executes $Q$;
this construct is useful in proofs, but does not occur
in games manipulated by CryptoVerif.
The semantics of the input $c[M_1, \ldots, M_l](\pat);P$
will be explained below together with the
semantics of the output. 

The output process $\Res{x[\tup{i}]}{T};P$ chooses a new
random value in $T$,
stores it in $x[\tup{i}]$, and executes $P$.
The abbreviation $\tup{i}$ stands for a sequence of replication indices
$i_1, \ldots, i_m$.
The random value is chosen according to the default distribution $D_T$
for type $T$, which is determined as follows:
\begin{itemize}

\item When the type $T$ is declared with option \emph{nonuniform}, 
the default probability 
distribution $D_T$ for type $T$ may be non-uniform. It is left unspecified.

\item Otherwise, if $T$ is \emph{fixed}, $T$ consists of all bitstrings
of a certain length, and the default distribution is the uniform distribution.
The probability of each element of $T$ is $1/|T|$.

\item If $T$ is \emph{bounded} but not \emph{fixed}, $T$ is finite, 
and the default distribution is 
an approximately uniform distribution, such that its distance
to the uniform distribution is at most $\epsilon_T$. 
The distance between
two probability distributions $D_1$ and $D_2$ for type $T$ is
\[d(D_1, D_2) = \frac{1}{2} \sum_{a \in T} \left| D_1(a) - D_2(a) \right|\]
Indeed, probabilistic Turing machines that run in bounded time cannot choose random elements exactly uniformly in sets whose cardinal is not a power of 2.

For example, a possible algorithm to obtain a random integer in $[0,
m-1]$ is to choose a random integer $x'$ uniformly among $[0, 2^k-1]$
for a certain $k$ large enough and return $x' \bmod m$. 
By euclidean division, we have $2^k = qm+r$ with $r \in [0,m-1]$.
With this algorithm
\[D(a) = \begin{cases}
\frac{q+1}{2^k} &\text{if }a \in [0,r-1]\\
\frac{q}{2^k} &\text{if }a \in [r,m-1]
\end{cases}\]
so
\[\left|D(a) -\frac{1}{m}\right| = \begin{cases}
\frac{q+1}{2^k} - \frac{1}{m}&\text{if }a \in [0,r-1]\\
\frac{1}{m} - \frac{q}{2^k}&\text{if }a \in [r,m-1]
\end{cases}\]
Therefore
\[\begin{split}
d(D_T, \mathit{uniform}) 
&=  \frac{1}{2} \sum_{a \in T} \left| D(a) - \frac{1}{m} \right|
= \frac{1}{2} r\left(\frac{q+1}{2^k} - \frac{1}{m}\right) - \frac{1}{2} (m-r)\left(\frac{1}{m} - \frac{q}{2^k}\right)\\
&=\frac{r(m-r)}{m.2^k} \leq \frac{m}{2^{k+1}}
\end{split}\]
%If r <= m/2, we upper bound r(m-r) by m/2.m
%If r >= m/2, m-r <= m/2, so we upper bound r(m-r) by m.m/2 
so we can take $\epsilon_T = \frac{m}{2^{k+1}}$. A given precision of $\epsilon_T = \frac{1}{2^{k'}}$ can be obtained by choosing $k = (k' + \text{number of bits of }m)$ random bits. 

By default, CryptoVerif does not display $\epsilon_T$ in probability
formulas, to make them more readable.

\end{itemize}
When $T$ is not declared with any of the options $\emph{nonuniform}$,
$\emph{fixed}$, or $\emph{bounded}$, CryptoVerif rejects the construct
$\Res{x[\tup{i}]}{T};P$.
Function symbols represent deterministic functions,
so all random numbers must be chosen by $\Res{x[\tup{i}]}{T}$.
Deterministic functions make automatic syntactic manipulations
easier: we can duplicate a term without changing its value.

The process $\assign{x[\tup{i}]:T}{M}{P}$
stores the value of $M$ (which must be in $T$)
in $x[\tup{i}]$ and executes $P$.
Furthermore, we say that a function 
$f : T_1 \times \ldots \times T_m \rightarrow T$
is \emph{efficiently injective} when it is injective and its inverses 
are efficiently computable,
that is, there exist functions $f_j^{-1}: T \rightarrow T_j$ 
($1 \leq j \leq m$) such that 
$f_j^{-1}(f(x_1, \ldots, x_m)) = x_j$ and $f_j^{-1}$ is efficiently computable. 
When $f$ is efficiently injective, we define a pattern 
matching construct
$\assign{f(x_1, \ldots, x_m)}{M}{P}\ \ELSE P'$ as an abbreviation for
$\assign{y:T}{M}\ab
\assign{x'_1:T_1}{f_1^{-1}(y)}\ab
\ldots \ab
\assign{x'_m:T_m}{f_m^{-1}(y)}\ab
\bguard{f(x'_1, \ldots, x'_m) = y}{(\assign{x_1:T_1}{x'_1}\ldots \assign{x_m:T_m}{x'_m}P)}{P'}$
where $y, x'_1, \dots, x'_m$ are fresh variables.
(The variables $x'_1, \dots, x'_m$ are introduced to make sure that none
of the variables $x_1, \dots, x_m$ is defined when the pattern-matching fails.)
We naturally generalize this construct to $\assign{\pat}{M}{P}\ \ELSE P'$
where $\pat$ is built from variables, efficiently injective functions, and
equality tests. When $\pat$ is simply a variable, the pattern-matching
always succeeds, so the $\kw{else}$ branch of the assignment is
never executed and can be omitted. 

The process $\kevent{e(M_1,\ldots, M_l)}; P$ executes the event
$e(M_1,\ldots, M_l)$, then runs $P$. This event records that a certain
program point has been reached with certain values of $M_1, \ldots,
M_l$, but otherwise does not affect the execution of the process.
Events are used in particular for specifying security properties.

The process $\keventabort{e}$ executes event $e$ (without argument)
and aborts the game.

Next, we explain the process 
$\FIND \uniqueopt \ (\mathop\bigoplus\nolimits_{j=1}^{m} \vf_{j1}[\tup{i}] = i_{j1} \leq n_{j1}, \ldots, \vf_{jm_j}[\tup{i}] = i_{jm_j} \leq n_{jm_j}$ $\SUCHTHAT $ $\defined (M_{j1}, \ldots, M_{jl_j})\fand M_j\THEN P_j)\ \ELSE P$.
The order and array indices on tuples are taken component-wise, so for
instance, $\vf_{j1}[\tup{i}] = i_{j1} \leq n_{j1}, \ldots, \vf_{jm_j}[\tup{i}]
= i_{jm_j} \leq n_{jm_j}$ can be further abbreviated $\tup{\vf_j}[\tup{i}]
= \tup{i_j} \leq \tup{n_j}$.
A simple example is the following:
$\FIND$ $\vf = i \leq n$ $\SUCHTHAT$ $\defined(x[i]) \fand {x[i] = a}$
$\kw{then}$ $P'$ $\kw{else}$ $P$
tries to find an index $i$ such that $x[i]$ is defined and
$x[i] = a$, and when such an $i$ is found, it stores it in $\vf$ and
executes $P'$ with that value of $\vf$;
otherwise, it executes $P$.
In other words, this $\FIND $ construct looks for the value
$a$ in the array $x$, and when $a$ is found, it stores in
$\vf$ an index such that $x[\vf] = a$. Therefore, the $\FIND $ construct
allows us to access arrays, which is key for our purpose.
More generally, $\FIND \ \vf_{1}[\tup{i}] = i_1 \leq n_{1}, \ldots, \vf_{m}[\tup{i}] = i_m \leq n_{m}$ $\SUCHTHAT $ $\defined (M_{1}, \ldots, M_{l})\fand M\THEN P'\ \ELSE P$ tries to find values of $i_1, \ldots, i_m$ for which
$M_1, \ldots, M_l$ are defined and $M$ is true. In case of success, it
stores the obtained values in $\vf_{1}[\tup{i}], \ldots, \vf_{m}[\tup{i}]$ and 
executes $P'$. In case of failure, it executes $P$.
This is further generalized to $m$ branches: $\FIND \ (\mathop\bigoplus\nolimits_{j=1}^{m} \vf_{j1}[\tup{i}] = i_{j1} \leq n_{j1}, \ldots, \vf_{jm_j}[\tup{i}] = i_{jm_j} \leq n_{jm_j}$ $\SUCHTHAT $ $\defined (M_{j1}, \ldots, M_{jl_j})\fand M_j\THEN P_j)\ \ELSE P$ tries to find a branch $j$ in $[1,m]$ such that there are 
values of $i_{j1}, \ldots, i_{jm_j}$ for which 
$M_{j1}, \ldots, M_{jl_j}$ are defined and $M_j$ is true. In case of 
success, it stores them in $\vf_{j1}[\tup{i}], \ldots, \vf_{jm}[\tup{i}]$ and executes $P_j$.
In case of failure for all branches, it executes $P$. 
More formally, it evaluates the conditions
$\defined(M_{j1}, \ldots, M_{jl_j}) \fand M_j$ for each $j$ and
each value of $i_{j1}, \ldots, i_{jm_j}$ in $[1, n_{j1}] 
\times \ldots \times [1, n_{jm_j}]$.
If none of these conditions is $\true$, it executes $P$.
Otherwise, it chooses randomly 
one $j$ and one value of $i_{j1}, \ldots, i_{jm_j}$
such that the corresponding condition is $\true$, according to the
distribution $D_{\FIND}(S)$ where $S$ is the set of possible solutions 
$j, i_{j1}, \ldots, i_{jm_j}$, stores
it in $\vf_{j1}[\tup{i}], \ldots, \vf_{jm_j}[\tup{i}]$, and executes $P_j$.
The distribution $D_{\FIND}(S)$ is almost uniform: formally,
the distance between $D_{\FIND}(S)$ and the uniform distribution 
is at most $\epsilon_{\FIND}/2$, that is, $d(D_{\FIND}(S), \mathit{uniform}) \leq \epsilon_{\FIND}/2$,
that is, $\frac{1}{2} \sum_{v \in S} \left| D_{\FIND}(S)(v) - \frac{1}{|S|} \right| \leq \frac{\epsilon_{\FIND}}{2}$,
so that, when $|S| = |S'|$, for any bijection $\phi : S \rightarrow S'$,
$\frac{1}{2} \sum_{v \in S} \left| D_{\FIND}(S)(v) - D_{\FIND}(S')(\phi(v)) \right| \leq \epsilon_{\FIND}$.
Moreover $D_{\FIND}(S)(v_i) = D_{|S|}(i)$ where $S = \{ v_1, \dots, v_{|S|} \}$ with the values $v_i$ ordered in increasing order lexicographically, for some distribution $D_{|S|}$ that depends only on the cardinal of $S$. In other words, the probability of a value $v_i$ in the distribution $D_{\FIND}(S)$ does not depend on the values in the set $S$ but only on the number $|S|$ of elements of $S$ and on the position $i$ of the value $v_i$ in $S$ ordered in increasing order lexicographically. Therefore, transformations that do not modify the number of successful values nor their order, that is, transformations that map elements $v$ of $S$ to elements $\phi(v)$ of $S'$ at the same position $i$, preserve the probabilities exactly and we do not need to add $\epsilon_{\FIND}$ to the probability when we apply such a transformation. This is true for instance when we remove a branch of $\FIND$ that is never taken.
By default, CryptoVerif does not display $\epsilon_{\FIND}$ in probability
formulas, to make them more readable.
We cannot take the first element found because the game transformations made by CryptoVerif may reorder the elements. For these transformations to preserve the behavior of the game, the distribution of the chosen element must be invariant by reordering, up to a small probability $\epsilon_{\FIND}$.
In this definition, the variables $i_{j1}, \ldots, i_{jm_j}$ are
considered as replication indices, while $\vf_{j1}[\tup{i}], \ldots, 
\vf_{jm_j}[\tup{i}]$ are considered as array variables.
The indication $\uniqueopt$ stands for either $\unique{e}$
or empty. The empty case has just been explained. 
When the $\FIND$ is marked $\unique{e}$ and there are several solutions
that make the condition of the $\FIND$ evaluate to $\true$,
we execute the event $e$ and abort
the game. When there is zero or one solution,
the $\FIND$ is executed as when $\uniqueopt$ is empty. 
This semantics allows us to perform game transformations
that require the $\FIND$ to have a single solution.

The conditional $\bguard{M}{P}{P'}$
executes $P$ if $M$ evaluates to $\true$. Otherwise, it executes $P'$. 
\bbnote{At some point, I said that $\bguard{M}{P}{P'}$ was a particular
case of $\FIND$, but that does not work well, because it does not satisfy
the same invariants: variables defined in $M$ may have array accesses, in contrast to those defined in conditions of $\FIND$; $\kw{event}$ and $\INSERT$ are allowed in $M$.}%
CryptoVerif also supports 
the conditional $\bguard{\defined(M_1, \ldots, M_l) \fand M}{P}{P'}$, which
executes $P$ if $M_1, \ldots, M_l$ are defined and $M$ evaluates to
$\true$. Otherwise, it executes $P'$. This conditional is internally
encoded as
$\FIND$ $\SUCHTHAT$ $\defined(M_1, \ldots, M_l) \fand M$ $\kw{then}$ $P$ $\ELSE P'$. The conjunct $M$ can be omitted when it is $\true$, writing
$\bguard{\defined(M_1, \ldots, M_l)}{P}{P'}$.

The constructs $\kw{insert}$ and $\kw{get}$ handle tables, used for instance
to store the keys of the protocol participants. A table can be represented 
as a list of tuples; $\kw{insert}$ $\tbl(M_1, \ab \ldots, \ab M_l); P$ inserts the
element $M_1, \ab \ldots, \ab M_l$ in the table $\tbl$; $\GET$ $\tbl(x_1:T_1,\ldots,x_l:T_l)$ $\SUCHTHAT$ $M$ $\IN$ $P$ $\kw{else}$ $P'$ tries to retrieve an element $(x_1, \ab \ldots, \ab x_l)$ in the table $\tbl$ such that $M$ is true. When such an element is found, it executes $P$ with $x_1, \ab \ldots, \ab x_l$ bound to that element. (When several such elements are found, one of them is chosen randomly according to distribution $D_{\kw{get}}(S)$ where $S$ is the set of indices of suitable elements, with $d(D_{\kw{get}}(S), \mathit{uniform}) \leq \epsilon_{\FIND}/2$.) When no such element is found, $P'$ is executed.
We can generalize this construct to patterns instead of variables similarly
to the $\kw{let}$ case.
As in the case of $\FIND$, the indication $\uniqueopt$ stands for either $\unique{e}$
or empty. The empty case has just been explained. 
When the $\GET$ is marked $\unique{e}$ and there are several solutions,
we execute the event $e$ and abort
the game. When there is zero or one solution,
the $\GET$ is executed as when $\uniqueopt$ is empty. 
CryptoVerif internally translates the $\kw{insert}$ and $\kw{get}$ constructs into $\FIND$.

Let us explain the output 
$\overline{c[M_1, \ldots, M_l]}\langle N\rangle;Q$.
A channel $c[M_1, \ab \ldots, \ab M_l]$ consists of both a channel name 
$c$ and a tuple of terms
$M_1, \ldots, M_l$. Channel names $c$ can be declared private by $\Reschan{c}$;
the adversary can never have access to channel $c[M_1, \ab \ldots, \ab M_l]$ when $c$ is private.
(This is useful in the proofs, although all channels of protocols
are often public.)
Terms $M_1, \ldots, M_l$ are intuitively analogous to IP addresses and
ports, which are numbers that the adversary may guess.
A semantic configuration always consists
of a single output process (the process currently being executed) 
and several input processes.
When the output process executes $\coutput{c[M_1, \ldots, M_l]}{N};Q$, 
one looks
for an input on channel $c[M'_l \ldots, M'_l]$, where $M'_1, \ldots, M'_l$
evaluate to the same bitstrings as $M_1, \ldots, M_l$, in the available input processes. 
If no such input process is found, the process blocks.
Otherwise, one such input process 
$c[M'_1, \ldots, M'_l](x[\tup{i}]:T);P $ is chosen randomly according
to the probability distribution $D_{\kw{in}}(S)$ where $S$ is the multiset of 
suitable input processes. The communication is then executed:
the output message $N$ is evaluated and
stored in $x[\tup{i}]$ if it is in $T$
(otherwise the process blocks). Finally, the output process $P$
that follows the input is executed. The input process $Q$ that
follows the output is stored in the available input processes
for future execution. 
The input construct can be generalized to patterns instead of variables 
similarly to the $\kw{let}$ case; when pattern-matching fails, 
the input process executes $\kw{yield}$.
The syntax requires an output
to be followed by an input process, as in~\cite{Laud05}. If one
needs to output several messages consecutively, one can simply
insert fictitious inputs between the outputs. The adversary can
then schedule the outputs by sending messages to these inputs.

Using different channels for each input and output allows the adversary
to control the network.\label{distinctchannels} For instance, we may write 
$\repl{i}{n} \cinput{c[i]}{x[i]:T} \ldots \coutput{c'[i]}{M}\ldots$
The adversary can then decide which copy of the replicated process
receives its message, simply by sending it on $c[i]$ for the
appropriate value of $i$.

The $\kw{yield}$ construct is an abbreviation for
$\coutput{\mathit{yield}}{()}$.
By performing an output, this construct returns control
to the adversary, which is going to receive the message.
An $\kw{else}$ branch of $\FIND$, $\kw{if}$, $\kw{get}$, 
or $\kw{let}$  may be omitted
when it is $\ELSE \kw{yield}$.
(Note that ``$\ELSE 0$'' would not be syntactically correct.)
Similarly, $; \kw{yield}$ may be omitted after $\kw{event}$,
$\kw{new}$, or $\kw{insert}$ and $\kw{in}\ \kw{yield}$ may be omitted 
after $\kw{let}$.
A trailing 0 after an output may be omitted.

The \emph{current replication indices} at a certain program point
in a process are the replication indices $i_1, \ldots, i_m$ 
bound by replications and $\kw{find}$ above that program point. 
The replication $\repl{i}{n} Q$ binds the replication index $i$ in $Q$.
The $\kw{find}$ construct
$\FIND\uniqueopt \ (\mathop\bigoplus\nolimits_{j=1}^m \vf_{j1}[\tup{i}] = i_{j1} \leq n_{j1}, \ldots, \vf_{jm_j}[\tup{i}] = i_{jm_j} \leq n_{jm_j}\ \SUCHTHAT\ 
\defined (M_{j1}, \ldots, M_{jl_j}) \fand M_j \THEN \ldots)\ \ELSE \ldots$
binds the replication indices $i_{j1}$, \ldots, $i_{jm_j}$ in
$\defined (M_{j1}, \ldots, M_{jl_j}) \fand M_j$. 
We often abbreviate
$x[i_1, \ldots, i_m]$ by $x$ when $i_1, \ldots, i_m$ are the current
replication indices at the definition of $x$, but it should be kept in
mind that this is only an abbreviation.  
Variables defined under a replication must be arrays: for example
$\repl{i_1}{n_1}\ldots\repl{i_m}{n_m} \assign{x[i_1, \ldots, i_m]: T}{M}\ldots$  
More formally,
we require the following invariant:

\begin{invariant}[Single definition]
\label{inv1}
The process $Q_0$ satisfies Invariant~\ref{inv1} if and only if
\begin{enumerate}

\item \label{inv1p1}in every definition of $x[i_1, \ldots, i_m]$ in $Q_0$, 
the indices $i_1, \ldots, i_m$ of $x$ are the current replication
indices at that definition, and

\item \label{inv1p2}two different definitions of the same variable $x$ in $Q_0$
  are in different branches of a $\FIND$, $\kw{if}$ (or $\kw{let}$), or $\kw{get}$.

  In a $\kw{let}$ with pattern-matching, $\assign{\pat}{M}{P}\ \ELSE P'$, the variables bound by $\pat$ are considered to be defined in the $\kw{in}$ branch; however, the variables defined in $M$ and in terms included in the pattern $\pat$ are defined before the branching.

  In $\GET\uniqueopt\ \tbl(\pat_1,\ldots,\pat_l)\ \SUCHTHAT\ M\ \IN\ P\ \ELSE P'$, the variables bound by $\pat_j$ for $j \leq l$ and the variables defined in $M$ and in terms included in the patterns $\pat_j$ are (temporarily) defined before the branching.

\end{enumerate}
\end{invariant}
Invariant~\ref{inv1} guarantees that each variable
is assigned at most once for each value of its indices. (Indeed,
item~\ref{inv1p2} shows that only one definition of each variable
can be executed for given indices in each trace.)
A definition of $x[\tup{i}]$ can be $\Res{x[\tup{i}]}{T}$,
a $\kw{let}$, $\kw{get}$, or input that contains the pattern
$x[\tup{i}] : T$, or $\kw{find} \ldots x[\tup{i}] = i \leq n \ldots$.

\begin{invariant}[Defined variables]
\label{inv2}
The process $Q_0$ satisfies Invariant~\ref{inv2} if and only if
every occurrence of a variable access $x[M_1, \ldots, M_m]$ in
$Q_0$ is either
\begin{itemize}

\item syntactically under the definition of $x[M_1, \ldots, M_m]$
(in which case $M_1, \ldots, M_m$ are in fact the current replication 
indices at the definition of $x$);

\item or in a $\defined $ condition in a $\FIND $ process or term;

\item or in $M'_j$ in a process or term of the form 
%$\cfind{j=1}{m''}{\tup{\vf_j}[\tup{i}] \leq \tup{n_j}}{
%M'_{j1}, \ldots, M'_{jl_j}}{M'_j}{P_j}{P}$ 
$\FIND$ $(\mathop\bigoplus\nolimits_{j=1}^{m''} \tup{\vf_j}[\tup{i}] = \tup{i_j} \leq \tup{n_j}$ $\SUCHTHAT$ $\defined (M'_{j1}, \ab \ldots, \ab M'_{jl_j}) \fand M'_j$ $\kw{then}$ $P_j)$ $\kw{else}$ $P$
where for some $k \leq l_j$, $x[M_1, \ldots, M_m]$ is a subterm of $M'_{jk}$.

\item or in $P_j$ in a process or term of the form 
%$\cfind{j=1}{m''}{\tup{\vf_j}[\tup{i}] \leq \tup{n_j}}{
%M'_{j1}, \ldots, M'_{jl_j}}{M'_j}{P_j}{P}$ 
$\FIND$ $(\mathop\bigoplus\nolimits_{j=1}^{m''} \tup{\vf_j}[\tup{i}] = \tup{i_j} \leq \tup{n_j}$ $\SUCHTHAT$ $\defined (M'_{j1}, \ab \ldots, \ab M'_{jl_j}) \fand M'_j$ $\kw{then}$ $P_j)$ $\kw{else}$ $P$
where for some $k \leq l_j$, there is a subterm $N$ of $M'_{jk}$
such that $N \{ \tup{\vf_j}[\tup{i}] / \tup{i_j} \} = x[M_1, \ldots, M_m]$.

\end{itemize}
\end{invariant}
Invariant~\ref{inv2} guarantees
that variables can be accessed only when they have been initialized. 
It checks that the definition of the variable access
is either in scope (first item) or checked by a $\FIND $ (last two items).
The scope of variable definitions is defined as follows:
$x[\tup{i}]$ is syntactically under its definition when it is 
\begin{itemize}
\item inside $P$ in $\Res{x[\tup{i}]}{T}; P$;
\item inside $N$ in $\Res{x[\tup{i}]}{T}; N$;
\item inside $P$ in $\assign{\pat}{M}{P}\ \ELSE P'$ when $x[\tup{i}] : T$
is bound in the pattern $\pat$;
\item inside $N$ in $\assign{\pat}{M}{N}\ \ELSE N'$ when $x[\tup{i}] : T$
is bound in the pattern $\pat$;
\item inside $N$ in $\assign{x[\tup{i}]:T}{M}{N}$;
\item inside $P_j$ in $\FIND\uniqueopt \ (\mathop\bigoplus\nolimits_{j=1}^m \vf_{j1}[\tup{i}] = i_{j1} \leq n_{j1}, \ldots, \vf_{jm_j}[\tup{i}] = i_{jm_j} \leq n_{jm_j}$ $\SUCHTHAT$ $\defined (M_{j1}, \ldots, M_{jl_j}) \fand M_j \THEN P_j)\ \ELSE P$
when $x$ is $\vf_{jk}$ for some $k \leq m_j$;
\item inside $N_j$ in $\FIND\uniqueopt \ (\mathop\bigoplus\nolimits_{j=1}^m \vf_{j1}[\tup{i}] = i_{j1} \leq n_{j1}, \ldots, \vf_{jm_j}[\tup{i}] = i_{jm_j} \leq n_{jm_j}$ $\SUCHTHAT$ $\defined (M_{j1}, \ldots, M_{jl_j}) \fand M_j \THEN N_j)\ \ELSE N$
when $x$ is $\vf_{jk}$ for some $k \leq m_j$;
\item inside $M$ or $P$ in $\GET\uniqueopt\ \tbl(\pat_1,\ldots,\pat_l)\ \SUCHTHAT\ M\ \IN\ P\ \ELSE P'$ when $x[\tup{i}] : T$ is bound in one of the patterns $\pat_1, \ldots, \pat_l$;
\item inside $M$ or $N$ in $\GET\uniqueopt\ \tbl(\pat_1,\ldots,\pat_l)\ \SUCHTHAT\ M\ \IN\ N\ \ELSE N'$ when $x[\tup{i}] : T$ is bound in one of the patterns $\pat_1, \ldots, \pat_l$;
\item inside $P$ in $\cinput{c[M_1, \ldots, M_l]}{\pat}; P$ when $x[\tup{i}] : T$ is bound in the pattern $\pat$.
\end{itemize}
A variable access that does not correspond to the first item of 
Invariant~\ref{inv2} is called an \emph{array access}.
We furthermore require the following invariant.

%not to do: move the last 3 invariants to the section on subsets used in
%the sequence of games? They are not necessary for a sound definition
%of the language, they are just used in the CryptoVerif implementation.
%For the definition of the semantics, it is better that
%- channels of inputs are only replication indices, variables, and function applications (so they are deterministic and never abort)
%- event does not occur in conditions of find and get (so that they do not abort)
%- there are no array accesses to variables defined in patterns and conditions of get (since I use the same variables during the lookup and afterwards).
%==> Even if only relaxed versions of these invariants are needed to define the semantics, I think it is better to keep them here.

\begin{invariant}[Variables defined in $\kw{find}$ and $\kw{get}$ conditions]\label{invfc}
The process $Q_0$ satisfies Invariant~\ref{invfc} with public variables $V$
if and only if
the variables defined in conditions of $\kw{find}$
and the variables defined in patterns and in conditions of $\kw{get}$ have
no array accesses and are not in the set of variables $V$.
\end{invariant}
These conditions are needed for variables of $\kw{get}$, because
they will be transformed into variables defined in conditions of $\kw{find}$
by the transformation of $\kw{get}$ into $\kw{find}$.

\begin{invariant}[Terms in $\kw{find}$ and $\kw{get}$ conditions]\label{invtfc}
  The process $Q_0$ satisfies Invariant~\ref{invtfc} if and only if
  $\kw{event}$ and $\kw{insert}$ do not occur in conditions of $\kw{find}$
  and $\kw{get}$. 
\end{invariant}
%DONE modify the implementation to match this invariant
%
Invariant~\ref{invtfc} guarantees that evaluating the condition of a
$\kw{find}$ or $\kw{get}$ does not change the state of the system.

\begin{definition}
A term is \emph{simple} when it contains only replication indices, variables, and function applications.
\end{definition}

\begin{invariant}[Terms in input channels and $\defined$ conditions]\label{invtic}
\ The process $Q_0$ satisfies Invariant~\ref{invtic} if and only if
all terms in input channels $c[M_1, \ab \ldots, \ab M_l]$ and in
conditions $\defined(M_1, \ab \ldots, \ab M_l)$ in $\FIND$ are simple.
\end{invariant}
Terms that are not simple
are handled by expanding them into their corresponding processes.
Invariant~\ref{invtic} is needed because terms in input channels
cannot be expanded, as we need an output process to put
the computations coming from expanded terms,
and similarly terms in $\defined$ conditions cannot be expanded
(see the transformation $\rn{expand}$ in Section~\ref{sec:transfexpand}).
By combining this invariant with Invariant~\ref{inv2}, 
we see that the terms of all variable accesses $x[M_1, \ldots, M_m]$
are simple.

The last 3 invariants did not appear in previous versions of the calculus
because all terms were simple.

\begin{invariant}[Events]\label{inv:event}
We distinguish three disjoint sets of events $e$, Shoup events, non-unique events, and other events. The process $Q_0$ satisfies Invariant~\ref{inv:event} if and only if
\begin{itemize}
\item Shoup events occur only in processes of the form
$\keventabort{e}$ in $Q_0$, 
\item non-unique events occur only in $\FIND\unique{e}$ or $\GET\unique{e}$ in $Q_0$, and
\item other events occur in $\kevent{e(M_1,\dots, M_l)}$ or in $\keventabort{e}$ in $Q_0$.
\end{itemize}
\end{invariant}

The name ``Shoup events'' is used because these events are introduced
when applying Shoup's lemma~\cite{Shoup04} (see Section~\ref{sec:insertevent}).
The non-unique events are those triggered when a $\FIND\unique{e}$ or $\GET\unique{e}$ actually has several solutions.
%DONE check that this invariant is preserved in the implementation, that is, we always abort immediately after a Shoup event. Would it be simpler to have a special construct event_abort(e) that would abort immediately after event e? Yes do that! Currently at least the manual insertion instructions can put things between event and abort. The event construct in terms should also become event_abort.

All these invariants are checked by the prover for the initial game and
preserved by all game transformations.

We denote by $\fvar(P)$ the set of variables that occur in $P$,
$\vardef(P)$ the set of variables defined in $P$
($\fvar(P)$ may contain more variables than $\vardef(P)$
in case some variables are read using $\FIND$ but never defined),
and by $\fc(P)$ the set of free channels of $P$.
(We use similar notations for input processes.)

\subsection{Example}

Let us introduce two cryptographic primitives that we use below.

\begin{definition}\label{def:macsec}
Let $T_{mk}$ and $T_{ms}$ be types that correspond intuitively
to keys and message authentication codes, respectively;
$T_{mk}$ is a fixed-length type.
A message authentication code scheme $\MAC$~\cite{Bellare00c} 
consists of two function symbols:
\begin{itemize}

\item $\mac : \bitstring \times T_{mk} \rightarrow T_{ms}$ 
is the MAC algorithm taking as
arguments a message and a key, and returning the
corresponding tag. (We assume here that $\mac$ is deterministic; we
could easily encode a randomized $\mac$ by adding random coins as
an additional argument.)

\item $\mverify : \bitstring \times T_{mk} \times T_{ms} \rightarrow \bool$
is a verification algorithm such that $\mverify(m,k,t) = \true$
if and only if $t$ is a valid MAC of
message $m$ under key $k$.  (Since $\mac$ is
deterministic, $\mverify(m,k,t)$ is typically 
$\mac(m,k) = t$.)

\end{itemize}
We have $\forall m \in \bitstring, \forall k \in T_{mk},
\mverify(m,\ab k,\mac(m,k)) = \true$.

The advantage of an adversary against unforgeability under chosen message attacks (UF-CMA) is
\[\succufcma(t, q_m, q_v, l) = \max_{\Adv}\, \Pr\left[\begin{array}{@{}l@{}}
k \randomchoice T_{mk}; 
(m,s) \leftarrow \Adv^{\mac(.,k), \mverify(.,k,.)} : \mverify(m,k,s)\\
{} \wedge m\text{ was never queried to the oracle }\mac(.,k)
\end{array}\right] \] 
where the adversary $\cal A$ is any probabilistic Turing machine 
that runs in time at most $t$,
calls $\mac(.,k)$ at most $q_m$ times with messages of length at most $l$, and
calls $\mverify(.,k,.)$ at most $q_v$ times with messages of length at most $l$.
\end{definition}

$\succufcma(t, q_m, q_v, l)$ is the probability that an
adversary forges a MAC, that is, returns a pair $(m,s)$ where $s$ is
a correct MAC for $m$, without having queried the MAC oracle $\mac(.,k)$
on $m$. Intuitively, when the MAC is secure, this probability is small:
the adversary has little chance of forging a MAC. Hence, the MAC
guarantees the integrity of the MACed message because one cannot
compute the MAC without the secret key.

Two frameworks exist for expressing security properties.
In the asymptotic framework, 
used in~\cite{Blanchet07c,Blanchet07}, 
the length of keys is determined by
a security parameter $\secp$, and a MAC is UF-CMA when 
$\succufcma(t, q_m, q_v, l)$ is a negligible function of $\secp$
when $t$ is polynomial in $\secp$. 
($f(\secp)$ is \emph{negligible} when for all
polynomials $q$, there exists $\secp_o \in \mathbb{N}$ such that for
all $\secp > \secp_0$, $f(\secp) \leq \frac{1}{q(\secp)}$.)
The assumption that functions are efficiently computable means that
they are computable in time polynomial in $\secp$ and in the length of their
arguments. 
The goal is to show that
the probability of success of an attack against the protocol is negligible,
assuming the parameters $n$ are polynomial in $\secp$
and the network messages are of length polynomial in $\secp$.
In contrast, in the exact security framework, on which we focus 
in this report, one computes the probability of success of an attack against
the protocol as a function of the probability of breaking the
primitives such as $\succufcma(t, q_m, q_v, l)$, of the runtime of 
functions, of the parameters $n$, and of the length of messages, thus providing
a more precise security result. Intuitively,
the probability $\succufcma(t, q_m, q_v, l)$ is assumed to be small
(otherwise, the computed probability of attack will
be large),
but no formal assumption on this probability is needed 
to establish the security theorem.

\begin{definition}\label{def:encsec}
Let $T_k$, $T_r$, and $T_e$ be types for random coins, keys, and ciphertexts respectively. $T_k$ and $T_r$ are fixed-length types.
A symmetric encryption scheme $\SE$~\cite{Bellare00c} 
consists of two function symbols: 
\begin{itemize}
\item $\enc : \bitstring \times T_k \times T_r \rightarrow T_e$ is the encryption algorithm taking as arguments the cleartext, the key, and random coins, and returning the ciphertext, 

\item $\dec : T_e \times T_k \rightarrow \bitstringbot$ is the decryption algorithm taking as arguments the ciphertext and the key, and returning either the cleartext when decryption succeeds or $\bot$ when decryption fails, 

\end{itemize}
such that $\forall k \in T_k, \forall m \in \bitstring, \forall r \in T_r$, 
$\dec(\enc(m,k,r),k) = m$.

Let $LR(x,y,b) = x$ if $b = 0$ and $LR(x,y,b) = y$ if $b = 1$,
defined only when $x$ and $y$ are bitstrings of the same length.
The advantage of an adversary against indistinguishability under chosen
plaintext attacks (IND-CPA) is
\[
\succindcpa(t, q_e, l) = \max_{\Adv}\, 2 \Pr\left[
\begin{array}{@{}l@{}}
b \randomchoice \{ 0, 1\}; k \randomchoice T_k;\\
b' \leftarrow \Adv^{r \randomchoice T_r; \enc(LR(.,.,b), k, r)}: b' = b
\end{array}\right] - 1
\]
where $\Adv$ is any probabilistic
Turing machine that runs in time at most $t$ and
calls $r \randomchoice T_r;\ab enc(LR(.,.,b), k, r)$ at most $q_e$ times 
on messages of length at most $l$.
\end{definition}
Given two bitstrings $a_0$ and $a_1$ of the same length, the left-right 
encryption oracle $r \randomchoice T_r; \ab enc(LR(.,.,b), k, r)$ returns 
$r \randomchoice T_r; \enc(LR(a_0,a_1,b), k, r)$,
that is, encrypts $a_0$ when $b = 0$ and $a_1$ when $b=1$.
$\succindcpa(t, q_e, l)$ is the probability that the
adversary distinguishes the encryption of the messages $a_0$ given as
first arguments to the left-right encryption oracle from the
encryption of the messages $a_1$ given as second arguments.
Intuitively, when the encryption scheme is IND-CPA secure, this
probability is small: the ciphertext gives almost no information
on what the cleartext is (one cannot determine whether it is $a_0$ or $a_1$
without having the secret key).

\begin{example}\label{exa:running}
Let us consider the following trivial protocol:
\[\begin{split}
A \rightarrow B: e, \mac(e,\vn{mk})\quad 
\text{ where $e = \enc(\vn{k}',\vn{k},\vn{r})$}\\
\text{and $\vn{r},\vn{k}' $ are fresh random numbers}
\end{split}\]
$A$ and $B$ are assumed to share a key $\vn{k}$ for a symmetric encryption scheme
and a key $\vn{mk}$ for a message authentication code.
$A$ creates a fresh key $\vn{k}'$ and sends it encrypted under $\vn{k}$ to $B$.
A MAC is appended to the message, in order to guarantee integrity.
In other words, the protocol sends the key $\vn{k}'$ encrypted 
using an encrypt-then-MAC scheme~\cite{Bellare00c}.
The goal of the protocol is that $\vn{k}'$ should be a secret key shared between
$A$ and $B$.
This protocol can be modeled in our calculus by the following
process $Q_0$:
{\allowdisplaybreaks\begin{align*}
\begin{split}
&Q_0 = \cinput{\startch}{};
\Res{\vn{k}}{T_l};
\Res{\vn{mk}}{T_{mk}};
\coutput{c}{};(Q_A \parpop Q_B)
\end{split}\\
\begin{split}
&Q_A = \repl{i}{n}\cinput{c_A[i]}{};\Res{\vn{k}'}{T_k};\Res{\vn{r}}{T_r};\\*
&\phantom{Q_A =\, }
\assign{\vn{m}:\bitstring}{\enc(\ktob(\vn{k}'),\vn{k},\vn{r})}
\coutput{c_A[i]}{\vn{m}, \mac(\vn{m},\vn{mk})}
\end{split}\\
&Q_B = \repl{i'}{n}\cinput{c_B[i']}{\vn{m}',\vn{ma}};
\IF \mverify(\vn{m}',\vn{mk},\vn{ma})\ \kw{then}\\*
&\phantom{Q_B =\, }
\assign{\injbot(\ktob(\vn{k}''))}{\dec(\vn{m}',\vn{k})}\coutput{c_B[i']}{}
\end{align*}}%
When $Q_0$ receives a message on channel $\startch$, it begins execution: 
it generates the keys $\vn{k}$ and $\vn{mk}$ randomly.
Then it yields control to the adversary,
by outputting on channel $c$. After this output, $n$ copies of 
processes for $A$ and $B$ are ready to be executed, when the adversary
outputs on channels $c_A[i]$ or $c_B[i]$ respectively. 
In a session that runs as expected, the adversary first sends a message
on $c_A[i]$. Then $Q_A$ creates a fresh key $\vn{k}'$ ($T_k$ is
assumed to be a fixed-length type), encrypts it under $\vn{k}$
with random coins $\vn{r}$, computes the MAC under $\vn{mk}$ of the 
ciphertext,
and sends the ciphertext and the MAC on $c_A[i]$. 
The function $\ktob : T_k \rightarrow \bitstring$ is the natural injection
$\ktob(x) = x$; it is needed only for type conversion.
The adversary is
then expected to forward this message on $c_B[i]$. When $Q_B$ receives
this message, it verifies the MAC, decrypts, and stores the obtained key in 
$\vn{k}''$. 
(The function $\injbot: \bitstring \rightarrow \bitstringbot$
is the natural injection; it is useful to check that decryption succeeded.)
This key $\vn{k}''$ should be secret.

The adversary is responsible for forwarding messages from $A$ to $B$.
It can send messages in unexpected ways in order to mount an attack.

This very small example is sufficient to illustrate the main features
of CryptoVerif.
\end{example}

\subsection{Type System}\label{sec:typesystem}

We use a type system to check that bitstrings of the proper
type are passed to each function and that array indices are used correctly.

To be able to type variable accesses used not under their definition
(such accesses are guarded by a $\FIND $ construct), the type-checking
algorithm proceeds in two passes. In the first pass, it builds a type
environment $\tyenv$, which maps variable names $x$ to types $[1, n_1]
\times \ldots \times [1, n_m] \rightarrow T$, where
the definition of $x[i_1, \ldots, i_m]$ of type $T$ occurs under
replications or $\kw{find}$ that bind $i_1, \ldots, i_m$ with
declaration $i_j \leq n_j$.
(For instance, the definition of $x[i_1, \ldots, i_m]$ occurs under
$\repl{i_1}{n_1}$, \ldots, $\repl{i_m}{n_m}$ or it occurs in the
condition of $\kw{find}\ \vf_1 = i_1 \leq n_1, \ldots, \vf_m = i_m \leq
n_m$ under no replication. The type $T$ is the one given in the
definition of $x$ in $\Res{x[\tup{i}]}{T}$ or in a pattern $x[\tup{i}]:T$
in an assignment, an input, or a $\kw{get}$. In the $\kw{find}$ construct,
$\kw{find} \ldots x[\tup{i}] = i \leq n$, the type $T$ of $x$ is $T = [1,n]$.)
The tool checks that all definitions of the same variable $x$ yield
the same value of $\tyenv(x)$, so that $\tyenv$ is properly defined.

\begin{figure}[tp]
Typing rules for terms:
\begin{gather}
\frac{\tyenv(i) = T}{
\tyenv \vdash i : T}\tag{TIndex}\\[1mm]
\frac{\tyenv(x) = T_1 \times \ldots \times T_m \rightarrow T \qquad 
\forall j \leq m, \tyenv \vdash M_j : T_j}{
\tyenv \vdash x[M_1, \ldots, M_m] : T}\tag{TVar}\\[1mm]
\frac{f : T_1 \times \ldots \times T_m \rightarrow T \qquad 
\forall j \leq m, \tyenv \vdash M_j : T_j}{
\tyenv \vdash f(M_1, \ldots, M_m) : T}\tag{TFun}\\[1mm]
\frac{T \text{ \emph{fixed}, \emph{bounded}, or \emph{nonuniform}} \qquad 
\tyenv \vdash x[\tup{i}] : T \qquad
\tyenv \vdash N: T'}{
\tyenv \vdash \Res{x[\tup{i}]}{T}; N : T'}
\tag{TNewT}\\[1mm]
\frac{\tyenv \vdash M : T \qquad
\tyenv \vdash \pat : T \qquad
\tyenv \vdash N : T'\qquad
\tyenv \vdash N' : T'}{
\tyenv \vdash \assign{\pat}{M}{N}\ \ELSE N': T'}\tag{TLetT}\\[1mm]
\frac{\tyenv \vdash M : T \qquad
\tyenv \vdash x[\tup{i}] : T \qquad
\tyenv \vdash N : T'}{
\tyenv \vdash \assign{x[\tup{i}] : T}{M}{N}: T'}\tag{TLetT2}\\[1mm]
\frac{\tyenv \vdash M : \bool \qquad
\tyenv \vdash N: T \qquad
\tyenv \vdash N' : T}{
\tyenv \vdash \bguard{M}{N}{N'}:T}\tag{TIfT}\\[1mm]
\frac{\begin{array}{c}
\forall j \leq m, \forall k \leq m_j,
\tyenv \vdash \vf_{jk}[\tup{i}] : [1, n_{jk}]\\
\forall j \leq m, \forall k \leq l_j, 
\tyenv[i_{j1} \mapsto [1, n_{j1}], \ldots, i_{jm_j} \mapsto [1, n_{jm_j}]\,] \vdash M_{jk} : T_{jk}\\
\forall j \leq m, \tyenv[i_{j1} \mapsto [1, n_{j1}], \ldots, i_{jm_j} \mapsto [1, n_{jm_j}]\,]  \vdash M_j : \bool\\
\forall j \leq m, \tyenv \vdash N_j : T
\qquad
\tyenv \vdash N : T
\end{array}}{
\begin{array}{@{}c@{}}
\tyenv \vdash 
%\cfind{j=1}{m}{
%\vf_{j1}[\tup{i}] \leq n_{j1}, \ldots, \vf_{jm_j}[\tup{i}] \leq n_{jm_j}}{M_{j1}, \ldots, M_{jl_j}}{
%M_j}{P_j}{P}
\FIND\uniqueopt \ (\mathop{\textstyle\bigoplus}\nolimits_{j=1}^{m}
\vf_{j1}[\tup{i}] = i_{j1} \leq n_{j1}, \ldots, \vf_{jm_j}[\tup{i}] = i_{jm_j} \leq n_{jm_j}\ 
\SUCHTHAT\\
\defined(M_{j1}, \ldots, M_{jl_j}) \fand M_j \THEN N_j)\ \ELSE N : T
\end{array}}\tag{TFindT}\\[1mm]
\frac{\tbl : T_1 \times \ldots \times T_l\qquad
\forall j \leq l, \tyenv \vdash M_j : T_j\qquad
\tyenv \vdash N:T}{
\tyenv \vdash \INSERT\ \tbl(M_1,\ldots,M_l);N : T}\tag{TInsertT}\\[1mm]
\frac{\tbl : T_1 \times \ldots \times T_l\qquad
\forall j \leq l, \tyenv \vdash \pat_j : T_j\qquad
\tyenv \vdash M : \bool\qquad
\tyenv \vdash N: T\qquad 
\tyenv \vdash N': T}{
\GET\uniqueopt\ \tbl(\pat_1,\ldots,\pat_l)\ \SUCHTHAT\ M\ \IN\ N\ \ELSE N': T}\tag{TGetT}\\[1mm]
\frac{e : T_1 \times \ldots \times T_l \qquad
\forall j \leq l, \tyenv \vdash M_j : T_j\qquad
\tyenv \vdash N: T}{
\tyenv \vdash \kevent{e(M_1, \ldots, M_l)}; N : T}\tag{TEventT}\\[1mm]
\frac{e:()}{\tyenv \vdash \keventabort{e}: T'}\tag{TEventAbortT}
\end{gather}
Typing rules for patterns:
\begin{gather}
\frac{\tyenv \vdash x[\tup{i}] :T}{
\tyenv \vdash (x[\tup{i}] :T) : T}\tag{TVarP}\\[1mm]
\frac{f : T_1 \times \ldots \times T_m \rightarrow T
\qquad \forall j \leq m, \tyenv \vdash \pat_j: T_j}{
\tyenv \vdash f(\pat_1, \ldots, \pat_m) : T}\tag{TFunP}\\[1mm]
\frac{\tyenv \vdash M : T}{
\tyenv \vdash {{=}M} : T}\tag{TEqP}
\end{gather}
\caption{Typing rules (1)}\label{fig:type1}
\end{figure}%

\begin{figure}[tp]
Typing rules for input processes:
\begin{gather}
\tyenv \vdash 0\tag{TNil}\\[1mm]
\frac{\tyenv \vdash Q \qquad \tyenv \vdash Q'}{\tyenv \vdash Q \parpop Q'}
\tag{TPar}\\[1mm]
\frac{\tyenv[i \mapsto [1,n]] \vdash Q}{\tyenv \vdash \repl{i}{n}{Q}}
\tag{TRepl}\label{typ:repl}\\[1mm]
\frac{\tyenv \vdash Q}{\tyenv \vdash \Reschan{c};Q}\tag{TNewChannel}\\[1mm]
\frac{\forall j \leq l, \tyenv \vdash M_j : T'_j\qquad
\tyenv \vdash \pat: T \qquad
\tyenv \vdash P}{
\tyenv \vdash \cinput{c[M_1, \ldots, M_l]}{\pat}; P}
\tag{TIn}\label{typ:in}
\end{gather}
Typing rules for output processes:
\begin{gather}
\frac{\forall j \leq l, \tyenv \vdash M_j : T'_j\qquad
\tyenv \vdash N : T\qquad
\tyenv \vdash Q}{
\tyenv \vdash \coutput{c[M_1, \ldots, M_l]}{N}; Q}
\tag{TOut}\label{typ:out}\\[1mm]
\frac{T \text{ \emph{fixed}, \emph{bounded}, or \emph{nonuniform}} \qquad 
\tyenv \vdash x[\tup{i}] : T \qquad
\tyenv \vdash P}{
\tyenv \vdash \Res{x[\tup{i}]}{T}; P}
\tag{TNew}\\[1mm]
\frac{\tyenv \vdash M : T \qquad
\tyenv \vdash \pat : T \qquad
\tyenv \vdash P\qquad
\tyenv \vdash P'}{
\tyenv \vdash \assign{\pat}{M}{P}\ \ELSE P'}\tag{TLet}\\[1mm]
\frac{\tyenv \vdash M : \bool \qquad
\tyenv \vdash P \qquad
\tyenv \vdash P'}{
\tyenv \vdash \bguard{M}{P}{P'}}\tag{TIf}\\[1mm]
\frac{\begin{array}{c}
\forall j \leq m, \forall k \leq m_j,
\tyenv \vdash \vf_{jk}[\tup{i}] : [1, n_{jk}]\\
\forall j \leq m, \forall k \leq l_j, 
\tyenv[i_{j1} \mapsto [1, n_{j1}], \ldots, i_{jm_j} \mapsto [1, n_{jm_j}]\,] \vdash M_{jk} : T_{jk}\\
\forall j \leq m, \tyenv[i_{j1} \mapsto [1, n_{j1}], \ldots, i_{jm_j} \mapsto [1, n_{jm_j}]\,] \vdash M_j : \bool\\
\forall j \leq m, \tyenv \vdash P_j
\qquad
\tyenv \vdash P 
\end{array}}{
\begin{array}{@{}c@{}}
\tyenv \vdash 
%\cfind{j=1}{m}{
%\vf_{j1}[\tup{i}] \leq n_{j1}, \ldots, \vf_{jm_j}[\tup{i}] \leq n_{jm_j}}{M_{j1}, \ldots, M_{jl_j}}{
%M_j}{P_j}{P}
\FIND\uniqueopt \ (\mathop{\textstyle\bigoplus}\nolimits_{j=1}^{m}
\vf_{j1}[\tup{i}] = i_{j1} \leq n_{j1}, \ldots, \vf_{jm_j}[\tup{i}] = i_{jm_j} \leq n_{jm_j}\ 
\SUCHTHAT\\
\defined(M_{j1}, \ldots, M_{jl_j}) \fand M_j \THEN P_j)\ \ELSE P
\end{array}}\tag{TFind}\\[1mm]
\frac{\tbl : T_1 \times \ldots \times T_l\qquad
\forall j \leq l, \tyenv \vdash M_j : T_j\qquad
\tyenv \vdash P}{
\tyenv \vdash \INSERT\ \tbl(M_1,\ldots,M_l);P}\tag{TInsert}\\[1mm]
\frac{\tbl : T_1 \times \ldots \times T_l\qquad
\forall j \leq l, \tyenv \vdash \pat_j : T_j\qquad
\tyenv \vdash M : \bool\qquad
\tyenv \vdash P\qquad 
\tyenv \vdash P'}{
\GET\uniqueopt\ \tbl(\pat_1,\ldots,\pat_l)\ \SUCHTHAT\ M\ \IN\ P\ \ELSE P'}\tag{TGet}\\[1mm]
\frac{e : T_1 \times \ldots \times T_l \qquad
\forall j \leq l, \tyenv \vdash M_j : T_j\qquad
\tyenv \vdash P}{
\tyenv \vdash \kevent{e(M_1, \ldots, M_l)}; P}\tag{TEvent}\\[1mm]
\frac{e:()}{\tyenv \vdash \keventabort{e}} \tag{TEventAbort}\\[1mm]
\tyenv \vdash \kw{yield} \tag{TYield}
\end{gather}
\caption{Typing rules (2)}\label{fig:type2}
\end{figure}

In the second pass, the process is typechecked in the type environment $\tyenv$
using the rules of Figures~\ref{fig:type1} and~\ref{fig:type2}. These figures defines four judgments:
\begin{itemize}

\item $\tyenv \vdash M : T$ means that the term $M$ has type $T$ in environment
$\tyenv$. 

\item $\tyenv \vdash \pat : T$ means that the pattern $\pat$ has type $T$
in environment $\tyenv$.

\item $\tyenv \vdash P$ and $\tyenv \vdash Q$ mean that the output process $P$
and the input process $Q$ are well-typed in environment $\tyenv$,
respectively.

\end{itemize}

In $x[M_1, \ldots, M_m]$, $M_1, \ldots, M_m$ must be of the suitable 
interval type.
When $f(M_1, \ldots, M_m)$ is called and $f: T_1 \times \ldots \times
T_m \rightarrow T$, $M_j$ must be of type $T_j$, and $f(M_1, \ldots, M_m)$
is then of type $T$.
%IMPLEMENTATION: $T_1, \ldots, T_m$ must be bitstring types.
%DONE say it in the section on subsets of the language

The term $\Res{x[\tup{i}]}{T}; N$ is accepted only when 
$T$ is declared \emph{fixed}, \emph{bounded}, or \emph{nonuniform}.
We check that $x[\tup{i}]$ is of type $T$ (which is in fact always
true when the construction of $\tyenv$ succeeds). $N$ must well-typed,
and its type is also the type of $\Res{x[\tup{i}]}{T}; N$.

In $\assign{\pat}{M}{N}\ \ELSE N'$, $\pat$ must have the same type
as $M$, and $N$ and $N'$ must have the same type, which is also
the type of $\assign{\pat}{M}{N}\ \ELSE N'$. The typing rules
for patterns $\pat$ are found at the bottom of Figure~\ref{fig:type1}.
The pattern $x[\tup{i}] :T$ has type $T$, provided $x[\tup{i}]$ has
type $T$ (which is in fact always true when the construction of
$\tyenv$ succeeds). The other typing rules for patterns are
straightforward.
The particular case $\assign{x[\tup{i}] : T}{M}{N}$ is typed
similarly, except that the $\kw{else}$ branch is omitted.

In $\bguard{M}{N}{N'}$, $M$ must be of type $\bool$ and
$N$ and $N'$ must have the same type, which is also
the type of $\bguard{M}{N}{N'}$.

In
\[\begin{split}
&\FIND\uniqueopt \ (\mathop{\textstyle\bigoplus}\nolimits_{j=1}^{m}
\vf_{j1}[\tup{i}] = i_{j1} \leq n_{j1}, \ldots, \vf_{jm_j}[\tup{i}] = i_{jm_j} \leq n_{jm_j}\ 
\SUCHTHAT\\
&\quad \defined(M_{j1}, \ldots, M_{jl_j}) \fand M_j \THEN N_j)\ \ELSE N
\end{split}\]
%\cfind{j=1}{m}{\vf_{j1}[\tup{i}] \leq n_{j1}, \ldots, 
%\vf_{jm_j}[\tup{i}] \leq n_{jm_j}}{
%M_{j1}, \ldots, M_{jl_j}}{M_j}{P_j}{P}\]
the replication indices $i_{j1}, \ldots, i_{jm_j}$ are bound in
$M_{j1}, \ldots, M_{jl_j}, M_j$, of types $[1,n_{j1}], \ldots, [1,n_{jm_j}]$
respectively;
$M_j$ is of type $\bool$ for all $j \leq m$;
$N_j$ for all $j \leq m$ and $N$ all have the same type, which is also
the type of the $\kw{find}$ term.

In $\INSERT\ \tbl(M_1,\ldots,M_l);N$, $M_1, \ldots, M_l$ must
be of the type declared for the elements of the table $\tbl$, 
and the type of $N$ is the type of the $\INSERT$ term.

In $\GET\uniqueopt\ \tbl(\pat_1,\ldots,\pat_l)\ \SUCHTHAT\ M\ \IN\ N\ \ELSE N'$,
$\pat_1, \ldots, \pat_l$ must
be of the type declared for the elements of the table $\tbl$
and $M$ must be of type $\bool$.
The terms $N$ and $N'$ must have the same type, which is also the type of
the $\GET$ term.

In $\kevent{e(M_1, \ldots, M_l)}; N$, $M_1, \ldots, M_l$ must
be of the type declared for the arguments of event $e$,
and the type of $N$ is the type of the $\kw{event}$ term.

The term $\keventabort{e}$ can have any type (because it aborts the game);
the event $e$ must be declared without argument, which we denote
by $e : ()$.

The type system for processes requires each subterm to be well-typed. 
In $\repl{i}{n} Q$, $i$ is of type $[1, n]$ in $Q$. 
The processes $\kw{new}$, $\kw{let}$, $\kw{if}$, $\kw{find}$, $\INSERT$, 
$\GET$, $\kw{event}$, and $\kw{event\string_abort}$ are typed
similarly to the corresponding terms.

%
%
%IMPLEMENTATION In the implementation, in the output T must be a bitstring 
%type; M must be of a bitstring type.
%DONE say it in the section on subsets of the language
%

We say that an occurrence of a term $M$ in a process $Q$ is of type
$T$ when $\tyenv \vdash M : T$ where $\tyenv$ is the type environment
of $Q$ extended with $i \mapsto [1, n]$ for each replication
$\repl{i}{n}$ above $M$ in $Q$ and with 
$i_{j1} \mapsto [1,n_{j1}], \ldots, i_{jm_j} \mapsto [1, n_{jm_j}]$ for each
$\FIND\uniqueopt$ $(\mathop{\textstyle\bigoplus}\nolimits_{j=1}^{m}
\vf_{j1}[\tup{i}] = i_{j1} \leq n_{j1}, \ldots, \vf_{jm_j}[\tup{i}] = i_{jm_j} \leq n_{jm_j}$
$\SUCHTHAT$ $\defined(M_{j1}, \ldots, M_{jl_j}) \fand M_j$ $\kw{then}$ $P_j)$ $\ELSE P$ such that the considered occurrence of $M$ is in the condition $\defined(M_{j1}, \ldots, M_{jl_j}) \fand M_j$.
\bbnote{defined for an occurrence but sometimes used for a term}%

\begin{invariant}[Typing]\label{inv3}
The process $Q_0$ satisfies Invariant~\ref{inv3} if and only if
the type environment $\tyenv$ for $Q_0$ is well-defined,
and $\tyenv \vdash Q_0$.
\end{invariant}

We require the adversary to be well-typed. This requirement does not
restrict its computing power, because it can always define type-cast
functions $f : T \rightarrow T'$ to bypass the type system.  
Similarly, the type system does not restrict the class of protocols
that we consider, since the protocol may contain type-cast functions.
The type system just makes explicit which set of values may appear
at each point of the protocol.

\subsection{Formal Semantics}

\subsubsection{Definition of the Semantics}

The formal semantics of our calculus 
is presented in Figures~\ref{fig:sem1}, \ref{fig:sem1bis},
\ref{fig:sem2}, \ref{fig:sem3}, and~\ref{fig:sem3bis}.

In this semantics, each term $M$ or process $P$ or $Q$ is labeled by
a program point $\pp$, replacing $M$ with $\pptag M$ and similarly for
$P$ and $Q$. We still use the notations $M$, $P$, $Q$ for terms and
processes tagged with program points.
The program points are used in order to track from where
each term or process comes from in the initial process. These program
points are simply constant tags, and the initial process 
is tagged with a distinct program point at each subterm
and subprocess.

A semantic configuration is a sextuple $E, (\sigma,P)\restconfig$,
where 
\begin{itemize}
\item $E$ is an environment mapping array cells to values.
\item $(\sigma,P)$ is the output process $P$ currently scheduled, 
with the associated mapping sequence $\sigma$ which gives values 
of replication indices. 

The mapping sequence $\sigma = [i_1 \mapsto a_1, \dots, i_m \mapsto a_m]$ is a sequence of mappings $i_j \mapsto a_j$, which can also be interpreted as a function: $\sigma(i_j) = a_j$ for all $j \leq m$, and $\sigma(\tup{i})$ is defined by the natural extension to sequences. However, using a sequence allows us to define $\dom(\sigma) = [i_1, \dots, i_m]$ to be the sequence of current replication indices and $\image(\sigma) = [a_1, \dots, a_m]$ to be the sequence of their values. When $\sigma = [i_1 \mapsto a_1, \dots, i_m \mapsto a_m]$, 
$\sigma[i_{m+1} \mapsto a_{m+1}, \dots, i_l \mapsto a_l] = [i_1 \mapsto a_1, \dots, i_l \mapsto a_l]$.

\item $\pset$ is the multiset of input processes running in 
parallel with $P$, with their associated mapping sequences giving values 
of replication indices.
\item $\cset$ is the set of channels already created.
\item $\tblcts$ defines the contents of tables. It is 
a list of $\tbl(a_1, \ldots, a_m)$ indicating that table $\tbl$
contains the element $(a_1, \ldots, a_m)$.
\item $\evseq$ is a sequence representing the events executed so far.
Each element of the sequence is of the form $(\pp, \tup{a}):e(a_1, \ldots, a_m)$,
meaning that the event $e(a_1, \ldots, a_m)$ has been executed at program point
$\pp$ with replication indices evaluating to $\tup{a}$.

We define $\evseqnopp = \removepp(\evseq) = [e(a_1, \ldots, a_m) \mid (\pp, \tup{a}):e(a_1, \ldots, a_m) \in \evseq]$ to be the sequence of events $\evseq$ without their associated program points and replication indices.

\end{itemize}
In addition to the grammar given in Figure~\ref{fig:syntax},
the terms $M$ of the semantics can
be values $a$ and abort event values $\keventabort{(\pp, \tup{a}):e}$,
and the processes $P$ can be
$\kw{abort}$, corresponding to the situation in which
the game has been aborted.
These additional terms and processes are not tagged with program points.
(They do not occur in the initial process.)

The semantics is defined by reduction rules of the form
$E,(\sigma,P)\restconfig \red{p}{\ix} E',\ab (\sigma',P'),\ab \pset', \ab \cset', \ab \tblcts', \ab \evseq'$
meaning that $E,(\sigma,P)\restconfig$ reduces to $E',\ab (\sigma',P'),\ab \pset', \ab \cset', \ab \tblcts', \ab \evseq'$
with probability $p$. 
The index $\ix$ just serves in distinguishing
reductions that yield the same configuration with the same
probability in different ways, so that the probability
of a certain reduction can be computed correctly:
\[\Pr[E, (\sigma,P), \pset, \cset, \tblcts, \evseq \rightarrow E', (\sigma',P'), \pset', \cset', \tblcts', \evseq'] = \hspace*{-12mm}\sum_{E,(\sigma,P),\pset, \cset, \tblcts, \evseq\red{p}{\ix} E', (\sigma',P'), \pset', \cset', \tblcts', \evseq'} \hspace*{-6mm} p\]
The probability of a trace 
$\trace = E_1, (\sigma_1, P_1), \pset_1, \cset_1, \tblcts_1, \evseq_1 
\red{p_1}{\ix_1} \ldots \red{p_{m-1}}{\ix_{m-1}} E_m, (\sigma_m, P_m), \ab \pset_m, \ab \cset_m, \ab \tblcts_m, \ab \evseq_m$ is $\Pr[\trace] = p_1 \times \ldots \times p_{m-1}$.
We define the semantics only for patterns $x[\tup{i}]:T$, the other
patterns can be encoded as outlined in Section~\ref{sec:syntax}.
\bb{This is not ideal, because it basically means that we are not going to
prove transformations that involve other patterns correct. 
On the other hand, it simplifies the semantics, which is already complicated 
enough.}%

\begin{figure}[tp]
\begin{gather}
E,\sigma, \pptag i, \tblcts, \evseq \red{1}{} E, \sigma, \sigma(i), \tblcts, \evseq\tag{ReplIndex}\label{sem:replindex}\\[1mm]
\frac{x[a_1, \ldots, a_m] \in \dom(E)}{
E, \sigma, \pptag x[a_1, \ldots, a_m], \tblcts, \evseq \red{1}{} E, \sigma, E(x[a_1, \ldots, a_m]), \tblcts, \evseq}\tag{Var}\label{sem:var}\\[1mm]
\frac{
f: T_1 \times \ldots \times T_m \rightarrow T\qquad
\forall j \leq m, a_j \in T_j\qquad f(a_1, \ldots, a_m) = a}{
E, \sigma, \pptag f(a_1, \ldots, a_m), \tblcts, \evseq \red{1}{} E, \sigma, a, \tblcts, \evseq}\tag{Fun}\label{sem:fun}\\[1mm]
\frac{a \in T \qquad E' = E[x[\sigma(\tup{i})] \mapsto a]}{
E, \sigma, \pptag \Res{x[\tup{i}]}{T}; N, \tblcts, \evseq \red{D_T(a)}{N(a)} E', \sigma, N, \tblcts, \evseq}
\tag{NewT}\label{sem:newt}\\[1mm]
\frac{a \in T \qquad E' = E[x[\sigma(\tup{i})] \mapsto a]}{
E, \sigma, \pptag \assign{x[\tup{i}] : T}{a}{N}, \tblcts, \evseq \red{1}{} 
E', \sigma, N, \tblcts, \evseq}\tag{LetT}\label{sem:lett}\\[1mm]
E, \sigma, \pptag \bguard{\true}{N}{N'}, \tblcts, \evseq \red{1}{}
E, \sigma, N, \tblcts, \evseq  \tag{IfT1}\label{sem:ift1}\\[1mm]
\frac{a \neq \true}{
E, \sigma, \pptag \bguard{a}{N}{N'}, \tblcts, \evseq \red{1}{}
E, \sigma, N', \tblcts, \evseq} \tag{IfT2}\label{sem:ift2}\\[1mm]
\frac{
\begin{array}{c}
(v_k)_{1\leq k \leq l} \text{ is the sequence of }(j, a_1, \ldots, a_{m_j})
\text{ for }a_1 \in [1, n_{j1}], \ldots, a_{m_j} \in [1, n_{jm_j}]\\
\text{ ordered in increasing lexicographic order}\\
\forall k \in [1, l], E, \sigma[i_{j1} \mapsto a_1, \ldots, i_{jm_j} \mapsto a_{m_j}], D_j \wedge M_j, \tblcts, \evseq \red{p_k}{\ix_k}^* E_k, \sigma_k, r_k, \tblcts, \evseq\\
\text{ where }v_k = (j, a_1, \ldots, a_{m_j})\text{ and }r_k \text{ is a value or }\keventabort{(\pp',\tup{a}):e}\\
S = \{ k \mid \exists (\pp',\tup{a}, e), r_k = \keventabort{(\pp',\tup{a}):e} \}\qquad k_0 \in S
\end{array}
}{\begin{array}{c}
E, \sigma, \pptag \FIND\uniqueopt\ (\mathop\bigoplus\nolimits_{j=1}^m
\vf_{j1}[\tup{i}] = i_{j1} \leq n_{j1}, \ldots, \vf_{jm_j}[\tup{i}] = i_{jm_j} \leq n_{jm_j}\ \kw{suchthat}\\ 
D_j \fand M_j\ \kw{then}\ N_j)\ \ELSE N, \tblcts, \evseq \red{p_1\ldots p_{l}D_{\FIND}(S)(k_0)}{\ix_1\ldots \ix_l FE(k_0)} E_{k_0}, \sigma_{k_0}, r_{k_0}, \tblcts, \evseq
\end{array}}\tag{FindTE}\label{sem:findte}\\[1mm]
\frac{
\begin{array}{c}
(v_k)_{1\leq k \leq l} \text{ is the sequence of }(j, a_1, \ldots, a_{m_j})
\text{ for }a_1 \in [1, n_{j1}], \ldots, a_{m_j} \in [1, n_{jm_j}]\\
\text{ ordered in increasing lexicographic order}\\
\forall k \in [1, l], E, \sigma[i_{j1} \mapsto a_1, \ldots, i_{jm_j} \mapsto a_{m_j}], D_j \wedge M_j, \tblcts, \evseq \red{p_k}{\ix_k}^* E'', \sigma', r_k, \tblcts, \evseq\\
\text{ where }v_k = (j, a_1, \ldots, a_{m_j})\text{ and }r_k \text{ is a value}\\
S = \{ v_k \mid r_k = \true \}\qquad 
|S| = 1 \text{ or } \uniqueopt \text{ is empty}\\
v_0 = (j', a'_1, \ldots, a'_{m_{j'}})\in S\qquad
E' = E[\vf_{j'1}[\sigma(\tup{i})] \mapsto a'_1, \ldots, \vf_{j'm_{j'}}[\sigma(\tup{i})] \mapsto a'_{m_{j'}}]
\end{array}
}{\begin{array}{c}
E, \sigma, \pptag \FIND\uniqueopt\ (\mathop\bigoplus\nolimits_{j=1}^m
\vf_{j1}[\tup{i}] = i_{j1} \leq n_{j1}, \ldots, \vf_{jm_j}[\tup{i}] = i_{jm_j} \leq n_{jm_j}\ \kw{suchthat}\\ 
D_j \fand M_j\ \kw{then}\ N_j)\ \ELSE N, \tblcts, \evseq \red{p_1\ldots p_l D_{\FIND}(S)(v_0)}{\ix_1\ldots \ix_l F1(v_0)} E', \sigma, N_{j'}, \tblcts, \evseq
\end{array}}\tag{FindT1}\label{sem:findt1}\\[1mm]
\frac{
\begin{array}{c}
\text{First four lines as in \eqref{sem:findt1}}\qquad
S = \{ v_k \mid r_k = \true \} = \emptyset
\end{array}
}{\begin{array}{c}
E, \sigma, \pptag \FIND\uniqueopt\ (\mathop\bigoplus\nolimits_{j=1}^m
\vf_{j1}[\tup{i}] = i_{j1} \leq n_{j1}, \ldots, \vf_{jm_j}[\tup{i}] = i_{jm_j} \leq n_{jm_j}\ \kw{suchthat}\\ 
D_j \fand M_j\ \kw{then}\ N_j)\ \ELSE N, \tblcts, \evseq \red{p_1\ldots p_l}{\ix_1\ldots \ix_l F2} E, \sigma, N, \tblcts, \evseq
\end{array}}\tag{FindT2}\label{sem:findt2}\\[1mm]
\frac{
\begin{array}{c}
\text{First four lines as in \eqref{sem:findt1}}\qquad
S = \{ v_k \mid r_k = \true \}\qquad |S| > 1
\end{array}
}{\begin{array}{c}
E, \sigma, \pptag \FIND\unique{e}\ (\mathop\bigoplus\nolimits_{j=1}^m
\vf_{j1}[\tup{i}] = i_{j1} \leq n_{j1}, \ldots, \vf_{jm_j}[\tup{i}] = i_{jm_j} \leq n_{jm_j}\ \kw{suchthat}\\ 
D_j \fand M_j\ \kw{then}\ N_j)\ \ELSE N, \tblcts, \evseq \red{p_1\ldots p_l}{\ix_1\ldots \ix_l F3} E, \sigma, \keventabort{(\pp, \image(\sigma)):e}, \tblcts, \evseq
\end{array}}\tag{FindT3}\label{sem:findt3}
\end{gather}%
\caption{Semantics (1): terms, first part}\label{fig:sem1}
\end{figure}

\begin{figure}
  \begin{gather}
    E, \sigma, \pptag \INSERT\ \tbl(a_1, \ldots, a_l); N, \tblcts, \evseq
    \red{1}{}
    E, \sigma, N, (\tblcts, \tbl(a_1, \ldots, a_l)), \evseq
\tag{InsertT}\label{sem:insertt}\\[1mm]
\frac{
\begin{array}{@{}c@{}}
[v_1, \dots, v_m] = [x \in \tblcts \mid \exists a_1, \dots, \exists a_l, x = \tbl(a_1, \ldots, a_l)]\\
\forall k \in [1, m], E[x_1[\sigma(\tup{i})] \mapsto a_1, \ldots, x_l[\sigma(\tup{i})] \mapsto a_l], \sigma, M, \tblcts, \evseq \red{p_k}{\ix_k}^* E_k, \sigma_k, r_k, \tblcts, \evseq\\
\text{ where }v_k = \tbl(a_1, \ldots, a_l)\text{ and }r_k \text{ is a value or }\keventabort{(\pp',\tup{a}):e}\\
S = \{ k \in [1,m] \mid \exists (\pp',\tup{a},e), r_k = \keventabort{(\pp',\tup{a}):e} \}\qquad k_0 \in S
\end{array}}{
\begin{array}{@{}c@{}}
E, \sigma, \pptag \GET\uniqueopt\ \tbl(x_1[\tup{i}]:T_1,\ldots,x_l[\tup{i}]:T_l)\ \SUCHTHAT\ M\ \IN\ N\ \ELSE N', \tblcts, \evseq\\
\red{p_1\dots p_m D_{\kw{get}}(S)(k_0)}{\ix_1\dots \ix_m GE(k_0)} E_{k_0}, \sigma_{k_0}, r_{k_0}, \tblcts, \evseq
\end{array}}\tag{GetTE}\label{sem:gette}\\[1mm]
\frac{
\begin{array}{@{}c@{}}
[v_1, \dots, v_m] = [x \in \tblcts \mid \exists a_1, \dots, \exists a_l, x = \tbl(a_1, \ldots, a_l)]\\
\forall k \in [1, m], E[x_1[\sigma(\tup{i})] \mapsto a_1, \ldots, x_l[\sigma(\tup{i})] \mapsto a_l], \sigma, M, \tblcts, \evseq \red{p_k}{\ix_k}^* E'', \sigma, r_k, \tblcts, \evseq\\
\text{ where }v_k = \tbl(a_1, \ldots, a_l)\text{ and }r_k \text{ is a value}\\
S = \{ k \in [1,m] \mid r_k = \true \}\\
|S| = 1 \text{ or } \uniqueopt \text{ is empty}\\
k_0 \in S \qquad \tbl(a_1, \ldots, a_l) = v_{k_0} \qquad
E' = E[x_1[\sigma(\tup{i})] \mapsto a_1, \ldots, x_l[\sigma(\tup{i})] \mapsto a_l]
\end{array}}{
\begin{array}{@{}c@{}}
E, \sigma, \pptag \GET\uniqueopt\ \tbl(x_1[\tup{i}]:T_1,\ldots,x_l[\tup{i}]:T_l)\ \SUCHTHAT\ M\ \IN\ N\ \ELSE N', \tblcts, \evseq\\
\red{p_1\dots p_m D_{\kw{get}}(S)(k_0)}{\ix_1\dots \ix_m G1(k_0)} E', \sigma, N, \tblcts, \evseq
\end{array}}\tag{GetT1}\label{sem:gett1}\\[1mm]
\frac{\text{First four lines as in \eqref{sem:gett1}}\qquad S = \emptyset}{
\begin{array}{c}
E, \sigma, \pptag \GET\uniqueopt\ \tbl(x_1[\tup{i}]:T_1,\ldots,x_l[\tup{i}]:T_l)\ \SUCHTHAT\ M\ \IN\ N\ \ELSE N', \tblcts, \evseq\\
\red{p_1\dots p_m}{\ix_1\dots \ix_m G2} E, \sigma, N', \tblcts, \evseq
\end{array}}\tag{GetT2}\label{sem:gett2}\\[1mm]
\frac{\text{First four lines as in \eqref{sem:gett1}}\qquad |S| > 1}{
\begin{array}{c}
E, \sigma, \pptag \GET\unique{e}\ \tbl(x_1[\tup{i}]:T_1,\ldots,x_l[\tup{i}]:T_l)\ \SUCHTHAT\ M\ \IN\ N\ \ELSE N', \tblcts, \evseq\\
\red{p_1\dots p_m}{\ix_1\dots \ix_m G3} E, \sigma, \keventabort{(\pp, \image(\sigma)):e}, \tblcts, \evseq
\end{array}}\tag{GetT3}\label{sem:gett3}\\[1mm]
E, \sigma, \pptag \kevent{e(a_1, \dots, a_l)}; N, \tblcts, \evseq
\red{1}{}
  E, \sigma, N, \tblcts, (\evseq, (\pp,\image(\sigma)):e(a_1, \dots, a_l))
\tag{EventT}\label{sem:eventt}\\[1mm]
E, \sigma, \pptag \keventabort{e}, \tblcts, \evseq
\red{1}{}
  E, \sigma, \keventabort{(\pp,\image(\sigma)):e}, \tblcts, \evseq
\tag{EventAbortT}\label{sem:eventabortt}\\[1mm]
\frac{E, \sigma, N, \tblcts, \evseq \red{p}{\ix} E', \sigma', N', \tblcts', \evseq'}{
E, \sigma, C[N], \tblcts, \evseq \red{p}{\ix} E', \sigma', C[N'], \tblcts', \evseq'}\tag{CtxT}\label{sem:ctxt}\\[1mm]
E, \sigma, C[\keventabort{(\pp,\tup{a}):e}], \tblcts, \evseq \red{1}{} E, \sigma, \keventabort{(\pp,\tup{a}):e}, \tblcts, \evseq\tag{CtxEventT}\label{sem:ctxeventt}\\[1mm]
%
%Defined conditions
\frac{\neg \forall j \leq l, \exists a_j, E, \sigma, M_j, \tblcts, \evseq \red{1}{}^* E, \sigma, a_j, \tblcts, \evseq}{
E, \sigma, \defined(M_1, \ldots, M_l) \wedge M, \tblcts, \evseq \red{1}{} E, \sigma, \false, \tblcts, \evseq}\tag{DefinedNo}\label{sem:definedno}\\[1mm]
\frac{\forall j \leq l, \exists a_j, E, \sigma, M_j, \tblcts, \evseq \red{1}{}^* E, \sigma, a_j, \tblcts, \evseq}{
E, \sigma, \defined(M_1, \ldots, M_l) \wedge M, \tblcts, \evseq \red{1}{} E, \sigma, M, \tblcts, \evseq}\tag{DefinedYes}\label{sem:definedyes}
\end{gather}%
\caption{Semantics (2): terms, second part, and $\defined$ conditions}\label{fig:sem1bis}
\end{figure}

In Figures~\ref{fig:sem1} and~\ref{fig:sem1bis},
we define an auxiliary relation for evaluating terms:
$E, \sigma, M, \tblcts, \evseq \red{p}{\ix} E', \sigma, M', \tblcts', \evseq'$
means that the term $M$ reduces to $M'$ in environment $E$ with the
replication indices defined by $\sigma$, the table contents $\tblcts$, and
the sequence of events $\evseq$, with probability $p$.
Rule~\eqref{sem:replindex} evaluates replication indices using the 
function $\sigma$.
Rule~\eqref{sem:var} looks for the value of the variable in the
environment $E$.
Rule~\eqref{sem:fun} evaluates the function call.
Rule~\eqref{sem:newt} chooses a random $a \in T$ according to distribution
$D_T$, and stores it in $x[\sigma(\tup{i})]$ by extending the environment $E$
accordingly.
Similarly, Rule~\eqref{sem:lett} extends the environment $E$ with the
value of $x[\sigma(\tup{i})]$.
Rule~\eqref{sem:ift1} evaluates the $\kw{then}$ branch of $\kw{if}$ when the condition is true, and Rule~\eqref{sem:ift2} evaluates the $\kw{else}$ branch otherwise.

Rules~\eqref{sem:findte} to~\eqref{sem:findt3} define the semantics of $\kw{find}$. First, they all evaluate the conditions for all branches $j$ and all values of the indices $i_{j1}, \ldots, i_{jm_j}$. If one of these evaluations executes an event (which can happen in case the condition contains an $\keventabort{e}$ or a $\kw{find}\unique{e}$), the whole $\kw{find}$ executes the same event; in case the evaluations of the conditions execute several different events, one of them is chosen randomly, according to distribution $D_{\kw{find}}(S)$, that is, almost uniformly over the choices of branches and indices (Rule~\eqref{sem:findte}). Otherwise, the branch and indices for which the condition is true are collected in a set $S$. If $S$ is empty, the $\kw{else}$ branch of the $\kw{find}$ is executed (Rule~\eqref{sem:findt2}). When $S$ is not empty, two cases can happen. Either the $\kw{find}$ is not marked $\unique{e}$, and we choose an element $v_0 = (j', a'_1, \ldots, a'_{m_{j'}})$ of $S$ randomly according to the distribution $D_{\kw{find}}(S)$, store the corresponding indices $a'_1, \ldots, a'_{m_{j'}}$ in $u_{j'1}[\sigma(\tup{i})], \ldots, u_{j'm_{j'}}[\sigma(\tup{i})]$ by extending the environment accordingly, and we continue with the selected branch $N'_j$. If the $\kw{find}$ is marked $\unique{e}$ and $S$ has a single element, we do the same. If the $\kw{find}$ is marked $\unique{e}$ and $S$ has several elements, we execute the event $e$ (Rule~\eqref{sem:findt3}). 
We recall that $D_{\kw{find}}(S)(v_0)$ denotes the probability of choosing $v_0$ in the distribution  $D_{\kw{find}}(S)$.
The terms in conditions of $\kw{find}$ may define variables, included in the environment $E''$; we ignore these additional variables and compute the final environment from the initial environment $E$, because these variables have no array accesses by Invariant~\ref{invfc}, so the values of these variables are not used after the evaluation of the condition.
The conditions of $\kw{find}$, $D \wedge M$, are evaluated using Rules~\eqref{sem:definedno} and~\eqref{sem:definedyes}. If an element of the $\defined$ condition $D$ is not defined, then the condition is $\false$ (Rule~\eqref{sem:definedno}); when all elements of the $\defined$ condition are defined, we evaluate $M$ (Rule~\eqref{sem:definedyes}). Since terms in conditions of $\kw{find}$ do not contain $\kw{insert}$ nor $\kw{event}$ (Invariant~\ref{invtfc}), the table contents and the sequence of events are left unchanged by the evaluation of the condition of $\kw{find}$.

Rule~\eqref{sem:insertt} inserts the new table element in $\tblcts$.
Rules~\eqref{sem:gette} to~\eqref{sem:gett3} define the semantics of $\kw{get}$.
We denote by $[ x \in L \mid f(x) ]$ the list of all elements $x$ of the list $L$ that satisfy $f(x)$, in the same order as in $L$. 
We denote by $|L|$ the length of list $L$.
We denote by $\nth{L}{j}$ the $j$-th element of the list $L$.
Rule~\eqref{sem:gette} executes $\keventabort{e}$ when the evaluation of the condition $M$ executes $\keventabort{e}$ for some element of the table $\tbl$; when several events may be executed, one of them is chosen randomly according to distribution $D_{\kw{get}}(S)$, that is, almost uniformly in the elements of the table $\tbl$.
Rules~\eqref{sem:gett1}, \eqref{sem:gett2}, and~\eqref{sem:gett3} compute the set $S$ of elements of indices of elements of table $\tbl$ in $\tblcts$ that satisfy condition $M$. If $S$ is empty, we execute $N'$ (Rule~\eqref{sem:gett2}).
When $S$ is not empty, two cases can happen. If the $\GET$ is not marked $\unique{e}$ or $S$ has a single element, then one of its elements is chosen randomly according to distribution $D_{\kw{get}}(S)$, we store this element in
$x_1[\sigma(\tup{i})], \ldots, x_l[\sigma(\tup{i})]$ by extending the environment $E$, and continue by executing $N$ (Rule~\eqref{sem:gett1}). 
If the $\GET$ is marked $\unique{e}$ and $S$ contains several elements, then we execute event $e$ and abort (Rule~\eqref{sem:gett3}).
Since terms in conditions of $\kw{get}$ do not contain $\kw{insert}$ nor $\kw{event}$ (Invariant~\ref{invtfc}), the table contents and the sequence of events are left unchanged by the evaluation of $M$.
The modified environment $E''$ obtained after evaluating a condition $M$ can be ignored because there are no array accesses to the variables defined in conditions of $\kw{get}$, by Invariant~\ref{invfc}, so the values of these variables are not used after the evaluation of the condition.

\begin{remark}
Another way of defining the semantics of tables would be to consider
two distinct calculi, one with tables (used for the initial game), and one
without tables (used for the other games). 
The semantics of the calculus without tables can be defined without the component
$\tblcts$.
We then need to relate the two semantics. 
\end{remark}

Rule~\eqref{sem:eventt} adds the executed event to $\evseq$. 
Rule~\eqref{sem:eventabortt} executes $\keventabort{e}$. The event $e$ is not immediately added
to $\evseq$, because for terms that occur in conditions of $\FIND$, in case several branches execute $\kw{event\_abort}$, we may need to choose
randomly which event will be added to $\evseq$. Hence, we use the result $\keventabort{(\pp,\tup{a}):e}$ instead.

\begin{figure}[t]
  \begin{align*}
    C ::= {}&\pptag x[a_1, \ldots, a_{k-1}, [\,], M_{k+1}, \ldots, M_m]\\
    &\pptag f(a_1, \ldots, a_{k-1}, [\,], M_{k+1}, \ldots, M_m)\\
    &\pptag \assign{x[\tup{i}]:T}{[\,]}N\\
    &\pptag \bguard{[\,]}{N}{N'}\\
    &\pptag \kevent{e(a_1, \ldots, a_{k-1}, [\,], M_{k+1}, \ldots, M_l)}; N\\
    &\pptag \INSERT\ \tbl(a_1, \ldots, a_{k-1}, [\,], M_{k+1}, \ldots, M_l); N
  \end{align*}
\caption{Term contexts}\label{fig:termcontexts}
\end{figure}

Rules~\eqref{sem:ctxt} and~\eqref{sem:ctxeventt} allow evaluating terms under a context. In these rules, $C$ is an elementary 
context, of one of the forms defined in Figure~\ref{fig:termcontexts}.
When the term $N$ reduces to some other term $N'$, Rule~\eqref{sem:ctxt}
allows one to reduce it in the same way under a context $C$.
When the term $N$ is an event, $C[N]$ also executes the same event
by Rule~\eqref{sem:ctxeventt}. 

These rules define a small-step semantics for terms. We consider the
reflexive and transitive closure $\red{p}{\ix}^*$ of the relation
$\red{p}{\ix}$ to reach directly the normal form of the term, which can
be either a value $a$ or an abort event value $\keventabort{(\pp,\tup{a}):e}$. 
We have $E, \sigma, M, \tblcts, \evseq
\red{1}{}^* E, \sigma, M, \tblcts, \evseq$ and, 
if $E, \sigma, M, \tblcts, \evseq \red{p}{\ix} E', \sigma', M', \tblcts', \evseq'$ 
and $E', \sigma', M', \tblcts', \evseq' \red{p'}{\ix'}^* E'', \ab \sigma'', \ab M'', \ab \tblcts'', \ab \evseq''$, then
$E, \sigma, M, \tblcts, \evseq \red{p \times p'}{\ix,\ix'}^* E'', \sigma'', M'', \tblcts'', \evseq''$: we
take the product of the probabilities to have the probability of a
sequence of reductions, and we specify which sequence was taken by a
list of indices $\ix, \ix'$.

\begin{figure}[t]
\begin{gather}
E, \{ (\sigma, \pptag 0) \} \multiunion \pset, \cset \redq E, \pset, \cset \tag{Nil}\label{sem:nil}\\
E, \{ (\sigma, \pptag (Q_1 \parpop Q_2)) \} \multiunion \pset, \cset \redq E, \{ (\sigma, Q_1), (\sigma, Q_2) \} \multiunion \pset, \cset\tag{Par}\label{sem:par}\\
E, \{ (\sigma, \pptag \repl{i}{n} Q) \} \multiunion \pset, \cset \redq E, \{ (\sigma[i \mapsto a], Q) \mid a \in [1, n] \} \multiunion \pset, \cset \tag{Repl}\label{sem:repl}\\
\frac{ c' \notin \cset }{
E, \{ (\sigma, \pptag \Reschan{c};Q) \} \multiunion \pset, \cset 
\redq E, \{ (\sigma, Q\{ c'/c \}) \} \multiunion \pset, \cset \cup \{ c' \}
}\tag{NewChannel}\label{sem:newchannel}\\
\frac{E, \sigma, M_j, \emptyset, \emptyset \red{1}{} E, \sigma, M'_j, \emptyset, \emptyset}{
\begin{array}{@{}l@{}}
E,\{ (\sigma, C[M_j]) \} \multiunion \pset, \cset \redq E, \{ (\sigma, C[M'_j])\} \multiunion \pset, \cset\\
\qquad \text{where $C = \pptag \cinput{c[a_1, \dots, a_{j-1}, [\,], M_{j+1}, \dots, M_l]}{x[\tup{i}]:T};P$}
\end{array}}
\tag{Input}\label{sem:input}\\
%The previous version with for all $j \leq l$, E, \sigma, M_j, \emptyset, \emptyset \red{1}{}^* E, \sigma, a_j, \emptyset, \emptyset allowed non termination: \cinput{c[a_1, \dots, a_l]}{x[\tup{i}]:T};P reduced to itself 
%
\reduce(E, \pset, \cset) \text{ is the normal form of }E, \pset, \cset \text{ by }\redq\notag
\end{gather}
\caption{Semantics (3): input processes}\label{fig:sem2}
\end{figure}

Figure~\ref{fig:sem2} defines the semantics of input processes. 
We use an auxiliary reduction relation $\redq$, 
for reducing
input processes. This relation transforms configurations of the form
$E, \pset, \cset$.
Rule~\eqref{sem:nil} removes nil processes.
Rules~\eqref{sem:par} and~\eqref{sem:repl} expand parallel compositions and replications,
respectively.
Rule~\eqref{sem:newchannel} creates a new channel and adds it 
to $\cset$.
Semantic configurations are considered equivalent modulo renaming of
channels in $\cset$, so that a single semantic configuration is
obtained after applying~\eqref{sem:newchannel}.
Rule~\eqref{sem:input} evaluates the terms in the input channel. The input
itself is not executed: the communication is done by the~\eqref{sem:output} rule.
In the \eqref{sem:input} rule, the terms $M_1, \ldots, M_l$ are simple by
Invariant~\ref{invtic}, so their evaluation is deterministic (the unique
result is obtained with probability 1), the environment $E$, the contents of tables $\tblcts$, and the sequence of events $\evseq$ are
unchanged, and $\tblcts$ and $\evseq$ are unused, that is why we can write 
$E, \sigma, M_j, \emptyset, \emptyset \red{1}{} E, \sigma, M'_j, \emptyset, \emptyset$
using empty $\tblcts$ and $\evseq$, written $\emptyset$.
The relation $\redq$ is convergent (confluent and terminating),
so it has normal forms.
Processes in $\pset$ in configurations $E, (\sigma,P)\restconfig$
are always in normal form by $\redq$, so they always start with an input.

\begin{figure}[tp]
\vspace*{-3mm}%
\begin{gather}
\frac{
\text{Same assumption as in \eqref{sem:findte}}\qquad r_{k_0} = \keventabort{(\pp',\tup{a}):e}
}{\begin{array}{c}
E, (\sigma, \pptag \FIND\uniqueopt\ (\mathop\bigoplus\nolimits_{j=1}^m
\vf_{j1}[\tup{i}] = i_{j1} \leq n_{j1}, \ldots, \vf_{jm_j}[\tup{i}] = i_{jm_j} \leq n_{jm_j}\ \kw{suchthat}\\ 
D_j \fand M_j\ \kw{then}\ P_j)\ \ELSE P)\restconfig \\
\red{p_1\ldots p_{l_0}D_{\FIND}(S)(k_0)}{\ix_1\ldots \ix_{l_0}FE(k_0)} E_{k_0}, (\sigma_{k_0}, \kw{abort}), \pset, \cset, \tblcts, (\evseq, (\pp',\tup{a}):e)
\end{array}}\tag{FindE}\label{sem:finde}\\[1mm]
\frac{
\begin{array}{c}
\text{First four lines as in \eqref{sem:findt1}}\qquad
S = \{ v_k \mid r_k = \true \}\qquad 
|S| = 1 \text{ or } \uniqueopt \text{ is empty}\\
v_0 = (j', a'_1, \ldots, a'_{m_{j'}})\in S\qquad
E' = E[\vf_{j'1}[\sigma(\tup{i})] \mapsto a'_1, \ldots, \vf_{j'm_{j'}}[\sigma(\tup{i})] \mapsto a'_{m_{j'}}]
\end{array}
}{\begin{array}{c}
E, (\sigma, \pptag \FIND\uniqueopt\ (\mathop\bigoplus\nolimits_{j=1}^m
\vf_{j1}[\tup{i}] = i_{j1} \leq n_{j1}, \ldots, \vf_{jm_j}[\tup{i}] = i_{jm_j} \leq n_{jm_j}\ \kw{suchthat}\\ 
D_j \fand M_j\ \kw{then}\ P_j)\ \ELSE P)\restconfig \red{p_1\ldots p_l D_{\FIND}(S)(v_0)}{\ix_1\ldots \ix_l F1(v_0)} E', (\sigma, P_{j'})\restconfig
\end{array}}\tag{Find1}\label{sem:find1}\\[1mm]
\frac{
\begin{array}{c}
\text{First four lines as in \eqref{sem:findt1}}\qquad
S = \{ v_k \mid r_k = \true \} = \emptyset
\end{array}
}{\begin{array}{c}
E, (\sigma, \pptag \FIND\uniqueopt\ (\mathop\bigoplus\nolimits_{j=1}^m
\vf_{j1}[\tup{i}] = i_{j1} \leq n_{j1}, \ldots, \vf_{jm_j}[\tup{i}] = i_{jm_j} \leq n_{jm_j}\ \kw{suchthat}\\ 
D_j \fand M_j\ \kw{then}\ P_j)\ \ELSE P)\restconfig \red{p_1\ldots p_l}{\ix_1\ldots \ix_l F2} E, (\sigma, P)\restconfig
\end{array}}\tag{Find2}\label{sem:find2}\\[1mm]
\frac{
\begin{array}{c}
\text{First four lines as in \eqref{sem:findt1}}\qquad
S = \{ v_k \mid r_k = \true \}\qquad |S| > 1
\end{array}
}{\begin{array}{@{}c@{}}
E, (\sigma, \pptag \FIND\unique{e}\ (\mathop\bigoplus\nolimits_{j=1}^m
\vf_{j1}[\tup{i}] = i_{j1} \leq n_{j1}, \ldots, \vf_{jm_j}[\tup{i}] = i_{jm_j} \leq n_{jm_j}\ \kw{suchthat} D_j \fand {}\\ 
M_j\ \kw{then}\ P_j)\ \ELSE P)\restconfig \red{p_1\ldots p_l}{\ix_1\ldots \ix_l F3} E, (\sigma, \kw{abort}), \pset, \cset, \tblcts, (\evseq, (\pp,\image(\sigma)):e)
\end{array}}\tag{Find3}\label{sem:find3}\\[1mm]
E, (\sigma, \pptag \INSERT\ \tbl(a_1, \ldots, a_l); P)\restconfig \red{1}{} E, (\sigma, P), \pset, \cset, (\tblcts, \tbl(a_1, \ldots, a_l)), \evseq
\tag{Insert}\label{sem:insert}\\[1mm]
\frac{\text{Same assumption as in \eqref{sem:gette}}\qquad r_{k_0} = \keventabort{(\pp',\tup{a}):e}
}{\begin{array}{@{}c@{}}
E, (\sigma, \pptag \GET\uniqueopt\ \tbl(x_1[\tup{i}]:T_1,\ldots,x_l[\tup{i}]:T_l)\ \SUCHTHAT\ M\ \IN\ P\ \ELSE P') \restconfig\\
\red{p_1\dots p_m D_{\kw{get}}(S)(k_0)}{\ix_1\dots \ix_m GE(k_0)} E_{k_0}, (\sigma_{k_0}, \kw{abort}), \pset, \cset, \tblcts, (\evseq, (\pp',\tup{a}):e)
\end{array}}\tag{GetE}\label{sem:gete}\\[1mm]
\frac{
\begin{array}{@{}c@{}}
  \text{First four lines as in \eqref{sem:gett1}}\qquad
  |S| = 1 \text{ or } \uniqueopt \text{ is empty}\\
k_0 \in S\qquad \tbl(a_1, \ldots, a_l) = v_{k_0} \qquad
E' = E[x_1[\sigma(\tup{i})] \mapsto a_1, \ldots, x_l[\sigma(\tup{i})] \mapsto a_l]
\end{array}}{
\begin{array}{@{}c@{}}
E, (\sigma, \pptag \GET\uniqueopt\ \tbl(x_1[\tup{i}]:T_1,\ldots,x_l[\tup{i}]:T_l)\ \SUCHTHAT\ M\ \IN\ P\ \ELSE P')\restconfig\\
\red{p_1\dots p_m D_{\kw{get}}(S)(k_0)}{\ix_1\dots \ix_m G1(k_0)} E', (\sigma, P)\restconfig
\end{array}}\tag{Get1}\label{sem:get1}\\[1mm]
\frac{\text{First four lines as in \eqref{sem:gett1}}\qquad S = \emptyset}{
\begin{array}{c}
E, (\sigma, \pptag \GET\uniqueopt\ \tbl(x_1[\tup{i}]:T_1,\ldots,x_l[\tup{i}]:T_l)\ \SUCHTHAT\ M\ \IN\ P\ \ELSE P')\restconfig\\
\red{p_1\dots p_m}{\ix_1\dots \ix_m G2} E, (\sigma, P')\restconfig
\end{array}}\tag{Get2}\label{sem:get2}\\[1mm]
\frac{\text{First four lines as in \eqref{sem:gett1}}\qquad |S|>1}{
\begin{array}{c}
E, (\sigma, \pptag \GET\unique{e}\ \tbl(x_1[\tup{i}]:T_1,\ldots,x_l[\tup{i}]:T_l)\ \SUCHTHAT\ M\ \IN\ P\ \ELSE P')\restconfig\\
\red{p_1\dots p_m}{\ix_1\dots \ix_m G3} E, (\sigma, \kw{abort}), \pset, \cset, \tblcts, (\evseq, (\pp, \image(\sigma)):e)
\end{array}}\tag{Get3}\label{sem:get3}
\end{gather}%
\caption{Semantics (4): output processes, first part}\label{fig:sem3}
\end{figure}

\begin{figure}[t]
\vspace*{-3mm}%
\begin{gather}
\frac{a \in T \qquad E' = E[x[\sigma(\tup{i})] \mapsto a]}{
E, (\sigma, \pptag \Res{x[\tup{i}]}{T};P)\restconfig 
\red{D_T(a)}{N(a)} E', (\sigma, P)\restconfig}\tag{New}\label{sem:new}\\[1mm]
\frac{a \in T \qquad E' = E[x[\sigma(\tup{i})] \mapsto a]
}{
E, (\sigma, \pptag \assign{x[\tup{i}] : T}{a}{P})\restconfig \red{1}{} E', (\sigma, P)\restconfig}\tag{Let}\label{sem:let}\\[1mm]
E, (\sigma, \pptag \bguard{\true}{P}{P'})\restconfig \red{1}{} 
E, (\sigma, P)\restconfig  \tag{If1}\label{sem:if1}\\[1mm]
\frac{a \neq \true}{
E, (\sigma, \pptag \bguard{a}{P}{P'})\restconfig \red{1}{} 
E, (\sigma, P')\restconfig}  \tag{If2}\label{sem:if2}\\[1mm]
\begin{split}
  &E, (\sigma, \pptag \kevent{e(a_1, \ldots, a_l)};P)\restconfig \red{1}{} \\
  &\qquad E, (\sigma, P), \pset, \cset, \tblcts, (\evseq, (\pp, \image(\sigma)):e(a_1, \ldots, a_l))
\end{split}\tag{Event}\label{sem:event}\\[1mm]
E, (\sigma, \pptag \keventabort{e})\restconfig \red{1}{} E, (\sigma, \kw{abort}), \pset, \cset, \tblcts, (\evseq, (\pp, \image(\sigma)):e)\tag{EventAbort}
\label{sem:eventabort}\\[1mm]
\frac{E, \sigma, N, \tblcts, \evseq \red{p}{\ix} E', \sigma', N', \tblcts', \evseq'}{
E, (\sigma, C[N])\restconfig \red{p}{\ix} E', (\sigma', C[N']),\pset,\cset,\tblcts',\evseq'} \tag{Ctx}\label{sem:ctx}\\[1mm]
E, (\sigma, C[\keventabort{(\pp, \tup{a}):e}])\restconfig \red{1}{} E, (\sigma, \kw{abort}), \pset, \cset, \tblcts, (\evseq, (\pp, \tup{a}):e) \tag{CtxEvent}\label{sem:ctxevent}\\[1mm]
\frac{\begin{array}{@{}c@{}}
E, \pset', \cset' = \reduce(E, \{ (\sigma, Q'')\}, \cset)\\
S= \{ (\sigma', Q) \in \pset \mid  Q = {}^{\pp''} \cinput{c[a_1, \ldots, a_l]}{x'[\tup{i}] : T'}.P'\text{ and }b \in T'\text{ for some }\pp'', \sigma',x',T', P'\}\\
(\sigma', Q_0) \in S \qquad Q_0 = {}^{\pp'} \cinput{c[a_1, \ldots, a_l]}{x[\tup{i}]:T}.P 
\end{array}}{
\begin{array}{@{}l@{}}
E, (\sigma, \pptag \coutput{c[a_1, \ldots, a_l]}{b}.Q'')\restconfig 
\red{S(\sigma', Q_0) \times D_{\kw{in}}(S)(\sigma', Q_0)}{O(\sigma',Q_0)}\\
\qquad 
E[x[\sigma'(\tup{i})] \mapsto b], (\sigma', P), \pset \multiunion \pset'
\setminus \{ (\sigma', Q_0) \}, \cset', \tblcts, \evseq
\end{array}}
\tag{Output}\label{sem:output}
\end{gather}%
\vspace*{-4mm}%
\caption{Semantics (5): output processes, second part}\label{fig:sem3bis}
\end{figure}

\begin{figure}[t]
  \begin{align*}
    C ::= {}&\pptag \assign{x[\tup{i}]:T}{[\,]}P\\
    &\pptag \bguard{[\,]}{P}{P'}\\
    &\pptag \coutput{c[a_1, \ldots, a_{k-1}, [\,], M_{k+1}, \ldots, M_l]}{N}.Q\\
    &\pptag \coutput{c[a_1, \ldots, a_l]}{[\,]}.Q\\
    &\pptag \INSERT\ \tbl(a_1, \ldots, a_{k-1}, [\,], M_{k+1}, \ldots, M_l); P\\
    &\pptag \kevent{e(a_1, \ldots, a_{k-1}, [\,], M_{k+1}, \ldots, M_l)};P
  \end{align*}
\caption{Process contexts}\label{fig:proccontexts}
\end{figure}

Finally, Figures~\ref{fig:sem3} and~\ref{fig:sem3bis} define the semantics of output processes.
Most of these rules are very similar to those for terms: they just use processes instead of terms as continuations, and include a whole semantic configuration.
Rule~\eqref{sem:eventabort} executes event $e$ and aborts the game, by reducing to the configuration with process $\kw{abort}$.
Similarly to the case of terms, Rules~\eqref{sem:ctx} and~\eqref{sem:ctxevent} allow evaluating terms under a context inside a process. 
In these rules, $C$ is an elementary context, of one of the forms defined
in Figure~\ref{fig:proccontexts}.

Rule~\eqref{sem:output} performs communications: it selects an
input on the desired channel randomly, and immediately executes the
communication. (The process blocks if no suitable input is available.) 
The scheduled process after this rule is the receiving process. 
The input processes that follow the output are stored in the
available input processes, after reducing them by rules of 
Figure~\ref{fig:sem2}. In this rule, $S$ is a multiset.
When we take probabilities over multisets, 
we consider that $D_{\kw{in}}(S)(\sigma',Q_0)$ is the probability of choosing
\emph{one} of the elements equal to $\sigma',Q_0$ in $S$ according to the
distribution $D_{\kw{in}}(S)$, so that the probability of choosing any
element equal to $(\sigma',Q_0)$ is in fact $S(\sigma',Q_0) \times D_{\kw{in}}(S)(\sigma',Q_0)$.

After finishing execution of a process, the system produces
the sequence of executed events $\evseqnopp$.
These events can be used to distinguish games, so we
introduce an additional algorithm, a \emph{distinguisher} $D$ that
takes as input a sequence of events $\evseqnopp$ (without program points and 
replication indices) and returns $\true$ or $\false$.

An example of distinguisher is $D_e$ defined by $D_e(\evseqnopp) = \true$ 
if and only if $e \in \evseqnopp$: this distinguisher detects the execution of event $e$. We will denote the distinguisher $D_e$ simply by $e$.
More generally, distinguishers can detect various properties
of the sequence of events $\evseqnopp$ executed by the game and of its result $a$.
We denote by $D \vee D'$, $D \wedge D'$, and $\neg D$ the distinguishers such that
$(D \vee D')(\evseqnopp) = D(\evseqnopp) \vee D'(\evseqnopp)$, 
$(D \wedge D')(\evseqnopp) = D(\evseqnopp) \wedge D'(\evseqnopp)$,
and $(\neg D)(\evseqnopp) = \neg D(\evseqnopp)$.
We denote by $\Pr[Q : D]$ the probability that $Q$
executes a sequence of events $\evseqnopp$ such that
$D(\evseqnopp) = \true$. This is formally
defined as follows.

\begin{definition}
The initial configuration for running process $Q$ is
$\initconfig(Q) = \emptyset, \ab (\sigma_0, \pptag\coutput{\startch}{}), \ab \pset, \ab \cset, \ab \emptyset, \ab \emptyset$
where $\emptyset, \pset, \cset = \reduce(\emptyset, \{ (\sigma_0, Q) \}, \fc(Q))$ and $\sigma_0$ is the empty mapping sequence.

A \emph{trace of} $Q$ is a trace that starts from $\initconfig(Q)$: $\trace = \initconfig(Q) \ab \red{p_1}{\ix_1} \ab \ldots \ab \red{p_{m-1}}{\ix_{m-1}} \ab \conf_m$.
Let $\traceset$ be the set of all traces of $Q$.

Let $\tracesetfull$ be the set of \emph{full traces} of $Q$, that is, the set of traces of $\traceset$ whose last configuration $\conf_m$ cannot be reduced.

A trace $\trace'$ is an \emph{extension} of $\trace$ when $\trace'$ is obtained by continuing execution from the last configuration of $\trace$. Equivalently, $\trace$ is a \emph{prefix} of $\trace'$.
\bb{Instead of using the notion of prefix, perhaps I should use $\trace \before \conf$. That strengthens a bit the property guaranteed by ``yields a contradiction'': it is a contradiction not only in the sequence of output process configurations of $\trace$ but in all configurations of the derivation. Then in the proof of success simplify, I could consider $\trace \before \conf'$ for some $\conf'$ not necessarily an output process configuration instead of $\trace'$.}%

Let $\varphi$ be a property of traces, that is, a function from traces to $\{\true, \false\}$. We say that $\trace$ satisfies $\varphi$, and we write $\trace \vdash \varphi$, when $\varphi(\trace) = \true$.

A property $\varphi$ is \emph{preserved by extension} when for all traces $\trace$ such that $\trace \vdash \varphi$, for all extensions $\trace'$ of $\trace$, $\trace'\vdash \varphi$.

Given a trace $\trace = \initconfig(Q) \red{p_1}{\ix_1} \ldots \red{p_{m-1}}{\ix_{m-1}} \conf_m$, recall that
$\Pr[\trace] = p_1 \times \ldots \times p_{m-1}$.
We define
\begin{align*}
  \Pr[Q:\varphi] &= \sum_{\trace \in \tracesetfull, \trace\vdash \varphi} \Pr[\trace]\,,\\
  \Prss{Q}{\varphi} &= \sum_{\trace \in \tracesetfull, \exists\trace'\text{ prefix of }\trace, \trace' \vdash \varphi} \Pr[\trace]\\
  &= \sum_{\trace \in \traceset, \trace\vdash \varphi, \text{for any strict prefix $\trace'$ of $\trace$}, \trace'\not\vdash \varphi} \Pr[\trace]\,.
\end{align*}
Given a distinguisher $D$, we can consider it as a property of traces by defining $\trace \vdash D$ if and only if $D(\removepp(\evseq_m)) = \true$ where the last configuration of $\trace$ is $\conf_m = E_m, \ab (\sigma_m, P_m), \ab \pset_m, \ab \cset_m, \ab \tblcts_m, \ab \evseq_m$.
The function $\removepp$ guarantees that the pair (program point, replication indices) is not used in the evaluation of the distinguisher. Actually, this pair could be removed from the semantics. It is useful for the proof of injective correspondences (Section~\ref{sec:injcorresp}).
\end{definition}

\begin{lemma}\label{lem:Pr-Prss}
  \begin{enumerate}
  \item $\Pr[Q:\varphi] \leq \Prss{Q}{\varphi}$.
  \item If $\varphi$ is preserved by extension, then $\Pr[Q:\varphi] = \Prss{Q}{\varphi}$.
  \end{enumerate}
\end{lemma}
\begin{proof}
  The first property holds because, when $\trace \vdash \varphi$,
  there exists a prefix $\trace'$ of $\trace$ (take $\trace' = \trace$) such that $\trace' \vdash \varphi$.

  The second property holds because, if $\varphi$ is preserved by extension and there exists a prefix $\trace'$ of $\trace$ such that $\trace '\vdash \varphi$, then $\trace \vdash \varphi$.
  \proofcomplete
\end{proof}

For simple terms $M$,
the evaluation can be defined without tables and events. We define $E, \venv, M \evalterm a$
if and only if $E, \venv, M, \emptyset, \emptyset \red{1}{}^* E, \venv, a, \emptyset, \emptyset$, 
where the environment $E$ gives the values of process variables (in particular arrays), and the 
environment $\venv$ gives the values of replication indices and other variables (e.g., those 
used in correspondences, see Section~\ref{sec:def:corresp}). The evaluation relation 
$E, \venv, M \evalterm a$ can also be defined by induction as follows:
\begin{gather*}
  E, \venv, i \evalterm \venv(i)\\[1mm]
  \frac{\forall j \leq m, E, \venv, M_j \evalterm a_j}{E, \venv, x[M_1, \dots, M_m] \evalterm E(x[a_1, \dots, a_m])}\\[1mm]
  \frac{\forall j \leq m, E, \venv, M_j \evalterm a_j}{E, \venv, f(M_1, \dots, M_m) \evalterm f(a_1, \dots, a_m)}
\end{gather*}
For terms that do not contain process variables, the environment $E$ can be omitted,
and we write $\venv, M \evalterm a$. It can also be defined by induction as follows:
\begin{gather*}
  \venv, i \evalterm \venv(i)\\[1mm]
  \frac{\forall j \leq m, \venv, M_j \evalterm a_j}{\venv, f(M_1, \dots, M_m) \evalterm f(a_1, \dots, a_m)}
\end{gather*}

\subsubsection{Properties}

Given a process $Q_0$, 
we write $\replidx{\pp}$ for the current replication indices at $\pp$ in $Q_0$.

\begin{lemma}\label{lem:curidx}
Let $\trace$ be a trace of $Q_0$.
In the derivation of $\trace$, 
for all configurations $E, \ab \sigma, \ab \pptag M, \ab \tblcts, \ab \evseq$
or $E, (\sigma, \pptag P)\restconfig$,
$\dom(\sigma) = \replidx{\pp}$ is the sequence of current replication indices at $\pp$
in $Q_0$ (or $\dom(\sigma) = \emptyset$ is the sequence of current replication indices at $\pptag\coutput{\startch}{}$ in the initial configuration) and for all configurations $E,\pset,\cset$ or $E, (\sigma, \pptag P)\restconfig$, for all $(\sigma', {}^{\pp'} Q) \in \pset$, 
$\dom(\sigma') = \replidx{\pp'}$.
\end{lemma}
\begin{proofsk}
We say that
\begin{itemize}
\item a configuration $E, \sigma, \pptag M, \tblcts, \evseq$ is ok when $\dom(\sigma) = \replidx{\pp}$;
\item a configuration $E,\pset,\cset$ is ok when for all $(\sigma', {}^{\pp'} Q) \in \pset$, 
$\dom(\sigma') = \replidx{\pp'}$;
\item a configuration $E, (\sigma, \pptag P)\restconfig$ is ok when $\dom(\sigma) = \replidx{\pp}$ (or $\dom(\sigma) = \emptyset$ is the sequence of current replication indices at $\pptag\coutput{\startch}{}$ in the initial configuration) and for all $(\sigma', {}^{\pp'} Q) \in \pset$, 
$\dom(\sigma') = \replidx{\pp'}$.
\end{itemize}
We show by induction on the derivations that
\begin{enumerate}
\item\label{curidx1} if $E_1, \sigma_1, M_1, \tblcts_1, \evseq_1$ is ok, then all configurations $E, \sigma, M, \tblcts, \evseq$ in the derivation of $E_1, \ab \sigma_1, \ab M_1, \ab \tblcts_1, \ab \evseq_1 \red{p}{\ix}
E_2, \sigma_2, M_2, \tblcts_2, \evseq_2$ are ok;

\item\label{curidx2} if $E_1, \pset_1, \cset_1$ is ok, then all configurations $E, \sigma, M, \tblcts, \evseq$ in the derivation of $E_1, \ab \pset_1, \ab \cset_1 \redq E_2, \pset_2, \cset_2$ are ok and $E_2, \pset_2, \cset_2$ is ok;

\item\label{curidx2bis} if $E_1, \pset_1, \cset_1$ is ok, then $\reduce(E_1, \pset_1, \cset_1)$ is ok and all configurations $E, \sigma, M, \tblcts, \evseq$ or $E,\pset,\cset$ and in the derivation of $E_1, \pset_1, \cset_1 \redq^* \reduce(E_1, \pset_1, \cset_1)$ are ok;

\item\label{curidx3} if $E_1, (\sigma_1, P_1), \pset_1, \cset_1, \tblcts_1, \evseq_1$ is ok, then all configurations $E, \sigma, M, \tblcts, \evseq$
or $E,\pset,\cset$ or $E, \ab (\sigma, P)\restconfig$ in the derivation of
$E_1, \ab (\sigma_1, P_1), \ab \pset_1, \ab \cset_1, \ab \tblcts_1, \ab \evseq_1 \red{p}{\ix}
E_2, \ab (\sigma_2, P_2), \ab \pset_2, \ab \cset_2, \ab \tblcts_2, \ab \evseq_2$ are ok.

\end{enumerate}
Indeed, the changes in $\sigma$ match the definition of replication indices.

For Property~\ref{curidx1}, in rules for $\FIND$, the indices of the $\FIND$ condition are added to the domain of $\sigma$ when evaluating the condition and in all other cases, $\sigma$ is unchanged. 

In the proof of Property~\ref{curidx2}, in~\eqref{sem:repl}, the replication index $i$ is added to $\sigma$ and in all other cases, $\sigma$ is unchanged. We use Property~\ref{curidx1} in the case of input~\eqref{sem:input}. 

Property~\ref{curidx2bis} follows immediately from Property~\ref{curidx2} by induction.

For Property~\ref{curidx3}, in rules for $\FIND$, the indices of the $\FIND$ condition are added to the domain of $\sigma$ when evaluating the condition, as in Property~\ref{curidx1}; in~\eqref{sem:output}, we use Property~\ref{curidx2bis}. In all other cases, $\sigma$ is unchanged. We use Property~\ref{curidx1} when evaluating terms, in conditions of $\FIND$ and $\GET$ and in~\eqref{sem:ctx}.

Moreover, in the computation of $\initconfig(Q_0)$, the configuration $\emptyset, \{ (\sigma_0, Q_0) \}, \fc(Q)$ is ok (the current replication indices at the root of $Q_0$ are empty), so by Property~\ref{curidx2bis}, $\initconfig(Q_0)$ is ok.
  \proofcomplete
\end{proofsk}

\begin{lemma}\label{lem:sem-ext}
  If $E, \sigma, N, \tblcts, \evseq \red{p}{\ix} E', \sigma', N', \tblcts', \evseq'$, then
  $E'$ is an extension of $E$, $\tblcts$ is a prefix of $\tblcts'$, $\evseq$ is a prefix of $\evseq'$,
  $\sigma'$ is an extension of $\sigma$, and if the term $N'$ is not of the form $C_1[\dots C_k[\keventabort{(\pp, \tup{a}):e}]\dots]$ for some $k \in \mathbb{N}$ and $C_1$, \dots, $C_k$ contexts defined in Figure~\ref{fig:termcontexts}, then $\sigma' = \sigma$.
  For all configurations $E'', \sigma'', N'', \tblcts'', \evseq''$ in the derivation of $E, \sigma, N, \tblcts, \evseq \red{p}{\ix} E', \sigma', N', \tblcts', \evseq'$, we have that $E''$ is an extension of $E$ and, if $E'', \sigma'', N'', \tblcts'', \evseq''$ is not in the derivation of an hypothesis of a rule for $\FIND$ or $\GET$, then $E'$ is an extension of $E''$; $\sigma''$ is an extension of $\sigma$; $\tblcts$ is a prefix of $\tblcts''$, which is a prefix of $\tblcts'$; and $\evseq$ is a prefix of $\evseq''$, which is a prefix of $\evseq'$.

  If $E, \pset, \cset \redq E', \pset', \cset'$, then $E' = E$ and $\cset \subseteq \cset'$. For all configurations $E'', \ab \sigma'', \ab N'', \ab \tblcts'', \ab \evseq''$ in the derivation of $E, \pset, \cset \redq E', \pset', \cset'$, we have $E'' = E$, $\tblcts'' = \emptyset$, and $\evseq'' = \emptyset$.
  
  If $E, (\sigma, P), \pset, \cset, \tblcts, \evseq \red{p}{\ix} E', (\sigma', P'), \pset', \cset', \tblcts', \evseq'$, then $E'$ is an extension of $E$, $\tblcts$ is a prefix of $\tblcts'$, $\evseq$ is a prefix of $\evseq'$, and $\cset \subseteq \cset'$.
  \begin{itemize}
  \item   If $E, (\sigma, P), \pset, \cset, \tblcts, \evseq \red{p}{\ix} E', (\sigma', P'), \pset', \cset', \tblcts', \evseq'$ is derived by~\eqref{sem:output}, then $\tblcts' = \tblcts$, $\evseq' = \evseq$, for all configurations $E'', \pset'', \cset''$ in that derivation, $E'' = E$ and $\cset \subseteq \cset'' \subseteq \cset'$, and for all configurations $E'', \ab \sigma'', \ab N'', \ab \tblcts'', \ab \evseq''$ in that derivation, $E'' = E$, $\sigma''$ is an extension of $\sigma$, $\tblcts'' = \emptyset$, and $\evseq'' = \emptyset$.
  \item In all other cases, $\cset' = \cset$, $\sigma'$ is an extension of $\sigma$, and if the process $P'$ is not $\kw{abort}$ or of the form $C_0[C_1[\dots C_k[\keventabort{(\pp, \tup{a}):e}]\dots]]$ for some $C_0$ context defined in Figure~\ref{fig:proccontexts}, $k \in \mathbb{N}$, and $C_1$, \dots, $C_k$ contexts defined in Figure~\ref{fig:termcontexts}, then $\sigma' = \sigma$. For all configurations $E'', \sigma'', N'', \tblcts'', \evseq''$ in the derivation of $E, (\sigma, P), \pset, \cset, \tblcts, \evseq \red{p}{\ix} E', \ab (\sigma', P'), \ab \pset', \ab \cset', \ab \tblcts', \ab \evseq'$, we have that $E''$ is an extension of $E$ and, if $E'', \sigma'', N'', \tblcts'', \evseq''$ is not in the derivation of an hypothesis of a rule for $\FIND$ or $\GET$, then $E'$ is an extension of $E''$; $\sigma''$ is an extension of $\sigma$; $\tblcts$ is a prefix of $\tblcts''$, which is a prefix of $\tblcts'$; and $\evseq$ is a prefix of $\evseq''$, which is a prefix of $\evseq'$.
  \end{itemize}
\end{lemma}
\begin{proofsk}
  By induction on the derivation of $E, \sigma, N, \tblcts, \evseq \red{p}{\ix} E', \sigma', N', \tblcts', \evseq'$ and by cases on the reductions $E, \pset, \cset \redq E', \pset', \cset'$ and $E, (\sigma, P), \pset, \cset, \tblcts, \evseq \red{p}{\ix} E', \ab (\sigma', P'), \ab \pset', \ab \cset', \ab \tblcts', \ab \evseq'$. 
  \proofcomplete
\end{proofsk}

\bbnote{in find and get:
In case we execute $\kw{event\_abort}$, the configuration $(E, \sigma)$ at the selected $\kw{event\_abort}$ is returned. This is needed in the proof of correspondences, to make sure that the facts that hold at the $\kw{event\_abort}$ also hold at the end of the trace.
In all other cases, final $E$ is built from $E$ before the find by just adding the defined indices. The variables bound in the evaluation of the branches are removed. This is ok because these variables do not have array accesses. In the collection of true facts, we make sure to remove the facts that use variables defined in the condition of find/get, to compute the facts that hold in the $\kw{then}$ branches.
Other solutions could be:
 - Keep the variables in the successful branch (to make sure that the conditions of the successful branch hold)
 - Keep all variables of the branches, by taking the union of $E$ after evaluating all branches. That does not work for "get" because the variables defined in the condition of get are the same for all evaluations of the branches. For "find", that would require a proof to show that there is no conflict (same variable given different values in these environments). That's tricky. E.g., SArename may create several branches of find that use the same variables; at runtime, they are used with different indices, because the defined conditions cannot succeed for the same indices. These variables are subsequently renamed by AutoSArename, but the situation in the intermediate game makes it not so obvious to define the semantics that way. If we do that, then the final environment in a reduction is an extension of all environments in the derivation of that reduction, except for "get".}%

%In a configuration, the pairs (\pp, \sigma) that occur are pairwise distinct
%That would be a strenghtening of Lemma~\eqref{lem:oneevent}. Is it needed?

We say that a term $N$ is \emph{in evaluation position} in a
configuration $E, \sigma, M, \tblcts, \evseq$ when
$M = C_1[\dots C_k[N]\dots]$ for some $k \in \mathbb{N}$ and $C_1$,
\dots, $C_k$ contexts defined in Figure~\ref{fig:termcontexts}.

An \emph{input context} is a context of the form
$\pptag\cinput{c[a_1, \dots, a_{j-1}, [\,], M_{j+1}, \dots, M_m]}{x[\tup{i}]:T}; P$.

We say that a term $N$ is \emph{in evaluation position} in a
configuration $E, \pset, \cset$ when $(\sigma',Q) \in \pset$ and
$Q = C_0[C_1[\dots C_k[N]\dots]]$ for some $C_0$ input context,
$k \in \mathbb{N}$ and $C_1$, \dots, $C_k$ contexts defined in
Figure~\ref{fig:termcontexts}.

We say that a term $N$ is \emph{in evaluation position} in a
configuration $E, (\sigma, P)\restconfig$ when
$P = C_0[C_1[\dots C_k[N]\dots]]$ for some $C_0$ context defined in
Figure~\ref{fig:proccontexts}, $k \in \mathbb{N}$, and $C_1$, \dots,
$C_k$ contexts defined in Figure~\ref{fig:termcontexts} or
$(\sigma',Q) \in \pset$ and $Q = C_0[C_1[\dots C_k[N]\dots]]$ for some
$C_0$ input context, $k \in \mathbb{N}$ and $C_1$, \dots, $C_k$
contexts defined in Figure~\ref{fig:termcontexts}.

The output process $P$ is \emph{in evaluation position} in 
configuration $E, (\sigma, P)\restconfig$.
The input processes $Q$ such that $(\sigma',Q) \in \pset$ for some $\sigma'$
are \emph{in evaluation position} in 
configurations $E, \pset, \cset$ and $E, (\sigma, P)\restconfig$.

\begin{lemma}\label{lem:evalpos}
Consider a trace $\trace$ of $Q_0$.

All subterms of $M$ that occur in non-evaluation position in a configuration $E, \sigma, M, \tblcts, \evseq$ in the derivation of $\trace$ are subterms of $Q_0$.

All subterms and subprocesses of processes in $\pset$ that occur in non-evaluation position in a configuration $E, \pset,\cset$ in the derivation of $\trace$ are subterms, resp. subprocesses, of $Q_0$ up to renaming of channels.

All subterms and subprocesses of $P$ and of processes in $\pset$ that occur in non-evaluation position in a configuration $E, (\sigma,P), \pset, \cset, \tblcts, \evseq$ in the derivation of $\trace$ are subterms, resp. subprocesses, of $Q_0$ up to renaming of channels (except for the process $0$ that follows $\coutput{\startch}{}$ in $\initconfig(Q_0)$).
\end{lemma}
\begin{proof}
We say that
\begin{itemize}
\item a configuration $E, \sigma, M, \tblcts, \evseq$ is ok when all subterms of $M$ that are not in evaluation position are subterms of $Q_0$;
\item a configuration $E,\pset,\cset$ is ok when all subterms and subprocesses of processes in $\pset$ that are not in evaluation position are subterms, resp. subprocesses, of $Q_0$ up to renaming of channels;
\item a configuration $E, (\sigma, P)\restconfig$ is ok when all subterms and subprocesses of $P$ and of processes in $\pset$ that are not in evaluation position are subterms, resp. subprocesses, of $Q_0$ up to renaming of channels.
\end{itemize}
We show by induction on the derivations that
\begin{enumerate}
\item\label{evalpos1} if $E_1, \sigma_1, M_1, \tblcts_1, \evseq_1$ is ok, then all configurations $E, \sigma, M, \tblcts, \evseq$ in the derivation of $E_1, \ab \sigma_1, \ab M_1, \ab \tblcts_1, \ab \evseq_1 \red{p}{\ix}
E_2, \sigma_2, M_2, \tblcts_2, \evseq_2$ are ok;

\item\label{evalpos2} if $E_1, \pset_1, \cset_1$ is ok, then all configurations $E, \sigma, M, \tblcts, \evseq$ in the derivation of $E_1, \ab \pset_1, \ab \cset_1 \redq E_2, \pset_2, \cset_2$ are ok and $E_2, \pset_2, \cset_2$ is ok;

\item\label{evalpos2bis} if $E_1, \pset_1, \cset_1$ is ok, then $\reduce(E_1, \pset_1, \cset_1)$ is ok and all configurations $E, \sigma, M, \tblcts, \evseq$ or $E,\pset,\cset$ and in the derivation of $E_1, \pset_1, \cset_1 \redq^* \reduce(E_1, \pset_1, \cset_1)$ are ok;

\item\label{evalpos3} if $E_1, (\sigma_1, P_1), \pset_1, \cset_1, \tblcts_1, \evseq_1$ is ok, then all configurations $E, \sigma, M, \tblcts, \evseq$
or $E,\pset,\cset$ or $E, \ab (\sigma, P)\restconfig$ in the derivation of
$E_1, \ab (\sigma_1, P_1), \ab \pset_1, \ab \cset_1, \ab \tblcts_1, \ab \evseq_1 \red{p}{\ix}
E_2, \ab (\sigma_2, P_2), \ab \pset_2, \ab \cset_2, \ab \tblcts_2, \ab \evseq_2$ are ok.

\end{enumerate}
Property~\ref{evalpos1}: In~\eqref{sem:replindex}, \eqref{sem:var}, \eqref{sem:fun}, and~\eqref{sem:eventabortt}, all terms are in evaluation position, so all configurations are ok.
In~\eqref{sem:newt}, \eqref{sem:lett}, \eqref{sem:ift1}, \eqref{sem:ift2}, \eqref{sem:insertt}, and~\eqref{sem:eventt}, $N$ (resp. $N'$) is not in evaluation position, so by hypothesis it is a subterm of $Q_0$. Therefore, the target configuration is ok.
In the rules for $\FIND$, the recursive calls are on $D_j \wedge M_j$ which is not in evaluation position, so it is a subterm of $Q_0$. Therefore, the initial configurations of the recursive calls are ok, and we conclude for the configurations inside the recursive calls by induction hypothesis. The target configuration is ok because
in~\eqref{sem:findte} and~\eqref{sem:findt3}, the resulting term is in evaluation position (it has no subterm), and in~\eqref{sem:findt1} and~\eqref{sem:findt2}, the resulting term is a term that is not evaluation position in the initial configuration, so it is a subterm of $Q_0$. 
In~\eqref{sem:definedno} and~\eqref{sem:definedyes}, the terms $M_1$, \dots, $M_l$ are not in evaluation position in the initial configuration, so they are subterms of $Q_0$. Hence the initial configurations of the recursive calls are ok, and we conclude for the configurations inside the recursive calls by induction hypothesis.
The target configuration is ok because in~\eqref{sem:definedno}, $\false$ is in evaluation position (it has no subterm) and in~\eqref{sem:definedyes}, the resulting term $M$ is a term that is not evaluation position in the initial configuration, so it is a subterm of $Q_0$. 
The case of $\GET$ is similar to the one of $\FIND$.
In~\eqref{sem:ctxt}, the initial configuration of the recursive call is ok because terms that are not in evaluation position in $N$ are also not in evaluation position in $C[N]$, so they are subterms of $Q_0$. We conclude for the configurations inside the recursive call by induction hypothesis. The target configuration is ok because terms that are not in evaluation position in $C[N']$ are either not in evaluation position inside $C$, in which case they are subterms of $Q_0$ because the initial configuration is ok, or they are not in evaluation position in $N'$, in which case they are also subterms of $Q_0$ because the target configuration of the recursive call is ok. 
In~\eqref{sem:ctxeventt}, the target configuration is ok because the term is in evaluation position (it has no subterm).

Property~\ref{evalpos2}: The desired property is preserved for unmodified processes in $\pset$. This is enough to conclude for~\eqref{sem:nil}.
For~\eqref{sem:par} and~\eqref{sem:repl}, the resulting processes $Q_1$, $Q_2$, $Q$ are not in evaluation position in the initial configuration, so they are subprocesses of $Q_0$ up to renaming of channels.
For~\eqref{sem:newchannel}, $Q$ is not in evaluation position in the initial configuration, so it is a subprocesses of $Q_0$ up to renaming of channels, and so is $Q\{c'/c\}$.
For~\eqref{sem:input}, the subterms of $M_j$ that are not in evaluation position are not in evaluation position in $C[M_j]$, so they are subterms of $Q_0$. We can then apply Property~\ref{evalpos2} for the recursive call, so all configurations in the derivation of the recursive call are ok. The target configuration is ok because the subterms or subprocesses of $C[M'_j]$ not in evaluation position are either subterms of $M'_j$ not in evaluation position, which are subterms of $Q_0$ since the target configuration is ok, or subterms of subprocesses of $C$ not in evaluation position, which are subterms or subprocesses of $Q_0$ up to renaming of channels.

Property~\ref{evalpos2bis} follows immediately from Property~\ref{evalpos2} by induction.

Property~\ref{evalpos3}: In~\eqref{sem:output}, $Q''$ is not in evaluation position in the initial configuration, so it is a subprocess of $Q_0$ up to renaming of channels. Therefore, the configuration $E, \{ (\sigma, Q'')\}, \cset$ is ok.
By Property~\ref{evalpos2bis}, all configurations in the computation of $E, \pset', \cset'$ and $E, \pset', \cset'$ itself are ok.
Moreover, $P$ is not in evaluation position in the initial configuration, so it is a subprocess of $Q_0$ up to renaming of channels. We can then conclude that the target configuration is ok.
All other cases can be treated similarly to terms in Property~\ref{evalpos1}.

Moreover, in the computation of $\initconfig(Q_0)$, the configuration $\emptyset, \{ (\sigma_0, Q_0) \}, \fc(Q)$ is ok, so by Property~\ref{evalpos2bis}, $\initconfig(Q_0)$ is ok except for the process $0$ that follows $\coutput{\startch}{}$ in $\initconfig(Q_0)$. That process disappears in the first reduction, which is by~\eqref{sem:output}. 
  \proofcomplete
\end{proof}

\begin{corollary}\label{cor:evalpos}
Consider a trace $\trace$ of $Q_0$.

In $\trace$, the target process of rules~\eqref{sem:new},
\eqref{sem:let}, \eqref{sem:if1}, \eqref{sem:if2}, \eqref{sem:find1}, \eqref{sem:find2}, \eqref{sem:insert}, \eqref{sem:get1}, \eqref{sem:get2}, \eqref{sem:output}, \eqref{sem:event} is a subprocess of $Q_0$ up to renaming of channels.

\sloppy

In $\trace$, the target term of rules~\eqref{sem:newt},
\eqref{sem:lett}, \eqref{sem:ift1}, \eqref{sem:ift2}, \eqref{sem:findt1}, \eqref{sem:findt2}, \eqref{sem:insertt}, \eqref{sem:gett1}, \eqref{sem:gett2}, \eqref{sem:eventt}, \eqref{sem:definedyes} is a subterm $Q_0$.
\end{corollary}
\begin{proof}
The target term or process of these rules appears in non-evaluation position
in the initial configuration of these rules, so by Lemma~\ref{lem:evalpos},
it is a subterm of $Q_0$ or a subprocess of $Q_0$ up to renaming of channels.
\proofcomplete
\end{proof}

%When a program point $\pp$ occurs for the first time in reduction position, $\pptag N$ (resp. $\pptag P$) is a subterm (subprocess) of $Q_0$. Then there are 0 or more context rules that preserve $\pp$, and finally $\pp$ itself is executed and disappears. (For events, that's the point at which the event is added to $\evseq$.)

We say that a configuration $\conf$ \emph{is at program point} $\pp$ in a trace $\trace$ of $Q_0$ when $\conf$ occurs in the derivation of $\trace$, and
either $\conf = E, \sigma, \pptag N, \tblcts, \evseq$ for some subterm $\pptag N$ of $Q_0$
or $\conf = E, (\sigma, \pptag P)\restconfig$ for some subprocess $\pptag P$ of $Q_0$ up to renaming of channels.

\begin{lemma}\label{lem:condfind}
Let $\trace$ be a trace of $Q_0$.
Let $\conf$ be a configuration at program point $\pp$ in $\trace$.
If $\pp$ is not inside a condition of $\FIND$ or $\GET$, 
then $\conf$ is not in the derivation of an hypothesis of a rule for $\FIND$ or $\GET$ inside the derivation of $\trace$.
\end{lemma}
\begin{proof}
The only rules that can conclude with a term
$\FIND$ or $\GET$ are \eqref{sem:newt},
\eqref{sem:lett}, \eqref{sem:ift1}, \eqref{sem:ift2}, \eqref{sem:findt1}, \eqref{sem:findt2}, \eqref{sem:insertt}, \eqref{sem:gett1}, \eqref{sem:gett2}, \eqref{sem:eventt}, \eqref{sem:definedyes} and, by Corollary~\ref{cor:evalpos}, their target term is a subterm of $Q_0$.
The situation is similar $\FIND$ and $\GET$ processes. So the executed term or process in the initial configuration of rules for $\FIND$ and $\GET$ is always a subterm or subprocess of $Q_0$ up to renaming of channels. When a configuration $\conf$ is in the derivation of an hypothesis of a rule for $\FIND$ or $\GET$, it therefore always deals with program points syntactically in conditions of $\FIND$ or $\GET$ in $Q_0$. (The semantic rules do not create program points.) Since $\pp$ is not inside a condition of $\FIND$ or $\GET$, we conclude that $\conf$ is not in the derivation of an hypothesis of a rule for $\FIND$ or $\GET$.
\proofcomplete
\end{proof}

Given a trace $\trace$, we define a partial ordering relation $\beforetr{\trace}$ (reflexive, transitive, antisymmetric) on the occurrences of configurations
in the derivation of $\trace$: if $\conf_1 \red{p}{\ix} \conf_2$ occurs in the derivation of $\trace$, then $\conf_1 \beforetr{\trace} \conf_2$
and for all $\conf$ that occur in the derivation of the assumptions of $\conf_1 \red{p}{\ix} \conf_2$, $\conf_1 \beforetr{\trace} \conf \beforetr{\trace} \conf_2$,
and similarly for $\redq$ instead of $\red{p}{\ix}$.
\bbnote{For processes without $\GET$, there is a single configuration at program point $\pp$ with replication indices $\image(\sigma) = \tup{a}$ inside the derivation of a trace. With $\GET$, we need to consider the occurrence of a configuration inside the trace.}%
If $\trace$ is a trace of $Q_0$, then for all $\conf$ that occur in the derivation of $\emptyset, \{ (\sigma_0, Q_0) \}, \fc(Q_0) \redq^* \emptyset, \pset, \cset$,
we have $\conf \beforetr{\trace} \initconfig(Q_0)$.
When $\conf_1 \beforetr{\trace} \conf_2$, we say that $\conf_1$ \emph{occurs before} $\conf_2$ in $\trace$,
or equivalently, that  $\conf_2$ \emph{occurs after} $\conf_1$ in $\trace$.

We say that that a configuration  $\conf = E, \sigma, \pptag N, \tblcts, \evseq$, $\conf = E, (\sigma, \pptag P)\restconfig$, $\conf = E, (\sigma, {}^{\pp'} P)\restconfig$ with $(\sigma', \pptag Q) \in \pset$, or  $\conf = E, \pset, \cset$ with $(\sigma', \pptag Q) \in \pset$ \emph{is inside program point} $\pp$. A configuration may be inside several program points: $\conf = E, \pset, \cset$ is inside the program points of the input processes in $\cset$, $\conf = E, (\sigma, P)\restconfig$ is inside the program points of one output process ($P$) as well as the input processes in $\cset$.

We say that the program point $\pp$ is immediately above the program
points $\pp_j$ in a process $Q_0$ when $Q_0$ contains one of the
following constructs:
{\allowdisplaybreaks\begin{align*}
&\pptag x[{}^{\pp_1} M_1, \ldots, {}^{\pp_m} M_m]\\
&\pptag f({}^{\pp_1} M_1, \ldots, {}^{\pp_m} M_m)\\
&\pptag \Res{x[\tup{i}]}{T}; {}^{\pp_1} N\\
&\pptag \assign{x[\tup{i}]:T}{{}^{\pp_1} M}{{}^{\pp_2} N}\\
&\pptag \bguard{{}^{\pp_1} M}{{}^{\pp_2} N}{{}^{\pp_3} N'}\\
%\entry{\cfind{j=1}{m}{
%\vf_{j1}[\tup{i}] \leq n_{j1}, \ldots, \vf_{jm_j}[\tup{i}] \leq n_{jm_j}}{
%M_{j1}, \ldots, M_{jl_j}}{M_j}{P_j}{P}}{}\\
&\pptag \kw{find}\uniqueopt\ (\mathop\bigoplus\nolimits_{j=1}^m
    \vf_{j1}[\tup{i}] = i_{j1} \leq n_{j1}, \ldots, \vf_{jm_j}[\tup{i}] = i_{jm_j} \leq n_{jm_j}\ \kw{suchthat}\\
&\quad\kw{defined}({}^{\pp_{j,1}}M_{j1}, \ldots, {}^{\pp_{j,l_j}}M_{jl_j}) \fand {}^{\pp_{j,l_j+1}}M'_j\ \kw{then}\ {}^{\pp_{j,l_j+2}}N_j)\ \ELSE {}^{\pp_0} N'\\
&\pptag \INSERT\ \tbl({}^{\pp_1} M_1,\ldots,{}^{\pp_l} M_l);{}^{\pp_0} N\\
&\pptag \GET\uniqueopt\ \tbl(x_1[\tup{i}]:T_1,\ldots,x_l[\tup{i}]:T_l)\ \SUCHTHAT\ {}^{\pp_1} M\ \IN\ {}^{\pp_2} N\ \ELSE {}^{\pp_3} N'\\
&\pptag \kevent{e({}^{\pp_1} M_1, \ldots, {}^{\pp_l} M_l)}; {}^{\pp_0} N\\
&\pptag ({}^{\pp_1} Q \parpop {}^{\pp_2} Q')\\
&\pptag \repl{i}{n}{{}^{\pp_1} Q}\\
&\pptag \Reschan{c}; {}^{\pp_1} Q\\
&\pptag \cinput{c[{}^{\pp_1} M_1, \ldots, {}^{\pp_l} M_l]}{\pat}; {}^{\pp_0} P\\
&\pptag \coutput{c[{}^{\pp_1} M_1, \ldots, {}^{\pp_l} M_l]}{{}^{\pp_0} N}; Q\\
&\pptag \Res{x[\tup{i}]}{T}; {}^{\pp_1} P\\
&\pptag \assign{x[\tup{i}]}{{}^{\pp_1} M}{{}^{\pp_2}P}\\
&\pptag \bguard{{}^{\pp_1} M}{{}^{\pp_2} P}{{}^{\pp_3} P'}\\
&\pptag \FIND\uniqueopt \ (\mathop\bigoplus\nolimits_{j=1}^m \vf_{j1}[\tup{i}] = i_{j1} \leq n_{j1}, \ldots, \vf_{jm_j}[\tup{i}] = i_{jm_j} \leq n_{jm_j}\ \SUCHTHAT\\
&\quad \defined ({}^{\pp_{j,1}}M_{j1}, \ldots, {}^{\pp_{j,l_j}}M_{jl_j}) \fand {}^{\pp_{j,l_j+1}}M_j \THEN {}^{\pp_{j,l_j+2}}P_j)\ \ELSE {}^{\pp_0} P\\
&\pptag \INSERT\ \tbl({}^{\pp_1} M_1,\ldots,{}^{\pp_l} M_l);{}^{\pp_0} P\\
&\pptag \GET\uniqueopt\ \tbl(x_1[\tup{i}]:T_1,\ldots,x_l[\tup{i}]:T_l)\ \SUCHTHAT\ {}^{\pp_1} M\ \IN\ {}^{\pp_2} P\ \ELSE {}^{\pp_3} P'\\
&\pptag \kevent{e({}^{\pp_1} M_1, \ldots, {}^{\pp_l} M_l)}; {}^{\pp_0} P
\end{align*}}%
The relation ``$\pp$ is above $\pp'$'' is the reflexive and transitive
closure of ``$\pp$ is immediately above $\pp'$''.

\begin{lemma}\label{lem:above_before}
  Let $\trace$ be a trace of $Q_0$. 
  If $\conf$ is a configuration in $\trace$ inside program point $\pp$ and $\pp$ is a program point in $Q_0$\bbnote{to exclude the start program point in $\initconfig(Q_0)$}, then
  either $\conf$ is at the program point $\pp$ at the top of $Q_0$,
  or there exists a configuration $\conf'\neq \conf$ such that
  $\conf' \beforetr{\trace} \conf$ and $\conf'$ is inside program point
  $\pp$ or inside the program point $\pp'$ immediately above $\pp$ in $Q_0$.
  
%%   the initial configuration of the rule is either inside program point
%%   $\pp$ or inside the program point $\pp'$ just above $\pp$
%%   (except for $\defined$ conditions of $\FIND$, for which there is
%%   an intermediate rule). And except for output!!! For Ctx, that's true
%%   for the target conf. of the conclusion, for the initial conf. of the
%%   assumption, but not for the target conf. of the assumption.
%%   That's general.

  As a consequence, if a configuration $\conf$ inside program point
  $\pp$ in $Q_0$ is in $\trace$, then there are configurations
  inside all program points above $\pp$ in $Q_0$ before $\conf$ in $\trace$.

\end{lemma}
\begin{proof}
  First property. Since $\conf$ is a configuration in $\trace$, we are in one of the following cases:
  \begin{itemize}
  \item $\conf = \emptyset, \{ (\sigma_0, Q_0) \}, \fc(Q_0)$, the very first configuration of $\trace$. The configuration $\conf$ is at the program point $\pp$ at the top of $Q_0$.

  \item $\conf = \initconfig(Q_0) = \emptyset, \ab (\sigma_0, {}^{\pp''}\coutput{\startch}{}), \ab \pset, \ab \cset, \ab \emptyset, \ab \emptyset$. Since $\pp'' $ is not in $Q_0$, $\pp$ is the program point of a process in $\pset$. Then $\conf' = \emptyset, \pset, \cset$ is also inside $\pp$ and $\conf' \beforetr{\trace} \conf$, by definition of $\beforetr{\trace}$.

  \item $\conf$ is the initial configuration of an assumption of a semantic rule. In rules for $\FIND$, the initial configuration of assumption is not inside a program point (because it evaluates the $\defined$ condition, not a term). In rules for $\GET$, \eqref{sem:ctxt}, \eqref{sem:ctx}, and \eqref{sem:input}, the initial configuration of the conclusion is inside the program point $\pp'$ immediately above $\pp$. In rules~\eqref{sem:definedno} and~\eqref{sem:definedyes}, these rules are used to conclude assumptions of $\FIND$, and the initial configuration of the $\FIND$ rule is inside the program point $\pp'$ immediately above $\pp$. In rule~\eqref{sem:output}, $\conf = E, \{(\sigma, Q'')\}, \cset$, and the initial configuration of the conclusion of the rule~\eqref{sem:output} is inside the program point $\pp'$ immediately above $\pp$.

  \item $\conf$ is the target configuration of a semantic rule. In rules~\eqref{sem:newt}, \eqref{sem:lett}, \eqref{sem:ift1}, \eqref{sem:ift2}, \eqref{sem:findt1}, \eqref{sem:findt2}, \eqref{sem:insertt}, \eqref{sem:gett1}, \eqref{sem:gett2}, and~\eqref{sem:eventt}, the initial configuration of the rule is inside the program point $\pp'$ immediately above $\pp$.
    In rule~\eqref{sem:ctxt}, the initial configuration of the rule is inside the same program point $\pp$.
    In rule~\eqref{sem:definedyes}, this rule is used to conclude assumptions of $\FIND$, and the initial configuration of the $\FIND$ rule is inside the program point $\pp'$ immediately above $\pp$.
    In the rules for input processes, if $\pp$ is the program point of an unchanged element of $\pset$, the initial configuration of the rule is also inside $\pp$. This is sufficient for~\eqref{sem:nil}. If $\pp$ is the program point of a modified process, then for rules~\eqref{sem:par}, \eqref{sem:repl}, and~\eqref{sem:newchannel}, the initial configuration of the rule is inside the program point $\pp'$ immediately above $\pp$, and for rule~\eqref{sem:input}, the initial configuration of the rule is inside the same program point $\pp$.
    In the rules for output processes other than~\eqref{sem:output}, $\pset$ is unchanged, so if $\pp$ is the program point of a process in $\pset$, then the initial configuration of the rule is inside the same program point $\pp$.
    In rules~\eqref{sem:new}, \eqref{sem:let}, \eqref{sem:if1}, \eqref{sem:if2}, \eqref{sem:find1}, \eqref{sem:find2}, \eqref{sem:insert}, \eqref{sem:get1}, \eqref{sem:get2}, and~\eqref{sem:event}, if $\pp$ is the program point of an output process, then the initial configuration of the rule is inside the program point $\pp'$ immediately above $\pp$.
    In rule~\eqref{sem:ctx}, the initial configuration of the rule is inside the same program point $\pp$.
    In rule~\eqref{sem:output}, if $\pp$ is the program point of a process in $\pset$, then the initial configuration of the rule is inside the same program point $\pp$. If $\pp$ is the program point of a process in $\pset'$, then the configuration $E, \pset', \cset'$ in the assumption of the rule is inside the same program point $\pp$. If $\pp$ is the program point of $P$, then the initial configuration of the rule is inside the program point $\pp'$ immediately above $\pp$, because $(\sigma', Q_0) \in S \subseteq \pset$ and $\pp'$ is the program point of $Q_0$.

  \end{itemize}
  Second property. Suppose that $\conf$ is inside $\pp$ and there is a program point $\pp'$ immediately above $\pp$ in $Q_0$. Let us show that there exists $\conf' \beforetr{\trace} \conf$ such that $\conf'$ is inside $\pp'$. The proof proceeds by well-founded induction on $\beforetr{\trace}$. The program point $\pp$ is not at the top of $Q_0$, so by the first property, there exists $\conf'\neq \conf$ such that $\conf' \beforetr{\trace} \conf$ and $\conf'$ is inside $\pp$ or inside $\pp'$. If $\conf'$ is inside $\pp$, we conclude by applying the induction hypothesis on $\conf'$. If $\conf'$ is inside $\pp'$, we have the result.
  By applying this property repeatedly, we obtain the second property.
  \proofcomplete
\end{proof}

\begin{lemma}\label{lem:ctx}
  Let $\trace$ be a trace of $Q_0$.
  \begin{enumerate}
  \item\label{ctx1} If $\conf = E, \sigma, C_{m}[\dots C_{1}[\pptag M]\dots], \tblcts, \evseq$ is the target configuration of a semantic rule in $\trace$, where $\pp$ is a program point in
  $Q_0$, $C_{1}$, \dots, $C_{m}$ are term contexts defined in Figure~\ref{fig:termcontexts},
    and $\pptag M$ is not a subterm of $Q_0$, then the reduction that yields $\conf$ is obtained by $m$ applications of~\eqref{sem:ctxt} with contexts $C_m$, \dots, $C_1$ from a reduction with target configuration $E, \sigma, \pptag M, \tblcts, \evseq$, itself proved by~\eqref{sem:ctxt}.

  \item\label{ctx2} If $\conf =  E, (\sigma, C_0[C_m[\dots C_{1}[\pptag M]\dots]]), \pset, \cset, \tblcts, \evseq$ is the target configuration of a semantic rule in $\trace$, where $\pp$ is a program point in
  $Q_0$, $C_0$ is a process context defined in Figure~\ref{fig:proccontexts}, $C_{1}$, \dots, $C_{m}$ are term contexts defined in Figure~\ref{fig:termcontexts}, and $\pptag M$ is not a subterm of $Q_0$, then the reduction that yields $\conf$ is obtained by one application of~\eqref{sem:ctx} with context $C_0$ and $m$ applications of~\eqref{sem:ctxt} with contexts $C_m$, \dots, $C_1$ from a reduction with target configuration $E, \sigma, \pptag M, \tblcts, \evseq$, itself proved by~\eqref{sem:ctxt}.

  \item\label{ctx3} Let $Q = C_0[C_m[\dots C_{1}[\pptag M]\dots]]$
    where $\pp$ is a program point in $Q_0$, $C_0$ is an input
    context, $C_1$, \dots, $C_m$ are term contexts defined in
    Figure~\ref{fig:termcontexts}, and $\pptag M$ is not a subterm of
    $Q_0$. If $\conf = E, \{ (\sigma, Q \} \multiunion \pset, \cset$ is target configuration of a semantic rule that affects $Q$ in $\trace$, then the reduction that yields $\conf$ is obtained by one application of~\eqref{sem:input} with context $C_0$ and $m$ applications of~\eqref{sem:ctxt} with contexts $C_m$, \dots, $C_1$ from a reduction with target configuration $E, \sigma, \pptag M, \tblcts, \evseq$, itself proved by~\eqref{sem:ctxt}.

  \end{enumerate}
\end{lemma}
\begin{proof}
  Property~\ref{ctx1}. This property is proved by induction of $m$.
  The reduction that yields $\conf$ cannot be obtained by~\eqref{sem:newt}, \eqref{sem:lett}, \eqref{sem:ift1}, \eqref{sem:ift2}, \eqref{sem:findt1}, \eqref{sem:findt2}, \eqref{sem:insertt}, \eqref{sem:gett1}, \eqref{sem:gett2}, \eqref{sem:eventt}, \eqref{sem:definedyes}, because in this case, by Corollary~\ref{cor:evalpos}, $C_{m}[\dots C_{1}[\pptag M]\dots]$  would be a subterm of $Q_0$, so $\pptag M$ would be a subterm of $Q_0$. So it is obtained by~\eqref{sem:ctxt}.
  For $m = 0$, this is enough to conclude. 
  For $m > 0$, the rule~\eqref{sem:ctxt} is applied with context $C_m$. Indeed, since $\pptag M$ is not a value, the hole of the context cannot be after $C_{m-1}[\dots C_1[\pptag M]\dots]$ and, since $\pptag M$ is not a subterm of $Q_0$, the hole of the context cannot be before $C_{m-1}[\dots C_1[\pptag M]\dots]$. (If it were before, $C_{m-1}[\dots C_1[\pptag M]\dots]$ would not be in evaluation position in the initial configuration of rule~\eqref{sem:ctxt}, so by Lemma~\ref{lem:evalpos}, $C_{m-1}[\dots C_1[\pptag M]\dots]$ would be a subterm of $Q_0$.) Hence the reduction that yields $\conf$ is obtained from a reduction that yields $E, \sigma, C_{m-1}[\dots C_{1}[\pptag M]\dots], \tblcts, \evseq$ by applying~\eqref{sem:ctxt} with context $C_m$. We conclude by applying the induction hypothesis.

  Property~\ref{ctx2}. The reduction that yields $\conf$ cannot be obtained by~\eqref{sem:new}, \eqref{sem:let}, \eqref{sem:if1}, \eqref{sem:if2}, \eqref{sem:find1}, \eqref{sem:find2}, \eqref{sem:insert}, \eqref{sem:get1}, \eqref{sem:get2}, \eqref{sem:output}, or~\eqref{sem:event}, because in this case, by Corollary~\ref{cor:evalpos}, the process of $\conf$ would be a subprocess of $Q_0$ up to renaming of channels, so $\pptag M$ would be a subterm of $Q_0$. Therefore, $\conf$ is the target configuration of~\eqref{sem:ctx}. Furthermore, by the same reasoning as in Property~\ref{ctx1}, this rule is applied with context $C_0$. Hence the reduction that yields $\conf$ is obtained from a reduction that yields $E, \sigma, C_m[\dots C_1[\pptag M]\dots], \tblcts, \evseq$ by applying~\eqref{sem:ctx} with context $C_0$. We conclude by Property~\ref{ctx1}.

\sloppy
  
  Property~\ref{ctx3}. The reduction that yields $\conf$ cannot be obtained by~\eqref{sem:par}, \eqref{sem:repl}, or~\eqref{sem:newchannel}, because in this case, by Lemma~\ref{lem:evalpos}, since $Q$ occurs in non-evaluation position in the initial configuration of the rule, $Q$ would be a subprocess of $Q_0$ up to renaming of channels, so $\pptag M$ would be a subterm of $Q_0$. Therefore, $\conf$ is the target configuration of~\eqref{sem:input}. Furthermore, by the same reasoning as in Property~\ref{ctx1}, this rule is applied with context $C_0$. Hence the reduction that yields $\conf$ is obtained from a reduction that yields $E, \sigma, C_m[\dots C_1[\pptag M]\dots], \tblcts, \evseq$ by applying~\eqref{sem:input} with context $C_0$. We conclude by Property~\ref{ctx1}.
  \proofcomplete
\end{proof}

\begin{lemma}\label{lem:start_from_pp}
  Let $\trace$ be a trace of $Q_0$.
  \begin{enumerate}
  \item\label{start_from_pp1}  If $\conf =  E, (\sigma, \pptag P), \pset, \cset, \tblcts, \evseq$ is a configuration in $\trace$ 
  where $\pp$ is a program point in
  $Q_0$, then this configuration is derived in $\trace$ from a configuration
  $\conf' = E', \ab (\sigma, \pptag P'), \ab \pset, \ab \cset, \ab \tblcts', \ab \evseq'$
  where $\pptag P'$ is a subprocess of $Q_0$ up to renaming of channels,
  by~\eqref{sem:ctx} any number of times.
  
  \item\label{start_from_pp2}  If $\conf = E, \{ (\sigma, \pptag Q) \} \multiunion \pset, \cset$ is a configuration in $\trace$ where $\pp$ is a program point in
  $Q_0$, then, possibly after swapping reductions in $\trace$, this configuration is derived from a configuration
$\conf' = E, \{ (\sigma, \pptag Q') \} \multiunion \pset, \cset$
  where $\pptag Q'$ is a subprocess of $Q_0$ up to renaming of channels,
  by~\eqref{sem:input} any number of times.

  \item\label{start_from_pp3} If $\conf =  E, (\sigma, C_0[C_m[\dots C_{1}[\pptag M]\dots]]), \pset, \cset, \tblcts, \evseq$ is a configuration in $\trace$ 
  where $\pp$ is a program point in
  $Q_0$, $C_0$ is a process context defined in Figure~\ref{fig:proccontexts}, and $C_1$, \dots, $C_{m}$ are term contexts defined in Figure~\ref{fig:termcontexts}, then we have
  \[E', \sigma, \pptag M', \tblcts', \evseq' \red{p}{\ix} \dots \red{p'}{\ix'}E, \sigma, \pptag M, \tblcts, \evseq\]
  by~\eqref{sem:ctxt} any number of times, where $\pptag M'$ is a subterm of $Q_0$ and in $\trace$, these reductions are in fact performed 
starting from 
  $\conf' = E', \ab (\sigma, C_0[C_m[\dots C_{1}[\pptag M']\dots]]), \ab \pset, \ab \cset, \ab \tblcts', \ab \evseq'$
  under one application of~\eqref{sem:ctx} with context $C_0$ and $m$ applications of~\eqref{sem:ctxt} with contexts $C_m$, \dots, $C_{1}$.

  \item\label{start_from_pp4} If $\conf = E, \{ (\sigma, C_0[C_m[\dots C_{1}[\pptag M]\dots]]) \} \multiunion \pset, \cset$ is a configuration in $\trace$, where $\pp$ is a program point in
  $Q_0$, $C_0$ is an input context, and
  $C_1$, \dots, $C_m$ are term contexts defined in Figure~\ref{fig:termcontexts}, then we have
  \[E, \sigma, \pptag M', \emptyset, \emptyset \red{p}{\ix} \dots \red{p'}{\ix'}E, \sigma, \pptag M, \emptyset, \emptyset\]
  by~\eqref{sem:ctxt} any number of times, where $\pptag M'$ is a subterm of $Q_0$ and possibly after swapping reductions in $\trace$, these reductions are in fact performed
  starting from a configuration
$\conf' = E, \{ (\sigma, C_0[C_m[\dots C_{1}[\pptag M']\dots]]) \} \multiunion \pset, \cset$
  under one application of~\eqref{sem:input} with context $C_0$ and $m$ applications of~\eqref{sem:ctxt} with contexts $C_m$, \dots, $C_{1}$.

  \item\label{start_from_pp5} If $\conf = E, \sigma, C_{l}[\dots C_{1}[\pptag M]\dots], \tblcts, \evseq$ is a configuration in $\trace$ 
  and $\pp$ is a program point in
  $Q_0$ where $C_{1}$, \dots, $C_{l}$ are term contexts defined in Figure~\ref{fig:termcontexts}, then we have
  \[E', \sigma, \pptag M', \tblcts', \evseq' \red{p}{\ix} \dots \red{p'}{\ix'}E, \sigma, \pptag M, \tblcts, \evseq\]
  by~\eqref{sem:ctxt} any number of times, where $\pptag M'$ is a subterm of $Q_0$ and in $\trace$, these reductions are in fact performed starting  
\begin{itemize}
\item from a configuration
  $\conf' = E', \sigma, C_m[\dots C_{1}[\pptag M']\dots], \tblcts', \evseq'$
  where $m \geq l$,
  $C_1$, \dots, $C_m$ are term contexts defined in Figure~\ref{fig:termcontexts},
  under $m$ applications of~\eqref{sem:ctxt} with contexts $C_m$, \dots, $C_{1}$.
\item \sloppy or from a configuration
  $\conf' = E', (\sigma, C_0[C_m[\dots C_{1}[\pptag M']\dots]]), \pset, \cset, \tblcts', \evseq'$
  where $m \geq l$, $C_0$ is a process context defined in Figure~\ref{fig:proccontexts}, and
  $C_1$, \dots, $C_l$ are term contexts defined in Figure~\ref{fig:termcontexts},
  under one application of~\eqref{sem:ctx} with context $C_0$ and $m$ applications of~\eqref{sem:ctxt} with contexts $C_m$, \dots, $C_{1}$.
\item or, possibly after swapping reductions in $\trace$, from a configuration
  $\conf' = E', \ab \{ (\sigma, \ab C_0[C_m[\dots C_{1}[\pptag M']\dots]]) \} \multiunion \pset, \ab \cset$
  where $m \geq l$,
  $C_0$ is an input context and
  $C_1$, \dots, $C_m$ are term contexts defined in Figure~\ref{fig:termcontexts},
  under one application of~\eqref{sem:input} with context $C_0$ and $m$ applications of~\eqref{sem:ctxt} with contexts $C_m$, \dots, $C_{1}$.
  
\end{itemize}
\end{enumerate}
\end{lemma}
\begin{proof}
  Property~\ref{start_from_pp1}. The proof proceeds by well-founded induction on $\beforetr{\trace}$. The configuration $\conf$ cannot be $\initconfig(Q_0)$ because $\pp$ is a program point in $Q_0$. Then, $\conf$ is the target configuration of some semantic rule. If $\conf$ is the target configuration
  of~\eqref{sem:new}, \eqref{sem:let}, \eqref{sem:if1}, \eqref{sem:if2}, \eqref{sem:find1}, \eqref{sem:find2}, \eqref{sem:insert}, \eqref{sem:get1}, \eqref{sem:get2}, \eqref{sem:output}, or~\eqref{sem:event}, then by Corollary~\ref{cor:evalpos}, $\pptag P$ is a subprocess of $Q_0$ up to renaming of channels, so the property holds with $\conf' = \conf$. If it is the target configuration of~\eqref{sem:ctx}, we conclude by applying the induction hypothesis to the initial configuration of this rule.
  
  Property~\ref{start_from_pp2}. The proof proceeds by well-founded induction on $\beforetr{\trace}$. If $\conf = \emptyset, \ab \{ (\sigma_0, Q_0) \}, \ab \fc(Q_0)$, then the property holds with $\conf' = \conf$, since $Q_0$ is a subprocess of $Q_0$. If $\conf$ is the initial configuration of the assumption of~\eqref{sem:output}, then by Lemma~\ref{lem:evalpos}, $Q$ is a subprocess of $Q_0$ up to renaming of channels, since $Q$ occurs in non-evaluation position in the initial configuration of~\eqref{sem:output}. So the property holds with $\conf' = \conf$. Otherwise, $\conf$ is the target configuration of a semantic rule.
  \begin{itemize}
  \item If that rule does not affect $Q$, then it is of the form $E, \{ (\sigma, \pptag Q) \} \multiunion \pset', \cset' \redq E, \{ (\sigma, \pptag Q) \} \multiunion \pset, \cset$. We apply the induction hypothesis to $E, \{ (\sigma, \pptag Q) \} \multiunion \pset', \cset'$, so possibly after swapping reductions in $\trace$, we have $E, \{ (\sigma, \pptag Q') \} \multiunion \pset', \cset' \redq^* E, \{ (\sigma, \pptag Q) \} \multiunion \pset', \cset'$ by~\eqref{sem:input} any number of times, where $\pptag Q'$ is a subprocess of $Q_0$ up to renaming of channels. By swapping reductions, we have $E, \{ (\sigma, \pptag Q') \} \multiunion \pset', \cset' \redq E, \{ (\sigma, \pptag Q') \} \multiunion \pset, \cset \redq^*   E, \{ (\sigma, \pptag Q) \} \multiunion \pset, \cset$, so we obtain the desired property with $\conf' = E, \{ (\sigma, \pptag Q') \} \multiunion \pset, \cset$.
  \item Otherwise, that rule affects $Q$. If that rule is~\eqref{sem:par}, \eqref{sem:repl}, or~\eqref{sem:newchannel}, then by Lemma~\ref{lem:evalpos}, $Q$ is a subprocess of $Q_0$ up to renaming of channels, since $Q$ occurs in non-evaluation position in the initial configuration of the rule. So the property holds with $\conf' = \conf$. If that rule is~\eqref{sem:input}, then we obtain the result by induction hypothesis applied to the initial configuration of the rule. 
  \end{itemize}

  Property~\ref{start_from_pp3}. The proof proceeds by well-founded induction on $\beforetr{\trace}$. The configuration $\conf$ cannot be $\initconfig(Q_0)$ because $\pp$ is a program point in $Q_0$. Then, $\conf$ is the target configuration of some semantic rule. If $\pptag M$ is a subterm of $Q_0$, then the property holds with $\conf' = \conf$.
  Otherwise, by Lemma~\ref{lem:ctx}, Property~\ref{ctx2}, the reduction that yields $\conf$ is obtained by one application of~\eqref{sem:ctx} with context $C_0$ and $m$ applications of~\eqref{sem:ctxt} with contexts $C_m$, \dots, $C_1$ from a reduction with target configuration $E, \sigma, \pptag M, \tblcts, \evseq$, itself proved by~\eqref{sem:ctxt}. We obtain the desired property by applying the induction hypothesis to the configuration before this reduction step by~\eqref{sem:ctx}.

  Property~\ref{start_from_pp4}. Let $Q = C_0[C_m[\dots C_{1}[\pptag M]\dots]]$. The proof proceeds by well-founded induction on $\beforetr{\trace}$. 
  If $\conf = \emptyset, \ab \{ (\sigma_0, Q_0) \}, \ab \fc(Q_0)$, then the property holds with $\conf' = \conf$, since $Q_0$ is a subprocess of $Q_0$, so $\pptag M$ is a subterm of $Q_0$.
  If $\conf$ is the initial configuration of the assumption of~\eqref{sem:output}, then by Lemma~\ref{lem:evalpos}, $Q$ is a subprocess of $Q_0$ up to renaming of channels, since $Q$ occurs in non-evaluation position in the initial configuration of~\eqref{sem:output}. So $\pptag M$ is a subterm of $Q_0$ and the property holds with $\conf' = \conf$.
  Otherwise, $\conf$ is the target configuration of a semantic rule. If that rule does not affect $Q$, then we swap reductions as in the proof of Property~\ref{start_from_pp2}. Otherwise, if $\pptag M$ is a subterm of $Q_0$, then the property holds with $\conf' = \conf$. Otherwise, by Lemma~\ref{lem:ctx}, Property~\ref{ctx3}, the reduction that yields $\conf$ is obtained by one application of~\eqref{sem:input} with context $C_0$ and $m$ applications of~\eqref{sem:ctxt} with contexts $C_m$, \dots, $C_1$ from a reduction with target configuration $E, \sigma, \pptag M, \tblcts, \evseq$, itself proved by~\eqref{sem:ctxt}. We obtain the desired property by applying the induction hypothesis to the configuration before this reduction step by~\eqref{sem:input}.

  Property~\ref{start_from_pp5}. The proof proceeds by well-founded induction on $\beforetr{\trace}$. 
  First case: $\conf$ is the initial configuration of an assumption of a semantic rule.
  This rule cannot be a rule for $\FIND$ (in rules for $\FIND$, the initial configuration of assumption is not inside a program point, because it evaluates the $\defined$ condition, not a term).
  In rules for $\GET$, \eqref{sem:definedno}, and~\eqref{sem:definedyes}, by Lemma~\ref{lem:evalpos}, since the term in $\conf$ occurs in non-evaluation position in the initial configuration of the rule, it is a subterm of $Q_0$, so $\pptag M$ is a subterm of $Q_0$. The property holds with $\conf' = \conf$.
  In~\eqref{sem:ctxt}, we let $C_{l+1}$ be the context used in this rule, and we conclude by induction hypothesis applied to the initial configuration of this rule: $E, \sigma, C_{l+1}[C_{l}[\dots C_{1}[\pptag M]\dots]], \tblcts, \evseq$.
  In~\eqref{sem:input}, we let $m = l$. The initial configuration of the rule is of the form $E, \{ C_0[C_l[\dots C_{1}[\pptag M]\dots]] \} \multiunion \pset, \cset$ for some input context $C_0$. We conclude by applying Property~\ref{start_from_pp4} to that configuration.
  In~\eqref{sem:ctx}, we let $m = l$. The initial configuration of the rule is of the form $E, (\sigma, C_0[C_l[\dots C_{1}[\pptag M]\dots]]), \pset, \cset, \tblcts, \evseq$ for some process context $C_0$ defined in Figure~\ref{fig:proccontexts}. We conclude by applying Property~\ref{start_from_pp3} to that configuration.

  Second case: $\conf$ is the target configuration of a semantic rule.
  If $\pptag M$ is a subterm of $Q_0$, then the property holds with $\conf' = \conf$. 
  Otherwise, by Lemma~\ref{lem:ctx}, Property~\ref{ctx1}, the reduction that yields $\conf$ is obtained by $l$ applications of~\eqref{sem:ctxt} with contexts $C_{l}$, \dots, $C_1$ from a reduction with target configuration $E, \sigma, \pptag M, \tblcts, \evseq$, itself proved by~\eqref{sem:ctxt}. We apply the induction hypothesis to the initial configuration of the reduction that yields $\conf$, and we continue the reduction one more step by applying possibly~\eqref{sem:ctx} or~\eqref{sem:input} with context $C_0$, \eqref{sem:ctxt} $m$ times with contexts $C_m$, \dots, $C_1$ under the reduction with target configuration $E, \sigma, \pptag M, \tblcts, \evseq$. Up to swapping of reductions in the case of input processes, this is also what happens in $\trace$. (By inspection of the rules, only the rule~\eqref{sem:input} with context $C_0$ can reduce an input process $C_0[N]$, where $C_0$ is an input context; only the rule~\eqref{sem:ctx} with context $C_0$ can reduce an output process $C_0[N]$, where $C_0$ is a process context defined in Figure~\ref{fig:proccontexts}; only the rule~\eqref{sem:ctxt} with context $C_j$ can reduce a term $C_j[N]$, where $C_j$ is a term context defined in Figure~\ref{fig:termcontexts}.)
  \proofcomplete
\end{proof}

\subsubsection{Each Variable is Defined at Most Once}\label{sec:inv1exec}

In this section, we show that Invariant~\ref{inv1} implies that
each array cell is assigned at most once during the execution of
a process.

We define the multiset of variable accesses that may be defined by a term or 
a process (given the replication indices fixed by a mapping sequence $\sigma$)
as follows:
{\allowdisplaybreaks\begin{align*}
&\defset(\sigma, \pptag i) = \emptyset\\
&\defset(\sigma, \pptag x[M_1, \ldots, M_m]) =  \bigmultiunion_{j=1}^m \defset(\sigma, M_j)\\
&\defset(\sigma, \pptag f(M_1, \ldots, M_m)) =  \bigmultiunion_{j=1}^m \defset(\sigma, M_j)\\
&\defset(\sigma, \pptag \Res{x[\tup{i}]}{T}; N) = \{ x[\sigma(\tup{i})] \} 
\multiunion \defset(\sigma, N)\\
&\defset(\sigma, \pptag \assign{x[\tup{i}]:T}{M}{N}) = \{ x[\sigma(\tup{i})] \} 
\multiunion \defset(\sigma, M) \multiunion \defset(\sigma, N)\\
&\defset(\sigma, \pptag \bguard{M}{N}{N'}) = \defset(M) \multiunion \max(\defset(N), \defset(N'))\\
&\defset(\sigma, \pptag \FIND\uniqueopt \ (\mathop{\textstyle\bigoplus}\nolimits_{j=1}^m
\tup{\vf_j}[\tup{i}] = \tup{i_j} \leq \tup{n_j}\ \SUCHTHAT \ 
\defined(\tup{M_j}) \fand M_j \THEN N_j)\ \ELSE N) = \\*
&\qquad \max(\max_{j = 1}^m \max_{\tup{a} \leq \tup{n_j}} \defset(\sigma[\tup{i_j} \mapsto \tup{a}], M_j),
\max_{j = 1}^m \{ \tup{\vf_j}[\sigma(\tup{i})] \} \multiunion
\defset(\sigma, N_j), \defset(\sigma, N))\\
&\defset(\sigma, \pptag \INSERT\ \tbl(M_1, \ldots, M_l); N) = 
   \bigmultiunion_{j=1}^l \defset(\sigma, M_j) \multiunion \defset(\sigma, N)\\
&\defset(\sigma, \pptag \GET\uniqueopt\ \tbl(x_1[\tup{i}] : T_1,\ldots,x_l[\tup{i}] : T_l)\ \SUCHTHAT\ M\ \IN\ N\ \ELSE N') =\\*
&\qquad  \max(\{ x_j[\sigma(\tup{i})] \mid j \leq l \}
\multiunion \max(\defset(\sigma, M), \defset(\sigma, N)), \defset(\sigma, N'))\\
&\defset(\sigma, \pptag \kevent{e(M_1, \ldots, M_l)};N) =
   \bigmultiunion_{j=1}^l \defset(\sigma, M_j) \multiunion \defset(\sigma, N)\\
&\defset(\sigma, \pptag \keventabort{e}) = \emptyset\\
&\defset(\sigma, a) = \defset(\sigma, \keventabort{(\pp,\tup{a}):e}) = \emptyset\\
&\defset(\sigma, \pptag 0) = \emptyset\\
&\defset(\sigma, \pptag (Q_1 \parpop Q_2)) = \defset(\sigma, Q_1) \multiunion \defset(\sigma, Q_2)\\
&\defset(\sigma, \pptag \repl{i}{n} Q) = \bigmultiunion_{a \in [1, n]} \defset(\sigma[i \mapsto a], Q)\\
&\defset(\sigma, \pptag \Reschan{c};Q) = \defset(\sigma, Q)\\
&\defset(\sigma, \pptag \cinput{c[M_1, \ldots, M_l]}{x[\tup{i}]:T};P) = \{x[\sigma(\tup{i})]\}\multiunion \defset(\sigma, P)\\
&\defset(\sigma, \pptag \coutput{c[M_1, \ldots, M_l]}{N};Q) = \bigmultiunion_{j=1}^l \defset(\sigma, M_j) \multiunion \defset(\sigma, N) \multiunion \defset(\sigma, Q)\\
&\defset(\sigma, \pptag \Res{x[\tup{i}]}{T}; P) = \{ x[\sigma(\tup{i})] \} 
\multiunion \defset(\sigma, P)\\
&\defset(\sigma, \pptag \assign{x[\tup{i}]:T}{M}{P}) = \{ x[\sigma(\tup{i})] \} 
\multiunion \defset(\sigma, M)\multiunion \defset(\sigma, P)\\
&\defset(\sigma, \pptag \bguard{M}{P}{P'}) = \defset(M) \multiunion \max(\defset(P), \defset(P'))\\
&\defset(\sigma, \pptag \FIND\uniqueopt \ (\mathop{\textstyle\bigoplus}\nolimits_{j=1}^m
\tup{\vf_j}[\tup{i}] = \tup{i_j} \leq \tup{n_j}\ \SUCHTHAT \ 
\defined(\tup{M_j}) \fand M_j \THEN P_j)\ \ELSE P) = \\*
&\qquad \max(\max_{j = 1}^m \max_{\tup{a} \leq \tup{n_j}} \defset(\sigma[\tup{i_j} \mapsto \tup{a}], M_j), \max_{j = 1}^m \{ \tup{\vf_j}[\sigma(\tup{i})] \} \multiunion
\defset(\sigma, P_j), \defset(\sigma, P))\\
&\defset(\sigma, \pptag \INSERT\ \tbl(M_1, \ldots, M_l); P) = 
   \bigmultiunion_{j=1}^l \defset(\sigma, M_j) \multiunion \defset(\sigma, P)\\
&\defset(\sigma, \pptag \GET\uniqueopt\ \tbl(x_1[\tup{i}] : T_1,\ldots,x_l[\tup{i}] : T_l)\ \SUCHTHAT\ M\ \IN\ P\ \ELSE P') =\\*
&\qquad  \max(\{ x_j[\sigma(\tup{i})] \mid j \leq l \}
\multiunion \max(\defset(\sigma, M), \defset(\sigma, P)), \defset(\sigma, P'))\\
&\defset(\sigma, \pptag\kevent{e(M_1, \ldots, M_l)};P) =
   \bigmultiunion_{j=1}^l \defset(\sigma, M_j) \multiunion \defset(\sigma, P)\\
&\defset(\sigma, \pptag \keventabort{e}) = \emptyset\\
&\defset(\sigma, \kw{abort}) = \emptyset
\end{align*}}%
% $\defset(P)$, $\defset(M)$ ($\max$ for branches of if/find, 
% $\multiunion$ elsewhere).
Notice that, by Invariant~\ref{invtic}, the terms $M_j$ in channels of inputs and
the terms $\tup{M_j}$ in $\defined$ conditions of $\kw{find}$ do not define any variable.
By Invariant~\ref{invfc}, the variables defined in conditions of $\kw{find}$ and $\kw{get}$ can be considered as defined temporarily only during the evaluation of the considered condition.
Given a configuration $\conf = E, \sigma, N, \tblcts, \evseq$ or
$\conf = E, (\sigma, P)\restconfig$ or $\conf = E, \pset, \cset$, we denote by $E_{\conf}$ the
environment $E$ in configuration $\conf$.
We define
\begin{align*}
  &\defsetfut(E, \sigma, M, \tblcts, \evseq) = \defset(\sigma, M)\\
  &\defsetfut(E, \pset, \cset) = \bigmultiunion_{(\sigma, Q) \in \pset} \defset(\sigma, Q)\\
  &\defsetfut(E,(\sigma, P)\restconfig) = \defset(\sigma, P) \multiunion \bigmultiunion_{(\sigma, Q) \in \pset} \defset(\sigma, Q)\\
  &\defset(\conf) = \dom(E_{\conf}) \multiunion \defsetfut(\conf)\,.
\end{align*}

\begin{invariant}[Single definition, for executing games]
\label{inv1exec}
\ The semantic configuration $\conf$ (which can be $E, \sigma, M, \tblcts, \evseq$ or
$E, \pset, \cset$ or $E, (\sigma,\ab P)\restconfig$) satisfies 
Invariant~\ref{inv1exec} if and only if 
$\defset(\conf)$ does not contain duplicate elements.
\end{invariant}

\begin{lemma}\label{lem:inv1exec}
  Let $\trace$ be trace of $Q_0$.
  If $Q_0$ satisfies Invariant~\ref{inv1}, then all semantic configurations in the derivation
  of $\trace$
  satisfy Invariant~\ref{inv1exec}.
\end{lemma}
\begin{proofsk}
  We first show that, for all program points $\pp$ in $Q_0$, if
  $\dom(\sigma) = \replidx{\pp}$ are the current replication indices at $\pp$
  and the process or term $Q$ at $\pp$
  satisfies Invariant~\ref{inv1}, then all elements of
  $\defset(\sigma, Q)$ are of the form $x[\tup{a}]$ where $x \in
  \vardef(Q)$ and $\image(\sigma)$ is a prefix of $\tup{a}$.  The
  proof proceeds by induction on $Q$. At the definition of a variable
  $x[\tup{i}]$, $x[\sigma(\tup{i})]$ is added to $\defset(\sigma, Q)$
  and we have $x \in \vardef(Q)$; by Invariant~\ref{inv1}, $\tup{i}$
  are the current replication indices at that definition, so
  $\sigma(\tup{i}) = \image(\sigma)$.  All recursive calls
  $\defset(\sigma', {}^{\pp'} Q')$ consider an extension $\sigma'$ of $\sigma$
  and a subprocess or subterm $Q'$ of $Q$ (so $\vardef(Q') \subseteq
  \vardef(Q)$) such that $\dom(\sigma') = \replidx{\pp'}$ are the current replication
  indices at $\pp'$.

  Next, we show that, for all program points $\pp$, if
  $\dom(\sigma) = \replidx{\pp}$ are the current replication indices at $\pp$
  and the process or term $Q$ at $\pp$
  satisfies Invariant~\ref{inv1}, then $\defset(\sigma, Q)$ does not
  contain duplicate elements.
  The proof proceeds by induction on $Q$.
  All multiset unions in the computation of $\defset(\sigma, Q)$ are
  disjoint unions by the property above, because either they use
  different extensions of $\sigma$ (case of replication) or they use
  disjoint variable definitions or subprocesses or subterms in the
  same branch of $\FIND$, $\kw{if}$, or $\kw{get}$, which must define different
  variables by Invariant~\ref{inv1}.
  
  We show by induction on the derivations that, if
  $\conf \red{p}{\ix} \conf'$, then
  $\defset(\conf) \supseteq \defset(\conf')$
  and for all semantic configurations $\conf''$ in the derivation of
  $\conf \red{p}{\ix} \conf'$,
  $\defset(\conf) \supseteq \defset(\conf'')$, and similarly with $\redq$ instead of $\red{p}{\ix}$.

The result follows: since $Q_0$ satisfies Invariant~\ref{inv1},
$\defset(\sigma_0, Q_0)$ does not
contain duplicate elements, where $\sigma_0$ is the empty mapping sequence.
Then $\emptyset, \{ (\sigma_0, Q_0) \}, \fc(Q_0)$ satisfies Invariant~\ref{inv1exec}
and so do $\reduce(\emptyset, \{ (\sigma_0, Q_0) \}, \fc(Q_0))$, $\initconfig(Q_0)$,
and the other configurations of~$\trace$.
\proofcomplete
\end{proofsk}

\begin{corollary}\label{cor:inv1exec}
  If $Q_0$ satisfies Invariant~\ref{inv1}, then each variable that is
  not defined in a condition of $\FIND$ or $\GET$ is defined at most
  once for each value of its array indices in a trace of $Q_0$.
\end{corollary}
\begin{proof}
  Let $x$ be the considered variable and $\trace$ be the considered trace of $Q_0$.
The only semantic rules that can add $x[\tup{a}]$ to the environment $E$ are~\eqref{sem:newt}, \eqref{sem:lett}, \eqref{sem:findt1}, \eqref{sem:gett1}, \eqref{sem:new}, \eqref{sem:let}, \eqref{sem:find1}, \eqref{sem:get1}, and~\eqref{sem:output}.
By Corollary~\ref{cor:evalpos}, the target term or process of these rules is a subterm or subprocess of $Q_0$ up to renaming of channels.
Hence, the target configuration $\conf'$ of these rules is at some program point $\pp$ in $\trace$. 
By hypothesis, $x$ is not defined in conditions of $\FIND$ or $\GET$,
so $\pp$ is not inside the condition of $\FIND$ or $\GET$ in $Q_0$,
so by Lemma~\ref{lem:condfind}, the configuration $\conf'$ is not in 
the derivation of an assumption of a rule for $\FIND$ or $\GET$.

In order to derive a contradiction, assume that
two transitions $\conf_1 \red{p_1}{\ix_1} \conf'_1$ and $\conf_2 \red{p_2}{\ix_2} \conf'_2$
inside $\trace$ define the same variable $x[\tup{a}]$.
\begin{itemize}

\item First case: one transition happens before the other, for instance $\conf'_1 \beforetr{\trace} \conf_2$.
(The case $\conf'_2 \beforetr{\trace} \conf_1$ is symmetric.)
Since $\conf_1 \red{p_1}{\ix_1} \conf'_1$ defines $x[\tup{a}]$, we have $x[\tup{a}] \in \dom(E_{\conf'_1})$.
Since $\conf'_1$ is not in 
the derivation of an assumption of a rule for $\FIND$ or $\GET$,
by Lemma~\ref{lem:sem-ext}, $E_{\conf_2}$ extends $E_{\conf'_1}$, so $x[\tup{a}] \in \dom(E_{\conf_2})$.
Moreover, since $\conf_2 \red{p_2}{\ix_2} \conf'_2$ defines $x[\tup{a}]$, we have $x[\tup{a}] \in \defsetfut(\conf_2)$,
by inspecting all rules that add elements to the environment.
Therefore $\defset(\conf_2) = \dom(E_{\conf_2}) \multiunion \defsetfut(\conf_2)$ contains
twice $x[\tup{a}]$. Contradiction with Invariant~\ref{inv1exec}.

\item Second case: the transitions cannot be ordered. 
By definition of $\beforetr{\trace}$, this can happen only when a semantic rule uses several derivations
for its assumptions, which happens only in rules for $\FIND$ or $\GET$. 
Contradiction.

\end{itemize}
That concludes the proof.
\proofcomplete
\end{proof}

Variables defined in conditions of $\GET$ may be defined several times, once
for each element of the table that is tested. Variables defined in conditions
of $\FIND$ may be defined several times in case the same variable
is used in several branches of the same $\FIND$. (We use the indices
of the $\FIND$ as indices of the variables defined in the condition, so
when we evaluate several times the condition of a certain branch of a $\FIND$,
we use variables with different indices.)
In Section~\ref{sec:subsets}, we define properties that exclude these situations
(Properties~\ref{prop:notables} and~\ref{prop:autosarename}),
and we prove in Lemma~\ref{lem:stronginv1exec} that every variable is defined
at most once for each value of its indices when these properties are satisfied.

\subsubsection{Variables are Defined Before Being Used}

In this section, we show that Invariant~\ref{inv2} implies that
all variables are defined before being used. In order to show this
property, we use the following invariant:

\begin{invariant}[Defined variables, for executing games]
\label{inv2exec}
The semantic configuration $E,(\sigma,\ab P)\restconfig$ satisfies
Invariant~\ref{inv2exec} if and only if every occurrence of a variable
access $x[M_1, \ab \ldots, \ab M_m]$ in $(\sigma,P)$ or $\pset$ is either
\begin{enumerate}

\item present in $\dom(E)$: if $x[M_1, \ldots, M_m]$ occurs in
a process $P'$ for $(\sigma', P') \in \{ (\sigma, P) \} \cup \pset$, then
for all $j \leq m$, $E, \sigma', M_j \evalterm a_j$ and 
$x[a_1, \ldots, a_m] \in \dom(E)$;

\item or syntactically under the definition of $x[M_1, \ldots, M_m]$
(in which case for all $j \leq m$, $M_j$ is a constant or variable
replication index);

\item or in a $\defined$ condition in a $\FIND$ process or term;

\item or in $M'_j$ in a process or term of the form 
%$\cfind{j=1}{m''}{\tup{\vf_j}[\tup{i}] \leq \tup{n_j}}{
%M'_{j1}, \ldots, M'_{jl_j}}{M'_j}{P_j}{P}$ 
$\FIND$ $(\mathop\bigoplus\nolimits_{j=1}^{m''} \tup{\vf_j}[\tup{i}] = \tup{i_j} \leq \tup{n_j}$ $\SUCHTHAT$ $\defined (M'_{j1}, \ab \ldots, \ab M'_{jl_j}) \fand M'_j$ $\kw{then}$ $P_j)$ $\kw{else}$ $P$
where for some $k \leq l_j$, $x[M_1, \ldots, M_m]$ is a subterm of $M'_{jk}$.

\item or in $P_j$ in a process or term of the form 
%$\cfind{j=1}{m''}{\tup{\vf_j}[\tup{i}] \leq \tup{n_j}}{
%M'_{j1}, \ldots, M'_{jl_j}}{M'_j}{P_j}{P}$ 
$\FIND$ $(\mathop\bigoplus\nolimits_{j=1}^{m''} \tup{\vf_j}[\tup{i}] = \tup{i_j} \leq \tup{n_j}$ $\SUCHTHAT$ $\defined (M'_{j1}, \ab \ldots, \ab M'_{jl_j}) \fand M'_j$ $\kw{then}$ $P_j)$ $\kw{else}$ $P$
where for some $k \leq l_j$, there is a subterm $N$ of $M'_{jk}$
such that $N \{ \tup{\vf_j}[\tup{i}] / \tup{i_j} \} = x[M_1, \ldots, M_m]$.

\end{enumerate}
Similarly, $E, \sigma, M, \tblcts, \evseq$ satisfies Invariant~\ref{inv2exec}
if and only if every occurrence of a variable
access $x[M_1, \ldots, M_m]$ in $M$ either is present in $\dom(E)$
(for all $j \leq m$, $E, \sigma, M_j \evalterm a_j$ and 
$x[a_1, \ldots, a_m] \in \dom(E)$) or satisfies one of the last four conditions
above.

$E, \sigma, \defined(M'_1, \ldots, M'_l) \wedge M, \tblcts, \evseq$ 
satisfies Invariant~\ref{inv2exec}
if and only if every occurrence of a variable
access $x[M_1, \ldots, M_m]$ in $M$ either is 
a subterm of $M'_1, \ldots, M'_l$, or is present in $\dom(E)$
(for all $j \leq m$, $E, \sigma, M_j \evalterm a_j$ and 
$x[a_1, \ldots, a_m] \in \dom(E)$) or satisfies one of the last four conditions
above.
\end{invariant}

Recall that, by Invariants~\ref{inv2} and~\ref{invtic}, the terms of
all variable accesses $x[M_1, \ldots, M_m]$ are simple. That is why we can
evaluate them by $E, \sigma', M_j \evalterm a_j$.

\begin{lemma}
If $Q_0$ satisfies Invariant~\ref{inv2}, then
$\initconfig(Q_0)$ satisfies Invariant~\ref{inv2exec}.
\end{lemma}

\begin{lemma}\label{lem:defsubterms}
Let $M$ be a simple term.
If $E, \sigma, M \evalterm a$, 
then for all subterms $x[M_1, \ldots, M_m]$ of $M$, for all $j' \leq m$,
$E, \sigma, M_{j'} \evalterm a_{j'}$ and $x[a_1,\ldots,a_m]$ 
is in $\dom(E)$.
\end{lemma}
\begin{proofsk}
By induction on $M$.
\proofcomplete
\end{proofsk}

\begin{lemma}\label{lem:substreplindex}
Let $N$, $M$ be simple terms.
If $E, \sigma[i \mapsto a'], N \evalterm a$
and $E, \sigma, M \evalterm a'$, then we have
$E, \sigma, N\{M/i\} \evalterm a$.
\end{lemma}
\begin{proofsk}
By induction on $N$.
\proofcomplete
\end{proofsk}

\begin{lemma}
If $E,\sigma,M, \tblcts, \evseq \red{p}{\ix} E',\sigma',M', \tblcts', \evseq'$
and $E,\sigma,M, \tblcts, \evseq$ satisfies Invariant~\ref{inv2exec},
then so does $E',\sigma',M', \tblcts', \evseq'$.

\sloppy

If $E,\sigma,\defined(M_1, \ldots, M_m) \wedge M, \tblcts, \evseq \red{p}{\ix} E',\sigma',M', \tblcts', \evseq'$
and $E,\sigma,\defined(M_1, \ab \ldots, \ab M_m) \wedge M, \tblcts, \evseq$ satisfies Invariant~\ref{inv2exec},
then so does $E',\sigma',M', \tblcts', \evseq'$.

If $E,(\sigma,P)\restconfig \red{p}{\ix} E',(\sigma',P'),\pset',\cset',\tblcts',\evseq'$
and $E,(\sigma,P)\restconfig$ satisfies Invariant~\ref{inv2exec},
then so does $E',(\sigma',P'),\pset',\cset',\tblcts',\evseq'$.

Moreover, if the rules that define
$E,\sigma,M, \tblcts, \evseq \red{p}{\ix} E',\sigma',M', \tblcts', \evseq'$ (resp. 
$E,\ab \sigma,\ab \defined(M_1, \ab \ldots, \ab M_m) \wedge M, \ab \tblcts, \ab \evseq \red{p}{\ix} E',\ab \sigma',\ab M', \ab \tblcts', \ab \evseq'$ or
$E,(\sigma,P)\restconfig \red{p}{\ix} E',\ab (\sigma',P'),\ab \pset',\ab \cset',\ab \tblcts',\ab \evseq'$)
require as assumption
$E'',\ab \sigma'',\ab M'', \ab \tblcts'', \ab \evseq'' \red{p}{\ix} \ldots$
or
$E'',\ab \sigma'',\ab \defined(M''_1, \ab \ldots, \ab M''_m) \wedge M'', \ab \tblcts'', \ab \evseq'' \red{p}{\ix} \ldots$,
and the initial configuration $E,\ab \sigma,\ab M, \ab \tblcts, \ab \evseq$
(resp. $E,\ab \sigma,\ab \defined(M_1, \ab \ldots, \ab M_m) \wedge M, \ab \tblcts, \ab \evseq$ or
$E,\ab (\sigma,P)\restconfig$) 
satisfies Invariant~\ref{inv2exec},
then so does the initial configuration of the assumption, $E'',\ab \sigma'',\ab M'', \ab \tblcts'', \ab \evseq''$
or $E'',\ab \sigma'',\ab \defined(M''_1, \ab \ldots, \ab M''_m) \wedge M'', \ab \tblcts'', \ab \evseq''$.
%Rules that make recursive calls: Find*, Ctx(T), DefinedYes/No, Input, Get1/2
\end{lemma}
\begin{proofsk}
The proof proceeds by induction following the definition of $\red{p}{\ix}$.
We just sketch the main arguments.

If $x[M_1,\ldots,M_m]$ is in the second case of Invariant~\ref{inv2exec}, 
and we execute the definition
of $x[M_1,\ldots,M_m]$, then for all $j \leq m$, $M_j$
is a variable replication index and $x[\sigma(M_1),\ldots,\sigma(M_m)]$ is
added to $\dom(E)$ by rules \eqref{sem:newt}, \eqref{sem:lett}, \eqref{sem:findt1}, \eqref{sem:gett1},
\eqref{sem:new}, \eqref{sem:let}, \eqref{sem:find1}, \eqref{sem:output}, or \eqref{sem:get1} 
so it moves to the first case of Invariant~\ref{inv2exec}.

If $x[M_1,\ldots,M_m]$ is in the third case of Invariant~\ref{inv2exec}, 
and we execute the
corresponding $\FIND$, this access to $x$ simply disappears.

If $x[M_1,\ldots,M_m]$ is in the fourth case of Invariant~\ref{inv2exec}, 
and we execute the $\FIND$, then $x[M_1,\ab \ldots,\ab M_m]$ is a subterm of $M'_{jk}$
for some $j \leq m''$ and $k \leq l_j$. 
Therefore, the initial configuration of the assumption $E, \sigma[\tup{i_j} \mapsto \tup{a}], D_j \wedge M'_j, \tblcts, \evseq \red{p_{k'}}{\ix_{k'}}^* E'', \sigma', r_{k'}, \tblcts, \evseq$
with $D_j \wedge M'_j = \defined(M'_{j1}, \ab \ldots, 
\ab M'_{jl_j}) \wedge M'_j$ 
and
$\sigma' = \sigma[\tup{i_j} \mapsto \tup{a}]$
also satisfies Invariant~\ref{inv2exec}.
In case this assumption is reduced by~\eqref{sem:definedyes},
we have $E, \sigma', M'_{jk}, \tblcts, \evseq \red{1}{}^* E, \sigma', a_{jk}, \tblcts, \evseq$, that is,
$E, \sigma', M'_{jk} \evalterm a_{jk}$.
Therefore, by Lemma~\ref{lem:defsubterms}, for
all $j' \leq m$, $E, \sigma', M_{j'} \evalterm a_{j'}$ and 
$x[a_1,\ldots,a_m]$
is in $\dom(E)$.  So $x[M_1,\ldots,M_m]$ moves to the first
case of Invariant~\ref{inv2exec} in $E, \sigma', M'_j, \tblcts, \evseq$ after reduction 
by~\eqref{sem:definedyes}.

If $x[M_1,\ldots,M_m]$ is in the last case of
Invariant~\ref{inv2exec}, and we execute the $\FIND$ selecting branch
$j$ by~\eqref{sem:findt1} or~\eqref{sem:find1}, then there is a
subterm $N$ of $M'_{jk}$ for some $k \leq l_j$ such that
$N\{\tup{u_j}[\tup{i}] / \tup{i_j} \} = x[M_1,\ldots,M_m]$.
By hypothesis of~\eqref{sem:findt1} or~\eqref{sem:find1}, we have
$E, \ab \sigma[\tup{i_j} \mapsto \tup{a'}], \ab D_j \wedge M'_j, \ab \tblcts, \ab \evseq \red{p_{k'}}{\ix_{k'}}^* E, \ab \sigma', \ab r_{k'}, \ab \tblcts, \ab \evseq$ 
where
$r_{k'} = \true$,
$v_0 = (j, \tup{a'}) \in S$,
$D_j \wedge M'_j = \defined(M'_{j1}, \ab \ldots, 
\ab M'_{jl_j}) \wedge M'_j$,
and
$\sigma' = \sigma[\tup{i_j} \mapsto \tup{a'}]$.
This assumption cannot reduce by~\eqref{sem:definedno} because the result
is $\true$, so it reduces by~\eqref{sem:definedyes}.
Therefore, we have $E, \ab \sigma', \ab M'_{jk}, \ab \tblcts, \ab \evseq \red{1}{}^* E, \ab \sigma', \ab a, \ab \tblcts, \ab \evseq$ for some $a$,
that is, $E, \sigma', M'_{jk} \evalterm a$.
The term $N = x[N_1, \ldots, N_m]$ is a subterm of $M'_{jk}$.
Therefore, by Lemma~\ref{lem:defsubterms}, for
all $j' \leq m$, $E, \sigma', N_{j'} \evalterm a_{j'}$ and $x[a_1,\ldots,a_m]$
is in $\dom(E)$.  
Moreover, the resulting environment $E'$ is an extension of $E$, so
a fortiori for
all $j' \leq m$, $E', \sigma', N_{j'} \evalterm a_{j'}$ and $x[a_1,\ldots,a_m]$ is in $\dom(E')$.
We have for all $j' \leq m$, $M_{j'} = N_{j'} \{ \tup{u_j}[\tup{i}]/\tup{i_j} \}$,
$E'(\tup{u_j}[\tup{i}]) = \tup{a'}$,
and $\sigma'(\tup{i_j}) = \tup{a'}$, so by Lemma~\ref{lem:substreplindex},
for
all $j' \leq m$, $E', \sigma, M_{j'} \evalterm a_{j'}$ and $x[a_1,\ldots,a_m]$ is in $\dom(E')$.
So $x[M_1,\ldots,M_m]$ also moves to the first
case of Invariant~\ref{inv2exec}.

In all other cases, the situation remains unchanged.
For context rules, this is because, in the allowed contexts,
the hole is never under a $\defined$ condition.
\proofcomplete
\end{proofsk}

Therefore, if $Q_0$ satisfies Invariant~\ref{inv2}, then in traces of
$Q_0$, the test $x[a_1, \ldots, a_m] \in \dom(E)$ in rule~\eqref{sem:var} always
succeeds, except when the considered term occurs in a $\defined $
condition of a $\FIND $.

Indeed, consider an application of rule~\eqref{sem:var}, where the array access $x[M_1, \ldots, M_m]$ is not in a
$\defined $ condition of a $\FIND $.  Then, this array access is
not under any variable definition or $\FIND $, so it is
present in $\dom(E)$: for all $j \leq
m$, $E, \sigma, M_j \evalterm a_j$ and 
$x[a_1, \ldots, a_m] \in \dom(E)$.
Hence, the test $x[a_1, \ldots, a_m] \in \dom(E)$ succeeds.

\subsubsection{Typing}

In this section, we show that our type system is compatible
with the semantics of the calculus, that is, we define a notion
of typing for semantic configurations and show that typing
is preserved by reduction (subject reduction).
Finally, the property that semantic configurations are well-typed
shows that certain conditions in the semantics always hold.

We use the following definitions:
\begin{itemize}
\item $\tyenv \vdash E$ if and only if $E(x[a_1, \ldots,
a_m]) = a$ implies $\tyenv(x) = T_1 \times \ldots \times T_m \rightarrow T$
with for all $j \leq m$, $a_j \in T_j$ and $a \in T$.
\item We define $\tyenv \vdash P$, $\tyenv \vdash Q$,
and $\tyenv \vdash M : T$ as in Section~\ref{sec:typesystem}, with the
additional rules $\tyenv \vdash a : T$ if and only if
$a \in T$, $\tyenv \vdash \keventabort{(\pp,\tup{a}):e} : T$ for all $T$,
and $\tyenv \vdash \kw{abort}$. 
(These rules are useful to type evaluated terms and processes.)
\item $\tyenv\vdash (\sigma, P)$ if and only if 
$\tyenv[i_1 \mapsto [1,n_1], \ldots, i_m \mapsto [1,n_m]\, ] \vdash P$
and for all $j \leq m$, $\sigma(i_j) \in [1, n_j]$
for some $n_1, \ldots, n_m$, where $\dom(\sigma) = [i_1, \ldots, i_m ]$. 
The judgments  $\tyenv\vdash (\sigma, Q)$
and $\tyenv\vdash (\sigma, M): T$
are defined in the same way.

\item $\tyenv \vdash \tblcts$ if and only if $\tbl(a_1, \ldots,
a_m) \in \tblcts$ implies $\tbl : T_1 \times \ldots \times T_m$
with for all $j \leq m$, $a_j \in T_j$.
\item $\tyenv \vdash \evseq$ if and only if $(\pp, \tup{a}): e(a_1, \ldots,
a_m) \in \evseq$ implies $e : T_1 \times \ldots \times T_m$
with for all $j \leq m$, $a_j \in T_j$.
\item $\tyenv \vdash E, (\sigma,P)\restconfig$ if and only if
$\tyenv \vdash E$, $\tyenv \vdash (\sigma, P)$,
$\tyenv \vdash \tblcts$, $\tyenv \vdash \evseq$, and for all $(\sigma', Q)
\in \pset$, $\tyenv \vdash (\sigma', Q)$.
\item $\tyenv \vdash E, \pset, \cset$ if and only if
$\tyenv \vdash E$ and for all $(\sigma', Q)
\in \pset$, $\tyenv \vdash (\sigma', Q)$.
\item $\tyenv \vdash E, \sigma, M : T, \tblcts, \evseq$ if and only if
$\tyenv \vdash E$, $\tyenv \vdash (\sigma, M): T$, $\tyenv \vdash \tblcts$,
and $\tyenv \vdash \evseq$.

\end{itemize}

\begin{lemma}\label{lem:subredm}
If $\tyenv \vdash E, \sigma, M: T, \tblcts, \evseq$
and $E, \sigma, M, \tblcts, \evseq \red{p}{\ix} E', \sigma', M', \tblcts', \evseq'$, then $\tyenv \vdash E', \sigma', \ab M' : T, \tblcts', \evseq'$.

So, $\tyenv \vdash E, \sigma, M: T, \tblcts, \evseq$
and $E, \sigma, M, \tblcts, \evseq\red{p}{\ix}^* E', \sigma', a, \tblcts', \evseq'$, then $\tyenv \vdash E', \sigma', \ab a : T, \tblcts', \evseq'$.
\end{lemma}
\begin{proofsk}
By induction on the derivation of $E, \sigma, M, \tblcts, \evseq \red{p}{\ix} E', \sigma', M', \tblcts', \evseq'$.
\proofcomplete
\end{proofsk}

\begin{lemma}\label{lem:subredq}
If $\tyenv \vdash E, \pset, \cset$ 
and $E, \pset, \cset \redq E', \pset', \cset'$, then
$\tyenv \vdash E', \pset', \cset'$.

So, if $\tyenv \vdash E, \pset, \cset$,
then $\tyenv \vdash \reduce(E, \pset, \cset)$.
\end{lemma}
\begin{proofsk}
By cases on the derivation of $E, \pset, \cset \redq E', \pset', \cset'$.
In the case of the replication, we have
$\tyenv \vdash (\sigma, \repl{i}{n} Q)$, 
so $\tyenv[i_1 \mapsto [1,n_1], \ldots, i_m \mapsto [1,n_m]\, ] \vdash \repl{i}{n} Q$
and for all $j \leq m$, $\sigma(i_j) \in [1, n_j]$
for some $n_1, \ldots, n_m$, where $\dom(\sigma) = [ i_1, \ldots, i_m ]$.
By~\eqref{typ:repl}, $\tyenv[i_1 \mapsto [1,n_1], \ldots, i_m \mapsto [1,n_m], i \mapsto [1, n]\, ] \vdash Q$, so $\tyenv \vdash (\sigma[i \mapsto a], Q)$ for $a \in [1, n]$.
In the case of the input, we use Lemma~\ref{lem:subredm}.
\proofcomplete
\end{proofsk}

\begin{lemma}\label{lem:inittype}
If $\tyenv \vdash Q_0$, then $\tyenv \vdash \initconfig(Q_0)$.
\end{lemma}
\begin{proofsk}
By Lemma~\ref{lem:subredq} and the previous definitions.
\proofcomplete
\end{proofsk}

\begin{lemma}[Subject reduction]\label{lem:subred}
If $\tyenv \vdash E, (\sigma, P)\restconfig$ and 
$E, \ab (\sigma, P)\restconfig \ab \red{p}{\ix} E', \ab (\sigma', P'), \ab \pset', \ab \cset', \ab \tblcts', \ab \evseq'$,
then $\tyenv \vdash E', (\sigma', P'), \pset', \cset', \tblcts', \evseq'$.
\end{lemma}
\begin{proofsk}
By cases on the derivation of
$E, P\restconfig \red{p}{\ix} E', \ab (\sigma', P'), \ab \pset', \ab \cset', \ab \tblcts', \ab \evseq'$,
using Lemmas~\ref{lem:subredm} and~\ref{lem:subredq}.
\proofcomplete
\end{proofsk}

Moreover, if the rules that define
$E,\sigma,M, \tblcts, \evseq \red{p}{\ix} E',\sigma',M', \tblcts', \evseq'$ (resp. 
$E,(\sigma,P),\ab \pset,\ab\cset,\ab\tblcts,\ab\evseq \red{p}{\ix} E',(\sigma',P'),\pset',\cset',\tblcts',\evseq'$)
require as assumption
$E'',\sigma'',M'', \tblcts'', \evseq'' \red{p}{\ix} \ldots$
and the initial configuration is well-typed $\tyenv \vdash E,\sigma,M : T, \tblcts, \evseq$
(resp. $\tyenv \vdash E,(\sigma,P)\restconfig$) 
then so is the initial configuration of the assumption, that is, 
there exists $T''$ such that $\tyenv \vdash E'',\sigma'',M'' : T'', \tblcts'', \evseq''$.
%Rules that make recursive calls: Find* followed by DefinedYes/No, Ctx(T), Input, Get1/2
\bbnote{should I add $\defined(M''_1, \ab \ldots, \ab M''_m) \wedge M''$? More rigorous, but not sure that it would be clearer.}%

As an immediate consequence of 
Lemmas~\ref{lem:inittype}, \ref{lem:subred}, and~\ref{lem:subredm}
and the observation above, we obtain:
if $Q_0$ satisfies Invariant~\ref{inv3}, then in traces of
$Q_0$, the tests $a \in T$ in rules~\eqref{sem:lett} and~\eqref{sem:let} and $\forall j \leq m, a_j \in
T_j$ in rule \eqref{sem:fun} always succeed.
Moreover, in rules \eqref{sem:newt} and \eqref{sem:new}, we always have that
$T$ is \emph{fixed}, \emph{bounded}, or \emph{nonuniform}.
In rules~\eqref{sem:ift1}, \eqref{sem:ift2}, \eqref{sem:if1}, and~\eqref{sem:if2},
the condition is in $\bool = \{ \false, \true\}$ (when it is a value, not an abort event value), so the condition $a \neq \true$ is equivalent to $a = \false$.
In the rules for $\kw{find}$, we have $r_k \in \{ \false, \true\}$ when $r_k$
is a value (not an abort event value).
In the rules~\eqref{sem:insertt}, \eqref{sem:gette}, \eqref{sem:gett1}, \eqref{sem:gett2}, \eqref{sem:insert}, \eqref{sem:gete}, \eqref{sem:get1}, and~\eqref{sem:get2}, we have $a_j \in T_j$ for $j \leq l$,
where $\tbl: T_1 \times \ldots \times T_l$.
In the rules~\eqref{sem:eventt} and~\eqref{sem:event}, we have  $a_j \in T_j$ for $j \leq l$,
where $e: T_1 \times \ldots \times T_l$.

\subsection{Subset for the Initial Game}

%now done
% \emph{What is described in this section is not implemented yet, it is
% a project. Currently, process macros with arguments are not implemented.
% CryptoVerif additionally requires that variables in $V$ and in
% $\kw{defined}$ conditions have a single definition.}

\begin{figure}[tp]
\begin{defn}
\categ{M, N}{terms}\\ 
\entry{i}{replication index}\\
\entry{x[M_1, \ldots, M_m]}{variable access}\\ 
\entry{f(M_1, \ldots, M_m)}{function application}\\
\entry{\Res{x}{T}; N}{random number}\\
\entry{\assign{\pat}{M}{N}\ \ELSE N'}{assignment (pattern-matching)}\\
\entry{\assign{x:T}{M}{N}}{assignment}\\
\entry{\bguard{M}{N}{N'}}{conditional}\\
\entry{\kw{find}\uniqueopt\ (\mathop\bigoplus\nolimits_{j=1}^m
\vf_{j1}[\tup{i}] = i_{j1} \leq n_{j1}, \ldots, \vf_{jm_j}[\tup{i}] = i_{jm_j} \leq n_{jm_j}\ \kw{suchthat}}{}\\ 
\entry{\quad\kw{defined}(M_{j1}, \ldots, M_{jl_j}) \fand M'_j\ \kw{then}\ N_j)\ \ELSE N'}{array lookup}\\
\entry{\INSERT\ \tbl(M_1,\ldots,M_l);N}{insert in table}\\
\entry{\GET\uniqueopt\ \tbl(\pat_1,\ldots,\pat_l)\ \SUCHTHAT\ M\ \IN\ N\ \ELSE N'}{get from table}\\
\entry{\kevent{e(M_1, \ldots, M_l)}; N}{event}\\
\entry{\keventabort{e}}{event $e$ and abort}
\end{defn}
\begin{defn}
\categ{\pat}{pattern}\\
\entry{x:T}{variable}\\
\entry{f(\pat_1, \ldots, \pat_m)}{function application}\\
\entry{{=}M}{comparison with a term}
\end{defn}
\begin{defn}
\categ{Q}{input process}\\
\entry{0}{nil}\\
\entry{Q \parpop Q'}{parallel composition}\\
\entry{\repl{i}{n}{Q}}{replication $n$ times}\\
\entry{\cinput{c}{\pat}; P}{input}
\end{defn}
\begin{defn}
\categ{P}{output process}\\
\entry{\coutput{c}{N}; Q}{output}\\
\entry{\Res{x}{T}; P}{random number}\\
\entry{\assign{\pat}{M}{P}\ \ELSE P'}{assignment}\\
\entry{\bguard{M}{P}{P'}}{conditional}\\
\entry{\FIND\uniqueopt \ (\mathop\bigoplus\nolimits_{j=1}^m \vf_{j1}[\tup{i}] = i_{j1} \leq n_{j1}, \ldots, \vf_{jm_j}[\tup{i}] = i_{jm_j} \leq n_{jm_j}\ \SUCHTHAT}{}\\
\entry{\quad \defined (M_{j1}, \ldots, M_{jl_j}) \fand M_j \THEN P_j)\ \ELSE P}{array lookup}\\
\entry{\INSERT\ \tbl(M_1,\ldots,M_l);P}{insert in table}\\
\entry{\GET\uniqueopt\ \tbl(\pat_1,\ldots,\pat_l)\ \SUCHTHAT\ M\ \IN\ P\ \ELSE P'}{get from table}\\
\entry{\kevent{e(M_1, \ldots, M_l)}; P}{event}\\
\entry{\keventabort{e}}{event $e$ and abort}\\
\entry{\kw{yield}}{end}
\end{defn}
\caption{Subset of the calculus for the initial game}\label{fig:syntax-initial}
\end{figure}

The variables are always defined with the current replication indices
$\tup{i}$, so we omit them, writing $x$ for $x[\tup{i}]$; they are
implicitly added by CryptoVerif. When a variable is used with the
current replication indices at its definition, we can also omit the indices.

Along similar lines, the channels $c$ are used without indices,
and the current replication indices are implicitly added by CryptoVerif.
This allows the adversary the select to which copy of processes 
it sends messages. 
The construct $\NEWCHANNEL$ cannot occur in games manipulated
by CryptoVerif. It is used only inside proofs.
The grammar of the resulting calculus is
summarized in Figure~\ref{fig:syntax-initial}.

We recommend using the constructs $\GET$ and $\INSERT$ to manage key
tables, instead of $\kw{find}$ or $\kw{if}$ with $\defined$
conditions.  When no $\kw{find}$ nor $\kw{if}$ with $\defined$
conditions occurs in the game, by Invariant~\ref{inv2}, all accesses
to variable $x$ are of the form $x[\tup{i}]$ where $\tup{i}$ are the
current replication indices at the definition of $x$.  Such accesses
are simply abbreviated as $x$. Variables can then be considered as
ordinary variables instead of arrays, since we only access the array
cell at the current replication indices.
This choice has several other advantages:
\begin{itemize}

\item Tables with $\GET$/$\INSERT$ are closer to lists usually used by
  cryptographers than $\kw{find}$, they should be easier to understand
  for the user.

\item Tables are supported by the symbolic protocol verifier ProVerif
  while $\kw{find}$ is not.  Similarly, ProVerif does not support
  channels with indices.  So avoiding $\kw{find}$ and channels with indices
  allows us to have a language compatible with ProVerif.

\item Our compiler that translates CryptoVerif specifications
into OCaml implementations~\cite{Cade12} does not support $\kw{find}$,
because tables with $\GET$/$\INSERT$ are also much easier to implement
than $\kw{find}$.

% \item In the presence of $\kw{find}$, it is difficult to give a semantics
% to the process macros described below: we cannot rename variables
% when they have array accesses, but when we write for instance
% $\qid(sk_A) \parpop \qid(sk_B)$, we need to rename the variables defined in
% the definition of $\qid$ with two different names to satisfy Invariant~\ref{inv1}.
% Without $\kw{find}$, variables can be renamed without problem.

\end{itemize}

We can define processes by macros:
  $\kw{let}\ \pid(x_1:T_1, \ldots, x_m:T_m) = P$
or
  $\kw{let}\ \qid(x_1:T_1, \ab \ldots, \ab x_m:T_m) = Q$.
If a process $\pid(M_1, \ldots, M_m)$ occurs in the initial game,
CryptoVerif verifies that $M_1, \ldots, M_m$ are of types $T_1, \ldots, T_m$
respectively, and replaces $\pid(M_1, \ldots, M_m)$ with the expansion
$P\{M_1/x_1, \ldots, M_m/x_m\}$.

We can also define functions by macros:
  $\kw{letfun}\ f(x_1:T_1, \ldots, x_m:T_m) = M$.
If a term $f(M_1, \ab \ldots, \ab M_m)$ occurs in the initial game,
CryptoVerif verifies that $M_1, \ab \ldots, \ab M_m$ are of types $T_1, \ab \ldots, \ab T_m$
respectively, and replaces $f(M_1, \ab \ldots, \ab M_m)$ with the expansion
$M\{M_1/x_1, \ab \ldots, \ab M_m/x_m\}$.

In the initial game, all bound variables with several incompatible definitions (different indices, different types, or variables defined in the same branch of a test) as well as variables declared without an explicit type
are not allowed to occur in $V$ nor in $\defined$ conditions of $\kw{find}$ or $\kw{if}$ and are renamed to distinct
names, so that Invariant~\ref{inv1} is satisfied for these variables.
The condition on input channels in Invariant~\ref{invtic} is always satisfied by
definition of the language. 
CryptoVerif checks the rest of the invariants.

\subsection{Subsets used inside the Sequence of Games}\label{sec:subsets}

\begin{figure}[tp]
\begin{defn}
\categ{M, N}{terms}\\ 
\entry{i}{replication index}\\
\entry{x[M_1, \ldots, M_m]}{variable access}\\ 
\entry{f(M_1, \ldots, M_m)}{function application}
\end{defn}
\begin{defn}
\categ{FC}{find condition}\\
\entry{M}{term}\\
\entry{\Res{x[\tup{i}]}{T}; FC}{random number}\\
\entry{\assign{\pat}{M}{FC}\ \ELSE FC'}{assignment (pattern-matching)}\\
\entry{\assign{x[\tup{i}]:T}{M}{N}}{assignment}\\
\entry{\bguard{M}{FC}{FC'}}{conditional}\\
%\entry{\cfind{j=1}{m}{
%\vf_{j1}[\tup{i}] \leq n_{j1}, \ldots, \vf_{jm_j}[\tup{i}] \leq n_{jm_j}}{
%M_{j1}, \ldots, M_{jl_j}}{M_j}{P_j}{P}}{}\\
\entry{\kw{find}\uniqueopt\ (\mathop\bigoplus\nolimits_{j=1}^m
\vf_{j1}[\tup{i}] = i_{j1} \leq n_{j1}, \ldots, \vf_{jm_j}[\tup{i}] = i_{jm_j} \leq n_{jm_j}\ \kw{suchthat}}{}\\ 
\entry{\quad\kw{defined}(M_{j1}, \ldots, M_{jl_j}) \fand FC'_j\ \kw{then}\ FC_j)\ \ELSE FC''}{array lookup}\\
\entry{\keventabort{e}}{event $e$ and abort}
\end{defn}
\begin{defn}
\categ{\pat}{pattern}\\
\entry{x[\tup{i}]:T}{variable}\\
\entry{f(\pat_1, \ldots, \pat_m)}{function application}\\
\entry{{=}M}{comparison with a term}
\end{defn}
\begin{defn}
\categ{Q}{input process}\\
\entry{0}{nil}\\
\entry{Q \parpop Q'}{parallel composition}\\
\entry{\repl{i}{n}{Q}}{replication $n$ times}\\
\entry{\cinput{c[M_1, \ldots, M_l]}{\pat}; P}{input}
\end{defn}
\begin{defn}
\categ{P}{output process}\\
\entry{\coutput{c[M_1, \ldots, M_l]}{N}; Q}{output}\\
\entry{\Res{x[\tup{i}]}{T}; P}{random number}\\
\entry{\assign{\pat}{M}{P}\ \ELSE P'}{assignment}\\
\entry{\bguard{M}{P}{P'}}{conditional}\\
\entry{\FIND\uniqueopt \ (\mathop\bigoplus\nolimits_{j=1}^m \vf_{j1}[\tup{i}] = i_{j1} \leq n_{j1}, \ldots, \vf_{jm_j}[\tup{i}] = i_{jm_j} \leq n_{jm_j}\ \SUCHTHAT}{}\\
\entry{\quad \defined (M_{j1}, \ldots, M_{jl_j}) \fand FC_j \THEN P_j)\ \ELSE P}{array lookup}\\
\entry{\kevent{e(M_1, \ldots, M_l)}; P}{event}\\
\entry{\keventabort{e}}{event $e$ and abort}\\
\entry{\kw{yield}}{end}
\end{defn}
\caption{Subset after game expansion}\label{fig:syntax-expanded}
\end{figure}

During the computation of the sequence of games, several properties
are used by CryptoVerif, either required by some game transformations
or guaranteed by others. We summarize them in this section.

\begin{property}\label{prop:nointervaltypes}
No function returns values of interval types.
The types of values chosen by $\Res{x[\tup{i}]}{T}$ are not interval types.
The type $T$ of the sent message in the~\eqref{typ:out} rule and
of the receiving pattern in the~\eqref{typ:in} rule are not 
interval types.
\end{property}

This property is satisfied by all games manipulated by CryptoVerif,
but not by processes that model the adversary. Combined with 
Invariants~\ref{inv3}, \ref{inv2}, and~\ref{invtic}, it implies that
the terms of variable accesses $x[M_1, \ldots, M_m]$
contain only replication indices and variables.
(Tuples, events, and tables can take interval types as arguments.
The constraint on inputs and outputs could probably be relaxed.)

For processes that model security assumptions on primitives, 
the receiving variable can be of an interval type. (This is used for
instance to specify the computational Diffie-Hellman assumption; see
Section~\ref{sec:primdef}.)

\begin{property}\label{prop:noreschan}
The $\Reschan{c}$ construct does not appear in games.
\end{property}

\begin{property}\label{prop:channelindices}
The indices of channels are always the current replication indices.
\end{property}
These properties are also satisfied by all games manipulated by CryptoVerif,
but not by processes that model the adversary.

\begin{property}\label{prop:notables}
The constructs $\INSERT$ and $\GET$ do not occur in the game.
\end{property}
This property is not valid in the initial game, but it is in all
other games of the sequence produced by CryptoVerif. The
very first game transformation applied by CryptoVerif,
\rn{expand\_tables}, encodes
$\INSERT$ and $\GET$ using $\kw{find}$ (see Section~\ref{sec:transftables}). 
The constructs $\INSERT$ and 
$\GET$ are never introduced by subsequent game transformations,
so this property remains valid in the rest of the sequence.

\begin{property}\label{prop:autosarename}
The variables defined in conditions of $\kw{find}$ have pairwise
distinct names.
\end{property}
This property is enforced by the transformation \rn{auto\_SArename}
(see Section~\ref{sec:transftables}) by renaming variables defined
in conditions of $\kw{find}$ to distinct names. (This is easy since
these variables do not have array accesses by Invariant~\ref{invfc}.)
Property~\ref{prop:autosarename} is required as a precondition by many game transformations,
and may be broken by game transformations that duplicate code.
Therefore, we apply \rn{auto\_SArename} after these game transformations.

\begin{property}\label{prop:expand}
The terms $M$ are simple except for conditions of $\kw{find}$.
\end{property}
The grammar of the language taking into account this property as well 
as Properties~\ref{prop:noreschan} and~\ref{prop:notables} is shown 
in Figure~\ref{fig:syntax-expanded}. By Invariant~\ref{invtfc},
$\kw{event}$ does not occur in conditions of $\kw{find}$,
so $\kw{event}$ never occurs as term.
Property~\ref{prop:expand} is enforced by the transformation \rn{expand}
(see Section~\ref{sec:transfexpand}) by converting other terms
into processes. This transformation is applied on the initial
game after \rn{expand\_tables}. Property~\ref{prop:expand} is broken by the
cryptographic transformation of Section~\ref{sec:primdef}, so
by default \rn{expand} is called again after this transformation.
Many game transformations require Property~\ref{prop:expand} as a precondition.

\subsection{Security Properties, Indistinguishability}

A context is a process containing a hole $[\,]$.
An evaluation context $C$ is a context built from
$[\,]$, $\Reschan{c}; C$, 
$Q \parpop C$, and $C \parpop Q$.
We use an evaluation context to represent the adversary.
We denote by $C[Q]$ the process obtained by replacing
the hole $[\,]$ in the context $C$ with the process $Q$.

We write $\cevent(D)$ for the set of events that occur in 
the distinguisher $D$ (i.e. are used by the distinguisher $D$).
We write $\cevent(Q)$ for the set of events that occur
in the process $Q$. We use similar notations for output
processes, contexts, \dots
We write $\cevent(Q,Q')$ for $\cevent(Q) \cup \cevent(Q')$.

\begin{definition}[Indistinguishability]\label{def:indist}
Let $Q$ and $Q'$ be two processes, $V$ a set of variables, and $\evset$
a set of events. 
Assume that $Q$ and $Q'$ satisfy Invariants~\ref{inv1}
to~\ref{inv3} with public variables $V$,
and the variables of $V$ are defined in $Q$ and $Q'$,
with the same types.

\bb{TO DO I should not require that the variables of $V$ are defined in $Q$ and $Q'$.
I should rather have $V$ contain variables with their types, 
to be able to type the context correctly even when $Q$ and/or $Q'$ do
not define some variables in $V$.
(It is possible that a variable is defined in a game, but its
definition is actually never executed, so CryptoVerif can remove it
in a future game, so we can end up with a variable in $V$ that is not
defined in the game.)
Then the Invariants~\ref{inv1} to~\ref{inv3} can be checked separately
on $C$ and on $Q$, using only $V$ in both checks. Since $Q$ is already
supposed to satisfy them, it is enough to check them on $C[0]$.
(For Invariant~\ref{inv1}, note that $\vardef(C) \cap \vardef(Q) = \emptyset$
because $\vardef(C) \cap \vardef(Q) 
\subseteq \vardef(C) \cap \fvar(Q) = \emptyset$ as mentioned below.)
Then the only conditions that depend on $Q$ are 
$\fvar(C) \cap \fvar(Q) \subseteq V$ and
$C$ and $Q$ do not use any common table.}%

An evaluation context $C$ is said to be \emph{acceptable} for $Q$
with public variables $V$ if and only if 
$\fvar(C) \cap \fvar(Q) \subseteq V$,
$\vardef(C) \cap V = \emptyset$,
$C$ and $Q$ do not use any common table,
and $C[Q]$ satisfies Invariants~\ref{inv1} to~\ref{inv3} with public 
variables $V$. 

We write $Q \approx^{V,\evset}_p Q'$ when,
for all evaluation contexts $C$ acceptable for $Q$ and $Q'$ with public variables $V$
and all distinguishers $D$ that run in time at most $t_D$ and such that 
$\cevent(D) \cap \cevent(Q,Q') \subseteq \evset$,
$|\Pr[C[Q] : D] - \Pr[C[Q'] : D]| \leq p(C,t_D)$.
\end{definition}
This definition formalizes that the probability that
algorithms $C$ and $D$ distinguish the games $Q$ and $Q'$
is at most $p(C,t_D)$. 
The probability $p$ typically depends on the runtime of $C$ and $D$,
but may also depend on other parameters, such as the number of 
queries to each oracle made by $C$. That is why $p$ takes as
arguments the whole algorithm $C$ and the runtime of $D$.
More specifically:
\begin{property}\label{prop:prob}
  All probabilities computed by CryptoVerif are built from the following
  components by mathematical operations:
  \begin{itemize}
  \item the runtime of the context;
  \item the maximum number of outputs made by the context on each channel;
  \item the value of replication bounds, which is also determined from the number of outputs performed by the context on channels;
  \item the maximum length of the bitstring represented by a term, in particular a variable; this length may depend on messages output by the context; it is used only for unbounded types; for bounded types, we use the maximum length of the type instead;
  \item the maximum length of bitstring of a type $T$;
  \item the length of the result of a function, expressed as a function of the length of its arguments;
  \item the time of some action, expressed as a function of other elements of the formula;
  \item probability functions, used in particular to express the probability of breaking each primitive from other elements of the formula;
  \item the cardinal $|T|$ of a type $T$;
  \item the probability of collision between two random values of a type $T$, or between a random value and a value independent from that random value; these probabilities depend on the default distribution on the type $D_T$;
  \item $\epsilon_T$, the distance between the default distribution $D_T$ of type $T$ and the uniform distribution;
  \item $\epsilon_{\FIND}$, where the distance between $D_{\FIND}(S)$ and the uniform distribution is  $\epsilon_{\FIND}/2$.
  \end{itemize}
\end{property}
Among the elements above, the first four depend on the context.  In
particular, probability formulas output by CryptoVerif do not depend
on the variable, table, event names in the context. They also do not
depend on the values of variables, but may depend on their length. For
variables of bounded types, the probabilities do not depend at all on
the values.

The set of events $\evset$ corresponds to events that the 
adversary is allowed to observe. When $\evset = \cevent(Q,Q')$,
we omit it and write $Q \approx^V_p Q'$.

The unusual requirement on variables of $C$ comes from the presence of
arrays and of the associated $\FIND $ construct which gives $C$
direct access to variables of $Q$ and $Q'$: the context $C$ is allowed
to access variables of $Q$ and $Q'$ only when they are in $V$. 
(In more standard settings, the calculus does not have
constructs that allow the context to access variables of $Q$ and
$Q'$.)
When $V$ is empty, we omit it and write $Q \approx^{\evset}_p Q'$.

When $C$ is acceptable for $Q$ with public variables $V$,
and we transform $Q$ into $Q'$, we can rename the fresh variables
of $Q'$ (introduced by the game transformation) so that they 
do not occur in $C$. 
Then $C$ is also acceptable for $Q'$ with public variables $V$.
(To establish this property, we use that the variables of $V$ are
defined in $Q$ and $Q'$, with the same types, so that, if $C[Q]$ is 
well-typed, then so is $C[Q']$.)

When $C$ is acceptable for $Q$ with public variables $V$, we have that
$\vardef(C) \cap \fvar(Q) = \emptyset$, because
$\vardef(C) \cap \fvar(Q) = \vardef(C) \cap \fvar(C) \cap \fvar(Q)
\subseteq \vardef(C) \cap V = \emptyset$.

The following lemma is a straightforward consequence of Definition~\ref{def:indist}:

\begin{lemma}\label{lem:evequi}
\begin{enumerate}

\item\label{evequip0} Reflexivity: $Q \approx^{V,\evset}_0 Q$.

\item\label{evequipsym} Symmetry: If $Q \approx^{V,\evset}_p Q'$, then $Q' \approx^{V,\evset}_p Q$.

\item\label{evequip1} Transitivity: If $Q \approx^{V,\evset}_p Q'$ and $Q' \approx^{V,\evset}_{p'} Q''$, then
$Q \approx^{V,\evset}_{p+p'} Q''$.

\item\label{evequip3} Application of a context:
If $Q \approx^{V,\evset}_p Q'$ and $C$ is an evaluation context acceptable
for $Q$ and $Q'$ with public variables $V$, 
then $C[Q] \approx^{V',\evset'}_{p'} C[Q']$,
where $p'(C',t_D) = p(C'[C[\,]],t_D)$, $V' \subseteq V \cup \fvar(C)$, and
$\evset' = \evset \cup (\cevent(C) \setminus \cevent(Q,Q'))$.

\end{enumerate}
\end{lemma}
\iffalse
For \ref{evequip1}, we show that, if $C$ is acceptable for $Q$ and $Q''$,
then it is acceptable for $Q'$ (after some renaming if needed), as we do below.
If $D$ uses an event of $Q'$ not in $\evseq$, we can rename that
event in $D$ without changing the probabilities for $Q,Q''$.

For \ref{evequip3}, the distinguishers $D$ for $C[Q]$ and $C[Q']$
are such that $\cevent(D) \cap \cevent(C[Q],C[Q']) \subseteq \evset'$.
Then $\cevent(D) \cap \cevent(Q,Q') \subseteq \evset' \cap \cevent(Q,Q') = 
\evset$, so $D$ can be used as distinguisher for $Q$ and $Q'$.
\fi

Next, we introduce a notion related to indistinguishability
that treats Shoup and non-unique events specially.

\begin{definition}[Property preservation with introduction of events]\label{def:indistev}
Let $Q$ and $Q'$ be two processes and $V$ a set of variables. 
Assume that $Q$ and $Q'$ satisfy Invariants~\ref{inv1}
to~\ref{inv3} with public variables $V$, 
and the variables of $V$ are defined in $Q$ and $Q'$,
with the same types.

Let $\Dfalse(\evseqnopp) = \false$ for all $\evseqnopp$.
Let $\nonunique{Q} = \bigvee \{ e \mid \unique{e} \text{ occurs in } Q\}$
and $\nonunique{Q,D} = \bigvee \{ e \mid \unique{e}$ occurs in $Q, e \notin D\}$,
where $D$ is a distinguisher consisting of a disjunction of Shoup and non-unique events,
and we write $e \notin D$ to say that $e$ does not occur in this disjunction.
We have $\nonunique{Q,D} = \nonunique{Q} \wedge\neg D$.

We write $\dset, \dsetsnu: Q, D, \usedevents \indistev{V}{p} Q', D', \usedevents'$ when
$\dset$ is a set of distinguishers,
$\dsetsnu$ is a set of Shoup and non-unique events in $\usedevents$,
$D$ and $D'$ are distinguishers consisting of a disjunction of Shoup and non-unique events,
the events that occur in $Q$ or in $D$ are in $\usedevents$,
$\usedevents \subseteq \usedevents'$,
the events that occur in $Q'$ but not in $Q$ are in $\usedevents'$ but 
not in $\usedevents$,
the events that occur in $D'$ but not in $D$ are in $\usedevents'$ but 
not in $\usedevents$,
and,
for all evaluation contexts $C$ acceptable for $Q$ and $Q'$ with public variables $V$ 
that do not contain events in $\usedevents'$,
all distinguishers $D_0 \in \dset\cup \{\Dfalse\}$ that run in time at most $t_{D_0}$,
all distinguishers $D_1$ that are disjunctions of events in $\dsetsnu$,
\bbnote{another approach would be to replace $D_0 \vee D_1$ with any distinguisher $D_0$ such that $\cevent(D_0) \cap \usedevents' \subseteq \evset$, where $\evset$ is a set of events in $\usedevents$ that replaces $\dset$ and $\dsetsnu$.
That would probably be sufficient for everything except for "success simplify"!!
And the current notion is not satisfied by "success simplify" (see explanation in that transfo).}
\begin{align}
  \begin{split}
    &\Pr[C[Q] : (D_0 \vee D_1 \vee D) \wedge \neg \nonunique{Q,D_1 \vee D}] \\
    &\qquad \leq \Pr[C[Q'] : (D_0 \vee D_1 \vee D') \wedge \neg \nonunique{Q',D_1 \vee D'}] + p(C,t_{D_0})
  \end{split}\label{eq:indistev}
\end{align}
\end{definition}

Intuitively, the events $\usedevents$ are those used by CryptoVerif in the sequence of games until the game $Q$ included, while the events $\usedevents'$ are those used until $Q'$. Hence, $\usedevents$ contains the events that occur in $Q$; $\usedevents'$ contains $\usedevents$ and the events that occur in $Q'$. The formula $D_0 \in \dset \cup \{\Dfalse\}$ corresponds to the initial query to prove: it is
\begin{itemize}
\item a correspondence distinguisher $\neg\varphi$ for a correspondence property (see Section~\ref{sec:def:corresp});
\item $\sevent$ or $\neg\sbarevent$ for (one-session) secrecy (see Section~\ref{sec:def:secrecy}) and bit secrecy (see Section~\ref{sec:secrbit});
\item any distinguisher for indistinguishability;
\item $\Dfalse$ when the initial query has already been proved, and only Shoup and non-unique events remain to be proved.
\end{itemize}
We need to specify precisely the distinguishers needed for the queries we want to prove, because some game transformations of CryptoVerif rely on that. For instance, \rn{simplify} (Section~\ref{sec:simplify}) removes events that are not used by the queries.
%; \rn{success\ simplify} (Section~\ref{sec:success_simplify}) removes parts of the game for which we can show that the adversary cannot break the query when this part of the game is reached.

The formula $D_1$ is a disjunction of Shoup and non-unique events that remain to be proved, both in $Q$ and in $Q'$. These events are in $\dsetsnu$ and in $\usedevents$.
The formula $D$ is a disjunction of Shoup and non-unique events that remain to be proved in $Q$, while the formula $D'$ is a disjunction of Shoup and non-unique events that remain to be proved in $Q'$.
Hence, the events that occur in $D$ and not in $D'$ are events proved while transforming $Q$ into $Q'$. (``Proving'' an event means proving that this event has a negligible probability of occurring, and adding that probability to $p$.)
In contrast, the events that occur in $D'$ and not in $D$ are fresh Shoup and non-unique events introduced during the transformation of $Q$ into $Q'$, and that will need to be proved later; hence these events are in $\usedevents'$ but not in $\usedevents$.
More generally, all events that occur in $Q'$ but not in $Q$ are fresh events introduced in the transformation of $Q$ into $Q'$, so they are in $\usedevents'$ but not in $\usedevents$.

When there are no Shoup nor non-unique events, we have $D_1 = D = D' = \Dfalse$ and
$\nonunique{Q,D_1\vee D} =  \nonunique{Q',D_1\vee D'} = \Dfalse$, so the inequality~\eqref{eq:indistev} reduces to
\[\Pr[C[Q] : D_0] \leq \Pr[C[Q'] : D_0] + p(C,t_{D_0})\]
Using $\neg D_0$ instead of $D_0$, we obtain
\[1-\Pr[C[Q] : D_0] \leq 1-\Pr[C[Q'] : D_0] + p(C,t_{D_0})\]
so by combining the two, $|\Pr[C[Q] : D_0] - \Pr[C[Q'] : D_0]| \leq p(C,t_{D_0})$
as in the definition of indistinguishability.
In the general case, the inequality~\eqref{eq:indistev} differs from this formula because~\eqref{eq:indistev} always counts the traces that execute Shoup and non-unique events that remain to be proved (these traces are always included in the probability by $D_1 \vee D$, resp. $D_1 \vee D'$; the probability of these events needs to be bounded), and never counts the traces that execute proved non-unique events (these traces are excluded by $\neg\nonunique{Q,D_1\vee D}$, resp. $\neg\nonunique{Q',D_1\vee D'}$; the probability of these events has already been bounded). We could also exclude traces that execute proved Shoup events, though it is less essential: Shoup events often simply disappear when they are proved, while non-unique events remain in the game. Excluding traces that execute proved non-unique events allows us to exploit that the corresponding $\FIND$ or $\GET$ is unique in the transformation from $Q$ to $Q'$: intuitively, the traces in which the $\FIND$ or $\GET$ is not unique are not counted, so they can be ignored. (Obviously, this point needs to be proved more precisely for each game transformation.)

We need to introduce a distinct event for each $\unique{e}$ because not all $\FIND\unique{e}$ and $\GET\unique{e}$ may be proved unique in the current process and because we need
to distinguish the non-unique events that occur in $Q$ from those that occur
in the context $C$.
However, once a $\unique{e}$ is proved, its name does not matter: all such events are counted in $\nonunique{Q,D_1 \vee D}$, and that remains true in future game transformations, since events are never re-added to $D_1$ or $D$. Therefore, we can rename all such events to the same name. So we simply abbreviate $\unique{e}$ by $\unique{}$ when event $e$ is proved.
The context $C$ must not contain the events used in $Q$ or $Q'$ and more generally events in $\usedevents'$.

The formula~\eqref{eq:indistev}
could also be written
\begin{align*}
  \begin{split}
    &\Pr[C[Q] : (D_0 \wedge \neg \nonunique{Q}) \vee D_1 \vee D] \\
    &\qquad \leq \Pr[C[Q'] : (D_0 \wedge \neg \nonunique{Q'}) \vee D_1 \vee D'] + p(C,t_{D_0})
  \end{split}
\end{align*}
since $\nonunique{Q, D_1 \vee D} = \nonunique{Q} \wedge \neg(D_1 \vee D)$.
The advantage of the latter formulation is that the dependency in
$D_1$, $D$, $D'$ is simpler, which we sometimes exploit in the proofs.
However, its drawback is that it is less clear
for which events the traces are always counted and for which ones they are never counted,
because some events appear both in $\nonunique{Q}$ and in $D_1 \vee D$.
That is why we chose formula~\eqref{eq:indistev},
to make it clear that the traces that execute events in $D_1 \vee D$ are
always counted and the traces that execute events in $\nonunique{Q, D_1 \vee D}$
are never counted.

When $f$ is a function from sequences of events to sequences of events
and $D$ is a distinguisher, we define the distinguisher $D \circ f$ by
$(D \circ f)(\evseqnopp) = D(f(\evseqnopp))$.
In particular, when $\sigma$ is a renaming of events,
the distinguisher $D \circ \sigma^{-1}$ is defined by $(D \circ \sigma^{-1})(\evseqnopp) = D(\sigma^{-1} \evseqnopp)$.
When $D$ is defined as a logical formula, that corresponds to renaming
the events in $D$. For instance, 
if $D = e_1 \vee \ldots \vee e_m$, then $D \circ \sigma^{-1} = 
\sigma e_1 \vee \ldots \vee \sigma e_m$.
When $\dset$ is a set of distinguishers, $\sigma \dset = \{ D \circ \sigma^{-1} \mid D \in \dset\}$.

We write $\dset_{\neg\evset}$ for the set of distinguishers that do not
use events in $\evset$.

\begin{lemma}\label{lem:indistev}
\begin{enumerate}

\item\label{linkindist} Link with indistinguishability:  
\begin{enumerate}
\item\label{linkindist1}
Suppose that 
$\dsetsnu$ is a set of Shoup events,
$D$ is a distinguisher consisting of a disjunction of Shoup events,
$Q$, $Q'$ do not contain non-unique events,
the events that occur in $Q$, $D$, $\dset$, or $\dsetsnu$ are in $\usedevents$, and
the events that occur in $Q'$ also occur in $Q$.
If $Q \approx^{V,\evset}_p Q'$,
then $\dset,\dsetsnu : Q, D, \usedevents \indistev{V}{p} Q', D, \usedevents$, for all $\dset, \dsetsnu, D$ such that $\cevent(\dset) \cap \usedevents \subseteq \evset$, $\dsetsnu \subseteq \evset$, $\cevent(D) \subseteq \evset$.

\item\label{linkindist2}
Suppose that 
$Q$ and $Q'$ do not contain non-unique events,
the events that occur in $Q$ are in $\usedevents$, and
the events that occur in $Q'$ also occur in $Q$.
Then
$\dset_{\neg(\usedevents\setminus\evset)}, \emptyset: Q, \Dfalse, \usedevents \indistev{V}{p} Q', \Dfalse, \usedevents$ if and only if $Q \approx^{V,\evset}_p Q'$.

\end{enumerate}

\item\label{indistref} Reflexivity: 
If $\dsetsnu$ is a set of Shoup and non-unique events in $\usedevents$,
$D$ is a distinguisher consisting of a disjunction of Shoup and non-unique events, and the events that occur in $Q$ or in $D$ are in $\usedevents$, then $\dset, \dsetsnu: Q, D, \usedevents \ab\indistev{V}{0} Q, D, \ab \usedevents$.

\item\label{indisttrans} Transitivity:
If $\dset, \dsetsnu: Q, D, \usedevents \indistev{V}{p} Q', D', \usedevents'$ and
$\dset, \dsetsnu: Q', D', \ab\usedevents' \ab \indistev{V}{p'} Q'', D'', \usedevents''$, then we have
$\dset, \dsetsnu: Q, D, \usedevents \ab \indistev{V}{p''} Q'', D'', \usedevents''$, where
$p''(C, \ab t_{D_0}) = p(C, \ab t_{D_0}) + p'(C, \ab t_{D_0})$.

\item\label{indistctx} Application of a context:
  %Restricted to \dset_{\neg\usedevents'}
  If $\dset_{\neg\usedevents'}, \dsetsnu: Q, D, \usedevents \indistev{V}{p} Q', D', \usedevents'$,
  $\sigma$ is a renaming of the events in $\usedevents'$ to events not in $\usedevents^+$, 
$C$ is a context acceptable for $\sigma Q$ and $\sigma Q'$ with public variables $V$ such that 
$\cevent(C) \subseteq \usedevents^+$, and
$\dsetsnu'$ is a set of Shoup and non-unique events in $\usedevents^+$, 
then we have $\dset_{\neg\sigma\usedevents'}, \sigma \dsetsnu \cup \dsetsnu': C[\sigma Q], D\circ \sigma^{-1}, \sigma \usedevents \cup \usedevents^+\indistev{V'}{p'} C[\sigma Q'],\ab D'\circ \sigma^{-1}, \ab \sigma \usedevents'\cup \usedevents^+$,
where 
$p'(C',t_{D_0}) = p(\sigma^{-1} (C'[C[\,]]),t_{D_0})$, and $V' \subseteq V \cup \fvar(C)$.

\item\label{indistadd} Adding distinguishers:
  If $\dset, \dsetsnu: Q, D, \usedevents \indistev{V}{p} Q', D', \usedevents'$ and $e \in \dsetsnu$, then
$\dset, \dsetsnu: Q, D \vee e, \usedevents \indistev{V}{p} Q', D' \vee e, \usedevents'$.

\item\label{indistremove} Removing distinguishers:
  If $\dset, \dsetsnu: Q, D, \usedevents \indistev{V}{p} Q', D', \usedevents'$, $\dset' \subseteq \dset$, and $\dsetsnu' \subseteq \dsetsnu$, then
$\dset', \dsetsnu': Q, D, \usedevents \indistev{V}{p} Q', D', \usedevents'$

\end{enumerate}
\end{lemma}

\begin{proof}
Property~\ref{linkindist1}: Given the hypothesis,
$\dset,\dsetsnu : Q, D, \usedevents \indistev{V}{p} Q', D, \usedevents$
reduces to:
for all evaluation contexts $C$ acceptable for $Q$ and $Q'$ with public variables $V$ 
that do not contain events in $\usedevents$,
all distinguishers $D_0 \in \dset \cup \{\Dfalse\}$ that run in time at most $t_{D_0}$,  
and all distinguishers $D_1$ that are disjunctions of events in $\dsetsnu$,
\[\Pr[C[Q] : D_0 \vee D_1 \vee D] \leq \Pr[C[Q'] : D_0 \vee D_1 \vee D] + p(C,t_{D_0})\]
We have $\cevent(D_0 \vee D_1 \vee D) \cap \cevent(Q,Q') \subseteq \cevent(D_0 \vee D_1 \vee D) \cap \usedevents \subseteq \evset$.
Moreover, $D_0 \vee D_1 \vee D$ can be implemented in the same time as $D_0$
since evaluating $D_0 \vee D_1 \vee D$ can be done by setting the final result to true as soon as an event in $D_1 \vee D$ is executed, and evaluating $D_0$ otherwise. This does not take more time than evaluating $D_0$. 
So this inequality is a consequence of $Q \approx^{V,\evset}_p Q'$.

Property~\ref{linkindist2}: Given the hypothesis,
$\dset_{\neg(\usedevents\setminus\evset)}, \emptyset: Q, \Dfalse, \usedevents \indistev{V}{p} Q', \Dfalse, \usedevents$ reduces to:
for all evaluation contexts $C$ acceptable for $Q$ and $Q'$ with public variables $V$ that do not contain events in $\usedevents$,
and all distinguishers $D_0 \in \dset_{\neg(\usedevents\setminus\evset)}$ that run in time at most $t_{D_0}$,
\[\Pr[C[Q] : D_0 ] \leq \Pr[C[Q'] : D_0 ] + p(C,t_{D_0})\,,\]
that is,
\begin{equation}
\Pr[C[Q] : D_0 ] - \Pr[C[Q'] : D_0 ] \leq p(C,t_{D_0})\,.\label{eq:equivindistev}
\end{equation}
We have $\cevent(D_0) \cap \cevent(Q,Q') \subseteq \cevent(D_0) \cap \usedevents \subseteq \evset$.
Therefore, $Q \approx^{V,\evset}_p Q'$ implies $\dset_{\neg(\usedevents\setminus\evset)}, \emptyset: Q, \Dfalse, \usedevents \indistev{V}{p} Q', \Dfalse, \usedevents$. 

Conversely, assume 
$\dset_{\neg(\usedevents\setminus\evset)}, \emptyset: Q, \Dfalse, \usedevents \indistev{V}{p} Q', \Dfalse, \usedevents$. 
Let $C$ be an evaluation context $C$ acceptable for $Q$ and $Q'$ with public variables $V$ 
and $D$ be a distinguisher such that $\cevent(D) \cap \cevent(Q,Q') \subseteq \evset$.
Let $\sigma$ be a renaming of events in $\usedevents$ to fresh events.
Then $\sigma C$ does not contain events in $\usedevents$.
Given a sequence of events $\evseqnopp$, let $f(\evseqnopp)$ be obtained
by removing all events in $\usedevents\setminus\evset$ from $\evseqnopp$
and, in the remaining sequence, 
renaming the events $e$ in $\sigma\usedevents$ to $\sigma^{-1}(e)$.
Let $D_0 = D \circ f$. By construction, $D_0 \in \dset_{\neg(\usedevents\setminus\evset)}$. So by~\eqref{eq:equivindistev}, we get
\[\Pr[(\sigma C)[Q] : D_0 ] - \Pr[(\sigma C)[Q'] : D_0 ] \leq p(C,t_{D_0})\,.\]
Since $\neg D_0 \in \dset_{\neg(\usedevents\setminus\evset)}$ and runs in the same time as $D_0$, 
\[1-\Pr[(\sigma C)[Q] : D_0 ] - 1+\Pr[(\sigma C)[Q'] : D_0 ] \leq p(C,t_{D_0})\]
so
\[|\Pr[(\sigma C)[Q] : D_0 ] - \Pr[(\sigma C)[Q'] : D_0| ] \leq p(C,t_{D_0})\,.\]
Furthermore, each trace of $(\sigma C)[Q]$ with sequence of events $\evseqnopp$ corresponds to a trace of $C[Q]$ with the same probability and a sequence of events $\evseqnopp'$ equal to $f(\evseqnopp)$ plus some events not used by $D$, so $D_0(\evseqnopp) = D(f(\evseqnopp)) = D(\evseqnopp')$.
Indeed, in a trace of $(\sigma C)[Q]$, if an event $e(\dots)$ is executed in $Q$, then it is executed by $C[Q]$ as well. If it is in $\evset$, then it is left unchanged by $f$. If it is not in $\evset$, then it is in $\cevent(Q) \setminus \evset \subseteq \usedevents\setminus \evset$, so it is removed by $f$; furthermore, this event is not used by $D$ since $\cevent(D) \cap \cevent(Q) \subseteq \evset$. If an event $e(\dots)$ is executed in $\sigma C$, then either $e\in\sigma\usedevents$, $e'(\dots)$ is executed by $C[Q]$ with $e' = \sigma^{-1}(e)$, and $f$ maps $e$ to $e'$; or $e \notin \sigma\usedevents$, $e \notin \usedevents$, $e(\dots)$ is executed by $C[Q]$, and $f$ leaves $e$ unchanged.
Therefore, $\Pr[(\sigma C)[Q] : D_0 ] = \Pr[C[Q] : D ]$. We have a similar situation for $Q'$ instead of $Q$, and $D$ can be implemented in the same time as $D_0$, so
\[|\Pr[C[Q] : D ] - \Pr[C[Q'] : D ]| \leq p(C,t_D)\]
so $Q \approx^{V,\evset}_p Q'$.

Property~\ref{indistref}: Obvious.

Property~\ref{indisttrans}:
The events in $Q$ or $D$ are in $\usedevents$.
We have $\usedevents \subseteq \usedevents' \subseteq \usedevents''$.
The events that occur in $Q''$ but not in $Q$ occur either in $Q''$ but not in $Q'$ or in $Q'$ but not in $Q$; in the former case, they are in $\usedevents'' \setminus \usedevents'$; in the latter case, they are in $\usedevents' \setminus \usedevents$; so in both cases they are in $\usedevents'' \setminus \usedevents$.
The same reasoning applies for the events that occur in $D''$ but not in $D$.

Let $C$ be any evaluation context acceptable for $Q$ and $Q''$ with public variables $V$ that does not contain events in $\usedevents''$. After renaming the variables of $C$ that do not occur in $Q$ and $Q''$ and the tables of $C$ that do not occur in $Q$ and $Q''$ so that they do not occur in $Q'$, $C$ is also acceptable for $Q'$ with public variables $V$.
Furthermore, by Property~\ref{prop:prob}, this renaming does not change the probabilities.
Since $\usedevents' \subseteq \usedevents''$, $C$ does not contain events in $\usedevents'$.

Let $D_0 \in \dset \cup \{\Dfalse\}$ that runs in time at most $t_{D_0}$.
Let $D_1$ be a disjunction of events of $\dsetsnu$.
Then we have:
\begin{align*}
  &\Pr[C[Q] : (D_0 \vee D_1 \vee D)) \wedge \neg \nonunique{Q,D_1 \vee D} ] \\*
  &\quad \leq \Pr[C[Q'] : (D_0 \vee D_1 \vee D') \wedge \neg \nonunique{Q',D_1 \vee D'} ] + p(C,t_{D_0})  \\*
  &\quad \leq \Pr[C[Q''] : (D_0 \vee D_1 \vee D'') \wedge \neg\nonunique{Q'',D_1 \vee D''}] + p(C,t_{D_0}) + p'(C, t_{D_0}) \\*
  &\quad\leq \Pr[C[Q''] : (D_0 \vee D_1 \vee D'')) \wedge \neg\nonunique{Q'',D_1 \vee D''}] + p''(C,t_{D_0})
\end{align*}
by definition of $p''$.
Therefore, we have $\dset, \dsetsnu: Q, D, \usedevents \indistev{V}{p''} Q'', D'', \usedevents''$.

Property~\ref{indistctx}:
$\sigma\dsetsnu$ is a set of Shoup and non-unique events in $\sigma\usedevents \subseteq \sigma\usedevents \cup \usedevents^+$, so $\sigma\dsetsnu \cup \dsetsnu'$ is a set of Shoup and non-unique events in $\sigma\usedevents \cup \usedevents^+$. 
The events that occur in $C[\sigma Q]$ are either in $C$ or in $\sigma Q$; the former case, they are in $\usedevents^+$ by hypothesis; in the latter case, they are also in $\sigma\usedevents$, since the events of $Q$ are in $\usedevents$. So the events that occur in $C[\sigma Q]$ are in $\sigma\usedevents \cup \usedevents^+$.
The events in $D\circ \sigma^{-1}$ are in $\sigma\usedevents \subseteq \sigma\usedevents \cup \usedevents^+$.
We have $\usedevents \subseteq \usedevents'$, so $\sigma\usedevents \cup \usedevents^+ \subseteq \sigma\usedevents' \cup \usedevents^+$.
The events that occur in $C[\sigma Q']$ but not in $C[\sigma Q]$ are in $\sigma Q'$ but not in $\sigma Q$, so they are in $\sigma \usedevents' \setminus \sigma \usedevents$, so in $(\sigma\usedevents' \cup \usedevents^+) \setminus (\sigma\usedevents \cup \usedevents^+)$.
Similarly, the events that occur in $D'\circ \sigma^{-1}$ but not in $D\circ\sigma^{-1}$ are in $\sigma\usedevents' \setminus\sigma\usedevents$, so they are in $(\sigma\usedevents' \cup \usedevents^+) \setminus (\sigma\usedevents \cup \usedevents^+)$.

Let $C'$ be any evaluation context acceptable for $C[\sigma Q]$ and $C[\sigma Q']$ with public variables $V'$ that does not contain events in $\sigma\usedevents' \cup \usedevents^+$.
We rename the variables of $C'$ not in $V'$ so that they are not in $V$; by Property~\ref{prop:prob}, this renaming does not change
the probabilities.
Then $\sigma^{-1} (C'[C[\,]])$ is an evaluation context acceptable for $Q$ and $Q'$ with public variables $V$. Indeed, 
\begin{align*}
\fvar(\sigma^{-1}(C'[C[\,]])) \cap \fvar (Q) 
&= (\fvar(C') \cup \fvar(C)) \cap \fvar(Q)\\
& \subseteq (V' \cup \fvar(C)) \cap \fvar(Q)\tag*{since $\fvar(C') \cap \fvar(Q) \subseteq \fvar(C') \cap \fvar(C[Q]) \subseteq V'$}\\
&\subseteq (V \cup \fvar(C)) \cap \fvar(Q)\tag*{since $V' \subseteq V \cup \fvar(C)$}\\
&\subseteq V \tag*{since $\fvar(C) \cap \fvar(Q) \subseteq V$}
\end{align*}
We have similarly $\fvar(\sigma^{-1} (C'[C[\,]])) \cap \fvar (Q')  \subseteq V$.
We also have
$\vardef(\sigma^{-1} (C'[C[\,]])) \cap V = (\vardef(C') \cap V) \cup (\vardef(C) \cap V) = \emptyset$
since
$\vardef(C) \cap V = \emptyset$ because $C$ is an acceptable evaluation context for $\sigma Q$ with public variables $V$
and $\vardef(C') \cap V \subseteq \vardef(C') \cap V' = \emptyset$ because we have renamed the variables of $C'$ not in $V'$ so that they are not in $V$ and $C'$ is an acceptable evaluation context for $C[\sigma Q]$ and with public variables $V'$.
Moreover, $C$ and $\sigma Q$ do not use any common table, and $C'$ and $C[\sigma Q]$ do not use any common table so a fortiori $C'$ and $\sigma Q$  do not use any common table. Therefore, $C'[C[\,]]$ and $\sigma Q$ do not use any common table, so $\sigma^{-1} (C'[C[\,]])$ and $Q$ do not use any common table.
Similarly, $\sigma^{-1} (C'[C[\,]])$ and $Q'$ do not use any common table.
The context $\sigma^{-1} (C'[C[\,]])$ does not contain events in $\usedevents'$, since $C'$ and $C$ do not contain events in $\sigma \usedevents'$, because $C'$ does not contain events in $\sigma\usedevents' \cup \usedevents^+$ and the events of $C$ are in $\usedevents^+$ which is disjoint from $\sigma \usedevents'$. (The renamings $\sigma$ and $\sigma^{-1}$ are bijections, so for instance $\sigma$ maps the fresh events introduced by $\sigma$ to $\usedevents'$ and $\sigma^{-1}$ maps $\usedevents'$ to the fresh events introduced by $\sigma$.)

By using the property $\dset_{\neg\usedevents'}, \dsetsnu: Q, D, \usedevents \indistev{V}{p} Q', D', \usedevents'$ with the context $\sigma^{-1} (C'[C[\,]])$, we get for any distinguishers $D_0 \in \dset_{\neg\usedevents'}\cup \{\Dfalse\}$ that runs in time $t_{D_0}$ and $D_1$ disjunction of events in $\dsetsnu$:
\begin{align}
\begin{split}
  &\Pr[\sigma^{-1}C'[\sigma^{-1}C[Q]] : (D_0 \vee D_1 \vee D)) \wedge \neg \nonunique{Q,D_1 \vee D}] \\
  &\qquad \leq \Pr[\sigma^{-1}C'[\sigma^{-1}C[Q']] : (D_0 \vee D_1 \vee D')) \wedge \neg \nonunique{Q',D_1 \vee D'}] + p(\sigma^{-1} (C'[C[\,]]),t_{D_0})
\end{split}\label{eq:context1}
\end{align}
Let $D'_0 \in \dset_{\neg\sigma\usedevents'}\cup \{\Dfalse\}$ that runs in time at most $t_{D_0}$,
Let $D'_1$ be a disjunction of events in $\sigma\dsetsnu \cup \dsetsnu'$.
We can write $D'_1$ under the form
$D'_1 = D'_2 \vee D'_3$ where $D'_2$ is a disjunction of events in $\sigma\dsetsnu$ and $D'_3$ is a disjunction of events in $\dsetsnu'$.

By applying~\eqref{eq:context1} to $D_0 = (D'_0 \circ \sigma \vee D'_3 \circ \sigma) \wedge \neg \nonunique{\sigma^{-1}C,D'_3 \circ \sigma}$, which uses events not in $\usedevents'$, and to $D_1 = D'_2 \circ \sigma$,
we get
\[\begin{split}
&\Pr[\sigma^{-1}C'[\sigma^{-1}C[Q]] : (((D'_0 \circ \sigma \vee D'_3 \circ \sigma) \wedge \neg \nonunique{\sigma^{-1}C,D'_3 \circ \sigma}) \vee D'_2 \circ \sigma \vee D))\\
  &\qquad \wedge \neg \nonunique{Q,D'_2 \circ \sigma \vee D}] \\
&\quad \leq \Pr[\sigma^{-1}C'[\sigma^{-1}C[Q']] : (((D'_0 \circ \sigma \vee D'_3 \circ \sigma) \wedge \neg \nonunique{\sigma^{-1}C,D'_3 \circ \sigma}) \vee D'_2 \circ \sigma \vee D'))\\
  &\qquad \wedge \neg \nonunique{Q',D'_2  \circ \sigma\vee D'}] + p(\sigma^{-1} (C'[C[\,]]),t_{D'_0})
\end{split}\]
since $D_0$ can be implemented to run in the same time as $D'_0$.
By applying $\sigma$, we have
\[\begin{split}
  &\Pr[C'[C[\sigma Q]] : (((D'_0 \vee D'_3) \wedge \neg \nonunique{C,D'_3}) \vee D'_2 \vee D\circ \sigma^{-1})) \wedge \neg \nonunique{\sigma Q,D'_2 \vee D\circ\sigma^{-1}}] \\
&\quad \leq \Pr[C'[C[\sigma Q']] : (((D'_0 \vee D'_3) \wedge \neg \nonunique{C,D'_3}) \vee D'_2 \vee D'\circ \sigma^{-1}))\\
  &\qquad \wedge \neg \nonunique{\sigma Q',D'_2\vee D'\circ\sigma^{-1}}] + p(\sigma^{-1} (C'[C[\,]]),t_{D'_0})
\end{split}\]
Since the events of $C$ are in $\usedevents^+$,
the events of $\nonunique{C,D'_3}$ are disjoint from those in $(D'_2 \vee D\circ \sigma^{-1})$,
so $(\neg \nonunique{C,D'_3}) \vee (D'_2 \vee D\circ \sigma^{-1}) = \neg \nonunique{C,D'_3}$
and similarly
$(\neg \nonunique{C,D'_3}) \vee (D'_2 \vee D'\circ \sigma^{-1}) = \neg \nonunique{C,D'_3}$.
So we have
\[\begin{split}
  &\Pr[C'[C[\sigma Q]] : (D'_0 \vee D'_3 \vee D'_2 \vee D\circ \sigma^{-1})) \wedge \neg \nonunique{C,D'_3} \wedge \neg \nonunique{\sigma Q,D'_2 \vee D\circ\sigma^{-1}}] \\
&\quad \leq \Pr[C'[C[\sigma Q']] : (D'_0 \vee D'_3 \vee D'_2 \vee D'\circ \sigma^{-1})) \wedge \neg \nonunique{C,D'_3}\\
  &\qquad \wedge \neg \nonunique{\sigma Q',D'_2\vee D'\circ\sigma^{-1}}] + p(\sigma^{-1} (C'[C[\,]]),t_{D'_0})
\end{split}\]
that is
\[\begin{split}
&\Pr[C'[C[\sigma Q]] : (D'_0 \vee D'_1 \vee D\circ \sigma^{-1})) \wedge \neg \nonunique{C[\sigma Q],D'_1 \vee D\circ\sigma^{-1}}] \\
&\quad \leq \Pr[C'[C[\sigma Q']] : (D'_0 \vee D'_1 \vee D'\circ \sigma^{-1})) \wedge \neg \nonunique{C[\sigma Q'],D'_1\vee D'\circ\sigma^{-1}}]\\
&\qquad + p(\sigma^{-1} (C'[C[\,]]),t_{D'_0})
\end{split}\]

Properties~\ref{indistadd} and~\ref{indistremove}: Obvious.
\proofcomplete
\end{proof}

When CryptoVerif transforms a game $G$ into a game $G'$, in most cases,
we have $G \approx^{V,\evset}_p G'$, where $p$ is the probability difference
coming from the transformation, and computed by CryptoVerif,
which implies $\dset,\dsetsnu : G, D, \usedevents \indistev{V}{p} G', D, \usedevents$, for all $\dset, \dsetsnu, D$ such that $\cevent(\dset) \cap \usedevents \subseteq \evset$, $\dsetsnu \subseteq \evset$, $\cevent(D) \subseteq \evset$ by Lemma~\ref{lem:indistev}, Property~\ref{linkindist1}.
However, there are exceptions to this situation:
\begin{itemize}
\item transformations that exploit the uniqueness of $\FIND\unique{e}$ or $\GET\unique{e}$,
which are valid only when event $e$ is not executed. These events
are taken into account by $\nonunique{Q,D}$.
\item transformations that insert events using Shoup's lemma.
This is the case of the transformations \rn{insert\_event} (see Section~\ref{sec:insertevent}) and \rn{insert} (see Section~\ref{sec:insert}). In this case, we have $\dset,\dsetsnu : G, \Dfalse, \usedevents \indistev{V}{p} G', e, \usedevents \cup \{e\}$ where $e$ is the introduced event.

The addition of Shoup events may also be combined with the cryptographic transformation of Section~\ref{sec:primdef}, for example for specifying the decisional Diffie-Hellman assumption.
In general, the cryptographic axioms are of the form
\[\dset_{\neg\usedevents_R},\emptyset : L, \Dfalse, \emptyset \indistev{}{p} R, D_R, \usedevents_R\]
where $L$ does not contain events,
$D_R = \bigvee \{ e \mid \keventabort{e} \text{ occurs in }R\}$
and $\usedevents_R = \{ e \text{ that occur in }R\}$ (both $\keventabort{e}$ and $\unique{e}$).
By Lemma~\ref{lem:indistev}, Property~\ref{indistctx}, we infer
\[\dset_{\neg\sigma\usedevents_R},\dsetsnu : C[L], \Dfalse, \usedevents \indistev{V}{p'} C[\sigma R], D_R\circ \sigma^{-1}, \usedevents\cup\sigma\usedevents_R \]
where 
$\sigma$ is a renaming of the events in $\usedevents_R$ to events not in $\usedevents$,
$C$ is a context acceptable for $L$ and $R$ such that 
the events that occur in $C$ are in $\usedevents$,
$\dsetsnu$ is a set of Shoup and non-unique events in $\usedevents$,
$p'(C',t_{D_0}) = p(\sigma^{-1}(C'[C[\,]]), t_{D_0})$,
and $V \subseteq \fvar(C)$.

Distinguishers in $\dset_{\neg\sigma\usedevents_R}$ include distinguishers that use events in $\usedevents$, in particular distinguishers for correspondences in $G$, as well as distinguishers in $\dset_{\neg(\usedevents\cup\sigma\usedevents_R)}$ used for secrecy and indistinguishability.

\item transformations that prove the absence of some events (up to some probability). For such transformations, we have
$\dset,\dsetsnu : G, e, \usedevents \indistev{V}{p} G, \Dfalse, \usedevents$
where $p$ is an upper bound of the probability of event $e$ in $G$.
For Shoup events, this generally happens when $G$ does not contain $e$ and $p(C,t_{D_0}) = 0$. For non-unique events, $p$ is the probability that the $\FIND\unique{e}$ or $\GET\unique{e}$ yields several possible choices; after this step, the event $e$ is in $\nonunique{G,D}$, so we can exploit uniqueness of $\FIND\unique{e}$ or $\GET\unique{e}$.

\end{itemize}
That is why, in general, when CryptoVerif transforms a game $G$ into a game $G'$, we have $\dset,\dsetsnu : G, D, \usedevents \indistev{V}{p} G', D', \usedevents'$.

There are still transformations that do not fit in this framework (\rn{guess} and \rn{guess\_branch}, because they multiply probabilities, as shown in Sections~\ref{sec:guessi}, \ref{sec:guessx}, and~\ref{sec:guess_branch}; \rn{success simplify} because it needs to compensate probabilities of traces that execute $\sevent$ with those that execute $\sbarevent$ to show soundness for secrecy, as shown in Section~\ref{sec:success_simplify}).

\subsubsection{Secrecy}\label{sec:def:secrecy}

Let us now define the secrecy properties that are proved
by CryptoVerif.

\begin{definition}[(One-session) secrecy]\label{def:secr}
  Let $Q$ be a process, $x$ a variable, and $V$ a set of variables.
  Let
{\allowdisplaybreaks\begin{align*}
\tproco{x} =\, & \cinput{\cSz}{}; \Res{b}{\bool}; \coutput{\cSz}{}; \\
&(\cinput{\cS}{\vf_1:[1, n_1], \ldots, \vf_m:[1, n_m]};\adeftest{x[\vf_1, \ldots, \vf_m]}{}\\*
&\phantom{(}\bguard{b}{\coutput{\cS}{x[\vf_1, \ldots, \vf_m]}}{\Res{y}{T}; \coutput{\cS}{y}}\\*
&\!\!\parpop \cinput{\cS'}{b':\bool}; \bguard{b = b'}{\keventabort{\sevent}}{\keventabort{\sbarevent}})\\
\tprocs{x} =\, & \cinput{\cSz}{}; \Res{b}{\bool}; \coutput{\cSz}{}; \\*
&(\repl{\iS}{\nS}\,\cinput{\cS[\iS]}{\vf_1:[1, n_1], \ldots, \vf_m:[1, n_m]};\adeftest{x[\vf_1, \ldots, \vf_m]}{}\\*
&\phantom{(}\bguard{b}{\coutput{\cS[\iS]}{x[\vf_1, \ldots, \vf_m]}}{}\\*
&\phantom{(}\kw{find} \ {\vfS' = \iS' \leq \nS}\ \kw{suchthat} \ \defined (y[\iS'],\vf_1[\iS'], \ldots, \vf_m[\iS']) \fand {}\\*
&\phantom{(}\qquad \vf_1[\iS'] = \vf_1 \fand \ldots \fand \vf_m[\iS'] = \vf_m\\*
&\phantom{(} \kw{then}\ \coutput{\cS[\iS]}{y[\vfS']}\\*
&\phantom{(} \ELSE\Res{y}{T}; \coutput{\cS[\iS]}{y}\\*
&\!\!\parpop \cinput{\cS'}{b':\bool}; \bguard{b = b'}{\keventabort{\sevent}}{\keventabort{\sbarevent}})
\end{align*}}%
where $\cSz, \cS, \cS' \notin \fc(Q)$, $\vf_1, \ldots, \vf_m, \vfS', y, b, b' \notin \fvar(Q) \cup V$, 
$\sevent$, $\sbarevent$ do not occur in $Q$, 
and $\tyenv(x) = [1,n_1] \times \ldots \times [1,n_m] \rightarrow T$.

Let $\prop$ be $\secrone(x)$ (\emph{one-session secrecy of $x$}) or $\secr(x)$ (\emph{secrecy of $x$}).
The events used by $\prop$ are $\sevent$ and $\sbarevent$.
Let $C_{\prop}= [\,] \parpop \tproc{\prop}$.

Let $C$ be an evaluation context acceptable for $C_{\prop}[Q]$ with public variables $V$ ($x \notin V$)
that does not contain the events used by $\prop$.
The advantage of the adversary $C$ against $\prop$ in process $Q$ is 
\[\Advt_Q^{\prop}(C) = \Pr[C[C_{\prop}[Q]] : \sevent] - \Pr[C[C_{\prop}[Q]] : \sbarevent]\]

The process $Q$ \emph{satisfies} $\prop$ with public variables $V$ ($x \notin V$) up to
probability $p$ when, for all
evaluation contexts $C$ acceptable for $C_{\prop}[Q]$ with public variables $V$
that do not contain the events used by $\prop$, $\Advt_Q^{\prop}(C) \leq p(C)$.
\end{definition}
Intuitively, when $Q$ satisfies $\prop$,
the adversary cannot guess the random bit $b$, that is, 
it cannot distinguish whether the test process $\tproc{\prop}$ outputs the
value of the secret ($b = \true$) or outputs a random number $(b = \false)$.

For one-session secrecy, the adversary performs a single test query, modeled by $\tproco{x}$.
In more detail, in $\tproco{x}$, 
we choose a random bit $b$; the 
adversary sends the indices $(\vf_1, \ldots, \vf_m)$ on channel $\cS$ to
perform a test query on $x[\vf_1, \ldots, \vf_m]$:
if $b = \true$, the test query sends back $x[\vf_1, \ldots, \vf_m]$;
if $b = \false$, it sends back a random value $y$.
Finally, the adversary should guess the bit $b$: it sends its guess $b'$ 
on channel $\cS'$ and, if the guess is correct, then event $\sevent$ is executed,
and otherwise, event $\sbarevent$ is executed.
The probability of getting some information on the secret
is the difference between the probability of $\sevent$
and the probability of $\sbarevent$.
(When the adversary always sends a guess on channel $\cS'$, we have $\Pr[C[C_{\prop}[Q]] : \sbarevent] = 1 - \Pr[C[C_{\prop}[Q]] : \sevent]$, so the advantage of the adversary is $\Advt_Q^{\prop}(C) = \Pr[C[C_{\prop}[Q]] : \sevent] - \Pr[C[C_{\prop}[Q]] : \sbarevent] = 2 \Pr[C[C_{\prop}[Q]]: \sevent] - 1$, which is a more standard formula.
By flipping a coin, the adversary can execute events $\sevent$ and $\sbarevent$ with the same probability, that is why the probability that the adversary really guesses $b$ is the difference between the probability of these two events. 
We need not take the absolute value of $\Pr[C[C_{\prop}[Q]] : \sevent] - \Pr[C[C_{\prop}[Q]] : \sbarevent]$ because, when it is negative, we can obtain the opposite, positive value by considering an adversary that sends the guess $\fnot b'$ instead of $b'$.)

For secrecy, the adversary
can perform several test queries, modeled by 
$\tprocs{x}$. This corresponds to the ``real-or-random''
definition of security~\cite{Abdalla06}. (As shown in~\cite{Abdalla06},
this notion is stronger than the more standard approach
in which the adversary can perform a single test query and some
reveal queries, which always reveal $x[\vf_1, \ldots, \vf_m]$.)
The replication bound $\nS$ in $\tproc{\secr(x)}$ is
chosen large enough so that it does not prevent communications that
would otherwise occur, so $\nS$ does not actually limit the
number of test queries. When we return a random value ($b = \false$) and
several tests queries are performed
on the same indices $\vf_1, \dots, \vf_m$, 
we must return the same random value. That is why, in this case, we look for
previous test queries ($\FIND\ \vfS'$\dots) and return the previous value of $y$ in case a previous test query was performed with the same indices.
For different indices $\vf_1, \dots, \vf_m$, the returned random values
are independent of each other, so the secrecy of $x$ requires that
the cells of array $x$ are indistinguishable from independent random
values. In contrast, the one-session secrecy of $x$ only requires
that all array cells of $x$ are indistinguishable from random values,
not that they are independent of each other.

By Invariant~\ref{invfc}, the variables defined in conditions of $\kw{find}$ and in patterns and in conditions of $\kw{get}$ have no array accesses. Therefore, the definition above applies only to variables $x$ that are not defined in conditions of $\kw{find}$ nor in patterns nor in conditions of $\kw{get}$. 

\begin{lemma}\label{lem:transfersec}
  Let $\prop$ be $\secrone(x)$ or $\secr(x)$.
  
  If $Q$ satisfies $\prop$ with public variables $V$ up to probability $p$ 
  and $C$ is an acceptable evaluation context for $Q$ with
  public variables $V$, then for all $V' \subseteq V \cup \fvar(C)$,
  $C[Q]$ satisfies $\prop$ with public variables $V'$ up to probability $p'$ 
  such that $p'(C') = p(C'[C])$.

  If $Q \approx^{V \cup \{x\},\evset}_p Q'$ and $Q$ satisfies $\prop$ with
  public variables $V$ up to probability $p'$, then $Q'$ satisfies $\prop$ with
  public variables $V$ up to probability $p''$ such that 
  $p''(C) = p'(C) + 2 \times p(C[C_{\prop}[\,]], t_{\sevent})$.
\end{lemma}
\begin{proof}
Suppose that $Q$ satisfies $\prop$ with public variables $V$ 
($x \notin V$)
and $C$ is an acceptable evaluation context for $Q$ with
  public variables $V$. Let $V' \subseteq V \cup \fvar(C)$.
Choose channels $\cSz, \cS, \cS'$, variables $\vf_1, \ldots, \vf_m, \vfS', y, b, b'$,
and events $\sevent$, $\sbarevent$  such that they do not occur in $C[Q]$.
Let $C'$ be an acceptable evaluation context for $C_{\prop}[C[Q]]$
with public variables $V'$ that does not contain $\sevent$ nor $\sbarevent$.
Then we have
\[\begin{split}
\Advt_{C[Q]}^{\prop}(C') 
&= \Pr[C'[C_{\prop}[C[Q]]] : \sevent] - \Pr[C'[C_{\prop}[C[Q]]] : \sbarevent]\\
&= \Pr[C'[C[C_{\prop}[Q]]] : \sevent] - \Pr[C'[C[C_{\prop}[Q]]] : \sbarevent]\\
&\leq p(C'[C])
\end{split}\]
We can commute the contexts $C$ and $C_{\prop}$
because the context $C$ does not bind the channels of $\tproc{\prop}$.
The context $C'[C]$ is an acceptable evaluation context for $C_{\prop}[Q]$
with public variables $V$ that does not contain $\sevent$ nor $\sbarevent$:
there is no common table between $C$ and $Q$, and between $C'$ and $C_{\prop}[C[Q]]$,
so a fortiori between $C'$ and $Q$ and $\tproc{\prop}$ does not use tables, so there is
no common table between $C'[C]$ and $C_{\prop}[Q]$; moreover
\begin{align*}
&\fvar(C'[C]) \cap \fvar(C_{\prop}[Q]) \\
&\quad = ((\fvar(C') \cap \fvar(C_{\prop}[Q])) \cup \fvar(C)) \cap \fvar(C_{\prop}[Q])\\
&\quad \subseteq (V' \cup \fvar(C)) \cap \fvar(C_{\prop}[Q]) \tag*{since $\fvar(C') \cap \fvar(C_{\prop}[C[Q]]) \subseteq V'$}\\
&\quad \subseteq (V \cup \fvar(C)) \cap \fvar(C_{\prop}[Q]) \tag*{since $V' \subseteq V \cup \fvar(C)$}\\
&\quad \subseteq V \tag*{since $\fvar(C) \cap \fvar(Q) \subseteq V$ and $\fvar(C) \cap \fvar(\tproc{\prop}) = \emptyset$}
\end{align*}

Suppose that $Q \approx^{V \cup \{x\},\evset}_p Q'$ and $Q$ satisfies $\prop$ with
  public variables $V$ up to probability $p'$.
Let $C$ be an acceptable evaluation context for $C_{\prop}[Q']$
with public variables $V$ that does not contain $\sevent$ nor $\sbarevent$.
\[\begin{split}
\Advt_{Q'}^{\prop}(C) 
&= \Pr[C[C_{\prop}[Q']] : \sevent] - \Pr[C[C_{\prop}[Q']] : \sbarevent]\\
&\leq \Pr[C[C_{\prop}[Q]] : \sevent] - \Pr[C[C_{\prop}[Q]] : \sbarevent] + {}\\
&\phantom{{}\leq {}} | \Pr[C[C_{\prop}[Q']] : \sevent] - \Pr[C[C_{\prop}[Q]] : \sevent] | + {}\\
&\phantom{{}\leq {}} | \Pr[C[C_{\prop}[Q]] : \sbarevent] - \Pr[C[C_{\prop}[Q']] : \sbarevent] |\\
&\leq p'(C) + 2 \times p(C[C_{\prop}[\,]], t_{\sevent})
\end{split}\]
since $t_{\sevent} = t_{\sbarevent}$.
Indeed, by renaming the variables and tables of $C$ that do not appear in $Q'$ to variables and tables 
that also do not occur in $Q$,
$C$ is also an acceptable evaluation context for $C_{\prop}[Q]$
with public variables $V$. 
Furthermore, by Property~\ref{prop:prob}, this renaming does not change the probabilities.
\proofcomplete
\end{proof}

\subsubsection{Secrecy for a Bit}\label{sec:secrbit}

\begin{definition}[Bit secrecy]\label{def:secrbit}
  Let $Q$ be a process, $x$ a boolean variable defined under no replication,
  and $V$ a set of variables.
  Let
\begin{align*}
  \tprocb{x} =\, &\cinput{\cS''}{b':\bool}; \adeftest{x}
  \bguard{x = b'}{\keventabort{\sevent}}{\keventabort{\sbarevent}}
\end{align*}
where $\cS'' \notin \fc(Q)$, $b' \notin \fvar(Q) \cup V$, 
$\sevent$, $\sbarevent$ do not occur in $Q$, and $\tyenv(x) = \bool$.

Let $\prop$ be $\secrbit(x)$ (\emph{bit secrecy of $x$}).
The events used by $\prop$ are $\sevent$ and $\sbarevent$.
Let $C_{\prop}= [\,] \parpop \tproc{\prop}$.
The definitions of $\Advt_Q^{\prop}(C)$ and ``$Q$ \emph{satisfies} $\prop$''
are as in Definition~\ref{def:secr}.
\end{definition}
Intuitively, when $Q$ satisfies $\prop$,
the adversary cannot guess the boolean $x$, that is, 
it cannot distinguish whether $x = \true$ or $x = \false$.
The adversary performs a single test query, modeled by $\tprocb{x}$.
This definition is simpler than the definition of (one-session) secrecy for $x$,
because we do not introduce an additional random bit $b$.

By Invariant~\ref{invfc}, the variables defined in conditions of $\kw{find}$ and in patterns and in conditions of $\kw{get}$ have no array accesses. Therefore, the definition above applies only to variables $x$ that are not defined in conditions of $\kw{find}$ nor in patterns nor in conditions of $\kw{get}$. 

Lemma~\ref{lem:transfersec} is also valid when $\prop = \secrbit(x)$, with the same
statement and proof.

\begin{lemma}\label{lem:secrecy-comp}
  If $b_0$ is a boolean variable defined under no replication and $Q$ preserves the one-session secrecy of $b_0$ with public variables $V$ up to probability $p$, then $Q$ preserves the bit secrecy of $b_0$ with public variables $V$ up to probability $2p$.
\end{lemma}

\begin{proof}
  Let $C$ be any acceptable evaluation context for $Q \parpop \tprocb{b_0}$
  with public variables $V$.
  Let 
  \[
    C' = C[\_\parpop \cinput{\cS''}{b'_0:\bool};
    \coutput{\cSz}{};\cinput{\cSz}{};
    \coutput{\cS}{};\cinput{\cS}{b''_0:\bool};\coutput{\cS'}{b'_0 = b''_0}]\,.
    \]
  We execute $C'[Q \parpop \tproco{b_0}]$.

When $b'_0 = b_0$ ($C[Q \parpop \tprocb{b_0}]$ executes $\sevent$),
\begin{itemize}
\item if $b = \true$ (probability 1/2), then $b''_0 = b_0$,
  so $b' = (b'_0 = b''_0) = \true$ and $\sevent$ is executed with probability 1/2
\item if $b = \false$ (probability 1/2), then $b''_0$ is random,
  so
  \begin{itemize}
    \item $b''_0 = b'_0$ with probability 1/4,
      so $b' = (b'_0 = b''_0) = \true$ and $\sbarevent$ is executed;
    \item $b''_0 = \neg b'_0$ with probability 1/4,
      so $b' = (b'_0 = b''_0) = \false$ and $\sevent$ is executed.
  \end{itemize}
\end{itemize}

When $b'_0 = \neg b_0$ ($C[Q \parpop \tprocb{b_0}]$ executes $\sbarevent$),
\begin{itemize}
\item if $b = \true$ (probability 1/2), then $b''_0 = b_0$,
  so $b'$ is false and $\sbarevent$ is executed with probability 1/2;
\item if $b = \false$ (probability 1/2), then $b''_0$ is random,
  so
  \begin{itemize}
    \item $b''_0 = b'_0$ with probability 1/4,
      so $b' = \true$ and $\sbarevent$ is executed;
    \item $b''_0 = \neg b'_0$ with probability 1/4,
      so $b' = \false$ and $\sevent$ is executed.
  \end{itemize}
\end{itemize}
So
\begin{align*}
&  \Pr[C'[Q \parpop \tproco{b_0}]:\sevent] = \frac{3}{4} \Pr[C[Q \parpop \tprocb{b_0}]:\sevent] + \frac{1}{4} \Pr[C[Q \parpop \tprocb{b_0}]:\sbarevent] \\
&\Pr[C'[Q \parpop \tproco{b_0}]:\sbarevent] = \frac{1}{4} \Pr[C[Q \parpop \tprocb{b_0}]:\sevent] + \frac{3}{4} \Pr[C[Q \parpop \tprocb{b_0}]:\sbarevent]
\end{align*}
Finally, we obtain
\begin{align*}
  &\Pr[C'[Q \parpop \tproco{b_0}]:\sevent] - \Pr[C'[Q \parpop \tproco{b_0}]:\sbarevent]\\
&\quad = \frac{1}{2}(\Pr[C[Q \parpop \tprocb{b_0}]:\sevent]-\Pr[C[Q \parpop \tprocb{b_0}]:\sbarevent])
\end{align*}
so
$\Pr[C[Q \parpop \tprocb{b_0}]:\sevent]-\Pr[C[Q \parpop \tprocb{b_0}]:\sbarevent] = 2(\Pr[C'[Q \parpop \tproco{b_0}]:\sevent] - \Pr[C'[Q \parpop \tproco{b_0}]:\sbarevent]) \leq 2p(C') = 2p(C)$, neglecting the additional runtime of $C'$.
\proofcomplete
\end{proof}

Intuitively, the factor 2 is necessary, because in the definition of one-session secrecy, even if the adversary knows the secret bit $b_0$ perfectly, it will not be able to distinguish $b_0$ from the random bit $y$ in half of the cases, because $b_0$ and $y$ have the same value.

In the rest of Section~\ref{sec:secrbit}, we consider a process
\[Q = \cinput{c}{}; \Res{b_0}{\bool}; \coutput{c}{};Q'\]
and let $Q_L = Q'\{\true/b_0\}$ and $Q_R = Q'\{\false/b_0\}$,
so that $Q$ chooses a random bit $b_0$ and runs as $Q_L$ when $b_0$ is true
and as $Q_R$ when $b_0$ is false.
We assume that $Q_L$ and $Q_R$ never abort, that is, they contain neither $\kw{event\string_abort}$ nor $\FIND\unique{e}$. Moreover, they do not use the variable $b_0$.
We assume that the channels of the inputs at the root of $Q_L$ and $Q_R$ are not used elsewhere in $Q_L$ or $Q_R$.
We have the following lemmas.

\newcommand{\ctxdist}[2]{{#1}+t_{#2}}%

\begin{lemma}\label{lem:secrecy-dist}
  If $Q$ preserves the bit secrecy of $b_0$ with public variables $V$ up to probability $p$, then $Q_L \approx^{V,\emptyset}_{p'} Q_R$ where $p'(C, t_D) = p(\ctxdist{C}{D})$ and the context $\ctxdist{C}{D}$ runs in time $t_C + t_D$ and its other parameters (replication bounds, lengths of bitstrings) are the same as for $C$.
\end{lemma}

The processes $Q_L$ and $Q_R$ can execute different
events without breaking the bit secrecy of $b_0$, because the adversary
for the bit secrecy of $b_0$ does not have access to the events executed
by $Q$. Hence, $Q_L \approx^V_{p'} Q_R$ would not hold in general.

\begin{proof}
  Let $C$ be any acceptable evaluation context for $Q_L$ and $Q_R$
  with public variables $V$, and $D$ a distinguisher such that
  $\cevent(D) \cap \cevent(Q_L,Q_R) = \emptyset$.
Let $c$ and $\cS''$ be a channel that $C$ does not use.

Let $C'$ be a context that outputs on channel $c$, inputs on channel $c$,
runs $C$ but stores events executed by $C$ in its internal state instead of actually executing events, 
computes $D$ on the stored sequence of events executed by $C$ and stores the result in $b'_0$, 
and sends $b'_0$ on channel $\cS''$.
Such a context $C'$ exists because it can be encoded
as a probabilistic Turing machine adversary, which can itself be
encoded as a context in CryptoVerif, as shown in Section~\ref{sec:Turing_adv}.

When $b_0$ is $\true$,
$C'[Q \parpop \tprocb{x}]$ stores in $b'_0$ the result of $C[Q_L]:D$.
When $b'_0 = \true$, $\sevent$ is executed.
When $b'_0 = \false$, $\sbarevent$ is executed.
So
\begin{align*}
&\Pr[C'[Q \parpop \tprocb{x}]:\sevent / b_0 = \true] = \Pr[C[Q_L]:D]\\
&\Pr[C'[Q \parpop \tprocb{x}]:\sbarevent / b_0 = \true] = 1-\Pr[C[Q_L]:D]
\end{align*}

When $b_0$ is $\false$,
$C'[Q \parpop \tprocb{x}]$ stores in $b'_0$ the result of $C[Q_R]:D$.
When $b'_0 = \true$, $\sbarevent$ is executed.
When $b'_0 = \false$, $\sevent$ is executed.
So
\begin{align*}
&\Pr[C'[Q \parpop \tprocb{x}]:\sevent / b_0 = \false] = 1-\Pr[C[Q_R]:D]\\
&\Pr[C'[Q \parpop \tprocb{x}]:\sbarevent / b_0 = \false] = \Pr[C[Q_R]:D]
\end{align*}
Finally, we obtain
\begin{align*}
  &\Pr[C'[Q \parpop \tprocb{x}]:\sevent] - \Pr[C'[Q \parpop \tprocb{x}]:\sbarevent]\\
  &\quad = \frac{1}{2}(\Pr[C[Q_L]:D] + 1-\Pr[C[Q_R]:D] - (1-\Pr[C[Q_L]:D]) - \Pr[C[Q_R]:D])\\
  &\quad = \Pr[C[Q_L]:D] - \Pr[C[Q_R]:D]
\end{align*}
so $\Pr[C[Q_L]:D] - \Pr[C[Q_R]:D] \leq p'(C, t_D)$.
By negating the bit $b'_0$, we swap the events $\sevent$ and $\sbarevent$ without
changing the probability, so $\Pr[C[Q_R]:D] - \Pr[C[Q_L]:D] \leq p'(C, t_D)$.
Therefore, $|\Pr[C[Q_L]:D] - \Pr[C[Q_R]:D]| \leq p'(C, t_D)$.
So $Q_L \approx^{V,\emptyset}_{p'} Q_R$.
\proofcomplete
\end{proof}

\newcommand{\diff}{\kwf{diff}}
\newcommand{\fst}{\kwf{fst}}
\newcommand{\snd}{\kwf{snd}}

Lemma~\ref{lem:secrecy-dist} is the main motivation for the notion of secrecy for a bit: it allows proving indistinguishability between two processes by showing secrecy of bit $b_0$. Using this notion instead of one-session secrecy of $b_0$ avoids losing a factor 2, as shown by Lemma~\ref{lem:secrecy-comp}.

We use this idea to encode the $\diff$ construct originally introduced in ProVerif~\cite{Blanchet07b}: given a process $Q_1$ that contains terms $\diff[M,M']$ and processes $\diff[P,P']$, we define $\fst(Q_1)$ as $Q_1$ with $\diff[M,M']$ replaced with $M$ and $\snd(Q_1)$ as $Q_1$ with $\diff[M,M']$ replaced with $M'$, and similarly for $\diff[P,P']$; the goal is to show that $\fst(Q_1) \approx^{V,\emptyset}_p \snd(Q_1)$ for some $p$ and determine $p$. In order to do that, we define $Q'$ as $Q_1$ with $\diff[M,M']$ replaced with $\iffun(b_0,M,M')$ when $M$ and $M'$ are simple and with $\bguard{b_0}{M}{M'}$ otherwise\footnote{$\iffun(b_0,M,M')$ differs from $\bguard{b_0}{M}{M'}$ in that it evaluates both $M$ and $M'$. Since $\diff[M,M']$ evaluates either $M$ or $M'$ but not both, we translate it into $\bguard{b_0}{M}{M'}$ when the evaluation of $M$ or $M'$ may modify the semantic state, e.g.~by executing an event or by defining a variable. The evaluation of simple terms does not modify the semantic state.}, $\diff[P,P']$ replaced with $\bguard{b_0}{P}{P'}$, and $Q = \cinput{c}{}; \Res{b_0}{\bool}; \coutput{c}{};Q'$. We have $Q_L = Q'\{\true/b_0\} \approx^{V,\emptyset}_0 \fst(Q_1)$ and $Q_R = Q'\{\false/b_0\} \approx^{V,\emptyset}_0 \snd(Q_1)$, so by Lemma~\ref{lem:secrecy-dist}, if $Q$ preserves the bit secrecy of $b_0$ with public variables $V$ up to probability $p$, then $\fst(Q_1) \approx^{V,\emptyset}_{p'} \snd(Q_1)$ where $p'(C, t_D) = p(\ctxdist{C}{D})$.

\begin{lemma}\label{lem:halfobsequi}
  If $Q_L \approx^{V,\emptyset}_{p} Q_R$ then $Q \approx^{V \cup \{b_0\},\emptyset}_{p/2} \cinput{c}{}; \Res{b_0}{\bool}; \coutput{c}{};Q_R$.
\end{lemma}
\begin{proof}
  Let $C$ be an evaluation context acceptable for $Q$ and $\cinput{c}{}; \Res{b_0}{\bool}; \coutput{c}{};Q_R$ with public variables $V \cup \{b_0\}$ 
and $D$ be a distinguisher.
  We have
  \begin{align*}
&   |\Pr[C[Q]:D] - \Pr[C[\cinput{c}{}; \Res{b_0}{\bool}; \coutput{c}{};Q_R]:D|\\
&\quad      \leq \frac{1}{2} |\Pr[C[Q]:D / b_0 = \true] - \Pr[C[\cinput{c}{}; \Res{b_0}{\bool}; \coutput{c}{};Q_R]:D / b_0 = \true ] | \\
&\qquad      + \frac{1}{2} |\Pr[C[Q]:D / b_0 = \false] - \Pr[C[\cinput{c}{}; \Res{b_0}{\bool}; \coutput{c}{};Q_R]:D / b_0 = \false ] |\\
      &\quad      \leq \frac{1}{2} |\Pr[C[\cinput{c}{}; \assign{b_0}{\true} \coutput{c}{};Q_L]:D] - \Pr[C[\cinput{c}{}; \assign{b_0}{\true}\coutput{c}{};Q_R]:D] |\\
&\quad      \leq \frac{1}{2} p(C, t_D)
  \end{align*}
  Indeed, $\Pr[C[\cinput{c}{}; \assign{b_0}{\true} \coutput{c}{};Q_L]:D] = \Pr[C''[Q_L]:D]$ (and similarly for $Q_R$), where $C'' = C[\Reschan{\tup{c}'}; (F_{\tup{c},\tup{c}'} \parpop \Reschan{\tup{c}}; (\cinput{c}{}; \assign{b_0}{\true} \coutput{c}{}; F_{\tup{c'},\tup{c}}) \parpop [\,])]$, $\tup{c}$ are the channels of the inputs at the root of $Q_L$ and $Q_R$ (which we assume not to be used elsewhere in $Q_L$, $Q_R$), $\tup{c'}$ are fresh channels corresponding to channels in $\tup{c}$, and $F_{\tup{c'},\tup{c}}$ forwards all messages sent on a channel in $\tup{c'}$ to the corresponding channel in $\tup{c}$, with replication bounds corresponding to the maximum of the replication bounds in $Q_L$ and $Q_R$. (All messages on channels in $\tup{c}$ are forwarded to channels in $\tup{c}'$ and then back on channels in $\tup{c}$ provided the code $\cinput{c}{}; \assign{b_0}{\true}\coutput{c}{}$ has already been executed. That prevents executing $Q_L$ or $Q_R$ before $\cinput{c}{}; \assign{b_0}{\true}\coutput{c}{}$.)
  By Property~\ref{prop:prob}, replacing $C''$ with $C$ as argument of $p(C, t_D)$ does not affect the probability. (We neglect the additional runtime of $C''$.)
\proofcomplete
\end{proof}

Lemma~\ref{lem:dist-secrecy} is the converse of Lemma~\ref{lem:secrecy-dist}.

\begin{lemma}\label{lem:dist-secrecy}
  If $Q_L \approx^{V,\emptyset}_{p} Q_R$, then $Q$ preserves the bit secrecy of $b_0$ with public variables $V$ up to probability $p'$ where $p'(C) = p(C[C_{\secrbit(b_0)}], t_{\sevent})$.
\end{lemma}
\begin{proof}
  If $Q_L \approx^{V,\emptyset}_{p} Q_R$, then by Lemma~\ref{lem:halfobsequi}, $Q \approx^{V \cup \{b_0\},\emptyset}_{p/2} \cinput{c}{}; \Res{b_0}{\bool}; \coutput{c}{};Q_R$. Moreover, $\cinput{c}{}; \Res{b_0}{\bool}; \coutput{c}{};Q_R$ preserves the bit secrecy of $b_0$ with public variables $V$ up to probability 0. (Since $Q_R$ does not use $b_0$, the variable $b'$ in $Q_{\secrbit(b_0)}$ is independent of $b_0$, so a trace that executes $\sevent$ corresponds to a trace of the same probability and that executes $\sbarevent$ by changing the value of $b_0$, so $\Pr[C[\cinput{c}{}; \Res{b_0}{\bool}; \coutput{c}{};Q_R \parpop Q_{\secrbit(b_0)}] :\sevent] = \Pr[C[\cinput{c}{}; \Res{b_0}{\bool}; \coutput{c}{};Q_R \parpop Q_{\secrbit(b_0)}] :\sbarevent]$.) So by Lemma~\ref{lem:transfersec} (version for bit secrecy), $Q$ preserves the bit secrecy of $b_0$ with public variables $V$ up to probability $p'$.
\proofcomplete
\end{proof}

Lemma~\ref{lem:secrecy-comp-rev} provides a converse of Lemma~\ref{lem:secrecy-comp} when $Q$ has the particular form given above. There is no probability loss in this case.

\begin{lemma}\label{lem:secrecy-comp-rev}
  If $Q$ preserves the bit secrecy of $b_0$ with public variables $V$ up to probability $p$, then $Q$ preserves the one-session secrecy of $b_0$ with public variables $V$ up to probability $p$.
\end{lemma}
\begin{proof}
  If $Q$ preserves the bit secrecy of $b_0$ with public variables $V$ up to probability $p$, then by Lemma~\ref{lem:secrecy-dist}, $Q_L \approx^{V,\emptyset}_{p'} Q_R$ where $p'(C, t_D) = p(\ctxdist{C}{D})$. By Lemma~\ref{lem:halfobsequi}, $Q \approx^{V \cup \{b_0\},\emptyset}_{p'/2} \cinput{c}{}; \Res{b_0}{\bool}; \coutput{c}{};Q_R$.

  Moreover, $\cinput{c}{}; \Res{b_0}{\bool}; \coutput{c}{};Q_R$  preserves the one-session secrecy of $b_0$ with public variables $V$ up to probability 0. Indeed, since $Q_R$ does not use $b_0$, $b_0$ can in fact be chosen in the test query in $Q_{\secrone(b_0)}$, so that test query always returns a random boolean, independently of the value of the variable $b$ of $Q_{\secrone(b_0)}$. Therefore, the variable $b'$ is independent of $b$, so a trace that executes $\sevent$ corresponds to a trace of the same probability and that executes $\sbarevent$ by changing the value of $b$, so $\Pr[C[\cinput{c}{}; \Res{b_0}{\bool}; \coutput{c}{};Q_R \parpop Q_{\secrone(b_0)}] :\sevent] = \Pr[C[\cinput{c}{}; \Res{b_0}{\bool}; \coutput{c}{};Q_R \parpop Q_{\secrone(b_0)}] :\sbarevent]$.
  
  So by Lemma~\ref{lem:transfersec} (version for one-session secrecy), $Q$ preserves the one-session secrecy of $b_0$ with public variables $V$ up to probability $p''$ such that $p''(C) = p'(C[C_{\secrone(b_0)}], t_{\sevent}) = p(\ctxdist{C[C_{\secrone(b_0)}]}{\sevent})$ which is about $p(C)$ by Property~\ref{prop:prob}, neglecting the additional runtime of the context.
\proofcomplete
\end{proof}

\subsubsection{Correspondences}\label{sec:def:corresp}

In this section, we define non-injective and injective correspondences.

\newcommand{\fevent}[1]{\kw{event}(#1)}
\newcommand{\fievent}[1]{\kw{inj\text{-}event}(#1)}

\paragraph{Non-injective Correspondences}

A non-injective correspondence is a property of the form
``if some events have been executed, then some other events have
been executed at least once''. Here, we generalize these correspondences to
implications between logical formulas $\psi \Rightarrow \phi$, 
which may contain events.
We use the following logical formulas:
\begin{defn}
\categ{\phi}{formula}\\
\entry{M}{term}\\
\entry{\fevent{e(M_1, \ldots, M_m)}}{event}\\
\entry{\phi_1 \wedge \phi_2}{conjunction}\\
\entry{\phi_1 \vee \phi_2}{disjunction}
\end{defn}
Terms $M, M_1, \ldots, M_m$ in formulas must contain only
variables $x$ without array indices and function applications,
and their variables are assumed to be distinct from variables of processes.
Formulas denoted by $\psi$ are conjunctions of events.
In a correspondence $\psi \Rightarrow \phi$, the variables of $\psi$ are universally quantified;
those of $\phi$ that do not occur in $\psi$ are existentially quantified.
Formally:

\begin{definition}\label{def:nicorresp}
The semantics of the correspondence 
$\forall \tup{x}:\tup{T}; \psi \Rightarrow \exists \tup{y}:\tup{T}'; \phi$,
also written
$\tup{x}:\tup{T}, \tup{y}:\tup{T}'; \psi \Rightarrow \phi$
in a less explicit syntax, is 
$\sem{\forall \tup{x}:\tup{T}; \psi \Rightarrow \exists \tup{y}:\tup{T}'; \phi} = \sem{\tup{x}:\tup{T}, \tup{y}:\tup{T}'; \psi \Rightarrow \phi} = \forall \tup{x} \in \tup{T}, (\psi \Rightarrow \exists \tup{y} \in \tup{T}', \phi)$,
where $\tup{x} = \fvar(\psi)$ and $\tup{y} = \fvar(\phi)\setminus \fvar(\psi)$.
\end{definition}

The formula $M$ holds when $M$ evaluates to $\true$. The formula
$\fevent{e(M_1, \ldots, M_n)}$ holds when the event $e(M_1, \ldots, M_n)$ 
has been executed. Conjunction, disjunction, implication, existential and 
universal quantifications are defined as usual.
More formally, we write $\venv, \evseqnopp \vdash \varphi$ when the sequence
of events $\evseqnopp$ satisfies the formula $\varphi$, in the environment
$\venv$ that maps variables to their values. We define $\venv,
\evseqnopp \vdash \varphi$ as follows:
\begin{tabbing}
$\venv,\evseqnopp \vdash M$ if and only if $\venv, M \evalterm \true$\\
$\venv,\evseqnopp \vdash \fevent{e(M_1, \ldots, M_m)}$ if and only if\\*
\qquad for all $j \leq m$, $\venv, M_j \evalterm a_j$ and 
$e(a_1, \ldots, a_m) \in \evseqnopp$\\
$\venv, \evseqnopp \vdash \varphi_1 \wedge \varphi_2$ if and only if 
$\venv, \evseqnopp \vdash \varphi_1$ and $\venv, \evseqnopp \vdash \varphi_2$\\
$\venv, \evseqnopp \vdash \varphi_1 \vee \varphi_2$ if and only if 
$\venv, \evseqnopp \vdash \varphi_1$ or $\venv, \evseqnopp \vdash \varphi_2$\\
$\venv, \evseqnopp \vdash \varphi_1 \Rightarrow \varphi_2$ if and only if 
$\venv, \evseqnopp \vdash \varphi_1$ implies $\venv, \evseqnopp \vdash \varphi_2$\\
$\venv, \evseqnopp \vdash \exists x \in T, \varphi$ if and only if 
there exists $a\in T$ such that $\venv[x \mapsto a], \evseqnopp \vdash \varphi$\\
$\venv, \evseqnopp \vdash \forall x \in T, \varphi$ if and only if 
for every $a\in T$, we have $\venv[x \mapsto a], \evseqnopp \vdash \varphi$
\end{tabbing}
When $\varphi$ is a closed formula, we write $\evseqnopp \vdash \varphi$ for
$\venv, \evseqnopp \vdash \varphi$ where $\venv$ is the empty function.

\begin{definition}\label{def:correspsat}
The sequence of events 
$\evseqnopp$ \emph{satisfies the correspondence} $\varphi$ if
and only if $\evseqnopp \vdash \varphi$.
\end{definition}

\begin{definition}\label{def:proccorresp}
We define a distinguisher $D(\evseqnopp) = \true$
if and only if $\evseqnopp \vdash \varphi$, and we denote
this distinguisher $D$ simply by $\varphi$.

The advantage of the adversary $C$ against the \emph{correspondence} $\varphi$ in process $Q$
is $\Advt_Q^{\varphi}(C) = \Pr[C[Q] : \neg \varphi]$,
where $C$ is an evaluation context acceptable for $Q$ with any public variables 
that does not contain events used by $\varphi$.

The process $Q$  \emph{satisfies the correspondence} $\varphi$ with public variables $V$ up to probability $p$
if and only if for all
evaluation contexts $C$ acceptable for $Q$ with public variables $V$
that do not contain events used by $\varphi$, 
$\Advt_Q^{\varphi}(C) \leq p(C)$.
\end{definition}

When $\prop$ is a correspondence $\varphi$, we define $C_{\prop} = [\,]$
and the events used by $\prop$ are the events that occur
in the formula $\varphi$.
Therefore, the definition of ``$Q$ satisfies the correspondence $\varphi$''
matches the definition of ``$Q$ satisfies $\prop$'' given in
Definition~\ref{def:secr}.

A process satisfies $\varphi$ up to probability $p$
when the probability that
it generates a sequence of events $\evseqnopp$ that does not satisfy 
$\varphi$ is at most $p(C)$, in the presence of an
adversary represented by the context $C$.

\begin{example}
%Referring to the example $G_0$ of Section~\ref{sec:calculus},
The semantics of the correspondence
\begin{equation}
\forall x:\pkey, y: \host, z:\nonce; \fevent{e_B(x,y,z)}\Rightarrow \fevent{e_A(x,y,z)}\label{c1}
\end{equation} 
is 
\begin{equation}
\forall x \in \pkey, \forall y \in  \host, \forall z\in \nonce, \fevent{e_B(x,y,z)}\Rightarrow \fevent{e_A(x,y,z)}\label{c1-full}
\end{equation} 
It means that, with overwhelming probability, for all $x,y,z$, if 
$e_B(x,y,z)$ has been executed, then $e_A(x,y,z)$ has been executed.

The semantics of the correspondence
\[\begin{split}
\forall x:T; {}&\fevent{e_1(x)} \wedge \fevent{e_2(x)} \Rightarrow\\
&\quad \exists y:T'; \fevent{e_3(x)} \vee (\fevent{e_4(x,y)} \wedge \fevent{e_5(x,y)})
\end{split}\]
is 
\[\begin{split}
&\forall x\in T, \fevent{e_1(x)} \wedge \fevent{e_2(x)} \Rightarrow\\
&\quad \exists y \in T', \fevent{e_3(x)} \vee (\fevent{e_4(x,y)} \wedge \fevent{e_5(x,y)})
\end{split}\]
It means that, with overwhelming probability, for all $x$, if $e_1(x)$ and $e_2(x)$ have been executed, 
then $e_3(x)$ has been executed or there
exists $y$ such that both $e_4(x,y)$ and $e_5(x,y)$ have been
executed.
\end{example}

\paragraph{Injective Correspondences}

\newcommand{\step}{\tau}
\newcommand{\phistep}{\phi^{\step}}
\newcommand{\psistep}{\psi^{\step}}
\newcommand{\Fstepfun}{f}
\newcommand{\Inj}{\mathrm{Inj}}

% or talk about ``distinct events''?
Injective correspondences are properties of the form
``if some event has been executed $n$ times, then some other events
have been executed at least $n$ times''. In order to model them
in our logical formulas, we extend the grammar of formulas $\phi$
with injective events $\fievent{e(M_1, \ldots, M_m)}$.
The formula $\psi$ is a conjunction of (injective or non-injective)
events.
The conditions on the number of executions of events apply
only to injective events.

The definition of formula satisfaction is also extended, to be able to
indicate at which step an event has been executed (that is, at which
index it appears in $\evseqnopp$): $\fevent{e(\tup{M})}@\step$ means that
event $e(\tup{M})$ has been executed at step $\step$. Formally:
\begin{tabbing}
$\venv,\evseqnopp \vdash \fevent{e(M_1, \ldots, M_m)}@M_0$ if and only if\\
\qquad for all $j \leq m$, $\venv, M_j \evalterm a_j$, 
$a_0 \neq \bot$, and 
$e(a_1, \ldots, a_m) = \evseqnopp(a_0)$
\end{tabbing}
With this definition, we have:

\begin{definition}\label{def:icorresp}
The semantics of the correspondence 
$\forall \tup{x}:\tup{T}; \psi \Rightarrow \exists \tup{y}:\tup{T}'; \phi$,
also written
$\tup{x}:\tup{T}, \tup{y}:\tup{T}'; \psi \Rightarrow \phi$ 
in a less explicit syntax, is 
\[\begin{split}
&\sem{\forall \tup{x}:\tup{T}; \psi \Rightarrow \exists \tup{y}:\tup{T}'; \phi} =
\sem{\tup{x}:\tup{T}, \tup{y}:\tup{T}'; \psi\Rightarrow \phi} = 
\exists \Fstepfun_1, \dots, \Fstepfun_k\in \mathbb{N}^m \times \prod \tup{T} \rightarrow \mathbb{N} \cup \{\bot\}, \\
&\quad \Inj(I,\Fstepfun_1) \wedge \dots \wedge \Inj(I,\Fstepfun_k) \wedge \forall \step_1, \dots, \step_m\in \mathbb{N}, \forall \tup{x}\in\tup{T}, (\psistep \Rightarrow \exists \tup{y}\in\tup{T}', \phistep)\,,
\end{split}\]
where $\tup{x} = \fvar(\psi)$, $\tup{y} = \fvar(\phi)\setminus \fvar(\psi)$,
$\psi = F_1\wedge \dots \wedge F_m$,
$\psistep = F_1^{\step} \wedge \dots \wedge F_m^{\step}$,
$F_j^{\step} = \fevent{e(\tup{M})}@\step_j$ if $F_j = \fevent{e(\tup{M})}$ or
$F_j = \fievent{e(\tup{M})}$,
$I = \{ j \mid F_j = \fievent{\dots}\}$, 
$\phistep$ is obtained from $\phi$ by replacing each injective event
$\fievent{e(\tup{M})}$ with $\fevent{e(\tup{M})}@\Fstepfun_j(\step_1, \ab \dots, \ab \step_m, \ab \tup{x})$ using a distinct function $\Fstepfun_j$ for each injective event in $\phistep$, and
$\Inj(I,f)$ if and only if $f(\step_1, \dots, \step_m, \tup{x}) = f(\step'_1, \dots, \step'_m, \tup{x}') \neq \bot \Rightarrow \forall j \in I, \ab \step_j = \step'_j$.
\end{definition}

In $\psistep$ and $\phistep$, events are labeled with their associated
execution step, $\step_j$ for the events in $\psistep$ and
$\Fstepfun_j(\step_1, \ab \dots, \ab \step_m, \ab \tup{x})$ for the
injective events in $\phistep$.
Therefore, the functions $\Fstepfun_j$ map the execution steps of
events in $\psi$, $\step_1$, \dots, $\step_m$, and the values
of the variables in $\psi$, $\tup{x}$, to the associated
execution steps of injective events in $\phi$. (The result $\bot$
corresponds to the case in which the event in $\phi$ is not executed: 
in case of disjunctions, not all events in $\phi$ are required to be executed.)
The correspondence is injective when these functions $\Fstepfun_j$
are injective in their arguments that correspond to injective events
in $\psi$. The indices of injective events in $\psi$ are collected
in the set $I$, and injectivity is guaranteed by $\Inj(I,f_j)$,
which means that, ignoring the result $\bot$, $f_j$ is injective
in its arguments of indices in $I$.

Definition~\ref{def:proccorresp} is unchanged for injective correspondences.

\begin{example}
%Referring to the example $G_0$ of Section~\ref{sec:calculus}, 
The semantics of the correspondence
\begin{equation}
\forall x:\pkey, y: \host, z:\nonce; \fievent{e_B(x,y,z)}\Rightarrow \fievent{e_A(x,y,z)}\label{c2}
\end{equation}
is 
\begin{equation}
\begin{split}
  &\exists \Fstepfun\in \mathbb{N} \times \pkey\times\host\times\nonce\rightarrow \mathbb{N} \cup \{\bot\}, \Inj(\{1\},\Fstepfun) \wedge {}\\
  &\quad \forall \step\in\mathbb{N}, \forall x\in\pkey, \forall y\in \host, \forall z\in\nonce, \\
&\qquad \fevent{e_B(x,y,z)}@\step\Rightarrow \fevent{e_A(x,y,z)}@\Fstepfun(\step,x,y,z)
\end{split}\label{c2-full}
\end{equation} 
It means that, with overwhelming probability,
each execution of $e_B(x,y,z)$ corresponds to a distinct execution of 
$e_A(x,y,z)$.
In this case, $\Fstepfun$ is a function that maps 
the execution step $\step$ of $e_B(x,y,z)$ and the variables $x$, $y$, $z$ 
to the execution step of $e_A(x,y,z)$. (This step is
never $\bot$.) This function is injective in its first argument, the step $\step$,
so if there are $n$ executions of $e_B(x,y,z)$, at steps $\step_1$, \dots, $\step_n$,
then there are at least $n$ executions of $e_A(x,y,z)$, at steps
$\Fstepfun(\step_1,x,y,z)$, \dots, $\Fstepfun(\step_n,x,y,z)$ and these steps are distinct
by injectivity of $\Fstepfun$ in its first argument.

The semantics of the correspondence
\[\begin{split}
&\forall x:T; \fevent{e_1(x)} \wedge \fievent{e_2(x)} \Rightarrow \exists y:T';\fievent{e_3(x)} \vee {}\\
&\qquad (\fievent{e_4(x,y)} \wedge \fievent{e_5(x,y)})
\end{split}\] 
is 
\[\begin{split}
&\exists \Fstepfun_1, \Fstepfun_2, \Fstepfun_3\in \mathbb{N}^2 \times T\rightarrow \mathbb{N} \cup \{\bot\}, \Inj(\{2\},\Fstepfun_1) \wedge \Inj(\{2\},\Fstepfun_2) \wedge \Inj(\{2\},\Fstepfun_3) \wedge {}\\
&\forall \step_1, \step_2\in\mathbb{N}, \forall x\in T,\fevent{e_1(x)}@\step_1 \wedge \fevent{e_2(x)}@\step_2 \Rightarrow \exists y\in T', \fevent{e_3(x)}@\Fstepfun_1(\step_1,\step_2,x) \vee {}\\
&\qquad (\fevent{e_4(x,y)}@\Fstepfun_2(\step_1,\step_2,x) \wedge \fevent{e_5(x,y)}@\Fstepfun_3(\step_1,\step_2,x))
\end{split}\] 
It means that, with overwhelming probability, for all $x$, if $e_1(x)$ has been executed, 
then each execution of $e_2(x)$ corresponds to distinct
executions of $e_3(x)$ or to distinct executions of $e_4(x,y)$
and $e_5(x,y)$.
The functions $\Fstepfun_1$, $\Fstepfun_2$, and $\Fstepfun_3$ map 
the execution steps $\step_1$ and $\step_2$ of $e_1$ and $e_2$ and the variable $x$
to the execution steps of $e_3$, $e_4$, and $e_5$ respectively.
Ignoring the result $\bot$, they are injective in their second argument, 
which corresponds to the execution step of the injective event $e_2$.
\end{example}

When no injective event occurs in $\forall \tup{x}:\tup{T}; \psi \Rightarrow \exists \tup{y}:\tup{T}';\phi$, 
Definition~\ref{def:icorresp} reduces to the
definition of non-injective correspondences: there are no functions $\Fstepfun_j$, $\phistep = \phi$, and $\fevent{e(\tup{M})}@\step$ holds for some $\step$
if and only if $\fevent{e(\tup{M})}$ holds, so $\psistep$ holds for some $\step_1$, \dots, $\step_m$ if and only if $\psi$ holds. 

\paragraph{Well-formedness condition}

When we consider a correspondence $\forall \tup{x}:\tup{T}; \psi
\Rightarrow \exists \tup{y}:\tup{T}';\phi$, with $\tup{x} = \fvar(\psi)$ and $\tup{y} = \fvar(\phi)\setminus \fvar(\psi)$,
we should have
\begin{equation}
\forall \tup{x} \in \tup{T},\forall \tup{x}'\in\tup{T}, \forall\tup{y}\in\tup{T}',
\psi = \psi\{\tup{x}'/\tup{x}\} \Rightarrow \phi = \phi\{\tup{x}'/\tup{x}\} 
\label{eq:corresp-wf}
\end{equation}
where $\tup{x}'$ are fresh variables, and the equality of terms is the equality
of their values, but disjunctions, conjunctions, and events are considered
syntactically. This condition guarantees that, given an execution of events
in $\psi$, the formula to verify $\exists \tup{y}\in \tup{T}', \phistep$ is uniquely
determined. It avoids pathological correspondences such as
\begin{equation}
\forall x:T; \fevent{e(f(x)} \Rightarrow \fevent{e'(x)}
\label{c5}
\end{equation}
with $f(a) = f(b) = c$, for which $\fevent{e(c)}$ corresponds to 
both $\fevent{e(f(a)}$ and $\fevent{e(f(b)}$, so when event $e(c)$ is executed, $x$ can take both values
$a$ and $b$, so \eqref{c5} requires the execution of events $e'(a)$
and $e'(b)$. An even more pathological case is when $f(x) = c$ for all $x$:
in this case, when event $e(c)$ is executed, \eqref{c5} requires the execution of event $e'(x)$ for all $x\in T$,
which is impossible when $T$ is infinite. However, the condition allows the
correspondence~\eqref{c5} when $f$ is injective, so $x$ is uniquely determined, and it
also allows the correspondences
\begin{equation}
\forall x:T; \fevent{e(f(x)} \Rightarrow \false
\label{c3}
\end{equation}
and
\begin{equation}
\forall x:T; \fevent{e(f(x)} \Rightarrow \fevent{e'(f(x))}
\label{c4}
\end{equation}
for any function $f$: \eqref{c3} requires that event $e(y)$ is never executed with $y$ in the image of $f$,
and \eqref{c4} requires that $e'(y)$ is executed when $e(y)$ is executed with $y$ in the image of $f$.

CryptoVerif displays a warning when it does not manage to prove the well-formedness condition~\eqref{eq:corresp-wf}.

\paragraph{Property}

\begin{lemma}\label{lem:transfercorr}
  If $Q$ satisfies a correspondence $\corresp$ with public variables
  $V$ up to probability $p$ and $C$ is an acceptable evaluation context for $Q$ with public
  variables $V$ that does not contain events used in $\corresp$, then
  for all $V' \subseteq V \cup \fvar(C)$, $C[Q]$ satisfies a
  correspondence $\corresp$ with public variables $V'$ up to probability $p'$ such that
  $p'(C') = p(C'[C])$.

  If $Q \approx^{V,\evset}_p Q'$, $Q$ satisfies a correspondence $\corresp$
  with public variables $V$ up to probability $p'$, and $\evset$ contains all events in $\corresp$, then $Q'$ satisfies $\corresp$
  with public variables $V$ up to probability $p''$ such that 
  $p''(C) = p'(C) + p(C, t_{\corresp})$.
\end{lemma}
\begin{proof}
Suppose that $Q$ satisfies a correspondence $\corresp$ with public variables
  $V$ and $C$ is an acceptable evaluation context for $Q$ with public
  variables $V$ that does not contain events used in $\corresp$.
Let $V' \subseteq V \cup \fvar(C)$.
Let $C'$ be an evaluation context acceptable for $C[Q]$ with public variables $V'$
that does not contain events used by $\corresp$.
We rename the variables of $C'$ not in $V'$ so that they are not in $V$; by Property~\ref{prop:prob}, this renaming does not change the probabilities.
We have 
\[\Advt_{C[Q]}^{\corresp}(C') = \Pr[C'[C[Q]] : \neg \corresp] \leq p(C'[C])\]
because $C'[C]$ is an evaluation context acceptable for $Q$ with public variables $V$:
there is no common table between $C$ and $Q$, and between $C'$ and $C[Q]$,
so a fortiori between $C'$ and $Q$, so there is
no common table between $C'[C]$ and $Q$; moreover
\begin{align*}
\fvar(C'[C]) \cap \fvar(Q) 
&= ((\fvar(C') \cap \fvar(Q)) \cup \fvar(C)) \cap \fvar(Q)\\
&\subseteq (V' \cup \fvar(C)) \cap \fvar(Q) \tag*{since $\fvar(C') \cap \fvar(C[Q]) \subseteq V'$}\\
&\subseteq (V \cup \fvar(C)) \cap \fvar(Q) \tag*{since $V' \subseteq V \cup \fvar(C)$}\\
&\subseteq V \tag*{since $\fvar(C) \cap \fvar(Q) \subseteq V$}
\end{align*}
We also have
$\vardef(C'[C[\,]]) \cap V = (\vardef(C') \cap V) \cup (\vardef(C) \cap V) = \emptyset$
since $\vardef(C) \cap V = \emptyset$ because $C$ is an acceptable
evaluation context for $Q$ with public variables $V$ and
$\vardef(C') \cap V \subseteq \vardef(C') \cap V' = \emptyset$ because
we have renamed the variables of $C'$ not in $V'$ so that they are not
in $V$ and $C'$ is an acceptable evaluation context for $C[Q]$ and
with public variables $V'$.

Suppose that $Q \approx^{V,\evset}_p Q'$, $Q$ satisfies a
correspondence $\corresp$ with public variables $V$ up to probability
$p'$, and $\evset$ contains all events in $\corresp$.
Let $C$ be an evaluation context acceptable for $Q'$ with public variables $V$
that does not contain events used by $\corresp$.
We have 
\[\begin{split}
\Advt_{Q'}^{\corresp}(C) 
&= \Pr[C[Q'] : \neg \corresp]\\
&\leq \Pr[C[Q] : \neg \corresp] + | \Pr[C[Q'] : \neg \corresp] - \Pr[C[Q] : \neg \corresp] |\\
&\leq p'(C) + p(C, t_{\corresp})
\end{split}\]
Indeed, by renaming the variables and tables of $C$ that do not appear in $Q'$ to variables and tables 
that also do not occur in $Q$,
$C$ is also an acceptable evaluation context for $Q$
with public variables $V$.
Furthermore, by Property~\ref{prop:prob}, this renaming does not change the probabilities.
\proofcomplete
\end{proof}

\paragraph{Reachability secrecy}

Reachability secrecy aims to show that the adversary cannot compute
the secret value. This notion is standard in the symbolic model,
but less common than the notion of secrecy as ``the adversary cannot
distinguish the secret from a random value'' (Section~\ref{sec:def:secrecy})
in the computational model. It is still used, e.g.\ in the property
of one-wayness or in the computational Diffie-Hellman assumption.

This notion makes sense only when the secret value is of a large type.
Otherwise, the adversary would have a non-negligible probability
of finding the secret value just by random guessing.

This notion is in fact encoded as a correspondence property.
We distinguish two variants.
One-session reachability secrecy of $x$ means that the adversary
cannot compute any cell of array $x$, even if it has access to the
public variables in $V$.
Reachability secrecy of $x$ means that the adversary cannot
compute any cell of array $x$, even if it has access to the other
cells of $x$ and to the public variables in $V$.

\begin{definition}[(One-session) reachability secrecy]\label{def:reachsecr}
  Let $Q$ be a process, $x$ a variable, and $V$ a set of variables.
  Let
  {\allowdisplaybreaks\begin{align*}
    \tproc{\secrreachone(x)} = {} & \repl{\iT}{\nT}\cinput{\cS[\iT]}{x':T, \vf_1:[1, n_1], \ldots, \vf_m:[1, n_m]}; \\*
    &\baguard{\defined(x[\vf_1, \ldots, \vf_m]) \wedge x' = x[\vf_1, \ldots, \vf_m]}{\kevent{\kwf{adv\_has\_x}}}\\
    \tproc{\secrreach(x)} = {} &\repl{\iR}{\nR}\cinput{\cR[\iR]}{\vf'_1:[1, n_1], \ldots, \vf'_m:[1, n_m]}; \adeftest{x[\vf'_1, \ldots, \vf'_m]} \\*
    &\assign{\reveal:\bool}{\true} \coutput{\cR[\iR]}{x[\vf'_1, \ldots, \vf'_m]} \\*
{}\parpop{} &\repl{\iT}{\nT}\cinput{\cS}{x':T, \vf_1:[1, n_1], \ldots, \vf_m:[1, n_m]}; \\*
&\FIND \ i \leq \iR\ \SUCHTHAT \ \kw{defined}(\reveal[i],\vf'_1[i],\dots,\vf'_m[i]) \wedge {}\\*
&\qquad\vf'_1[i] = \vf_1 \wedge \dots \wedge \vf'_m[i] = \vf_m \THEN \kw{yield}\ \ELSE\\*
    &\baguard{\defined(x[\vf_1, \ldots, \vf_m]) \wedge x' = x[\vf_1, \ldots, \vf_m]}{\kevent{\kwf{adv\_has\_x}}}
  \end{align*}}%
where $\cS, \cR \notin\fc(Q)$, $x', \vf_1, \ldots, \vf_m, \vf'_1, \ldots, \vf'_m, \reveal \notin \fvar(Q) \cup V$, $\kwf{adv\_has\_x}$ does not occur in $Q$, 
and $\tyenv(x) = [1,n_1] \times \ldots \times [1,n_m] \rightarrow T$.

  The process $Q$ \emph{satisfies one-session reachability secrecy of} $x$ with public variables $V$ ($x \notin V$) up to probability $p$ if and only if the process $Q \parpop \tproc{\secrreachone(x)}$ satisfies the correspondence $\fevent{\kwf{adv\_has\_x}} \Rightarrow \false$ with public variables $V$ up to probability $p$.

  The process $Q$ \emph{satisfies reachability secrecy of} $x$ with public variables $V$ ($x \notin V$) up to probability $p$ if and only if $Q \parpop \tproc{\secrreach(x)}$ satisfies the correspondence $\fevent{\kwf{adv\_has\_x}} \Rightarrow \false$ with public variables $V$ up to probability $p$.
\end{definition}

The process $\tproc{\secrreachone(x)}$ waits on channel $\cS[\iT]$ for a candidate
value $x'$ and indices $\vf_1, \dots, \vf_m$. If $x[\vf_1, \dots, \vf_m]$ is defined and equal to $x'$, the adversary managed to compute $x[\vf_1, \dots, \vf_m]$, hence to break one-session reachability secrecy. In this case, we execute event $\kwf{adv\_has\_x}$, and our goal will be to bound the probability of this event, by showing the correspondence $\fevent{\kwf{adv\_has\_x}} \Rightarrow \false$.

The process $\tproc{\secrreach(x)}$ additionally provides a reveal query: by sending indices $\vf'_1, \dots, \vf'_m$ on channel $\cR[\iR]$, the adversary can obtain the value of $x[\vf'_1, \dots, \vf'_m]$ if it is defined. Obviously, the adversary breaks reachability secrecy if it computes $x[\vf_1, \dots, \vf_m]$ without having first made a successful reveal query on the indices $\vf_1, \dots, \vf_m$. The absence of such a reveal query is verified by $\FIND\ i \leq \iR \dots$ before executing event $\kwf{adv\_has\_x}$.

The bounds on the number of queries ($\nT$, $\nR$) are chosen large enough that they do not limit the adversary.

\subsubsection{Computation of Advantages}\label{sec:computadv}

\begin{definition}\label{def:adv}
Let $\prop$ be a security property: $\prop$ is $\secrone(x)$, $\secr(x)$, $\secrbit(x)$, or
a trace property, represented by any distinguisher that does not use $\sevent$, $\sbarevent$, nor non-unique events. (Trace properties include correspondences $\corresp$, as well as $\true$, the property that is always true.)
Let $D$ be a disjunction of Shoup and non-unique events that does not contain $\sevent$ nor $\sbarevent$.
Let $C$ be an evaluation context acceptable for $Q$ with any public variables.

When $\prop$ is a trace property, we define $V_{\prop} = \emptyset$ and
\begin{equation}
  \Advtev{Q}{\prop}{C}{D} = \Pr[C[Q] : (\neg\prop \vee D) \wedge \neg \nonunique{Q,D}]\label{eq:adv:corr}
\end{equation}

When $\prop$ is $\secrone(x)$, $\secr(x)$, or $\secrbit(x)$, we define $V_{\prop} = \{x\}$ and
\begin{equation}
\Advtev{Q}{\prop}{C}{D} = \Pr[C[Q] : \sevent \vee D] - \Pr[C[Q] : \sbarevent \vee \nonunique{Q,D}]\label{eq:adv:secr}
\end{equation}

We write $\bound{Q}{V}{\prop}{D}{p}$ when $V_{\prop} \subseteq V$ and for all evaluation contexts $C'$ acceptable for $C_{\prop}[Q]$ with public variables $V \setminus V_{\prop}$ that does not contain events used by $\prop$ or $D$ nor non-unique events in $Q$, we have
$\Advtev{Q}{\prop}{C}{D} \leq p(C)$ for $C = C'[C_{\prop}[\,]]$.
\end{definition}

The events $\sevent$ and $\sbarevent$ are only executed by
$\kw{event\_abort}$.  In~\eqref{eq:adv:secr}, we could write $(\sevent
\vee D) \wedge \neg \nonunique{Q,D}$ instead of $\sevent \vee D$, to
be more similar to~\eqref{eq:adv:corr}.  That would be equivalent
because the game immediately aborts after executing $\sevent$ as well
as events in $D$ and in $\nonunique{Q,D}$, so only one of these events
is executed.
In~\eqref{eq:adv:secr}, we expect that $C = C'[C_{\prop}[\,]]$ for some context $C'$. This is what happens in the
definition of $\bound{Q}{V}{\prop}{D}{p}$. In that definition, $C$ is
a context acceptable for $Q$ with public variables $V$.

%
%% We assume that $D$ is written as a logical formula (for instance, one
%% of the correspondence formulas defined previously). 
%% In the following
%% lemmas, the events used by $D$ are the events that occur in this
%% formula.
%% We consider $\vee$ as commutative and associative, with $\Dfalse$
%% as neutral element, so that disjuncts can be reorganized
%% to put $D$ in the form $D_0 \vee D_1$.

\begin{lemma}\label{lem:adv}
\begin{enumerate}

\item\label{item:prop:init}
  In the initial game $Q$, let
  %$\usedevents$ consist of the events that occur in $Q$ or in $D_0$. Let
  $D_U = \bigvee \{ e \mid \unique{e}\text{ occurs in }Q\} = \nonunique{Q}$.

  If $\bound{Q}{V}{\prop}{D_U}{p}$, then $Q$ satisfies property $\prop$
  with public variables $V\setminus V_{\prop}$ up to probability $p'$ where
  $p'(C') = p(C'[C_{\prop}[\,]])$.

%%   If $C$ is an evaluation context acceptable for $C_{\prop}[Q]$
%% with public variables $V$ that does not contain events used by $\prop$,
%% then $\Advt_{Q}^{\prop}(C) \leq \Advtev{Q}{\prop}{C[C_{\prop}[\,]]}{D_U}$. 

\item\label{item:prop:step}
If $\dset, \dsetsnu: Q, D, \usedevents \indistev{V}{p} Q', D', \usedevents'$, 
the events $\sevent$ and $\sbarevent$ are not in $\usedevents'$,
$\bound{Q'}{V}{\prop}{D'}{p'}$, and
\begin{itemize}
\item either $\prop$ is a trace property, $\neg\prop \in \dset$,
  and $p''(C) = p(C, t_{\neg\prop}) + p'(C)$;
\item or $\prop$ is $\secrone(x)$, $\secr(x)$, or $\secrbit(x)$, $\{ \sevent, \neg\sbarevent \} \subseteq \dset$, and $p''(C) = 2 p(C,\ab t_\sevent) + p'(C)$
\end{itemize}
then $\bound{Q}{V}{\prop}{D}{p''}$.

\item\label{item:prop:split}
Let $D$ and $D'$ be disjunctions of Shoup and non-unique events that do not contain $\sevent$ nor $\sbarevent$.
\begin{itemize}
\item If $\prop$ is a trace property,
  $\bound{Q}{V}{\prop}{D}{p}$, and
  $\bound{Q}{V}{\true}{D'}{p'}$, then we have
  $\bound{Q}{V}{\prop}{D \vee D'}{p + p'}$.
\item If $\prop$ is $\secrone(x)$, $\secr(x)$, or $\secrbit(x)$,
  $\Dnu' = \bigvee \{ e \mid e$ occurs in $D'$ and $e$ is a non-unique event$\}$,
  $\bound{Q}{V}{\prop}{D}{p}$,
  $\bound{Q}{V}{\true}{D'}{p'}$, and
  $\bound{Q}{V}{\ab \true}{\ab \Dnu'}{\ab p''}$, then
  $\bound{Q}{V}{\prop}{D \vee D'}{p + p' + p''}$.
\end{itemize}

\item\label{item:prop:merge}
If $\dset, \dsetsnu: Q, D, \usedevents \indistev{V}{p} Q', D', \usedevents'$,
the distinguisher $D''$ is a disjunction of Shoup and non-unique events in $\usedevents$ that does not contain $\sevent$ nor $\sbarevent$,
and $\bound{Q}{V}{\ab \true}{\ab D''}{\ab p'}$, then
$\dset, \dsetsnu: Q, D \vee D'', \usedevents \indistev{V}{p+p'} Q', D', \usedevents'$.

\end{enumerate}
\end{lemma}
\begin{proof}
  Property~\ref{item:prop:init}: We have $\nonunique{Q,D_U} = \Dfalse$.
Let $C'$ be an evaluation context acceptable for $C_{\prop}[Q]$
with public variables $V\setminus V_{\prop}$ that does not contain events used by $\prop$,
and $C = C'[C_{\prop}[\,]]$.
Since $\bound{Q}{V}{\prop}{D_U}{p}$, we have
$\Advtev{Q}{\prop}{C}{D_U} \leq p(C)$.
\begin{itemize}
\item In case $\prop$ is a trace property,
$\Advt_Q^{\prop}(C') = \Pr[C'[Q]: \neg\prop] \leq \Pr[C'[Q]: \neg\prop \vee D_U] = \Pr[C'[C_{\prop}[Q]]: \neg\prop \vee D_U] = \Advtev{Q}{\prop}{C'[C_{\prop}[\,]]}{D_U}$. 
\item In case $\prop$ is $\secrone(x)$, $\secr(x)$, or $\secrbit(x)$,
$\Advt_{Q}^{\prop}(C') = \Pr[C'[C_{\prop}[Q]] : \sevent] - \Pr[C'[C_{\prop}[Q]] : \sbarevent] \leq
\Pr[C'[C_{\prop}[Q]] : \sevent \vee D_U] -
\Pr[C'[C_{\prop}[Q]] : \sbarevent]
= \Advtev{Q}{\prop}{C'[C_{\prop}[\,]]}{D_U}$.
\end{itemize}
In both cases, $\Advt_{Q}^{\prop}(C') \leq \Advtev{Q}{\prop}{C'[C_{\prop}[\,]]}{D_U} = \Advtev{Q}{\prop}{C}{D_U} \leq p(C) = p'(C')$, so $Q$ satisfies property $\prop$ with public variables $V\setminus V_{\prop}$ up to probability $p'$.

Property~\ref{item:prop:step}, case $\prop$ is a trace property:
Let $C$ be an evaluation context acceptable for $Q$ with public variables $V$
that does not contain events used by $\prop$ nor $D$ nor non-unique events
of $Q$.
Let $C'$ be obtained by renaming the variables and tables of $C$ that do not occur in $Q$ to variables and tables that also do not occur in $Q'$, and by renaming the events of $C$ so that they are not in $\usedevents'$.
The context $C'$ is then acceptable for $Q$ and $Q'$ with public variables $V$ and does not contain events in $\usedevents'$.
%The distinguishers $D$ and $D'$ do not contain $\sevent$ nor $\sbarevent$ since $\sevent$ and $\sbarevent$ are not in $\usedevents'$, and
Since $\dset, \dsetsnu: Q, D, \usedevents \indistev{V}{p} Q', D', \usedevents'$ and
$\neg\prop\in \dset$, we have by taking $D_1 = \Dfalse$,
\begin{align*}
\Advtev{Q}{\prop}{C}{D} 
&= \Pr[C[Q] : (\neg\prop \vee D) \wedge \neg \nonunique{Q,D} ] \\
&= \Pr[C'[Q] : (\neg\prop \vee D) \wedge \neg \nonunique{Q,D} ]
\tag*{since the renaming of events does not affect the events of the distinguisher}\\
&\leq p(C',t_{\neg\prop}) + \Pr[C'[Q'] : (\neg\prop \vee D') \wedge \neg \nonunique{Q',D'}]\\
&\leq p(C', t_{\neg\prop}) + \Advtev{Q'}{\prop}{C'}{D'}\\
&\leq p(C', t_{\neg\prop}) + p'(C')
\tag*{since $\bound{Q'}{V}{\prop}{D'}{p'}$}\\
&\leq p(C, t_{\neg\prop}) + p'(C)
\tag*{since the renaming does not modify the probability formulas by
  Property~\ref{prop:prob}}\\
&\leq p''(C)
\end{align*}
so $\bound{Q}{V}{\prop}{D}{p''}$.

Property~\ref{item:prop:step}, case $\prop$ is $\secrone(x)$, $\secr(x)$, or $\secrbit(x)$:
Let $C'$ be an evaluation context acceptable for $C_{\prop}[Q]$ with public variables $V \setminus V_{\prop}$
that does not contain $\sevent$, $\sbarevent$, events used by $D$, nor non-unique events
of $Q$. Let $C = C'[C_{\prop}[\,]]$.
Let $C''$ be obtained by renaming the variables and tables of $C$ that do not occur in $Q$ to variables and tables that also do not occur in $Q'$, and by renaming the events of $C$ so that they are not in $\usedevents'$. (The events $\sevent$ and $\sbarevent$ are left unchanged by this renaming; they are not in $\usedevents'$.)
The context $C''$ is then acceptable for $Q$ and $Q'$ with public variables $V$ and does not contain events in $\usedevents'$.
%The distinguishers $D$ and $D'$ do not contain $\sevent$ nor $\sbarevent$ since $\sevent$ and $\sbarevent$ are not in $\usedevents'$, and
Then we have 
\begin{align*}
\Advtev{Q}{\prop}{C}{D} 
&= \Pr[C[Q] : \sevent \vee D] - \Pr[C[Q] : \sbarevent \vee \nonunique{Q,D}]\\
&= \Pr[C''[Q] : \sevent \vee D] - \Pr[C''[Q] : \sbarevent \vee \nonunique{Q,D}]
\tag*{since the renaming of events does not affect the events of the distinguisher}\\
&= \Pr[C''[Q] : (\sevent \vee D) \wedge \neg \nonunique{Q,D}]\\
&\quad - \Pr[C''[Q] : (\sbarevent \wedge \neg D) \vee \nonunique{Q,D}]
\intertext{since $\sevent \vee D$ is equivalent to $(\sevent \vee D) \wedge \neg \nonunique{Q,D}$ and $\sbarevent$ is equivalent to $\sbarevent \wedge \neg D$: the game aborts immediately after
executing $\sevent$, $\sbarevent$, and the events in $D$ and $\nonunique{Q,D}$ so only one of them can be executed}
&= \Pr[C''[Q] : (\sevent \vee D) \wedge \neg \nonunique{Q,D}]\\
&\quad + \Pr[C''[Q] : (\neg\sbarevent \vee D) \wedge \neg \nonunique{Q,D}] - 1\\
&\leq \Pr[C''[Q'] : (\sevent \vee D') \wedge \neg \nonunique{Q',D'}] + p(C'',t_{\sevent})\\
&\quad + \Pr[C''[Q'] : (\neg\sbarevent \vee D') \wedge \neg \nonunique{Q',D'}] + p(C'', t_{\sevent}) - 1
\tag*{since $\neg\sbarevent$ can be implemented in the same time as $\sevent$}\\
&\leq \Pr[C''[Q'] : (\sevent \vee D') \wedge \neg \nonunique{Q',D'}] + p(C'',t_{\sevent})\\
&\quad - \Pr[C''[Q'] : (\sbarevent \wedge\neg D') \vee \nonunique{Q',D'}] + p(C'', t_{\sevent})\\ 
&\leq \Pr[C''[Q'] : \sevent \vee D'] - \Pr[C''[Q'] : \sbarevent \vee \nonunique{Q',D'}] + 2 p(C'', t_{\sevent})
\intertext{since $\sevent \vee D'$ is also equivalent to $(\sevent \vee D') \wedge \neg \nonunique{Q,D'}$ and $\sbarevent$ is also equivalent to $\sbarevent \wedge \neg D'$}
&\leq \Advtev{Q'}{\prop}{C''}{D'} + 2 p(C'', t_{\sevent})\\
&\leq p'(C'') + 2 p(C'', t_{\sevent})
\tag*{since $\bound{Q'}{V}{\prop}{D'}{p'}$}\\
&\leq p'(C) + 2 p(C, t_{\sevent})
\tag*{since the renaming does not modify the probability formulas by
  Property~\ref{prop:prob}}\\
&\leq p''(C)
\end{align*}
so $\bound{Q}{V}{\prop}{D}{p''}$.

Property~\ref{item:prop:split}, case $\prop$ is a trace property:
Let $C$ be an evaluation context acceptable for $Q$ with public variables $V$
that does not contain events used by $\prop$, $D$, $D'$, nor non-unique events
of $Q$.
We have $\nonunique{Q,D} = \nonunique{Q} \wedge\neg D$.
Therefore
\[\begin{split}
\Advtev{Q}{\prop}{C}{D} 
&= \Pr[C[Q] : (\neg\prop \vee D) \wedge \neg \nonunique{Q, D}]\\
&= \Pr[C[Q] : (\neg\prop \vee D) \wedge \neg (\nonunique{Q} \wedge\neg D)]\\
&= \Pr[C[Q] : (\neg\prop \vee D) \wedge (\neg \nonunique{Q} \vee D)]\\
&= \Pr[C[Q] : (\neg\prop \wedge \neg \nonunique{Q}) \vee D]
\end{split}\]
In particular, $\Advtev{Q}{\true}{C}{D} = \Pr[C[Q] : D]$.
Hence
\[\begin{split}
\Advtev{Q}{\prop}{C}{D \vee D'}
&= \Pr[C[Q] : (\neg\prop \wedge \neg \nonunique{Q})\vee D\vee D']\\
&\leq \Pr[C[Q] : (\neg\prop \wedge \neg \nonunique{Q})\vee D] + \Pr[C[Q] : D']\\
&\leq \Advtev{Q}{\prop}{C}{D} + \Advtev{Q}{\true}{C}{D'}\\
&\leq p(C) + p'(C)
\end{split}\]
since $\bound{Q}{V}{\prop}{D}{p}$ and
  $\bound{Q}{V}{\true}{D'}{p'}$.
So $\bound{Q}{V}{\prop}{D \vee D'}{p + p'}$.

Property~\ref{item:prop:split}, case $\prop$ is $\secrone(x)$, $\secr(x)$, or $\secrbit(x)$:
Let $C'$ be an evaluation context acceptable for $C_{\prop}[Q]$ with public variables $V \setminus V_{\prop}$
that does not contain $\sevent$, $\sbarevent$, events used by $D$, $D'$, nor non-unique events
of $Q$. Let $C = C'[C_{\prop}[\,]]$.
Since $\nonunique{Q,D} = \nonunique{Q} \wedge\neg D$, we have
\[
\Advtev{Q}{\prop}{C}{D} = \Pr[C[Q] : \sevent \vee D] - \Pr[C[Q] : \sbarevent \vee (\nonunique{Q} \wedge\neg D)]
\]
Hence
\[\begin{split}
\Advtev{Q}{\prop}{C}{D \vee D'}
&= \Pr[C[Q] : \sevent \vee D \vee D']\\
&\quad - \Pr[C[Q] : \sbarevent \vee (\nonunique{Q} \wedge\neg D \wedge \neg D')]\\
&\leq \Pr[C[Q] : \sevent \vee D] + \Pr[C[Q] : D'] \\
&\quad - \Pr[C[Q] : \sbarevent \vee (\nonunique{Q} \wedge\neg D)] + \Pr[C[Q] : \Dnu']\\
&\leq \Advtev{Q}{\prop}{C}{D} + \Advtev{Q}{\true}{C}{D'} + \Advtev{Q}{\true}{C}{\Dnu'}\\
&\leq p(C) + p'(C) + p''(C)
\end{split}\]
since $\bound{Q}{V}{\prop}{D}{p}$,
$\bound{Q}{V}{\true}{D'}{p'}$, 
$\bound{Q}{V}{\true}{\Dnu'}{p''}$,
and $C$ is also an evaluation context acceptable for $Q$ with public variables $V$.
So $\bound{Q}{V}{\prop}{D \vee D'}{p + p' + p''}$.

The inequality $- \Pr[C[Q] : \sbarevent \vee (\nonunique{Q} \wedge\neg D \wedge \neg D')] \leq - \Pr[C[Q] : \sbarevent \vee (\nonunique{Q} \wedge\neg D)] + \Pr[C[Q] : \Dnu']$ used above is justified as follows:
\[\begin{split}
&\Pr[C[Q] : \sbarevent \vee (\nonunique{Q} \wedge\neg D)]\\
&\quad\leq \Pr[C[Q] : \sbarevent \vee (\nonunique{Q} \wedge\neg D\wedge \neg D') \vee (\nonunique{Q} \wedge D')]\\
&\quad\leq \Pr[C[Q] : \sbarevent \vee (\nonunique{Q} \wedge\neg D \wedge \neg D')] + \Pr[C[Q] : \nonunique{Q} \wedge D']\\
&\quad\leq \Pr[C[Q] : \sbarevent \vee (\nonunique{Q} \wedge\neg D \wedge \neg D')] + \Pr[C[Q] : \Dnu']
\end{split}\]

Property~\ref{item:prop:merge}:
The distinguisher $D \vee D''$ is a disjunction of Shoup and non-unique events in $\usedevents$.
Let $C$ be any evaluation context acceptable for $Q$ with public variables $V$ that does not contain events in $\usedevents'$.
Let $D_0 \in \dset\cup \{\Dfalse\}$.
Let $D_1$ be a disjunction of events in $\dsetsnu$.
We have
\begin{align*}
&\Pr[C[Q] : (D_0 \vee D_1 \vee D \vee D'') \wedge \neg \nonunique{Q,D_1 \vee D \vee D''}]\\
&\quad = \Pr[C[Q] : (D_0 \wedge \neg \nonunique{Q}) \vee D_1 \vee D \vee D'']
\tag*{as in Property~\ref{item:prop:split}, case $\prop$ is a trace property}\\
&\quad \leq \Pr[C[Q] : (D_0 \wedge \neg \nonunique{Q}) \vee D_1 \vee D] + \Pr[C[Q]: D'']\\
&\quad \leq \Pr[C[Q] : (D_0 \vee D_1 \vee D) \wedge \neg \nonunique{Q,D_1 \vee D}] + \Advtev{Q}{\true}{C}{D''}
\tag*{since $\Advtev{Q}{\true}{C}{D''} = \Pr[C[Q]: D'']$ (see Property~\ref{item:prop:split}, case $\prop$ is a trace property)}\\
&\quad \leq  \Pr[C[Q'] : (D_0 \vee D_1 \vee D') \wedge \neg \nonunique{Q,D_1 \vee D'}] + p(C, t_{D_0}) + p'(C)
\end{align*}
since $\bound{Q}{V}{\true}{D''}{p'}$. (Note that the context $C$ does not contain events used by $D''$ nor non-unique events of $Q$.)
\proofcomplete
\end{proof}
This lemma allows one to bound the advantage of the adversary against 
secrecy and correspondences. 
Property~\ref{item:prop:init} is used in the initial game, to express
the desired probability from $\bound{Q}{V}{\prop}{D_U}{p}$. (Using the
distinguisher $D_U$ can also
be understood by saying that we consider that the adversary wins 
if some non-unique event is executed, that is, if a $\FIND$ or $\GET$
declared $\kw{unique}$ by the user actually has several possible choices.
That allows the implementation to make any choice when a $\FIND\unique{e}$ or
$\GET\unique{e}$ has several possible choices: the security proof remains
valid. In particular, a $\FIND\unique{e}$ or $\GET\unique{e}$ can be
implemented by always choosing the first found element.)
Property~\ref{item:prop:step} is used when a game $Q$ is transformed
into a game $Q'$ during the proof. It allows one to bound the probability
in $Q$ from a bound in $Q'$.
Property~\ref{item:prop:split} is useful when distinct sequences of games
are used for bounding the probabilities of breaking $\prop$ and of $D$ on one
side and of $D'$ on the other side.
We bound these two probabilities by $\bound{Q}{V}{\prop}{D}{p}$ and $\bound{Q}{V}{\true}{D'}{p'}$
separately, then obtain a bound $\bound{Q}{V}{\prop}{D \vee D'}{p'''}$ by
computing a sum. (When we deal with secrecy and $\Dnu'\neq\Dfalse$, the
probability of $\Dnu'$ can be bounded by looking at the proof for $D'$.) 
 
More formally, consider the following cases, using Lemma~\ref{lem:adv},
Property~\ref{item:prop:init}:
\begin{itemize}
\item If we want to prove that $Q_0$ satisfies the
correspondence $\corresp$ with public variables $V$,
then we let $\prop = \corresp$.

\item If we want prove that $Q_0$ satisfies the (one-session or bit)
  secrecy of $x$ with public variables $V'$ ($x \notin V'$),
  then we let $\prop$ be $\secrone(x)$, $\secr(x)$, or $\secrbit(x)$
  and $V = V' \cup \{ x\}$.
  
\end{itemize}
In both cases, we show $\bound{Q_0}{V}{\prop}{D_U}{p}$.
The proof produced by CryptoVerif can be represented as a 
tree whose nodes are labeled with quintuples $(Q, V, \prop', D, \usedevents)$
and whose edges are labeled with triples $(\prop'', D'', p'')$,
where $Q$ is the current game, $V$ is the set of public variables,
$\prop'$ and $\prop''$ are either the initial property to prove $\prop$ or $\true$,
$D$ and $D''$ are disjunctions of Shoup and non-unique events,
$\usedevents$ is the set of events used so far, and
$p''$ is a probability formula. The edges have a single source node,
but may have 0, 1, or several target nodes.
We associate to each node labeled with $(Q, V, \prop', D, \usedevents)$
a property $\bound{Q}{V}{\prop'}{D}{p}$, and to each edge labeled 
with $(\prop'', D'', p'')$ with source node labeled with $(Q, V, \prop', D, \usedevents)$ a property $\bound{Q}{V}{\prop''}{D''}{p'}$. CryptoVerif computes the probabilities $p$, $p'$ such that these properties hold from the leaves of the tree to its root, as explained next.

The root of the tree is labeled with $(Q_0, V,  \prop, D_U, \usedevents_0)$
such that $\usedevents_0$ is the set containing the events used
by correspondences to prove or that occur in $Q_0$.

\begin{figure}
\begin{center}
\begin{tikzpicture}[std/.style={rounded corners,draw=black,thick}]
  \node[std] (T) at (3,0) {$Q, V, \_, \_, \usedevents$};
\coordinate (E1m) at (-3,-1);
\coordinate (E2m) at (3,-1);
\coordinate (E3m) at (9,-1);
\node[std] (E21) at (-2,-2) {$Q_1, V_1, \prop'_1, D_1, \usedevents_1$};
\node at (0.5,-2) {\dots};
  \node[std] (E22) at (3,-2) {$Q_j, V_j, \prop'_j, D_j, \usedevents_j$};
\node at (5.5,-2) {\dots};
  \node[std] (E23) at (8,-2) {$Q_l, V_l, \prop'_l, D_l, \usedevents_l$};
  \draw[thick] (T)--(E1m) node[below=1mm] {\dots};
  \draw[thick] (T)--(E2m) node[midway,left=-0.1cm] {{\small $\prop'', D'', p''$}};
  \draw[->,>=latex,thick] (E2m)--(E21);
  \draw[->,>=latex,thick] (E2m)--(E22);
  \draw[->,>=latex,thick] (E2m)--(E23);
  \draw[thick] (T)--(E3m) node[below=1mm] {\dots};
\end{tikzpicture}
\end{center}
\centerline{(a) edge, source, and target nodes}

\begin{center}
\begin{tikzpicture}[std/.style={rounded corners,draw=black,thick}]
  \node[std] (T) at (3,0) {$Q, V, \prop', D, \usedevents$};
  \node[std] (E11) at (-1,-2) {\dots};
  \node[std] (E12)at (0,-2) {\dots};
  \node[std] (E13) at (1,-2) {\dots};
  \node[std] (E2) at (3,-2) {\dots};
  \node[std] (E31) at (5,-2) {\dots};
  \node[std] (E32) at (6,-2) {\dots};
\coordinate (E1m) at (0,-1);
\node at (1,-1) {\dots};
\coordinate (E2m) at (3,-1);
\node at (4.2,-1) {\dots};
\coordinate (E3m) at (5.5,-1);
  \draw[thick] (T)--(E1m) node[midway,left=0.3cm] {{\small $\prop'_1, D'_1, \_$}};
  \draw[->,>=latex,thick] (E1m)--(E11);
  \draw[->,>=latex,thick] (E1m)--(E12);
  \draw[->,>=latex,thick] (E1m)--(E13);
  \draw[thick] (T)--(E2m) node[left=-0.1cm] {{\small $\prop'_j, D'_j, \_$}};
  \draw[->,>=latex,thick] (E2m)--(E2);
  \draw[thick] (T)--(E3m) node[midway,right=0.1cm] {{\small $\prop'_l, D'_l, \_$}};  
  \draw[->,>=latex,thick] (E3m)--(E31);
  \draw[->,>=latex,thick] (E3m)--(E32);
\end{tikzpicture}
\end{center}

\centerline{(b) node and outgoing edges}

\caption{Structure of a proof tree}\label{fig:prooftreegen}
\end{figure}

When an edge labeled with $(\prop'', D'', p'')$ has a source node
labeled with  $(Q, V, \_, \_, \usedevents)$ and target nodes
labeled with $(Q_j, V_j, \prop'_j, D_j, \usedevents_j)$ for
$j \in \{1, \dots, l\}$ (Figure~\ref{fig:prooftreegen}(a)), the bound associated to the edge can be computed
from the bound associated to the target nodes by the probability formula
$p''$ that labels the edge: if $\bound{Q_j}{V_j}{\prop'_j}{D_j}{p_j}$ for
$j \in \{1, \dots, l\}$, then
$\bound{Q}{V}{\prop''}{D''}{p}$ where $p(C) = p''(C, p_1(C), \dots, p_l(C))$.
This situation corresponds to a game transformation that transforms game $Q$
into games $Q_j$ ($j \in \{1, \dots, l\}$).
We can distinguish several cases depending on where the edge comes from:
\begin{itemize}

\item For most game transformations, the edge has a single target node ($l=1$),
  $V_1 = V$, $\prop'_1 = \prop''$,
and the game transformation transforms $Q$ into $Q_1$ and satisfies
  $\dset_1,\emptyset: Q, D'', \usedevents \indistev{V}{p'} Q_1, D_1, \usedevents_1$
where $\dset_1 = \{ \neg\prop'' \}$ when $\prop''$ is a trace property, $\dset_1 = \{ \sevent, \neg\sbarevent\}$ when $\prop''$ is $\secrone(x)$, $\secr(x)$, or $\secrbit(x)$, and $\dset_1 = \emptyset$ when $\prop'' = \true$.
(During the building of the proof tree,
we have
$\dset,\dsetsnu: Q, D, \usedevents \indistev{V}{p'} Q_1, D', \usedevents_1$,
where the distinguishers $\dset$ correspond to the active queries not for introduced Shoup and non-unique events, so $\dset_1 \subseteq \dset$,
$\dsetsnu$ are the active queries for Shoup and non-unique events both before and after this step so $D''$ and $D_1$ are disjunctions of events in $\dsetsnu$,
$D = D''\wedge\neg D_1$ are the Shoup/non-unique events proved at this step,
$D' = D_1\wedge\neg D''$ are the Shoup/non-unique events introduced at this step.
By applying several times Lemma~\ref{lem:indistev}, Property~\ref{indistadd},
we obtain 
$\dset,\dsetsnu: Q, D'', \usedevents \indistev{V}{p'} Q_1, D_1, \usedevents_1$.
By Lemma~\ref{lem:indistev}, Property~\ref{indistremove},
we obtain 
$\dset_1,\emptyset: Q, D'', \usedevents \indistev{V}{p'} Q_1, D_1, \usedevents_1$.)
The bound is inferred by Lemma~\ref{lem:adv}, Property~\ref{item:prop:step}.

If $\prop''$ is a trace property and $\bound{Q_1}{V}{\prop''}{D_1}{p_1}$, then $\bound{Q}{V}{\prop''}{D''}{p}$ where $p(C) = p'(C, t_{\neg\prop''}) + p_1(C)$, so we can define $p''(C, p_t) = p'(C, t_{\neg\prop''}) + p_t$.

If $\prop''$ is $\secrone(x)$, $\secr(x)$, or $\secrbit(x)$ and $\bound{Q_1}{V}{\prop''}{D_1}{p_1}$, then we have $\bound{Q}{V}{\ab \prop''}{\ab D''}{\ab p}$ where $p(C) = 2 p'(C, t_{\sevent}) + p_1(C)$, so we can define $p''(C, p_t) = 2 p'(C, t_{\sevent}) + p_t$.

To unify these two cases, we define
\[\stdtransfproba{\prop''}{p'}(C,p_t) =
\begin{cases}
  p'(C, t_{\neg\prop''}) + p_t&\text{when $\prop''$ is a trace property}\\
  2p'(C, t_{\sevent}) + p_t&\text{when $\prop''$ is $\secrone(x)$, $\secr(x)$, or $\secrbit(x)$}
\end{cases}\]
so that, when an edge comes from a transformation
$\dset_1,\emptyset: Q, D'', \usedevents \indistev{V}{p'} Q_1, \ab D_1, \ab \usedevents_1$ and the considered security property is $\prop''$, the edge can be labeled with the probability formula $\stdtransfproba{\prop''}{p'}$.
If $\bound{Q_1}{V}{\prop''}{D_1}{p_1}$, then $\bound{Q}{V}{\prop''}{D''}{p}$ where $p(C) = \stdtransfproba{\prop''}{p'}(C,p_1(C))$.

\item When a query is proved (by the command \rn{success}, Section~\ref{sec:success}),
  the edge has no target node ($l = 0$), and we simply obtain $\bound{Q}{V}{\prop''}{D''}{p''}$. The probability $p''$ that labels the edge is determined by the \rn{success} command (Proposition~\ref{prop:sec}, \ref{prop:nicorresp}, or~\ref{prop:icorresp}).

  These propositions require that $D'' = \Dfalse$. This is obtained by splitting the properties to prove one property at a time (with one edge for each property starting from the source node), yielding bounds of the form
  $\bound{Q}{V}{\prop''}{\Dfalse}{p''}$ or $\bound{Q}{V}{\true}{e}{p''}$.
  For the first form, $p''$ can immediately be computed from Proposition~\ref{prop:sec}, \ref{prop:nicorresp}, or~\ref{prop:icorresp}. For the second form, when $e$ is a non-unique event, we use Section~\ref{sec:nuevents} and when $e$ is a Shoup event, we notice that $\Advtev{Q}{\true}{C}{e} = \Pr[C[Q]: e \wedge \neg\nonunique{Q,e}] = \Pr[C[Q]: e \wedge \neg\nonunique{Q}] = \Advtev{Q}{\sem{\fevent{e} \Rightarrow \false}}{C}{\Dfalse}$, so we can use $\bound{Q}{V}{\sem{\fevent{e} \Rightarrow \false}}{\Dfalse}{p''}$ instead.

  The \rn{success} command is the only one that removes an event from
  $D''$, which then happens only when we evaluate
  $\bound{Q}{V}{\true}{e}{p''}$, so $D'' = e$, $\prop'' = \true$, $l =
  0$.

\item In case of other transformations such as the \rn{guess} transformation, the relation between bounds that defines $p''$ is given directly in the soundness lemma for the transformation (Lemma~\ref{lem:guess_i}, \ref{lem:guess_x}, or~\ref{lem:guess_branch}). In particular, for the transformation \rn{guess\_branch}, the edge has as many target nodes as there are branches in the guessed instruction. For the transformation $\rn{guess}\ i$, Lemma~\ref{lem:guess_i} requires $D'' = \Dfalse$, which can be achieved as for \rn{success} above.

\end{itemize}

When a node labeled with $(Q, V, \prop', D, \usedevents)$ 
has outgoing edges labeled respectively $(\prop'_1, D'_1, \_)$, 
\ldots, $(\prop'_l, D'_l, \_)$ (Figure~\ref{fig:prooftreegen}(b)), then 
$D = D'_1 \vee \ldots \vee D'_l$
($D$ is a disjunction of the form $e_1 \vee \ldots \vee e_m$, $D'_1, \ldots, D'_l$ are disjunctions
that form a partition of the disjuncts of $D$),
there exists $j_0 \leq l$ such that $\prop'_{j_0} = \prop'$ and
for all $j \neq j_0$, $\prop'_{j} = \true$.
The bound associated to the node is computed from the bound associated to the edges by Lemma~\ref{lem:adv}, Property~\ref{item:prop:split}:
\begin{itemize}
\item If $\prop'$ is a trace property, then we have
 $\bound{Q}{V}{\prop'}{D}{p'_1 + \dots + p'_l}$, where
  $\bound{Q}{V}{\ab \prop'_j}{\ab D'_j}{\ab p'_j}$ for all $j \in \{ 1, \dots, l\}$.
 
\item If $\prop'$ is $\secrone(x)$, $\secr(x)$, or $\secrbit(x)$, then we have
 $\bound{Q}{V}{\prop'}{D}{\sum_{j=1}^l p'_j + \sum_{j=1,\dots,l; j\neq j_0} p''_j}$, where
  $\bound{Q}{V}{\prop'_j}{D'_j}{p'_j}$ for all $j \in \{ 1, \dots, l\}$,
  $\bound{Q}{V}{\prop'_j}{\Dnuarg{j}'}{\ab p''_j}$ for all $j \in \{ 1, \dots, l\}$ with $j \neq j_0$,
  and $\Dnuarg{j}'$ is obtained from $D'_j$ by keeping only the non-unique events.

To prove $\bound{Q}{V}{\prop'_j}{\Dnuarg{j}'}{p''_j}$, we build a proof tree 
with root $Q, V, \prop'_j, \Dnuarg{j}', \usedevents$ from the subtree 
$Q, V, \prop', D, \usedevents \mathop{\longrightarrow}\limits^{\prop'_j, D'_j, p_j} \dots$
by
\begin{itemize}
\item replacing the root $Q, V, \prop', D, \usedevents$ with $Q, V, \prop'_j, D'_j, \usedevents$.
\item removing all Shoup events of $D'_j$ from distinguishers that
label nodes and edges. In particular, the root $Q, V, \prop'_j, D'_j, \usedevents$
then becomes $Q, V, \prop'_j, \Dnuarg{j}', \usedevents$.
\item removing subtrees that start with an edge labeled $\prop'_j, \Dfalse, p$ for some
  $p$.
\end{itemize}
The proof steps remain valid: all proof steps are a fortiori valid when we ignore some Shoup events. That can be verified for each transformation using its soundness lemma. For instance, for proof steps that come from
usual transformations that satisfy property preservation with introduction of events and have a Shoup event $e$
of $D'_j$ both before and after, we can avoid adding that event $e$ when we
derive $\dset_1,\emptyset: Q, D'_j, \usedevents \indistev{V}{p'} Q_1, D_1, \usedevents_1$
from
$\dset,\dsetsnu: Q, D, \usedevents \indistev{V}{p'} Q_1, D', \usedevents_1$.
For proof steps that prove a Shoup event $e$ of $D'_j$ (via \rn{success}), that proof step
starts with just event $e$, so it is simply removed.
We can then apply the previous reasoning to that proof tree.

\end{itemize}

When the proof is a basic sequence of games, each node has one son, which
is the next game in the sequence, except the last game of the sequence 
which has no son. Only the final proof step is distinct for each query.
However, it may happen that distinct sequences of games
are used to bound several events occurring in the game; in this
case, there is a branching in the proof and a node has several sons.
Examples of proof trees can be found in Figure~\ref{fig:prooftrees};
they are explained below.

The bound associated to the leaves of the tree
is computed by \rn{success}; the bound associated to an edge is computed
from the bounds associated to its target nodes, and the bound
associated to a node is computed from the bounds associated to its
outgoing edges. We can then
compute the bounds associated to all nodes of the tree, 
by induction from the leaves to the root. At the root, we obtain a 
bound $\bound{Q_0}{V}{\prop}{D_U}{p}$ that yields the desired result.

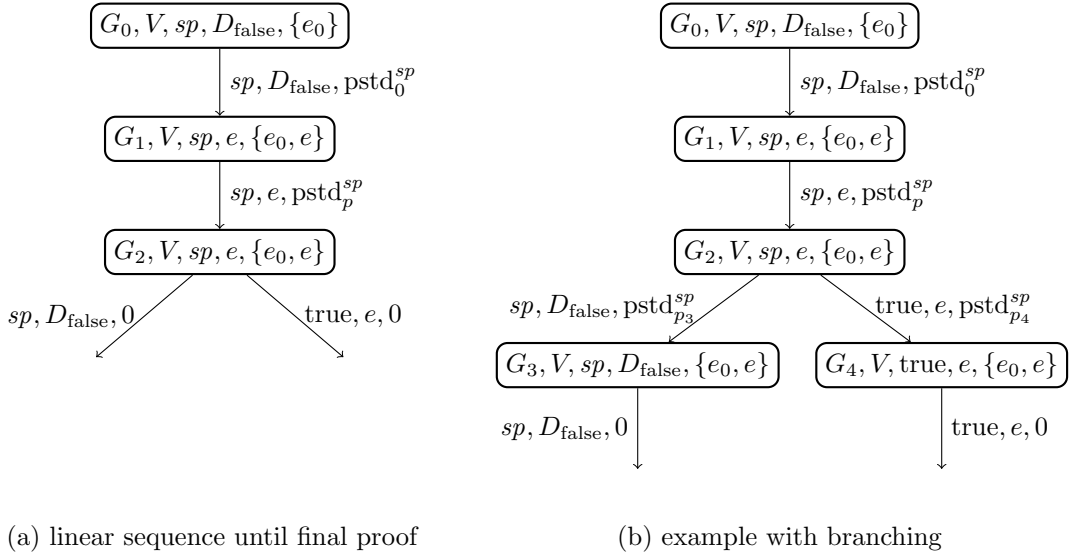
\begin{figure}
\begin{center}
\begin{minipage}{6.1cm}
\begin{center}
\begin{tikzpicture}[std/.style={rounded corners,draw=black,thick},
edge from parent path={[->] (\tikzparentnode) -- (\tikzchildnode)}]
\node[std] {$G_0, V, \prop, \Dfalse, \{e_0\}$} 
child{ node[std] {$G_1, V, \prop, e, \{e_0, e\}$} 
  child { node[std] {$G_2, V, \prop, e, \{e_0, e\}$} [sibling distance=3.5cm]
    child { node {}
      edge from parent node[left]{$\prop, \Dfalse, 0$}
    }
    child { node {}
      edge from parent node[right]{$\true, e, 0$}
    }
    edge from parent node[right]{$\prop, e, \stdtransfproba{\prop}{p}$}
  }
  edge from parent node[right]{$\prop, \Dfalse, \stdtransfproba{\prop}{0}$}
} ;
\end{tikzpicture}%
\end{center}
\vspace*{1.5cm}

\centerline{(a) linear sequence until final proof}
\end{minipage}
\quad
\begin{minipage}{7.9cm}
\begin{center}
\begin{tikzpicture}[std/.style={rounded corners,draw=black,thick},
edge from parent path={[->] (\tikzparentnode) -- (\tikzchildnode)}]
\node[std] {$G_0, V, \prop, \Dfalse, \{e_0\}$} 
child{ node[std] {$G_1, V, \prop, e, \{e_0, e\}$} 
  child { node[std] {$G_2, V, \prop, e, \{e_0, e\}$} [sibling distance=4cm]
    child { node[std] {$G_3, V, \prop, \Dfalse, \{e_0, e\}$}
      child { node {}
        edge from parent node[left] {$\prop, \Dfalse, 0$}
      }
      edge from parent node[left=1mm]{$\prop, \Dfalse, \stdtransfproba{\prop}{p_3}$}
    }
    child { node[std] {$G_4, V, \true, e, \{e_0, e\}$}
      child { node {}
        edge from parent node[right] {$\true, e,0$}
      }
      edge from parent node[right]{$\true, e, \stdtransfproba{\prop}{p_4}$}
    }
    edge from parent node[right]{$\prop, e, \stdtransfproba{\prop}{p}$}
  }
  edge from parent node[right]{$\prop, \Dfalse, \stdtransfproba{\prop}{0}$}
} ;
\end{tikzpicture}
\end{center}

\centerline{(b) example with branching}
\end{minipage}

\end{center}
\caption{Examples of proof trees}\label{fig:prooftrees}
\end{figure}

Lemma~\ref{lem:adv} allows us to obtain more precise probability bounds 
than the standard computation of probabilities generally done by
cryptographers, when we use Shoup's lemma~\cite{Shoup04}.
By Shoup's lemma, if $G'$ is obtained from $G$ by inserting an event $e$ and modifying the code executed after $e$, the probability of distinguishing $G'$ from $G$ is bounded by the probability
of executing $e$: for all contexts $C$ acceptable for $G$ and $G'$ (with any public variables) and all distinguishers $D$, $|\Pr[C[G] : D] - \Pr[C[G']: D]| \leq \Pr[C[G'] : e]$.
Hence, 
\[\Pr[C[G] : D] \leq \Pr[C[G'] : e] + \Pr[C[G']: D].\] 
We improve over this computation of probabilities by considering $e$ and $D$ simultaneously instead of making the sum of the two probabilities: 
\[\Pr[C[G] : D] \leq \Pr[C[G'] : D \vee e].\]
For example, suppose that we want to bound the probability of event $e_0$ in $G_0$: we define $\prop = \sem{\fevent{e_0} \Rightarrow \false} = \neg e_0$. We transform $G_0$ into $G_1$ using Shoup's lemma, so that $G_1$ differs from $G_0$ only when $G_1$ executes event $e$, and we have 
$\{ e_0\},\emptyset: G_0, \ab \Dfalse, \ab \usedevents_0 \indistev{}{0} G_1, \ab e, \ab \usedevents_1$;
then we transform $G_1$ into $G_2$, so that $G_1 \approx^{\{e_0,e\}}_p G_2$,
so we have $\{ e_0\},\emptyset: G_1, \ab e, \ab \usedevents_1 \indistev{}{p} G_2, \ab e, \ab \usedevents_1$ by Lemma~\ref{lem:indistev}, Property~\ref{linkindist1};
and $G_2$ executes neither $e_0$ nor $e$.
We suppose for simplicity that no $\unique{e'}$ occurs, 
so that $\nonunique{G_i,D}$ is always false.
The corresponding proof tree is given in Figure~\ref{fig:prooftrees}(a).
\begin{itemize}
  \item 
Since $e_0$ does not occur in $G_2$, we have $\Advtev{G_2}{\prop}{C}{\Dfalse} = 0$ for all evaluation contexts $C$ acceptable for $G_2$ with public variables $V$ that do not contain event $e_0$. So $\bound{G_2}{V}{\ab \prop}{\ab \Dfalse}{\ab 0}$.
Similarly, $\bound{G_2}{V}{\ab \true}{\ab e}{\ab 0}$. These two properties are represented
in the proof tree by the two edges outgoing from node $G_2, \ab V, \ab \prop, \ab e, \ab \{e_0, e\}$.

\item By Lemma~\ref{lem:adv}, Property~\ref{item:prop:split},
$\bound{G_2}{V}{\ab \prop}{\ab e}{\ab p'_3}$ where $p'_3(C) = 0$.

\item 
Since $\{ e_0\},\emptyset: G_1, \ab e, \ab \usedevents_1 \indistev{}{p} G_2, \ab e, \ab \usedevents_1$,
by Lemma~\ref{lem:adv}, Property~\ref{item:prop:step},
we obtain $\bound{G_1}{V}{\ab \prop}{\ab e}{\ab p'_2}$
where $p'_2(C) = p(C, t_{\neg\prop}) + p'_3(C)$. This is represented in the proof tree
by the edge from node $G_1, \ab V, \ab \prop, \ab e, \ab \{e_0, e\}$ to node $G_2, \ab V, \ab \prop, \ab e, \ab \{e_0, e\}$ labeled with $\prop, e, \stdtransfproba{\prop}{p}$.

\item 
Since $\{ e_0\},\emptyset: G_0, \ab \Dfalse, \ab \usedevents_0 \indistev{}{0} G_1, \ab e, \ab \usedevents_1$,
by Lemma~\ref{lem:adv}, Property~\ref{item:prop:step},
we obtain $\bound{G_0}{V}{\ab \prop}{\ab \Dfalse}{\ab p'_1}$
where $p'_1(C) = 0 + p'_2(C)$. This is represented in the proof tree
by the edge from node $G_0, \ab V,\ab \prop,\ab  \Dfalse, \ab \{e_0\}$ to node $G_1, \ab V, \ab \prop, \ab e, \ab \{e_0, e\}$ labeled with $\prop, \Dfalse, \stdtransfproba{\prop}{0}$.

\item
  Finally, by Lemma~\ref{lem:adv}, Property~\ref{item:prop:init},
  we conclude that $G_0$ satisfies $\prop$ with public variables $V$
  up to probability $p'_1$, where $p'_1(C) = p'_2(C) = p(C, t_{\neg\prop}) + p'_3(C) = p(C, t_{e_0})$, which means that $\Pr[C[G_0] : e_0] \leq p(C, t_{e_0})$
  for all evaluation contexts $C$ acceptable for $G_0$ with public variables $V$ that do not contain event $e_0$.
\end{itemize}
%% \begin{align*}
%% \Pr[C[G_0]: e_0] 
%% &= \Advtev{G_0}{\prop}{C}{\Dfalse} \\
%% &\leq \Advtev{G_1}{\prop}{C}{e}\\
%% &\leq p(C, t_{e_0}) + \Advtev{G_2}{\prop}{C}{e} = p(C, t_{e_0 \vee e})
%% \end{align*}

Let $C$ be an evaluation context acceptable for $G_0$ with public variables $V$ that does not contain event $e_0$. 
Since we suppose for simplicity that no $\unique{e'}$ occurs, 
so that $\nonunique{G_i,D}$ is always false, we have $\Advtev{G_i}{\prop}{C}{D} = 
\Pr[C[G_i] : \neg\prop \vee D]$, so we can write the previous computation simply using
probabilities:
\begin{align*}
  \Pr[C[G_0] : e_0] &\leq \Pr[C[G_1] : e_0 \vee e]
\tag*{since $\{ e_0\},\emptyset: G_0, \ab \Dfalse, \ab \usedevents_0 \indistev{}{0} G_1, \ab e, \ab \usedevents_1$}\\
& \leq p(C, t_{e_0}) + \Pr[C[G_2] : e_0 \vee e]
\tag*{since $\{ e_0\},\emptyset: G_1, \ab e, \ab \usedevents_1 \indistev{}{p} G_2, \ab e, \ab \usedevents_1$}\\
&\leq p(C, t_{e_0})
\tag*{since $G_2$ executes neither $e_0$ nor $e$.}
\end{align*}
In contrast, the standard computation of probabilities yields
\[\Pr[C[G_0] : e_0] \leq \Pr[C[G_1] : e_0] + \Pr[C[G_1] : e] \leq p(C, t_{e_0}) + p(C, t_e).\]
The runtime $t_D$ of $D$ is essentially the same for $e_0$, $e$, and $e_0 \vee e$, so $\Pr[C[G_0] : e_0] \leq p(C, t_D)$ by Lemma~\ref{lem:adv}, while $\Pr[C[G_0] : e_0] \leq 2 p(C,t_D)$ by the standard computation, so we have gained a factor 2. The probability that comes from the transformation of $G_1$ into $G_2$ is counted once (for distinguisher $e_0 \vee e$) instead of counting it twice (once for $e_0$ and once for $e$).

The standard computation of probabilities corresponds to applying
point~\ref{item:prop:split} of Lemma~\ref{lem:adv} to bound each
probability separately and compute the sum, as soon as the considered
distinguisher $D$ has several disjuncts. Instead, we use
point~\ref{item:prop:split} of Lemma~\ref{lem:adv} only when the proof
uses different sequences of games to bound the probabilities of the
events, as in Figure~\ref{fig:prooftrees}(b).

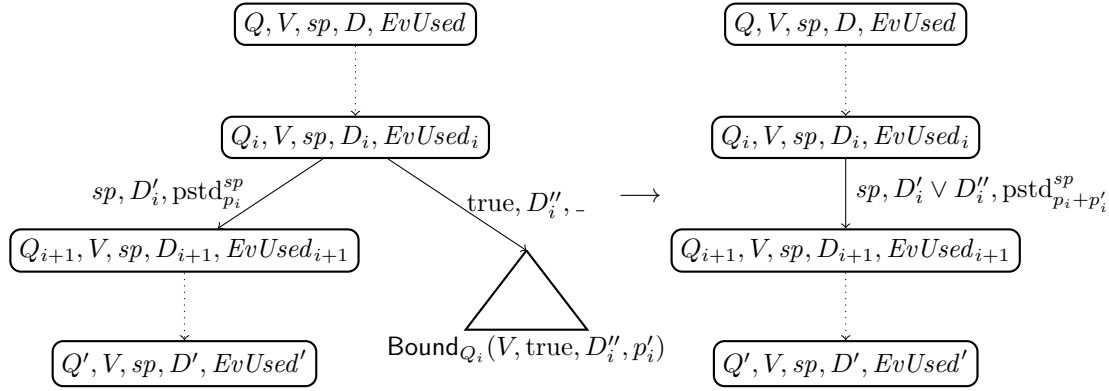
\begin{figure}[t]
\begin{center}
\begin{minipage}{8.9cm}
\begin{tikzpicture}[std/.style={rounded corners,draw=black,thick},
edge from parent path={[->] (\tikzparentnode) -- (\tikzchildnode)}]
\node[std] (A) {$Q, V, \prop, D, \usedevents$} 
child{ node[std] (B) {$Q_i, V, \prop, D_i, \usedevents_i$} [sibling distance=4.5cm]
  child { node[std] {$Q_{i+1}, V, \prop, D_{i+1}, \usedevents_{i+1}$}
    child { node[std] {$Q', V, \prop, D', \usedevents'$}
      edge from parent node { } [dotted]
    }
    edge from parent node[left=2mm]{$\prop, D'_i, \stdtransfproba{\prop}{p_i}$}
  }
  child { node[coordinate] (C) { }
    edge from parent node[right]{$\true, D''_i, \_$}
  };
} ;
\draw [->,dotted] (A) -- (B);
\node[coordinate,below left=1.05cm and 0.8cm] (D) at (C) {};
\node[coordinate,below right=1.05cm and 0.8cm] (E) at (C) {};
\draw [thick] (C) -- (D) -- (E) -- (C);
\node[below=1cm] at (C) {$\bound{Q_i}{V}{\true}{D''_i}{p'_i}$} ;
\end{tikzpicture}%
\end{minipage}
$\hspace*{-10mm}\longrightarrow$
\begin{minipage}{5.9cm}
\begin{tikzpicture}[std/.style={rounded corners,draw=black,thick},
edge from parent path={[->] (\tikzparentnode) -- (\tikzchildnode)}]
\node[std] (A) {$Q, V, \prop, D, \usedevents$} 
child{ node[std] (B) {$Q_i, V, \prop, D_i, \usedevents_i$}
  child { node[std] {$Q_{i+1}, V, \prop, D_{i+1}, \usedevents_{i+1}$}
    child { node[std] {$Q', V, \prop, D', \usedevents'$}
      edge from parent node { } [dotted]
    }
    edge from parent node[right]{$\prop, D'_i \vee D''_i, \stdtransfproba{\prop}{p_i + p'_i}$}
  };
} ;
\draw [->,dotted] (A) -- (B);
\end{tikzpicture}%
\end{minipage}
\end{center}
\caption{Removing subtrees}\label{fig:remsubtree}
\end{figure}

Consider a proof tree that consists of a main branch
that is a sequence of applications of transformations
that satisfy property preservation with introduction of events
(properties of the form
$\dset, \emptyset: Q_i, D'_i, \usedevents_i \indistev{V}{p_i} Q_{i+1}, D_{i+1}, \usedevents_{i+1}$),
and side branches that may use any transformation to bound
the probability of Shoup and non-unique events.
All nodes on the main branch use the same security property $\prop$,
while nodes on side branches use $\true$ as security property.
Such a proof tree happens when we prove indistinguishability properties,
where $\prop$ is any distinguisher used in the definition of
indistinguishability (see Section~\ref{sec:proof_indist}). In particular,
the transformations \rn{guess}, \rn{guess\_branch}, and \rn{success simplify}
are not allowed in the main branch but may be used in side branches.
Lemma~\ref{lem:adv}, Property~\ref{item:prop:merge} allows one
to transform such a proof tree into a single property of the form
$\dset, \emptyset: Q, D, \usedevents \indistev{V}{p} Q', D', \usedevents'$.
Indeed, the main branch starts from the root $Q, V, \prop, D, \usedevents$
to a leaf $Q', V, \prop, D', \usedevents$, with additional subtrees starting from various nodes on this branch. Each node on the main branch is labeled $Q_i, V, \prop, D_i, \usedevents_i$, and the edge of the main branch that starts from $Q_i, V, \prop, D_i, \usedevents_i$ is labeled $\prop, D'_i, \stdtransfproba{\prop}{p_i}$ (see Figure~\ref{fig:remsubtree}).
Suppose an additional subtree starts from a node $Q_i, V, \prop, D_i, \usedevents_i$ with an edge labeled $\true, D''_i, \_$. This subtree yields a bound $\bound{Q_i}{V}{\true}{D''_i}{p'_i}$. The edge of the main branch from $Q_i, V, \prop, D_i, \usedevents_i$ yields a property $\dset, \emptyset: Q_i, D'_i, \usedevents_i \indistev{V}{p_i} Q_{i+1}, D_{i+1}, \usedevents_{i+1}$.
By Lemma~\ref{lem:adv}, Property~\ref{item:prop:merge}, the additional subtree can then be removed from the proof tree by replacing $\prop, D'_i, \stdtransfproba{\prop}{p_i}$ with $\prop, D'_i \vee D''_i, \stdtransfproba{\prop}{p_i + p'_i}$ as label of the edge of the main branch starting from the node $Q_i, V, \prop, D_i, \usedevents_i$.
By repeating this operation, we remove all additional subtrees, obtaining
a proof tree that consists of a single branch with nodes labeled $Q_i, V, \prop, D_i, \usedevents_i$, such that the edge that starts from $Q_i, V, \prop, D_i, \usedevents_i$ is labeled $\prop, D_i, \stdtransfproba{\prop}{p_i}$. This yields a sequence of properties 
$\dset, \emptyset: Q_i, D_i, \usedevents_i \indistev{V}{p_i} Q_{i+1}, D_{i+1}, \usedevents_{i+1}$, which yields a single such property 
$\dset, \emptyset: Q, D, \usedevents \indistev{V}{p} Q', D', \usedevents'$
by transitivity
(Lemma~\ref{lem:indistev}, Property~\ref{indisttrans}).

\subsubsection{Proof of Indistinguishability}\label{sec:proof_indist}

To prove indistinguishability between two games $G_0$ and $G_1$,
CryptoVerif finds a game $G_2$ such that
$\dset, \emptyset : G_0, D_{U0}, \usedevents \indistev{V}{p} G_2, D_2, \usedevents'$
and
$\dset, \emptyset : G_1, D_{U1}, \ab \usedevents_1 \ab \indistev{V}{p',D^+} G_2, D_2, \usedevents_1'$
where
$\dset$ is the set of all distinguishers,
$\usedevents = \cevent(G_0)$, $\usedevents_1 = \cevent(G_1)$,
$D_{U0} = \bigvee \{ e \mid \unique{e}\text{ occurs in }G_0\}$
and $D_{U1} = \bigvee \{ e \mid \unique{e}\text{ occurs in }G_1\}$.
The active queries $D_2$ are also required to be the same in both
sequences of games.
(In general, CryptoVerif builds proof trees; they can be transformed into the properties above by Lemma~\ref{lem:adv}, Property~\ref{item:prop:merge} as explained above. Only transformations that satisfy property preservation with introduction of events are allowed in the sequence of games that proves indistinguishability. The transformations \rn{guess}, \rn{guess\_branch}, and~\rn{success simplify} are not allowed in that sequence, but are allowed in side branches that bound the probability of introduced events.)
So for all evaluation contexts $C$ acceptable for $G_0$ and $G_2$ with public variables $V$ that do not contain events $\usedevents'$,
and all distinguishers $D_0\in\dset$ that run in time at most $t_{D_0}$,  
\begin{equation}
\Pr[C[G_0] : D_0 \vee D_{U0} ] \leq \Pr[C[G_2] : (D_0 \vee D_2) \wedge \neg \nonunique{G_2,D_2}] + p(C,t_{D_0})
\end{equation}
since $\nonunique{G_0,D_{U0}} = \Dfalse$,
and for all evaluation contexts $C$ acceptable for $G_1$ and $G_2$ with public variables $V$ that do not contain events in $\usedevents_1'$,
and all distinguishers $D_0 \in\dset$ that run in time at most $t_{D_0}$,  
\begin{equation}
\Pr[C[G_1] : D_0 \vee D_{U1} ] \leq \Pr[C[G_2] : (D_0 \vee D_2) \wedge \neg \nonunique{G_2,D_2}] + p'(C,t_{D_0})\,.\label{eq:indist2ndstep}
\end{equation}

Let $C$ be an evaluation context acceptable for $G_0$ and $G_1$ with public variables $V$.
After renaming the variables of $C$ that do not occur in $G_0$ and $G_1$ and the tables of $C$ that do not occur in $G_0$ and $G_1$ so that they do not occur in $G_2$, $C$ is also acceptable for $G_2$ with public variables $V$.
Furthermore, by Property~\ref{prop:prob}, this renaming does not change the probabilities.
Let $D_0 \in \dset$ be a distinguisher that runs in time at most $t_{D_0}$.
We rename the events of $C$ in $\usedevents'$ or $\usedevents'_1$
to some fresh events, and modify $D_0$ so that it considers the renamed
events as if they were the original events.
That does not change the probability $|\Pr[C[G_0] : D_0] - \Pr[C[G_1] : D_0]|$, and guarantees that $C$ does not contain events in $\usedevents'$ nor in $\usedevents'_1$.
So
\begin{align*}
  \Pr[C[G_0] : D_0] &\leq \Pr[C[G_0] : D_0 \vee D_{U0} ] \\*
  &\leq \Pr[C[G_2] : (D_0 \vee D_2) \wedge \neg \nonunique{G_2,D_2}] + p(C,t_{D_0})\\*
  &\leq \Pr[C[G_2] : ((D_0 \wedge \neg D_2) \vee D_2) \wedge \neg \nonunique{G_2,D_2}] + p(C,t_{D_0})  \\*
  &\leq \Pr[C[G_2] : (D_0 \wedge \neg D_2) \wedge \neg \nonunique{G_2,D_2}]\\*
  &\qquad {} + \Pr[C[G_2] : D_2 \wedge \neg \nonunique{G_2,D_2}] + p(C,t_{D_0})\\*
  &\leq \Pr[C[G_2] : D_0 \wedge \neg D_2] + \Pr[C[G_2] : D_2 \wedge \neg \nonunique{G_2,D_2}] + p(C,t_{D_0})
\end{align*}
By applying~\eqref{eq:indist2ndstep} to $\neg D_0$, which is also in $\dset$
and runs in the same time as $D_0$, we have
\begin{align*}
  1-\Pr[C[G_1] : D_0] & = \Pr[C[G_1] : \neg D_0]\\
  &\leq \Pr[C[G_1] : \neg D_0 \vee D_{U1} ] \\
  &\leq \Pr[C[G_2] : (\neg D_0 \vee D_2) \wedge \neg \nonunique{G_2,D_2}] + p'(C,t_{D_0})\\
  &\leq \Pr[C[G_2] : \neg D_0 \vee D_2] + p'(C,t_{D_0})\\
  &\leq 1- \Pr[C[G_2] : D_0 \wedge \neg D_2] + p'(C,t_{D_0})
\end{align*}
so
\[- \Pr[C[G_1] : D_0] \leq - \Pr[C[G_2] : D_0 \wedge \neg D_2] + p'(C,t_{D_0})\]
so
\[
\Pr[C[G_0] : D_0] - \Pr[C[G_1] : D_0]\leq \Pr[C[G_2] : D_2 \wedge \neg \nonunique{G_2,D_2}] + p(C,t_{D_0}) + p'(C,t_{D_0})
\]
By applying the formula above to $\neg D_0$, which runs in the same time as $D_0$, we have
\[  \Pr[C[G_1] : D_0] - \Pr[C[G_0] : D_0]\leq \Pr[C[G_2] : D_2 \wedge \neg \nonunique{G_2,D_2}] + p(C,t_{D_0}) + p'(C,t_{D_0})\]
so
\[
|\Pr[C[G_0] : D_0] - \Pr[C[G_1] : D_0]| \leq \Pr[C[G_2] : D_2 \wedge \neg \nonunique{G_2,D_2}] + p(C,t_{D_0}) + p'(C,t_{D_0})
\]
so $G_0 \approx_{p''}^V G_1$ where $p''(C, t_{D_0}) = \Pr[C[G_2] : D_2 \wedge \neg \nonunique{G_2,D_2}] + p(C,t_{D_0}) + p'(C,t_{D_0})$.

\subsubsection{Proof of $\kw{query\_equiv}$}

\paragraph{Decisional case ($\kw{query\_equiv}$ without $[\kwf{computational}]$ annotation)}

The situation is similar to the proof of indistinguishability,
but the property we want to prove is $\dset_{\neg\usedevents_1}, \emptyset : G_0, \Dfalse, \emptyset \indistev{V}{p} G_1, D_1, \usedevents_1$
where $G_0$ contains no events, 
$D_1 = e_1 \vee \ldots \vee e_m$, $e_1, \ldots, e_m$ are the Shoup events occurring in $G_1$, $e'_1, \dots, e'_l$ are the non-unique events occurring in $G_1$,
$\usedevents_1 = \{ e_1, \ldots, e_m, e'_1, \dots, e'_l \}$, $V = \emptyset$.

We need to show that, for all evaluation contexts $C$ acceptable for $G_0$ and $G_1$ without public variables that do not contain events in $\usedevents_1$ and
all distinguishers $D_0 \in \dset_{\neg\usedevents_1}$,
\[
\Pr[C[G_0] : D_0 ] \leq \Pr[C[G_1] : (D_0 \vee D_1) \wedge \neg \nonunique{G_1,D_1}] + p(C,t_{D_0})
\]
CryptoVerif finds a game $G_2$ such that
$\dset_{\neg\usedevents_0'}, \emptyset : G_0, \Dfalse, \emptyset \indistev{V}{p_0} G_2, D_2, \usedevents_0'$
% in this sequence, event in $D_1$ are introduced without creating queries for them... so same queries on both sides will yield $D_2 = D_1 \vee D_2'$
and
$\dset_{\neg\usedevents_1''}, \emptyset : G_1, D_{U1}, \usedevents_1 \indistev{V}{p_1} G_2, D_2', \usedevents_1'$
where $\usedevents_1'' = \usedevents_1' \setminus \{ e_1, \ab \ldots, \ab e_m \}$,
$D_{U1} = e'_1 \vee \ldots \vee e'_l = \nonunique{G_1,D_1}$,
and $D_2 = D_1 \vee D_2'$.
(The events $e_1, \ldots, e_m$ must be preserved by the second proof,
hence we allow distinguishers in $\dset_{\neg\usedevents_1''}$ to use these
events. 
The events $e_1, \ldots, e_m$ will be introduced in the first proof,
and the active queries in $G_2$ are also required to match so
$D_2 = D_1 \vee D_2'$.)
So for all evaluation contexts $C$ acceptable for $G_0$ and $G_2$ without public variables that do not contain events in $\usedevents'_0$,
and all distinguishers $D_0 \in \dset_{\neg\usedevents_0'}$ that run in time at most $t_{D_0}$,  
\[
\Pr[C[G_0] : D_0] \leq \Pr[C[G_2] : (D_0 \vee D_2) \wedge \neg \nonunique{G_2,D_2}] + p_0(C,t_{D_0})
\]
and for all evaluation contexts $C$ acceptable for $G_1$ and $G_2$ without public variables that do not contain events in $\usedevents_1'$,
and all distinguishers $D_0 \in \dset_{\neg\usedevents_1''}$ that run in time at most $t_{D_0}$,
\[
\Pr[C[G_1] : D_0 \vee D_{U1} ] \leq \Pr[C[G_2] : (D_0 \vee D'_2) \wedge \neg \nonunique{G_2,D'_2}] + p_1(C,t_{D_0})
\]
In the last equation, we replace $D_0$ with $\neg D_0 \wedge \neg D_1$ for $D_0 \in \dset_{\neg\usedevents_1'}$, yielding
\begin{align*}
  &\Pr[C[G_1] : (\neg D_0 \wedge\neg D_1) \vee D_{U1} ]\\*
  &\qquad \leq \Pr[C[G_2] : ((\neg D_0 \wedge\neg D_1) \vee D_2') \wedge \neg \nonunique{G_2,D'_2}] + p_1(C,t_{D_0})
\end{align*}
so
\begin{align*}
  &\Pr[C[G_2] : ((D_0 \vee D_1) \wedge \neg D_2') \vee \nonunique{G_2,D'_2}]\\*
  &\qquad \leq   \Pr[C[G_1] : (D_0 \vee D_1) \wedge \neg D_{U1} ] + p_1(C,t_{D_0})
\end{align*}
Then we get
\begin{align*}
  \Pr[C[G_0] : D_0] &\leq \Pr[C[G_2] : (D_0 \vee D_2) \wedge \neg \nonunique{G_2,D_2}] + p_0(C,t_{D_0})\\
  &\leq \Pr[C[G_2] : (D_0 \vee D_1 \vee D'_2) \wedge \neg \nonunique{G_2,D'_2}] + p_0(C,t_{D_0})\\
  &\leq \Pr[C[G_2] : ((D_0 \vee D_1) \wedge \neg D'_2)\vee \nonunique{G_2,D'_2}]\\*
  &\qquad {} + \Pr[C[G_2] : D'_2\wedge \neg \nonunique{G_2, D'_2}] + p_0(C,t_{D_0})\\
  &\leq \Pr[C[G_1] : (D_0 \vee D_1) \wedge\neg D_{U1}]\\*
  &\qquad {} + p_1(C, t_{D_0}) + \Pr[C[G_2] : D'_2 \wedge \neg \nonunique{G_2,D'_2}] + p_0(C,t_{D_0})\\
  &\leq \Pr[C[G_1] : (D_0 \vee D_1) \wedge \neg \nonunique{G_1,D_1}]\\*
  &\qquad{}+ p_1(C, t_{D_0}) + \Pr[C[G_2] : D'_2\wedge \neg \nonunique{G_2,D'_2}] + p_0(C,t_{D_0})
\end{align*}
so we get the desired result with $p(C, t_{D_0}) = p_1(C, t_{D_0}) + \Pr[C[G_2] : D'_2\wedge \neg \nonunique{G_2,D'_2}] + p_0(C,t_{D_0})$.

%\emph{Currently, when Shoup events are introduced, they are always fresh (except the events of $D_1^+$),
%  so the same events cannot be introduced on both sides; since both proofs
%reach the same intermediate game $G_2$, we have $\Pr[C[G_2] : D^+_2] = 0$.}

\paragraph{Computational case ($\kw{query\_equiv}$ with $[\kwf{computational}]$ annotation)}

\newcommand{\distinguish}{\kwf{distinguish}}

As in the decisional case, we want to prove $\dset_{\neg\usedevents_1}, \emptyset : G_0, \Dfalse, \emptyset \indistev{V}{p} G_1, D_1, \usedevents_1$
where $G_0$ contains no events, 
$D_1 = e_1 \vee \ldots \vee e_m$, $e_1, \ldots, e_m$ are the Shoup events occurring in $G_1$, $e'_1, \dots, e'_l$ are the non-unique events occurring in $G_1$,
$\usedevents_1 = \{ e_1, \ldots, e_m, e'_1, \dots, e'_l \}$, $V = \emptyset$.
Additionally, we want to show that the random values of $G_0$ and
$G_1$ marked $[\kwf{unchanged}]$ can be used in events (different from
$e_1, \ldots, e_m$ since $e_1, \ldots, e_m$ have no arguments) in the
game transformed using this assumption. That corresponds to adding
oracles that execute the same arbitrary events using
$[\kwf{unchanged}]$ random values to both $G_0$ and $G_1$. We write
$G_0'$ and $G'_1$ for the games $G_0$ and $G_1$ respectively with
additional events and show $\dset_{\neg\usedevents_1}, \emptyset : G'_0, \Dfalse, \usedevents_0 \indistev{V}{p} G'_1, D_1, \usedevents_1 \cup \usedevents_0$
where $\usedevents_0$ contains these additional events.
These additional events can be observed by the adversary, so
they are allowed in distinguishers in $\dset_{\neg\usedevents_1}$.

Let $D_{U1} = e'_1 \vee \ldots \vee e'_l = \nonunique{G_1,D_1}$.

Let us write $O_0(\tup{r}, \tup{\mathit{args}})$ (resp. $O_1(\tup{r},
\tup{\mathit{args}})$) for the result of oracle $O$ in game $G_0$ (resp. $G_1$)
with randomness $\tup{r}$ and arguments $\tup{\mathit{args}}$.

In order to establish this property, we show that there exists
a mapping $\phi$ of the randomness, such that if random value variable
$r$ has value $v$ in $G_0$, then it has value $\phi_r(v)$ in $G_1$,
$\phi_r$ is the identity when the variable $r$ is marked $[\kwf{unchanged}]$,
$\phi_r$ preserves the probability distribution of variable $r$,
and we define a game $G_2$ in which oracle $O$ with randomness $\tup{r}$
and arguments $\tup{\mathit{args}}$ returns
\begin{align*}
  & \assign{x_0}{O_0(\tup{r}, \tup{\mathit{args}})}\\
  & \assign{x_1}{O_1(\phi(\tup{r}), \tup{\mathit{args}})}\\
  & \bguard{x_0 = x_1}{x_0}{\keventabort{\distinguish}}
\end{align*}
We bound
\begin{align*}
  p(C) &= \Pr[C[G_2] : \distinguish \vee \nonunique{G_2,\Dfalse}] \\
  &= \Pr[C[G_2] : \distinguish \vee D_{U1}]\\
  & = \Advtev{G_2}{\distinguish \Rightarrow\false}{C}{D_{U1}}
\end{align*}
From this bound, we infer the desired property.

We consider a game $G_3$ in which oracle $O$ returns
\begin{align*}
  & \assign{x_0}{O_0(\tup{r}, \tup{\mathit{args}})}\\
  & \assign{x_1}{O_1(\phi(\tup{r}), \tup{\mathit{args}})}\\
  & x_0
\end{align*}
We define $G_2'$ as $G_2$ with the same additional events as in $G_0'$ and $G_1'$, and $G_3'$ as $G_3$ with the same additional events as in $G_0'$ and $G_1'$.
The game $G_3'$ behaves as $G_0'$ except that it executes a Shoup event $e_i$ or a non-unique event when $G_1$ does, so we have, for any evaluation context $C$ acceptable for $G'_0$ and $G'_3$ without public variables, and any distinguisher $D_0$,
\[
\Pr[C[G'_0] : D_0] \leq \Pr[C[G'_3] : D_0 \vee D_1 \vee D_{U1}]
\]
Moreover, $G'_3$ behaves as $G_1'$ except when $G'_2$ executes event $\distinguish$, that is, when $G_2$ executes event $\distinguish$ (the additional events introduced in $G'_2$ are not needed to evaluate the probability of $\distinguish$ since we do not consider their probability), so for all $D$,
\begin{equation}
  |\Pr[C[G'_3] : D] - \Pr[C[G'_1] : D]| \leq \Pr[C[G_2]:\distinguish]
  \label{eq:dist-comput}
\end{equation}
Let $C$ be any evaluation context acceptable for $G'_0$ and $G'_2$ without public variables. By renaming $x_0$ and $x_1$ to variables not in $C$, $C$ is also acceptable for $G'_0$ and $G'_3$ without public variables.
Let $D_0$ be any distinguisher.
With $D = D_0 \vee D_1 \vee D_{U1}$ in~\eqref{eq:dist-comput},
we obtain
\begin{align*}
  \Pr[C[G'_0] : D_0] &\leq \Pr[C[G'_3] : D_0 \vee D_1 \vee D_{U1}]\\
  &\leq \Pr[C[G'_1] : D_0 \vee D_1 \vee D_{U1}] + \Pr[C[G_2]:\distinguish]\\
  &\leq \Pr[C[G'_1] : (D_0 \vee D_1) \wedge \neg \nonunique{G_1,D_1}] + \Pr[C[G'_1] :D_{U1}]\\*
  &\qquad {} + \Pr[C[G_2]:\distinguish]\\
  &\leq \Pr[C[G'_1] : (D_0 \vee D_1) \wedge \neg \nonunique{G_1,D_1}] + p(C)
\end{align*}
since $G'_1$ and $G_2$ execute events $D_{U1}$ in the same cases, and $D_{U1}$ and $\distinguish$ are mutually exclusive.
Therefore, we have $\dset_{\neg\usedevents_1}, \emptyset : G'_0, \Dfalse, \usedevents_0 \indistev{V}{p} G'_1, D_1, \usedevents_1 \cup \usedevents_0$.

\emph{Currently, CryptoVerif can prove $\kw{query\_equiv}$ only when the mapping $\phi$ is the identity for all variables. Other cases can be proved manually and used as assumptions in $\kw{equiv}$ statements.}

\subsection{Turing Machine Adversary}\label{sec:Turing_adv}

In CryptoVerif, the adversary is modeled as an evaluation context.
However, usually, in cryptographic results, an adversary is a
bounded-time probabilistic Turing machine. In this section, we explain
how any bounded-time probabilistic Turing machine that communicates on channels
can be represented as a CryptoVerif evaluation context.

Let $Q_0$ be the initial game that interacts with an adversary.
Let $c_1, \ldots, c_k$ be the channels used in $Q_0$. 
Let $\Tall$ be the union of all types that occur in $Q_0$.
Let $\Tall'$ be the type of pairs containing the encoding a channel as first
component and an element of $\Tall$ as second component.
The encoding of a channel is either the constant $\mathit{yield}$
or a tuple of integers $(j, i_1, \ldots, i_{k'})$ with $1 \leq j \leq k$.
(We assume that unambiguous tuples can be encoded as CryptoVerif values,
and that the constant $\mathit{yield}$ is different from a tuple.)
Let $d_0$, $d_1$, and $d_2$ be channels that do not occur in $Q_0$.

Let $Q_1$ be a process that contains the parallel composition of processes
\[\repl{i_1}{n_1}\ldots \repl{i_{k'}}{n_{k'}} \cinput{c_j[i_1, \ldots, i_{k'}]}{x:\Tall}.\coutput{d_0}{((j, i_1, \ldots, i_{k'}), x)}\]
for each output $\coutput{c_j[i'_1, \ldots, i'_{k'}]}{N}$ that occurs under
$\repl{i'_1}{n_1}\ldots \repl{i'_{k'}}{n_{k'}}$ in $Q_0$.
Since, in the initial game $Q_0$, the channels of all outputs use the current replication indices as channel indices, as in $c_j[i'_1, \ldots, i'_{k'}]$, a single output is executed for each value of the indices and for each syntactic occurrence of the output, so the inputs in $Q_1$ can receive all outputs made by $Q_0$. The process $Q_1$ forwards all these outputs to the same channel $d_0$, with a message that specifies both the channel $c_j[i_1, \ldots, i_{k'}]$ on which $Q_0$ emitted (encoded as a bitstring) and the message $x$ sent by $Q_0$.

In addition, $Q_1$ also contains the parallel composition of processes
\[\repl{i_1}{n_1}\ldots \repl{i_{k'}}{n_{k'}} \cinput{\mathit{yield}}{}.\coutput{d_0}{(\mathit{yield}, ())}\]
for each occurrence of $\kw{yield}$ that occurs under
$\repl{i'_1}{n_1}\ldots \repl{i'_{k'}}{n_{k'}}$ in $Q_0$, to receive all outputs
that come from the $\kw{yield}$ construct.

\begin{figure}
\begin{align*}\firstline\linelabel{line:start}
\hspace*{-1cm}Q_2 = {}&\repl{i}{n} \cinput{d_1[i]}{s:\bitstring};\\\nextline
\linelabel{line:Turingrun}&\assign{(s',o,v)}{f(s)}\\
&&&\qquad \textit{Lines~\ref{line:startcopy}--\ref{line:endcopy} are repeated for each $j \leq k$ and each $k'$}\\
&&&\qquad \textit{such that there is an input on channel $c_j[i'_1, \ldots, i'_{k'}]$ in $Q_0$.}\\\nextline
\linelabel{line:startcopy}\linelabel{line:testmessage}&\baguard{o = (j,k')}\\\nextline
\linelabel{line:out}&\qquad \assign{(a_1, \ldots, a_{k'}, b)}{v}
%&\qquad \assign{(a_1:[1,n_{\max}], \ldots, a_{k'}:[1,n_{\max}], b:\Tall)}{v}
 \coutput{c_j[a_1, \ldots, a_{k'}]}{b};\\\nextline
\linelabel{line:in}&\qquad \cinput{d_0}{s'':\Tall'};
 \coutput{d_1[i+1]}{f'(s',s'')}\\\nextline
\linelabel{line:endcopy}&\ELSE\\[5mm]\nextline
\linelabel{line:testrandom}&\baguard{o = \mathit{random}}\\\nextline
\linelabel{line:random}&\qquad \Res{x}{\bool};
 \coutput{d_1[i+1]}{f''(s', x)}\\\nextline
&\ELSE\\\nextline
\linelabel{line:testabort}&\baguard{o = \mathit{abort}}\\\nextline
\linelabel{line:abort}&\qquad \keventabort{e}\\\nextline
&\ELSE\\\nextline
\linelabel{line:stop}&\qquad \coutput{d_2}{}
\end{align*}
\caption{Looping process}\label{fig:loopprocess}
\end{figure}

Let $C = \Reschan{d_0}; \Reschan{d_1}; \Reschan{d_2}; (\cinput{\startch}{}.\coutput{d_1[1]}{s_0} \parpop Q_1 \parpop Q_2 \parpop [\,])$, where the process $Q_2$ is defined in Figure~\ref{fig:loopprocess}. Let us explain how the context $C$ can simulate any Turing machine interacting with the process $Q_0$.

The current state of the Turing machine is sent on channel $d_1[i]$ where $i$ is a loop index that starts at 1 and increases during execution.
As shown in the semantics of CryptoVerif, upon startup, a message is sent on channel $\startch$. When $C$ receives that message, it sends the initial state of the Turing machine $s_0$ on channel $d_1[1]$. This message is received by process $Q_2$ (line~\ref{line:start}). Then $Q_2$ calls the function $f$ on the current state $s$ of the Turing machine (line~\ref{line:Turingrun}). This function executes the Turing machine, until one of the following situations happens:
\begin{itemize}
\item The Turing machine sends a message $b$ on a channel $c_j[a_1, \ldots, a_{k'}]$; in this case, $f$ returns $(s', (j, k'), (a_1, \ldots, a_{k'}, b))$, where $s'$ is the new state of the Turing machine. 
The test at line~\ref{line:testmessage} is then going to succeed for the appropriate value of $j,k'$, and the desired message is going to be sent at line~\ref{line:out}. After receiving a message, the process $Q_0$ always replies by sending a message (except if it aborts). This message is going to be received by $Q_1$, which is going to forward on $d_0$ the channel and the received message. These channel and message are then received as $s''$ at line~\ref{line:in}. Then $f'(s',s'')$ is the new state of the Turing machine after receiving that message. This state is sent on channel $d_1[i+1]$, which restarts a new iteration of $Q_2$.

\item The Turing machine generates a fresh random bit; in this case, $f$ returns $(s', \ab \mathit{random}, ())$ where $s'$ is the new state of the Turing machine.
The test at line~\ref{line:testrandom} is then going to succeed. At line~\ref{line:random}, a random bit $x$ is chosen. Then $f''(s', x)$ is the new state of the Turing machine with that random bit. This state is sent on channel $d_1[i+1]$, which restarts a new iteration of $Q_2$ as in the previous case.

\item The Turing machine aborts; in this case, $f$ returns $(s', \mathit{abort}, ())$. The test at line~\ref{line:testabort} is then going to succeed, and the process aborts at line~\ref{line:abort}. (The event $e$ is any event not used elsewhere; the event is not really useful, it is present because the CryptoVerif language always executes an event before aborting.)

\item The Turing machine stops; in this case, $f$ returns $(s', \mathit{stop}, ())$. No test succeeds, so line~\ref{line:stop} is executed. The process tries to send a message on channel $d_2$, but there is no input on this channel, so the process blocks.

\end{itemize}
The constants $\mathit{random}$, $\mathit{abort}$, and $\mathit{stop}$ are
assumed to be pairwise distinct, and distinct from all pairs.

The function $f$ is a CryptoVerif primitive, because it can be implemented
by a deterministic bounded-time Turing machine. (Recall that $f$ stops when
the initial probabilistic Turing machine makes a random choice, and the
random choice is performed by CryptoVerif at 
lines~\ref{line:testrandom}--\ref{line:random}.)
Similarly, the function $f'$ that computes the new state of the Turing machine from the old state and the received message, and the function $f''$ that computes the new state of the Turing machine from the old state and a random bit are CryptoVerif primitives.

The replication bound $n$ (used in $Q_2$, line~\ref{line:start}) is chosen large enough so that the loop never stops due to that bound: the Turing machine aborts or stops before the bound is reached. This is possible since the Turing machine runs in bounded time, so sends a bounded number of messages and chooses a bounded number of random bits.

Notice that, if $Q_0$ sends and receives messages on the same channels, it may happen that a message sent by $Q_0$ is immediately received by $Q_0$ without being intercepted by the adversary. In this case, since both $Q_0$ and $Q_1$ are going to listen on the same channels, the destination of the message (either the honest process $Q_0$ or the adversary $Q_1$) is chosen randomly with uniform probability, depending on the number of available receivers. Therefore, adding more copies of the receiving processes in $Q_1$ increases the probability that the adversary receives the message. 
Moreover, when the same channel is used for both inputs and outputs, the messages sent by $Q_2$ at line~\ref{line:out} may be received back by the adversary via $Q_1$, instead of being received by $Q_0$.
We recommend avoiding this strange situation, by using distinct channels for inputs on the one hand and outputs on the other hand. 
More generally, we recommend using distinct channels for each input and output, so that the adversary gets full control of the network, as already mentioned page~\pageref{distinctchannels}.

As a slight extension, it would still be possible to allow $Q_0$ to output on $c_j[i_1, \ldots, i_{k'}]$ after receiving a message on the same channel $c_j[i_1, \ldots, i_{k'}]$. In this case, a message sent by $Q_0$ on $c_j[i_1, \ldots, i_{k'}]$ cannot be received by $Q_0$, because the input on $c_j[i_1, \ldots, i_{k'}]$ is no longer available when the output on $c_j[i_1, \ldots, i_{k'}]$ is performed by $Q_0$. Moreover, the problem that messages sent by $Q_2$ at line~\ref{line:out} may be received back by the adversary via $Q_1$, instead of being received by $Q_0$, can be avoided by putting the receiver process
\[\cinput{c_j[a_1, \ldots, a_{k'}]}{x:\Tall}.\coutput{d_0}{((j, a_1, \ldots, a_{k'}), x)}\]
after $\coutput{c_j[a_1, \ldots, a_{k'}]}{b}$ in parallel with $\cinput{d_0}{s'':\Tall'}; \coutput{d_1[i+1]}{f'(s',s'')}$ in $Q_2$, instead of including
\[\repl{i_1}{n_1}\ldots \repl{i_{k'}}{n_{k'}} \cinput{c_j[i_1, \ldots, i_{k'}]}{x:\Tall}.\coutput{d_0}{((j, i_1, \ldots, i_{k'}), x)}\]
in $Q_1$.

The context $C$ does not allow the Turing machine to execute events of its choice, while a CryptoVerif context can execute events. We could obviously extend the model to allow the Turing machine to execute events, but this is not needed for the cases we consider. Indeed, if the adversary represented as a CryptoVerif context executes events, these events can be deleted without changing the final result returned by the distinguisher: for correspondences, by Definition~\ref{def:proccorresp}, the context is not allowed to contain events used by $\corresp$, and all other events are ignored by the distinguisher $\neg \corresp$; for one-session secrecy, secrecy, and bit secrecy, by Definitions~\ref{def:secr} and~\ref{def:secrbit}, the context is not allowed to contain $\sevent$ nor $\sbarevent$, and all other events are ignored by the distinguishers $\sevent$ and $\sbarevent$.

To sum up, the context given in this section allows us to run any probabilistic
bounded-time Turing machine as a CryptoVerif context, so CryptoVerif contexts are powerful enough to represent the adversaries usually considered by cryptographers.

\section{Collecting True Facts}\label{sec:truefacts}

In this section, we consider only processes that satisfy Properties~\ref{prop:notables} and~\ref{prop:autosarename}.
We can assume without loss of generality that the adversary also satisfies these properties: 
the Turing machine adversary encoded in Section~\ref{sec:Turing_adv} satisfies them
and tables ($\INSERT$ and $\GET$) can be removed by encoding them using $\FIND$ by transformation \rn{expand\_tables} (Section~\ref{sec:transftables}) and variables defined in conditions of $\FIND$ can be renamed to have distinct names by transformation \rn{auto\_SArename} (Section~\ref{sec:autosarename}).

Given a configuration $\conf = E, \sigma, N, \tblcts, \evseq$ or
$\conf = E, (\sigma, P)\restconfig$ or $\conf = E, \pset, \cset$, we denote by $E_{\conf}$ the
environment $E$ in configuration $\conf$.
We denote by $E_{\trace \before \conf}$ the union of $E_{\conf'}$ for all configurations $\conf' \beforetr{\trace} \conf$
in $\trace$. It is a set of mappings $x[\tup{a}] \mapsto b$. At this stage, it may include conflicting mappings
$x[\tup{a}] \mapsto b$ and $x[\tup{a}] \mapsto b'$ with $b \neq b'$. We prove below (Lemma~\ref{lem:stronginv1exec}) 
that this situation never happens.
The notation $E_{\trace \before \conf}$ is useful because the environment computed in the semantics does not keep the values of variables defined in conditions of $\FIND$ after these conditions are evaluated. Considering the union of all environments of previous configurations allows us to recover the values of these variables, and to use them in the facts that we collect.
Given a configuration $\conf = E, \sigma, N, \tblcts, \evseq$ or
$\conf = E, (\sigma, P)\restconfig$, we denote by $\sigma_{\conf}$ the
mapping sequence for replication indices $\sigma$ in the configuration $\conf$
and by $\evseq_{\conf}$ the sequence of events $\evseq$ in configuration $\conf$.

Let us define $\defset$ as in Section~\ref{sec:inv1exec}, except that
{\allowdisplaybreaks\begin{align*}
&\defset(\sigma, \pptag \FIND\uniqueopt \ (\mathop{\textstyle\bigoplus}\nolimits_{j=1}^m
\tup{\vf_j}[\tup{i}] = \tup{i_j} \leq \tup{n_j}\ \SUCHTHAT \ 
\defined(\tup{M_j}) \fand M_j \THEN N_j)\ \ELSE N) = \\*
&\ \left(\bigmultiunion_{j = 1}^m \bigmultiunion_{\tup{a} \leq \tup{n_j}} \defset(\sigma[\tup{i_j} \mapsto \tup{a}], M_j)\!\right)\! \multiunion
\max\left(\max_{j = 1}^m \left(\{ \tup{\vf_j}[\sigma(\tup{i})] \} \multiunion
\defset(\sigma, N_j)\right), \defset(\sigma, N)\!\right)\\
&\defset(\sigma, \pptag \FIND\uniqueopt \ (\mathop{\textstyle\bigoplus}\nolimits_{j=1}^m
\tup{\vf_j}[\tup{i}] = \tup{i_j} \leq \tup{n_j}\ \SUCHTHAT \ 
\defined(\tup{M_j}) \fand M_j \THEN P_j)\ \ELSE P) = \\*
&\ \left(\bigmultiunion_{j = 1}^m \bigmultiunion_{\tup{a} \leq \tup{n_j}} \defset(\sigma[\tup{i_j} \mapsto \tup{a}], M_j)\!\right)\! \multiunion
\max\left(\max_{j = 1}^m \left(\{ \tup{\vf_j}[\sigma(\tup{i})] \} \multiunion
\defset(\sigma, P_j)\right), \defset(\sigma, P)\!\right)
\end{align*}}%
so that the variables defined in conditions of $\FIND$ are now considered as defined forever, and not temporarily
during the evaluation of the considered condition. We also define
\[\defset(\trace \before \conf) = \dom(E_{\trace \before \conf}) \multiunion \defsetfut(\conf)\,.\]

\begin{lemma}\label{lem:stronginv1exec}
Let $Q_0$ be a process that satisfies Properties~\ref{prop:notables} and~\ref{prop:autosarename}.
Let $\trace$ be a trace of $Q_0$ and $\conf$ be a configuration in the derivation of $\trace$.
Then the following properties hold:
\begin{enumerate}
\item $\defset(\trace \before \conf)$ does not contain duplicate elements.
\item Each variable is defined at most once for each value of its array
indices in $\trace$.
\item $E_{\trace \before \conf}$ contains at most one binding for each $x[\tup{a}]$.
\end{enumerate}
\end{lemma}
\begin{proofsk}
The proof is similar to the proof of Lemma~\ref{lem:inv1exec}.
  We first show as in Lemma~\ref{lem:inv1exec} that, 
  for all program points $\pp$ in $Q_0$, if
  $\dom(\sigma) = \replidx{\pp}$ are the current replication indices at $\pp$
  and the process or term $Q$ at $\pp$
  satisfies Invariant~\ref{inv1}, then all elements of
  $\defset(\sigma, Q)$ are of the form $x[\tup{a}]$ where $x \in
  \vardef(Q)$ and $\image(\sigma)$ is a prefix of $\tup{a}$. 

  Next, we show that, for all program points $\pp$, if
  $\dom(\sigma) = \replidx{\pp}$ are the current replication indices at $\pp$
  and the process or term $Q$ at $\pp$
  satisfies Invariant~\ref{inv1}, then $\defset(\sigma, Q)$ does not
  contain duplicate elements.
  The proof proceeds by induction on $Q$.
  All multiset unions in the computation of $\defset(\sigma, Q)$ are
  disjoint unions by the property above, because either they use
  different extensions of $\sigma$ (cases of replication and of conditions of $\FIND$) or they use
  disjoint variable definitions or subprocesses or subterms in the
  same branch of $\FIND$ or $\kw{if}$, which must define different
  variables by Invariant~\ref{inv1} and by Property~\ref{prop:autosarename}.

  We show by induction on the derivations that, if
  $\conf \red{p}{\ix} \conf'$, then
  $\defset(\trace\before \conf) \supseteq \defset(\trace\before
  \conf')$
  and for all semantic configurations $\conf''$ in the derivation of
  $\conf \red{p}{\ix} \conf'$,
  $\defset(\trace\before \conf) \supseteq \defset(\trace\before
  \conf'')$, and similarly with $\redq$ instead of $\red{p}{\ix}$.

The first result follows: since $Q_0$ satisfies Invariant~\ref{inv1},
$\defset(\sigma_0, Q_0)$ does not
contain duplicate elements, where $\sigma_0$ is the empty mapping sequence.
Let $\conf_0 = \emptyset, \{ (\sigma_0, Q_0) \}, \fc(Q_0)$, 
 $\conf_1 = \reduce(\emptyset, \{ (\sigma_0, Q_0) \}, \fc(Q_0))$, 
 $\conf_2 = \initconfig(Q_0)$, and $\conf_3$ be any other configuration of~$\trace$.
Then $\defset(\trace \before \conf_0)$, $\defset(\trace \before \conf_1)$,
$\defset(\trace \before \conf_2)$, and therefore $\defset(\trace \before \conf_3)$ do not contain duplicate elements.

Let us prove the second result. In order to derive a contradiction, assume that
two transitions $\conf_1 \red{p_1}{\ix_1} \conf'_1$ and $\conf_2 \red{p_2}{\ix_2} \conf'_2$
inside $\trace$ define the same variable $x[\tup{a}]$.
\begin{itemize}
\item First case: one transition happens before the other, for instance $\conf'_1 \beforetr{\trace} \conf_2$.
(The case $\conf'_2 \beforetr{\trace} \conf_1$ is symmetric.)
Since $\conf_1 \red{p_1}{\ix_1} \conf'_1$ defines $x[\tup{a}]$, we have $x[\tup{a}] \in \dom(E_{\conf'_1})$, so $x[\tup{a}] \in \dom(E_{\trace \before \conf_2})$.
Moreover, since $\conf_2 \red{p_2}{\ix_2} \conf'_2$ defines $x[\tup{a}]$, we have $x[\tup{a}] \in \defsetfut(\conf_2)$,
by inspecting all rules that add elements to the environment.
Therefore $\defset(\trace \before \conf_2) = \dom(E_{\trace \before \conf_2}) \multiunion \defsetfut(\conf_2)$ contains
twice $x[\tup{a}]$. Contradiction.

\item Second case: the transitions cannot be ordered. 
By definition of $\beforetr{\trace}$, this can happen only when a semantic rule uses several derivations
for its assumptions, which happens only in rules for $\FIND$. 
(Recall that $\GET$ is excluded by Property~\ref{prop:notables}.)
Therefore, there exists $k_1$ and $k_2$ such that 
 $\conf_1 \red{p_1}{\ix_1} \conf'_1$ is in the derivation of
$\conf_{0,k} = E, \sigma[\tup{i_{j_k}} \mapsto \tup{a}_k], D_{j_k} \wedge M_{j_k}, \tblcts, \evseq \red{p_k}{\ix_k}^* \conf'_{0,k} = E_k, \sigma_k, r_k, \tblcts, \evseq$ with $v_k = (j_k, \tup{a}_k)$ for $k = k_1$ and $\conf_2 \red{p_2}{\ix_2} \conf'_2$ is in that derivation for $k = k_2$, with $k_1 \neq k_2$.
We have $x[\tup{a}] \in \defset(\trace \before \conf'_{0,k_1}) \subseteq \defset(\trace \before\conf_{0,k_1}) = \dom(E_{\trace \before\conf_{0,k_1}}) \cup \defset(\sigma[\tup{i_{j_{k_1}}} \mapsto \tup{a}_{k_1}], M_{j_{k_1}})$.
Moreover, $\conf_{0,k} \before \conf_1$, so $\dom(E_{\trace \before\conf_{0,k_1}}) \subseteq \dom(E_{\trace \before \conf_1})$.
Since $\conf_1 \red{p_1}{\ix_1} \conf'_1$ defines $x[\tup{a}]$, we have $x[\tup{a}] \in \defsetfut(\conf_1)$,
by inspecting all rules that add elements to the environment.
Since $\defset(\trace \before \conf_1) = \dom(E_{\trace \before \conf_1}) \multiunion \defsetfut(\conf_1)$
does not contain duplicate elements, we have $x[\tup{a}] \notin \dom(E_{\trace \before \conf_1})$,
so $x[\tup{a}] \notin \dom(E_{\trace \before\conf_{0,k_1}})$.
Hence we have $x[\tup{a}] \in \defset(\sigma[\tup{i_{j_{k_1}}} \mapsto \tup{a}_{k_1}], M_{j_{k_1}})$.
Similarly, $x[\tup{a}] \in \defset(\sigma[\tup{i_{j_{k_2}}} \mapsto \tup{a}_{k_2}], M_{j_{k_2}})$.
Let us show that the sets $\defset(\sigma[\tup{i_{j_{k_1}}} \mapsto \tup{a}_{k_1}], M_{j_{k_1}})$
and $\defset(\sigma[\tup{i_{j_{k_2}}} \mapsto \tup{a}_{k_2}], M_{j_{k_2}})$ are disjoint.
We have $v_{k_1} \neq v_{k_2}$, so either $j_{k_1} \neq j_{k_2}$ and in this case
these sets are disjoint because
$M_{j_{k_1}}$ and $M_{j_{k_2}}$ define different variables by Property~\ref{prop:autosarename},
or $j_{k_1} = j_{k_2}$ and $\tup{a}_{k_1} \neq \tup{a}_{k_2}$ and in this case
these sets are disjoint because they use different extensions of $\sigma$.
Since these sets are disjoint, they cannot both contain $x[\tup{a}]$.
Contradiction.

\end{itemize}

The last result is an immediate consequence of the second one.
\proofcomplete
\end{proofsk}

\begin{lemma}\label{lem:add_find_cond}
Let $Q_0$ be a process that satisfies Properties~\ref{prop:notables} and~\ref{prop:autosarename}.
Let $\trace$ be a trace of $Q_0$ and $\conf$ be a configuration in the derivation of $\trace$.
We have $E_{\trace\before\conf} = E_{\conf}[x[\tup{a}]\mapsto b $ for some variables $x$ defined in a condition of $\FIND$ and some indices $\tup{a}$ and values $b]$.
\bbnote{This lemma will be useful to prove the soundness of elsefind facts.}%
\end{lemma}
\begin{proofsk}
  By induction on the derivation.
  \proofcomplete
\end{proofsk}

The previous lemma shows that the only difference between $E_{\conf}$ and
$E_{\trace\before\conf}$ is that variables defined in conditions of $\FIND$ are
added to $E_{\trace\before\conf}$. These variables have no array accesses,
so they do not appear in $\defined$ conditions of $\FIND$. Therefore,
these $\defined$ conditions yield the same result whether they are
evaluated in $E_{\conf}$ or in $E_{\trace\before\conf}$.

We use \emph{facts} the represent properties that hold at certain program
points in processes.
We consider the following facts:
\begin{itemize}
\item The boolean term $M$ means that $M$ evaluates to $\true$.
\item $\defined(M)$ means that $M$ is defined (all array accesses in $M$ are defined).
\item $\fevent{e(\tup{M})}$ means that event $e(\tup{M})$ has been executed.
\item $\fevent{e(\tup{M})}@\step$ means that event $e(\tup{M})$ has been executed at step $\step$ (index in the sequence of events $\evseq$).
\item $M_1:\fevent{e(\tup{M})}$ means that event $e(\tup{M})$ has been executed with pair (program point, replication indices) equal to $M_1$.
\item $M_1:\fevent{e(\tup{M})}@\step$ means that event $e(\tup{M})$ has been executed at step $\step$ with pair (program point, replication indices) equal to $M_1$.
\item $\ppf(\pp, \tup{M})$ means that program point $\pp$ has been executed with replication indices equal to $\tup{M}$.
\item $\ppf(\sset_1, \tup{M}_1) \before \dots \before \ppf(\sset_m, \tup{M}_m)$ means that, for $j \leq m$, some program point $\pp_j \in \sset_j$ has been executed with replication indices equal to $\tup{M}_j$, and furthermore these program points have been executed in the order of increasing $j$.
\bbnote{models the history used in the collection of true facts}%
\item $\lppf(\pp, \tup{M})$ means that program point $\pp$ has been executed with replication indices equal to $\tup{M}$ and the values of variables and replication indices are unchanged since that program point (that is, no variable definition nor output that changes the replication indices was executed since that program point).
\end{itemize}
Given an environment $E$ mapping process variables to their values,
an environment $\venv$ mapping replication indices and non-process variables
of the formula $\varphi$ to their values, and a sequence of events $\evseq$,
we define $E, \venv, \evseq \vdash \varphi$, meaning that $E, \venv, \evseq$ satisfy $\varphi$, as follows:
\bbnote{Will be useful when defining facts that hold at input processes}%
\begin{itemize}
\item $E,\venv,\evseq \vdash M$ if and only if $E, \venv, M \evalterm \true$.

\item $E,\venv,\evseq \vdash \defined(M)$ if and only if $E, \venv, M \evalterm a$ for some $a$.

\item $E,\venv,\evseq \vdash \fevent{e(\tup{M})}$ if and only if $E, \venv, \tup{M} \evalterm \tup{a}$ and $(\pp, \tup{a'}):e(\tup{a}) \in \evseq$ for some $\pp$ and $\tup{a'}$.

\item $E,\venv,\evseq \vdash \fevent{e(\tup{M})}@M_0$  if and only if $E, \venv, \tup{M} \evalterm \tup{a}$,
  $E, \venv, M_0 \evalterm a_0$, and $\evseq(a_0) = (\pp, \tup{a'}):e(\tup{a})$ for some $\pp$ and $\tup{a'}$.

\item $E,\venv,\evseq \vdash M_1:\fevent{e(\tup{M})}$ if and only if $E, \venv, \tup{M} \evalterm \tup{a}$, $E, \venv, M_1 \evalterm (\pp,\tup{a'})$ and $(\pp, \tup{a'}):e(\tup{a}) \in \evseq$.

\item $E,\venv,\evseq \vdash M_1:\fevent{e(\tup{M})}@M_0$  if and only if $E, \venv, \tup{M} \evalterm \tup{a}$,
  $E, \venv, M_0 \evalterm a_0$, $E, \venv, M_1 \evalterm (\pp,\tup{a'})$, and $\evseq(a_0) = (\pp, \tup{a'}):e(\tup{a})$.

\end{itemize}
Logical connectives are defined as usual. When $\varphi$ does not contain events,
$\evseq$ can be omitted, writing $E, \venv \vdash \varphi$.

Let $\trace$ be a trace of $Q_0$.
Let $\conf = E, \sigma, N, \tblcts, \evseq$
or $\conf = E, (\sigma, P)\restconfig$ be a configuration that occurs in the derivation of $\trace$.
We define $\trace\before\conf, \venv \vdash \varphi$, meaning that the prefix of $\trace$ until $\conf$ satisfies the formula $\varphi$ with environment $\venv$ (giving values of non-process variables of $\varphi$) as follows:
\begin{itemize}
\item $\trace\before\conf,\venv \vdash F$ if and only if
  $E_{\trace \before \conf}, \sigma_{\conf} \cup \venv, \evseq_{\conf} \vdash F$,
  when $F$ is a term $M$, a defined fact $\defined(M)$, or an event
  $\fevent{e(\tup{M})}$, $\fevent{e(\tup{M})}@M_0$, $M_1:\fevent{e(\tup{M})}$,
  or $M_1:\fevent{e(\tup{M})}@M_0$.

\item $\trace\before\conf,\venv \vdash \ppf(\sset_1, \tup{M}_1) \before \dots \before \ppf(\sset_m, \tup{M}_m)$ if and only if, for all $j \in \{1, \dots, m\}$,
there exists $\conf_j$ at program point $\pp_j \in \sset_j$ in $\trace$ such that $E_{\trace \before \conf}, \sigma_{\conf} \cup \venv, \tup{M}_j \evalterm \image(\sigma_{\conf_j})$ and $\conf_1 \beforetr{\trace} \dots \beforetr{\trace} \conf_m \beforetr{\trace} \conf$.

The fact $\ppf(\pp, \tup{M})$ is actually a particular case of $\ppf(\sset_1, \tup{M}_1) \ab \before \ab \dots \ab \before \ab \ppf(\sset_m, \tup{M}_m)$ with $m = 1$ and $\sset_1 = \{ \pp \}$. 
By specializing the definition above, we have $\trace\before\conf,\venv \vdash \ppf(\pp, \tup{M})$ if and only if there is a configuration $\conf'$ at program point $\pp$ in $\trace$ such that $\conf' \beforetr{\trace} \conf$ and $E_{\trace \before \conf}, \sigma_{\conf} \cup \venv, \tup{M} \evalterm \image(\sigma_{\conf'})$.%
\bbnote{Cannot be defined on a configuration only, because it needs the initial process $Q_0$ to determine that a configuration is at a program point $\pp$. Or could I replace the condition that the term or process is a subterm/subprocess of $Q_0$ by the fact that it contains no value or abort event value or $\kw{abort}$? No, that does not work: one can reduce for instance $C[\Res{x}{T}; N] \red{}{} C[N]$; $C[N]$ contains no value or abort event value, but it is not a subterm/subprocess of $Q_0$.}%

\item $\trace\before\conf,\venv \vdash \lppf(\pp, \tup{M})$ if and only if there is a configuration $\conf'$ at program point $\pp$ in $\trace$ such that $\conf' \beforetr{\trace} \conf$, $E_{\trace \before \conf'} = E_{\trace \before \conf}$, $\sigma_{\conf'} = \sigma_{\conf}$, and $E_{\trace \before \conf}, \sigma_{\conf} \cup \venv, \tup{M} \evalterm \image(\sigma_{\conf'})$.

\end{itemize}
Logical connectives are defined as usual. Most facts are
evaluated in the environment $E_{\trace \before \conf}$ and the mapping sequence $\sigma_{\conf}$.  Events
are evaluated using the sequence of events $\evseq_{\conf}$ in $\conf$, 
but correspond to an execution of the event at some point before $\conf$ in
the trace.

We define that a trace $\trace$ satisfies a logical formula $\varphi$
with environment $\venv$ (giving values of non-process variables of
$\varphi$), denoted $\trace,\venv \vdash \varphi$ as
$\trace \before \conf,\venv \vdash \varphi$, where $\trace$ ends with
$\conf$. Along the same line, we define $E_{\trace} = E_{\trace \before \conf}$
and $\evseq_{\trace} = \evseq_{\conf}$
where $\trace$ ends with $\conf$.

When the formula $\varphi$ does not contain free non-process variables,
we may write $\trace \vdash \varphi$ instead of $\trace,\venv \vdash \varphi$
since the environment $\venv$ is useless.
When $\fset$ is a set of formulas (in particular, of facts),
we write $\bigwedge \fset$ for $\bigwedge_{F \in \fset} F$
and $\bigvee \fset$ for $\bigvee_{F \in \fset} F$.
We also write $\trace,\venv \vdash \fset$ when for all $\varphi \in \fset$,
$\trace,\venv \vdash \varphi$. This is equivalent to
$\trace,\venv \vdash \bigwedge \fset$.
We use similar notations for prefixes $\trace \before \conf$ 
instead of traces $\trace$.

Additionally, we define the following facts:
\begin{itemize}
\item $\sem{\elsefind((i_1 \leq n_1, \ldots, i_m \leq n_m), (M_1, \ldots, M_l), M)} = \forall i_1 \in [1, n_1], \ldots, \forall i_m \in [1, n_m],\ab
\neg (\defined(M_1) \wedge \dots \wedge \defined(M_l) \wedge M)$.
\item $\sem{\elselet(\tup{x}:\tup{T},N,M)} = \forall \tup{x}\in \tup{T}, N \neq M$.
\end{itemize}

\subsection{User-defined Rewrite Rules}\label{sec:userdefinedrewriterules}

The user can give two kinds of information:
\begin{itemize}

\item claims of the form $\forall x_1:T_1, \ldots, \forall x_m:T_m, M$
which mean that for all environments $E$, if for all $j \leq m$, 
$E(x_j) \in T_j$, then $E, M \evalterm \true$.

Such claims must be well-typed, that is, $\{ x_1 \mapsto T_1,\ab
\ldots, \ab x_m \mapsto T_m \} \vdash M : \bool$.

They are translated into rewrite rules as follows:
\begin{itemize}

\item If $M$ is of the form $M_1 = M_2$ and $\vardef(M_2) \subseteq \vardef(M_1)$, 
we generate the rewrite rule
$\forall x_1:T_1, \ab \ldots, \ab \forall x_m:T_m, M_1 \rewrite M_2$.

\item If $M$ is of the form $M_1 \neq M_2$, we generate the rewrite rules
$\forall x_1:T_1, \ldots, \forall x_m:T_m, (M_1 = M_2) \rewrite \false$,
$\forall x_1:T_1, \ldots, \forall x_m:T_m, (M_1 \neq M_2) \rewrite \true$.
(Such rules are used for instance to express that different constants
are different.)
%Also for tuple functions...

\item Otherwise, we generate the rewrite rule
$\forall x_1:T_1, \ab \ldots, \ab \forall x_m:T_m, \ab M \rewrite \true$.

\end{itemize}
The term $M$ reduces into $M'$ by the rewrite rule 
$\forall x_1:T_1, \ab \ldots, \ab \forall x_m:T_m, \ab M_1
\rewrite M_2$ if and only if $M = C[\sigma M_1]$, $M' = C[\sigma
M_2]$, where $C$ is a term context and $\sigma$ is a substitution that
maps $x_j$ to any term of type $T_j$ for all $j \leq m$.

\item claims of the form $\Res{y_1}{T'_1}, \ab \ldots, \ab \Res{y_l}{T'_l}, \ab 
\forall x_1:T_1, \ab \ldots, \ab \forall x_m:T_m, \ab M_1 \approx_p M_2$
with $\vardef(M_2) \subseteq \vardef(M_1)$.
Informally, these claims mean that $M_1$ and $M_2$ evaluate to 
the same bitstring except in cases of probability at most $p$, 
provided that $y_1, \ldots, y_l$ are chosen randomly with
uniform probability and independently among $T'_1, \ldots, T'_l$ respectively,
and that $x_1, \ldots, x_m$ are of type $T_1, \ldots, T_m$.
($x_1, \ldots, x_m$ may depend on $y_1, \ldots, y_l$.)
Formally, these claims are defined as:
\[\begin{split}
\Pr[& E(y_1) \randomchoice T'_1; \ldots
E(y_l) \randomchoice T'_l;\\[-1mm]
&(E(x_1), \ldots, E(x_m)) \leftarrow {\Adv}(E(y_1), \ldots, E(y_l));\\
&E, M_1 \evalterm a; E, M_2 \evalterm a': a \neq a'
] \leq p(\Adv)
\end{split}\]
where $\Adv$ is a probabilistic Turing machine.

The above claim must be well-typed, that is, $\{ x_1 \mapsto T_1, \ab \ldots, \ab x_m
\mapsto T_m, \ab y_1 \mapsto T'_1, \ab \ldots, \ab y_l \mapsto T'_l \} \vdash M_1
= M_2$.

This claim is translated into the rewrite rule
$\Res{y_1}{T'_1}, \ab \ldots, \ab \Res{y_l}{T'_l}, \ab 
\forall x_1:T_1, \ab \ldots, \ab \forall x_m:T_m, \ab M_1 \rewrite M_2$.

\end{itemize}

The prover has built-in rewrite rules for defining boolean functions:
{\allowdisplaybreaks\begin{align*}
&\fnot \true \rewrite \false\qquad \fnot \false \rewrite \true\qquad 
\forall x:\bool, \fnot (\fnot x) \rewrite x\\
&\forall x:T, \forall y:T, \fnot(x = y) \rewrite x \neq y\\
&\forall x:T, \forall y:T, \fnot(x \neq y) \rewrite x = y\\
&\forall x:T, x = x \rewrite \true
\qquad
\forall x:T, x \neq x \rewrite \false\\
&\forall x:\bool, \forall y:\bool, \fnot(x \fand y) \rewrite (\fnot x) \for (\fnot y)\\
&\forall x:\bool, \forall y:\bool, \fnot(x \for y) \rewrite (\fnot x) \fand (\fnot y)\\
&\forall x:\bool, x \fand \true \rewrite x \qquad 
\forall x:\bool, x \fand \false \rewrite \false\\
&\forall x:\bool, x \for \true \rewrite \true\qquad
\forall x:\bool, x \for \false \rewrite x\\
&\forall x:T, \forall y:T, \iffun(\true,x,y) \rewrite x\qquad
\forall x:T, \forall y:T, \iffun(\false,x,y) \rewrite y\\
&\forall x:\bool, \forall y:T, \iffun(x,y,y) \rewrite y\\
&\forall x_1:T_1, \dots, \forall x_m:T_m, \forall x:\bool, \forall y: T_k, \forall z:T_k, \\
&\qquad f(x_1, \dots,x_{k-1}, \iffun(x,y,z), x_{k+1}, \dots, x_m) \rewrite\\
&\qquad\qquad
\iffun(x,f(x_1, \dots,x_{k-1}, y, x_{k+1}, \dots, x_m),
f(x_1, \dots,x_{k-1}, z, x_{k+1}, \dots, x_m))\\
&\qquad \text{when $f: T_1 \times \ldots \times T_m \rightarrow T$ has option \rn{autoSwapIf}}
\end{align*}}%
% as well as the symmetrics of the last four rules, swapping
% the arguments of $\fand$ and $\for$.

The prover also has support for commutative function symbols,
that is, binary function symbols $f: T \times T \rightarrow T'$ 
such that for all $x, y \in T$, 
$f(x,y) = f(y,x)$.
For such symbols, all equality and matching tests are performed
modulo commutativity. The functions $\fand$, $\for$, $=$, and $\neq$
are commutative. 
So, for instance, the rewrite rules above may also be used
to rewrite $\true \fand M$ into $M$, $\false \fand M$ into $\false$, 
$\true \for M$ into $\true$, and $\false \for M$ into $M$.
Used-defined functions may also be declared
commutative; $\xor$ is an example of such a commutative function.

\begin{example}
For example, considering MAC and encryption schemes as in 
Definitions~\ref{def:macsec} and~\ref{def:encsec} respectively, we
have:
{\allowdisplaybreaks\begin{align}
\begin{split}
&\forall k: T_{mk}, \forall m:\bitstring, \\*
&\quad \mverify(m, k, \mac(m, k)) = \true
\end{split}\tag{$\mac$}\label{eq:mac}\\
\begin{split}
&\forall m: \bitstring; \forall k: T_k, \forall r:T_r,\\*
&\quad \dec(\enc(m,k,r),k) = \injbot(m)
\end{split}\tag{$\enc$}\label{eq:enc}
\end{align}}%
We express the poly-injectivity of the function $\ktob$ of Example~\ref{exa:running}
by
\begin{equation}
\begin{split}
&\forall x:T_k, \forall y:T_k,(\ktob(x) = \ktob(y)) = (x = y)\quad\\
&\forall x:T_k,\ktob^{-1}(\ktob(x)) = x
\end{split}\tag{$\ktob$}\label{eq:k2b}
\end{equation}
where $\ktob^{-1}$ is a function symbol that denotes the inverse of $\ktob$.
We have similar formulas for $\injbot$.
\end{example}

\subsection{Collecting True Facts from a Game}

CryptoVerif collects a set of facts $\fset_\pp$ that hold at each program point
$\pp$ in the current game $Q_0$.
Additionally, CryptoVerif also collects facts $\futfset_\pp$ (future facts at $\pp$), which hold at the end of the block of code that contains $\pp$ and ends with an output or an $\kw{event\_abort}$ instruction that aborts the end. For instance, $\futfset_\pp$ may contain equalities that come from assignments performed after $\pp$ in the same block of code. (However, the facts in $\futfset_\pp$ m
ay not hold in case a $\FIND\unique{e}$ aborts because several choices make theconditions of that $\FIND$ succeed.)
These sets of facts may contain facts $M$, $\defined(M)$, $M_1:\fevent{e(\tup{M})}$, and $\ppf(\pp, \tup{M})$.
In these sets of facts, all terms $M$ must be simple.

Previous versions of the algorithm that collects facts were presented in~\cite[Appendix~C.2]{Blanchet07c} and~\cite[Appendix~B.2]{Blanchet07}.
The current algorithm is an extension that relies on the same principles.
The facts $\ppf(\pp, \tup{M})$ is new.
In particular, we have $\ppf(\pp, \replidx{\pp}) \in \facts_{\pp}$.
The algorithm that collects facts satisfies the following properties.

\begin{lemma}\label{lem:factsP}
Let $C$ be an evaluation context acceptable for $Q_0$, $\trace$ be a trace of $C[Q_0]$, $\pp$ be a program point in $Q_0$, and $\facts_{\pp}$ be computed in $Q_0$. 
  If a configuration $\conf$ is at program point $\pp$ in $\trace$, then $\trace \before \conf \vdash \facts_{\pp}$. 
\end{lemma}
\bbnote{In the implementation, keep the information about program points provided by the known\_history (reorganized as (program\_point list * term list) list), and use that to test compatibility? That would be more precise and more efficient than using the defined variables.}%
%Use $Q_0$ rather than $G$ when a single game is involved, like here.

Additionally, there is a more precise version of $\facts_{\pp}$ that distinguishes cases depending on the program points at which the various variables are defined, generating several $\facts_{\pp,\case}$ for the various cases $\case$. For this version, we have: 

\begin{lemma}\label{lem:factsPcases}
  Let $C$ be an evaluation context acceptable for $Q_0$, $\trace$ be a trace of $C[Q_0]$, $\pp$ be a program point in $Q_0$, and $\facts_{\pp,\case}$ be computed in $Q_0$. 
  If a configuration $\conf$ is at program point $\pp$ in $\trace$, then there exists $\case$ such that $\trace \before \conf \vdash \facts_{\pp,\case}$. 
\end{lemma}

$\facts_{\pp}$ can be seen as a particular case of $\facts_{\pp,\case}$ by considering a single case $\case$.

\begin{corollary}\label{cor:factsP}
  Let $C$ be an evaluation context acceptable for $Q_0$, $\trace$ be a trace of $C[Q_0]$, $\pp$ be a program point in $Q_0$, and $\facts_{\pp}$ (resp. $\facts_{\pp,\case}$) be computed in $Q_0$. Let $\conf$ be a configuration at program point $\pp$ in $\trace$. Let $\theta$ be a renaming of $\replidx{\pp}$ to fresh replication indices and $\venv = \{ \theta\replidx{\pp} \mapsto \sigma_{\conf}\replidx{\pp} \}$. Let $\conf'$ be a term or output process configuration in $\trace$ such that $\conf \beforetr{\trace} \conf'$.
  
  We have $\trace \before \conf', \venv \vdash \theta\facts_{\pp}$ and there exists $\case$ such that $\trace \before \conf', \venv \vdash \theta\facts_{\pp,\case}$.

  In particular, $\trace, \venv \vdash \theta\facts_{\pp}$ and there exists $\case$ such that $\trace, \venv \vdash \theta\facts_{\pp,\case}$.
\end{corollary}
\begin{proof}
  By Lemma~\ref{lem:factsP}, $\trace \before \conf \vdash \facts_{\pp}$.
  By Lemma~\ref{lem:factsPcases}, there exists $\case$ such that $\trace \before \conf \vdash \facts_{\pp,\case}$.
  Let $\fset = \facts_{\pp}$ (resp. $\fset = \facts_{\pp,\case}$) such that $\trace \before \conf \vdash \fset$.
  By definition of $\venv$, we have $\trace \before \conf, \venv \vdash \theta\fset$.
  Since $\conf \beforetr{\trace} \conf'$, the environment $E_{\trace \before \conf'}$ is an extension of $E_{\trace \before \conf}$, so the terms and $\defined$ facts in $\theta\fset$ are preserved when considering $\trace \before \conf'$ instead of $\trace \before \conf$. (They do not use $\sigma_{\conf}$, resp. $\sigma_{\conf'}$, by the renaming $\theta$.)
  Moreover, by Lemma~\ref{lem:sem-ext}, $\evseq_{\conf'}$ is an extension of $\evseq_{\conf}$, so the events are also preserved.
  By definition of $\trace\before\conf,\venv \vdash \ppf(\sset_1, \tup{M}_1) \before \dots \before \ppf(\sset_m, \tup{M}_m)$, the sequences of program points  $\ppf(\sset_1, \tup{M}_1) \before \dots \before \ppf(\sset_m, \tup{M}_m)$ are also preserved.
  Since $\fset$ is a set of facts containing only terms, $\defined$ facts, events, and sequences of program points,
  we conclude that $\trace \before \conf', \venv \vdash \theta\fset$.

  The last point is obtained by choosing $\conf'$ to be the last configuration of $\trace$.
  \proofcomplete
\end{proof}

\begin{lemma}\label{lem:factsPfut}
  Let $C$ be an evaluation context acceptable for $Q_0$, $\trace = \initconfig(C[Q_0]) \red{p}{\ix} \ldots \red{p'}{\ix'}E, (\sigma, P)\restconfig$ be a trace that does not execute any non-unique event of $Q_0$ with $P = \coutput{c[\tup{a}]}{a};Q$ for some $c$, $\tup{a}$, $a$, and $Q$ or $P = \kw{abort}$, $\pp$ be a program point in $Q_0$, and $\futfset_{\pp}$ be computed in $Q_0$. 
  If the configuration $\conf$ is at program point $\pp$ in $\trace$ and no executed process in the configurations between the configuration at the end of reduction step that contains $\conf$ (included) and $E, (\sigma, P)\restconfig$ (excluded) is of the form $\coutput{c[\tup{a}]}{a};Q$ for some $\tup{a}$, $a$, and $Q$ (when $\conf$ is a process configuration with process $\coutput{c[\tup{a}]}{a};Q$ for some $c$, $\tup{a}$, $a$, and $Q$, we have $\conf = E, (\sigma, P)\restconfig$), then $\trace \vdash \futfset_{\pp}$. 
  %what about abortion of find[unique]? it seems ok not to count it:
  %future true facts are used for collection of true facts, and in that case we could in fact consider future facts obtained until an output (because future facts are used when we know we reach another block of code); that would be stronger
  %they are also used for the proof of correspondences, and in that case we can ignore traces that execute a non-unique event (they are not counted in $Advt_G(C, \neg\corresp)$)
\end{lemma}
  \bb{TODO in the implementation, the collection of future facts is not very precise when terms are not simple (no future facts for terms, future facts of an output do not take into account facts obtained by evaluating the channel and message of the output, ...)}%

\subsection{Equational Prover} \label{sec:equationalprover}
%(Knuth-Bendix? see dev/abadi/weaksecr/rewcourse/v8.ps.gz)

In order to reason on facts, CryptoVerif uses an equational prover: from a set
of facts $\fset$, this equational prover tries to derive a contradiction by rewriting terms, using an algorithm inspired by Knuth-Bendix completion. It also eliminates collisions between independent random values, thus the contradiction is
obtained up to the probability of the eliminated collisions, that is, the probability that $\fset$ holds is bounded by the probability of these collisions.
When this algorithm succeeds, we say that ``$\fset$ yields a contradiction in game $Q_0$'', and CryptoVerif computes the probbaility of the eliminated collisions. (We may omit the current game $Q_0$ when it is clear from the context.)
Previous versions of this algorithm were presented in~\cite[Appendix~C.5]{Blanchet07c} and~\cite[Appendix~B.3]{Blanchet07}.
Those versions did not evaluate the probability because they considered asymptotic security:
they showed that the probability was negligible in the security parameter.
Here, we use exact security: we compute the value of the probabilities, so the
soundness of this algorithm
can be expressed by the following lemma, adapted from~\cite[Proposition~7]{Blanchet07}.

\begin{lemma}\label{lem:yields_contrad}
If for all $j \in J$, $\fset_j$ yields a contradiction in a game $Q_0$, then
CryptoVerif returns a probability $p$ such that for all evaluation
contexts $C$ acceptable for $Q_0$ with any public variables,
$\Prss{C[Q_0]}{\bigvee_{j \in J}\exists \tup{x}_j \in \tup{T}_j, \bigwedge \fset_j} \leq p(C)$,
where $\tup{x}_j$ are the replication indices and non-process variables that occur in $\fset_j$
and $\tup{T}_j$ are their types.
\end{lemma}

In particular, the lemma states that, when several sets of facts
$\fset_j$ yield a contradiction in the same game, CryptoVerif counts
only once in the probability $p$ the collisions that are eliminated in
proofs that $\fset_j$ yields a contradiction for several $j$.

More generally, let us consider an algorithm $\algyc$ built from
the following grammar:
\begin{defn}
  \categ{\algyc}{algorithm}\\
  \entry{\fset \text{ yields a contradiction}}{equational proof}\\
  \entry{\algyc_1 \wedge \algyc_2}{conjunction}\\
  \entry{\algyc_1 \vee \algyc_2}{disjunction}\\
  \entry{\mathf}{mathematical formula}\\
  \entry{\algtest{\mathf}{\algyc_1}{\algyc_2}}{test}
\end{defn}
The mathematical formulas $\mathf$ in such algorithms must not depend on the
executed trace. (They may depend on the syntax of the game $Q_0$ or on the set of public variables $V$, for instance.)

We translate such algorithms into logical formulas on traces:
{\allowdisplaybreaks\begin{align*}
  &\yctran{\fset \text{ yields a contradiction}} = \neg \exists \tup{x} \in \tup{T}, \bigwedge \fset \quad \text{where $\tup{x}$ are the replication indices and}\\*
  &\quad\text{non-process variables that occur in $\fset$ and $\tup{T}$ are their types.}\\
  &\yctran{\algyc_1 \wedge \algyc_2} = \yctran{\algyc_1} \wedge \yctran{\algyc_2}\\
  &\yctran{\algyc_1 \vee \algyc_2} = \begin{cases}
    \yctran{\algyc_1}&\text{if $\algyc_1$}\\
    \yctran{\algyc_2}&\text{otherwise}
  \end{cases}\\
  &\yctran{\mathf} = \mathf\\
  &\yctran{\algtest{\mathf}{\algyc_1}{\algyc_2}} = \algtest{\mathf}{\yctran{\algyc_1}}{\yctran{\algyc_2}}
\end{align*}}%
Intuitively, when algorithm $\algyc$ returns true, CryptoVerif shows
that the formula $\yctran{\algyc}$ holds for most traces. It bounds
the probability of the traces for which this formula does not hold, as
shown by the following lemma.

\begin{lemma}\label{lem:yields_contrad_algo}
If algorithm $\algyc$ returns true in a game $Q_0$, then
CryptoVerif returns a probability $p$ such that for all evaluation
contexts $C$ acceptable for $Q_0$ with any public variables,
$\Prss{C[Q_0]}{\neg \yctran{\algyc}} \leq p(C)$.
\end{lemma}
\begin{proof}
  We show by induction on the definition of $\algyc$ that, if $\algyc$ returns true and $\trace \vdash \neg \yctran{\algyc}$, then there exists $\fset$ such that ``$\fset$ yields a contradiction'' has been called in the evaluation of $\algyc$ and returned true, and $\trace \vdash \exists \tup{x} \in \tup{T}, \bigwedge \fset$ where $\tup{x}$ are the replication indices and non-process variables that occur in $\fset$ and $\tup{T}$ are their types.
  \begin{itemize}
  \item Case $\algyc = (\fset \text{ yields a contradiction})$: obvious.
  \item Case $\algyc = \algyc_1 \wedge \algyc_2$: Since $\algyc$ returns true, $\algyc_1$ and $\algyc_2$ both return true. Since $\trace \vdash \neg (\yctran{\algyc_1} \wedge \yctran{\algyc_2})$, we have either $\trace \vdash \neg \yctran{\algyc_1}$ or $\trace \vdash \neg\yctran{\algyc_2}$. In the first case, by induction hypothesis on $\algyc_1$, there exists $\fset$ such that ``$\fset$ yields a contradiction'' has been called in the evaluation of $\algyc_1$ and returned true, and $\trace \vdash \exists \tup{x} \in \tup{T}, \bigwedge \fset$ where $\tup{x}$ are the replication indices and non-process variables that occur in $\fset$ and $\tup{T}$ are their types. Moreover, ``$\fset$ yields a contradiction'' has been called in the evaluation of $\algyc$. The second case is symmetric.
  \item Case $\algyc = \algyc_1 \vee \algyc_2$: If $\algyc_1$ returns true, then $\yctran{\algyc} = \yctran{\algyc_1}$, so $\trace \vdash \neg\yctran{\algyc_1}$. We conclude by induction hypothesis on $\algyc_1$, as above. If $\algyc_1$ returns false, then $\algyc_2$ returns true, since $\algyc$ returns true. Hence $\yctran{\algyc} = \yctran{\algyc_2}$, so $\trace \vdash \neg\yctran{\algyc_2}$. We conclude by induction hypothesis on $\algyc_2$.
  \item Case $\algyc = \mathf$: since $\mathf$ evaluates to true, there is no trace $\trace$ such that $\trace \vdash \neg\mathf$, so the property holds trivially.
  \item Case $\algyc = \algtest{\mathf}{\algyc_1}{\algyc_2}$: if $\mathf$ evaluates to true, then we conclude by induction hypothesis on $\algyc_1$. Indeed, since $\algyc$ returns true, $\algyc_1$ returns true. Since $\trace \vdash \neg \yctran{\algyc}$, we have $\trace \vdash \neg \yctran{\algyc_1}$. By induction hypothesis, there exists $\fset$ such that ``$\fset$ yields a contradiction'' has been called in the evaluation of $\algyc_1$ and returned true, and $\trace \vdash \exists \tup{x} \in \tup{T}, \bigwedge \fset$ where $\tup{x}$ are the replication indices and non-process variables that occur in $\fset$ and $\tup{T}$ are their types. Then ``$\fset$ yields a contradiction'' has also been called in the evaluation of $\algyc$. Similarly, if $\mathf$ evaluates to false, then we conclude by induction hypothesis on $\algyc_2$.
  \end{itemize}
  We conclude by Lemma~\ref{lem:yields_contrad}.
  \proofcomplete
\end{proof}

\section{\rn{success}: Criteria for Proving Security Properties}\label{sec:success}

The command \rn{success} tries to prove the active queries, as explained below.
We consider a process $Q_0$ that satisfies Properties~\ref{prop:notables} and~\ref{prop:autosarename},
and prove secrecy and correspondence properties for $Q_0$.

\subsection{Secrecy}

\begin{figure}[t]
  \begin{align*}
&\noleak(x[\tup{M}'], \sri, \fset) = (\fset\text{ yields a contradiction}) \vee ((x \notin V) \wedge \bigwedge_{\pptag x[\tup{M}] \text{ in }Q_0}\\
&\qquad \text{-- if $\pptag x[\tup{M}]$ is in $M$ in an assignment $\bassign{y[\tup{i}]}{M}$ in $Q_0$, $M$ is built from replication}\\[-1mm]
  &\qquad \text{indices, variables, function applications, and conditionals, and
    the current call is not}\\[-1mm]
  &\qquad \text{inside a call to $\noleak(y[\_], \_, \_)$, then}\\
&\qquad\quad \noleak(y[\theta\tup{i}], \sri \cup \{ \theta\tup{i} \}, \fset \cup \theta \fset_{\pp} \cup \{ \theta\tup{M} = \tup{M}' \})\\[-1mm]
&\qquad\quad \text{where $\theta$ is a renaming of $\tup{i}$ to fresh replication indices}\\
&\qquad \text{-- if $\pptag x[\tup{M}]$ is in $\kevent{e(M_1, \dots, M_{k-1}, C[\pptag x[\tup{M}]], M_{k+1}, \dots, M_m)}$ in $Q_0$ for $C$ defined in}\\[-1mm]
&\qquad \text{Figure~\ref{fig:eventcontexts}, then true}\\
&\qquad \text{-- otherwise}, \fset \cup \theta \fset_{\pp} \cup \{ \theta\tup{M} = \tup{M}' \}\text{ yields a contradiction}\\[-1mm]
&\qquad\quad \text{where $\theta$ is a renaming of $\replidx{\pp}$ to fresh replication indices})
\end{align*}
  \caption{Function $\noleak$}\label{fig:noleak}
\end{figure}

\begin{figure}[t]
\begin{align*}
C ::= {} & [\,]\\
&y[M_1, \dots, M_{k-1}, C, M_{k+1}, \dots, M_m]\\
&f(M_1, \dots, M_{k-1}, C, M_{k+1}, \dots, M_m)\\
&\Res{y[\tup{i}]}{T}; C\\
&\assign{y[\tup{i}]}{M}{C}\\
&\bguard{M}{C}{N'}\\
&\bguard{M}{N}{C}\\
&\FIND\uniqueopt \ (\mathop\bigoplus\nolimits_{j=1,\dots,m; j \neq k} \tup{\vf_j}[\tup{i}] = \tup{i_j} \leq \tup{n_j}\ \SUCHTHAT\ \defined(\tup{M_j}) \wedge M_j \THEN N_j)\\
&\qquad \oplus \tup{\vf_k}[\tup{i}] = \tup{i_k} \leq \tup{n_k}\ \SUCHTHAT\ \defined(\tup{M_k}) \wedge M_k \THEN C \ \ELSE N\\
&\FIND\uniqueopt \ (\mathop\bigoplus\nolimits_{j=1}^m \tup{\vf_j}[\tup{i}] = \tup{i_j} \leq \tup{n_j}\ \SUCHTHAT\ \defined(\tup{M_j}) \wedge M_j \THEN N_j)\ \ELSE C\\
&\kevent{e(\tup{M})}; C
\end{align*}
\caption{Event contexts}\label{fig:eventcontexts}
\end{figure}

Let us now define syntactic criteria that allow us to prove secrecy
properties of protocols. 
We first define the function $\noleak$ in Figure~\ref{fig:noleak} and explain it below.
This function implicitly depends on the current game $Q_0$ and the public variables $V$.
The function call $\noleak(x[\tup{M}'], \sri, \fset)$ shows that $x[\tup{M}']$
does not leak to the adversary, assuming $\fset$ holds. The set $\sri$ contains all replication indices that appear in $\fset$.
If $\fset$ yields a contradiction, $\noleak$ is true, since it shows the absence of leak \emph{assuming $\fset$ holds}. Otherwise, $x[\tup{M}']$ may leak either because $x \in V$ so $x$ is a public variable, or because of an occurrence of a term $x[\tup{M}]$ in the game that reads $x[\tup{M}']$ (so $\tup{M} = \tup{M'}$ holds) and such that the result of $x[\tup{M}]$ leaks.
\begin{itemize}
\item In case $x[\tup{M}]$ occurs in the term $M$ in an assignment $\bassign{y[\tup{i}]}{M}$, the function $\noleak$ recursively tries to prove that $y[\tup{i}]$ does not leak, when $y[\tup{i}]$ may use $x[\tup{M}']$, that is, when $\tup{M}' = \tup{M}$. The fact $\tup{M}' = \tup{M}$ and the facts $\fset_\pp$ that hold at the program point $\pp$ of $x[\tup{M}]$ are added to the known facts $\fset$ in the recursive call. Indeed, these facts are known to hold in this case. The replication indices are renamed to fresh indices in order to avoid using the same index variable for indices that can actually take different values.
\item In case $x[\tup{M}]$ occurs in the arguments of an event, the arguments of the event do not leak to the adversary, so this occurrence of $x[\tup{M}]$ does not make $x[\tup{M}']$ leak.
\item In all other cases, we consider that the result of $x[\tup{M}]$ may leak, so, in order to prove that $x[\tup{M}']$ does not leak, we show that the occurrence of $x[\tup{M}]$ at $\pp$ cannot read $x[\tup{M}']$, by showing that $\tup{M} = \tup{M'}$, $\fset_\pp$, and $\fset$ together yield a contradiction.
\end{itemize}

\begin{definition}[$\pp$ follows a definition of $x$]
We say that $\pp$ \emph{follows a definition of} $x$ when 
$\Res{x[\tup{i}]}{T};\pptag \dots$, 
$\assign{x[\tup{i}]}{M}{\pptag \dots}$,
$\FIND\uniqueopt \ (\mathop\bigoplus\nolimits_{j=1}^m \tup{\vf_j}[\tup{i}] = \tup{i_j} \leq \tup{n_j}\ \SUCHTHAT\ 
\defined (\tup{M_j}) \fand M_j \THEN {}^{\pp_j}\dots)\ \ELSE \dots$ 
with $\pp_j = \pp$ and $x$ in $\tup{\vf_j}$ for some $j \leq m$,
or $\cinput{c[\tup{M}]}{x[\tup{i}]:T}; \pptag P$ occurs in $Q_0$.
\end{definition}

We do not mention $\GET$ in the previous definition, because it is excluded by Property~\ref{prop:notables}.
For each $\pp$ that follows a definition of $x$ in $Q_0$, we define
$\defrestr_{\pp}(x)$ as follows:
\begin{equation*}
\defrestr_{\pp}(x) =
\begin{cases}
x[\tup{i}]&\text{if $\Res{x[\tup{i}]}{T}; \pptag\dots$ occurs in $Q_0$}\\
y[\tup{M}]&\text{if $\assign{x[\tup{i}]:T}{y[\tup{M}]}{\pptag \dots}$ occurs in $Q_0$ and}\\
&\text{$y$ is defined only by random choices in $Q_0$}
\end{cases}
\end{equation*}
In all other cases, $\defrestr_{\pp}(x)$ is not defined.
The variable $\defrestr_{\pp}(x)$ is the random variable that defines $x$ just before program point $\pp$.
When $x$ itself is chosen randomly at that point, $\defrestr_{\pp}(x)$ is simply $x[\tup{i}]$, where $\tup{i}$ are the current replication indices.
When $x$ is defined by an assignment of a variable $y[\tup{M}]$ that is random, $\defrestr_{\pp}(x)$ is that variable.
Otherwise, we give up and do not define $\defrestr_{\pp}(x)$.
\begin{align*}
&\prove{\secrone(x)}(\pp) = \defrestr_{\pp}(x)\text{ is defined and}\\[-1mm]
&\qquad \noleak(\theta \defrestr_{\pp}(x), \{ \theta\replidx{\pp} \}, \theta\fset_{\pp})\\[-1mm]
  &\qquad \text{where $\theta$ is a renaming of $\replidx{\pp}$ to fresh replication indices}\\
  &\prove{\secrone(x)}(\sset) =  \bigwedge_{\pp \in \sset} \prove{\secrone(x)}(\pp)
\end{align*}
The function call $\prove{\secrone(x)}(\pp)$ proves one-session secrecy for the definition of $x$ just before program point $\pp$. It considers only the cases in which $x$ is defined either by a random choice or by an assignment from a random choice. In other cases, the proof fails. (These other cases can typically be handled by first removing assignments as needed.)
Intuitively, $\prove{\secrone(x)}(\pp)$ guarantees that, when $x[\tup{i}]$ is defined just before program point $\pp$, the random variable $\defrestr_{\pp}(x)$ that defines $x[\tup{i}]$ does not leak, knowing that the facts $\fset_\pp$ hold.
Only events and variables $y[\tup{i}']$ that do not leak depend on the random choice that defines $x[\tup{i}]$; the sent messages and the control flow of the process are independent of $x[\tup{i}]$, so the adversary obtains no information on $x[\tup{i}]$.
That guarantees the one-session secrecy of $x[\tup{i}]$ when it is defined just before $\pp$. This is verified for all program points in $\sset$ by $\prove{\secrone(x)}(\sset)$.
When $x$ is defined by assignment of $z[\tup{M}]$, this proof of one-session secrecy allows some array cells of $z$ to leak, provided the array cells $z[\tup{M}]$ used to define $x$ do not leak.

\bbnote{In order to prove secrecy, we also define
\[\fset_{x,\pp_1,\pp_2} = \begin{cases}
  \{ \false \} &\text{if $z_1 \neq z_2$}\\
  \theta_1\fset_{\pp_1} \cup \theta_2\fset_{\pp_2} \cup \{ \theta_1 \tup{M_1} = \theta_2 \tup{M_2},\ab \tup{i}_1 \neq \tup{i}_2\}&\text{if $z_1 = z_2$}
  \end{cases}\]
That would be a nice definition, but we also need $\tup{i}_1$ and $\tup{i}_2$ elsewhere...}%
In order to prove secrecy, we also define $\prove{\distinct(x)}(\pp_1, \pp_2) = (z_1 \neq z_2) \vee (\theta_1\fset_{\pp_1} \cup \theta_2\fset_{\pp_2} \cup \{ \theta_1 \tup{M_1} = \theta_2 \tup{M_2},\ab \tup{i}_1 \neq \tup{i}_2\}$ yields a contradiction$)$,
where $\defrestr_{\pp_1}(x) = z_1[\tup{M_1}]$,
$\defrestr_{\pp_2}(x) = z_2[\tup{M_2}]$,
$\tup{i}$ are the current replication indices at the definition of $x$,
$\theta_1$ and $\theta_2$ are two distinct renamings of $\tup{i}$ to fresh replication indices, 
$\tup{i}_1 = \theta_1\tup{i}$, and $\tup{i}_2 = \theta_2\tup{i}$.
Intuitively, $\prove{\distinct(x)}(\pp_1, \pp_2)$ guarantees that, if
$x[\tup{i}_1]$ is defined at $\pp_1$, so
$x[\tup{i}_1] = z_1[\theta_1\tup{M_1}]$,
and $x[\tup{i}_2]$ is defined at $\pp_2$, so
$x[\tup{i}_2] = z_2[\theta_2\tup{M_2}]$,
with $\tup{i}_1 \neq \tup{i}_2$, then
the random variables that define $x$ in these two cases, $z_1[\theta_1\tup{M_1}]$ and $z_2[\theta_2\tup{M_2}]$,
are different, that is, $z_1 \neq z_2$ or $\theta_1\tup{M_1} \neq \theta_2\tup{M_2}$.
Therefore, $z_1[\theta_1\tup{M_1}]$ is independent of $z_2[\theta_2\tup{M_2}]$,
so $x[\tup{i}_1]$ is independent of $x[\tup{i}_2]$.
Combining this information with the proof of one-session secrecy, we can prove secrecy of $x$:
we define
\[  \prove{\secr(x)}(\sset) = \prove{\secrone(x)}(\sset) \wedge
\bigwedge_{\pp_1, \pp_2 \in \sset} \prove{\distinct(x)}(\pp_1, \pp_2)\]
The proof of bit secrecy is the same as for one-session secrecy:
\[ \prove{\secrbit(x)}(\sset) = \prove{\secrone(x)}(\sset)\]
The proof of (one-session or bit) secrecy is justified by the following proposition.

\begin{proposition}[(One-session or bit) secrecy]\label{prop:sec}
Consider a process $Q_0$ that satisfies Properties~\ref{prop:notables} and~\ref{prop:autosarename}.
Let $\prop$ be $\secrone(x)$, $\secr(x)$, or $\secrbit(x)$.
Let $\sset = \{\pp \mid \pp$ follows a definition of $x\}$.
If $\prove{\prop}(\sset)$ and
for all evaluation contexts $C$ acceptable for $Q_0$,
$\Prss{C[Q_0]}{\neg \yctran{\prove{\prop}(\sset)}}\leq p(C)$, 
then
$Q_0$ satisfies $\prop$ with public variables $V$ ($x \notin V$)
up to probability $p'$ such that $p'(C) = p(C[C_{\prop}[\,]])$
and $\bound{Q_0}{V \cup\{x\}}{\prop}{\Dfalse}{p}$.
%% for all evaluation contexts $C$ acceptable for $C_{\prop}[Q_0]$ with public variables $V$ ($x \notin V$)
%% that do not contain $\sevent$ nor $\sbarevent$,
%% $\Advtev{Q_0}{\prop}{C[C_{\prop}[\,]]}{\Dfalse} \leq \Advt_{Q_0}^{\prop}(C) \leq p(C)$.
\end{proposition}

The proof of Proposition~\ref{prop:sec} relies on the following definitions and
lemma.
We have
\begin{align*}
&\yctran{\noleak(x[\tup{M}'], \sri, \fset)} = (\forall \sri, \neg\bigwedge \fset) \vee ((x \notin V) \wedge \bigwedge_{\pptag x[\tup{M}] \text{ in }Q_0}\\
&\qquad \text{-- if $\pptag x[\tup{M}]$ is in $M$ in an assignment $\bassign{y[\tup{i}]}{M}$, $M$ is built from replication indices,}\\[-1mm]
  &\qquad \text{variables, function applications, and conditionals, and
    the current call is not inside a}\\[-1mm]
  &\qquad \text{call to $\noleak(y, \_, \_)$, then}\\
&\qquad\quad \yctran{\noleak(y[\theta\tup{i}], \sri \cup \{ \theta\tup{i} \}, \fset \cup \theta \fset_{\pp} \cup \{ \theta\tup{M} = \tup{M}' \})}\\[-1mm]
&\qquad\quad \text{where $\theta$ is a renaming of $\tup{i}$ to fresh replication indices}\\
&\qquad \text{-- if $\pptag x[\tup{M}]$ is in $\kevent{e(M_1, \dots, M_{k-1}, C[\pptag x[\tup{M}]], M_{k+1}, \dots, M_m)}$ for $C$ defined in}\\[-1mm]
&\qquad \text{Figure~\ref{fig:eventcontexts}, then true}\\
&\qquad \text{-- otherwise}, \forall (\sri \cup\theta\replidx{\pp}), \neg \bigwedge (\fset \cup \theta \fset_{\pp} \cup \{ \theta\tup{M} = \tup{M}' \})\\[-1mm]
&\qquad\quad \text{where $\theta$ is a renaming of $\replidx{\pp}$ to fresh replication indices})
\end{align*}
$\yctran{\noleak(x[\tup{M}'], \sri, \fset)}$ is the logical formula that is guaranteed when
$\noleak(x[\tup{M}'], \sri, \fset)$ succeeds, up to a small 
probability computed by the equational prover and that bounds 
$\Prss{C[Q_0]}{\neg \yctran{\noleak(x[\tup{M}'], \sri, \fset)}}$.
It is obtained by collecting formulas guaranteed by each call to 
``$\fset$ yields a contradiction'': such a call guarantees 
$\forall \tup{z} \in \tup{T}'', \neg \bigwedge \fset$
up to a small probability that it evaluates, where
$\tup{z}$ are the non-process variables in $\fset$ and 
$\tup{T}''$ are their types.
In the definition of $\yctran{\noleak(x[\tup{M}'], \sri, \fset)}$,
the notation $\forall \sri$ means that all 
variables in $\sri$ are universally quantified
in their respective types. The notation $\forall (\sri \cup\theta\replidx{\pp})$ is similar.
We have similarly
{\allowdisplaybreaks\begin{align*}
\yctran{\prove{\secrone(x)}(\pp)} &=
\begin{cases}
  \yctran{\noleak(\theta\defrestr_{\pp}(x), \{ \theta\replidx{\pp} \}, \theta\fset_{\pp})}\\
  \phantom{\false \quad}\text{if $\defrestr_{\pp}(x)$ is defined},\\
  \phantom{\false \quad}\text{where $\theta$ is a renaming of $\replidx{\pp}$ to fresh replication indices}\\
  \false\quad\text{otherwise}
\end{cases}\\
\yctran{\prove{\secrone(x)}(\sset)} &= \bigwedge_{\pp \in \sset} \yctran{\prove{\secrone(x)}(\pp)}\\
  \yctran{\prove{\distinct(x)}(\pp_1, \pp_2)} &= (z_1 \neq z_2) \vee (\forall \tup{i}_1, \forall \tup{i}_2, \neg \bigwedge \theta_1\fset_{\pp_1} \cup \theta_2\fset_{\pp_2} \cup \{ \theta_1 \tup{M_1} = \theta_2 \tup{M_2},\ab \tup{i}_1 \neq \tup{i}_2\})\\
  \yctran{\prove{\secr(x)}(\sset)} &= \yctran{\prove{\secrone(x)}(\sset)} \wedge
  \bigwedge_{\pp_1, \pp_2 \in \sset} \yctran{\prove{\distinct(x)}(\pp_1, \pp_2)}\\
  \yctran{\prove{\secrbit(x)}(\sset)} &= \yctran{\prove{\secrone(x)}(\sset)}
\end{align*}}%

The only semantic rules that can add $x[\tup{a}]$ to the environment $E$ are~\eqref{sem:newt}, \eqref{sem:lett}, \eqref{sem:findt1}, \eqref{sem:new}, \eqref{sem:let}, \eqref{sem:find1}, and~\eqref{sem:output}.
($\GET$ is excluded by Property~\ref{prop:notables}.)
By Corollary~\ref{cor:evalpos}, the target term or process of these rules is a subterm or subprocess of $Q_0$ up to renaming of channels.
Hence, the target configuration $\conf$ of these rules is at some program point $\pp$ in $\trace$. 
In this case, we say that $x[\tup{a}]$ \emph{is defined just before} $\pp$ in a
trace $\trace$.
Furthermore, given $x[\tup{a}]$ and $\trace$, there is at most one program point $\pp$ such that $x[\tup{a}]$ is defined just before $\pp$ in $\trace$, by Lemma~\ref{lem:stronginv1exec}.

Let $\prop$ be $\secrone(x)$, $\secr(x)$, or $\secrbit(x)$.
Let $\trace$ be a trace of $C[C_{\prop}[Q_0]]$.
Let $E = E_{\trace}$.
We define the set $\Stestidx(\trace)$ of indices of successful test queries as follows:
\begin{itemize}
\item When $\prop$ is $\secrone(x)$:
  Let $\pp_t$ be the program point of the input that performs the test query in
  $\tproco{x}$: ${}^{\pp_t}\cinput{\cS}{\vf_1:[1, n_1], \ldots, \vf_m:[1, n_m]}$.
  If there is an~\eqref{sem:output} reduction in $\trace$ with $\pp' = \pp_t$,
  we define $E'$ to be the environment after that reduction.
%%
%%   If $\{ \vf_1, \dots, \vf_m \} \subseteq \dom(E)$, we define $E'$ to be
%% the minimal environment of $\trace$ such that
%% $\{ \vf_1, \dots, \vf_m \} \subseteq \dom(E')$.
%% (The variables $\vf_1$, \dots, $\vf_m$ are variables of $\tproco{x}$.)
%%%%% Does not work in case $m=0$!
%%
  If $x[E(\vf_1), \dots, E(\vf_m)] \in \dom(E')$, we let $\Stestidx(\trace) = \{ \epsilon \}$ (the empty sequence of indices).
Otherwise, $\Stestidx(\trace) = \emptyset$.

\item When $\prop$ is $\secr(x)$:
  Let $\pp_t$ be the program point of the input that performs the test query in
  $\tprocs{x}$: ${}^{\pp_t}\cinput{\cS}{\vf_1:[1, n_1], \ldots, \vf_m:[1, n_m]}$.
Let $\Stestidx(\trace) = \{ a \in [1, \nS] \mid {}$
there is an~\eqref{sem:output} reduction in $\trace$ with $\pp' = \pp_t$, $\sigma'(\tup{i}) = a$, and $x[E(\vf_1[a]), \dots, E(\vf_m[a])] \in \dom(E')$ where $E'$ is the environment after that reduction, and for all \eqref{sem:output} reductions in $\trace$ before the latter reduction, with $\pp' = \pp_t$, $\sigma'(\tup{i}) = a'$, $E(\vf_1[a']) = E(\vf_1[a])$, \dots, and $E(\vf_m[a']) = E(\vf_m[a])$, we have $x[E(\vf_1[a]), \dots, E(\vf_m[a])] \notin \dom(E'')$ where $E''$ is the environment after that reduction$\}$.
%%
%%Let $\Stestidx(\trace) = \{ a \in [1, \nS] \mid \{ \vf_1[a], \dots \vf_m[a] \} \subseteq \dom(E)$, $x[E(\vf_1[a]), \dots, E(\vf_m[a])] \in \dom(E')$ where $E'$ is the minimal environment in $\trace$ that contains such that $\{ \vf_1[a], \dots \vf_m[a] \} \subseteq \dom(E')$ and there is no $a' \neq a$ such that $E'(\vf_1[a']) = E'(\vf_1[a])$, \dots, $E'(\vf_m[a']) = E'(\vf_m[a]) \}$.
%%%%% Does not work in case $m=0$! Plus forgot to require that $x[E(\vf_1[a']), \dots, E(\vf_m[a'])] \in \dom(E'')$ where $E''$ is the minimal environment in $\trace$ that contains such that $\{ \vf_1[a'], \dots \vf_m[a'] \} \subseteq \dom(E'')$

The test query with index $a$ is the first successful test query for $x[E(\vf_1[a]), \dots, E(\vf_m[a])]$, so $x[E(\vf_1[a]), \dots, E(\vf_m[a])]$ is defined at that test query, that is, $x[E(\vf_1[a]), \dots, E(\vf_m[a])] \in \dom(E')$. For all previous test queries on the same indices, $x[E(\vf_1[a]), \dots, E(\vf_m[a])]$ was not defined, that is, $x[E(\vf_1[a]), \dots, E(\vf_m[a])] \notin \dom(E'')$. The bound $\nS$ and the variables $\vf_1$, \dots, $\vf_m$ come from $\tprocs{x}$.

\item When $\prop$ is $\secrbit(x)$:
Let $\pp_t$ be the program point of the input in
  $\tprocb{x}$: ${}^{\pp_t}\cinput{\cS''}{b':\bool}$.
  If there is an~\eqref{sem:output} reduction in $\trace$ with $\pp' = \pp_t$,
  we define $E'$ to be the environment after that reduction.
  If $x \in \dom(E')$, we let $\Stestidx(\trace) = \{ \epsilon \}$.
  Otherwise, $\Stestidx(\trace) = \emptyset$.
  
\end{itemize}
Let $\Stestpp(\trace) = \{ \pp \mid \exists a \in \Stestidx(\trace), x[E(\vf_1[a]), \dots, E(\vf_m[a])]$ is defined just before $\pp$ in $\trace \}$.
When $\prop$ is $\secrbit(x)$, $m = 0$, so this definition reduces to
$\Stestpp(\trace) = \emptyset$ if $\Stestidx(\trace) = \emptyset$
and $\Stestpp(\trace) = \{ \pp_{\epsilon} \}$ where $x$ is defined just before
$\pp_{\epsilon}$ in $\trace$ otherwise.
We write $\trace \vdash \prop$ when 
$\trace \vdash \yctran{\prove{\prop}(\Stestpp(\trace))}$.

\begin{lemma}\label{lem:sec}
Consider a process $Q_0$ that satisfies Properties~\ref{prop:notables} and~\ref{prop:autosarename}.
Let $\prop$ be $\secrone(x)$, $\secr(x)$, or $\secrbit(x)$.
Let $C$ be an evaluation context acceptable for $C_{\prop}[Q_0]$ with any public variables $V$ ($x \notin V$) that does not contain $\sevent$ nor $\sbarevent$.
We have
\[\Pr[C[C_{\prop}[Q_0]]: \sevent \wedge \prop] = \Pr[C[C_{\prop}[Q_0]]: \sbarevent \wedge \prop]\,.\]
\end{lemma}

\begin{proofof}{the cases $\prop=\secrone(x)$ and $\prop=\secr(x)$}
Let $\trace$ be a full trace of $C[C_{\prop}[Q_0]]$ such that $\trace \vdash \prop$ and $b$ is defined in $\trace$, that is, $b \in \dom(E_{\trace})$.

Let $E = E_{\trace}$.
Let $\tup{i}$ be the current replication indices at the definition of $x$ in $Q_0$.
For $j\in \Stestidx(\trace)$, let $\tup{a}_j = E(u_1[j]), \dots, E(u_m[j])$, 
so that the test query at index $j$ tests $x[\tup{a}_j]$.
Let $\conf_j$ be the target configuration of the semantic rule that adds $x[\tup{a}_j]$ to $E$, and $\pp_j$ be such that $x[\tup{a}_j]$ is defined just before $\pp_j$ in $\trace$. So $\conf_j$ is at program point $\pp_j$ in $\trace$.
Let $z_j[\tup{M}_j] = \defrestr_{\pp_j}(x)$ (which is always defined since $\pp_j \in \Stestpp(\trace)$ and $\trace \vdash \prop$, so $\trace \vdash \yctran{\prove{\secrone(x)}(\pp_j)}$). 
Let $\tup{b}_j$ be such that $E, \{\tup{i} \mapsto \tup{a}_j\}, \tup{M}_j \evalterm \tup{b}_j$.
Then $x[\tup{a}]$ is added to the environment $E$ by~\eqref{sem:newt}, \eqref{sem:lett}, \eqref{sem:new}, or~\eqref{sem:let}, we have $E(x[\tup{a}_j]) = E(z_j[\tup{b}_j])$, and $z_j[\tup{b}_j]$ is chosen at random by~\eqref{sem:newt} or~\eqref{sem:new} in $\trace$ by definition of $\defrestr_{\pp_j}(x)$.
\bbnote{Formally, I should prove that using other semantic rules is not possible, by examining the derivation of the trace in these other cases.}%
Let us prove that, for all $j_1 \neq j_2$ in $\Stestidx(\trace)$, we have $z_{j_1} \neq z_{j_2}$ or $\tup{b}_{j_1} \neq \tup{b}_{j_2}$. 
\begin{itemize}
\item When $\prop$ is $\secrone(x)$, this is trivially true since $\Stestidx(\trace)$ contains at most one element.
\item When $\prop$ is $\secr(x)$, we have $\pp_{j_1}, \pp_{j_2} \in \Stestpp(\trace)$ and $\trace \vdash \secr(x)$, so we have
$\trace \vdash \yctran{\prove{\distinct(x)}(\pp_{j_1}, \pp_{j_2})}$, so $z_{j_1} \neq z_{j_2}$ or 
$\trace \vdash \forall \tup{i}_1, \forall \tup{i}_2, \neg \bigwedge \theta_1\fset_{\pp_{j_1}} \cup \theta_2\fset_{\pp_{j_2}} \cup \{ \theta_1 \tup{M}_{j_1} = \theta_2 \tup{M}_{j_2},\ab \tup{i}_1 \neq \tup{i}_2\}$
where $\theta_1$ and $\theta_2$ are two distinct renamings of $\tup{i}$ to fresh replication indices, 
$\tup{i}_1 = \theta_1\tup{i}$, and $\tup{i}_2 = \theta_2\tup{i}$.
In the latter case, 
let $\venv = \{ \tup{i}_1 \mapsto \tup{a}_{j_1}, \tup{i}_2 \mapsto \tup{a}_{j_2} \}$.
We have $\sigma_{\conf_{j_1}} = [\tup{i} \mapsto \tup{a}_{j_1}]$.
By Corollary~\ref{cor:factsP}, $\trace, \venv \vdash \theta_1\fset_{\pp_{j_1}}$.
Similarly, $\trace, \venv \vdash \theta_2\fset_{\pp_{j_2}}$.
Since $j_1 \neq j_2$, $\tup{a}_{j_1} \neq \tup{a}_{j_2}$. (By construction of $\Stestidx(\trace)$, we consider only the first successful test query for a certain $x[\tup{a}_j]$.\bbnote{More formally?}) So $\trace, \venv \vdash \tup{i}_1 \neq \tup{i}_2$. Therefore, $\trace, \venv \vdash \theta_1 \tup{M}_{j_1} \neq \theta_2 \tup{M}_{j_2}$,
so $\tup{b}_{j_1} \neq \tup{b}_{j_2}$.
\end{itemize}
For $j\in \Stestidx(\trace)$, let us choose elements $v_j$ in $T$, where $T$ is the type of $x$.
Let us consider the following two sets of traces:
\begin{enumerate}
\item\label{tr:btrue} $\trace$ modified by choosing $b = \true$ and $z_j[\tup{b_j}] = v_j$ for all $j\in \Stestidx(\trace)$. (The variable $b$ is chosen and used in $\tproc{\prop}$. Note that the variable $y$ of $\tproc{\prop}$ is not defined when $b = \true$. This set contains a single trace.)
\item\label{tr:bfalse} $\trace$ modified by choosing $b = \false$ and $y[j] = v_j$ and $z_j[\tup{b_j}] = v'_j$ for any $v'_j\in T$, for all $j\in \Stestidx(\trace)$. (This set contains $|T|^{|\Stestidx(\trace)|}$ traces.)
\end{enumerate}
The trace $\trace$ is one of the traces in these two sets: just choose the values of $b$, $z_j[\tup{b}_j]$, and $y[j]$ if $b$ is false that are used in $\trace$.

We show by induction on the derivation of these traces $\trace_s$ ($\trace_1$ is in set~\ref{tr:btrue} and
$\trace_2$ is in set~\ref{tr:bfalse})
that they have matching configurations $\conf'_s = E_s, \sigma, M_s, \tblcts, \evseq_s$, 
$\conf'_s = E_s, \pset_s, \cset$, or
$\conf'_s = E_s, (\sigma, P_s), \pset_s, \cset, \tblcts, \evseq_s$ for $s \in \{1,2\}$
that differ as follows: 
\begin{itemize}
\item $E_1(b) = \true$ while $E_2(b) = \false$.
\item $E_1(\vfS'[j])$ when $\prop$ is $\secr(x)$ and $E_1(y[j])$ are undefined even when $E_2(\vfS'[j])$ and $E_2(y[j])$ are defined for some indices $j \in \Stestidx(\trace)$.
\item $E_1(z[\tup{a}])$ differs from $E_2(z[\tup{a}])$ for some $z$ and $\tup{a}$ such that for some $\tup{M}$, $\sri$, $\fset$, we have $\trace \vdash \yctran{\noleak(z[\tup{M}], \sri, \fset)}$ and
$\trace \vdash \exists \sri, (\tup{M} = \tup{a})\wedge \bigwedge \fset$.
\item Values inside $M_1$ and $M_2$ may differ when $\conf'_s$ (for $s \in \{1,2\}$) occurs in the derivation of 
\begin{align*}
&E_s, \sigma, \assign{y[\tup{i}]}{M} M', \tblcts, \evseq_s \red{1}{}^* E_s[y[\sigma\tup{i}] \mapsto a_s], \sigma, M', \tblcts, \evseq_s\\
\text{or of }&E_s, (\sigma, \assign{y[\tup{i}]}{M} P), \pset_s, \cset,  \tblcts, \evseq_s \red{1}{}^* E_s[y[\sigma\tup{i}] \mapsto a_s], (\sigma, P), \pset_s, \cset,  \tblcts, \evseq_s
\end{align*}
where $M$ is built from replication indices, variables, function applications, and conditionals and for some $\tup{M}$, $\sri$, $\fset$, we have 
$\trace \vdash \yctran{\noleak(y[\tup{M}], \sri, \fset)}$ and
$\trace \vdash \exists \sri, (\tup{M} = \sigma\tup{i})\wedge \bigwedge \fset$.
\item Terms $M_1$ and $M_2$ may differ when $\conf'_s$ (for $s \in \{1,2\}$) occurs in the derivation of 
\begin{align*}
&E_s, \sigma, \kevent{e(\tup{M}_s)}; M', \tblcts, \evseq_s \red{1}{} E_s, \sigma, \kevent{e(\tup{M}'_s)}; M', \tblcts, \evseq_s\\
\text{or of }&E_s, (\sigma, \kevent{e(\tup{M}_s)}; P), \pset_s, \cset,  \tblcts, \evseq_s \red{1}{} E_s, (\sigma, \kevent{e(\tup{M}'_s)}; P), \pset_s, \cset,  \tblcts, \evseq_s
\end{align*}
where the terms $M_s$ match and the only rules above this reduction and under $\conf'_s$ are~\eqref{sem:ctxt} with matching simple contexts any number of times followed by~\eqref{sem:ctxt} with context $\kevent{e(a_{s,1}, \ldots, a_{s,k-1}, [\,], N_{k+1}, \ldots, N_l)}; N$ or~\eqref{sem:ctx} with context $\kevent{e(a_{s,1}, \ab\ldots, \ab a_{s,k-1}, \ab[\,], \ab N_{k+1}, \ab\ldots, \ab  N_l)};P$ once,
where 
\begin{itemize}
\item a \emph{simple context} is a context of the form  $x[a_1, \ab \ldots,\ab a_{k-1}, \ab [\,], \ab N_{k+1}, \ab\ldots, \ab N_m]$ or $f(a_1,\ab \ldots, \ab  a_{k-1},\ab [\,], \ab N_{k+1}, \ab \ldots, \ab N_m)$ for some $x$, $f$, $k$, $m$, $a_1$, \dots, $a_{k-1}$, $N_{k+1}$, \dots, $N_m$,
\item simple contexts $C_s$ (for $s \in \{1,2\}$) \emph{match} when $C_s = x[a_{s,1}, \ab  \ldots,\ab a_{s,k-1}, \ab [\,], \ab N_{k+1}, \ab\ldots, \ab N_m]$ or $C_s = f(a_{s,1}, \ab \ldots, \ab a_{s,k-1}, \ab [\,], \ab N_{k+1}, \ab \ldots, \ab N_m)$ for some $x$, $f$, $k$, $m$, $a_{s,1}$, \dots, $a_{s,k-1}$, $N_{k+1}$, \dots, $N_m$, and
\item  terms $M_s$ (for $s \in \{1,2\}$) \emph{match} when $M_s = a_{s,0}$, or $M_s = x[a_{s,1}, \ab  \ldots,\ab a_{s,k-1}, \ab M'_s, \ab N_{k+1}, \ab\ldots, \ab N_m]$ or $M_s = f(a_{s,1}, \ab \ldots, \ab a_{s,k-1}, \ab M'_s, \ab N_{k+1}, \ab \ldots, \ab N_m)$ for some $x$, $f$, $k$, $m$, $a_{s,0}$, $a_{s,1}$, \dots, $a_{s,k-1}$, $N_{k+1}$, \dots, $N_m$ and matching $M'_s$.
\end{itemize}
\item Terms $M_1$ and $M_2$ (resp. processes $P_1$ and $P_2$) may differ when 
\begin{align*}
&M_s = C_1[\dots C_k[\kevent{e(\tup{M}_s)}; M']\dots]\,, \\
&P_s = C_0[C_1[\dots C_k[\kevent{e(\tup{M}_s)}; M']\dots]]\,,\\
\text{or }&P_s = \kevent{e(\tup{M}_s)}; P
\end{align*} 
where
$k \in \mathbb{N}$,
$C_1$, \dots, $C_k$ are term contexts defined in Figure~\ref{fig:termcontexts}, 
$C_0$ is a process context defined in Figure~\ref{fig:proccontexts},
and the terms $\tup{M}_s$ match, for $s \in \{1,2\}$.

\item Arguments of events in $\evseq_1$ and $\evseq_2$ may differ.
\item The events $\sevent$ and $\sbarevent$ are swapped: when $\evseq_1$ contains $\sevent$, $\evseq_2$ contains $\sbarevent$, and conversely.
\item Some additional configurations corresponding to the execution $\tproc{\prop}$ differ.
\end{itemize}
The proof can be sketched as follows.
The different choice of $b$ leads to $E_1(b) = \true$ and $E_2(b) = \false$.
The only semantic rule that reads the environment is~\eqref{sem:var},
when it evaluates an occurrence of the variable in question. 
By Definition~\ref{def:secr}, $b \notin \fvar(Q_0) \cup V$, so the only occurrences of $b$ are in $\tproc{\prop}$.

If the adversary sends $\tup{a}$ on channel $\cS$ and
$x[\tup{a}]$ is not defined, then $\tproc{\prop}$ simply yields.
If the adversary sends $\tup{a}$ on channel $\cS$ and
$x[\tup{a}]$ is defined, then $\tup{a} = \tup{a}_j$ for some $j \in \Stestidx(\trace)$. In set~\ref{tr:btrue}, $E_1(b) = \true$, so $\tproc{\prop}$ outputs
$E_1(x[\tup{a}_j]) = E_1(z_j[\tup{b}_j]) = v_j$.
In set~\ref{tr:bfalse}, $E_2(b) = \false$, so when it is the first time that
the adversary sends $\tup{a}$ on channel $\cS$ and
$x[\tup{a}]$ is defined, $\tproc{\prop}$ chooses a fresh $y[j]$ equal to $v_j$,
and outputs
$E_2(y[j]) = v_j$; when the adversary sends again $\tup{a}$ on channel $\cS$,
$\tproc{\prop}$ finds $\vfS' = j$
(by construction of $\Stestidx(\trace)$), and outputs
$E_2(y[j]) = v_j$. So in both sets, $\tproc{\prop}$ outputs the same value.

If the adversary sends $b'$ to $\cS'$, then the result of the test $b' = b$
differs between set~\ref{tr:btrue} and set~\ref{tr:bfalse}, since $E_1(b) \neq E_2(b)$, so if set~\ref{tr:btrue} executes $\sevent$, then set~\ref{tr:bfalse} executes $\sbarevent$ and conversely. That is why the events $\sevent$ and $\sbarevent$ are swapped.

By Definition~\ref{def:secr}, $\vfS', y \notin \fvar(Q_0) \cup V$, so the only occurrences of $\vfS'$ and $y$ are in $\tproc{\prop}$. Therefore, the changes that come from differences in the definition of $\vfS'$ and $y$ are already taken into account above.

The value of $E_1(z_j[\tup{b_j}])$ is different from the one of $E_2(z_j[\tup{b_j}])$. Since $\trace \vdash \prop$, we have $\trace \vdash \yctran{\prove{\secrone(x)}(\pp_j)}$,
so $\trace \vdash \yctran{\noleak(\theta z_j[\tup{M}_j], \{ \theta \replidx{\pp_j} \}, \theta \fset_{\pp_j})}$
where $\theta$ is a renaming of $\replidx{\pp_j}$ to fresh replication indices.
We have $\sigma_{\conf_j} = [\replidx{\pp_j} \mapsto \tup{a}_j]$.
Let $\venv = \{ \theta\replidx{\pp_j} \mapsto \tup{a}_j \}$.
By Corollary~\ref{cor:factsP}, $\trace, \venv \vdash \theta \fset_{\pp_j}$.
Moreover, $\trace, \venv \vdash \theta \tup{M}_j = \tup{b}_j$.
So $\trace \vdash \exists \theta \replidx{\pp_j}, (\theta \tup{M}_j = \tup{b}_j) \wedge \bigwedge\theta \fset_{\pp_j}$. Hence, for $\tup{M} = \theta \tup{M}_j$, $\sri = \{ \theta \replidx{\pp_j} \}$, and $\fset = \theta \fset_{\pp_j}$, we have $\trace \vdash \yctran{\noleak(z_j[\tup{M}], \sri, \fset)}$ and $\trace \vdash \exists \sri, (\tup{M} = \tup{b}_j) \wedge \bigwedge\fset$.

The difference between $E_1(z[\tup{a}])$ and $E_2(z[\tup{a}])$ for $z$ and $\tup{a}$ such that for some $\tup{M}$, $\sri$, $\fset$, we have $\trace \vdash \yctran{\noleak(z[\tup{M}], \sri, \fset)}$ and
$\trace \vdash \exists \sri, (\tup{M} = \tup{a})\wedge \bigwedge \fset$ has consequences when~\eqref{sem:var} evaluates $z[\tup{a}]$:
\[E_s, \sigma, \pptag z[\tup{a}], \tblcts, \evseq_s \red{1}{} E_s, \sigma, E_s(z[\tup{a}]), \tblcts, \evseq_s\,.\]
By Lemma~\ref{lem:start_from_pp}, Property~\ref{start_from_pp5} applied to the configuration $\conf = E_s, \sigma, \pptag z[\tup{a}], \tblcts, \evseq_s$ with $l = 0$,
we have
\[E'_s, \sigma, \pptag z[\tup{M'}], \tblcts', \evseq'_s
\red{1}{}^* E_s, \sigma, \pptag z[\tup{a}], \tblcts, \evseq_s\]
by any number of applications of~\eqref{sem:ctxt}, where
$\pptag z[\tup{M}']$ is a subterm of $C[C_{\prop}[Q_0]]$.
Furthermore, if the evaluation of $\tup{M}'$ itself uses~\eqref{sem:var} that
evaluates a $z[\tup{a}]$ that differs, we replace $z[\tup{M}']$ by the smallest subterm of $\tup{M}'$ that evaluates a $z[\tup{a}]$ that differs. By this replacement, we guarantee that the evaluation of $\tup{M'}$ proceeds in the same way in set~\ref{tr:btrue} and set~\ref{tr:bfalse} and yields the same $\tup{a}$.
Recall that $\tup{M}'$ are simple terms by Invariants~\ref{inv2} and~\ref{invtic}, so the evaluation of
$\tup{M'}$ does not change $E_s$, $\tblcts$, $\evseq_s$.
So we have
\[E_s, \sigma, \pptag z[\tup{M'}], \tblcts, \evseq_s
\red{1}{}^* E_s, \sigma, \pptag z[\tup{a}], \tblcts, \evseq_s
\red{1}{} E_s, \sigma, E_s(z[\tup{a}]), \tblcts, \evseq_s\]
for $s \in \{1,2\}$, where $\pptag z[\tup{M}']$ is a subterm of $C[C_{\prop}[Q_0]]$,
by any number of applications of~\eqref{sem:ctxt} followed by one application of~\eqref{sem:var}.
  Let $\theta$ be a renaming of $\replidx{\pp}$ to fresh replication indices and $\venv = \{ \theta\replidx{\pp} \mapsto \sigma\replidx{\pp} \}$.
  By Corollary~\ref{cor:factsP}, $\trace_s, \venv \vdash \theta\fset_{\pp}$.
  Moreover, $E_s, \sigma, \tup{M}' \evalterm \tup{a}$,
  so $E_s, \venv, \theta\tup{M'} \evalterm \tup{a}$.
Since $E_{\trace_s}$ extends $E_s$, we have $\trace_s, \venv \vdash \theta\tup{M}' = \tup{a}$. Since $\trace$ is among the traces $\trace_s$, we have
  $\trace, \venv \vdash \theta\fset_{\pp}$
  and $\trace, \venv \vdash \theta\tup{M}' = \tup{a}$.
  There exists $\venv'$ with domain $\sri$ such that
  $\trace, \venv'\vdash (\tup{M} = \tup{a})\wedge \bigwedge \fset$.
  So $\trace, \venv \cup \venv'\vdash \bigwedge (\fset
  \cup \theta\fset_{\pp} \cup \{ \theta\tup{M}' = \tup{a}, \tup{M} = \tup{a} \})$.
Since $\trace \vdash \yctran{\noleak(z[\tup{M}], \sri, \fset)}$ and $\trace \vdash \exists \sri, (\tup{M} = \tup{a})\wedge \bigwedge \fset$, we have $\trace \vdash \neg (\forall \sri, \neg \bigwedge \fset)$, so $\trace$ satisfies the second disjunct of $\yctran{\noleak(z[\tup{M}], \sri, \fset)}$. Therefore, $z \notin V$, so the occurrence of $\pptag z[\tup{M}']$ evaluated above is either in $\tproc{\prop}$, and in this case $z$ is actually $x$ and this case has already been studied above, or in $Q_0$.
In the latter situation, we are in one of the following three cases:
\begin{itemize}
\item $\pptag z[\tup{M}']$ is in $M$ in an assignment ${}^{\pp'} \bassign{y[\tup{i}]}{M}$ in $Q_0$, $M$ is built from replication indices, variables, function applications, and conditionals and $\trace \vdash \yctran{\noleak(y[\theta\tup{i}], \sri \cup \{ \theta\tup{i} \}, \fset \cup \theta \fset_{\pp} \cup \{ \theta\tup{M}' = \tup{M} \})}$ where $\theta$ is a renaming of $\tup{i}$ to fresh replication indices.
  Let $\tup{M}'' = \theta\tup{i}$, $\sri' =  \sri \cup \{ \theta\tup{i} \}$, and $\fset' = \fset \cup \theta \fset_{\pp} \cup \{ \theta\tup{M}' = \tup{M} \}$.
  Since $\pp'$ is above $\pp$, by Lemma~\ref{lem:above_before}, there is a configuration inside $\pp'$ before $\conf$ in $\trace$. By Lemma~\ref{lem:start_from_pp} applied to that configuration (Property~\ref{start_from_pp1} when the assignment is a process, Property~\ref{start_from_pp5} with $l = 0$ when it is a term),
  the assignment ${}^{\pp'} \bassign{y[\tup{i}]}{M}$ is evaluated by
\begin{align*}
&E_s, \sigma, {}^{\pp'} \assign{y[\tup{i}]}{M} M', \tblcts, \evseq_s \red{1}{}^* E_s[y[\sigma\tup{i}] \mapsto a_s], \sigma, M', \tblcts, \evseq_s\\
\text{or }&E_s, (\sigma, {}^{\pp'} \assign{y[\tup{i}]}{M} P), \pset_s, \cset,  \tblcts, \evseq_s \red{1}{}^* E_s[y[\sigma\tup{i}] \mapsto a_s], (\sigma, P), \pset_s, \cset,  \tblcts, \evseq_s
\end{align*}
for $s \in \{1, 2\}$. We have $\trace \vdash \yctran{\noleak(y[\tup{M''}], \sri', \fset')}$. Moreover $\tup{i} = \replidx{\pp}$ and $\trace, \venv \cup \venv'\vdash \theta\tup{i} = \sigma\tup{i}$ by definition of $\venv$. Hence $\trace, \venv \cup \venv'\vdash \tup{M}'' = \sigma\tup{i}$ and $\trace, \venv \cup \venv'\vdash \fset'$, so
$\trace\vdash \exists\sri', (\tup{M}'' = \sigma\tup{i}) \wedge \bigwedge \fset'$.
Hence we are in a case in which different values inside terms  $M_1$ and $M_2$ are allowed.
Furthermore, the added values $E_s(y[\sigma\tup{i}]) = a_s$ may differ. Let $\tup{a}' = \sigma\tup{i}$.
We have $\trace \vdash \yctran{\noleak(y[\tup{M''}], \sri', \fset')}$ and $\trace\vdash \exists\sri', (\tup{M}'' = \tup{a}') \wedge \bigwedge \fset'$,
so $E_1(y[\tup{a}'])$ is indeed allowed to differ from $E_2(y[\tup{a}'])$.

\item $\pptag z[\tup{M}']$ is in ${}^{\pp'} \kevent{e(M_1, \dots, M_{k-1}, C[\pptag z[\tup{M}']], M_{k+1}, \dots, M_m)}$ in $Q_0$, for $C$ defined in Figure~\ref{fig:eventcontexts}.
Since $\pp'$ is above $\pp$, by Lemma~\ref{lem:above_before}, there is a configuration inside $\pp'$ before $\conf$ in $\trace$. By Lemma~\ref{lem:start_from_pp} applied to that configuration (Property~\ref{start_from_pp1} when the event is a process, Property~\ref{start_from_pp5} with $l = 0$ when it is a term), the evaluation of the event starts from a configuration at $\pp'$ in $\trace$. The evaluation of $\kevent{e(M_1, \dots, M_{k-1}, C[\pptag z[\tup{M}']], M_{k+1}, \dots, M_m)}$ first evaluates $M_1$, \dots, $M_{k-1}$ to values using~\eqref{sem:ctxt} or~\eqref{sem:ctx} with an event context. (If they evaluated to abort event values, $C[\pptag z[\tup{M}']]$ would not be evaluated.) Then if evaluates the context $C$: $\Res{y[\tup{i}]}{T}; C$ is evaluated by~\eqref{sem:newt},
$\assign{y[\tup{i}]}{M}{C}$ is evaluated by~\eqref{sem:lett}, $\bguard{M}{C}{N'}$ is evaluated by~\eqref{sem:ift1} ($M$ must evaluate to $\true$ because otherwise, $\pptag z[\tup{M}']$ would not be evaluated), $\bguard{M}{N}{C}$ is evaluated by~\eqref{sem:ift2} ($M$ must not evaluate to $\true$ because otherwise, $\pptag z[\tup{M}']$ would not be evaluated), $\kevent{e(\tup{M})}; C$ is evaluated by~\eqref{sem:eventt}, and $\FIND$ contexts are evaluated by rules for $\FIND$, until we reach
\[\kevent{e(a_{s,1}, \dots, a_{s,k-1}, C_{s,1}[\dots C_{s,l}[\pptag z[\tup{M}']]\dots], M_{k+1}, \dots, M_m)}\]
for $s \in \{1, 2\}$, where $C_{s,1}$, \dots, $C_{s,l}$ are matching simple contexts. (Values may differ in case $M_1$, \dots, $M_{k-1}$, or terms in $C$ contain other occurrences of variables whose value differs.)
At this point, the reduction proceeds as follows:
\begin{align*}
&E_s, \sigma, \kevent{e(\tup{M}_s)}; M', \tblcts, \evseq_s \red{1}{} E_s, \sigma, \kevent{e(\tup{M}'_s)}; M', \tblcts, \evseq_s\\
\text{or }&E_s, (\sigma, \kevent{e(\tup{M}_s)}; P), \pset_s, \cset,  \tblcts, \evseq_s \red{1}{} E_s, (\sigma, \kevent{e(\tup{M}'_s)}; P), \pset_s, \cset,  \tblcts, \evseq_s
\end{align*}
by~\eqref{sem:var}, \eqref{sem:ctxt} with matching simple contexts any number of times followed by~\eqref{sem:ctxt} or~\eqref{sem:ctx} with context 
$\kevent{e(a_{s,1}, \ldots, a_{s,k-1}, [\,], M_{k+1}, \ldots, M_m)}; \dots$ once,
where 
\begin{align*}
&\tup{M}_s = a_{s,1}, \ldots, a_{s,k-1}, C_{s,1}[\dots C_{s,l}[\pptag z[\tup{M}']]\dots], M_{k+1}, \dots, M_m\\
\text{and }&\tup{M}'_s = a_{s,1}, \ldots, a_{s,k-1}, C_{s,1}[\dots C_{s,l}[a_s]\dots], M_{k+1}, \dots, M_m
\end{align*}
for $s \in \{1,2\}$. Further reductions still manipulate configurations of the same form until the event itself
is executed by~\eqref{sem:eventt} or~\eqref{sem:event}, which adds the event $e$ with possibly different arguments to $\evseq_s$.

\item $\trace \vdash \forall (\sri \cup\theta\replidx{\pp}), \neg \bigwedge (\fset \cup \theta \fset_{\pp} \cup \{ \theta\tup{M}' = \tup{M} \})$ where $\theta$ is a renaming of $\replidx{\pp}$ to fresh replication indices.
  We have $\trace, \venv \cup \venv'\vdash \neg \bigwedge (\fset \cup \theta \fset_{\pp} \cup \{ \theta\tup{M}' = \tup{M} \})$.
  That yields a contradiction, so this case does not happen.
  
\end{itemize}
The sequence of events $\evseq_1$ (resp. $\evseq_2$) is never read by the semantic rules. It is only read by the distinguisher. Therefore, changes in this sequence of events do not modify the rest of the trace.

That concludes the proof that traces in set~\ref{tr:btrue} and set~\ref{tr:bfalse} match.

Furthermore, the trace in set~\ref{tr:btrue} and the traces in set~\ref{tr:bfalse}
have the same probability. All full traces of $C[C_{\prop}[Q_0]]$
that define $b$ and that satisfy $\prop$ belong to set~\ref{tr:btrue} or to set~\ref{tr:bfalse} for some $\trace$ (for instance using the trace in question as $\trace$).
Therefore, these sets form a partition of the full traces of $C[C_{\prop}[Q_0]]$ that define $b$ and that satisfy $\prop$, and the sets that execute $\sevent$ have the same probability as the sets that execute $\sbarevent$.
Moreover, the traces of $C[C_{\prop}[Q_0]]$ that do not define $b$ execute neither $\sevent$ nor $\sbarevent$.
So $\Pr[C[C_{\prop}[Q_0]]: \sevent \wedge \prop] = \Pr[C[C_{\prop}[Q_0]]: \sbarevent \wedge \prop]$.
  \proofcomplete
\end{proofof}

\begin{proofof}{the case $\prop=\secrbit(x)$}
Let $\trace$ be a full trace of $C[C_{\prop}[Q_0]]$ such that $\trace \vdash \prop$.
Let $E = E_{\trace}$.

If $\Stestidx(\trace) = \emptyset$, then $\trace$ executes neither $\sevent$ nor $\sbarevent$.

Otherwise, $\Stestidx(\trace) = \{\epsilon\}$ and $x$ is defined in $\trace$, that is, $x \in \dom(E)$.
Let $\conf_\epsilon$ be the target configuration of the semantic rule that adds $x$ to $E$, and $\pp_\epsilon$ be such that $x$ is defined just before $\pp_\epsilon$ in $\trace$. So $\conf_\epsilon$ is at program point $\pp_\epsilon$ in $\trace$.
Let $z_\epsilon[\tup{M}_\epsilon] = \defrestr_{\pp_\epsilon}(x)$ (which is always defined since $\pp_\epsilon \in \Stestpp(\trace)$ and $\trace \vdash \prop$, so $\trace \vdash \yctran{\prove{\secrone(x)}(\pp_\epsilon)}$). 
Let $\tup{b}_\epsilon$ be such that $E, \emptyset, \tup{M}_\epsilon \evalterm \tup{b}_\epsilon$.
Then $x$ is added to the environment $E$ by~\eqref{sem:newt}, \eqref{sem:lett}, \eqref{sem:new}, or~\eqref{sem:let}, we have $E(x) = E(z_\epsilon[\tup{b}_\epsilon])$, and $z_\epsilon[\tup{b}_\epsilon]$ is chosen at random by~\eqref{sem:newt} or~\eqref{sem:new} in $\trace$ by definition of $\defrestr_{\pp_\epsilon}(x)$.
\bbnote{Formally, I should prove that using other semantic rules is not possible, by examining the derivation of the trace in these other cases.}%

Let us consider the following two traces:
\begin{enumerate}
\item\label{tr:btrue_bit} $\trace_1$ is $\trace$ modified by choosing $z_\epsilon[\tup{b_\epsilon}] = \true$.
\item\label{tr:bfalse_bit} $\trace_2$ is $\trace$ modified by choosing $z_\epsilon[\tup{b_\epsilon}] = \false$.
\end{enumerate}
The trace $\trace$ is one of these two traces: just choose the value $z_\epsilon[\tup{b}_\epsilon]$ that is used in $\trace$.

We show by induction on the derivation of these traces $\trace_s$
that they have matching configurations $\conf'_s = E_s, \sigma, M_s, \tblcts, \evseq_s$, 
$\conf'_s = E_s, \pset_s, \cset$, or
$\conf'_s = E_s, (\sigma, P_s), \pset_s, \cset, \tblcts, \evseq_s$ for $s \in \{1,2\}$
that differ as follows: 
\begin{itemize}
\item $E_1(z[\tup{a}])$ differs from $E_2(z[\tup{a}])$ for some $z$ and $\tup{a}$ such that for some $\tup{M}$, $\sri$, $\fset$, we have $\trace \vdash \yctran{\noleak(z[\tup{M}], \sri, \fset)}$ and
$\trace \vdash \exists \sri, (\tup{M} = \tup{a})\wedge \bigwedge \fset$.
\item Values inside $M_1$ and $M_2$ may differ when $\conf'_s$ (for $s \in \{1,2\}$) occurs in the derivation of 
\begin{align*}
&E_s, \sigma, \assign{y[\tup{i}]}{M} M', \tblcts, \evseq_s \red{1}{}^* E_s[y[\sigma\tup{i}] \mapsto a_s], \sigma, M', \tblcts, \evseq_s\\
\text{or of }&E_s, (\sigma, \assign{y[\tup{i}]}{M} P), \pset_s, \cset,  \tblcts, \evseq_s \red{1}{}^* E_s[y[\sigma\tup{i}] \mapsto a_s], (\sigma, P), \pset_s, \cset,  \tblcts, \evseq_s
\end{align*}
where $M$ is built from replication indices, variables, function applications, and conditionals and for some $\tup{M}$, $\sri$, $\fset$, we have 
$\trace \vdash \yctran{\noleak(y[\tup{M}], \sri, \fset)}$ and
$\trace \vdash \exists \sri, (\tup{M} = \sigma\tup{i})\wedge \bigwedge \fset$.
\item Terms $M_1$ and $M_2$ may differ when $\conf'_s$ (for $s \in \{1,2\}$) occurs in the derivation of 
\begin{align*}
&E_s, \sigma, \kevent{e(\tup{M}_s)}; M', \tblcts, \evseq_s \red{1}{} E_s, \sigma, \kevent{e(\tup{M}'_s)}; M', \tblcts, \evseq_s\\
\text{or of }&E_s, (\sigma, \kevent{e(\tup{M}_s)}; P), \pset_s, \cset,  \tblcts, \evseq_s \red{1}{} E_s, (\sigma, \kevent{e(\tup{M}'_s)}; P), \pset_s, \cset,  \tblcts, \evseq_s
\end{align*}
where the terms $M_s$ match and the only rules above this reduction and under $\conf'_s$ are~\eqref{sem:ctxt} with matching simple contexts any number of times followed by~\eqref{sem:ctxt} with context $\kevent{e(a_{s,1}, \ldots, a_{s,k-1}, [\,], N_{k+1}, \ldots, N_l)}; N$ or~\eqref{sem:ctx} with context $\kevent{e(a_{s,1}, \ab\ldots, \ab a_{s,k-1}, \ab[\,], \ab N_{k+1}, \ab\ldots, \ab  N_l)};P$ once,
where simple contexts and matching are defined as in the cases $\prop = \secrone(x)$ and $\prop = \secr(x)$.
\item Terms $M_1$ and $M_2$ (resp. processes $P_1$ and $P_2$) may differ when 
\begin{align*}
&M_s = C_1[\dots C_k[\kevent{e(\tup{M}_s)}; M']\dots]\,, \\
&P_s = C_0[C_1[\dots C_k[\kevent{e(\tup{M}_s)}; M']\dots]]\,,\\
\text{or }&P_s = \kevent{e(\tup{M}_s)}; P
\end{align*} 
where
$k \in \mathbb{N}$,
$C_1$, \dots, $C_k$ are term contexts defined in Figure~\ref{fig:termcontexts}, 
$C_0$ is a process context defined in Figure~\ref{fig:proccontexts},
and the terms $\tup{M}_s$ match, for $s \in \{1,2\}$.

\item Arguments of events in $\evseq_1$ and $\evseq_2$ may differ.
\item The events $\sevent$ and $\sbarevent$ are swapped: when $\evseq_1$ contains $\sevent$, $\evseq_2$ contains $\sbarevent$, and conversely.
\item Some additional configurations corresponding to the execution $\tproc{\prop}$ differ.
\end{itemize}
The proof can be sketched as follows.

If the adversary sends $b'$ to $\cS''$, then $x$ is defined (since $\Stestidx(\trace) = \{\epsilon\}$) and the result of the test $x = b'$
differs between $\trace_1$ and $\trace_2$, since $E_1(x) = E_1(z_\epsilon[\tup{b_\epsilon}]) = \true \neq E_2(x) = E_2(z_\epsilon[\tup{b_\epsilon}]) = \false$, so if $\trace_1$ executes $\sevent$, then $\trace_2$ executes $\sbarevent$ and conversely. That is why the events $\sevent$ and $\sbarevent$ are swapped.

The value of $E_1(z_\epsilon[\tup{b_\epsilon}])$ is different from the one of $E_2(z_\epsilon[\tup{b_\epsilon}])$. Since $\trace \vdash \prop$, we have $\trace \vdash \yctran{\prove{\secrone(x)}(\pp_\epsilon)}$,
so $\trace \vdash \yctran{\noleak(\theta z_\epsilon[\tup{M}_\epsilon], \{ \theta \replidx{\pp_\epsilon} \}, \theta \fset_{\pp_\epsilon})}$
where $\theta$ is a renaming of $\replidx{\pp_\epsilon}$ to fresh replication indices.
Here, $\replidx{\pp_\epsilon}$ is empty since $x$ is defined under no replication.
We have $\sigma_{\conf_\epsilon} = []$.
Let $\venv = \emptyset$.
By Corollary~\ref{cor:factsP}, $\trace, \venv \vdash \theta \fset_{\pp_\epsilon}$.
Moreover, $\trace, \venv \vdash \theta \tup{M}_\epsilon = \tup{b}_\epsilon$.
So $\trace \vdash \exists \theta \replidx{\pp_\epsilon}, (\theta \tup{M}_\epsilon = \tup{b}_\epsilon) \wedge \bigwedge\theta \fset_{\pp_\epsilon}$. Hence, for $\tup{M} = \theta \tup{M}_\epsilon$, $\sri = \{ \theta \replidx{\pp_\epsilon} \}$, and $\fset = \theta \fset_{\pp_\epsilon}$, we have $\trace \vdash \yctran{\noleak(z_\epsilon[\tup{M}], \sri, \fset)}$ and $\trace \vdash \exists \sri, (\tup{M} = \tup{b}_\epsilon) \wedge \bigwedge\fset$.

The difference between $E_1(z[\tup{a}])$ and $E_2(z[\tup{a}])$ for $z$ and $\tup{a}$ such that for some $\tup{M}$, $\sri$, $\fset$, we have $\trace \vdash \yctran{\noleak(z[\tup{M}], \sri, \fset)}$ and
$\trace \vdash \exists \sri, (\tup{M} = \tup{a})\wedge \bigwedge \fset$ has consequences when~\eqref{sem:var} evaluates $z[\tup{a}]$, as in the cases $\prop = \secrone(x)$ and $\prop = \secr(x)$.
That concludes the proof that traces $\trace_1$ and $\trace_2$ match.

Furthermore, the traces $\trace_1$ and $\trace_2$
have the same probability. All full traces of $C[C_{\prop}[Q_0]]$
such that $\Stestidx$ is non-empty and that satisfy $\prop$ are $\trace_1$ or $\trace_2$ for some $\trace$ (for instance using the trace in question as $\trace$).
Therefore, $\trace_1$ and $\trace_2$ form a partition of the full traces of $C[C_{\prop}[Q_0]]$ such that $\Stestidx$ is non-empty and that satisfy $\prop$, and half of these traces execute $\sevent$, the other half execute $\sbarevent$.
Moreover, the traces of $C[C_{\prop}[Q_0]]$ such that $\Stestidx$ is empty execute neither $\sevent$ nor $\sbarevent$.
So $\Pr[C[C_{\prop}[Q_0]]: \sevent \wedge \prop] = \Pr[C[C_{\prop}[Q_0]]: \sbarevent \wedge \prop]$.
  \proofcomplete
\end{proofof}

\begin{proofof}{Proposition~\ref{prop:sec}}
  Let $\prop$ be $\secrone(x)$, $\secr(x)$, or $\secrbit(x)$.
Let $C$ be an evaluation context acceptable for $C_{\prop}[Q_0]$ with public variables $V$ ($x \notin V$) that does not contain $\sevent$ nor $\sbarevent$.
  We have
  \begin{align*}
  \Advt_{Q_0}^{\prop}(C)
  & = \Pr[C[C_{\prop}[Q_0]]: \sevent] - \Pr[C[C_{\prop}[Q_0]]: \sbarevent]\\
  & = \Pr[C[C_{\prop}[Q_0]]: \sevent \wedge \prop] +
  \Pr[C[C_{\prop}[Q_0]]: \sevent \wedge \neg\prop]\\
&\qquad\!\! - \Pr[C[C_{\prop}[Q_0]]: \sbarevent \wedge \prop]
  - \Pr[C[C_{\prop}[Q_0]]: \sbarevent \wedge \neg\prop]\\
  & = \Pr[C[C_{\prop}[Q_0]]: \sevent \wedge \neg\prop] -
  \Pr[C[C_{\prop}[Q_0]]: \sbarevent \wedge \neg\prop]
  \tag*{by Lemma~\ref{lem:sec}}\\
  & \leq \Pr[C[C_{\prop}[Q_0]]: \neg \prop]\\
  & \leq \Pr[C[C_{\prop}[Q_0]]: \neg \yctran{\prove{\prop}(\{ \pp \mid \pp\text{ follows a definition of }x\})}]\tag*{because, for all $\trace$, $\Stestpp(\trace) \subseteq \{ \pp \mid \pp\text{ follows a definition of }x\}$}\\
  & \leq \Prss{C[C_{\prop}[Q_0]]}{\neg \yctran{\prove{\prop}(\{ \pp \mid \pp\text{ follows a definition of }x\})}}
  \tag*{by Lemma~\ref{lem:Pr-Prss}}\\
  & \leq p(C[C_{\prop}[\,]]) = p'(C)
  \end{align*}
  So $Q_0$ satisfies $\prop$ with public variables $V$ up to probability $p'$.
  Moreover,
  \begin{align*}
    \Advtev{Q_0}{\prop}{C[C_{\prop}[\,]]}{\Dfalse}
    & = \Pr[C[C_{\prop}[Q_0]]: \sevent] - \Pr[C[C_{\prop}[Q_0]]: \sbarevent \vee \nonunique{Q_0}]\\
    &\leq \Advt_{Q_0}^{\prop}(C) \leq p(C[C_{\prop}[\,]])
  \end{align*}
  so $\bound{Q_0}{V \cup \{x\}}{\prop}{\Dfalse}{p}$.
  \proofcomplete
\end{proofof}

\bb{The local dependency analysis is disabled because it gives information
valid only at a certain process occurrence, and here we combine facts
obtained at two occurrences $\pp_1$ and $\pp_2$.}%

\begin{example}\label{exa:running:final}
  Using assumptions on cryptographic primitives,
  the process $Q_0$ of Example~\ref{exa:running} can be transformed
  into the following process $Q''_0$:
{\allowdisplaybreaks\begin{align*}
&Q_0'' = \cinput{\startch}{};
\Res{\vn{k}}{T_k}; \Res{\vn{mk}}{T_{mk}}; \coutput{c}{};(Q''_A \parpop Q''_B)\\
&Q''_A = \repl{i}{n}\cinput{c_A[i]}{};\Res{\vn{k}'}{T_k};\Res{\vn{r}}{T_r};\\*
&\quad \assign{\vn{m}:\bitstring}{\enc'(\Z_k,k,\vn{r}')}\\*
&\quad \coutput{c_A[i]}{\vn{m}, \mac'(\vn{m},\vn{mk})}\\
&Q''_B = \repl{i'}{n}\cinput{c_B[i']}{\vn{m}',\vn{ma}};\\
%\bafind{\vf \leq n}{\vn{m}[\vf], \vn{k}'[\vf]}{\vn{m}[\vf] = \vn{m}' \fand \mverify'(\vn{m}', \vn{mk}, \vn{ma})}\\
&\quad \FIND \ {\vf \leq n}\ \SUCHTHAT \ \defined (\vn{m}[\vf], \vn{k}'[\vf])\fand{}\\*
&\qquad  \vn{m}' = \vn{m}[\vf] \fand \mverify'(\vn{m}', \vn{mk}, \vn{ma}) \THEN \\
&\quad \assign{\vn{k}'' : T_k}{\vn{k}'[\vf]}\coutput{c_B[i']}{}
\end{align*}}%
and $Q_0 \approx^{\vn{k}''}_{p} Q''_0$.
In order to prove the one-session secrecy of $\vn{k}''$, we notice
that $\vn{k}''$ is defined by $\bassign{\vn{k}'' : T_k}{\vn{k}'[\vf]}$,
the only variable access to $\vn{k}'$ in $Q''_0$ is 
$\bassign{\vn{k}'' : T_k}{\vn{k}'[\vf]}$, and $\vn{k}''$ is not used in
$Q''_0$. So by Proposition~\ref{prop:sec}, $Q''_0$ satisfies the one-session secrecy of $\vn{k}''$ without public variables up to probability 0.
(We have $\defrestr_{\pp}(\vn{k}'') = \vn{k}'[\vf]$
and
$\yctran{\noleak(\vn{k}'[\vf], \sri, \fset)} = (\forall\sri, \neg\bigwedge \fset) \vee ((\vn{k}' \notin V) \wedge 
\yctran{\noleak(\vn{k}''[\theta\tup{i}], \sri', \fset')}) = \true$
since
$\vn{k}' \notin V$, $\vn{k}'' \notin V$, and
$\yctran{\noleak(\vn{k}''[\theta\tup{i}], \sri', \fset')} = (\forall\sri', \ab \neg\bigwedge \fset') \vee ((\vn{k}'' \notin V) \wedge \true) = \true$.
So $\Prss{C[Q''_0]}{\neg \yctran{\prove{\secrone(\vn{k}'')}(\sset)}} = 0$.)
By Lemma~\ref{lem:transfersec}, the process $Q_0$ of Example~\ref{exa:running}
also satisfies the one-session secrecy of $\vn{k}''$ without public variables
up to probability $p'(C) = 2p(C[C_{\secrone(\vn{k}'')}[\,]], t_{\sevent})$.
However, this process does not preserve the secrecy of $\vn{k}''$,
because the adversary can force several sessions of $B$ to
use the same key $\vn{k}''$, by replaying the message sent by $A$.
(Accordingly, $\prove{\secr(x)}(\sset)$ is not
satisfied.)
\end{example}

The criteria given in this section might seem restrictive, but in fact,
they should be sufficient for all protocols, provided the previous
transformation steps are powerful enough to transform the protocol
into a simpler protocol, on which these criteria can then be applied.

\subsection{Correspondences}

\subsubsection{Example}\label{sec:exacorresp}

We illustrate the proof of correspondences on the following example, 
inspired by the corrected Woo-Lam public key 
protocol~\cite{Woo97}:
\begin{align*}
&B \rightarrow A: (N, B)\\
&A \rightarrow B: \{ pk_A, B, N \}_{sk_A}
\end{align*}
This protocol is a simple nonce challenge: $B$ sends to $A$ a fresh 
nonce $N$ and its identity. $A$ replies by signing the nonce $N$, $B$'s
identity, and $A$'s public key
(which we use here instead of $A$'s
identity for simplicity: this avoids having to relate identities
and keys; CryptoVerif can obviously also handle the version with
$A$'s identity). 
The signatures are assumed to be (existentially) unforgeable under chosen 
message attacks (UF-CMA)~\cite{Goldwasser88}, so, when $B$ receives the signature, $B$ is convinced
that $A$ is present. The signature cannot be a replay because the
nonce $N$ is signed.

In our calculus, this protocol is encoded by the following process $G_0$,
explained below:
{\allowdisplaybreaks\begin{align*}
G_0 = {}& \cinput{c_0}{}; \Res{rk_A}{\keyseed}; \assign{pk_A}{\pkgen(rk_A)} \\*
&\assign{sk_A}{\skgen(rk_A)} \coutput{c_1}{pk_A}; (Q_A \parpop Q_B)\\
Q_A = {}&\repl{i_A}{n} \cinput{c_2[i_A]}{x_N:\nonce, x_B:\host}; \\*
& \kevent{e_A(pk_A, x_B, x_N)}; \Res{r}{\seed};\\*
& \coutput{c_3[i_A]}{\sign(\concat(pk_A, x_B, x_N), sk_A, r)}\\
Q_B = {}&\repl{i_B}{n} \cinput{c_4[i_B]}{x_{pk_A}:\pkey}; \Res{N}{\nonce};
\\*
& \coutput{c_5[i_B]}{N, B};\cinput{c_6[i_B]}{s:\signature};\\*
& \baguard{\verify(\concat(x_{pk_A}, B, N), x_{pk_A}, s)}\\*
& \baguard{x_{pk_A} = pk_A}
\kevent{e_B(x_{pk_A}, B, N)}
\end{align*}}%
The process $G_0$ is assumed to run in interaction with an adversary,
which also models the network.
$G_0$ first receives an empty message on channel $c_0$, sent by the adversary.
Then, it chooses randomly with uniform probability a bitstring $rk_A$ 
in the type $\keyseed$, by the construct $\Res{rk_A}{\keyseed}$. 
Then, $G_0$ generates the public key $pk_A$ corresponding to the coins $rk_A$,
by calling the public-key generation algorithm $\pkgen$. 
Similarly, $G_0$ generates the secret key $sk_A$ by calling 
$\skgen$. It outputs the public key $pk_A$
on channel $c_1$, so that the adversary has this public key.

After outputting this message, the control passes to the receiving
process, which is part of the adversary. Several processes are then made
available, which represent the roles of $A$ and $B$ in the protocol:
the process $Q_A \parpop Q_B$ is the parallel composition of  
$Q_A$ and $Q_B$; it makes simultaneously available the processes defined 
in $Q_A$ and $Q_B$. 
Let $Q'_A$ and $Q'_B$ be such that $Q_A = \repl{i_A}{n} Q'_A$ and
$Q_B = \repl{i_B}{n} Q'_B$.
The replication $\repl{i_A}{n} Q'_A$ represents $n$ copies of the
process $Q_A'$, indexed by the replication index $i_A$.
The process $Q'_A$
begins with an input on channel $c_2[i_A]$; the channel is indexed
with $i_A$ so that the adversary can choose which copy of
the process $Q'_A$ receives the message by sending it on channel
$c_2[i_A]$ for the appropriate value of $i_A$. The situation is
similar for $Q'_B$, which expects a message on channel $c_4[i_B]$.
The adversary can then run each copy of $Q'_A$ or $Q'_B$
simply by sending a message on the appropriate channel $c_2[i_A]$ or
$c_4[i_B]$.

% explain the beginning of Q'_B
The process $Q'_B$ first expects on channel
$c_4[i_B]$ a message $x_{pk_A}$ in the type $\pkey$ of public keys. This message is not really part of the protocol. It serves
for starting a new session of the protocol, in which $B$ interacts
with the participant of public key $x_{pk_A}$. For starting a session
between $A$ and $B$, this message should be $pk_A$.
Then, $Q'_B$ chooses randomly with uniform probability a nonce $N$ in
the type $\nonce$. 
The type $\nonce$ is \emph{large}: collisions between independent random numbers chosen uniformly
in a large type are eliminated by CryptoVerif.
$Q'_B$ sends the message $(N, B)$ on channel
$c_5[i_B]$.  The control then passes to the receiving process,
included in the adversary. This process is expected to forward this
message $(N, B)$ on channel $c_2[i_A]$, but may proceed differently in
order to mount an attack against the protocol.

%explain Q'_A
Upon receiving a message $(x_N, x_B)$ on channel $c_2[i_A]$, where the
bitstring $x_N$ is in the type $\nonce$ and $x_B$ in the type $\host$,
the process $Q'_A$ executes the event
$e_A(pk_A, \ab x_B, \ab x_N)$. This event does not change the state of the
system. Events just record that a certain program point has been
reached, with certain values of the arguments of the event. 
Then, $Q'_A$ chooses randomly with uniform probability a
bitstring $r$ in the type $\seed$; this random bitstring is next used
as coins for the signature algorithm.  
Finally,
$Q'_A$ outputs the signed message $\{ pk_A, x_B, x_N \}_{sk_A}$.
(The function $\concat$ concatenates its arguments, with information
on the length of these arguments, so that the arguments can be recovered
from the concatenation.)
The control then passes to the receiving process, which should 
forward this message on channel $c_6[i_B]$ if it wishes to run the protocol
correctly.

% explain the end of Q'_B
Upon receiving a message $s$ on $c_6[i_B]$, $Q'_B$ verifies that the
signature $s$ is correct and, if $x_{pk_A} = pk_A$, that is, if $B$ runs a
session with $A$, it executes the event $e_B(x_{pk_A}, B, N)$.  Our
goal is to prove that, if event $e_B$ is executed, then event $e_A$
has also been executed. However, when $B$ runs a session with a
participant other than $A$, it is perfectly correct that $B$
terminates without event $e_A$ being executed; that is why event $e_B$
is executed only when $B$ runs a session with $A$.

By the unforgeability of signatures, the signature verification
with $pk_A$ succeeds only for signatures generated with $sk_A$.
So, when we verify that the signature is correct, we can furthermore
check that it has been generated using $sk_A$.
So, after game transformations explained below, we obtain the following final game:
{\allowdisplaybreaks\begin{align*}
G_1 = {} &\cinput{c_0}{}; \Res{rk_A}{\keyseed}; \\*
&\assign{pk_A}{\pkgen'(rk_A)} \coutput{c_1}{pk_A}; (Q_{1A} \parpop Q_{1B})\\
Q_{1A} = {} &\repl{i_A}{n} \cinput{c_2[i_A]}{x_N:\nonce, x_B:\host};\\*
&\kevent{e_A(pk_A, x_B, x_N)};\\*
&\assign{m}{\concat(pk_A, x_B, x_N)}\\*
&\Res{r}{\seed};\coutput{c_3[i_A]}{\sign'(m, \skgen'(rk_A), r)}\\
Q_{1B} = {} &\repl{i_B}{n} \cinput{c_4[i_B]}{x_{pk_A}:\pkey}; \Res{N}{\nonce};\\*
& \coutput{c_5[i_B]}{N, B};\cinput{c_6[i_B]}{s:\signature};\\*
&\FIND\ u \leq n\ \SUCHTHAT\ \defined(m[u], x_B[u], x_N[u])\\*
& \ {} \wedge (x_{pk_A} = pk_A) \wedge (B = x_B[u])  \wedge (N = x_N[u]) \\*
& \ {} \wedge \verify'(\concat(x_{pk_A}, B, N), x_{pk_A}, s) \THEN\\*
&\kevent{e_B(x_{pk_A}, B, N)})
\end{align*}}%

The assignment $sk_A = \skgen(rk_A)$ has been removed and
$\skgen(rk_A)$ has been substituted for $sk_A$, in order to make
the term $\sign(m, \skgen(rk_A), r)$ appear. This term is needed for
the security of the signature scheme to apply.

In $Q_{1A}$, the signed message is stored in variable $m$, and this variable
is used when computing the signature.

Finally, using the unforgeability of signatures,
the signature verification has been replaced with an array
lookup: the signature verification can succeed only when
$\concat(x_{pk_A}, B, N)$ has been signed with $sk_A$, so we look for
the message $\concat(x_{pk_A}, B, N)$ in the array $m$ and the event
$e_B$ is executed only when this message is found. In other words, we
look for an index $u \leq n$ such that $m[u]$ is defined and $m[u] =
\concat(x_{pk_A}, B, N)$. By definition of $m$, $m[u] = \concat(pk_A,
x_B[u], x_N[u])$, so the equality $m[u] =
\concat(x_{pk_A}, B, N)$ can be replaced with
$(x_{pk_A} = pk_A) \wedge (B = x_B[u]) \wedge (N = x_N[u])$. (Recall
that the result of the $\concat$ function contains enough information
to recover its arguments.)
%DONE (not to do) I did not explain how the test $x_{pk_A} = pk_A$ of the initial
%game is manipulated. Should I simplify the initial game to remove x_{pk_A}?
This transformation replaces the function symbols $\pkgen$, $\skgen$,
$\sign$, and $\verify$ with primed function symbols $\pkgen'$,
$\skgen'$, $\sign'$, and $\verify'$ respectively, to avoid repeated
applications of the unforgeability of signatures with the same key.
(The unforgeability of signatures is applied only to unprimed symbols.)

The soundness of the game transformations shows that $G_0 \approx G_1$.
We will prove that $G_1$ satisfies the correspondences~\eqref{c1}
and~\eqref{c2} with any public variables $V$, in particular with $V =
\emptyset$. By Lemma~\ref{lem:transfercorr}, 
$G_0$ also satisfies these correspondences with public
variables $V=\emptyset$.
Let us sketch how the proof of correspondence~\eqref{c1} for the game $G_1$ 
will proceed. 
Let $Q'_{1A}$ and $Q'_{1B}$ such that $Q_{1A} = \repl{i_A}{n} Q'_{1A}$
and $Q_{1B} = \repl{i_B}{n} Q'_{1B}$.
Assume that event $e_B$ is executed in the copy of $Q'_{1B}$
of index $i_B$, that is, $e_B(x_{pk_A}[i_B], \ab B, \ab N[i_B])$ is executed. (Recall that the variables $x_{pk_A}$, $N$, $u$, \ldots{} are 
implicitly arrays.)
Then the condition of the $\FIND$ above $e_B$ holds, that is,
$m[u[i_B]]$, $x_B[u[i_B]]$, and $x_N[u[i_B]]$ are defined, 
$x_{pk_A}[i_B] = pk_A$, $B = x_B[u[i_B]]$, and $N[i_B] = x_N[u[i_B]]$.
Moreover, since $m[u[i_B]]$ is defined, the assignment that
defines $m$ has been executed in the copy of $Q'_{1A}$ of index $i_A = u[i_B]$.
Then the event $e_A(pk_A, x_B, x_N)$, located above the definition
of $m$, must have been executed in that copy of $Q'_{1A}$, that is,
$e_A(pk_A, \ab x_B[u[i_B]], \ab x_N[u[i_B]])$ has been executed.
The equalities in the condition of the $\FIND$ imply that this
event is also $e_A(x_{pk_A}[i_B], \ab B, \ab N[i_B])$. To sum up,
if $e_B(x_{pk_A}[i_B], \ab B, \ab N[i_B])$ has been executed, then
$e_A(x_{pk_A}[i_B], \ab B, \ab N[i_B])$ has been executed, so
we have the correspondence~\eqref{c1}. 
This reasoning is typical of the way the prover shows correspondences.
In particular, the conditions of array lookups are key in these
proofs, because they allow us to relate values in processes that run
in parallel (here, the processes that represent $A$ and $B$), and
interesting correspondences relate events
that occur in such processes.
Next, we detail and formalize this reasoning, both
for non-injective and injective correspondences.

\subsubsection{Non-unique Events}\label{sec:nuevents}

The only correspondence that involves a non-unique event $e$ is $\fevent{e} \Rightarrow \false$, and it is simply proved by noticing that the event $e$ no longer occurs in the game after the transformation \rn{prove\_unique} (Section~\ref{sec:prove_unique}). Therefore, non-unique events are not concerned by the proofs of Sections~\ref{sec:nicorresp} and~\ref{sec:injcorresp}.

\subsubsection{Non-injective Correspondences}\label{sec:nicorresp}

Intuitively, in order to prove that $Q_0$ satisfies 
a non-injective correspondence 
$\forall \tup{x}:\tup{T}; \psi \Rightarrow \exists \tup{y}:\tup{T'}; \phi$, with $\tup{x} = \fvar(\psi)$ and $\tup{y} = \fvar(\phi)\setminus\fvar(\psi)$, we collect all facts that hold at events in 
$\psi$ and show that these facts imply $\phi$ using the equational prover.

When~\rn{casesInCorresp = false}, CryptoVerif uses $\facts_{\pp}$ to collect these
facts. When \rn{casesInCorresp = true} (the default), it uses $\facts_{\pp,\case}$
for more precision.
In this section, we detail the proof with $\facts_{\pp,\case}$. The usage of
$\facts_{\pp}$ can be considered as using a single case $\case$,
relying on Lemma~\ref{lem:factsP} instead of Lemma~\ref{lem:factsPcases}.
Formally, we collect facts that hold when the event $F$ in $\psi$ has been executed, as follows.

\begin{definition}[$\pp$ executes $F$, $\fset_{F, \pp, \case}$]\label{def:followsfset}
When $F = \fevent{e(M_1, \ab \ldots, \ab M_m)}$ and
$\pptag\kevent{e(M'_1, \ab \ldots, \ab M'_m)}; \dots$ occurs in $Q_0$ or, for $m=0$, $\pptag\keventabort{e}$ or $\pptag\FIND\unique{e}\dots$ occurs in $Q_0$, we say that
\emph{$\pp$ executes $F$}.

If $\pp$ executes $F$ and for all $\pptag\kevent{e(M'_1, \ab \ldots, \ab M'_m)}; \dots$ in $Q_0$, $M'_1$, \dots, $M'_m$ are simple terms, then we define
$\fset^0_{F,\pp,\case} = \facts_{\pp,\case} \cup \{ M'_j = M_j \mid j \leq m \} \cup \{\lppf(\pp, \replidx{\pp})$ if $\pptag\keventabort{e}$ or $\pptag\FIND\unique{e}\dots$ occurs in $Q_0\}$.
If additionally $F$ is not a non-unique event, then we define
$\fset_{F,\pp,\case} = \fset^0_{F,\pp,\case} \cup \futfset_{\pp}$.
\end{definition}
Intuitively, when the event $F$ in $\psi$ has been executed, it has been executed
by some subterm or subprocess of $Q_0$, so there exists a
subterm or subprocess $\pptag\kevent{e(M'_1, \ldots, M'_m)}; \dots$ or, 
for $m=0$, $\pptag\keventabort{e}$ or 
$\pptag\FIND\unique{e}\dots$ in $Q_0$ such that,
the event $e(M'_1, \ldots, M'_m)$ has been executed and it is equal to
the event $F$, hence $M'_j = M_j$ holds for $j \leq m$.
Moreover, since the program point $\pp$, which executes $F$, has been reached,
$\facts_{\pp,\case}$ holds for some case $\case$ (Lemma~\ref{lem:factsPcases}). 
Furthermore, when the event aborts, it is the last step of the trace,
so $\lppf(\pp, \replidx{\pp})$ also holds. Hence $\fset^0_{F, \pp, \case}$ holds
for some case $\case$.
Additionally, assuming we consider traces that do not execute non-unique events,
since the adversary cannot stop execution of the process until the next
output or $\keventabort{e}$, $\futfset_{\pp}$ also holds (Lemma~\ref{lem:factsPfut}), so $\fset_{F, \pp, \case}$ holds for some case $\case$.
This is proved more formally in Lemma~\ref{lem:fsetFP} below.
(The case of $\GET\unique{e}$ is not mentioned in Definition~\ref{def:followsfset} because it is excluded by Property~\ref{prop:notables}.)

We restrict ourselves to the case in which $M'_1$, \dots, $M'_m$ are simple terms
because only simple terms allowed in sets of facts.

Let $\theta$ be a substitution equal to the identity on the variables $\tup{x}$
of $\psi$. This substitution gives values to existentially quantified
variables $\tup{y}$ of $\phi$.
We say that $\fset \shows_{\theta} \phi$ when we can show that $\fset$
implies $\theta\phi$. Formally, we define:
%\begin{definition}[$\fset \shows_{\theta} \phi$]
\begin{tabbing}
%\\
$\fset \shows_{\theta} M$ if and only if
$\fset \cup \{ \fnot \theta M\}$ yields a contradiction\\[1mm]
$\fset \shows_{\theta} \fevent{e(M_1, \ldots, M_m)}$ if and only if there exist \\*
\qquad $M'_0, \ldots, M'_m$ such that $M'_0:\fevent{e(M'_1, \ldots, M'_m)} \in \fset$\\*
\qquad and $\fset \cup \{ \bigvee_{j = 1}^m \theta M_j \neq M'_j \}$ yields a contradiction\\[1mm]
$\fset \shows_{\theta} \phi_1 \wedge \phi_2$ if and only if 
$\fset \shows_{\theta} \phi_1$ and $\fset \shows_{\theta} \phi_2$\\[1mm]
$\fset \shows_{\theta} \phi_1 \vee \phi_2$ if and only if 
$\fset \shows_{\theta} \phi_1$ or $\fset \shows_{\theta} \phi_2$
\end{tabbing}
%\end{definition}
Terms $\theta M$ are proved by contradiction, using the equational
prover. Events $\theta F$ are proved by looking for some event $F'$ in
$\fset$ and showing by contradiction that $\theta F = F'$, 
using the equational prover. 

Let $\corresp = \sem{\forall \tup{x}:\tup{T}; \psi \Rightarrow \exists \tup{y}:\tup{T'}; \phi}$ be a non-injective correspondence that does not use non-unique events, 
with $\psi = F_1 \wedge \ldots \wedge F_m$, $\tup{x} = \fvar(\psi)$, and $\tup{y} = \fvar(\phi)\setminus\fvar(\psi)$. Suppose that, in $Q_0$, the arguments of the events that occur in $\psi$ are always simple terms.
Suppose that, for all $j \leq m$, $\pp_j$ that executes $F_j$ and $\case_j$ is a case for $\fset_{\pp_j, \case_j}$.
For $j \leq m$, let $\theta_j$ be a renaming of $\replidx{\pp_j}$ to fresh replication indices. (The renamings $\theta_j$ have pairwise disjoint images.)
Let $\theta$ be a family parameterized by $\pp_1, \case_1,\dots, \pp_m, \case_m$ of substitutions equal to the identity on $\tup{x}$.
We define $\prove{\corresp}(\theta, \pp_1, \case_1, \dots, \pp_m, \case_m) = (\theta_1 \fset_{F_1,\pp_1,\case_1} \cup \dots \cup \theta_m \fset_{F_m,\pp_m,\case_m} \shows_{\theta(\pp_1, \case_1,\dots, \pp_m, \case_m)} \phi)$.
This function defines the algorithm that we use to prove the correspondence
$\corresp$ assuming for all $j \leq m$, $F_j$ is executed in $\pp_j$
and we are in case $\case_j$.
We also define $\prove{\corresp}(\theta, \sset) = \bigwedge_{(\pp_1,\case_1,\dots,\pp_m,\case_m) \in \sset} \prove{\corresp}(\theta, \ab \pp_1, \ab \case_1, \ab \dots, \ab \pp_m, \ab \case_m)$.

Non-injective correspondences are proved as follows.
\begin{proposition}\label{prop:nicorresp}
Let $\corresp = \sem{\forall \tup{x}:\tup{T}; \psi \Rightarrow \exists \tup{y}:\tup{T'}; \phi}$ be a non-injective correspondence that does not use non-unique events, 
with $\psi = F_1 \wedge \ldots \wedge F_m$, $\tup{x} = \fvar(\psi)$, and $\tup{y} = \fvar(\phi)\setminus\fvar(\psi)$. 
Let $Q_0$ be a process that satisfies Properties~\ref{prop:notables} and~\ref{prop:autosarename}.
Suppose that, in $Q_0$, the arguments of the events that occur in $\psi$ are always simple terms.
Let $\sset = \{ (\pp_1, \case_1, \dots, \pp_m, \case_m) \mid \forall j \leq m, \pp_j$ executes $F_j$ and $\case_j$ is a case for $\fset_{\pp_j, \case_j}\}$.
If there exists a family of substitutions $\theta$ equal to the 
identity on $\tup{x}$ such that
$\prove{\corresp}(\theta, \sset)$
and for all evaluation contexts $C$ acceptable for $Q_0$,
$\Prss{C[Q_0]}{\neg \yctran{\prove{\corresp}(\theta, \sset)}}\leq p(C)$,
then $\bound{Q_0}{V}{\corresp}{\Dfalse}{p}$ for any $V$.
\end{proposition}
Intuitively, when $\psi = F_1 \wedge \ldots \wedge F_m$ holds, 
$\theta_1\fset_{F_1,\pp_1,\case_1} \cup \dots \cup \theta_m\fset_{F_m,\pp_m,\case_m}$ hold
for some $\pp_1$, $\case_1$, \dots, $\pp_m$, $\case_m$.
For some $\theta$ equal to the identity on $\psi$,
$\theta_1\fset_{F_1,\pp_1,\case_1} \cup \dots \cup \theta_m\fset_{F_m,\pp_m,\case_m}$ implies $\theta \phi$,
so $\theta \phi$ holds.
Hence the correspondence is satisfied.
% Intuitively, when $\psi = F_1 \wedge \ldots \wedge F_m$ holds, 
% $\fset_{F_1,P_1} \cup \dots \cup \fset_{F_m,P_m}$ hold,
% so $\theta \phi$ holds for some $\theta$ equal to the identity on $\psi$, 
% hence the correspondence is satisfied.
%
The proof of Proposition~\ref{prop:nicorresp} relies on the following properties and lemmas. We have
{\allowdisplaybreaks\begin{align*}
&\yctran{\fset \shows_{\theta} M} = \forall \tup{z} \in \tup{T}'', \neg \left(\bigwedge\fset \wedge \fnot \theta M\right)\\
  &\textstyle\yctran{\fset \shows_{\theta} \fevent{e(M_1, \ldots, M_m)}} =
  \forall \tup{z} \in \tup{T}'', \neg \left(\bigwedge\fset \wedge \bigvee_{j = 1}^m \theta M_j \neq M'_j\right)\\*
  &\textstyle\qquad \text{for some $M'_0:\fevent{e(M'_1, \ldots, M'_m)} \in \fset$}\\
&\yctran{\fset \shows_{\theta} \phi_1 \wedge \phi_2} =
\yctran{\fset \shows_{\theta} \phi_1} \wedge \yctran{\fset \shows_{\theta} \phi_2}\\[1mm]
&\yctran{\fset \shows_{\theta} \phi_1 \vee \phi_2} = 
\begin{cases}
\yctran{\fset \shows_{\theta} \phi_1}&\text{if }\fset \shows_{\theta} \phi_1\\
\yctran{\fset \shows_{\theta} \phi_2}&\text{otherwise}
\end{cases}
\end{align*}}%
where $\tup{z}$ are the non-process variables in $\fset$, in the image of $\theta$, and in $\tup{x}$, and $\tup{T}''$
are their types.

\begin{lemma}\label{lem:shows}
$\yctran{\fset\shows_{\theta}\phi} \Rightarrow \forall \tup{z} \in \tup{T}'', \neg \left(\bigwedge\fset \wedge \fnot \theta \phi\right)$, where $\tup{z}$ are the non-process variables in $\fset$, in the image of $\theta$, and in $\tup{x}$, and $\tup{T}''$ are their types.
\end{lemma}
\begin{proof}
By induction on $\phi$.
\proofcomplete
\end{proof}

\begin{lemma}\label{lem:fsetFP}
Let $Q_0$ be a process that satisfies Properties~\ref{prop:notables} and~\ref{prop:autosarename}.
  Let $\trace$ be a full trace of $C[Q_0]$. Let $\evseq$ be the sequence of events in the last configuration of $\trace$. Let $F = \fevent{e(\tup{M})}$ where $\tup{M}$ is a tuple of terms and $e$ is an event that does not occur in $C$. Suppose that the arguments of $e$ in $Q_0$ are always simple terms.
  Let $\venv$ be a mapping of the variables of $\tup{M}$ and $\step$ to their values.
  Suppose that $\trace, \venv \vdash F@\step$.

  Then there exist a program point $\pp$ (in $Q_0$) that executes $F$ and a case $\case$ such that, for any $\theta'$ renaming of $\replidx{\pp}$ to fresh replication indices, there exists a mapping $\sigma$ with domain $\theta'\replidx{\pp}$ such that $\evseq(\venv(\step)) = (\pp, \sigma(\theta'\replidx{\pp})):e(\dots)$ and $\trace, \sigma\cup \venv \vdash \theta'\fset^0_{F, \pp, \case}$.
  If additionally, $\trace$ does not execute a non-unique event of $Q_0$, then 
$\trace, \sigma\cup \venv \vdash \theta'\fset_{F, \pp, \case}$.
\end{lemma}
\begin{proof}
Let $E = E_{\trace}$.
Since $\trace, \venv \vdash F@\step$ and the variables of $\tup{M}$ and $\step$ are defined in $\venv$, there exists $\tup{a}$ such that $\venv, \tup{M} \evalterm \tup{a}$ and 
$\evseq(\venv(\step)) = (\pp, \tup{a}_0): e(\tup{a})$ for some $\pp$ and $\tup{a}_0$.
The rule of the semantics that may have added this element to $\evseq$
is \eqref{sem:event}, \eqref{sem:eventabort}, \eqref{sem:ctxevent},
\eqref{sem:finde}, \eqref{sem:find3}, \eqref{sem:gete}, \eqref{sem:get3}, or~\eqref{sem:eventt}.
\begin{itemize}
\item
  In case~\eqref{sem:event}, the initial configuration of rule~\eqref{sem:event} is of the form
  $E_1, \ab (\sigma_1, \ab \pptag\kevent{e(\tup{a})};\ab P), \ab \pset_0, \ab \cset_0, \ab \tblcts_1,\ab \evseq_1$. By Lemma~\ref{lem:start_from_pp}, Property~\ref{start_from_pp1} applied to this configuration, 
%%   the rules that can conclude with a process
%% $\pptag\kevent{e(\tup{M'})};P$ are \eqref{sem:new},
%%   \eqref{sem:let}, \eqref{sem:if1}, \eqref{sem:if2}, \eqref{sem:find1}, \eqref{sem:find2}, \eqref{sem:insert}, \eqref{sem:get1}, \eqref{sem:get2}, \eqref{sem:output}, \eqref{sem:event} and by Corollary~\ref{cor:evalpos}, their target process is a subprocess of $C[Q_0]$ up to renaming of channels; or~\eqref{sem:ctx}, which recursively has an initial configuration of the form $\pptag\kevent{e(\tup{M'})};P$. Hence,
  we have reductions
\begin{align*}
&\conf = E_0, (\sigma_0, \pptag\kevent{e(\tup{M'})};P), \pset_0, \cset_0, \tblcts_0,\evseq_0 \\
&\qquad \red{p_0}{\ix_0} \dots \red{p_1}{\ix_1} E_1, (\sigma_1, \pptag\kevent{e(\tup{a})};P), \pset_0, \cset_0, \tblcts_1,\evseq_1 \\
&\qquad \red{1}{} E_1, (\sigma_1, P), \pset_0, \cset_0, \tblcts_1,(\evseq_1, (\pp, \image(\sigma_1)):e(\tup{a}))
\end{align*}
where $\pptag\kevent{e(\tup{M'})};P$ is a subprocess of $C[Q_0]$ up to renaming of channels, by any number of applications of~\eqref{sem:ctx} and a final application of~\eqref{sem:event}.

\item In case~\eqref{sem:eventabort}, \eqref{sem:find3}, or~\eqref{sem:get3}, the rules that can conclude with a process
  $\pptag P$ with $P = \keventabort{e}$, $P = \FIND\unique{e}\dots$, or $P = \GET\unique{e}\dots$
  are \eqref{sem:new},
\eqref{sem:let}, \eqref{sem:if1}, \eqref{sem:if2}, \eqref{sem:find1}, \eqref{sem:find2}, \eqref{sem:insert}, \eqref{sem:get1}, \eqref{sem:get2}, \eqref{sem:output}, \eqref{sem:event} and by Corollary~\ref{cor:evalpos}, their target process is a subprocess of $C[Q_0]$ up to renaming of channels.
So we have a reduction
\begin{align*}
  &\conf = E_0, (\sigma_0, \pptag P), \pset_0, \cset_0, \tblcts_0,\evseq_0 \\
  &\qquad \red{p}{\ix} E_0, (\sigma_0, \kw{abort}), \pset_0, \cset_0, \tblcts_0,(\evseq_0, (\pp, \image(\sigma_0)):e)
\end{align*}
where $\pptag P$ is a subprocess of $C[Q_0]$ up to renaming of channels, by~\eqref{sem:eventabort}, \eqref{sem:find3}, or~\eqref{sem:get3}.

\item In case~\eqref{sem:eventt}, the initial configuration of the rule~\eqref{sem:eventt} is of the form
$E_1, \ab \sigma_1, \ab \pptag\kevent{e(\tup{a})};N, \ab \tblcts_1,\ab \evseq_1$. By Lemma~\ref{lem:start_from_pp}, Property~\ref{start_from_pp5} applied to this configuration with $l = 0$, 
%%   the rules that can conclude with a term
%% $\pptag\kevent{e(\tup{M'})};N$ are \eqref{sem:newt},
%%   \eqref{sem:lett}, \eqref{sem:ift1}, \eqref{sem:ift2}, \eqref{sem:findt1}, \eqref{sem:findt2}, \eqref{sem:insertt}, \eqref{sem:gett1}, \eqref{sem:gett2}, \eqref{sem:eventt}, \eqref{sem:definedyes} and by Corollary~\ref{cor:evalpos}, their target term is a subterm of $C[Q_0]$; or~\eqref{sem:ctxt}, which recursively has an initial configuration of the form $\pptag\kevent{e(\tup{M'})};N$. Hence,
  we have reductions
\begin{align*}
&\conf = E_0, \sigma_0, \pptag\kevent{e(\tup{M'})};N, \tblcts_0,\evseq_0 \\
&\qquad \red{p_0}{\ix_0} \dots \red{p_1}{\ix_1} E_1, \sigma_1, \pptag\kevent{e(\tup{a})};N, \tblcts_1,\evseq_1 \\
&\qquad \red{1}{} E_1, \sigma_1, N, \tblcts_1, (\evseq_1, (\pp, \image(\sigma_1)):e(\tup{a}))
\end{align*}
where $\pptag\kevent{e(\tup{M'})};N$ is a subterm of $C[Q_0]$, by any number of applications of~\eqref{sem:ctxt} and a final application of~\eqref{sem:eventt}.

\item In case~\eqref{sem:ctxevent}, the only rule that can conclude with a process
$C[\keventabort{(\pp, \tup{a}_0):e}]$ is \eqref{sem:ctx}. (The rules~\eqref{sem:new},
\eqref{sem:let}, \eqref{sem:if1}, \eqref{sem:if2}, \eqref{sem:find1}, \eqref{sem:find2}, \eqref{sem:insert}, \eqref{sem:get1}, \eqref{sem:get2}, \eqref{sem:output}, \eqref{sem:event} cannot conclude with $C[\keventabort{(\pp, \tup{a}_0):e}]$ because, by Corollary~\ref{cor:evalpos}, their target process is a subprocess of $C[Q_0]$ up to renaming of channels, and the initial process $C[Q_0]$ does not contain the abort event value $\keventabort{(\pp, \tup{a}_0):e}$.)
Hence, there is a rule that concludes with a term $\keventabort{(\pp, \tup{a}_0):e}$.

In cases~\eqref{sem:finde} and \eqref{sem:gete}, there is also a rule that concludes with a term $\keventabort{(\pp, \ab \tup{a}_0):e}$.

The only rules that conclude with a term $\keventabort{(\pp, \tup{a}_0):e}$ are
\eqref{sem:findte}, \eqref{sem:findt3}, \eqref{sem:gette}, \eqref{sem:gett3}, \eqref{sem:eventabortt}, and~\eqref{sem:ctxeventt}. (It cannot be~\eqref{sem:newt},
\eqref{sem:lett}, \eqref{sem:ift1}, \eqref{sem:ift2}, \eqref{sem:findt1}, \eqref{sem:findt2}, \eqref{sem:insertt}, \eqref{sem:gett1}, \eqref{sem:gett2}, \eqref{sem:eventt}, \eqref{sem:definedyes} because, by Corollary~\ref{cor:evalpos}, their target term is a subterm of $C[Q_0]$, and the initial process $C[Q_0]$ does not contain the abort event value $\keventabort{(\pp, \tup{a}_0):e}$.)
In cases~\eqref{sem:findte} and \eqref{sem:gette}, there is recursively another rule that concludes with $\keventabort{(\pp, \tup{a}_0):e}$.
In case~\eqref{sem:ctxeventt}, the only rule that can conclude with $C[\keventabort{(\pp, \tup{a}_0):e}]$ is~\eqref{sem:ctxt}, so there is recursively another rule that concludes with $\keventabort{(\pp, \tup{a}_0):e}$.
Therefore, $\keventabort{(\pp, \tup{a}_0):e}$ ultimately comes from an application of~\eqref{sem:eventabortt}, \eqref{sem:findt3}, or~\eqref{sem:gett3}:
\[\conf = E_0, \sigma_0, \pptag N, \tblcts_0,\evseq_0 \red{p}{\ix} E_0, \sigma_0, \keventabort{(\pp, \image(\sigma_0)):e}, \tblcts_0,\evseq_0\]
where $N = \keventabort{e}$, $N = \FIND\unique{e}\dots$, or $N = \GET\unique{e}\dots$
Furthermore, the only rules that can conclude with such a term
$\pptag N$ are \eqref{sem:newt}, \eqref{sem:ift1}, \eqref{sem:ift2}, 
\eqref{sem:lett}, \eqref{sem:findt1}, \eqref{sem:findt2}, \eqref{sem:insertt}, \eqref{sem:gett1}, \eqref{sem:gett2}, \eqref{sem:eventt}, \eqref{sem:definedyes} and, by Corollary~\ref{cor:evalpos}, their target term is a subterm of $C[Q_0]$.

\end{itemize}
In all cases, since $e$ does not occur in $C$, $\pp$ is in fact a program point of $Q_0$.
Therefore, the process or term at program point $\pp$ in $Q_0$ is of the form
$\pptag\kevent{e(\tup{M'})}; \dots$, $\pptag\keventabort{e}$,
$\pptag\FIND\unique{e}\dots$, or $\pptag\GET\unique{e}\dots$.
In cases~\eqref{sem:event} and~\eqref{sem:eventt}, 
we have $\sigma_0 = \sigma_1$ by Lemma~\ref{lem:sem-ext}.
In all cases, $\dom(\sigma_0) = \replidx{\pp}$ are the current replication indices at program point $\pp$ by Lemma~\ref{lem:curidx},
and $\image(\sigma_0) = \tup{a}_0$, so
$\sigma_0 = [\replidx{\pp} \mapsto \tup{a}_0]$.
Moreover, $\tup{M'}$ are simple terms, so their evaluation can be written
$E_0, \{ \replidx{\pp} \mapsto \tup{a}_0 \}, \tup{M'} \evalterm \tup{a}$.
Let $\sigma = \{ \theta'\replidx{\pp} \mapsto \tup{a}_0 \}$.
We have $\evseq(\venv(\step)) = (\pp, \sigma(\theta'\replidx{\pp})):e(\tup{a})$.

We have $E_0, \sigma_0, \tup{M'} \evalterm \tup{a}$ so
$E_0, \sigma, \theta'\tup{M'} \evalterm \tup{a}$.
The environment $E_{\trace}$ extends $E_0$, so
$E_{\trace}, \ab \sigma, \ab \theta'\tup{M'} \evalterm \tup{a}$, so
$\trace, \sigma\cup\venv \vdash \theta'\tup{M'} = \tup{M}$.

The configuration $\conf$ is at program point $\pp$ in $\trace$, so
by Corollary~\ref{cor:factsP}, we have
$\trace, \sigma \vdash \theta'\facts_{\pp,\case}$.
%, and by definition, we have $\trace, \sigma \vdash \ppf(\pp, \theta'\replidx{\pp})$.

When the process or term at $\pp$ is $\pptag\keventabort{e}$, $\pptag\FIND\unique{e}\dots$, or $\pptag\GET\unique{e}\dots$, the environment and replication indices in $\conf$ are the same as at the end of the trace, since the execution after $\conf$ applies \eqref{sem:eventabort}, \eqref{sem:find3}, or \eqref{sem:get3}, which terminate the trace keeping the same environment and replication indices as in $\conf$ or \eqref{sem:eventabortt}, \eqref{sem:findt3}, or \eqref{sem:gett3} which build an abort event value keeping the same environment and replication indices as in $\conf$, followed by some rules among \eqref{sem:findte}, \eqref{sem:gette}, \eqref{sem:ctxt}, \eqref{sem:ctxeventt}, \eqref{sem:finde}, \eqref{sem:gete}, \eqref{sem:ctx}, and~\eqref{sem:ctxevent}, which preserve the environment and replication indices that come with the abort event value.
So $\trace, \sigma \vdash \lppf(\pp, \theta'\replidx{\pp})$.

Therefore, $\pp$ executes $F$ and $\trace, \sigma\cup \venv \vdash \theta'\fset^0_{F, \pp, \case}$ for some $\case$.

Suppose additionally that $\trace$ does not execute any non-unique event of $Q_0$.
Let $\trace''$ be the prefix of $\trace$ that stops at the first evaluated output $\coutput{c[\tup{a'}]}{b};Q$ that follows $\conf$, or $\trace'' = \trace$ if $\trace$ contains no output after $\conf$. 
By Lemma~\ref{lem:factsPfut}, we have
$\trace'' \vdash \futfset_{\pp}$, so $\trace'', \sigma \vdash \theta'\futfset_{\pp}$.
Moreover, $E$ extends $E_{\trace''}$ and by Lemma~\ref{lem:sem-ext}, $\evseq_{\trace}$ extends $\evseq_{\trace''}$\bb{TODO Write a corollary like Corollary~\ref{cor:factsP}?}, so $\trace, \sigma \vdash \theta'\futfset_{\pp}$.
Therefore, $\trace, \sigma\cup \venv \vdash \theta'\fset_{F, \pp, \case}$ for some $\case$.
  \proofcomplete
\end{proof}

\begin{lemma}\label{lem:nicorresp}
Let $\corresp = \sem{\forall \tup{x}:\tup{T}; \psi \Rightarrow \exists \tup{y}:\tup{T'}; \phi}$ be a non-injective correspondence that does not use non-unique events, 
with $\psi = F_1 \wedge \ldots \wedge F_m$, $\tup{x} = \fvar(\psi)$, and $\tup{y} = \fvar(\phi)\setminus\fvar(\psi)$. 
Let $Q_0$ be a process that satisfies Properties~\ref{prop:notables} and~\ref{prop:autosarename}.
Suppose that, in $Q_0$, the arguments of the events that occur in $\psi$ are always simple terms.

Let $\sset = \{ (\pp_1, \case_1, \dots, \pp_m, \case_m) \mid \forall j \leq m, \pp_j$ executes $F_j$ and $\case_j$ is a case for $\fset_{\pp_j, \case_j}\}$.
Let $C$ be an evaluation context acceptable for $Q_0$ with public variables $V$ that does not contain events used by $\corresp$.
Let $\trace$ be a full trace of $C[Q_0]$ that does not execute any non-unique event of $Q_0$.
If $\trace \vdash \neg\corresp$, then 
for any $\theta$ family of substitutions equal to the identity on $\tup{x}$,
$\trace \vdash \neg \yctran{\prove{\corresp}(\theta, \sset)}$.
\end{lemma}
\begin{proof}
Since $\trace \vdash \neg\corresp$, we have $\trace \vdash \exists \tup{x} \in \tup{T}, F_1 \wedge \dots \wedge F_m \wedge \forall \tup{y} \in \tup{T'}, \neg \phi$.
So there exists $\venv$ that maps $\tup{x}$ to elements of $\tup{T}$ such that
$\trace, \venv \vdash F_1 \wedge \dots \wedge F_m \wedge \forall \tup{y} \in \tup{T'}, \neg \phi$.
By Lemma~\ref{lem:fsetFP}, for all $j \leq m$, there exists a program point $\pp_j$ (in $Q_0$) that executes $F_j$ and a case $\case_j$ such that, for any $\theta_j$ renaming of $\replidx{\pp_j}$ to fresh replication indices, there exists a mapping $\sigma_j$ with domain $\theta_j\replidx{\pp_j}$ such that $\trace, \sigma_j\cup \venv \vdash \theta_j\fset_{F_j, \pp_j, \case_j}$.
Since $\trace, \venv \vdash \forall \tup{y} \in \tup{T'}, \neg \phi$, we have
$\trace, \venv \vdash \neg \theta(\pp_1, \case_1,\dots, \pp_m,\case_m) \phi$.
Therefore, $\trace, \sigma_1 \cup \dots \cup \sigma_m \cup \venv \vdash
\theta_1\fset_{F_1,\pp_1,\case_1} \cup \dots \cup \theta_m\fset_{F_m,\pp_m,\case_m} \cup \{ \neg \theta(\pp_1, \case_1,\dots, \pp_m,\case_m) \phi \}$, so
$\trace \vdash \exists \theta_1\replidx{\pp_1}, \dots, \exists \theta_m\replidx{\pp_m}, \exists \tup{x} \in \tup{T}, \bigwedge\fset_0(\pp_1,\case_1, \dots, \pp_m,\case_m) \wedge \neg\theta(\pp_1, \ab \case_1, \ab \dots, \ab \pp_m,\ab \case_m)\phi$ where $\fset_0(\pp_1, \case_1, \dots, \pp_m, \case_m) = \theta_1\fset_{F_1,\pp_1,\case_1} \cup \dots \cup \theta_m\fset_{F_m,\pp_m,\case_m}$.
We have 
\[\begin{split}
&\yctran{\prove{\corresp}(\theta, \pp_1, \case_1, \dots, \pp_m, \case_m)} \\
&\quad\textstyle\Rightarrow \forall \theta_1\replidx{\pp_1}, \dots, \forall \theta_m\replidx{\pp_m}, \forall \tup{x} \in \tup{T}, \neg \left(\bigwedge\fset_0(\pp_1, \case_1, \dots, \pp_m, \case_m) \wedge \neg\theta(\pp_1, \ab \case_1, \ab \dots, \ab \pp_m,\ab \case_m)\phi\right)
\end{split}\]
by Lemma~\ref{lem:shows}. (The non-process variables in $\fset_0(\pp_1, \case_1, \dots, \pp_m, \case_m)$, in the image of $\theta(\pp_1, \ab \case_1, \ab \dots, \ab \pp_m,\ab \case_m)$, and in $\tup{x}$ are in $\theta_1\replidx{\pp_1}$, \dots, $\theta_m\replidx{\pp_m}$, $\tup{x}$.)
So $\trace \vdash \neg \yctran{\prove{\corresp}(\theta, \pp_1, \case_1, \dots, \pp_m, \case_m)}$,
so $\trace \vdash \neg \yctran{\prove{\corresp}(\theta, \sset)}$.
\proofcomplete
\end{proof}

\begin{proofof}{Proposition~\ref{prop:nicorresp}}
Let $C$ be an evaluation context acceptable for $Q_0$ with public variables $V$ that does not contain events used by $\corresp$.
We have
\begin{align*}
\Advtev{Q_0}{\corresp}{C}{\Dfalse} 
& = \Pr[C[Q_0] : \neg \corresp \wedge \neg\nonunique{Q_0}]\\
& \leq \Pr[C[Q_0] : \neg \yctran{\prove{\corresp}(\theta, \sset)}]
\tag*{by Lemma~\ref{lem:nicorresp}}\\
& \leq \Prss{C[Q_0]}{\neg \yctran{\prove{\corresp}(\theta, \sset)}}
\tag*{by Lemma~\ref{lem:Pr-Prss}}\\
& \leq p(C)
\end{align*}
So $\bound{Q_0}{V}{\corresp}{\Dfalse}{p}$.
\proofcomplete
\end{proofof}

\begin{example}\label{exa:ni}
  Let us prove that the example $G_1$ satisfies \eqref{c1}.
  We first study the facts $\facts_{\pp_B}$ that hold at the program point
  $\pp_B$ that executes event $e_B$.
  $\facts_{\pp_B}$ contains $\defined(m[u[i_B]])$, $\defined(x_B[u[i_B]])$,
$\defined(x_N[u[i_B]])$, $x_{pk_A}[i_B] = pk_A$, $B = x_B[u[i_B]]$,
and $N[i_B] = x_N[u[i_B]]$, because the condition of $\FIND$ holds
at $\pp_B$. 
Moreover, we have $\defined(m[u[i_B]]) \ab \in \facts_{\pp_B}$,
and, when $m[i_A]$ is defined, $(\pp_A,i_A):\fevent{e_A(pk_A, \ab x_B[i_A], \ab x_N[i_A])}$
holds, so $(\pp_A,i_A):\fevent{e_A(pk_A, \ab x_B[i_A], \ab x_N[i_A])} \{ u[i_B]/i_A \} \in \facts_{\pp_B}$, that is,
$(\pp_A, u[i_B]):\fevent{e_A(pk_A, \ab x_B[u[i_B]], \ab x_N[u[i_B]])} \in \facts_{\pp_B}$.
In other words, since $m$ is defined at index $u[i_B]$, 
event $e_A$ has been executed in the copy of $Q'_{1A}$ of index
$u[i_B]$.
($\facts_{\pp_B}$ also contains other facts, 
which are useless for proving the desired correspondences,
so we do not list them.)

For $\psi = F = \fevent{e_B(x,y,z)}$, 
$\pp_B$ is the only program point  that executes $F$,
so this event has been executed
in some copy of $Q'_{1B}$ of index $i'_B$, with $x_{pk_A}[i'_B] = x, B = y, N[i'_B] = z$.
Then, when $\psi$ holds, the facts $\theta'\fset_{F,\pp_B} \supseteq \theta'\facts_{\pp_B}
\cup \{ x_{pk_A}[i'_B] = x, B = y, N[i'_B] = z \}$ hold for some value of
$i'_B$, with $\theta' = \{i'_B/i_B\}$. (We consider a single case $\case$ here, so we can
simply omit the case $\case$.)

Furthermore, the substitution
$\theta(\pp_B)$ is the identity since all variables of $\phi$
also occur in $\psi$. Then we just have to show that
$\theta'\fset_{F,\pp_B}$ implies $\phi = \fevent{e_A(x,y,z)}$, that is,
$\theta'\fset_{F,\pp_B} \shows_{\theta(\pp_B)} \ab \fevent{e_A(x,y,z)}$.
Since $(\pp_A,u[i_B]):\fevent{e_A(pk_A, \ab x_B[u[i_B]], \ab x_N[u[i_B]])} \in \facts_{\pp_B}$, 
we have $(\pp_A,\ab u[i_B']):\fevent{e_A(pk_A, \ab x_B[u[i'_B]], \ab x_N[u[i'_B]])} \in \theta'\fset_{F, \pp_B}$,
so the equational prover just has to prove by contradiction
that $e_A(pk_A, \ab x_B[u[i'_B]], \ab x_N[u[i'_B]])  = 
e_A(x,\ab y, \ab z)$, that is, 
$pk_A = x$, $x_B[u[i'_B]] = y$, and $x_N[u[i'_B]] = z$.
The proof succeeds using the following equalities of $\theta'\fset_{F,\pp_B}$:
$x_{pk_A}[i'_B] = x$, $B = y$, $N[i'_B] = z$, 
$x_{pk_A}[i'_B] = pk_A$, $B = x_B[u[i'_B]]$,
and $N[i'_B] = x_N[u[i'_B]]$.

Hence, $G_1$ satisfies~\eqref{c1} with any public variables $V$: if
$\psi = \fevent{e_B(x,y,z)}$ has been executed, then
$\phi = \fevent{e_A(x,y,z)}$ has been executed.
\end{example}

In the implementation, 
the substitution $\theta$ is initially defined as the identity on 
$\tup{x} = \fvar(\psi)$. 
It is defined on other variables when checking $\fset
\shows_{\theta} M$ by trying to find $\theta$ such that $\theta M \in
\fset$, and when checking $\fset \shows_{\theta} \fevent{e(M_1,\ab 
\ldots, \ab 
M_m)}$ by trying to find $\theta$ such that $\theta \fevent{e(M_1,\ab
\ldots, \ab M_m)} \in \fset$.  When we do not
manage to find the image by $\theta$ of all variables of $M$, 
resp. $M_1, \ldots,
M_m$, the check fails. 
When there are several suitable facts $\theta M \in \fset$ or 
$\theta \fevent{e(M_1, \ab \ldots, \ab M_m)} \in \fset$,
the system tries all possibilities.

\subsubsection{Injective Correspondences}\label{sec:injcorresp}

Injective correspondences are more difficult to check than
non-injective ones, because they require distinguishing between
several executions of the same event. We achieve that by relying on
the pair (program points, replication indices) that is recorded in the
sequence $\evseq$ together with each event: distinct executions of
events either occur at different program points or have different
values of replication indices.

We extend Definition~\ref{def:followsfset} to injective events, 
with exactly the same definition as for non-injective events.
% When $F = \fievent{e(M_0, \ldots, M_m)}$ or $F =
% \fevent{e(M_0, \ldots, M_m)}$, $\fset_{F, P}$ is
% defined as for $F = \fevent{e(M_1, \ldots, M_m)}$ in the case of
% non-injective correspondences.

The proof of injective correspondences extends that for non-injective
correspondences: for a correspondence $\forall \tup{x}:\tup{T}; \psi \Rightarrow \exists \tup{y}:\tup{T}'; \phi$, we
additionally prove that distinct executions of the injective
events of $\psi$ correspond to distinct executions of each injective event of
$\phi$, that is, if the injective events of
$\psi$ have different pairs (program point, replication indices), then each injective event of $\phi$ 
has a different pair (program point, replication indices). In order to achieve this
proof, we collect the following information for each injective event of~$\phi$:
\begin{itemize}

\item the set of facts $\fset$ that are known to hold, which will be used 
to reason on replication indices of events; 

\item the program point and replication indices of the considered
  injective event of $\phi$, stored in a pair $M_0$; these program
  point and indices are computed when we prove that this event is
  executed;

\item the program point and replication indices of the injective events of $\psi$,
stored as a mapping $\sri = \{ j \mapsto (\pp_j,\theta_j\replidx{\pp_j}) \mid F_j$
is an injective event$\}$, where $\psi = F_1 \wedge \ldots \wedge F_m$, $\pp_j$ is the program point that executes $F_j$, and $\theta_j$ is a renaming of $\replidx{\pp_j}$ to fresh replication indices, for $j \leq m$;

\item the set $\vset$ containing the replication indices in $\fset$ and 
the variables $\tup{x}$ of $\psi$;
these variables will be renamed to fresh variables in order to
avoid conflicts of variable names between different events.

\end{itemize}
This information is stored in a set $\scoll$, which contains quadruples
$(\fset, M_0, \sri, \vset)$.
We will show that, if the pair (program point, replication indices) of two executions 
of the injective events
of $\psi$ are different, then the pair (program, replication indices) of the
corresponding executions of the considered injective event of $\phi$
are also different.
The equality between pairs (program point, replication indices)
is obviously defined as the equality between program points and 
between replication indices.
Formally, we consider $(\fset, M_0, \sri, \vset)$ and
$(\fset', M_0', \sri', \vset')$ in $\scoll$. We rename the
variables $\vset'$ of the second element to fresh variables by a
substitution $\theta''$ and show that, if $\sri \neq \theta'' \sri'$,
then $M_0 \neq \theta'' M_0'$ (knowing $\fset$ and
$\theta'' \fset'$). This property implies injectivity.

Since this reasoning is done for each injective event in $\phi$,
we collect the associated sets $\scoll$ in a pseudo-formula
$\coll$, obtained by replacing each injective event of $\phi$ with
a set $\scoll$ and all other leaves of $\phi$ with $\bot$.

We say that $\vdash \coll$ when for all non-bottom leaves $\scoll$ of
$\coll$, for all $(\fset, M_0, \sri, \vset)$, $(\fset', M_0',
\sri', \vset')$ in $\scoll$, $\fset \cup \theta'' \fset' \cup \{
\bigvee_{j \in \dom(\sri)} \sri(j) \neq \theta'' \sri'(j), 
M_0  = \theta'' M_0'\}$ yields a contradiction, where 
the substitution $\theta''$ is a renaming of variables in $\vset'$ to 
distinct fresh variables. As explained above, the condition $\vdash \coll$
guarantees injectivity.

We extend the definition of $\fset \shows_{\theta} \phi$ used for
non-injective correspondences to $\fset \shows_{\theta}^{\sri, \vset, \coll} \phi$,
which means that $\fset$ implies $\theta \phi$
and $\coll$ correctly collects the tuples $(\fset, M_0, \sri, \vset)$
associated to this proof.
Formally, we define:
%\begin{definition}[$\fset \shows_{\sri, \vset, \coll} \phi$]
\begin{tabbing}
%\\
$\fset \shows_{\theta}^{\sri,\vset,\bot} M$ if and only if 
$\fset \cup \{ \fnot \theta M\}$ yields a contradiction\\[1mm]
$\fset \shows_{\theta}^{\sri,\vset,\bot} \fevent{e(M_1, \ldots, M_m)}$ if and only if \\*
\qquad there exist $M'_0, \ldots, M'_m$ such that
 $M'_0:\fevent{e(M'_1, \ldots, M'_m)} \in \fset$ and \\*
\qquad $\fset \cup \{  \bigvee_{j = 1}^m \theta M_j \neq M'_j \}$ yields a contradiction\\[1mm]
$\fset \shows_{\theta}^{\sri,\vset,\scoll} \fievent{e(M_1, \ldots, M_m)}$ if and only if \\*
\qquad  there exist $M'_0, \ldots, M'_m$ such that
$M'_0:\fevent{e(M'_1, \ldots, M'_m)} \in \fset$,\\*
\qquad $\fset \cup \{  \bigvee_{j = 1}^m \theta M_j \neq M'_j \}$ yields a contradiction, and $(\fset, M'_0, \sri, \vset) \in \scoll$.\\[1mm]
$\fset \shows_{\theta}^{\sri,\vset,\coll_1 \wedge \coll_2} \phi_1 \wedge \phi_2$ 
if and only if $\fset \shows_{\theta}^{\sri,\vset,\coll_1} \phi_1$ and 
$\fset \shows_{\theta}^{\sri,\vset,\coll_2} \phi_2$\\[1mm]
$\fset \shows_{\theta}^{\sri,\vset,\coll_1 \vee \coll_2} \phi_1 \vee \phi_2$ 
if and only if $\fset \shows_{\theta}^{\sri,\vset,\coll_1} \phi_1$ or 
$\fset \shows_{\theta}^{\sri,\vset,\coll_2} \phi_2$
\end{tabbing}
%\end{definition}
These formulas differ from the non-injective case in that we propagate
$\sri$, $\vset$, $\coll$ and, in the case of injective events, 
we make sure that quadruples $(\fset, M'_0, \sri, \vset)$
are collected correctly by requiring that $(\fset, M'_0, \sri, \vset) 
\in \scoll$.

Let $\corresp = \sem{\forall \tup{x}:\tup{T}; \psi \Rightarrow \exists \tup{y}:\tup{T}'; \phi}$ be a 
correspondence that does not use non-unique events, with $\psi = F_1 \wedge \ldots \wedge F_m$, $\tup{x} = \fvar(\psi)$, and $\tup{y} = \fvar(\phi)\setminus\fvar(\psi)$.
Suppose that, in $Q_0$, the arguments of the events that occur in $\psi$ are always simple terms.
Suppose that, for all $j \leq m$, $\pp_j$ executes $F_j$ and $\case_j$ is a case for $\fset_{\pp_j, \case_j}$.
For $j \leq m$, let $\theta_j$ be a renaming of $\replidx{\pp_j}$ to fresh replication indices. (The renamings $\theta_j$ have pairwise disjoint images.)
Let $\coll$ be a pseudo formula and $\theta$ be a family parameterized by $\pp_1, \case_1,\dots, \case_m, \pp_m$ of substitutions equal to the identity on $\tup{x}$.
We define $\prove{\corresp}(\coll, \theta, \pp_1, \case_1, \dots, \pp_m, \case_m) = (\fset\shows_{\theta(\pp_1, \case_1, \dots, \pp_m, \case_m)}^{\sri, \vset, \coll} \phi)$ where
$\fset = \theta_1\fset_{F_1,\pp_1,\case_1} \cup \dots \cup \theta_m\fset_{F_m,\pp_m,\case_m}$,
$\sri = \{ j \mapsto (\pp_j,\theta_j\replidx{\pp_j}) \mid F_j$ is an injective event$\}$, and
$\vset = \fvar(\theta_1\replidx{\pp_1}) \cup \dots \cup \fvar(\theta_m\replidx{\pp_m}) \cup \{\tup{x}\}$.
The algorithm $\prove{\corresp}(\coll, \theta, \pp_1, \dots, \pp_m)$ shows that the non-injective version of the correspondence $\corresp$ holds assuming the events in $\psi = F_1 \wedge \ldots \wedge F_m$ are executed at program points $\pp_1, \dots, \pp_m$ respectively. Indeed, in this case, the facts $\fset = \theta_1\fset_{F_1,\pp_1} \cup \dots \cup \theta_m\fset_{F_m,\pp_m}$ hold and the formula $\fset\shows_{\theta(\pp_1, \dots, \pp_m)}^{\sri, \vset, \coll} \phi$ shows that this implies $\theta(\pp_1, \dots, \pp_m) \phi$. (The substitution $\theta(\pp_1, \dots, \pp_m) \phi$ determines the values of $\tup{y}$.) Additionally, $\prove{\corresp}(\coll, \theta, \pp_1, \dots, \pp_m)$ makes sure that $\coll$ correctly collects the information needed to prove injectivity.
We also define $\prove{\corresp}(\coll, \theta, \sset) = (\vdash \coll) \wedge \bigwedge_{(\pp_1,\case_1,\dots,\pp_m,\case_m) \in \sset} \prove{\corresp}(\coll, \theta, \pp_1, \case_1, \dots, \pp_m, \case_m)$.
This algorithm proves the correspondence $\corresp$ assuming the events in $\psi$ are executed at program points in $\sset$. It verifies injectivity via $\vdash \coll$.

The following proposition shows the soundness of the proof of injective correspondences based on this algorithm.

\begin{proposition}\label{prop:icorresp}
Let $\corresp = \sem{\forall \tup{x}:\tup{T}; \psi \Rightarrow \exists \tup{y}:\tup{T}'; \phi}$ be a 
correspondence that does not use non-unique events, with $\psi = F_1 \wedge \ldots \wedge F_m$, $\tup{x} = \fvar(\psi)$, and $\tup{y} = \fvar(\phi)\setminus\fvar(\psi)$.
%
%We assume that the variables of $\psi \Rightarrow \phi$ are fresh.
Let $Q_0$ be a process that satisfies Properties~\ref{prop:notables} and~\ref{prop:autosarename}.
Suppose that, in $Q_0$, the arguments of the events that occur in $\psi$ are always simple terms.

Let $\sset = \{ (\pp_1, \case_1, \dots, \pp_m, \case_m) \mid \forall j \leq m, \pp_j$ executes $F_j$ and $\case_j$ is a case for $\fset_{\pp_j, \case_j}\}$.
Assume that there exist a pseudo-formula $\coll$
and a family of substitutions $\theta$ equal to the
identity on $\tup{x}$ such that 
$\prove{\corresp}(\coll, \theta, \sset)$.
Assume that for all evaluation contexts $C$ acceptable for $Q_0$,
$\Prss{C[Q_0]}{\neg \yctran{\prove{\corresp}(\coll, \theta, \sset)}}\leq p(C)$.

Then $\bound{Q_0}{V}{\corresp}{\Dfalse}{p}$ for any $V$.
\end{proposition}
In the implementation, the value of $\coll$ is computed by adding
$(\fset, M'_0, \sri, \vset)$ to $\scoll$
when handling injective events during the checking of 
$\prove{\corresp}(\coll, \theta, \sset)$.
We check $\vdash \coll$ incrementally, after each addition of an
element to $\coll$. 
The proof of Proposition~\ref{prop:icorresp} relies on the following definitions and lemmas.
We have
{\allowdisplaybreaks\begin{align*}
&\yctran{\fset \shows_{\theta}^{\sri,\vset,\bot} M} = 
\forall \tup{z} \in \tup{T}'', \neg \left(\bigwedge \fset \wedge \fnot \theta M\right)\\
&\textstyle\yctran{\fset \shows_{\theta}^{\sri,\vset,\bot} \fevent{e(M_1, \ldots, M_m)}} =
\forall \tup{z} \in \tup{T}'', \neg \left(\bigwedge \fset \wedge \bigvee_{j = 1}^m \theta M_j \neq M'_j \right)\\*
&\qquad \text{for some $M'_0:\fevent{e(M'_1, \ldots, M'_m)} \in \fset$}\\
&\textstyle\yctran{\fset \shows_{\theta}^{\sri,\vset,\scoll} \fievent{e(M_1, \ldots, M_m)}} = 
\forall \tup{z} \in \tup{T}'', \neg \left(\bigwedge \fset \wedge \bigvee_{j = 1}^m \theta M_j \neq M'_j \right)\\*
&\qquad \text{for some $M'_0:\fevent{e(M'_1, \ldots, M'_m)} \in \fset$ and $(\fset, M'_0, \sri, \vset) \in \scoll$}\\
&\yctran{\fset \shows_{\theta}^{\sri,\vset,\coll_1 \wedge \coll_2} \phi_1 \wedge \phi_2} 
= \yctran{\fset \shows_{\theta}^{\sri,\vset,\coll_1} \phi_1} \wedge
\yctran{\fset \shows_{\theta}^{\sri,\vset,\coll_2} \phi_2}\\
&\yctran{\fset \shows_{\theta}^{\sri,\vset,\coll_1 \vee \coll_2} \phi_1 \vee \phi_2} = 
\begin{cases} 
\yctran{\fset \shows_{\theta}^{\sri,\vset,\coll_1} \phi_1}&\text{if }\fset \shows_{\theta}^{\sri,\vset,\coll_1} \phi_1\\
\yctran{\fset \shows_{\theta}^{\sri,\vset,\coll_2} \phi_2}&\text{otherwise}
\end{cases}
\end{align*}}%
where $\tup{z} = \vset$ and $\tup{T}''$ are the types of these variables. 
We also have
\[\begin{split}
\yctran{\vdash\coll} =
&\bigwedge_{\scoll\neq \bot\text{ leaf of }\coll}
\bigwedge_{(\fset, M_0, \sri, \vset)\in\scoll}
\bigwedge_{(\fset', M_0', \sri', \vset')\in\scoll}\forall \tup{z}\in\tup{T}'',\\
&\qquad \neg \left(\bigwedge \fset \wedge \bigwedge \theta'' \fset' \wedge
\left(\bigvee_{j \in \dom(\sri)} \sri(j) \neq \theta'' \sri'(j)\right) \wedge
M_0  = \theta'' M_0'\right)
\end{split}\]
where the substitution $\theta''$ is a renaming of variables in $\vset'$ to 
distinct fresh variables, $\tup{z} = \vset \cup \theta''\vset'$,
and $\tup{T}''$ are the types of these variables.

We define $\formula(\fset \shows_{\theta}^{\sri,\vset,\coll} \phi)$ as follows:
{\allowdisplaybreaks\begin{align*}
&\formula(\fset \shows_{\theta}^{\sri,\vset,\bot} M) = \theta M\\
&\formula(\fset \shows_{\theta}^{\sri,\vset,\bot} \fevent{e(M_1, \ldots, M_m)}) = \theta \fevent{e(M_1, \ldots, M_m)}\\
&\formula(\fset \shows_{\theta}^{\sri,\vset,\scoll} \fievent{e(M_1, \ldots, M_m)}) = {}\\
&\quad \bigvee_{(\fset, N_0, \sri, \vset) \in \scoll} \exists \step \in \mathbb{N}, N_0: \theta \fevent{e(M_1, \ldots, M_m)}@\step \\
&\formula(\fset \shows_{\theta}^{\sri,\vset,\coll_1 \wedge \coll_2} \phi_1 \wedge \phi_2) =
\formula(\fset \shows_{\theta}^{\sri,\vset,\coll_1} \phi_1) \wedge 
\formula(\fset \shows_{\theta}^{\sri,\vset,\coll_2} \phi_2)\\
&\formula(\fset \shows_{\theta}^{\sri,\vset,\coll_1 \vee \coll_2} \phi_1 \vee \phi_2) =
\begin{cases}
\formula(\fset \shows_{\theta}^{\sri,\vset,\coll_1} \phi_1) &\text{if }\fset \shows_{\theta}^{\sri,\vset,\coll_1} \phi_1\\
\formula(\fset \shows_{\theta}^{\sri,\vset,\coll_2} \phi_2) &\text{otherwise}
\end{cases}
\end{align*}}%
The formula $\formula(\fset \shows_{\theta}^{\sri,\vset,\coll} \phi)$
generalizes $\theta \phi$ to the case of injective events.

\begin{lemma}\label{lem:ishows}
$\yctran{\fset \shows_{\theta}^{\sri,\vset,\coll} \phi} \Rightarrow \forall \tup{z} \in \tup{T}'', \neg\left(\bigwedge \fset \wedge \neg \formula(\fset \shows_{\theta}^{\sri,\vset,\coll} \phi) \right)$
where $\tup{z} = \vset$ and $\tup{T}''$ are the types of these variables. 
\end{lemma}
\begin{proof}
  By induction on $\phi$. This result is similar to Lemma~\ref{lem:shows}. The case of injective events is new.
  The case of disjunction differs, but is straightforward by induction hypothesis.
\proofcomplete
\end{proof}

The next lemma shows that, for events $e$ in the considered correspondence,
two distinct executions of event $e$ have distinct pairs (program point, replication indices).
When the term $M$ contains no array accesses, we
define $\sigma(M)$ by $\sigma, M \evalterm \sigma(M)$.

\newcommand{\execevset}{\mathit{Events}}
\begin{lemma}\label{lem:oneevent}
Assume that the event $e$ is used in the correspondence $\corresp$.
Let $Q_0$ be a process that satisfies Property~\ref{prop:notables}.
Let $C$ be an evaluation context acceptable for $Q_0$ with public variables $V$
that does not contain events used by $\corresp$.
If $\initconfig(C[Q_0])  \red{p_1}{\ix_1} \ldots \red{p_{m-1}}{\ix_{m-1}} E, (\sigma, P)\restconfig$,
$\evseq(\step) = (\pp, \tup{a}):e(a_1, \dots, a_m)$,  $\evseq(\step') = (\pp', \tup{a}'):e(a'_1, \dots, a'_m)$, and $\step \neq \step'$,
then $(\pp, \tup{a}) \neq (\pp', \tup{a}')$.
\end{lemma}
\begin{proof}
  Let us fix the event symbol $e$.
  We define the multisets $\execevset(\tup{a}, M)$, %$\execevset(\tup{a}, \pat)$,
  $\execevset(\tup{a}, \ab P)$, and $\execevset(\tup{a}, \ab Q)$ by
  {\allowdisplaybreaks\begin{align*}
&\execevset(\tup{a}, \pptag i) = \emptyset\\
&\execevset(\tup{a}, \pptag x[M_1, \ldots, M_m]) = \bigmultiunion_{j \in \{1, \dots, m\}} \execevset(\tup{a}, M_j)\\ 
&\execevset(\tup{a}, \pptag f(M_1, \ldots, M_m)) = \bigmultiunion_{j \in \{1, \dots, m\}} \execevset(\tup{a}, M_j)\\
&\execevset(\tup{a}, \pptag\Res{x[\tup{i}]}{T}; N) = \execevset(\tup{a}, N)\\
%&\execevset(\tup{a}, \assign{\pat}{M}{N}\ \ELSE N') =  \execevset(\tup{a}, \pat) \multiunion \execevset(\tup{a}, M) \multiunion \max(\execevset(\tup{a}, N), \execevset(\tup{a}, N'))\\
&\execevset(\tup{a}, \pptag\assign{x[\tup{i}]:T}{M}{N}) = \execevset(\tup{a}, M) \multiunion \execevset(\tup{a}, N)\\
&\execevset(\tup{a}, \pptag\bguard{M}{N}{N'}) = \execevset(\tup{a}, M) \multiunion \max(\execevset(\tup{a}, N), \execevset(\tup{a}, N'))\\
%\entry{\cfind{j=1}{m}{
%\vf_{j1}[\tup{i}] \leq n_{j1}, \ldots, \vf_{jm_j}[\tup{i}] \leq n_{jm_j}}{
%M_{j1}, \ldots, M_{jl_j}}{M_j}{P_j}{P}}{}\\
&\execevset(\tup{a}, \pptag\kw{find}\uniqueopt\ (\mathop\bigoplus\nolimits_{j=1}^m
    \tup{\vf_j}[\tup{i}] = \tup{i_j} \leq \tup{n_j}\ \kw{suchthat}\ \kw{defined}(\tup{M_j}) \fand M'_j\ \kw{then}\ N_j)\ \ELSE N') = \\*
&\qquad \bigmultiunion_{j=1}^m \bigmultiunion_{\tup{a_j} \leq \tup{n_j}} \execevset((\tup{a}, \tup{a_j}), M'_j) \multiunion \max(\max_{j = 1}^m \execevset(\tup{a}, N_j), \execevset(\tup{a}, N'))\\
&\execevset(\tup{a}, \pptag\kevent{e'(M_1, \ldots, M_m)};M) = \bigmultiunion_{j \in \{1, \dots, m\}} \execevset(\tup{a}, M_j) \multiunion \execevset(\tup{a}, M) \text{ if $e' \neq e$}\\
&\execevset(\tup{a}, \pptag\kevent{e(M_1, \ldots, M_m)};M) = \{ (\pp,\tup{a}) \} \multiunion \bigmultiunion_{j \in \{1, \dots, m\}} \execevset(\tup{a}, M_j) \multiunion \execevset(\tup{a}, M)\\
&\execevset(\tup{a}, \pptag\keventabort{e'}) = \emptyset\text{ if $e' \neq e$}\\
&\execevset(\tup{a}, \pptag\keventabort{e}) = \{ (\pp,\tup{a}) \}\\[2mm]
%
%&\execevset(\tup{a}, x[\tup{i}]:T) = \emptyset\\
%&\execevset(\tup{a}, f(\pat_1, \ldots, \pat_m)) = \bigmultiunion_{j \in \{1, \dots, m\}} \execevset(\tup{a}, \pat_j)\\
%&\execevset(\tup{a}, {=}M) = \execevset(\tup{a}, M)\\[2mm]
%
&\execevset(\tup{a}, \pptag 0) = \emptyset\\
&\execevset(\tup{a}, \pptag (Q_1 \parpop Q_2)) = \execevset(\tup{a}, Q_1) \multiunion \execevset(\tup{a}, Q_2)\\
&\execevset(\tup{a}, \pptag\repl{i}{n} Q) = \bigmultiunion_{a \in [1, n]} \execevset((\tup{a}, a), Q)\\
&\execevset(\tup{a}, \pptag\Reschan{c};Q) = \execevset(\tup{a}, Q)\\
%&\execevset(\tup{a}, \pptag\cinput{c[M_1, \ldots, M_l]}{\pat};P) = \bigmultiunion_{j \in \{1, \dots, l\}} \execevset(\tup{a}, M_j) \multiunion \execevset(\tup{a}, \pat) \multiunion \execevset(\tup{a}, P)\\[2mm]
&\execevset(\tup{a}, \pptag\cinput{c[M_1, \ldots, M_l]}{x[\tup{i}]:T};P) = \bigmultiunion_{j \in \{1, \dots, l\}} \execevset(\tup{a}, M_j) \multiunion \execevset(\tup{a}, P)\\[2mm]
&\execevset(\tup{a}, \pptag\coutput{c[M_1, \ldots, M_l]}{N};Q) = \bigmultiunion_{j \in \{1, \dots, l\}} \execevset(\tup{a}, M_j) \multiunion \execevset(\tup{a}, N) \multiunion \execevset(\tup{a}, Q)\\
&\execevset(\tup{a}, \pptag\Res{x[\tup{i}]}{T}; P) = \execevset(\tup{a}, P)\\
%&\execevset(\tup{a}, \assign{\pat}{M}{P}\ \ELSE P') = \execevset(\tup{a}, \pat) \multiunion \execevset(\tup{a}, M) \multiunion \max(\execevset(\tup{a}, P), \execevset(\tup{a}, P'))\\
&\execevset(\tup{a}, \pptag\assign{x[\tup{i}]:T}{M}{P}) = \execevset(\tup{a}, M) \multiunion \execevset(\tup{a}, P)\\
&\execevset(\tup{a}, \pptag\bguard{M}{P}{P'}) = \execevset(\tup{a}, M) \multiunion \max(\execevset(\tup{a}, P), \execevset(\tup{a}, P'))\\
&\execevset(\tup{a}, \pptag\FIND\uniqueopt \ (\mathop{\textstyle\bigoplus}\nolimits_{j=1}^{m}
\tup{\vf_j}[\tup{i}] = \tup{i_j} \leq \tup{n_j} \ \SUCHTHAT\ \kw{defined}(\tup{M_j}) \fand M'_j \THEN P_j)\ \ELSE P) =\\*
&\quad \bigmultiunion_{j=1}^m \bigmultiunion_{\tup{a_j} \leq \tup{n_j}} \execevset((\tup{a},\tup{a_j}), M'_j) \multiunion \max(\max_{j = 1}^m \execevset(\tup{a}, P_j), \execevset(\tup{a}, P))\\
&\execevset(\tup{a}, \pptag\kevent{e'(M_0, \ldots, M_m)};P) = \bigmultiunion_{j \in \{1, \dots, m\}} \execevset(\tup{a}, M_j) \multiunion \execevset(\tup{a}, P) \text{ if $e' \neq e$}\\*
&\execevset(\tup{a}, \pptag\kevent{e(M_0, \ldots, M_m)};P) = \{ (\pp, \tup{a}) \} \multiunion \bigmultiunion_{j \in \{1, \dots, m\}} \execevset(\tup{a}, M_j) \multiunion \execevset(\tup{a}, P)\\
&\execevset(\tup{a}, \pptag\keventabort{e'}) = \emptyset\text{ if $e' \neq e$}\\
&\execevset(\tup{a}, \pptag\keventabort{e}) = \{ (\pp,\tup{a}) \}\\
&\execevset(\tup{a}, \pptag\kw{yield}) = \emptyset
\end{align*}}%
($\GET$ and $\INSERT$ are omitted because they do not occur in the game by Property~\ref{prop:notables}.)

  We define the multisets
  {\allowdisplaybreaks\begin{align*}
    \execevset(\evseq) &= \{ (\pp,\tup{a}) \mid (\pp,\tup{a}): e(\dots) \in \evseq \}\\
    \execevset(E,\sigma, M, \tblcts, \evseq) &=
    \execevset(\image(\sigma), M)  \multiunion \execevset(\evseq)\\
    \execevset(E, \pset, \cset) & = \bigmultiunion_{(\sigma', Q') \in \pset} \execevset(\image(\sigma'), Q')\\
    \execevset(E, (\sigma, P)\restconfig) &=
    \execevset(\image(\sigma), P) \multiunion \!\!\bigmultiunion_{(\sigma', Q') \in \pset}\!\! \execevset(\image(\sigma'), Q') \multiunion \execevset(\evseq)
  \end{align*}}%
  The latter multiset contains all pairs $(\pp,\tup{a})$ (program point, value 
of replication indices)
for events $e(\ldots)$ that may be executed in a trace that contains
the configuration $E, (\sigma, P)\restconfig$.

The multiset $\execevset(\sigma_0, C[Q_0])$ contains no duplicates. Indeed,
we show by induction on $M$ that $\execevset(\tup{a}, M)$ is included in the multiset $\{ (\pp, \tup{a}') \}$ where $\pp$ is a program point inside $M$ and $\tup{a}$ is a prefix of $\tup{a}'$, and similarly for $P$ and $Q$. That allows to show that all multiset unions in the computation of $\execevset$ are disjoint unions, since all recursive calls in the computation of $\execevset$ are either with disjoint processes or terms, or with different extensions of $\tup{a}$.

Moreover, by induction on the derivations, if $E,\sigma, M, \tblcts, \evseq \red{p}{\ix} E',\sigma', M', \tblcts', \evseq'$, then $\execevset(E,\sigma, M, \tblcts, \evseq) \supseteq \execevset(E',\ab \sigma', \ab M', \ab \tblcts', \ab \evseq')$; if $E, \pset, \cset \redq E', \pset', \cset'$, then $\execevset(E, \ab \pset, \ab \cset) \supseteq \execevset(E', \pset', \cset')$; and if $\conf \red{p}{\ix} \conf'$, then $\execevset(\conf) \supseteq \execevset(\conf')$.

Therefore, the multiset $\execevset(\emptyset, \{ (\sigma_0, C[Q_0]) \}, \fc(C[Q_0]))$ contains no duplicates, and neither do the multisets $\execevset(\reduce(\emptyset, \{ (\sigma_0, C[Q_0]) \}, \fc(C[Q_0])))$, $\execevset(\initconfig(C[Q_0]))$, and\break $\execevset(\conf)$, where $\conf = E, (\sigma, P)\restconfig$ is the final configuration of the considered trace. Hence, $\execevset(\evseq)$ contains no duplicates, which implies the desired result.
\proofcomplete
\end{proof}

\begin{lemma}\label{lem:icorresp}
Let $\corresp = \sem{\forall \tup{x}:\tup{T}; \psi \Rightarrow \exists \tup{y}:\tup{T}'; \phi}$ be a 
correspondence that does not use non-unique events, with $\psi = F_1 \wedge \ldots \wedge F_m$, $\tup{x} = \fvar(\psi)$, and $\tup{y} = \fvar(\phi)\setminus\fvar(\psi)$.
%
%We assume that the variables of $\psi \Rightarrow \phi$ are fresh.
Let $Q_0$ be a process that satisfies Properties~\ref{prop:notables} and~\ref{prop:autosarename}.
Suppose that, in $Q_0$, the arguments of the events that occur in $\psi$ are always simple terms.

Let $\sset = \{ (\pp_1, \case_1, \dots, \pp_m, \case_m) \mid \forall j \leq m, \pp_j$ executes $F_j$ and $\case_j$ is a case for $\fset_{\pp_j, \case_j}\}$.
Let $C$ be an evaluation context acceptable for $Q_0$ with public variables $V$ that does not contain events used by $\corresp$.
Let $\trace$ be a full trace of $C[Q_0]$ that does not execute any non-unique event of $Q_0$.
If $\trace \vdash \neg\corresp$, then 
for any family of substitutions $\theta$ equal to the 
identity on $\tup{x}$,
for any pseudo-formula $\coll$,
$\trace \vdash \neg \yctran{\prove{\corresp}(\coll, \theta, \sset)}$.
\end{lemma}
\begin{proof}
  By contraposition, we suppose that, $\trace \vdash \yctran{\prove{\corresp}(\coll, \theta, \sset)}$, so
  for every $\pp_1$ that executes $F_1$, \dots,
for every $\pp_m$ that executes $F_m$, for every $\case_1$, \dots, $\case_m$, $\trace \vdash \yctran{\prove{\corresp}(\coll, \theta, \ab \pp_1,\case_1, \ab \dots, \ab \pp_m,\case_m)}$ and $\trace \vdash \yctran{\vdash\coll}$, and we show that $\trace \vdash \corresp$.

Let $\evseq$ be the sequence of events in the last configuration of $\trace$.

We use the notations of Definition~\ref{def:icorresp}. We construct the functions $\Fstepfun_1, \dots, \Fstepfun_k$ as follows. Let $\venv$ be a mapping of $\step_1$, \dots, $\step_m$ to elements of $\mathbb{N}$ and of $\tup{x}$ to elements of $\tup{T}$. Suppose that $\trace,\venv \vdash \psistep$. Then, for all $j \leq m$, $\trace,\venv \vdash F_j^\step$. By Lemma~\ref{lem:fsetFP}, there exists a program point $\pp_j$ (in $Q_0$) that executes $F_j$ and a case $\case_j$ such that, for any $\theta_j$ renaming of $\replidx{\pp_j}$ to fresh replication indices, there exists a mapping $\sigma_j$ with domain $\theta_j\replidx{\pp_j}$ such that $\evseq(\venv(\step_j)) = (\pp_j, \sigma_j(\theta_j\replidx{pp_j})):\dots$ and $\trace, \sigma_j\cup \venv \vdash \theta_j\fset_{F_j, \pp_j, \case_j}$. Let $\venv_1 = \sigma_1 \cup \dots \cup \sigma_m \cup \venv$.
We have $\trace, \venv_1 \vdash \theta_1\fset_{F_1, \pp_1, \case_1} \cup \dots \cup \theta_m\fset_{F_m,\pp_m, \case_m}$.

Let $\fset(\coll, \ab \theta, \pp_1,\case_1, \dots, \pp_m,\case_m) = \theta_1\fset_{F_1,\pp_1,\case_1} \cup \dots \cup \theta_m\fset_{F_m,\pp_m,\case_m} \cup \{ \neg \formula(\prove{\corresp}(\coll, \ab \theta, \ab \pp_1,\case_1, \ab \dots, \ab \pp_m,\case_m))\}$.
By Lemma~\ref{lem:ishows}, 
$\yctran{\prove{\corresp}(\coll, \theta, \pp_1,\case_1, \dots, \pp_m,\case_m)} \Rightarrow \forall \theta_1\replidx{\pp_1}, \ab \dots, \ab \forall \theta_m\replidx{\pp_m}, \ab \forall \tup{x} \in \tup{T}, \ab \neg\bigwedge\fset(\coll, \theta, \ab \pp_1, \case_1, \ab \dots, \ab \pp_m, \ab \case_m)$.
So  $\trace \vdash \forall \theta_1\replidx{\pp_1}, \ab \dots, \ab \forall \theta_m\replidx{\pp_m}, \ab \forall \tup{x} \in \tup{T}, \ab \neg\bigwedge\fset(\coll, \theta, \ab \pp_1, \case_1, \ab \dots, \ab \pp_m, \case_m)$.

Then $\trace, \venv_1 \vdash \neg\bigwedge\fset(\coll, \theta, \pp_1, \case_1, \dots, \pp_m, \case_m)$,
so $\trace, \venv_1 \vdash \formula(\prove{\corresp}(\coll, \theta, \ab \pp_1,\case_1, \ab \dots, \ab \pp_m,\case_m))$, that is, 
$\trace, \venv_1 \vdash \formula(\fset \shows_{\theta}^{\sri,\vset,\coll} \phi)$, with $\fset = \theta_1\fset_{F_1,\pp_1, \case_1} \cup \dots \cup \theta_m\fset_{F_m,\pp_m, \case_m}$, 
$\sri = \{ j \mapsto (\pp_j,\theta_j\replidx{\pp_j}) \mid F_j$ is an injective event$\}$, 
$\vset = \fvar(\theta_1\replidx{\pp_1}) \cup \dots \cup \fvar(\theta_m\replidx{\pp_m}) \cup \{\tup{x}\}$,
and $\theta = \theta(\pp_1, \case_1, \dots, \pp_m, \case_m)$.

Consider an injective event in $\phi$, associated to function $\Fstepfun_l$. 
\begin{itemize}

\item If that injective event corresponds to 
\[\bigvee_{(\fset, N_0, \sri, \vset) \in \scoll} \exists \step \in \mathbb{N}, N_0:\theta \fevent{e(M_1, \ldots, M_m)}@\step\]
in $\formula(\fset \shows_{\theta}^{\sri,\vset,\coll} \phi)$,
we have 
\[\trace, \venv_1 \vdash \bigvee_{(\fset, N_0, \sri, \vset) \in \scoll} \exists \step \in \mathbb{N}, N_0: \theta \fevent{e(M_1, \ldots, M_m)}@\step\]
since $\formula(\fset \shows_{\theta}^{\sri,\vset,\coll} \phi)$ is a conjunction.
So $\trace, \venv_1[\step \mapsto a] \vdash N_0 : \theta \fevent{e(M_1, \ab \ldots, \ab M_m)}@\step$ for some $a \in \mathbb{N}$ and some $N_0$ such that $(\fset, N_0, \sri, \vset) \in \scoll$.
We define $\Fstepfun_l(\venv(\step_1),\ab \dots,\ab \venv(\step_m), \ab \venv(\tup{x})) = a$, 
so that $\trace, \venv_1\vdash N_0: \theta \fevent{e(M_1, \ab \ldots, \ab M_m)}@\Fstepfun_l(\step_1,\ab \dots,\ab \step_m, \ab \tup{x})$.

Moreover, if $j \in I$, then $F_j$ is an injective event, $F_j = \fievent{e_j(M_{j,1}, \dots, M_{j,m})}$.
Moreover, $\trace, \venv \vdash \psistep$, so $\trace, \venv \vdash F_j^{\step}$,
so $\trace, \venv \vdash \fevent{e_j(M_{j,1}, \dots, M_{j,m})}@\step_j$.
Since $\evseq(\venv(\step_j)) = (pp_j, \sigma_j(\theta_j\replidx{\pp_j}):\dots = (\pp_j, \venv_1(\theta_j\replidx{\pp_j})):\dots$ and $\sri(j) = (\pp_j, \theta_j\replidx{\pp_j})$, we have $\trace, \venv_1 \vdash \sri(j):\fevent{e_j(M_{j,1}, \dots, M_{j,m})}@\step_j$.

\item If that injective event is in a removed disjunct in $\formula(\fset \shows_{\theta}^{\sri,\vset,\coll} \phi)$, 
then we define $\Fstepfun_l(\venv(\step_1),\dots,\venv(\step_m), \venv(\tup{x})) = \bot$.

\end{itemize}
Then we have $\trace, \venv_1 \vdash \theta \phistep$, so $\trace, \venv_1 \vdash \exists \tup{y}\in\tup{T}', \phistep$.

Hence, applying this construction for all $\venv$, 
we obtain $\trace \vdash \forall \step_1, \dots, \step_m\in \mathbb{N}, \forall \tup{x}\in\tup{T}, (\psistep \Rightarrow \exists \tup{y}\in\tup{T}', \phistep)$. It remains to show $\Inj(I,\Fstepfun_l)$ for each $l \in \{1, \dots, k\}$.

Suppose $\Fstepfun_l(a_1, \dots, a_m, \tup{a}) = \Fstepfun_l(a'_1, \dots, a'_m, \tup{a}') \neq \bot$.
Let $\venv = \{\step_1 \mapsto a_1, \dots, \step_m \mapsto a_m, \tup{x} \mapsto \tup{a}\}$.
Let $\scoll$ be the leaf of $\coll$ corresponding to the event associated to $\Fstepfun_l$.
By the construction above, we have $\theta$, $(\fset, N_0, \sri, \vset) \in \scoll$, and an extension $\venv_1$ of $\venv$ such that
\begin{align}
  &\trace, \venv_1\vdash N_0: \theta \fevent{e(M_1, \ab \ldots, \ab M_m)}@\Fstepfun_l(\step_1,\ab \dots,\ab \step_m, \ab \tup{x})\label{eq:nic5}\\
  &\trace, \venv_1 \vdash \fset\label{eq:nic6}\\
  &\text{for $j \in I$, }\trace, \venv_1 \vdash \sri(j):\fevent{e_j(M_{j,1}, \dots, M_{j,m})}@\step_j\label{eq:nic3}
\end{align}

Let $\venv' = \{\step_1 \mapsto a'_1, \dots, \step_m \mapsto a'_m, \tup{x} \mapsto \tup{a'}\}$.
In the same way, we have $\theta'$, $(\fset', N'_0, \sri', \vset') \in \scoll$, and an extension $\venv'_1$ of $\venv'$ such that
\begin{align*}
  &\trace, \venv'_1\vdash N'_0: \theta' \fevent{e(M_1, \ab \ldots, \ab M_m)}@\Fstepfun_l(\step_1,\ab \dots,\ab \step_m, \ab \tup{x})\\
  &\trace, \venv'_1 \vdash \fset'\\
  &\text{for $j \in I$, }\trace, \venv'_1 \vdash \sri'(j):\fevent{e_j(M_{j,1}, \dots, M_{j,m})}@\step_j\,.
\end{align*}

Let $\theta''$ be a renaming of the domain of $\venv'_1$ to fresh variables.
We have
\begin{align}
  &\trace, \venv'_1\theta''^{-1}\vdash \theta''N'_0 : \theta''\theta' \fevent{e(M_1, \ab \ldots, \ab M_m)}@\Fstepfun_l(\step_1,\ab \dots,\ab \step_m, \ab \tup{x})\label{eq:nic8}\\
  &\trace, \venv'_1\theta''^{-1} \vdash \theta''\fset'\label{eq:nic9}\\
  &\text{for $j \in I$, }\trace, \venv'_1\theta''^{-1} \vdash \theta''\fevent{e_j(\sri'(j), M_{j,1}, \dots, M_{j,m})}@\step_j\,.\label{eq:nic4}
\end{align}

Therefore, by~\eqref{eq:nic5} and~\eqref{eq:nic8},
\begin{align*}
  &\trace, \venv_1 \cup \venv'_1\theta''^{-1}\vdash \theta N_0: \fevent{e(M_1, \ab \ldots, \ab M_m)}@\Fstepfun_l(a_1,\ab \dots,\ab a_m, \ab \tup{a})\\
  &\trace, \venv_1 \cup \venv'_1\theta''^{-1}\vdash \theta''N'_0 : \theta''\theta' \fevent{e(M_1, \ab \ldots, \ab M_m)}@\Fstepfun_l(a'_1,\ab \dots,\ab a'_m, \ab \tup{a'})\,.
\end{align*}
Since $\Fstepfun_l(a_1,\ab \dots,\ab a_m, \ab \tup{a}) = \Fstepfun_l(a'_1,\ab \dots,\ab a'_m, \ab \tup{a'})$, the events are the same, so
\begin{equation}
  \trace, \venv_1 \cup \venv'_1\theta''^{-1}\vdash N_0 = \theta'' N'_0\,.
  \label{eq:nic1}
\end{equation}
We also have by~\eqref{eq:nic6} and~\eqref{eq:nic9},
\begin{equation}
  \trace, \venv_1 \cup \venv'_1\theta''^{-1}\vdash \fset \cup \theta''\fset'\,.\label{eq:nic2}
\end{equation}
Since $\trace \vdash \yctran{\vdash\coll}$, we have
\[\trace, \venv_1 \cup \venv'_1\theta''^{-1} \vdash 
\neg \left(\bigwedge\fset \wedge \bigwedge \theta'' \fset' \wedge
\left(\bigvee_{j \in \dom(\sri)} \sri(j) \neq \theta'' \sri'(j)\right) \wedge
N_0  = \theta'' N_0'\right)\]
so using~\eqref{eq:nic2} and~\eqref{eq:nic1}, we conclude that
\[\trace, \venv_1 \cup \venv'_1\theta''^{-1} \vdash \bigwedge_{j \in I} \sri(j) = \theta'' \sri'(j)\]
since $\dom(\sri) = I$.
Let $\evseq$ be the sequence of events at the end of $\trace$.
For $j \in I$, let $b_j = \venv_1(\sri(j)) = \venv'_1(\sri'(j))$.
By~\eqref{eq:nic3} and~\eqref{eq:nic4}, we have $\evseq(a_j) =  b_j:e_j(\dots)$ and $\evseq(a'_j) =  b_j:e_j(\dots)$.
By Lemma~\ref{lem:oneevent}, we have $a_j = a'_j$.
That proves $\Inj(I, \Fstepfun_l)$, and concludes the proof that
$\trace \vdash \corresp$.
\proofcomplete
\end{proof}

\begin{proofof}{Proposition~\ref{prop:icorresp}}
Let $C$ be an evaluation context acceptable for $Q_0$ with public variables $V$ that does not contain events used by $\corresp$.
We have
\begin{align*}
\Advtev{Q_0}{\corresp}{C}{\Dfalse} 
& = \Pr[C[Q_0] : \neg \corresp \wedge \neg\nonunique{Q_0}]\\
& \leq \Pr[C[Q_0] : \neg \yctran{\prove{\corresp}(\coll, \theta, \sset)}]
\tag*{by Lemma~\ref{lem:icorresp}}\\
& \leq \Prss{C[Q_0]}{\neg \yctran{\prove{\corresp}(\coll, \theta, \sset)}}
\tag*{by Lemma~\ref{lem:Pr-Prss}}\\
& \leq p(C)
\end{align*}
So $\bound{Q_0}{V}{\corresp}{\Dfalse}{p}$.
\proofcomplete
\end{proofof}

\begin{example}\label{exa:i}
Let us prove that the example $G_1$ satisfies~\eqref{c2}.
We prove the correspondence
$(\forall \tup{x}:\tup{T}; \psi \Rightarrow \exists \tup{y}:\tup{T}'; \phi) = 
(\forall x: \pkey, y: \host, z: \nonce; \fievent{e_B(x, y, z)} \Rightarrow \fievent{e_A(x, \ab y, \ab z)})$.
The program point
$\pp_B$ executes $F = \fevent{e_B(x, y, z)}$ and $\fset = \theta'\fset_{F,\pp_B} \supseteq
\theta'\facts_{\pp_B}\{i'_B/i_B\} \cup \{ x_{pk_A}[i'_B] = x, B = y,  N[i'_B] = z \}$ with $\theta' = \{i'_B/i_B\}$.

As in the proof of $\fset \shows_{\theta(\pp_B)} \fevent{e_A(x, y, z)}$ 
in Example~\ref{exa:ni}, we show $\fset
\shows_{\theta(\pp_B)}^{\sri, \vset, \coll} \fevent{e_A(x, \ab y, \ab z)}$
where $\sri = \{ 1 \mapsto (\pp_B,i_B') \}$ encodes the program points and replication indices
of the events of $\psi$,
$\vset = \{ i_B', x, y, z \}$ contains the replication indices of
$\fset$ and the variables of $\psi$,
$\coll = \scoll = \{ (\fset, (\pp_A, u[i_B']), \sri, \vset) \}$.
($\coll = \scoll$ because the formula $\psi$ is reduced to a single event;
$M'_0 = (\pp_A, u[i_B'])$ contains the program point and replication indices of the event
$e_A$ contained in $\fset$: $(\pp_A, u[i_B']):\fevent{e_A(pk_A, x_B[u[i_B']], 
x_N[u[i_B']])} \in \fset$.)

In order to prove injectivity, it remains to show that $\vdash \coll$.
Let $\theta'' = \{ i_B''/i_B', \ab x''/x, \ab y''/y, \ab z''/z\}$.
We need to show that $\fset \cup \theta''\fset \cup \{ (\pp_B, i_B') \neq (\pp_B, i_B''),
(\pp_A, u[i_B']) = (\pp_A, u[i_B'']) \}$ yields a contradiction, that is, if the pairs (program point, replication
indices) of the event $e_B$ in $\psi$ are distinct ($(\pp_B, i_B') \neq (\pp_B, i_B'')$),
then the pairs (program point, replication indices) of the event $e_A$ in $\phi$
are also distinct ($(\pp_A, u[i_B']) \neq (\pp_A, u[i_B''])$).

$\fset$ contains $N[i_B'] = x_N[u[i_B']]$, so $\theta''\fset$
contains $N[i_B''] = x_N[u[i_B'']]$. These two equalities
combined with $u[i_B'] = u[i_B'']$ imply that $N[i_B'] = x_N[u[i_B']]
= x_N[u[i_B'']] = N[i_B'']$. Since $N$ is defined by random choices
of the large type $\nonce$, $N[i_B'] = N[i_B'']$ implies
$i_B' = i_B''$ up to probability $n^2/2|\nonce|$,
by eliminating collisions.
This equality contradicts $(\pp_B, i_B') \neq (\pp_B, i_B'')$, so we obtain
the desired injectivity.
Therefore, the game $G_1$ satisfies~\eqref{c2} with any 
public variables $V$ up to probability $n^2/2|\nonce|$.
\end{example}

\section{Game Transformations}\label{sec:gametransf}

\subsection{Syntactic Game Transformations}\label{sec:synttransfo}

\subsubsection{\rn{auto\_SArename}}\label{sec:autosarename}

The transformation \rn{auto\_SArename} renames all variables defined
in conditions of $\FIND$, so that they have distinct names. This
transformation is a particular case of \rn{SArename}
(Section~\ref{sec:SArename}) that is particularly simple because these
variables do not have array accesses by Invariant~\ref{invfc}.

\begin{lemma}
  The transformation \rn{auto\_SArename} requires and preserves
  Properties~\ref{prop:nointervaltypes}, \ref{prop:noreschan},
  \ref{prop:channelindices}, and~\ref{prop:notables}.
  It preserves Property~\ref{prop:expand}.
  If transformation \rn{auto\_SArename} transforms $G$ into $G'$, then
  $\dset,\dsetsnu : G, \ab D, \ab \usedevents \ab \indistev{V}{0} G', D,
  \usedevents$ and $G'$ satisfies Property~\ref{prop:autosarename}.
\end{lemma}

\subsubsection[\rn{expand\_tables}]{\rn{expand\_tables} \cite{Cade13b}}\label{sec:transftables}

The transformation $\rn{expand\_tables}$ transforms $\kw{insert}$ and $\kw{get}$ into $\kw{find}$, since the other transformations do not support tables.
It proceeds by storing the inserted list elements in
fresh array variables, and looking up in these arrays instead of performing
$\kw{get}$. More precisely, when $\INSERT\
 \tbl(M_1,\ldots,M_k);P$ is under the replications $\repl{i_1}{n_1} \ldots \repl{i_l}{n_l}$, it is transformed into
\[\assign{y_1[i_1, \ldots, i_l]}{M_1}\ldots \assign{y_k[i_1, \ldots, i_l]}{M_k} P \]
where $y_1, \ldots, y_k$ are fresh array variables, and we add
$(y_1, \ldots, y_k; i_1 \leq n_1, \ldots, i_l \leq n_l)$ in a set $S'$,
to remember them. The construct 
$\GET\uniqueopt\ \tbl(x_1:T_1,\ldots,x_k:T_k)\ \SUCHTHAT\ M\ \IN\ P\ \ELSE P'$ is then transformed into
\[\begin{split}
&\kw{find}\uniqueopt \left(
\!\!\bigoplus_{\begin{array}{l}
\scriptstyle (y_1, \ldots, y_k; i_1 \leq n_1, 
\scriptstyle \ldots, i_l \leq n_l) \in S'
\end{array}}\!\!
\begin{array}{l}
u_1 = i'_1 \leq n_1, \ldots, u_l = i'_l \leq n_l\ \SUCHTHAT\\ 
\defined(y_1[\tup{i}'], \ldots, y_k[\tup{i}']) \fand {}M\{ y_1[\tup{i}'] / x_1, \ldots, y_k[\tup{i}']/ x_k \}\\
\kw{then}\ \assign{x_1}{y_1[\tup{u}]} \ldots\ \assign{x_k}{y_k[\tup{u}]} P
\end{array}\!\!\!\!\right)\\
&\qquad \ELSE P'
\end{split}\]
where $\uniqueopt$ is either $\unique{e}$ or empty and has the same value
at both occurrences, $\tup{u}$ stands for $u_1, \ldots, u_l$, and $\tup{i}'$ stands for $i'_1, \dots, i'_l$. This construct looks in
all arrays used for translating insertion in table $\tbl$, for indices
$\tup{i}'$ such that $y_1[\tup{i}'], \ldots, y_k[\tup{i}']$ are defined, that is,
an element has been inserted at indices $\tup{i}'$, and 
$M\{ y_1[\tup{i}'] / x_1, \ldots, y_k[\tup{i}']/ x_k \}$ is true,
that is, that element satisfies $M$. When it finds such an element,
it stores it in $x_1, \ldots x_k$, and runs $P$.
(When it finds several elements, one of them is chosen randomly with
uniform probability when $\uniqueopt$ is empty and the non-unique event $e$
is raised when $\uniqueopt$ is $\unique{e}$.)
When it finds no element, it executes $P'$.
These transformations are described for processes, but similar 
transformations are performed for $\INSERT$ and $\GET$ terms.

After this transformation, \rn{expand\_tables} calls
\rn{auto\_SArename} to guarantee Property~\ref{prop:autosarename}.

\begin{lemma}
  The transformation \rn{expand\_tables} requires and preserves
  Properties~\ref{prop:nointervaltypes}, \ref{prop:noreschan}, and~\ref{prop:channelindices}.
  It preserves Property~\ref{prop:expand}.
  If transformation \rn{expand\_tables} transforms $G$ into $G'$, then
  $\dset,\dsetsnu : G, \ab D, \ab \usedevents \ab \indistev{V}{p} G', D, \usedevents$,
  where $p$ is $\epsilon_{\FIND}$ times the maximal number of executions of $\GET$ in $G$
  that is not obviously unique (that is, such that the $\INSERT$ in that table may be executed several times),
  and $G'$ satisfies Properties~\ref{prop:notables} and~\ref{prop:autosarename}.
\end{lemma}

\subsubsection{\rn{expand}}\label{sec:transfexpand}

The transformation \rn{expand} transforms terms $\NEW$, $\LET$, $\kw{if}$, $\FIND$, 
$\kw{event}$, and $\kw{event\string_abort}$ into processes, so
that Property~\ref{prop:expand} is guaranteed.
It simplifies the generated game on the fly, using many of the rules
of \rn{simplify} (Section~\ref{sec:simplify}), to avoid generating branches
that can actually not be executed.

\bb{TODO more details would be welcome}%

After this transformation, \rn{expand} calls
\rn{auto\_SArename} to guarantee Property~\ref{prop:autosarename}.

\begin{lemma}
  The transformation \rn{expand} requires and preserves
  Properties~\ref{prop:nointervaltypes}, \ref{prop:noreschan}, \ref{prop:channelindices}, \ref{prop:notables}, and~\ref{prop:autosarename}.
  If transformation \rn{expand} transforms $G$ into $G'$, then
  $\dset,\dsetsnu : G, \ab D, \ab \usedevents \ab \indistev{V}{p} G', D, \usedevents$,
  where $p$ is an upper bound of the probability that the simplification steps modify the execution,
  and $G'$ satisfies Property~\ref{prop:expand}.
\end{lemma}

\subsubsection{\rn{prove\_unique}}\label{sec:prove_unique}

The transformation \rn{prove\_unique} tries to prove that each $\FIND\unique{e}$
really has a unique possibility at runtime (up to a small probability),
so that event $e$ is executed with at most that probability.
More precisely, $\FIND\unique{}$ are already proved;
$\FIND\unique{e}$ for which no query $\fevent{e} \Rightarrow \false$ is active
are also considered as already proved (they will be proved elsewhere; with the notations
of Definition~\ref{def:indistev}, $e$ does not occur in $D_1 \vee D$, so $e$ occurs in $\nonunique{Q,D_1 \vee D}$), and they
are replaced with $\FIND\unique{}$. That corresponds to renaming event $e$
to a special non-unique event that is always in $\nonunique{Q,D_1 \vee D}$.
It remains to prove $\FIND\unique{e}$ for which a query $\fevent{e} \Rightarrow \false$ is active.

Suppose that $P_0 = \FIND\unique{e}\ (\mathop\bigoplus\nolimits_{j=1}^m \tup{\vf_{j}}[\tup{i}] = \tup{i_{j}} \leq \tup{n_{j}}$ $\SUCHTHAT$ $\defined (M_{j1}, \ldots, M_{jl_j}) \fand M_j$ $\kw{then}$ $P_j)$ $\kw{else}$ $P$ is such a $\FIND\unique{e}$.
(The same transformation is performed for $\FIND\unique{e}$ terms.)
CryptoVerif proves uniqueness by proving 
\begin{itemize}
\item that we obtain a contradiction if the condition of a certain branch
  holds for two different values of the indices $\tup{i_j}$,
  that is, for all $j \in \{1, \dots, m\}$, $\fset_{P_0} \cup \{ \defined (M_{j1}), \ab \ldots, \ab \defined(M_{jl_j}), \ab M_j, \ab \defined (\theta M_{j1}), \ab \ldots, \ab \defined(\theta M_{jl_j}), \ab \theta M_j, \ab \tup{i_{j}} \neq \theta \tup{i_{j}}\}$ yields a contradiction, where the substitution $\theta$ maps $\tup{i_{j}}$ to fresh replication indices;
\item and that we obtain a contradiction if the conditions of two different
branches hold simultaneously, that is, for all $j, j' \in \{1, \dots, m\}$ with $j < j'$, $\fset_{P_0} \cup \{ \defined (M_{j1}), \ab \ldots, \ab \defined(M_{jl_j}), \ab M_j, \ab \defined (\theta M_{j'1}), \ab \ldots, \ab \defined(\theta M_{j'l_{j'}}), \ab \theta M_{j'}\}$ yields a contradiction, where the substitution $\theta$ maps $\tup{i_{j'}}$ to fresh replication indices. (The substitution $\theta$ is useful in case the same replication indices are used in both branches $j$ and $j'$.)
\end{itemize}
When uniqueness is proved, $\FIND\unique{e}$ is replaced with $\FIND\unique{}$.
A subsequent call to \rn{success} (Section~\ref{sec:success})
will remove the query $\fevent{e} \Rightarrow \false$
when event $e$ no longer occurs in the game.

\begin{lemma}
  The transformation \rn{prove\_unique} requires and preserves
  Properties~\ref{prop:nointervaltypes}, \ref{prop:noreschan}, \ref{prop:channelindices}, \ref{prop:notables}, and~\ref{prop:autosarename}. It preserves Property~\ref{prop:expand}.
  If transformation \rn{prove\_unique} transforms $G$ into $G'$, then
  $\dset,\dsetsnu : G, \ab D, \ab \usedevents \ab \indistev{V}{p} G', D, \usedevents$,
  where $p$ is an upper bound of the probability that some $\FIND\unique{e}$ proved in the transformation actually executes event $e$.
\end{lemma}

\subsubsection[\rn{remove\_assign}]{\rn{remove\_assign} \cite{Blanchet07c}}

The transformation
\rn{remove\_assign} applied to an assignment
$\assign{x[i_1, \ldots, i_l]:T}{M}{P}$
replaces $x$ with its value $M$.
(The same transformation is performed for assignment terms.)
Precisely, the transformation is performed only when $x$ does not
occur in $M$ (non-cyclic assignment) and $M$ contains only variables,
function applications, and tests (otherwise, copying the definition
of $x$ may break the invariant that each variable is assigned at most once).
When $x$ has several distinct
definitions, we simply replace $x[i_1, \ldots, i_l]$ with $M$ in
$P$. (For accesses to $x$ guarded by $\FIND $, we do not know
which definition of $x$ is actually used.)
When $x$ has a single definition or several identical definitions,
we replace everywhere in the
game $x[M_1, \ldots, M_l]$ with $M\{ M_1 / i_1, \ldots, M_l / i_l\}$. 
We additionally update the $\defined$ conditions of $\FIND $ to
preserve Invariant~\ref{inv2} and to make sure that, if
a condition of $\FIND$ guarantees that $x[M_1, \ldots,
M_l]$ is defined in the initial game, then so does 
the corresponding condition of $\FIND$ in the
transformed game.
(Essentially, when $y[M'_1, \ldots, M'_{l'}]$ occurs in $M$,
the transformation typically creates new occurrences of
$y[M''_1, \ldots, M''_{l'}]$ for some $M''_1, \ldots, M''_{l'}$,
so the condition that $y[M''_1, \ldots, M''_{l'}]$ is defined 
must sometimes be explicitly added to conditions of $\FIND$
in order to preserve Invariant~\ref{inv2}.)
Moreover, we replace as often as possible defined conditions $x[M_1,
  \ldots, M_l]$ with defined conditions $y[M_1, \ldots, M_l]$ where
$y$ is defined at the same time as $x$.
When $x\in V$, its definition is kept unchanged.
Otherwise, when $x$ is not referred to at all after the transformation,
we remove the definition of $x$. When $x$ is referred to only at the root of 
$\defined$ tests, we replace its definition with a constant.
(The definition point of $x$ is important, but not its value.)

This removal of assignments is applied to all variables whose value is not used
(those are used only at the root of $\kw{defined}$ conditions, or not at all).
Depending on the argument of the transformation, it is also applied to other assignments:
\begin{itemize}
\item \rn{findcond}: all assignments in conditions of $\FIND$;
\item \rn{useless}: assignments that store a variable or a replication index, when the setting
  \rn{expandAssignXY} is true; otherwise, no other assignment;
\item \rn{binder} $x_1$ \dots $x_n$: the assignments to variables $x_1$, \dots, $x_n$.
\end{itemize}
After this transformation, \rn{remove\_assign} calls
\rn{auto\_SArename} to guarantee Property~\ref{prop:autosarename}.

With the arguments \rn{findcond} and \rn{useless}, this is repeated as
many times as specified by the setting
\rn{maxIterRemoveUselessAssign}. Repetition stops if a fixpoint is
reached.

\begin{lemma}
  The transformation \rn{remove\_assign} requires and preserves
  Properties~\ref{prop:nointervaltypes}, \ref{prop:noreschan}, \ref{prop:channelindices}, \ref{prop:notables}, and~\ref{prop:autosarename}.
  It preserves Property~\ref{prop:expand}.
  If transformation \rn{remove\_assign} transforms $G$ into $G'$, then
  $\dset,\dsetsnu : G, \ab D, \ab \usedevents \ab \indistev{V}{0} G', D, \usedevents$.
\end{lemma}

\begin{example}\label{exa:running2}
In the game $G_0$ of Section~\ref{sec:exacorresp},
the transformation
$\rn{remove\_assign\ binder}\ sk_A$ substitutes $\skgen(rk_A)$ for $sk_A$ 
in the whole process and removes the assignment 
$\bassign{sk_A}{\skgen(rk_A)}$.
After this substitution, $\sign(\concat(pk_A, x_B, x_N), sk_A, r)$
becomes $\sign(\concat(pk_A, \ab x_B, \ab x_N), \ab \skgen(rk_A), \ab r)$
thus exhibiting a term required to apply the security assumption
on signatures in the cryptographic transformation of
Section~\ref{sec:primdef}.
\end{example}

\subsubsection{\rn{use\_variable}}

The transformation \rn{use\_variable} $x_1$ \dots $x_m$ tries to use
variables $x_1$, \dots, $x_n$ instead of recomputing their value.
More precisely, at each program point that corresponds to a simple term $M$
and where $x_j[\tup{M}]$ is guaranteed
to be defined (because $x_j$ is defined above that program point or directly
or indirectly because of $\defined$ conditions above that program point),
if all definitions of $x_j$ that can be executed before reaching that
program point are $\bassign{x_j[\tup{i}]}{M_0}\ \IN$, then we test
whether $M$ is equal to $M_0\{\tup{M}/\tup{i}\}$ modulo the built-in
equations, and if yes, we replace $M$ with $x_j[\tup{M}]$.

The $\kw{defined}$ conditions of $\FIND$ above the modified program points are 
updated to make sure that Invariant~\ref{inv2} is preserved. This is needed
in particular to make sure that $x_j[\tup{M}]$ syntactically occurs
in the $\defined$ conditions when it is used.

\begin{lemma}
  The transformation \rn{use\_variable} requires and preserves
  Properties~\ref{prop:nointervaltypes}, \ref{prop:noreschan}, \ref{prop:channelindices}, \ref{prop:notables}, and~\ref{prop:autosarename}.
  It preserves Property~\ref{prop:expand}.
  If transformation \rn{use\_variable} transforms $G$ into $G'$, then
  $\dset,\dsetsnu : G, \ab D, \ab \usedevents \ab \indistev{V}{0} G', D, \usedevents$.
\end{lemma}

The transformation \rn{use\_variable} is a convenient way to perform
common subexpression elimination, possibly by first inserting the
definition of the desired variable(s) by \rn{insert} (Section~\ref{sec:insert}).
This transformation could be done by several applications
of the transformation \rn{replace} (Section~\ref{sec:replace}). However, \rn{use\_variable} is
easier to use when it performs the desired replacement. For performance
reasons, the equality tests performed by \rn{use\_variable} are considerably
less powerful than those performed by \rn{replace}, so if \rn{use\_variable}
does not replace a term with $x_j[\tup{M}]$ at some occurrence, it is worth
trying \rn{replace}. 

\subsubsection[\rn{SArename}]{\rn{SArename} \cite{Blanchet07c}}\label{sec:SArename}

The transformation \rn{SArename}\ $x$ (single assignment rename)
aims at renaming $x$ so that distinct definitions of $x$ have
different names; this is useful for distinguishing cases depending
on which definition of $x$ has set $x[\tup{i}]$.
This transformation can be applied only when $x \notin V$.
When $x$ has $m > 1$ definitions, we rename each definition of
$x$ to a different variable $x_1, \ldots, x_m$. Terms $x[\tup{i}]$
under a definition of $x_j[\tup{i}]$ are then replaced with 
$x_j[\tup{i}]$.
Each branch of find 
$\FB = \tup{\vf}[\tup{i}] = \tup{i}' \leq \tup{n}\ \SUCHTHAT \ 
\defined (M'_1, \ldots, M'_{l'})\fand M \THEN \dots$
where $x[M_1, \ldots, M_l]$ is a subterm of some $M'_k$ for $k \leq l'$
is replaced with $m$ branches $\FB\{ x_j[M_1, \ldots, M_l] /
x[M_1, \ldots, M_l]\}$ for $1 \leq j \leq m$.

Moreover, the implementation
takes into account that some variables cannot be simultaneously
defined, to reduce the number of branches of $\FIND$ to generate.

After this transformation, \rn{SArename} calls
\rn{auto\_SArename} to guarantee Property~\ref{prop:autosarename}.

As a particular case, \rn{SArename\ random} performs the following transformation:
when $y$ is defined by $\Res{y}{T}$ and has $m > 1$
definitions and all variable accesses to $y$ are of the form $y[i_1,
\ldots, i_l]$ under a definition of $y[i_1, \ldots, i_l]$, where
$i_1, \ldots, i_l$ are the current replication indices at this
definition of $y$ (that is, $y$ has no array access using $\FIND $), 
it renames $y$ to $y_1,
\ldots, y_m$ with a different name for each definition of $y$ by $\Res{y}{T}$.

\begin{lemma}
  The transformation \rn{SArename} requires and preserves
  Properties~\ref{prop:nointervaltypes}, \ref{prop:noreschan}, \ref{prop:channelindices}, \ref{prop:notables}, and~\ref{prop:autosarename}.
  It preserves Property~\ref{prop:expand}.
  If transformation \rn{SArename} transforms $G$ into $G'$, then
  $\dset,\dsetsnu : G, \ab D, \ab \usedevents \ab \indistev{V}{p} G', D, \usedevents$,
  where $p$ is $\epsilon_{\FIND}$ times the maximal number of executions of a modified
  $\FIND$ that is not proved unique.
\end{lemma}

\begin{example}\label{exa:SSA}
Consider the following process
{\allowdisplaybreaks\begin{align*}
&\cinput{\startch}{};
\Res{k_A}{T_k}; \Res{k_B}{T_k}; 
\coutput{yield}{};(Q_K \parpop Q_S)\\
&Q_K = \repl{i}{n}\cinput{c[i]}{h:T_h, k:T_k}\\*
&\phantom{Q_K =\, }\bguard{h = A}{\assign{k':T_k}{k_A}\coutput{yield}{}}{}\\*
&\phantom{Q_K =\, }\bguard{h = B}{\assign{k':T_k}{k_B}\coutput{yield}{}}{}\\*
&\phantom{Q_K =\, }\assign{k':T_k}{k}\coutput{yield}{}\\
&Q_S = \repl{i'}{n'}\cinput{c'[i']}{h':T_h};\\*
&\phantom{Q_S =\, } \FIND\ \vf = i'' \leq n\ \SUCHTHAT\ \defined(h[i''],k'[i'']) \fand h' = h[i'']\,\kw{then}\, P_1(k'[\vf])\,\kw{else}\, P_2
\end{align*}}%
The process $Q_K$ stores in $(h,k')$ a table of pairs (host name,
key): the key for $A$ is $k_A$, for $B$, $k_B$, and for any other $h$,
the adversary can choose the key $k$.  The process $Q_S$ queries this
table of keys to find the key $k'[\vf]$ of host $h'$, then executes
$P_1(k'[\vf])$. If $h'$ is not found, it executes $P_2$.

By the transformation $\rn{SArename}\ k'$, we can perform a case
analysis, to distinguish the cases in which $k' = k_A$, $k' = k_B$,
or $k' = k$, by renaming the three definitions of $k'$ to
$k'_1$, $k'_2$, and $k'_3$ respectively.
After transformation, we obtain the following processes:
{\allowdisplaybreaks\begin{align*}
&Q'_K = \repl{i}{n}\cinput{c[i]}{h:T_h, k:T_k}\\*
&\phantom{Q'_K =\, }\bguard{h = A}{\assign{k'_1:T_k}{k_A}\coutput{yield}{}}{}\\*
&\phantom{Q'_K =\, }\bguard{h = B}{\assign{k'_2:T_k}{k_B}\coutput{yield}{}}{}\\*
&\phantom{Q'_K =\, }\assign{k'_3:T_k}{k}\coutput{yield}{}\\
&Q'_S = \repl{i'}{n'}\cinput{c'[i']}{h':T_h};\\*
&\phantom{Q'_S =\, }\FIND\ \vf = i'' \leq n\ \SUCHTHAT\ \defined(h[i''],k'_1[i''])\fand h' = h[i''] \THEN P_1(k'_1[\vf])\\*
&\phantom{Q'_S =\, }\ \ \oplus \vf = i'' \leq n\ \SUCHTHAT\ \defined(h[i''],k'_2[i''])\fand h' = h[i''] \THEN P_1(k'_2[\vf])\\*
&\phantom{Q'_S =\, }\ \ \oplus \vf = i'' \leq n\ \SUCHTHAT\ \defined(h[i''],k'_3[i''])\fand h' = h[i''] \THEN P_1(k'_3[\vf])\ \ELSE{P_2}
\end{align*}}%
The $\FIND$ in $Q_S$, which looks for elements in array $k'$, is transformed
in $Q'_S$ into a $\FIND$ with three branches, one for each new name of $k'$
($k'_1$, $k'_2$, and $k'_3$ respectively).
After the simplification (Section~\ref{sec:simplify}), $Q'_S$ becomes:
\begin{align*}
&Q''_S = \repl{i'}{n'}\cinput{c'[i']}{h':T_h};\\
&\phantom{Q'_S =\, }\FIND\ \vf = i'' \leq n\ \SUCHTHAT\ \defined(h[i''],k'_1[i''])\fand h' = A \THEN P_1(k_A)\\
&\phantom{Q'_S =\, }\ \ \oplus \vf = i'' \leq n\ \SUCHTHAT\ \defined(h[i''],k'_2[i''])\fand h' = B \THEN P_1(k_B)\\
&\phantom{Q'_S =\, }\ \ \oplus \vf = i'' \leq n\ \SUCHTHAT\ \defined(h[i''],k'_3[i''])\fand h' = h[i''] \THEN P_1(k[\vf])\ \ELSE{P_2}
\end{align*}
since, when $k'_1[\vf]$ is defined, $k'_1[\vf] = k_A$ and $h[\vf] = A$,
and similarly for $k'_2[\vf]$ and $k'_3[\vf]$.
%The simplification can also take into account that in P_1(k[\vf]), h \neq A
%and h \neq B.
\end{example}

\subsubsection[\rn{move}]{\rn{move} \cite{Blanchet07c}}

The transformation \rn{move} moves random choices and assignments downwards 
in the code as much as possible.
A random choice $\Res{x[\tup{i}]}{T}$ or assignment $\bassign{x[\tup{i}]:T}{M}$
cannot be moved under a replication, or 
under a parallel composition when both sides use $x$, or a 
let $\assign{y[\tup{i}]:T}{M}{\ldots}$, 
input $\cinput{c[M_1, \ldots, M_l]}{x_1[\tup{i}]:T_1, \ldots, 
x_k[\tup{i}]:T_k}$, output 
$\coutput{c[M_1, \ldots, M_l]}{N_1, \ldots, N_k}$ when $x$ occurs in 
$M, M_1, \ldots, \ab M_l, N_1, \ldots, N_k$, or a $\FIND $ (or $\kw{if}$) when
the conditions use $x$. It can be moved under the other constructs,
duplicating it if necessary, when we move it under a $\FIND $ (or $\kw{if}$)
that uses $x$ in several branches. 
Note that when the random choice $\Res{x[\tup{i}]}{T}$ or assignment
$\bassign{x[\tup{i}]:T}{M}$ cannot be 
moved under
an input, a parallel composition, or a replication, it must be written
above the output that is located above the considered input, parallel 
composition or replication, so that the syntax of processes is not violated.
When there are array accesses to $x$, the random choice $\Res{x[\tup{i}]}{T}$ or assignment $\bassign{x[\tup{i}]:T}{M}$ can be moved only inside the same output process, without moving it under an output or under a $\FIND$ that makes an array access to $x$.

The conditions above are necessary for the soundness of the move.
Furthermore, the move is considered beneficial when it satisfies the following conditions:
\begin{itemize}
  
\item for random choices, when the random choice can
be moved under a $\FIND$ (or $\kw{if}$).
When this transformation duplicates a $\Res{x[\tup{i}]}{T}$ by 
moving it under a $\FIND$ that uses $x$ in several branches,
a subsequent \rn{SArename}($x$) enables us to distinguish several
cases depending in which branch $x$ is created, which is
useful in some proofs.

\item for assignments, when there are no
array accesses to $x$, the assignment to $x$ can be moved under a
$\FIND$ (or $\kw{if}$), and $x$ is used in a single branch of that
$\FIND$ (or $\kw{if}$). In this case, the assignment can be performed
only in the branch that uses $x$, so it will be computed in fewer cases
thanks to the move.

\end{itemize}
The performed moves are determined by the argument of the transformation:
\begin{itemize}
\item \rn{all}: moves all random choices and assignments, provided the move is beneficial.
\item \rn{noarrayref}: moves all random choices and assignments that do not have array references,
  provided the move is beneficial.
\item \rn{random}: moves all random choices, provided the move is beneficial.
\item \rn{random\_noarrayref}: moves all random choices that do not have array references,
  provided the move is beneficial.
\item \rn{assign}: moves all assignments, provided the move is beneficial.
\item \rn{binder}\ $x_1$ \dots $x_n$: move the variables $x_1$, \dots, $x_n$ (even when the move
  is not beneficial).
\end{itemize}
In all cases, only random choices and assignments at the process level (not inside terms) are moved.

\begin{lemma}
  The transformation \rn{move} requires and preserves
  Properties~\ref{prop:nointervaltypes}, \ref{prop:noreschan}, \ref{prop:channelindices}, \ref{prop:notables}, and~\ref{prop:autosarename}. It preserves Property~\ref{prop:expand}.
  If transformation \rn{move} transforms $G$ into $G'$, then
  $\dset,\dsetsnu : G, \ab D, \ab \usedevents \ab \indistev{V}{0} G', D, \usedevents$.
\end{lemma}

\subsubsection[\rn{move array}]{\rn{move array} \cite{BlanchetEPrint12}}

The transformation $\rn{move array}\ X$ delays the generation of a random value $X$ until the point at which it is first used (lazy sampling). This transformation is implemented as a particular case of a cryptographic transformation by the following equivalence:
\begin{tabbing}
$\repl{i}{\kwp{n}} \Res{X}{\kwt{T}};$ \\
\quad  $(\repl{\kvar{iX}}{\kwp{nX}} \kwf{OX}() := \kw{return}(X) \mid$\\
\quad  $\phantom{(}   \repl{\kvar{ieq}}{\kwp{neq}} \kwf{Oeq}(X':\kwt{T}) := \kw{return}(X' = X))$\\
$\approx_{\#\kwf{Oeq} / |\kwt{T}|}\ [\mathit{manual}]$\\
$\repl{i}{\kwp{n}}$\\ 
\quad  $(\repl{\kvar{iX}}{\kwp{nX}} \kwf{OX}() :=$\\
\qquad $\kw{find}[\kw{unique}]\ j\leq\kwp{nX}\ \kw{suchthat}\ \kw{defined}(Y[j])$\\
\qquad $\kw{then}\ \kw{return}(Y[j])\ \kw{else}\ \Res{Y}{\kwt{T}}; \kw{return}(Y) \mid$\\
\quad  $\phantom{(} \repl{\kvar{ieq}}{\kwp{neq}} \kwf{Oeq}(X':\kwt{T}) :=$\\
\qquad $\kw{find}[\kw{unique}]\ j\leq\kwp{nX}\ \kw{suchthat}\ \kw{defined}(Y[j])$\\
\qquad $\kw{then}\ \kw{return}(X' = Y[j])\ \kw{else}\ \kw{return}(\kwf{false}))$
\end{tabbing}
where $\kwt{T}$ is the type of $X$. Two oracles are defined,
$\kwf{OX}$ and $\kwf{Oeq}$. In the left-hand side, $\kwf{OX}$ returns
the random $X$ itself. In the right-hand side, $\kwf{OX}$ uses a
lookup to test if the random value was already generated; if yes, it
returns the previously generated random value $Y[j]$; if no, it
generates a fresh random value $Y$. Transforming the left-hand side
into the right-hand side therefore moves the generation of the random
number $X$ to the first call to $\kwf{OX}$, that is, the first usage
of $X$. The oracle $\kwf{Oeq}$ provides an optimized treatment of
equality tests $X' = X$: when the random value $X$ was not already
generated, we return $\kwf{false}$ instead of generating a fresh
$X$, so we exclude the case that $X'$ is equal to a fresh
$X$. This case has probability $1/|T|$ for each call to $\kwf{Oeq}$,
so the probability of distinguishing the two games is
$\#\kwf{Oeq}/|T|$. 
(Notice that
%similarly to the reasoning done in Section~\ref{sec:hash} for the
%Random Oracle Model,
there never exist several choices of $j$ that satisfy the conditions of the $\kw{find}$s in the right-hand side of this equivalence, so these $\kw{find}$s can be marked $[\kw{unique}]$ without modifying their behavior.)

%\extend{A generalization of $\kwf{Oeq}$ with other collisions is now supported by CryptoVerif.}%

\subsubsection{\rn{move\ up}}\label{sec:moveup}

  The transformation $\rn{move\ up\ }x_1\ \dots\ x_n\rn{ to }\pp$ moves the
  random number generations or assignments of $x_1$, \dots, $x_n$ upwards
  in the syntax tree, to the program point $\pp$. This program point
  must correspond to an output process.

  \label{sec:pp}
  The program
point $\pp$ is an integer, which can be determined using the
command \rn{show\_game\ occ}: this command displays the
current game with the corresponding label $\{\pp\}$ at each program
point.  The command \rn{show\_game\ occ} also
allows one to inspect the game, for instance to know the names
of fresh variables created by CryptoVerif during previous transformations.
Program points and variable names may depend on the version of
CryptoVerif.
Since CryptoVerif version 2.01, program points can also be designated
by expressions like $\rn{before}\ \mathit{regexp}$, which designates the
program point at the beginning of the line that matches the regular
expression $\mathit{regexp}$;
$\rn{after}\ \mathit{regexp}$, which designates the
program point just after the line that matches $\mathit{regexp}$;
$\rn{before\_nth }n\ \mathit{regexp}$, which designates the program
  point at the beginning of the $n$-th line that matches the regular expression
  $\mathit{regexp}$;
$\rn{after\_nth }n\ \mathit{regexp}$, which designates the program point at the beginning of the first line
  that has an occurrence number after 
  the $n$-th line that matches the regular expression
  $\mathit{regexp}$;
$\rn{at }n'\ \mathit{regexp}$, which designates the program point at the $n'$-th occurrence number
  that occurs inside the string that matches the regular expression
  $\mathit{regexp}$ in the displayed game;
$\rn{at\_nth }n\ n'\ \mathit{regexp}$, which designates the program point at the $n'$-th occurrence number
  that occurs inside the string corresponding to the
  $n$-th match of the regular expression
  $\mathit{regexp}$ in the displayed game.
This way of designating program points is more stable across versions
of CryptoVerif.

  After the game transformation, a variable $x$ is defined at program
  point $\pp$, and all other variables $x_k$ are defined by
  $\bassign{x_k}{x}\ \IN$.  The variable $x$ is a variable
  $x_k$ itself when $x_k$ has no array accesses and the current replication
  indices at the definition of $x_k$ are the same as at
  $\pp$. Otherwise, the variable $x$ is a fresh variable.
  
  All variables $x_1$, \dots, $x_n$ must have the same type.
  They must not be defined syntactically above the program point $\pp$.
  The definitions of the variables $x_1$, \dots, $x_n$ must be in distinct
  branches of $\kw{if}$, $\kw{find}$, $\kw{let}$, so that they cannot be
  simultaneously defined.
  Either all variables $x_1$, \dots, $x_n$ must be defined by
  random number generations or all of them must be defined by assignments.
  \begin{itemize}
    
  \item If $x_1$, \dots, $x_n$ are defined by random number
    generations, this transformation performs eager sampling of $x_i$.
    The random number generation of $x_1$, \dots $x_n$ must be
    executed at most once for each execution of program point
    $\pp$. This is proved by combining that the definitions of the
    variables $x_1$, \dots, $x_n$ are in distinct branches of $\kw{if}$,
    $\kw{find}$, $\kw{let}$ with the fact that each of these definitions
    (at $\pp_j$) is executed at most once for each value of the
    current replication indices at $\pp$. To show the latter fact, we notice
    that, since $\pp_j$ is syntactically under $\pp$, the current replication
    indices at $\pp$ are a prefix of the replication indices at
    $\pp_j$. If the replication indices at $\pp_j$ are the same as at
    $\pp$, then the fact is proved. Otherwise, the replication indices
    at $\pp_j$ are $\tup{i},\tup{i}_j$ while the replication indices
    at $\pp$ are $\tup{i}$ and we show that $\fset_{\pp_j} \cup
    \fset_{\pp_j}\{\tup{i}_j'/\tup{i}_j\} \cup \{ \tup{i}_j \neq
    \tup{i}_j'\}$ yields a contradiction, where $\tup{i}_j'$ are fresh
    replication indices.

  \item If $x_1$, \dots, $x_n$ are defined by assignments of terms
    $M_j$, then all $M_j$ must consist of variables, function
    applications, and tests; there must be one $M_j$ defined at
    program point $\pp$, using all defined variables collected in
    $\fset_{\pp}$ (let $M_{j_0}$ be that $M_j$, which will be used as
    the definition of $x$: $\bassign{x}{M_{j_0}}\ \IN$); and all terms
    $M_j$ must be equal: for all $j \neq j_0$, $M_j = M_{j_0}$ knowing the facts
    $\fset_{\pp_j}$ that hold at the program point $\pp_j$ of $M_j$.

    The $\kw{defined}$ conditions of $\FIND$ above $\pp$ are updated
    to syntactically guarantee the definition of $M_{j_0}$, as required by
    Invariant~\ref{inv2}.
    
  \end{itemize}

\begin{lemma}
  The transformation \rn{move up} requires and preserves
  Properties~\ref{prop:nointervaltypes}, \ref{prop:noreschan},
  \ref{prop:channelindices}, \ref{prop:notables},
  and~\ref{prop:autosarename}.  It preserves
  Property~\ref{prop:expand}.  If transformation \rn{move up}
  transforms $G$ into $G'$, then $\dset,\dsetsnu : G, \ab D, \ab
  \usedevents \ab \indistev{V}{p} G', D, \usedevents$, where $p$ is an
  upper bound of the probability of collisions eliminated in the proof
  that each $\pp_j$ is executed at most once for each execution of $\pp$ or
  that for all $j \neq j_0$, $M_j = M_{j_0}$.
\end{lemma}

\subsubsection{\rn{move\_if\_fun}}\label{sec:moveiffun}

The transformation \rn{move\_if\_fun} moves the predefined function
  $\iffun$ or transforms it into a term $\kw{if}\ \dots\ \kw{then}\ \dots\ \kw{else} \dots$ It supports the following variants:
  \begin{itemize}
  \item \rn{move\_if\_fun\ }$\mathit{loc}_1 \dots \mathit{loc}_n$, where each $\mathit{loc}_j$ is either a program point or a function symbol. When $\mathit{loc}_j$ is a program point, it moves occurrences of $\iffun$ from inside the term at that program point to the root of that term. (The program point $\pp$
is designated as explained in Section~\ref{sec:pp}.) When $\mathit{loc}_j$ is a function symbol, it moves occurrences of $\iffun$ from under that function symbol to just above it. The move corresponds to rewriting $C[\iffun(M_1,M_2,M_3)]$ into $\iffun(M_1,C[M_2],C[M_3])$, where $C$ is a term context built from the following grammar:
\begin{defn}
\categ{C}{simple term context}\\
\entry{[\,]}{hole}\\
\entry{x[M_1,\dots,M_{k-1},C,M_{k+1}, \dots, M_m]}{variable}\\
\entry{f(M_1,\dots,M_{k-1},C,M_{k+1}, \dots, M_m)}{function application}
\end{defn}
and the root of $C$ corresponds to a $\mathit{loc}_j$. (It is at program point $\mathit{loc}_j$ when $\mathit{loc}_j$ is a program point; its root symbol is $\mathit{loc}_j$ when $\mathit{loc}_j$ is a function symbol.) These moves are possible only when $C$ is a simple term context, since otherwise they might lead to defining several times the same variable or repeating events since the context $C$ is duplicated in the second and third arguments of $\iffun$ and $\iffun$ evaluates all its arguments. For simplicity, we allow them only when $C[\iffun(M_1,M_2,M_3)]$ is a simple term.
  \item \rn{move\_if\_fun\ level\ }$n$, where $n$ is a positive integer, moves occurrences of $\iffun$ $n$ function symbols up in the syntax tree (provided those $\iffun$ occur under at least $n$ function symbols). As above, these moves are allowed only when they occur inside a simple term.
  \item \rn{move\_if\_fun\ to\_term\ }$\pp_1 \dots \pp_n$ transforms terms $\iffun(M_1,M_2,M_3)$ that occur at program points $\pp_1$, \dots, $\pp_n$ into terms $\kw{if}$ $M_1$ $\kw{then}$ $M_2$ $\kw{else}$ $M_3$. When no program point is given, it performs that transformation everywhere in the game.

    When $M_2$ and $M_3$ have a visible effect, that is, they define some variable with array accesses (including by their usage in various kinds of secrecy queries) or they execute events, the transformation above would not be correct, because $\iffun(M_1,M_2,M_3)$ evaluates both $M_2$ and $M_3$ while $\kw{if}$ $M_1$ $\kw{then}$ $M_2$ $\kw{else}$ $M_3$ evaluates either $M_2$ or $M_3$. In this case, we transform $\iffun(M_1,M_2,M_3)$ into
    $\assign{\kvar{xcond}}{M_1} \assign{\kvar{xthen}}{M_2} \assign{\kvar{xelse}}{M_3} \bguard{\kvar{xcond}}{\kvar{xthen}}{\kvar{xelse}}$
    to make sure that $M_1$, $M_2$, and $M_3$ are always evaluated, and in that order.
    
    When \rn{autoExpand = true} (the default), a call to \rn{expand} is automatically performed after \rn{move\_if\_fun}, which transforms the terms $\assign{\dots}{\dots}\dots$ and $\kw{if}$ {\dots} $\kw{then}$ {\dots} $\kw{else}$ {\dots} into processes. 
  \end{itemize}

\begin{lemma}
  The transformation \rn{move\_if\_fun} requires and preserves
  Properties~\ref{prop:nointervaltypes}, \ref{prop:noreschan}, \ref{prop:channelindices}, \ref{prop:notables}, and~\ref{prop:autosarename}.
  The variants \rn{move\_if\_fun\ }$\mathit{loc}_1 \dots \mathit{loc}_n$ and
  \rn{move\_if\_fun\ level\ }$n$ preserve Property~\ref{prop:expand}.
  If transformation \rn{move\_if\_fun} transforms $G$ into $G'$, then
  $\dset,\dsetsnu : G, \ab D, \ab \usedevents \ab \indistev{V}{0} G', \ab D, \ab \usedevents$.
\end{lemma}

\subsubsection[\rn{insert\_event}]{\rn{insert\_event} \cite{BlanchetEPrint12}}\label{sec:insertevent}

The transformation $\rn{insert\_event}\ e\ \pp$ inserts
$\keventabort{e}$ at program point $\pp$. (The program point $\pp$
is designated as explained in Section~\ref{sec:pp}.)

The transformation $\rn{insert\_event}\ e\ \pp$ also adds
to query $\fevent{e} \Rightarrow \false$ in order to bound
the probability of event $e$.

\begin{lemma}
  The transformation \rn{insert\_event} requires and preserves
  Properties~\ref{prop:nointervaltypes}, \ref{prop:noreschan}, \ref{prop:channelindices}, \ref{prop:notables}, and~\ref{prop:autosarename}. It preserves Property~\ref{prop:expand} if the event is inserted at the process level.
  If transformation \rn{insert\_event} $e$ $\pp$ transforms $G$ into $G'$, then
  $\dset,\dsetsnu : G, \ab D, \ab \usedevents \ab \indistev{V}{0} G', D \vee e, \usedevents \cup \{ e\}$.
\end{lemma}

\subsubsection[\rn{insert}]{\rn{insert} \cite{BlanchetEPrint12}}\label{sec:insert}

The transformation $\rn{insert}\ \pp\ \mathit{ins}$ adds instruction
$\mathit{ins}$ at the program point $\pp$. The program point $\pp$
is designated as explained in Section~\ref{sec:pp}. The instruction
$\mathit{ins}$ can for instance be a test, in which case all branches
of the test will be copies of the code that follows program point $\pp$
(so that the semantics of the game is unchanged). It can also be an
assignment or a random generation of a fresh variable or an
$\kw{event\_abort}$ instruction. In all cases, CryptoVerif checks that
this instruction preserves the semantics of the game except when we
execute an inserted Shoup event, and rejects it with an error
message if it does not.

After this transformation, \rn{insert} calls
\rn{auto\_SArename} to guarantee Property~\ref{prop:autosarename}.

When the user inserts $\keventabort{e}$, the transformation
$\rn{insert}$ adds a query $\fevent{e} \Rightarrow \false$ in order to
bound the probability of event $e$.

When the user inserts a $\FIND\unique{e}$ (the user actually types
$\FIND\unique{}$ but CryptoVerif automatically generates a fresh event $e$
and inserts $\FIND\unique{e}$ instead, since uniqueness is not proved yet),
the transformation $\rn{insert}$ adds a query $\fevent{e} \Rightarrow \false$
and calls \rn{prove\_unique} (Section~\ref{sec:prove_unique})
in order to try proving uniqueness.

\begin{lemma}
  The transformation \rn{insert} requires and preserves
  Properties~\ref{prop:nointervaltypes}, \ref{prop:noreschan}, \ref{prop:channelindices}, \ref{prop:notables}, \ref{prop:autosarename}, and~\ref{prop:expand}.
  If transformation \rn{insert} transforms $G$ into $G'$, then
  $\dset,\dsetsnu : G, \ab D, \ab \usedevents \ab \indistev{V}{p} G', D \vee e_1 \vee \dots \vee e_m,
  \usedevents \cup \{ e_1, \dots, e_m \}$
  where $e_1$, \dots, $e_m$ are the events in the inserted instruction ($\keventabort{e_j}$, $\FIND\unique{e_j}$) and the probability $p$ comes from \rn{prove\_unique}.
\end{lemma}

\subsubsection{\rn{replace}}\label{sec:replace}

The transformation $\rn{replace}\ \pp\ M$ replaces the term $M_0$ at program point
$\pp$ with the term $M$. ($M_0$ and $M$ must be simple. The program point $\pp$
is designated as explained in Section~\ref{sec:pp}.) Before performing the replacement, it checks that 
$M$ is equal to $M_0$ at that program point (up to a small probability): first, it collects all facts $\fset_{\pp}$ that hold at program point $\pp$; second, it tests equality between $M_0$ and $M$ using $\fset_{\pp}$ and built-in equations (it uses equalities inferred from $\fset_{\pp}$ to replace variables with their values, trying to make the terms equal); third, it simplifies $M_0$ and $M$ using user-defined rewrite rules of Section~\ref{sec:userdefinedrewriterules}, and tests equality between the results using $\fset_{\pp}$ and built-in equations; fourth, it rewrites $M_0$ and $M$ at most \rn{maxReplaceDepth} times using equalities inferred from $\fset_{\pp}$ and user-defined rewrite rules of Section~\ref{sec:userdefinedrewriterules}, until it finds a common term modulo the built-in equations. The transformation is performed as soon as the equality between $M_0$ and $M$ is proved.

The $\kw{defined}$ conditions of $\FIND$ above the program point $\pp$ are 
updated to make sure that Invariant~\ref{inv2} is preserved. This is needed
in particular when $M$ makes array accesses that $M_0$ does not make.

\begin{lemma}
  The transformation \rn{replace} requires and preserves
  Properties~\ref{prop:nointervaltypes}, \ref{prop:noreschan}, \ref{prop:channelindices}, \ref{prop:notables}, and~\ref{prop:autosarename}. It preserves Property~\ref{prop:expand}.
  If transformation \rn{replace} $\pp$ $M$ transforms $G$ into $G'$, then
  $\dset,\dsetsnu : G, \ab D, \ab \usedevents \ab \indistev{V}{p} G', D, \usedevents$, where $p$ is an upper bound of the probability that $M$ is different from $M_0$ at $\pp$.
\end{lemma}

The variant $\rn{assume\ replace}\ \pp\ M$ performs the same replacement,
without checking the equality between $M_0$ and $M$. This transformation
is obviously not sound, but can be used to experiment with modifications
in the games. As soon as this transformation is used, CryptoVerif
does not claim that any property is proved.

\subsubsection[\rn{merge\_branches}]{\rn{merge\_branches} \cite{BlanchetEPrint12}}\label{sec:mergebranches}

The transformation $\rn{merge\_branches}$ performs the following transformations:
\begin{enumerate}

\item If some $\kw{then}$ branches of a $\kw{find}[\kw{unique}]$ 
  execute the same code as the $\kw{else}$ branch (up to renaming of
  variables defined in these branches and that do not have array
  accesses, and up to equality of terms proved using facts $\fset_{\pp}$ that hold at the program point $\pp$ of the considered $\kw{find}[\kw{unique}]$), the index variables bound in these $\kw{then}$ branches
  have no array accesses, and the conditions of these $\kw{then}$ branches
  do not contain $\kw{event\_abort}$ nor unproved $\FIND\unique{e}$, 
  then we remove these $\kw{then}$ branches.

Indeed, these $\kw{then}$ branches have the same effect
as the $\kw{else}$ branch. The hypotheses are needed for
the following reasons:
\begin{itemize}
\item The renamed variables must not have array accesses because 
renaming variables that have array accesses requires transforming
these array accesses. The transformation $\rn{merge\_arrays}$
presented in Section~\ref{sec:mergearrays} can rename
variables with array accesses.

\item The index variables bound in the removed branches must not
have array accesses, because removing the definitions of these
variables would modify the behavior of the array accesses.

\item The conditions must not contain $\kw{event\_abort}$ nor unproved $\FIND\unique{e}$, because
  if they do, the $\kw{find}$ may abort while code after transformation
  would not abort.
\end{itemize}

\item If all branches of $\kw{if}$, $\kw{let}$ with pattern matching,
  or $\kw{find}$ execute the same code (up to renaming of variables
  defined in these branches and that do not have array accesses, and up to equality of terms proved using facts $\fset_{\pp}$ that hold at the program point $\pp$ of the considered $\kw{if}$, $\kw{let}$, or $\kw{find}$), and in case of $\kw{find}$,
  it is not marked $\unique{e}$,
  %forbid merging of branches on unproved find[unique]
  the index variables bound in the $\kw{then}$ branches
  have no array accesses, and the conditions of the $\kw{then}$ branches 
  do not contain $\kw{event\_abort}$ nor unproved $\FIND\unique{e}$, then
  we replace that $\kw{if}$ or $\kw{find}$ with its $\kw{else}$
  branch.  

  In this transformation, we ignore the array accesses that
  occur in the conditions of the $\kw{find}$ under consideration,
  since these conditions will disappear after the transformation.

\end{enumerate}
Furthermore, $\rn{merge\_branches}$ applies these transformations
globally to all $\kw{find}$s of the game for which the simplification
is possible. As a consequence, one can ignore array accesses to all
variables in conditions of $\kw{find}$ that will be removed, so more
transformations are enabled.

\begin{lemma}
  The transformation \rn{merge\_branches} requires and preserves
  Properties~\ref{prop:nointervaltypes}, \ref{prop:noreschan}, \ref{prop:channelindices}, \ref{prop:notables}, \ref{prop:autosarename}, and~\ref{prop:expand}.
  If transformation \rn{merge\_branches} transforms $G$ into $G'$, then
  $\dset,\dsetsnu : G, \ab D, \ab \usedevents \ab \indistev{V}{p} G', D, \usedevents$, where $p$ is an upper bound on the probability that equalities between terms of merged branches do not hold.
\end{lemma}

\subsubsection[\rn{merge\_arrays}]{\rn{merge\_arrays} \cite{BlanchetEPrint12}}\label{sec:mergearrays}

The transformation $\rn{merge\_arrays}$ $x_{11}\ \ldots\ x_{1n}$, $\ldots$, $x_{m1}\ \ldots\ x_{mn}$ merges the variables $x_{j1}, \ab \ldots, \ab x_{jn}$ into a single variable $x_{j1}$ for each $j \leq m$. Each variable $x_{jk}$ must have a single definition. For each $j \leq n$, the variables $x_{j1}, \ldots, x_{jn}$ must have the same type and indices of the same type. They must not be defined for the same value of their indices (that is, $x_{jk}$ and $x_{jk'}$ must be defined in different branches of $\kw{if}$ or $\kw{find}$ when $k \neq k'$). The arrays $x_{j1}, \ldots, x_{jn}$ are merged into a single array $x_{j1}$ for each $j \leq m$. The transformation proceeds as follows:
\begin{itemize}

\item If, for each $k \leq n$,  $x_{1k}$ is defined above $x_{jk}$ for all $1 < j < m$, we introduce a fresh variable $b_k$ defined by $b_k \gets \kwf{mark}$
just after the definition of $x_{1k}$. We call $b_k$ a \emph{branch variable}; it is used to detect that $x_{jk}$ has been defined: $x_{jk}[\tilde M]$ is defined before the transformation if and only if $x_{j1}[\tilde M]$ and $b_k[\tilde M]$
are defined after the transformation, and $x_{j1}[\tilde M]$ after the transformation is equal to $x_{jk}[\tilde M]$ before the transformation. 

\item For each $\kw{find}$ that requires that some variables $x_{jk}$ are defined, we leave the branches that do not require the definition of $x_{jk}$ unchanged and we try to transform the other branches $\FB_l = (\tilde u_l = \tilde i_l \leq \tilde n_l\ \kw{suchthat}\ \kw{defined}(\tilde M_l) \wedge M_l\ \kw{then}\ P_l)$ as follows.
\begin{enumerate}

\item We require that, for each $l$, there exists a distinct $k$ such that the $\defined$ condition of $\FB_l$ refers to $x_{jk}$ for some $j$ but not to $x_{jk'}$ for any other $k'$. 
%I simplify a bit, because this is required only for definitions of $x_{jk}$ newly requested by the current $\kw{find}$, not for definitions already guaranteed by instructions above in the game.
(Otherwise, the transformation fails.) We denote by $l(k)$ the value of $l$ that corresponds to $k$.

\item We choose a ``target'' branch $\targetFB  = (\tilde u = \tilde i \leq \tilde n\ \kw{suchthat}\ \kw{defined}(\tilde M) \wedge M\ \kw{then}\ P)$: if the $\defined$ condition of some branch $\FB_l$ refers to $x_{j1}$ for some $j$, we choose that branch $\FB_l$. Otherwise, we choose any branch $\FB_l$ and rename its variables $x_{jk}$ to $x_{j1}$.
%\bb{In the implementation, this renaming is postponed to the last step.}%
We require that the references $x_{j1}[\tilde M]$ to the variables $x_{j1}$ in the $\defined$ condition of the target branch all have the same indices $\tilde M$.
%I simplify a bit, because this is required only for definitions of $x_{jk}$ newly requested by the current $\kw{find}$, not for definitions already guaranteed by instructions above in the game.
%
%Replaces:
% \item We require that all references to $x_{jk}$ in the defined conditions of $\FB_l$ have the same indices $\tilde M$. 
%(after renaming of the \tilde u_l to the target \tilde u)
If the transformation succeeds, we will replace all branches $\FB_l$ with the target branch.

\item The branch $\targetFB$ after transformation is equivalent to
branches $\bigorfind_{k = 1}^n \targetFB\{x_{jk}/x_{j1}, \ab j = 1,\dots, m\}$ before
transformation. We show that these branches are equivalent to 
the branches $\FB_l$.

For each $k \leq n$, 
\begin{itemize}

\item
if $l(k)$ exists, then we show that $\targetFB\{x_{jk}/x_{j1}, \ab j = 1,\dots,m\}$ is equivalent to $\FB_{l(k)}$. Let $l = l(k)$. We first rename the variables $\tilde u_l$ of $\FB_l$ to the variables $\tilde u$ of the target branch. For simplicity, we still denote by $\FB_l = (\tilde u_l = \tilde i_l \leq \tilde n_l$ $\kw{suchthat}$ $\kw{defined}(\tilde M_l) \wedge M_l$ $\kw{then}$ $P_l)$ the obtained branch. Then we show that, if the variables of $\tilde M_l$ are defined, then the variables of $\tilde M\{x_{jk}/x_{j1}, \ab j = 1, \dots, m\}$ are defined, and conversely; $M_l = M\{x_{jk}/x_{j1}, \ab j = 1, \dots, m\}$ (knowing the equalities that hold at that program point), and $P_l$ and $P\{x_{jk}/x_{j1}, \ab j = 1, \dots, m\}$ execute the same code up to renaming of variables defined in $P_l$ or $P\{x_{jk}/x_{j1}, \ab j = 1, \dots, m\}$ and that do not have array accesses, and up to equality of terms proved using facts $\fset_{\pp}$ that hold at the program point $\pp$ of the considered $\kw{find}$.

\item
if $l(k)$ does not exist, then we show that $\targetFB\{x_{jk}/x_{j1}, j = 1, \dots, m\}$ can in fact not be executed, because its condition cannot hold: the variables of $\tilde M\{x_{jk}/x_{j1}, \ab j = 1, \dots, m\}$ cannot be simultaneously defined or $M\{x_{jk}/x_{j1}, j = 1, \dots, m\}$ cannot hold.

\end{itemize}
\end{enumerate}

If the transformation above fails and we have introduced branch variables, we replace each condition $\kw{defined}(x_{jk}[\tilde M])$ with $\kw{defined}(x_{j1}[\tilde M], b_k[\tilde M])$.

If the transformation above fails and we have not introduced branch variables, the whole $\rn{merge\_arrays}$ transformation fails.

\item The definition of $x_{jk}$ is renamed to $x_{j1}$ and each reference 
to $x_{jk}[\tilde M]$ is renamed to $x_{j1}[\tilde M]$.

\end{itemize}

\begin{lemma}
  The transformation \rn{merge\_arrays} requires and preserves
  Properties~\ref{prop:nointervaltypes}, \ref{prop:noreschan}, \ref{prop:channelindices}, \ref{prop:notables}, \ref{prop:autosarename}, and~\ref{prop:expand}.
  If transformation \rn{merge\_arrays} transforms $G$ into $G'$, then
  $\dset,\dsetsnu : G, \ab D, \ab \usedevents \ab \indistev{V}{p} G', D, \usedevents$, where $p$ is an upper bound on the probability that required equalities do not hold.
\end{lemma}

\subsubsection{\rn{guess} $i$}\label{sec:guessi}

When \rn{guessRemoveUnique = true} and some (one-session) secrecy
queries are present, the transformation \rn{guess} $i$ first
transforms the game $G$ into $\Gru$, by replacing all proved
$\FIND\unique{}$ with $\FIND$. Lemma~\ref{lem:Gru} shows the soundness
of this preliminary transformation. It may be advantageous for
(one-session) secrecy proofs because the removed $\FIND\unique{}$ do
not need to be proved, while the remaining ones must be reproved after
the transformation, as we show below.

\begin{lemma}\label{lem:Gru}
  Let $\Gru$ be the game obtained from $G$ by replacing all proved  $\FIND\unique{}$ with $\FIND$.
  Let $\prop$ be a security property ($\prop$ is $\secrone(x)$, $\secr(x)$, $\secrbit(x)$, or
a correspondence $\corresp$, which does not use $\sevent$, $\sbarevent$, nor non-unique events).
Let $D$ be a disjunction of Shoup events and a subset of non-unique events $e_i$ corresponding to unproved $\FIND\unique{e_i}$ in $G$ ($D$ does not contain $\sevent$ nor $\sbarevent$).
If $\bound{\Gru}{V}{\prop}{D}{p}$, then $\bound{G}{V}{\prop}{D}{p}$.
\end{lemma}
\begin{proof}
  Let $C'$ be an evaluation context acceptable for $C_{\prop}[G]$ with public variables $V \setminus V_{\prop}$ that does not contain events used by $\prop$ or $D$ nor non-unique events in $G$. Let $C = C'[C_{\prop}[\,]]$.
  The context $C'$ is also acceptable for $C_{\prop}[\Gru]$ with public variables $V \setminus V_{\prop}$, and a fortiori does not contain non-unique events in $\Gru$. Since $\bound{\Gru}{V}{\prop}{D}{p}$, we have $\Advtev{\Gru}{\prop}{C}{D} \leq p(C)$.

  First case: $\prop$ is a correspondence $\corresp$. We have
  \begin{align*}
    \Advtev{G}{\prop}{C}{D} &= \Pr[C[G] : (\neg \corresp \vee D) \wedge \neg\nonunique{G,D}] \\
    &\leq \Pr[C[\Gru] : (\neg \corresp \vee D) \wedge \neg\nonunique{\Gru,D}] = \Advtev{\Gru}{\prop}{C}{D}
  \end{align*}
  because traces of $G$ that satisfy $(\neg \corresp \vee D) \wedge \neg\nonunique{G,D}$ correspond to similar traces of $\Gru$ that satisfy $(\neg \corresp \vee D) \wedge \neg\nonunique{\Gru,D}$. Only traces that satisfy $e$ for some proved $\FIND\unique{e}$ in $G$ are mapped to different traces in $\Gru$. These traces satisfy $\nonunique{G,D}$.

  Second case: $\prop$ is $\secrone(x)$, $\secr(x)$, or $\secrbit(x)$. We have
\[\Pr[C[G]: \sevent \vee D] \leq \Pr[C[\Gru]: \sevent \vee D]\]
   since the traces of $C[G]$ that execute an event in $\sevent \vee D$ correspond to
   similar traces of $C[\Gru]$ that also execute an event in $\sevent \vee D$ ($\sevent \vee D$ does not contain any proved non-unique event of $G$),
   and
    \[\Pr[C[\Gru]: \sbarevent \vee \nonunique{\Gru,D}] \leq \Pr[C[G]: \sbarevent \vee \nonunique{G,D}]\]
   since the traces of $C[\Gru]$ that execute $\sbarevent$ or an
   event in $\nonunique{\Gru,D}$ correspond to
   either to similar traces of $C[G]$ that execute the same event
   or to traces that execute a proved non-unique event in $C[G]$, so an event in $\nonunique{G,D}$. Therefore,
   \begin{align*}
     \Advtev{G}{\prop}{C}{D} &= \Pr[C[G]: \sevent \vee D] - \Pr[C[G]: \sbarevent \vee \nonunique{G,D}] \\
     &\leq \Pr[C[\Gru]: \sevent \vee D] - \Pr[C[\Gru]: \sbarevent \vee \nonunique{\Gru,D}] = \Advtev{\Gru}{\prop}{C}{D}.
   \end{align*}
   In both cases, we obtain $\Advtev{G}{\prop}{C}{D} \leq p(C)$ and $\bound{G}{V}{\prop}{D}{p}$.
   \proofcomplete
\end{proof}

Next, the main guessing transformation is performed.
The transformation \rn{guess\ $i$} consists in guessing the tested session of a principal
in a protocol, which is a step frequently done in cryptographic proofs.
In CryptoVerif, we consider a game $G$ and define a transformed game $G'$
by guessing a replication index $i$: we replace
   $\repl{i}{n} Q$ with $\repl{i}{n} Q'$
   where $Q'$ is obtained from $Q$ by replacing the processes $P$
   under the first inputs with $\baguard{i = i_{\mathrm{tested}}}{P'}$ $\kw{else}$ $P''$
   and $i_{\mathrm{tested}}$ is a constant.
   The constant $i_{\mathrm{tested}}$ is the index of the tested session. We distinguish the process executed in the tested session, $P'$, on which we are going to prove security properties, from the process $P''$ for other sessions which are executed, but for which we do not prove security properties.
   (In case \rn{diff\_constants = true}, the constant
   $i_{\mathrm{tested}}$ must not be considered different from other constants
   of the same type.)
   The process $P'$ is obtained from $P$ by
   \begin{itemize}
   \item duplicating all events: $\kevent{e(\tup{M})}$ is replaced with
     $\kevent{e(\tup{M})}; \kevent{e'(\tup{M})}$ and similarly
     $\keventabort{e}$ is replaced with $\kevent{e}; \keventabort{e'}$.
     We require that in the game $G$,
     the same event $e$ cannot occur both under the modified replication
     $\repl{i}{n} Q$ and elsewhere in the game.
     (Otherwise, queries that use $e$ are left unchanged.)
   \item duplicating definitions of every variable $x$ used in queries
     for secrecy and one-session secrecy: $\bassign{x'}{x}\ \kw{in}$
     is added after each definition of $x$. We require that in the
     game $G$, the same variable $x$ used in queries for secrecy or
     one-session secrecy cannot be defined both under the modified
     replication $\repl{i}{n} Q$ and elsewhere in the game. (Otherwise,
     the considered query is left unchanged.)
   \end{itemize}
   The process $P''$ is obtained from $P$ by duplicating definitions of every variable $x$ used
   in queries for secrecy (not one-session secrecy): $\bassign{x''}{x}\ \kw{in}$ is added after each definition of $x$.
   
   We replace variables $x$ in secrecy and one-session secrecy queries
   with their duplicated version $x'$. For secrecy queries, the duplicated
   version $x''$ is added to public variables.
   In both cases, we prove (one-session) secrecy for the variable
   $x'$ defined in the tested session. In case of one-session secrecy,
that is enough: it shows that $x'$ is indistinguishable from a random
value, and that proves one-session secrecy of $x$ for all sessions by
symmetry. However, for secrecy, we additionally want to show that the
values of $x$ in the various sessions are independent of each other;
this is achieved by considering the value of $x$ in sessions other
than the tested session (that is, $x''$) as public: if $x'$ is
indistinguishable from random even when $x''$ is public, then $x'$ is
independent of $x''$.

   In non-injective correspondence queries, we replace \emph{one}
   non-injective event $e$ before the arrow $\Rightarrow$ with its
   duplicated version $e'$. Hence, we prove the query for the tested
   session, which uses event $e'$. The proof is valid for all sessions
   by symmetry.

   The probability of attack must basically be multiplied by $n$ for
   all modified queries.  The proof depends on the considered query
   and is detailed below.

   For queries that are left unchanged, i.e. secrecy and one-session
   secrecy queries for variables not defined under the modified
   replication, non-injective correspondence queries with no event
   before the arrow $\Rightarrow$ under the modified replication (the
   previous queries prove properties about other roles than the one
   for which we guess the tested session), bit secrecy queries
   (because the secret is defined under no replication, so it is not
   under the guessed replication), as well as injective correspondence
   queries (see details below), the probability is unchanged. It is
   clear that these queries are not affected by the transformation.

   After this transformation, \rn{guess} calls
   \rn{auto\_SArename} to guarantee Property~\ref{prop:autosarename}.

\begin{lemma}\label{lem:guess_i}
    The transformation \rn{guess} $i$ requires and preserves
    Properties~\ref{prop:nointervaltypes}, \ref{prop:noreschan},
    \ref{prop:channelindices}, \ref{prop:notables}, and~\ref{prop:autosarename}.
    It preserves Property~\ref{prop:expand}.

    Suppose the game $G$ is transformed into $G'$ by the transformation \rn{guess} $i$, where $i$ is a replication index bounded by $n$. Below, we consider only the modified queries.
     
     Let $\corresp$ and $\corresp'$ be respectively the semantics of a
     non-injective correspondence and its transformed
     correspondence.
     If $\bound{G'}{V}{\corresp'}{\Dfalse}{p}$  and $p$
     is independent of the value of $i_{\mathrm{tested}}$, then
     $\bound{G}{V}{\corresp}{\Dfalse}{n p}$.

%%      Let $C$ be an evaluation context acceptable for
%%      $G$ with any public variables that does not contain events used
%%      by $\corresp$. 
%%      %
%%      If $\Advtev{G'}{\corresp'}{C}{\Dfalse} \leq p(C)$ and $p$
%%      is independent of the value of $i_{\mathrm{tested}}$, then
%%      $\Advtev{G}{\corresp}{C}{\Dfalse} \leq n \times p(C)$.

     If $G'$ satisfies the one-session secrecy of $x'$ with public variables $V$ ($x, x', x'' \notin V$) up to probability $p$ and $p$ is independent of the value of $i_{\mathrm{tested}}$, then $G$ satisfies the one-session secrecy of $x$ with public variables $V$ up to probability $n p$.
%%      Let $C$ be an evaluation context acceptable for $C_{\secrone(x)}[G]$ with public variables $V$ ($x, x', x'' \notin V$) that does not contain $\sevent$ nor $\sbarevent$.
%%      If $\Advt_{G'}^{\secrone(x')}(C) \leq p(C)$ and $p$ is independent of the value of $i_{\mathrm{tested}}$, then $\Advt_G^{\secrone(x)}(C) \leq n \times p(C)$.
     %
     
%%      If $\Advtev{G'}{\secrone(x')}{\ab C[C_{\secrone(x')}[\,]]}{ \ab \nonunique{G'}} \leq p(C)$
%%      and $p$ is independent of the value of $i_{\mathrm{tested}}$, then $\Advtev{G}{\ab \secrone(x)}{C[C_{\secrone(x)}[\,]]}{ \ab \Dfalse} \leq n \times p(C)$
     
     If $G'$ satisfies the secrecy of $x'$ with public variables $V \cup \{ x'' \}$ ($x, x', x'' \notin V$) up to probability $p$ and $p$ satisfies Property~\ref{prop:prob}, then
     $G$ satisfies the secrecy of $x$ with public variables $V$ up to probability $n \times p$ (neglecting a small additional runtime of the context).

     If $\bound{G'}{V \cup \{ x'\}}{\secrone(x')}{\nonunique{G'}}{p}$ and $p$ satisfies Property~\ref{prop:prob}, then we have $\bound{G}{V \cup \{x\}}{\secrone(x)}{\Dfalse}{n p}$.

     If $\bound{G'}{V \cup \{x',x''\}}{\secr(x')}{\nonunique{G'}}{p}$ and $p$ satisfies Property~\ref{prop:prob}, then we have $\bound{G}{V \cup \{ x \}}{\secr(x)}{\Dfalse}{n p}$ (neglecting a small additional runtime of the context).
%%      If $\Advtev{G'}{\secr(x)}{C[C_{\secr(x')}[\,]]}{\nonunique{G'}} \leq p(C)$
%%      for every evaluation context $C$ acceptable for $C_{\secr(x')}[G']$ with public variables $V \cup \{ x'' \}$ that does not contain $\sevent$ nor $\sbarevent$ and $p$ satisfies the same condition, then $\Advtev{G}{\secr(x)}{C[C_{\secr(x)}[\,]]}{\Dfalse} \leq n \times p(C)$ for every evaluation context $C$ acceptable for $C_{\secr(x)}[G]$ with public variables $V$ that does not contain $\sevent$ nor $\sbarevent$ (neglecting a small additional runtime of the context).
   \end{lemma}
   Property~\ref{prop:prob} guarantees that $p$ is independent of the value of $i_{\mathrm{tested}}$, as well as other independence conditions needed for secrecy and for the properties on $\boundfun$ because we modify the context in the proof.
Since the third argument of $\boundfun$ is always $\Dfalse$ in the conclusion of Lemma~\ref{lem:guess_i}, we cannot use the optimization of considering the disjunction of several properties simultaneously, as outlined in Section~\ref{sec:computadv}: we must consider each property and event separately. Indeed, if we applied guessing to several properties at once, we might need to guess the tested session for each property, which would introduce several factors $n$. Since the third argument of $\boundfun$ is $\nonunique{G'}$ in the hypothesis of Lemma~\ref{lem:guess_i} for (one-session) secrecy properties, uniqueness of $\FIND\unique{e}$ must be reproved in game $G'$ after the guess transformation (that is, the probability that these $\FIND\unique{e}$ have several successful choices must be bounded again in $G'$).

   \begin{proof} 
     
     \paragraph{Non-injective correspondences}
   We suppose that the events under the transformed replication
   contain as argument the replication index $i$ of that replication.
   (CryptoVerif implicitly adds the current program point and
   replication indices to each event, and uses fresh distinct
   variables for the added replication indices in the queries. That
   does not change the meaning of the query.)

   Let $\forall i_0:[1,n],\tup{x}:\tup{T}; \fevent{e(\tup{M})} \wedge \psi \Rightarrow \exists \tup{y}:\tup{T}'; \phi$ be the initial query and
   $\forall i_0:[1,n],\tup{x}:\tup{T}; \fevent{e'(\tup{M})} \wedge \psi \Rightarrow \exists \tup{y}:\tup{T}'; \phi$ be the transformed query,
   where $\{i_0, \tup{x}\} = \fvar(\fevent{e'(\tup{M})} \wedge \psi)$,
   $\tup{y} = \fvar(\phi) \setminus \fvar(\fevent{e'(\tup{M})} \wedge \psi)$,
   and $i_0$ be the variable for the index of the transformed
   replication in $\tup{M}$.
   
   Let $C$ be an evaluation context acceptable for $G$ with public variables $V$ 
   that does not contain events used by $\corresp$.
   {\allowdisplaybreaks\begin{align*}
     \Advtev{G}{\corresp}{C}{\Dfalse} &= \Pr[C[G] : \neg\corresp \wedge \neg\nonunique{G}]\\
     &= \Pr\left[\begin{array}{@{}l@{}}
         C[G] : (\exists i_0 \in [1,n], \exists \tup{x} \in \tup{T}, \fevent{e(\tup{M})} \\
         \quad {}\wedge \psi \wedge \neg \exists \tup{y}\in\tup{T}', \phi)\wedge \neg\nonunique{G}
       \end{array}\right]\\
     &= \sum_{i_{\mathrm{val}} = 1}^n \Pr\left[\begin{array}{@{}l@{}}
         C[G] : (\exists i_0 \in [1,n], \exists \tup{x} \in \tup{T}, i_0 = i_{\mathrm{val}} \wedge \fevent{e(\tup{M})} \\
         \quad {} \wedge \psi \wedge \neg \exists \tup{y}\in\tup{T}', \phi)\wedge \neg\nonunique{G}
     \end{array}\right]\\
     &= \sum_{i_{\mathrm{val}} = 1}^n 
       \Pr\left[\begin{array}{@{}l@{}}
         C[G'] : (\exists i_0 \in [1,n], \exists \tup{x} \in \tup{T}, \fevent{e'(\tup{M})} \\
         \quad {} \wedge \psi \wedge \neg \exists \tup{y}\in\tup{T}', \phi)\wedge \neg\nonunique{G'}
         \end{array}\right]
         \text{ for } i_{\mathrm{tested}} = i_{\mathrm{val}}
       \\
     &= \sum_{i_{\mathrm{val}} = 1}^n \Advtev{G'}{\corresp'}{C}{\Dfalse}\text{ for } i_{\mathrm{tested}} = i_{\mathrm{val}}
   \end{align*}}%
   Since $\bound{G'}{V}{\corresp'}{\Dfalse}{p}$, we have $\Advtev{G'}{\corresp'}{C}{\Dfalse} \leq p(C)$ and the probability $p$ is independent of the value of $i_{\mathrm{tested}}$, so we obtain
   $\Advtev{G}{\corresp}{C}{\Dfalse} \leq n \times p(C)$.
   Therefore $\bound{G}{V}{\corresp}{\Dfalse}{n p}$.
   
   \paragraph{One-session secrecy}
   Let $C$ be an evaluation context acceptable for $C_{\secrone(x)}[G]$ with public variables $V$ that does not contain $\sevent$ nor $\sbarevent$.
   Suppose that the modified replication corresponds to the $j$-th index of variable $x$. 
   We have
   \begin{align*}
     \Advt_G^{\secrone(x)}(C) &= \Pr[C[C_{\secrone(x)}[G]] : \sevent] - \Pr[C[C_{\secrone(x)}[G]] : \sbarevent]\\
     &= \sum_{v = 1}^n \begin{array}{@{}l@{}}
\Pr[C[C_{\secrone(x)}[G]] : \sevent \wedge \vf_j = v] \\
{} - \Pr[C[C_{\secrone(x)}[G]] : \sbarevent \wedge \vf_j = v]
\end{array}\\
     &= \sum_{v = 1}^n \begin{array}{@{}l@{}}
       \Pr[C[C_{\secrone(x')}[G']] : \sevent \wedge \vf_j = v] \\
       {} - \Pr[C[C_{\secrone(x')}[G']] : \sbarevent \wedge \vf_j = v]
       \end{array}\text{ for }i_{\mathrm{tested}} = v
   \end{align*}
   because in $G'$ with $i_{\mathrm{tested}} = v$, $x'[\vf_1, \dots, \vf_m] = x[\vf_1, \dots, \vf_m]$ when $\vf_j = v$, so $C_{\secrone(x)}[G]$ behaves like $C_{\secrone(x')}[G']$ (the events added in $G'$ are not used).

   Moreover, $\Pr[C[C_{\secrone(x')}[G']] : \sevent \wedge (\vf_j$ not defined${} \vee \vf_j \neq v)] = \Pr[C[C_{\secrone(x')}[G']] : \sbarevent \wedge (\vf_j$ not defined${} \vee \vf_j \neq v)]$ when $i_{\mathrm{tested}} = v$. Indeed, when $\vf_j$ is not defined or $\vf_j \neq v = i_{\mathrm{tested}}$, the query on $\cS$ either is not executed or always yields, independently of the value of $b$. Hence, changing the value of $b$ just swaps the events $\sevent$ and $\sbarevent$. So
   \begin{align*}
     &\Pr[C[C_{\secrone(x')}[G']] : \sevent \wedge (\vf_j\text{ not defined} \vee \vf_j \neq v) \wedge b = \true] = {}\\
     &\qquad\Pr[C[C_{\secrone(x')}[G']] : \sbarevent \wedge (\vf_j\text{ not defined} \vee \vf_j \neq v) \wedge b = \false]\\
     &\Pr[C[C_{\secrone(x')}[G']] : \sbarevent \wedge (\vf_j\text{ not defined} \vee \vf_j \neq v) \wedge b = \true] = {}\\
     &\qquad \Pr[C[C_{\secrone(x')}[G']] : \sevent \wedge (\vf_j\text{ not defined} \vee \vf_j \neq v) \wedge b = \false].
   \end{align*}
   We obtain the announced result by swapping the two sides of the second equality and adding the first equality to it. (The variable $b$ is always defined when $\sevent$ or $\sbarevent$ is executed.)

   Therefore,
   \begin{align*}
     \Advt_G^{\secrone(x)}(C) &= \sum_{v = 1}^n \Pr[C[C_{\secrone(x')}[G']] : \sevent] - \Pr[C[C_{\secrone(x')}[G']] : \sbarevent]\text{ for }i_{\mathrm{tested}} = v\\
     &= \sum_{v = 1}^n \Advt_{G'}^{\secrone(x')}(C) \text{ for }i_{\mathrm{tested}} = v
   \end{align*}
   Since $G'$ satisfies the one-session secrecy of $x'$ with public variables $V$ up to probability $p$, we have $\Advt_{G'}^{\secrone(x')}(C) \leq p(C)$ and the probability $p$ is independent of the value of $i_{\mathrm{tested}}$, so we obtain $\Advt_G^{\secrone(x)}(C) \leq n \times p(C)$. Therefore, $G$ satisfies the one-session secrecy of $x$ with public variables $V$ up to probability $n p$.
   
   \paragraph{Secrecy}
Let $C$ be an evaluation context acceptable for $C_{\secr(x)}[G]$ with public variables $V$ that does not contain $\sevent$ nor $\sbarevent$.
Suppose that the modified replication corresponds to the $j$-th index of variable $x$.
   We have
   \begin{align*}
     &\Advt_G^{\secr(x)}(C) \\
     &\quad {} = \frac{1}{2}\left(\begin{array}{@{}r@{}}
       \Pr[C[G \parpop \tprocs{x}] : \sevent \mid b = \true] + \Pr[C[G \parpop \tprocs{x}] : \sevent \mid b = \false]\\
       {} - \Pr[C[G \parpop \tprocs{x}] : \sbarevent \mid b = \true] - \Pr[C[G \parpop \tprocs{x}] : \sbarevent \mid b = \false]
       \end{array}\right)
   \end{align*}
   Let 
   {\allowdisplaybreaks\begin{align*}
\tproc{\secr(x), \mathrm{real}} =\, & \cinput{\cSz}{}; \coutput{\cSz}{}; \\*
&(\repl{\iS}{\nS}\,\cinput{\cS[\iS]}{\vf_1:[1, n_1], \ldots, \vf_m:[1, n_m]};\adeftest{x[\vf_1, \ldots, \vf_m]}{}\\*
&\phantom{(}\coutput{\cS[\iS]}{x[\vf_1, \ldots, \vf_m]}\\*
&\!\!\parpop \cinput{\cS'}{b':\bool}; \bguard{b'}{\keventabort{\sevent}}{\keventabort{\sbarevent}})\\
\tproc{\secr(x), \mathrm{random}} =\, & \cinput{\cSz}{}; \coutput{\cSz}{}; \\*
&(\repl{\iS}{\nS}\,\cinput{\cS[\iS]}{\vf_1:[1, n_1], \ldots, \vf_m:[1, n_m]};\adeftest{x[\vf_1, \ldots, \vf_m]}{}\\*
&\phantom{(}\kw{find} \ {\vfS' = \iS' \leq \nS}\ \kw{suchthat} \ \defined (y[\iS'],\vf_1[\iS'], \ldots, \vf_m[\iS']) \fand {}\\*
&\phantom{(}\qquad \vf_1[\iS'] = \vf_1 \fand \ldots \fand \vf_m[\iS'] = \vf_m\\*
&\phantom{(} \kw{then}\ \coutput{\cS[\iS]}{y[\vfS']}\\*
&\phantom{(} \ELSE\Res{y}{T}; \coutput{\cS[\iS]}{y}\\*
&\!\!\parpop \cinput{\cS'}{b':\bool}; \bguard{b'}{\keventabort{\sevent}}{\keventabort{\sbarevent}})
   \end{align*}}%
   Then
   \begin{align*}
     \Advt_G^{\secr(x)}(C) &= \frac{1}{2}\left(\begin{array}{@{}r@{}}
       \Pr[C[G \parpop \tproc{\secr(x), \mathrm{real}}] : \sevent] + \Pr[C[G \parpop \tproc{\secr(x), \mathrm{random}}] : \sbarevent]\\
       {} - \Pr[C[G \parpop \tproc{\secr(x), \mathrm{real}}] : \sbarevent] - \Pr[C[G \parpop \tproc{\secr(x), \mathrm{random}}] : \sevent]
       \end{array}\right)
   \end{align*}
Let
   \begin{align*}
\tproc{\secr(x), v} =\, & \cinput{\cSz}{}; \coutput{\cSz}{}; \\
&(\repl{\iS}{\nS}\,\cinput{\cS[\iS]}{\vf_1:[1, n_1], \ldots, \vf_m:[1, n_m]};\adeftest{x[\vf_1, \ldots, \vf_m]}{}\\*
&\phantom{(}\bguard{\vf_j \leq v}{\coutput{\cS[\iS]}{x[\vf_1, \ldots, \vf_m]}}{}\\
&\phantom{(}\kw{find} \ {\vfS' = \iS' \leq \nS}\ \kw{suchthat} \ \defined (y[\iS'],\vf_1[\iS'], \ldots, \vf_m[\iS']) \fand {}\\
&\phantom{(}\qquad \vf_1[\iS'] = \vf_1 \fand \ldots \fand \vf_m[\iS'] = \vf_m\\
&\phantom{(} \kw{then}\ \coutput{\cS[\iS]}{y[\vfS']}\\
&\phantom{(} \ELSE\Res{y}{T}; \coutput{\cS[\iS]}{y}\\
&\!\!\parpop \cinput{\cS'}{b':\bool}; \bguard{b'}{\keventabort{\sevent}}{\keventabort{\sbarevent}})
\end{align*}
   The process $\tproc{\secr(x),0}$ behaves like $\tproc{\secr(x),\mathrm{random}}$ and
   $\tproc{\secr(x), n}$ behaves like $\tproc{\secr(x),\mathrm{real}}$ ($n_j = n$).
   Therefore,
   \begin{align*}
     \Advt_G^{\secr(x)}(C) &= \frac{1}{2}\left(\begin{array}{@{}r@{}}
       (\Pr[C[G \parpop \tproc{\secr(x), n}] : \sevent] - \Pr[C[G \parpop \tproc{\secr(x), 0}] : \sevent])\\
       {} - (\Pr[C[G \parpop \tproc{\secr(x), n}] : \sbarevent] - \Pr[C[G \parpop \tproc{\secr(x), 0}] : \sbarevent]) 
     \end{array}\right)\\
     &=\frac{1}{2}\left(\begin{array}{@{}r@{}}
       \sum_{v=1}^n (\Pr[C[G \parpop \tproc{\secr(x), v}] : \sevent] - \Pr[C[G \parpop \tproc{\secr(x), v-1}] : \sevent]) \\
       {} - \sum_{v=1}^n (\Pr[C[G \parpop \tproc{\secr(x), v}] : \sbarevent] - \Pr[C[G \parpop \tproc{\secr(x), v-1}] : \sbarevent])
     \end{array}\right)
   \end{align*}
\newcommand{\cSu}{c_{s1}}

   We define a context $C'_v$ that returns a random value for $\vf_j > v$,
   the real value of $x$ obtained from the public variable $x''$ in $G'$ for $\vf_j < v$,
   and calls $\tproc{\secr(x'), \mathrm{real}}$ or $\tproc{\secr(x'), \mathrm{random}}$ for $\vf_j = v$.
   \begin{align*}
     C'_v ={} &\Reschan{\cSu};( [\,] \mid \repl{\iS}{\nS}\,\cinput{\cS[\iS]}{\vf_1:[1, n_1], \ldots, \vf_m:[1, n_m]}; \\
     &\qquad \bguard{\vf_j = v}{\coutput{\cSu[\iS]}{\vf_1, \ldots, \vf_m}; \cinput{\cSu[\iS]}{z: T}; \coutput{\cS[\iS]}{z}}{}\\
     &\qquad \adeftest{x''[\vf_1, \ldots, \vf_m]}{} \\
     &\qquad \bguard{\vf_j < v}{\coutput{\cS[\iS]}{x''[\vf_1, \ldots, \vf_m]}}{}\\
     &\qquad \kw{find} \ {\vfS' = \iS' \leq \nS}\ \kw{suchthat} \ \defined (y[\iS'],\vf_1[\iS'], \ldots, \vf_m[\iS']) \fand {}\\
     &\qquad\qquad \vf_1[\iS'] = \vf_1 \fand \ldots \fand \vf_m[\iS'] = \vf_m\\
     &\qquad \kw{then}\ \coutput{\cS[\iS]}{y[\vfS']}\\
     &\qquad \ELSE\Res{y}{T}; \coutput{\cS[\iS]}{y})
   \end{align*}
   where the processes $\tproc{\secr(x'), \mathrm{real}}$ and $\tproc{\secr(x'), \mathrm{random}}$ use channel $\cSu$ instead of $\cS$.
   When $\vf_j = v$, the query on $\cS[\iS]$ is forwarded to $\tproc{\secr(x'), \mathrm{real}}$ (resp.~$\tproc{\secr(x'), \mathrm{random}}$) on channel $\cSu[\iS]$.
   The values of $x$ for sessions other than $v$ are collected in $x''$ by $G'$; these are the values returned by the query on $\cS[\iS]$ when $\vf_j < v$.
   Finally, when $\vf_j > v$, the query on $\cS[\iS]$ is answered with a random value $y$.

   Then
   \begin{align*}
     &\Pr[C[G \parpop \tproc{\secr(x),v}] : \sevent] = \Pr[C[C'_v[G' \parpop \tproc{\secr(x'), \mathrm{real}}]] : \sevent],\\
     &\Pr[C[G \parpop \tproc{\secr(x),v}] : \sbarevent] = \Pr[C[C'_v[G' \parpop \tproc{\secr(x'), \mathrm{real}}]] : \sbarevent],\\
     &\Pr[C[G \parpop \tproc{\secr(x),v-1}] : \sevent] = \Pr[C[C'_v[G' \parpop \tproc{\secr(x'), \mathrm{random}}]] : \sevent],\text{ and}\\
     &\Pr[C[G \parpop \tproc{\secr(x),v-1}] : \sbarevent] = \Pr[C[C'_v[G' \parpop \tproc{\secr(x'), \mathrm{random}}]] : \sbarevent],
   \end{align*}
   %   $G \parpop \tproc{\secr(x),v} \approx_0^V C'_v[G' \parpop \tproc{\secr(x'), \mathrm{real}}]$
   %   and $G \parpop \tproc{\secr(x),v-1} \approx_0^V C'_v[G' \parpop \tproc{\secr(x'), \mathrm{random}}]$
   %but only for distinguishers that do not use new events introduced in $G'$
   for $i_{\mathrm{tested}} = v$.

   Then
   \begin{align*}
     &\Advt_G^{\secr(x)}(C) \\
     &\quad = \sum_{v=1}^n \frac{1}{2} \left(\begin{array}{@{}r@{}}
       \Pr[C[C'_v[G' \parpop \tproc{\secr(x'), \mathrm{real}}]] : \sevent] - \Pr[C[C'_v[G' \parpop \tproc{\secr(x'), \mathrm{random}}]] : \sevent]\\
          {} - \Pr[C[C'_v[G' \parpop \tproc{\secr(x'), \mathrm{real}}]] : \sbarevent] + \Pr[C[C'_v[G' \parpop \tproc{\secr(x'), \mathrm{random}}]] : \sbarevent]
     \end{array}\right)\\*
     &\qquad\quad\text{for }i_{\mathrm{tested}} = v\\
     &\quad = \sum_{v=1}^n \Advt_{G'}^{\secr(x')}(C[C'_v]) \text{ for }i_{\mathrm{tested}} = v
   \end{align*}
   by the link between the advantage for secrecy and $\tproc{\secr(x'), \mathrm{real}}$, $\tproc{\secr(x'), \mathrm{random}}$ shown above.
   Since $G'$ satisfies the secrecy of $x'$ with public variables $V \cup \{ x'' \}$ up to probability $p$ and $C[C'_v]$ is an evaluation context acceptable for $G'\parpop \tproc{\secr(x')}$ with public variables $V \cup \{ x'' \}$ that does not contain $\sevent$ nor $\sbarevent$, we have $\Advt_{G'}^{\secr(x')}(C[C'_v]) \leq p(C[C'_v])$. Moreover, by Property~\ref{prop:prob}, the probability $p(C[C'_v])$ is independent of $i_{\mathrm{tested}} = v$ (because the type of $i_{\mathrm{tested}}$ is bounded) and depends only on the runtime of $C$, the number of outputs $C$ makes on the various channels (which determine replication bounds), and the length of bitstrings, so we have $p(C[C'_v]) = p(C)$ (the runtime of $C'_v$ can be neglected). Hence, we obtain
   $\Advt_G^{\secr(x)}(C) \leq n \times p(C)$.
   Therefore, $G$ satisfies the secrecy of $x$ with public variables $V$ up to probability $n \times p$.

   \paragraph{One-session secrecy, secrecy, and bit secrecy}
   In this lemma, bit secrecy queries do not occur. However, we reuse this proof in Lemmas~\ref{lem:guess_x} and~\ref{lem:guess_branch} where bit secrecy queries occur, so we also handle them here.
   Let $\prop$ be $\secrone(x)$, $\secr(x)$, or $\secrbit(x)$, and $\prop'$ be the same property
   with $x'$ instead of $x$. Let $V' = V$ when $\prop$ is $\secrone(x)$ or $\secrbit(x)$, and
   $V' = V \cup \{ x''\}$ when $\prop$ is $\secr(x)$.
   Since $\bound{G'}{V'\cup \{x'\}}{\prop'}{\nonunique{G'}}{p}$, then 
   by Lemma~\ref{lem:adv}, Property~\ref{item:prop:init},
   $G'$ satisfies $\prop'$ with public variables $V'$ up to probability $p'$
   such that $p'(C) = p(C[C_{\prop'}])$.
   So by the previous result for (one-session or bit) secrecy,
   $G$ satisfies $\prop$ with public variables $V$ up to probability $n p'$.
   Let $C$ be an evaluation context acceptable for $C_{\prop}[G]$ with public variables $V$ that does not contain $\sevent$ nor $\sbarevent$.
   Hence
   \begin{align*}
     &\Advtev{G}{\prop}{C[C_{\prop}[\,]]}{\Dfalse}\\
     &\quad = \Pr[C[C_{\prop}[G]] : \sevent] - \Pr[C[C_{\prop}[G]] : \sbarevent \vee \nonunique{G}]\\
     &\quad \leq \Pr[C[C_{\prop}[G]] : \sevent] - \Pr[C[C_{\prop}[G]] : \sbarevent] \\
     &\quad = \Advt_G^{\prop}(C) \leq n \times p'(C) = n \times p(C[C_{\prop}[\,]])
   \end{align*}
   Indeed, replacing $C_{\prop'}$ with $C_{\prop}$ in the argument of $p$ does not change its result, by Property~\ref{prop:prob}.
   Therefore, we have $\bound{G}{V\cup\{x\}}{\prop}{\Dfalse}{n p}$.
   \proofcomplete
   \end{proof}
   %% The bound found in the new game shows that
   %% Forall i_tested, Pr[attack & i = i_tested] <= p
   %% Pr[attack] = \sum_{i_tested \in [1,n]} Pr[attack & i = i_tested] <= n * p 
   %% For query secret, use the link between Find_Then_Guess and
   %% Real_Or_Random. https://www.di.ens.fr/~fouque/pub/2006_iee.pdf

   In the \rn{guess} transformation, we cannot modify injective correspondence queries, because two executions
   of some injective event $e$ with different indices $i$ could be mapped to the same
   events in the conclusion of the query. In general, it even does not work for
   non-injective events inside injective queries. As a counter-example, consider the query:
   $\forall i:[1,n],x:T';\fevent{e_1(i,x)} \wedge \fievent{e_2(x)} \Rightarrow \fievent{e_3()}$
   with events $e_1(i_1,x_1)$, $e_1(i_2,x_2)$, $e_2(x_1)$, $e_2(x_2)$ and $e_3$
   each executed once. This query is false: we have two executions of $e_2$ (with matching executions of $e_1$) for a single execution of $e_3$. That contradicts injectivity. However, it is
   true if we restrict ourselves to one value of $i$ (the index of the tested session), because we consider
   $e_1(i_1,x_1)$, $e_2(x_1)$ and $e_3$ for $i = i_1$ and
   $e_1(i_2,x_2)$, $e_2(x_2)$ and $e_3$ for $i = i_2$. Requiring the same
   value of $x$ in $e_1$ and $e_2$ restricts the events $e_2$ that
   we consider when we guess the session for $e_1$. Therefore,
   proving the query for the tested session does not allow us to prove it in the initial game.
   
   Hence, 
   for injective correspondence queries $\forall \tup{x}:\tup{T}; \psi \Rightarrow \exists \tup{y}:\tup{T'};\phi$, we proceed as follows:
we define a non-injective query $\noninj(\forall \tup{x}:\tup{T}; \psi \Rightarrow \exists \tup{y}:\tup{T'};\phi)$ simply obtained by replacing injective events with non-injective events,
and we try to prove that $\noninj(\forall \tup{x}:\tup{T}; \psi \Rightarrow \exists \tup{y}:\tup{T'};\phi)$ implies $\forall \tup{x}:\tup{T}; \psi \Rightarrow \exists \tup{y}:\tup{T'};\phi$ in the current game. 
This proof is a modified version of the proof of injective queries (Section~\ref{sec:injcorresp}):
we define the pseudo-formula $\coll(\psi,\phi)$ by
\begin{tabbing}
$\coll(\psi,M) = \bot$\\[1mm]
$\coll(\psi,\fevent{e(\tup{M})}) = \bot$\\[1mm]
$\coll(\psi, \fievent{e(\tup{M})} = \scoll$ such that
assuming $\psi = F_1 \wedge \dots \wedge F_m$,\\
\quad for every $\pp_1$ that executes $F_1$, \ldots, 
for every $\pp_m$ that executes $F_m$,
letting\\
\qquad $\fset = \theta_1\fset_{F_1,\pp_1} \cup \dots \cup \theta_m\fset_{F_m,\pp_m}$, \\
\qquad $\sri = \{ j \mapsto (\pp_j,\theta_j\replidx{\pp_j}) \mid F_j$ is an injective event$\}$, \\
\qquad $\vset = \fvar(\theta_1\replidx{\pp_1}) \cup \dots \cup \fvar(\theta_m\replidx{\pp_m}) \cup \{ \tup{x}, \tup{y} \}$,\\
\quad where for $j \leq m$, $\theta_j$ is a renaming of $\replidx{\pp_j}$ to fresh replication indices,\\
\quad we have $(\fset, (\tup{M}), \sri, \vset) \in \scoll$.\\[1mm]
$\coll(\psi, \phi_1 \wedge \phi_2) = \coll(\psi, \phi_1) \wedge \coll(\psi, \phi_2)$\\[1mm]
$\coll(\psi, \phi_1 \vee \phi_2) = \coll(\psi, \phi_1) \vee \coll(\psi, \phi_2)$ 
\end{tabbing}
Given a pseudo-formula $\coll$, we define $\vdash\coll$ as in Section~\ref{sec:injcorresp}.

%program points of events made formal by adding the program point and value of replication indices as argument to each event (see proof of injective correspondences)
\begin{proposition}\label{prop:corresp-inj}
Let $\forall \tup{x}:\tup{T}; \psi \Rightarrow \exists \tup{y}:\tup{T'}; \phi$ be a 
correspondence, with $\tup{x} = \fvar(\psi)$
and $\tup{y} = \fvar(\phi) \setminus \fvar(\psi)$.
Let $\corresp = \sem{\forall \tup{x}:\tup{T}; \psi \Rightarrow \exists \tup{y}:\tup{T'}; \phi}$ be the semantics the
correspondence $\forall \tup{x}:\tup{T}; \psi \Rightarrow \exists \tup{y}:\tup{T'}; \phi$ (Definition~\ref{def:icorresp}),
and $\nicorresp = \sem{\noninj(\forall \tup{x}:\tup{T}; \psi \Rightarrow \exists \tup{y}:\tup{T'}; \phi)}$ be the semantics of the
correspondence $\noninj(\forall \tup{x}:\tup{T}; \psi \Rightarrow \exists \tup{y}:\tup{T'}; \phi)$.
Let $Q_0$ be a process that satisfies Properties~\ref{prop:notables} and~\ref{prop:autosarename}.
Suppose that, in $Q_0$, the arguments of the events that occur in $\psi$ are always simple terms.

Assume that $\vdash \coll(\psi, \phi)$ and
for all evaluation contexts $C$ acceptable for $Q_0$,
$\Prss{C[Q_0]}{\neg \yctran{\vdash\coll(\psi, \phi)}} \leq p(C)$.
If $\bound{Q_0}{V}{\nicorresp}{\Dfalse}{p'}$, then $\bound{Q_0}{V}{\corresp}{\ab \Dfalse}{\ab p' + p}$.
%% For all evaluation contexts $C$ acceptable for $Q_0$ with public variables $V$
%% that does not contain events used by $\corresp$,
%% \[\Advtev{Q_0}{\corresp}{C}{\Dfalse} \leq \Pr[C[Q_0] : \neg \yctran{\vdash\coll}] + \Advtev{Q_0}{\nicorresp}{C}{\Dfalse}\,.\]
\end{proposition}

Proposition~\ref{prop:corresp-inj} proves injectivity much like in Section~\ref{sec:injcorresp}, but using the arguments $\tup{M}$ of the events in $\phi$
instead of the program points and replication indices at their execution (which we do not have since in $Q_0$ we are not able to prove that these events have been executed; otherwise we would simply prove the correspondence in $Q_0$).
Intuitively, $\vdash \coll(\psi, \phi)$ shows that, if we have different executions of injective events in $\psi$, that is, executions of such events with different pairs (program point, replication indices), then the arguments $\tup{M}$ of each injective event in $\phi$ must be different, which implies different executions of this event.

\begin{proof}
Let $C$ be an evaluation context acceptable for $Q_0$ with public variables $V$
that does not contain events used by $\corresp$.
Let $\coll =\coll(\psi,\phi)$.
Consider a trace $\trace$ of $C[Q_0]$ such that $\trace \vdash \nicorresp$, $\trace \vdash \yctran{\vdash\coll}$, $\trace$ does not execute a non-unique event of $Q_0$, and the last configuration of $\trace$ cannot be reduced. Let $\evseq$ be the sequence of events in the last configuration of $\trace$.
Since $\evseq \vdash \nicorresp$,
\[
\evseq \vdash \forall \step_1, \dots, \step_m\in \mathbb{N}, \forall \tup{x}\in\tup{T}, (\psistep \Rightarrow \exists \tup{y}\in\tup{T}', \phi)
\]
with the notations of Definition~\ref{def:icorresp}.
By defining functions $\Fstepfun_j \in \mathbb{N}^m \times \prod \tup{T} \rightarrow \mathbb{N} \cup \{\bot\}$
that map $\step_1$, \dots, $\step_m$, $\tup{x}$ to the execution steps of injective events in the proof of $\phi$,
and to $\bot$ when the event is not used in the proof of $\phi$,
we have
\begin{equation}
\evseq \vdash \exists \Fstepfun_1, \dots, \Fstepfun_k\in \mathbb{N}^m \times \prod \tup{T} \rightarrow \mathbb{N} \cup \{\bot\},
\forall \step_1, \dots, \step_m\in \mathbb{N}, \forall \tup{x}\in\tup{T}, (\psistep \Rightarrow \exists \tup{y}\in\tup{T}', \phistep)
\label{eq:f1}
\end{equation}
It remains to show $\Inj(I,\Fstepfun)$ for all $\Fstepfun \in \{ \Fstepfun_1, \dots, \Fstepfun_k\}$.

Let $\Fstepfun \in \{ \Fstepfun_1, \dots, \Fstepfun_k\}$  and $\fevent{e(\tup{M})}@\Fstepfun(\step_1, \dots, \step_m, \tup{x})$
be the event labeled with $\Fstepfun$ in $\phistep$.
Suppose that $\Fstepfun(\step'_1, \dots, \step'_m, \tup{a}') = \Fstepfun(\step''_1, \dots, \step''_m, \tup{a}'') \neq \bot$ and 
there exists $j \in I$ such that $\step'_j \neq \step''_j$, and let us prove a contradiction.

Since $\Fstepfun(\step'_1, \dots, \step'_m, \tup{a}') \neq \bot$, 
$\fevent{e(\tup{M})}@\Fstepfun(\step'_1, \dots, \step'_m, \tup{a}')$ is used in the proof of~\eqref{eq:f1},
so letting  $\venv_1 = \{ \step_1 \mapsto \step'_1, \dots, \step_m \mapsto \step'_m, \tup{x} \mapsto \tup{a}' \}$,
$\venv_1, \evseq \vdash \psistep$ and there exists an extension $\venv'_1$ of $\venv_1$ to $\tup{y}$ such that 
$\venv'_1, \evseq \vdash \phistep$ and $\venv'_1, \evseq \vdash \fevent{e(\tup{M})}@\Fstepfun(\step_1, \dots, \step_m, \tup{x})$.

Since $\venv_1, \evseq \vdash \psistep$, for all events $F_\ell = \fevent{e_\ell(\tup{M}_\ell)}@\step_\ell$ in $\psistep$, $\trace, \venv_1 \vdash F_\ell$
and $\evseq(\step'_\ell) = (\dots):e_\ell(\tup{a}_{\ell,1})$ for $\tup{a}_{\ell,1}$ such that
$\venv_1, \tup{M}_\ell \evalterm \tup{a}_{\ell,1}$. 
By Lemma~\ref{lem:fsetFP}, there exists a program point $\pp_{\ell,1}$ that executes $F_\ell$ (in $Q_0$) and a case $\case_{\ell,1}$ such that, for any $\theta_{\ell,1}$ renaming of $\replidx{\pp_{\ell,1}}$ to fresh replication indices, there exists a mapping $\sigma_{\ell,1}$ with domain $\theta_{\ell,1}\replidx{\pp_{\ell,1}}$ such that $\evseq(\venv_1(\step_\ell)) = (\pp_{\ell,1}, \sigma_{\ell,1}(\theta_{\ell,1}\replidx{\pp_{\ell,1}})):\dots$ and $\trace, \sigma_{\ell,1}\cup \venv_1 \vdash \theta_{\ell,1}\fset_{F_\ell, \pp_{\ell,1},\case_{\ell,1}}$.
So $\evseq(\step'_\ell) = (\pp_{\ell,1}, \sigma_{\ell,1}(\theta_{\ell,1}\replidx{\pp_{\ell,1}})):e_\ell(\tup{a}_{\ell,1})$.

Let $\scoll$ be the label of $\coll$ at the occurrence corresponding to $\Fstepfun$,
and $\fievent{e(\tup{M})}$ be the injective event at that occurrence in $\phi$.
Let $\fset_1 = \bigcup_{\ell} \theta_{\ell,1}\fset_{F_\ell, \pp_{\ell,1},\case_{\ell,1}}$,
$\sri_1 = \{ \ell \mapsto (\pp_{\ell,1},\theta_{\ell,1}\replidx{\pp_{\ell,1}}) \mid F_\ell$ is an injective event$\}$, and
$\vset_1 = \fvar(\theta_{1,1}\replidx{\pp_{1,1}}) \cup \dots \cup \fvar(\theta_{m,1}\replidx{\pp_{m,1}}) \cup \{\tup{x},\tup{y}\}$.
By construction of $\coll$, we have $(\fset_1, (\tup{M}), \sri_1, \vset_1) \in \scoll$.

So we have $\trace, \bigcup_\ell \sigma_{\ell,1}\cup \venv_1 \vdash \bigcup_\ell \theta_{\ell,1}\fset_{F_\ell, \pp_{\ell,1},\case_{\ell,1}}$.
Letting $\sigma_1 = \bigcup_\ell \sigma_{\ell,1}$, we have $\trace, \sigma_1\cup \venv'_1 \vdash \fset_1$;
for $\ell$ such that $F_\ell$ is an injective event, $\evseq(\step'_{\ell}) = \sigma_1\sri_1(\ell):e_{\ell}(\dots)$;
$\vset_1 = \dom(\sigma_1) \cup \{\tup{x},\tup{y}\}$.
Since $\venv'_1, \evseq \vdash \fevent{e(\tup{M})}@\Fstepfun(\step_1, \ab \dots, \ab \step_m, \ab \tup{x})$, we have
$\venv'_1, \tup{M} \evalterm \tup{a}_1$ and  $\evseq(\Fstepfun(\step'_1, \dots, \step'_m, \tup{a}')) = (\dots):e(\tup{a}_1)$ for some $\tup{a}_1$.

Since $\Fstepfun(\step''_1, \dots, \step''_m, \tup{a}'') \neq \bot$, we have similarly $(\fset_2, (\tup{M}), \sri_2, \vset_2) \in \scoll$, $\venv'_2$, and $\sigma_2$ such that $\trace, \sigma_2\cup \venv'_2 \vdash \fset_2$;
for $\ell$ such that $F_\ell$ is an injective event, $\evseq(\step''_{\ell}) = \sigma_2\sri_2(\ell):e_{\ell}(\dots)$;
$\vset_2 = \dom(\sigma_2) \cup \{\tup{x},\tup{y}\}$;
$\venv'_2, \tup{M} \evalterm \tup{a}_2$ and $\evseq(\Fstepfun(\step''_1, \dots, \step''_m, \tup{a}'')) = (\dots):e(\tup{a}_1)$ for some $\tup{a}_2$.

Let $\theta''$ be a renaming of variables in $\vset_2$.
Then $\trace, \sigma_2\theta''^{-1}\cup \venv'_2\theta''^{-1} \vdash \theta''\fset_2$;
for $\ell$ such that $F_\ell$ is an injective event, $\evseq(\step''_{\ell}) = \sigma_2\theta''^{-1} \theta''\sri_2(\ell):e_{\ell}(\dots)$;
$\venv'_2\theta''^{-1}, \theta''\tup{M} \evalterm \tup{a}_2$ and $\evseq(\Fstepfun(\step''_1, \ab \dots, \ab \step''_m, \ab \tup{a}'')) = (\dots):e(\tup{a}_2)$ for some $\tup{a}_2$.

Then $\trace, \sigma_1 \cup \sigma_2\theta''^{-1}\cup \venv'_1 \cup \venv'_2\theta''^{-1} \vdash \fset_1 \cup \theta''\fset_2$.

There exists $j\in I$ such that $\step'_j \neq \step''_j$, so  
$\sigma_1\sri_1(j)\neq \sigma_2\theta''^{-1} \theta''\sri_2(j)$ (distinct events have distinct pairs (program point, replication indices) by Lemma~\ref{lem:oneevent}), so
there exists $j \in \dom(\sri_1) = I$ such that
$\trace, \sigma_1 \cup \sigma_2\theta''^{-1}\cup \venv'_1 \cup \venv'_2\theta''^{-1} \vdash \sri_1(j)\neq \theta''\sri_2(j)$.

Since $\Fstepfun(\step'_1, \dots, \step'_m, \tup{a}') = \Fstepfun(\step''_1, \dots, \step''_m, \tup{a}'')$, 
$\evseq(\Fstepfun(\step'_1, \dots, \step'_m, \tup{a}')) = \evseq(\Fstepfun(\step''_1, \dots, \step''_m, \tup{a}''))$, so
$\tup{a}_1 = \tup{a}_2$, so
$\trace, \sigma_1 \cup \sigma_2\theta''^{-1}\cup \venv'_1 \cup \venv'_2\theta''^{-1} \vdash \tup{M} = \theta''\tup{M}$.

So $\trace, \sigma_1 \cup \sigma_2\theta''^{-1}\cup \venv'_1 \cup \venv'_2\theta''^{-1} \vdash \fset_1 \cup \theta''\fset_2 \cup \{ \bigvee_{j \in \dom(\sri_1)} \sri_1(j)\neq \theta''\sri_2(j), \tup{M} = \theta''\tup{M} \}$.

Since the trace satisfies $\yctran{\vdash\coll}$, this is a contradiction. 
Therefore, we conclude that the considered trace satisfies $\corresp$.
Hence, every full trace of $C[Q_0]$ that satisfies $\nicorresp$, $\yctran{\vdash\coll}$, and does not execute a non-unique event of $Q_0$ also satisfies $\corresp$.
Therefore, every full trace of $C[Q_0]$ that satisfies $\neg \corresp$
satisfies $\neg (\nicorresp \wedge \yctran{\vdash\coll} \wedge \neg \nonunique{Q_0})$, so every full trace of $C[Q_0]$ that satisfies $\neg \corresp \wedge \neg \nonunique{Q_0}$ satisfies $\neg (\nicorresp \wedge \yctran{\vdash\coll} \wedge \neg \nonunique{Q_0})\wedge \neg \nonunique{Q_0} = (\neg \nicorresp \vee \neg \yctran{\vdash\coll}) \wedge \neg \nonunique{Q_0} = (\neg \nicorresp  \wedge \neg \nonunique{Q_0}) \vee ( \neg \yctran{\vdash\coll} \wedge \neg \nonunique{Q_0,\Dfalse})$, so it satisfies $(\neg \nicorresp  \wedge \neg \nonunique{Q_0}) \vee \neg \yctran{\vdash\coll}$. 
So
\begin{align*}
\Advtev{Q_0}{\corresp}{C}{\Dfalse}
&= \Pr[C[Q_0]: \neg \corresp \wedge \neg \nonunique{Q_0}] \\
&\leq \Pr[C[Q_0] : (\neg \nicorresp  \wedge \neg \nonunique{Q_0}) \vee \neg \yctran{\vdash\coll}]\\
&\leq \Pr[C[Q_0] : \neg \yctran{\vdash\coll}] + \Pr[C[Q_0]: \neg\nicorresp\wedge \neg \nonunique{Q_0})]\\
&\leq \Pr[C[Q_0] : \neg \yctran{\vdash\coll}] + \Advtev{Q_0}{\nicorresp}{C}{\Dfalse}\\
&\leq \Pr[C[Q_0] : \neg \yctran{\vdash\coll}] + p(C)\tag*{since $\bound{Q_0}{V}{\nicorresp}{\Dfalse}{p}$}\\
&\leq p'(C)\,.
\end{align*}
Therefore $\bound{Q_0}{V}{\corresp}{\Dfalse}{p'}$.
\proofcomplete
\end{proof}
If the proof that $\noninj(\forall \tup{x}:\tup{T}; \psi \Rightarrow \exists \tup{y}:\tup{T'}; \phi)$ implies $\forall \tup{x}:\tup{T}; \psi \Rightarrow \exists \tup{y}:\tup{T'}; \phi$ works, we 
just have to prove $\noninj(\forall \tup{x}:\tup{T}; \psi \Rightarrow \exists \tup{y}:\tup{T'}; \phi)$ and we can apply the \rn{guess} transformation for non-injective correspondences.
Otherwise, we simply leave the query $\forall \tup{x}:\tup{T}; \psi \Rightarrow \exists \tup{y}:\tup{T'}; \phi$ unchanged.

A transformation \rn{guess} $i$ \rn{\&\&\ above}, similar to \rn{guess} $i$, can be used to guess the whole sequence $\tup{i}$ of replication indices above and including the modified replication, by testing the equality $\tup{i} = \tup{i}_{\mathrm{tested}}$ instead of $i = i_{\mathrm{tested}}$.

\subsubsection{\rn{guess} $x[c_1,\dots,c_m]$}\label{sec:guessx}

Like the transformation \rn{guess} $i$,
when \rn{guessRemoveUnique = true} and some (one-session or bit) secrecy
queries are present, the transformation \rn{guess} $x[c_1,\dots,c_m]$ first
transforms the game $G$ into $\Gru$, by replacing all proved
$\FIND\unique{}$ with $\FIND$. Lemma~\ref{lem:Gru} shows the soundness
of this preliminary transformation.

Next, the transformation \rn{guess} $x[c_1,\dots,c_m]$ transforms a game $G$ into a game $G'$
by guessing the value of a variable $x[c_1,\dots,c_m]$: it replaces
the processes $P$ under the definition of $x[c_1,\dots,c_m]$
with
\[\bguard{x[c_1,\dots,c_m] = v_{\mathrm{tested}}}{P}{\kw{event\_abort}\ \kwf{bad\_guess}}\]
and $v_{\mathrm{tested}}$ is a constant, which is the guessed value of
$x[c_1,\dots,c_m]$.
   (At each definition of $x$, CryptoVerif must be able to determine
   whether it is a definition of $x[c_1,\dots,c_m]$ or not.
   The variable $x$ must not be defined inside a term.
   In case \rn{diff\_constants = true}, the constant
   $v_{\mathrm{tested}}$ must not be considered different from other constants
   of the same type.)

   In case there is a (one-session or bit) secrecy query, it uses instead
\[\assign{\mathit{guess\_x\_defined}}{\true}\bguard{x[c_1,\dots,c_m] = v_{\mathrm{tested}}}{P}{\kw{event\_abort}\ \kwf{bad\_guess}}\]
  where $\mathit{guess\_x\_defined}$ is a fresh variable, and we add $\mathit{guess\_x\_defined}$ to the public variables of (one-session or bit) secrecy queries. That gives the adversary knowledge of whether the guessed variable is defined or not. This is useful because the adversary may need to swap its answer differently depending on whether the guessed variable is defined or not, so that the cases in which this variable is not defined always increase the probability of breaking (one-session or bit) secrecy.

\medskip
  When there are only correspondence queries, we can actually execute any code when $x[c_1,\ab\dots,\ab c_m]$ is different from the guessed value $v_{\mathrm{tested}}$. In particular, we can execute $P$ with $x[c_1,\ab \dots,\ab c_m]$ set to $v_{\mathrm{tested}}$, which has the effect of replacing the definition of $x[c_1,\ab \dots,\ab c_m]$ with $\bassign{x}{v_{\mathrm{tested}}}$ and removing the test $x[c_1,\dots,c_m] = v_{\mathrm{tested}}$.

  To sum up, we also define a transformation \rn{guess} $x[c_1,\dots,c_m]$ \rn{no\_test} that can be applied when there are only correspondence queries and when $x[c_1,\dots,c_m]$ is defined only by definitions of the form $\bassign{x}{M}$. (This is the most useful case, since the definition of $x$ can then be simplified.) This transformation replaces these definitions $\bassign{x}{M}$ with
  $\bassign{x}{v_{\mathrm{tested}}}$ when $M$ is a simple term and with
  $\assign{\kvar{ignore}}{M}\bassign{x}{v_{\mathrm{tested}}}$ otherwise,
  where $\kvar{ignore}$ is a fresh variable whose value is not used.

The transformation \rn{guess} $x[c_1,\dots,c_m]$ \rn{no\_test} would not be valid in the presence of secrecy queries (at least not with the same probability), because before transformation the value of $x[c_1,\ab \dots, \ab c_m]$ may contain part of the secret variable and this value may leak, while after transformation, that leaks disappears and the variable may be perfectly secret for all values $x[c_1,\dots,c_m] = v_{\mathrm{tested}}$.

  \sloppy

\begin{lemma}\label{lem:guess_x}
    The transformations \rn{guess} $x[c_1,\dots,c_m]$ and \rn{guess} $x[c_1,\dots,c_m]$ \rn{no\_test} require and preserve
    Properties~\ref{prop:nointervaltypes}, \ref{prop:noreschan},
    \ref{prop:channelindices}, \ref{prop:notables}, and~\ref{prop:autosarename}.
    They preserve Property~\ref{prop:expand}.

Suppose the game $G$ is transformed into $G'$ by the transformation $\rn{guess}\ x[c_1,\dots,c_m]$ or \rn{guess} $x[c_1,\dots,c_m]$ \rn{no\_test}, where $x$ is of type $T$ and $T \neq \emptyset$. 

Let $\corresp$ be the semantics of a correspondence.
Let $D$ be a disjunction of Shoup and non-unique events that does not contain $\sevent$ nor $\sbarevent$.
If $\bound{G'}{V}{\corresp}{D}{p}$ and $p$ is independent of the value of $v_{\mathrm{tested}}$, then $\bound{G}{V}{\corresp}{D}{|T| p}$.
%% Let $C$ be an evaluation context acceptable for $G$ with any public
%% variables that does not contain events used by $\corresp$ or $D$.
%% If $\Advtev{G'}{\corresp}{C}{D} \leq p(C)$ and $p$ is independent of the value of $v_{\mathrm{tested}}$, then $\Advtev{G}{\corresp}{C}{D} \leq |T| \times p(C)$.

Let $\prop$ be $\secrone(y)$, $\secr(y)$, or $\secrbit(x)$.
Then $G$ is transformed into $G'$ by the transformation $\rn{guess}\ x[c_1,\dots,c_m]$.
If $G'$ satisfies $\prop$ with public variables $V \cup \{ \mathit{guess\_x\_defined} \}$ ($y \notin V$) up to probability $p$ and $p$ satisfies Property~\ref{prop:prob}, then
$G$ satisfies $\prop$ with public variables $V$ up to probability $|T| \times p$ (neglecting a small additional runtime of the context).
If $\bound{G'}{V' \cup \{ \mathit{guess\_x\_defined} \}}{\prop}{\nonunique{G'}}{p}$ and $p$ satisfies Property~\ref{prop:prob}, then $\bound{G}{V'}{\prop}{\Dfalse}{|T| p}$ (neglecting a small additional runtime of the context).
%%      If $\Advtev{G'}{\prop}{C[C_{\prop}[\,]]}{\nonunique{G'}} \leq p(C)$
%%      for every evaluation context $C$ acceptable for $C_{\prop}[G']$ with public variables $V \cup \{ \mathit{guess\_x\_defined} \}$ that does not contain $\sevent$ nor $\sbarevent$ and $p$ satisfies the same condition, then $\Advtev{G}{\prop}{C[C_{\prop}[\,]]}{\Dfalse} \leq |T| \times p(C)$ for every evaluation context $C$ acceptable for $C_{\prop}[G]$ with public variables $V$ that does not contain $\sevent$ nor $\sbarevent$ (neglecting a small additional runtime of the context).
\end{lemma}
\fussy
\begin{proof}

   \paragraph{Correspondences}
Let $C$ be an evaluation context acceptable for $G$ with any public
variables that does not contain events used by $\corresp$ or $D$.
   \begin{align*}
     &\Advtev{G}{\corresp}{C}{D} \\
     &\quad = \Pr[C[G] : (\neg\corresp \vee D) \wedge\neg\nonunique{G,D}]\\
     &\quad= \Pr[C[G] : (\neg\corresp \vee D)  \wedge\neg\nonunique{G,D} \wedge x[c_1,\dots,c_m] \text{ not defined}]  + {}\\
     &\quad\phantom{{}={}}\sum_{v \in T} \Pr[C[G] : (\neg\corresp \vee D) \wedge\neg\nonunique{G,D} \wedge x[c_1,\dots,c_m] = v]\\
     &\quad\leq \sum_{v \in T} \Pr[C[G] : (\neg\corresp \vee D) \wedge\neg\nonunique{G,D} \wedge (x[c_1,\dots,c_m] = v \vee x[c_1,\dots,c_m] \text{ not defined})]
   \end{align*}
   Moreover,
   \[\begin{split}
   &\Pr[C[G'] : (\neg\corresp \vee D) \wedge\neg\nonunique{G',D}]\\
   &\ \, \geq \Pr[C[G] : (\neg\corresp \vee D) \wedge\neg\nonunique{G,D} \wedge (x[c_1,\dots,c_m] = v_{\mathrm{tested}} \vee x[c_1,\dots,c_m] \text{ not defined})]
   \end{split}\]
   This property holds because $G'$ behaves like $G$ when $x[c_1,\dots,c_m]$ is not defined or $x[c_1,\dots,c_m] = v_{\mathrm{tested}}$, in both transformations \rn{guess} $x[c_1,\dots,c_m]$ and \rn{guess} $x[c_1,\dots,c_m]$ \rn{no\_test}.
   So      
   \begin{align*}
     \Advtev{G}{\corresp}{C}{D} 
     &\leq \sum_{v \in T} \Pr[C[G'] : (\neg\corresp \vee D) \wedge\neg\nonunique{G',D}] \text{ for } v_{\mathrm{tested}} = v\\
     &\leq \sum_{v \in T} \Advtev{G'}{\corresp}{C}{D}\text{ for } v_{\mathrm{tested}} = v
   \end{align*}
   Since $\bound{G'}{V}{\corresp}{D}{p}$, we have $\Advtev{G'}{\corresp}{C}{D} \leq p(C)$ and $p$ is independent of the value of $v_{\mathrm{tested}}$, so we obtain
   $\Advtev{G}{\corresp}{C}{\ab D} \leq |T| \times p(C)$.
   Therefore $\bound{G}{V}{\corresp}{D}{|T| p}$.
   
   \paragraph{(One-session or bit) secrecy}
   Let $C$ be an evaluation context acceptable for $C_{\prop}[G]$ with public variables $V$ that does not contain $\sevent$ nor $\sbarevent$.
   We have
   \begin{align*}
     &\Advt_G^{\prop}(C)\\
     &= \Pr[C[C_{\prop}[G]] : \sevent] - \Pr[C[C_{\prop}[G]] : \sbarevent]\\
     &=  \Pr[C[C_{\prop}[G]] : \sevent \wedge x[c_1,\dots,c_m] \text{ not defined}] - \Pr[C[C_{\prop}[G]] : \sbarevent \wedge x[c_1,\dots,c_m] \text{ not defined}]\\
     &+ \sum_{v = 1}^{|T|} \Pr[C[C_{\prop}[G]] : \sevent \wedge x[c_1,\dots,c_m] = v] - \Pr[C[C_{\prop}[G]] : \sbarevent \wedge x[c_1,\dots,c_m] = v]
   \end{align*}
   From the adversary $C$, we define four adversaries $C'$ that output $b''$ instead of $b'$ on channel $\cS'$ ($\cS''$ for bit secrecy), where $b'' = \bguard{\kw{defined}(\mathit{guess\_x\_defined})}{f_1(b')}{f_2(b')}$ where $f_1(b')$ is either $b'$ or $\neg b'$, and similarly for $f_2$, and consider the adversary $C'_{\max,v_{\mathrm{tested}}}$ among those four that yields the maximum $\Advt_{G'}^{\prop}(C')$. Changing $b'$ into $\neg b'$ swaps the events $\sevent$ and $\sbarevent$, and therefore swaps their probabilities. Hence, for this adversary $C'_{\max,v_{\mathrm{tested}}}$,
   \begin{align*}
     &\Advt_{G'}^{\prop}(C'_{\max,v_{\mathrm{tested}}})\\
     &= | \Pr[C[C_{\prop}[G]] : \sevent \wedge x[c_1,\dots,c_m] = v_{\mathrm{tested}}] -
     \Pr[C[C_{\prop}[G]] : \sbarevent \wedge x[c_1,\dots,c_m] = v_{\mathrm{tested}}] |\\
     &+ | \Pr[C[C_{\prop}[G]] : \sevent \wedge x[c_1,\dots,c_m] \text{ not defined}] -
     \Pr[C[C_{\prop}[G]] : \sbarevent \wedge x[c_1,\dots,c_m] \text{ not defined}] |
   \end{align*}
   Since $G'$ satisfies $\prop$ with public variables $V \cup \{ \mathit{guess\_x\_defined} \}$ ($y \notin V$) up to probability $p$ and $C'_{\max,v_{\mathrm{tested}}}$ is an evaluation context acceptable for $C_{\prop}[G']$ with public variables $V \cup \{ \mathit{guess\_x\_defined} \}$ that does not contain $\sevent$ nor $\sbarevent$, we have $\Advt_{G'}^{\prop}(C'_{\max,v_{\mathrm{tested}}}) \leq p(C'_{\max,v_{\mathrm{tested}}})$. So we have
   \begin{align*}
   \sum_{v_{\mathrm{tested}} = 1}^{|T|} p(C'_{\max,v_{\mathrm{tested}}})
   &\geq \sum_{v_{\mathrm{tested}} = 1}^{|T|} \Advt_{G'}^{\prop}(C'_{\max,v_{\mathrm{tested}}})\\
   &\geq \sum_{v_{\mathrm{tested}} = 1}^{|T|}
   \left|\begin{array}{l}
   \Pr[C[C_{\prop}[G]] : \sevent \wedge x[c_1,\dots,c_m] = v_{\mathrm{tested}}] \\
      {} - \Pr[C[C_{\prop}[G]] : \sbarevent \wedge x[c_1,\dots,c_m] = v_{\mathrm{tested}}]
   \end{array}\right|\\
   &\qquad {} + |T| \times \left|\begin{array}{l}
   \Pr[C[C_{\prop}[G]] : \sevent \wedge x[c_1,\dots,c_m] \text{ not defined}]\\
      {} - \Pr[C[C_{\prop}[G]] : \sbarevent \wedge x[c_1,\dots,c_m] \text{ not defined}]
   \end{array}\right|\\
   &\geq \Advt_G^{\prop}(C)
   \end{align*}
   Moreover, by Property~\ref{prop:prob}, $p$ is independent of the value of $v_{\mathrm{tested}}$ (since the type $T$ is bounded) and $p(C)$ depends only on the runtime of $C$, the number of outputs $C$ makes on the various channels (which determine replication bounds), and the length of bitstrings, so we have $p(C'_{\max,v_{\mathrm{tested}}}) = p(C)$. (The additional runtime of the context can be neglected.) So $\Advt_G^{\prop}(C) \leq |T| p(C)$. Therefore, $G$ satisfies $\prop$ with public variables $V$ up to probability $|T| \times p$.
   The proof of the second property for (one-session or bit) secrecy proceeds as for $\rn{guess}\ i$ in Lemma~\ref{lem:guess_i}.
   \proofcomplete
   \end{proof}
   
\subsubsection{\rn{guess\_branch}}\label{sec:guess_branch}

Like the transformation \rn{guess} $i$,
when \rn{guessRemoveUnique = true} and some (one-session or bit) secrecy
queries are present, the transformation \rn{guess\_branch} $\pp$ first
transforms the game $G$ into $\Gru$, by replacing all proved
$\FIND\unique{}$ with $\FIND$. Lemma~\ref{lem:Gru} shows the soundness
of this preliminary transformation.

Next, the transformation \rn{guess\_branch} $\pp$ guesses the branch taken by a branching instruction ($\kw{if}$, $\kw{let}$, $\kw{find}$) at program point $\pp$. The program point $\pp$
is designated as explained in Section~\ref{sec:pp}. The instruction at $\pp$ must be executed at most once (either because it is not under replication or because this is proved by CryptoVerif, showing that two executions with distinct replication indices lead to a contradiction: $\fset_{\pp} \cup \fset_{\pp}\{\tup{i}'/\tup{i}\} \cup \{ \tup{i}' \neq \tup{i}\}$ yields a contradiction, where $\tup{i}$ are the current replication indices at $\pp$ and $\tup{i}'$ are fresh replication indices, using a mode of the equational prover of Section~\ref{sec:equationalprover} that does not allow elimination of collisions, so that this property is proved without probability loss). Suppose this instruction has $k$ branches. 

We consider a game $G$ and define transformed games $G'_j$ ($0 \leq j < k$) in which branch $j$ of the instruction at $\pp$ is kept and all other branches are replaced with $\kw{event\_abort}\ \kwf{bad\_guess}$.

In case there is a (one-session or bit) secrecy query, $\bassign{\mathit{guess\_br\_defined}}{\true}\ \IN$ is added before the instruction at $\pp$ where $\mathit{guess\_br\_defined}$ is a fresh variable, and we add $\mathit{guess\_br\_defined}$ to the public variables of (one-session or bit) secrecy queries. That gives the adversary knowledge of whether the instruction at $\pp$ is executed or not. This is useful because the adversary may need to swap its answer differently depending on whether that instruction is executed or not, so that the cases in which that instruction is not executed always increase the probability of breaking (one-session or bit) secrecy.

\medskip
  When there are only correspondence queries, we can actually execute any code when the taken branch is different from the guessed one. In particular, we can execute the same code as in the tested branch, which has the effect of removing the test at $\pp$ when that test is $\kw{if}$. (The tests $\kw{find}$ and $\kw{let}$ with pattern-matching have additional effects: guaranteeing the definition of variables for $\kw{find}$; defining variables for $\kw{let}$. In general, that prevents their removal.)

  To sum up, we also define a transformation \rn{guess\_branch} $\pp$ \rn{no\_test} that can be applied when there are only correspondence queries and the instruction at $\pp$ is $\bguard{M}{P_1}{P_0}$. This transformation defines two transformed games $G'_j$ ($j \in \{0,1\}$) in which the instruction at $\pp$ is replaced with $P_j$ when $M$ is a simple term and with
  $\assign{\kvar{ignore}}{M} P_j$ otherwise,
  where $\kvar{ignore}$ is a fresh variable whose value is not used.

The transformation \rn{guess\_branch} $\pp$ \rn{no\_test} would not be valid in the presence of secrecy queries (at least not with the same probability), because before transformation the test at $\pp$ may make the value of $M$ leak, which can reveal for instance one bit of the secret variable, while after transformation, that leaks disappears and the variable may be perfectly secret both when $P_0$ and when $P_1$ are executed.
  
  \sloppy

\begin{lemma}\label{lem:guess_branch}
    The transformations \rn{guess\_branch} $\pp$ and \rn{guess\_branch} $\pp$ \rn{no\_test} require and preserve
    Properties~\ref{prop:nointervaltypes}, \ref{prop:noreschan},
    \ref{prop:channelindices}, \ref{prop:notables}, and~\ref{prop:autosarename}.
    They preserve Property~\ref{prop:expand}.

Suppose the game $G$ is transformed into games $G'_j$ ($0 \leq j < k$) by the transformation \rn{guess\_branch} $\pp$ or \rn{guess\_branch} $\pp$ \rn{no\_test}.

Let $\corresp$ be the semantics of a correspondence.
Let $D$ be a disjunction of Shoup and non-unique events that does not contain $\sevent$ nor $\sbarevent$.
If for all $0 \leq j < k$, $\bound{G'_j}{V}{\corresp}{D}{p_j}$,
then $\bound{G}{V}{\corresp}{D}{\sum_{j=0}^{k-1} p_j}$.
%% Let $C$ be an evaluation context acceptable for $G$ with any public
%% variables that does not contain events used by $\corresp$ or $D$.
%% We have $\Advtev{G}{\corresp}{C}{D} \leq \sum_{j = 0}^{k-1} \Advtev{G'_j}{\corresp}{C}{D}$.

Let $\prop$ be $\secrone(y)$, $\secr(y)$, or $\secrbit(x)$.
Then $G$ is transformed into $G'_j$ by the transformation \rn{guess\_branch} $\pp$.
If $G'_j$ satisfies $\prop$ with public variables $V \cup \{ \mathit{guess\_br\_defined} \}$ ($y \notin V$) up to probability $p_j$ for $0 \leq j < k$ and the probabilities $p_j$ satisfy Property~\ref{prop:prob}, then $G$ satisfies $\prop$ with public variables $V$ up to probability $\sum_{j = 0}^{k-1} p_j$ (neglecting a small additional runtime of the context).
If for all $0 \leq j < k$, $\bound{G'_j}{V' \cup \{ \mathit{guess\_br\_defined} \}}{\prop}{\nonunique{G'_j}}{p_j}$ and the probabilities $p_j$ satisfy Property~\ref{prop:prob},
then $\bound{G}{V'}{\prop}{\Dfalse}{\sum_{j=0}^{k-1} p_j}$ (neglecting a small additional runtime of the context).
%% If $\Advtev{G'_j}{\prop}{C[C_{\prop}[\,]]}{\nonunique{G'_j}} \leq p_j(C)$
%%      for $0 \leq j < k$ and for every evaluation context $C$ acceptable for $C_{\prop}[G'_j]$ with public variables $V \cup \{ \mathit{guess\_br\_defined} \}$ that does not contain $\sevent$ nor $\sbarevent$ and the probabilities $p_j$ satisfy the same condition, then $\Advtev{G}{\prop}{C[C_{\prop}[\,]]}{\Dfalse} \leq \sum_{j = 0}^{k-1} p_j(C)$ for every evaluation context $C$ acceptable for $C_{\prop}[G]$ with public variables $V$ that does not contain $\sevent$ nor $\sbarevent$ (neglecting a small additional runtime of the context).
\end{lemma}
\fussy
\begin{proof}

  \paragraph{Correspondences}
Let $C$ be an evaluation context acceptable for $G$ with any public
variables that does not contain events used by $\corresp$ or $D$.
\begin{align}
     \Advtev{G}{\corresp}{C}{D} &= \Pr[C[G] : (\neg\corresp \vee D) \wedge\neg\nonunique{G,D}]\notag\\
     &= \Pr[C[G] : (\neg\corresp \vee D)  \wedge\neg\nonunique{G,D}\wedge \pp \text{ not executed}] + {}\notag\\
     &\phantom{{}={}}
     \sum_{j = 0}^{k-1} \Pr[C[G] : (\neg\corresp \vee D) \wedge\neg\nonunique{G,D} \wedge \text{branch $j$ is taken at }\pp]\notag\\
     &\leq \sum_{j = 0}^{k-1} \Pr\left[\begin{array}{@{}l@{}}
         C[G] : (\neg\corresp \vee D) \wedge\neg\nonunique{G,D} \wedge {}\\
         \phantom{C[G] : {}}(\pp \text{ not executed} \vee \text{branch $j$ is taken at }\pp)
     \end{array}\right]\notag\\
     &\leq \sum_{j = 0}^{k-1} \Pr[C[G'_j] : (\neg\corresp \vee D) \wedge\neg\nonunique{G'_j,D}]\label{eq:stepguessbranch}\\
     &\leq \sum_{j = 0}^{k-1} \Advtev{G'_j}{\corresp}{C}{D}\notag\\
     &\leq \sum_{j = 0}^{k-1} p_j(C) \tag*{since $\bound{G'_j}{V}{\corresp}{D}{p_j}$}
\end{align}
The step~\eqref{eq:stepguessbranch} is valid because $G'_j$ behaves like $G$ when $\pp$ is not executed or branch $j$ is taken at $\pp$, in both transformations \rn{guess\_branch} $\pp$ and \rn{guess\_branch} $\pp$ \rn{no\_test}.
Therefore, we obtain $\bound{G}{V}{\corresp}{D}{\sum_{j=0}^{k-1} p_j}$.

   \paragraph{(One-session or bit) secrecy}
   Let $C$ be an evaluation context acceptable for $C_{\prop}[G]$ with public variables $V$ that does not contain $\sevent$ nor $\sbarevent$.
   We have
   \begin{align*}
     &\Advt_G^{\prop}(C)\\
     &= \Pr[C[C_{\prop}[G]] : \sevent] - \Pr[C[C_{\prop}[G]] : \sbarevent]\\
     &=  \Pr[C[C_{\prop}[G]] : \sevent \wedge \pp \text{ not executed}] - \Pr[C[C_{\prop}[G]] : \sbarevent \wedge \pp \text{ not executed}]\\
     &+ \sum_{j = 0}^{k-1} \Pr[C[C_{\prop}[G]] : \sevent \wedge \text{branch $j$ is taken at }\pp] - \Pr[C[C_{\prop}[G]] : \sbarevent \wedge \text{branch $j$ is taken at }\pp]
   \end{align*}
   From the adversary $C$, we define four adversaries $C'$ that output $b''$ instead of $b'$ on channel $\cS'$ ($\cS''$ for bit secrecy), where $b'' = \bguard{\kw{defined}(\mathit{guess\_br\_defined})}{f_1(b')}{f_2(b')}$ where $f_1(b')$ is either $b'$ or $\neg b'$, and similarly for $f_2$, and consider the adversary $C'_{\max,j}$ among those four that yields the maximum $\Advt_{G'_j}^{\prop}(C')$. Changing $b'$ into $\neg b'$ swaps the events $\sevent$ and $\sbarevent$, and therefore swaps their probabilities. Hence, for this adversary $C'_{\max,j}$,
   \begin{align*}
     &\Advt_{G'_j}^{\prop}(C'_{\max,j})\\
     &= | \Pr[C[C_{\prop}[G]] : \sevent \wedge \text{branch $j$ is taken at }\pp] -
     \Pr[C[C_{\prop}[G]] : \sbarevent \wedge \text{branch $j$ is taken at }\pp] |\\
     &+ | \Pr[C[C_{\prop}[G]] : \sevent \wedge \pp \text{ not executed}] -
     \Pr[C[C_{\prop}[G]] : \sbarevent \wedge \pp \text{ not executed}] |
   \end{align*}
   So
   \begin{align*}
   &\sum_{j=0}^{k-1} \Advt_{G'_j}^{\prop}(C'_{\max,j})\\
     &\geq \sum_{j=0}^{k-1} \left|\begin{array}{l}
     \Pr[C[C_{\prop}[G]] : \sevent \wedge \text{branch $j$ is taken at }\pp] \\
     {} -
     \Pr[C[C_{\prop}[G]] : \sbarevent \wedge \text{branch $j$ is taken at }\pp]
     \end{array}\right|\\
   &\qquad {} + \sum_{j=0}^{k-1} | \Pr[C[C_{\prop}[G]] : \sevent \wedge \pp \text{ not executed}] -
     \Pr[C[C_{\prop}[G]] : \sbarevent \wedge  \pp \text{ not executed}] |\\
   &\geq \Advt_G^{\prop}(C)
   \end{align*}
   Since $G'_j$ satisfies $\prop$ with public variables $V \cup \{ \mathit{guess\_br\_defined} \}$ up to probability $p_j$ and $C'_{\max,j}$ is an evaluation context acceptable for $C_{\prop}[G'_j]$ with public variables $V \cup \{ \mathit{guess\_br\_defined} \}$ that does not contain $\sevent$ nor $\sbarevent$, we have $\Advt_{G'_j}^{\prop}(C'_{\max,j}) \leq p_j(C'_{\max,j})$.
   Moreover, by Property~\ref{prop:prob}, $p_j(C)$ depends only on the runtime of $C$, the number of outputs $C$ makes on the various channels (which determine replication bounds), and the length of bitstrings, so we have $p_j(C'_{\max,j}) = p_j(C)$. (The additional runtime of the context can be neglected.)
   So $\Advt_G^{\prop}(C) \leq \sum_{j = 0}^{k-1} p_j(C)$.
   Therefore, $G$ satisfies $\prop$ with public variables $V$ up to probability $\sum_{j = 0}^{k-1} p_j$.
   The proof of the second property for (one-session or bit) secrecy proceeds as for $\rn{guess}\ i$ in Lemma~\ref{lem:guess_i}.
   \proofcomplete
   \end{proof}
   
\subsubsection[\rn{global\_dep\_anal}]{\rn{global\_dep\_anal} \cite{Blanchet07c}}\label{sec:globaldepanal}

%\extend{This description may be outdated.}

The global dependency analysis \rn{global\_dep\_anal} $x$ tries to find
a set of variables $S$ such that only variables in $S$ depend on $x$.
In particular, when the global dependency analysis succeeds,
the control flow and the view of the adversary do not depend on $x$,
except in cases of negligible probability.

Let $x$ be a variable defined only by random choices $\Res{x}{T}$ where $T$
is a large type.
Let $\Sdef$ be a set of variables defined only by 
assignments.
Let $\Sdep$ be a set of variables containing $x$.
(Intuitively, $\Sdep$ will be a superset of variables
that depend on $x$.)

We say that a function $f : T \rightarrow T'$ is \emph{uniform} when 
each element of $T'$ has at most
$|T|/|T'|$ antecedents by $f$.
In particular, this is true in the following two cases:
\begin{itemize}

\item $f$ is such that $f(x)$ is uniformly distributed in $T'$ if
$x$ is uniformly distributed in $T$.

\item $f$ is the restriction to the image of $f'$ of an inverse of
$f'$, where $f'$ is a poly-injective function. (We consider that
$f(x)$ is undefined
when $x$ is not in the image of $f'$. Here, in contrast to the rest of
the paper, we allow $f : T \rightarrow T'$ to be defined only on a
subset of $T$.) Precisely, when $x_k \in \Sdef$
is defined by a pattern-matching $\assign{f'(x_1, \ldots, x_n)}{M}{P}\
\ELSE P'$, we have $x_k = {f'}_k^{-1}(M)$, but furthermore when $x_k$
is defined we know that the value of $M$ is in the image of $f'$, so
we have $x_k = f(M)$ where $f = {{f'}_k^{-1}}_{|\im f'}$.

\end{itemize}
\bb{TO DO? I could generalize to functions that have at most $p(\secp) * |I_{\secp}(T)|/|I_{\secp}(T')|$ antecedents for each value, where p is a polynomial.}%

We say that $M$ \emph{characterizes a part of} $x$ with
$\Sdef$,$\Sdep$ when for all $M_0$
obtained from $M$ by substituting variables of $\Sdef$ with
their definition (when there is a dependency cycle among variables of $\Sdef$,
we do not substitute a variable inside its definition), 
$\alpha M_0 = M_0$ implies $f_1(\ldots f_k((\alpha
x)[\tup{M'}])) = f_1(\ldots f_k(x[\tup{M}]))$ for some uniform functions 
$f_1, \ldots, f_k$ and for
some $\tup{M}$ and $\tup{M'}$, where $\alpha$ is a renaming of variables
of $\Sdep$ to fresh variables, 
\bbnote{$\alpha$ does not touch replication indices.}%
$x[\tup{M}]$ is a subterm of $M_0$, 
$(\alpha x)[\tup{M'}]$ is a subterm of $\alpha M_0$, 
the variables in $\Sdep$ do not occur in $\tup{M}$ or $\tup{M'}$,
$T$ is the type of the
result of $f_1$ (or of $x$ when $k = 0$), and $T$ is a large type. In
that case, the value of $M$ uniquely determines the value of
$f_1(\ldots f_k(x[\tup{M}]))$.

We use a simple rewriting prover to determine that.  We consider the
set of terms $\mset_0 = \{ \alpha M_0 = M_0 \}$, and we rewrite
elements of $\mset_0$ using the first kind of user-defined rewrite
rules mentioned in Section~\ref{sec:userdefinedrewriterules} and the rule $\{
M_1 \wedge M_2 \} \cup \mset' \rewrite \{ M_1, M_2 \} \cup \mset'$.

When $\mset_0$ can be rewritten to a set that contains an equality
of the form  
$f_1(\ldots
f_k(x[\tup{M}])) = f_1(\ldots f_k((\alpha x)[\tup{M'}]))$ or
$f_1(\ldots f_k((\alpha x)[\tup{M'}])) = f_1(\ldots f_k(x[\tup{M}]))$
for some $\tup{M}$ and $\tup{M'}$ such that the variables in $\Sdep$
do not occur in $\tup{M}$ or $\tup{M'}$, we have that $M$ characterizes 
a part of $x$ with $\Sdef, \Sdep$.

We say that $M$ \emph{characterizes a part of} $x$ when $M$
characterizes a part of $x$ with $\emptyset,S'$ where $S'$ is $\{ x\}$
union the set of all variables except those defined by random choices.
(We know that variables different from $x$ and defined by random choices
do not depend on $x$, so in the absence of more precise information,
we can set $\Sdep = S'$.)

We say that $\indep(x) = S$ when intuitively, only variables in $S$
depend on $x$, and the adversary cannot see the value of $x$.
Formally, $\indep(x) = S$ when 
\begin{itemize}

\item $S \cap V = \emptyset$.

\item Variables of $S$ do not occur in input or output
channels or messages, that is, they do not occur in the terms 
$M_1$, \ldots, $M_m$,
$N_1$, \ldots, $N_k$ in the input $\cinput{c[M_1, \ldots,
M_m]}{x_1[\tup{i}]:T_1, \ab \ldots, \ab x_k[\tup{i}]:T_k}$ or in the output
$\coutput{c[M_1,\ldots, M_m]}{N_1, \ab \ldots, \ab N_k}$.

\item Variables of $S$ except $x$ are defined only by assignments.

\item If a variable $y \in S$ occurs in $M$ in $\assign{z:T}{M}{P}$,
then $z \in S$.

\item Variables in $S$ may occur in $\defined $ conditions of $\FIND $
but only at the root of them.

\item All terms $M_j$ in processes 
%$\cfind{j=1}{m}{ \tup{\vf_j}[\tup{i}] \leq
%\tup{n_j}}{ M_{j1}, \ldots, M_{jl_j}}{M_j}{P_j}{P'}$ 
$\FIND$ $(\mathop\bigoplus\nolimits_{j=1}^{m} \tup{\vf_j}[\tup{i}] \leq
\tup{n_j}$ $\SUCHTHAT$ $\defined (M_{j1}, \ldots, M_{jl_j}) \fand M_j$
$\kw{then}$ $P_j)$ $\kw{else}$ $P'$ 
are combinations
by $\fand$, $\for$, or $\fnot$ of terms that either do not contain
variables in $S$ or are of the form $M_1 = M_2$ or $M_1 \neq M_2$
where $M_1$ characterizes a part of $x$ with $S\setminus \{ x\},S$ and no
variable of $S$ occurs in $M_2$, or $M_2$ characterizes a part of $x$ 
with $S\setminus \{ x\},S$ and
no variable of $S$ occurs in $M_1$.

\end{itemize}
The last item implies that the result of tests does not depend on
the values of variables in $S$, except in cases of negligible probability.
Indeed, the tests $M_1 = M_2$ with $M_1$ characterizes a part of $x$ 
with $S\setminus \{ x\},S$ and $M_2$
does not depend on variables in $S$ are false except 
in cases of negligible probability, since the value of $M_1$ 
uniquely determines the value
of $f_1(\ldots f_k(x[\tup{M}]))$ and $M_2$ does not depend on
$f_1(\ldots f_k(x[\tup{M}]))$, so the equality $M_1 = M_2$ happens for a single
value of $f_1(\ldots f_k(x[\tup{M}]))$, which yields a negligible
probability because $f_1, \ldots, f_k$ are uniform, $x$ is chosen
with uniform probability, and the type of the result of $f_1$ is large. 
Similarly, the tests $M_1 \neq M_2$
are true except in cases of negligible probability.

In checking the conditions of $\indep(x) = S$, we do not consider the 
parts of the code that are unreachable due to tests whose result is
known by the conditions above.

The set $S$ is computed by a fixpoint iteration, 
starting from $\{ x \}$ and adding variables defined by assignments
that depend on variables already in $S$.

If we manage to show that $\indep(x) = S$, we transform the game
as follows:
\begin{itemize}

\item We replace with $\false$ terms $M_1 = M_2$ in conditions of $\FIND$
where $M_1$ characterizes a part of $x$ with $S\setminus \{ x\},S$ and no
variable of $S$ occurs in $M_2$, or symmetrically.

\item We replace with $\true$ terms $M_1 \neq M_2$ in conditions of $\FIND$
where $M_1$ characterizes a part of $x$ with $S\setminus \{ x\},S$ and no
variable of $S$ occurs in $M_2$, or symmetrically.

\end{itemize}

\begin{lemma}
  The transformation \rn{global\_dep\_anal} requires and preserves
  Properties~\ref{prop:nointervaltypes}, \ref{prop:noreschan}, \ref{prop:channelindices}, \ref{prop:notables}, \ref{prop:autosarename}, and~\ref{prop:expand}.
  If transformation \rn{global\_dep\_anal} transforms $G$ into $G'$, then
  $\dset,\dsetsnu : G, \ab D, \ab \usedevents \ab \indistev{V}{p} G', D, \usedevents$, where $p$ is an upper bound on the probability that required equalities do not hold.
\end{lemma}

\subsubsection[\rn{simplify}]{\rn{simplify} \cite{Blanchet07c,BlanchetEPrint12}}
\label{sec:simplify}

%\extend{This description is outdated.}

We use the following transformations in order to simplify games.
These transformations exploit the information collected as explained
in Section~\ref{sec:truefacts}.
\begin{enumerate}

\item Each term $M$ in the game is replaced with a simplified term
$M'$ obtained by reducing $M$ by user-defined rewrite rules knowing $\facts_{P_M}$ 
(see Sections~\ref{sec:userdefinedrewriterules} and~\ref{sec:equationalprover}) and 
the rewrite rules obtained from $\facts_{P_M}$
by the above equational prover where $P_M$ is the smallest process 
containing $M$. The replacement is performed
only when at least one user-defined rewrite rule has been used, to avoid 
complicating the game by substituting all variables with their value.

\item When setting \rn{inferUnique} is true, CryptoVerif tries to prove uniqueness of $\FIND\unique{e}$,
  as in transformation \rn{prove\_unique} (Section~\ref{sec:prove_unique}).
  
\item If $P = \FIND\uniqueopt \ (\mathop{\textstyle\bigoplus}\nolimits_{j=1}^m
\vf_{j1}[\tup{i}] = i_{j1} \leq n_{j1},
\ldots, \vf_{jm_j}[\tup{i}] = i_{jm_j} \leq n_{jm_j}$ $\SUCHTHAT$ $\defined(M_{j1}, \ab
\ldots, \ab M_{jl_j}) \fand M_j$ $\kw{then}$ $P_j)$ $\kw{else}$ $P'$,
$M_j$ does not contain $\kw{event\_abort}$ nor unproved $\FIND\unique{e}$,
$\vf_{jk}[\tup{i}]$ reduces into $M'$ by user-defined rewrite rules knowing $\facts_{P_j}$ (see Sections~\ref{sec:userdefinedrewriterules} and~\ref{sec:equationalprover}) and the rewrite rules obtained from $\facts_{P_j}$, 
and $\vf_{jk}$ does not occur in $M'$, 
then $\vf_{jk}$ is removed from the $j$-th branch of this
$\FIND$, 
$i_{jk}$ is replaced with $M'\{i_{j'k'}/\vf_{j'k'},j'\leq m, k'\leq m_j\}$ in
$M_{j1}, \ldots, M_{jl_j}, M_j$ and $P_j$ is replaced with
$\assign{\vf_{jk}[\tup{i}]:[1, n_{jk}]}{M'}{P_j}$.
(Intuitively, $\vf_{jk}[\tup{i}] = M'$, so the value of $\vf_{jk}[\tup{i}]$
can be computed by evaluating $M'$ instead of performing an array lookup.
We remove $\vf_{jk}[\tup{i}]$ from the variables looked up by $\FIND$ and
replace $\vf_{jk}[\tup{i}]$ with its value $M'$.)
\bb{TODO problem if $M'$ uses variables $\vf_{j'k'}$ with indices different
from $\tup{i}$? It seems impossible to build a counter-example, because
we cannot guarantee that $\vf_{jk} = M'$ while having $M'$ use $\vf_{j'k'}$ 
with indices syntactically different from $\tup{i}$ and still equal to 
$\tup{i}$ at runtime.}%

%CompatibleDefs.check_compatible_deflist
\item Suppose that  $P = \FIND\uniqueopt \ (\mathop{\textstyle\bigoplus}\nolimits_{j=1}^{m}$  $\vf_{j1}[\tup{i}]  = i_{j1} \leq n_{j1}, \ab 
\ldots, \ab \vf_{jm_j}[\tup{i}] = i_{jm_j} \leq n_{jm_j}$ $\SUCHTHAT$
$\defined(M_{j1}, \ab \ldots, \ab M_{jl_j}) \fand M_j$ $\kw{then}$ $P_j)$ $\kw{else}$
$P'$,
there exists a term $M$ such that $\defined(M) \in \fset_{P_j}$,
$x[N_1, \ldots, N_l]$ is a subterm of $M$, 
$x \neq \vf_{jk}$ for all $k \leq m_j$,
%$x[N_1, \ldots, N_l]$ is a subterm of $M_{jk}$, 
and none of the following conditions holds: 
a) $P$ is under a definition of $x$ in $Q_0$;
b) $Q_0$ contains $Q_1 \parpop Q_2$ such that a definition of $x$
occurs in $Q_1$ and $P$ is under $Q_2$ or a definition of $x$
occurs in $Q_2$ and $P$ is under $Q_1$;
c) $Q_0$ contains $lp+1$ replications above a process $Q$ that contains
a definition of $x$ and $P$, where $lp$ is the length of the longest
common prefix between $N_1, \ldots, N_l$ and the current replication
indices at the definitions of $x$.
Then the $j$-th branch of the $\FIND $ is removed.
(In this case, $x[N_1, \ldots, N_l]$ cannot be defined at $P$,
so the $j$-th branch of the $\FIND$ cannot be taken.)

%CompatibleDefs2.check_compatible2_deflist
\item Suppose that  $P = \FIND\uniqueopt \ (\mathop{\textstyle\bigoplus}\nolimits_{j=1}^{m}$  $\vf_{j1}[\tup{i}] = i_{j1} \leq n_{j1}, \ab 
\ldots, \ab \vf_{jm_j}[\tup{i}] = i_{jm_j} \leq n_{jm_j}$ $\SUCHTHAT$
$\defined(M_{j1}, \ab \ldots, \ab M_{jl_j}) \fand M_j$ $\kw{then}$ $P_j)$ $\kw{else}$
$P'$,
there exist terms $M$, $M'$ such that $\defined(M) \in \fset_{P_j}$,
$x[N_1, \ldots, N_l]$ is a subterm of $M$, 
$\defined(M') \in \fset_{P_j}$,
$x'[N'_1, \ldots, N'_{l'}]$ is a subterm of $M'$, 
$N_k = N'_k$ for all $k \leq \min(l, l')$,
$x \neq x'$,
and $x$ and $x'$ are incompatible, then the $j$-th branch of the 
$\FIND $ is removed. Two variables $x$ and $x'$ are said to be compatible
when either there exists $Q_1 \parpop Q_2$ in the game such that 
$x$ is defined in $Q_1$ and $x'$ is defined in $Q_2$, or 
there is a definition of $x'$ under a definition of $x$, or symmetrically.

\item If $P = \FIND\uniqueopt \ (\mathop{\textstyle\bigoplus}\nolimits_{j=1}^m
\tup{\vf_j}[\tup{i}] = \tup{i_j} \leq \tup{n_j}$ $\SUCHTHAT$ $\defined(M_{j1}, 
\ab \ldots, \ab M_{jl_j}) \fand M_j$ $\kw{then}$ $P_j)$ $\kw{else}$ $P'$ and $\facts_{P_j}$ yields a
contradiction, then 
the $j$-th branch of the $\FIND $ is removed if $M_j$ does not contain $\kw{event\_abort}$ nor unproved $\FIND\unique{e}$, 
and $P_j$ is replaced with $\coutput{yield}{}$ if $M_j$ contains $\kw{event\_abort}$ or some unproved $\FIND\unique{e}$.

\item If $P = \FIND\uniqueopt \ \ELSE P'$, then $P$ is
replaced with $P'$.

\item If $\FIND\uniqueopt \ (\mathop{\textstyle\bigoplus}\nolimits_{j=1}^m
\tup{\vf_j}[\tup{i}] = \tup{i_j} \leq \tup{n_j}$ $\SUCHTHAT$ $\defined(M_{j1}, 
\ab \ldots, \ab M_{jl_j}) \fand M_j$ 
$\kw{then}$ $P_j)$ $\kw{else}$ $P'$ and $\facts_{P'}$ yields a 
contradiction, then $P'$ is replaced with $\coutput{yield}{}$. 
%(The probability that $P'$ is executed is negligible.)

\item If $P = \FIND\uniqueopt\ \tup{\vf}[\tup{i}] = \tup{i'} \leq \tup{n} \ \SUCHTHAT\ \defined(M_{1}, \ab \ldots, \ab M_{l}) \wedge M \THEN 
P_1\ \ELSE P'$,
$\uniqueopt$ is not $\unique{e}$ for some non-unique event $e$ that is not proved yet to have negligible probability,
$\facts_{P'}$ yields a contradiction, 
$M$ is simple (so $M$ never aborts), and 
the variables in $\tup{\vf}$ are not used outside $P$ and are not in $V$, 
then $P$ is replaced with $P_1$.
(When the $\FIND$ defines variables $\tup{\vf}$ used elsewhere, 
we cannot remove it.)

\item If $P = \FIND\uniqueopt \ (\mathop{\textstyle\bigoplus}\nolimits_{j=1}^{m}$
$ \tup{\vf_j}[\tup{i}] \leq \tup{n_j}$ $\SUCHTHAT$
$\defined(M_{j1}, \ab \ldots, \ab M_{jl_j}) 
\fand M_j$ $\kw{then}$ $\coutput{yield}{})$
$\kw{else}$ $\coutput{yield}{}$,
$\uniqueopt$ is not $\unique{e}$ for some non-unique event $e$ that is not proved yet to have negligible probability,
the terms $M_j$ do not contain $\kw{event\_abort}$ nor unproved $\FIND\unique{e}$,
and the variables in $\tup{\vf_j}$ are not used outside $P$ and are not in $V$, 
then $P$ is replaced with $\coutput{yield}{}$.

\item The $\defined $ conditions of $\FIND $ are updated so that
Invariant~\ref{inv2} is satisfied. (When such a $\defined$ condition guarantees
that $M$ is defined, $\defined(M)$ implies $\defined(M')$, and 
after simplification $M'$ appears in the scope of this condition,
then $M'$ has to be added to this condition if it is not already present.)
\bbnote{remove from def list elements that do not occur in the condition $M_j$
or in the process $M_j$ of the find and whose definition is implied by 
other elements of the def list}%

\item If $P = \Res{x}{T};P'$ or $\assign{x:T}{M}{P'}$ and $x$ is not used
in the game and is not in $V$, then $P$ is replaced with $P'$.

\item If one of the $\kw{then}$ branches of a $\kw{find}[\kw{unique}]$
  always succeeds and the conditions of this $\kw{find}[\kw{unique}]$ do not contain
  $\kw{event\_abort}$ nor unproved $\FIND\unique{e}$, then we keep only that branch.

  Indeed, the other branches are never taken: the conditions of this
  $\kw{find}[\kw{unique}]$ never abort in traces counted in the probability, 
  and the $\kw{find}$ itself aborts when there are several successful choices.%
  \bb{If some conditions contain $\kw{event\_abort}$, I could remove the
  branches whose conditions do not contain $\kw{event\_abort}$, and
  replace the $\kw{then}$ part of branches whose conditions contain $\kw{event\_abort}$
  by $\kw{yield}$.}%

\item We reorganize a $\kw{find}[\kw{unique}]$ that occurs in a
  $\kw{then}$ branch of a $\kw{find}[\kw{unique}]$:
  we transform
  \newcommand{\brfun}{\mathsf{b}}%
\[\kw{find}[\kw{unique}]\ (\bigorfind_{j=1}^k \tilde u_j = \tilde i_j \leq \tilde n_j\ \kw{suchthat}\ c_j\ \kw{then}\ P_j)\ \kw{else}\ P\]
where $P_{j_0} = \kw{find}[\kw{unique}]\ (\bigorfind_{j'=1}^{k'} \FB_{j'})$ $\kw{else}$ $P''_{j_0}$
into
\[\begin{split}
&\kw{find}[\kw{unique}]\ (\bigorfind_{j=1,\dots,k; j \neq j_0} \tilde u_j = \tilde i_j \leq \tilde n_j\ \kw{suchthat}\ c_j\ \kw{then}\ P_j) \\
&\!\oplus (\bigorfind_{j'=1}^{k'} \bigorfind_{(\tilde u' = \tilde i' \leq \tilde n'\ \kw{suchthat}\ c'\ \kw{then}\ P') \in \brfun(\FB_{j'}, L_{j'})}\\
&\phantom{\!\oplus (\quad\qquad} \tilde u_{j_0} = \tilde i_{j_0} \leq \tilde n_{j_0}, \tilde u' = \tilde i' \leq \tilde n'\ \kw{suchthat}\ c_{j_0} \wedge c'\ \kw{then}\ P')\\
& \kw{else}\ \kw{find}[\kw{unique}]\ \tilde u_{j_0} = \tilde i_{j_0} \leq \tilde n_{j_0}\ \kw{suchthat}\ c_{j_0}\ \kw{then}\ P''_{j_0}\ \kw{else}\ P
\end{split}\]
where\\
-- either for all $j\leq k$, $c_j$ does not contain $\kw{event\_abort}$ nor unproved $\FIND\unique{e}$ or $c_{j_0}$ never aborts (this is true in particular when $c_{j_0}$ does not contain $\kw{event\_abort}$ nor proved or unproved $\FIND\unique{e}$)\\
\bbnote{If for all $j\leq k$, $c_j$ does not contain $\kw{event\_abort}$ nor unproved $\FIND\unique{e}$, the only abortion event that appears in the conditions of $\FIND$ is the proved non-unique event. Before transformation, this event is executed in a $c_j$ when there exists $j \neq j_0$ and $\tup{i}_j$ such that $c_j$ executes that event, or there exist $\tup{i}_{j_0}$ such that $c_{j_0}$ executes that event. 
After transformation, it is executed in a $c_j$ when there exists $j \neq j_0$ and $\tup{i}_j$ such that $c_j$ executes that event, or there exist $\tup{i}_{j_0}$ such that $c_{j_0}$ executes that event and either there exists $\tup{i'}$ such that the defined condition of $c'$ holds or the first $\FIND$ fails. 
That respects the probability conditions.

If $c_{j_0}$ never aborts, then in case $c_j$ aborts for $j \neq j_0$, the same abortion happens after transformation. 

However, in case $c_{j_0}$ contains a proved $\FIND\unique{}$ and for some $j \neq j_0$, $c_j$ contains $\keventabort{e}$, suppose
that $\keventabort{e}$ is executed for $N_1 > 0$ values of $\tup{i}_j$, there are $N_2>0$ values of $\tup{i}_{j_0}$ such that the proved $\FIND\unique{}$ in $c_{j_0}$ actually has several solutions. Suppose that for each $\tup{i}_{j_0}$ the defined condition of $c'$ holds for $N$ values of $\tup{i}'$ (for simplicity, we suppose that the last condition does not depend on the value of $\tup{i}_{j_0}$ and that $L_{j'} = []$ for all $j'$). Then, before transformation, there are $N_1$ solutions of the top $\FIND$ that yield $\keventabort{e}$ and $N_2$ that yield abortion with the non-unique event,
so $\Pr[C[G]:e] = N_1/(N_1+N_2)$ times the probability of reaching the top $\FIND$.
After transformation, there are $N_1$ solutions of the top $\FIND$ that yield $\keventabort{e}$ and $N_2 N$ that yield abortion with the non-unique event, so $\Pr[C[G']:e] = N_1/(N_1+N_2N)$ times the probability of reaching the top $\FIND$. That breaks $\Pr[C[G]:e] \leq \Pr[C[G']:e]$.
So we forbid this situation.
}%
-- $\tilde i$ are the current replication indices at the transformation point\\
-- for all $j'$, $\FB_{j'} = (\tilde u'_{j'} = \tilde i'_{j'} \leq \tilde n'_{j'}\ \kw{suchthat}\ c'_{j'}\ \kw{then}\ P'_{j'})$,
$c'_{j'} = \kw{defined}(\tilde M_{j'}) \wedge \dots$,
$c'_{j'}$ never aborts, and
\bbnote{To take an extreme counter-example, suppose $c_{j_0}$ always evaluates to false and $c'_{j'}$ is a $\FIND\unique{e}$ that always has several solutions. In the game before transformation, $c'_{j'}$ is not reachable, so the $\FIND\unique{e}$ is trivially proved unique: $\Pr[C[G]: e] = 0$. However, after transformation, $c_{j_0} \wedge c'$ evaluates $c' = c'_{j'}$ (simple case when $L_{j'} = \emptyset$) so executes event $e$. If no other condition of the find aborts, $\Pr[C[G']: e]$ is the probability of reaching the find. That breaks $\Pr[C[G]: D \wedge\neg e] = \Pr[C[G]: D] \leq \Pr[C[G']: D \wedge\neg e]$.}%
$L_{j'}$ is the list of $x[\tilde N]$ subterm of $\tilde M_{j'}$ with $x \in \tilde u_{j_0}$ and $\tilde N \neq \tilde i$, ordered by increasing size\\
-- $\brfun(\tilde u' = \tilde i' \leq \tilde n'\ \kw{suchthat}\ c'\ \kw{then}\ P', [\,]) = \{ \tilde u' = \tilde i' \leq \tilde n'\ \kw{suchthat}\ c'\{\tilde i_{j_0}/\tilde u_{j_0}\} \ \kw{then}\ P'\}$\\
-- $\brfun(\FB, x[\tilde N] :: L) = \brfun(\FB, L) \cup \{ \tilde u' = \tilde i' \leq \tilde n'\ \kw{suchthat}\ c'\{i/x[\tilde N]\} \wedge \tilde N = \tilde i\ \kw{then}\ P' \mid (\tilde u' = \tilde i' \leq \tilde n'\ \kw{suchthat}\ c'\ \kw{then}\ P') \in \brfun(\FB, L) \}$ where $i = x\{\tilde i_{j_0}/\tilde u_{j_0}\}$.

The function $\brfun$ takes into account that, before the transformation, $\tilde u_{j_0}[\tilde i]$ is defined when we test $\kw{defined}(\tilde M_{j'})$ in $\FB_{j'}$, while after the transformation, $\tilde u_{j_0}[\tilde i]$ is not defined yet when we perform this test. Furthermore, the value of $\tilde u_{j_0}[\tilde i]$ will be $\tilde i_{j_0}$. Therefore, 1) when we access $x[\tilde i]$ in $c'_{j'} = \kw{defined}(\tilde M_{j'}) \wedge \dots$, we replace this access with $i$, where $i = x\{\tilde i_{j_0}/\tilde u_{j_0}\}$; this is done in $\brfun(\FB, [\,])$ by the substitution $\{\tilde i_{j_0}/\tilde u_{j_0}\}$; and 2) when we access $x[\tilde N]$ in $c'_{j'} = \kw{defined}(\tilde M_{j'}) \wedge \dots$ for $\tilde N$ not syntactically equal to $\tilde i$, we need to distinguish two cases: either at runtime $\tilde N = \tilde i$ and we replace this access with $i$ (second part of the union in $\brfun(\FB, x[\tilde N] :: L)$), or at runtime $\tilde N \neq \tilde i$ and we continue using $x[\tilde N]$ (first part of the union in $\brfun(\FB, x[\tilde N] :: L)$). The array accesses $x[\tilde N]$ in $L_{j'}$ are ordered by increasing size because, in case of nested array accesses, we need to handle the bigger array access first (so it must occur last in the list), because after substitution of the smaller one with $i$, we would not recognize the bigger one.

\bbnote{We tried to transform
\[\begin{split}
\kw{find}[\kw{unique}]\ (\bigorfind_{j=1}^k &\tilde u_j = \tilde i_j \leq \tilde n_j\ \kw{suchthat}\ c_j\ \kw{then}\ P_j)\ \kw{else}\ P
\end{split}\]
where $P_{j_0} = \kw{find}\uniqueopt\ (\bigorfind_{j'=1}^{k'} \FB_{j'})$ $\kw{else}$ $P''_{j_0}$
into
\[\begin{split}
&\kw{find}\uniqueopt\ \\
&\quad \bigorfind_{j'=1}^{k'} \bigorfind_{(\tilde u' = \tilde i' \leq \tilde n'\ \kw{suchthat}\ c'\ \kw{then}\ P') \in \brfun(\FB_{j'}, L_{j'})}\\
&\qquad \tilde u_{j_0} = \tilde i_{j_0} \leq \tilde n_{j_0}, \tilde u' = \tilde i' \leq \tilde n'\\
&\qquad \kw{suchthat}\ c_{j_0} \wedge c'\ \kw{then}\ P'\\
&\kw{else}\ \kw{find}[\kw{unique}]\\
&\phantom{\kw{else}\ }(\bigorfind\limits_{j=1,\dots,k; j \neq j_0} \tilde u_j = \tilde i_j \leq \tilde n_j\ \kw{suchthat}\ c_j\ \kw{then}\ P_j) \\
&\phantom{\kw{else}\ }\oplus \tilde u_{j_0} = \tilde i_{j_0} \leq \tilde n_{j_0}\ \kw{suchthat}\ c_{j_0}\ \kw{then}\ P''_{j_0}\\
&\phantom{\kw{else}\ }\kw{else}\ P
\end{split}\]
However, using this transformation with $\uniqueopt = [\kw{unique}]$ had a negative impact on OEKE:
in OEKE, the branches in the inner $\kw{find}$ after transformation are merged, so this $\kw{find}$
disappears nicely.
We also tried using the version above only when $\uniqueopt$ is empty. It had a negative impact on
\url{examples/textsecure/indiff_HKDF_2.ocv}: it prevented a merge of $\kw{find}$ branches.

Moreover, this transformation can already be done manually when $\uniqueopt$ is empty:
just insert the outer $\kw{find}$ and simplification will yield the desired result
by removing useless branches of $\kw{find}$.

A transformation along the same lines would be: transform
\[\begin{split}
\kw{find}[\kw{unique}]\ (\bigorfind_{j=1}^k &\tilde u_j = \tilde i_j \leq \tilde n_j\ \kw{suchthat}\ c_j\\[-1mm]
&\kw{then}\ P_j)\ \kw{else}\ P
\end{split}\]
where $P_{j_0} = \kw{if}$ $c$ $\kw{then}$ $P''_1$ $\kw{else}$ $P''_2$ into
\[\begin{split}
&\kw{find}[\kw{unique}]\\
&\quad (\bigorfind_{j=1,\dots,k; j \neq j_0} \tilde u_j = \tilde i_j \leq \tilde n_j\ \kw{suchthat}\ c_j\ \kw{then}\ P_j)\\
&\quad \oplus \tilde u_{j_0} = \tilde i_{j_0} \leq \tilde n_{j_0}\ \kw{suchthat}\ c_{j_0} \wedge c\ \kw{then}\ P''_1\\
&\quad \oplus \tilde u_{j_0} = \tilde i_{j_0} \leq \tilde n_{j_0}\ \kw{suchthat}\ c_{j_0} \wedge \neg c\ \kw{then}\ P''_2\\
&\kw{else}\ P
\end{split}\]
but it might have a negative impact as well.}%

\bbnote{I abuse notation when I write $c_{j_0} \wedge c'_{j'}$: the conjunction of the ``defined'' parts is computed by concatenating the lists of defined variables.}%
This transformation cannot be performed when the outer $\kw{find}$ is not
unique because it might change the probability of taking each branch.
Moreover, we tried performing such a transformation when the inner $\kw{find}$
is not unique (in this case, after transformation, the outer $\kw{find}$ is
not unique), but it had a negative impact in some examples.
Furthermore, in the latter case, the transformation can be performed manually
by inserting the desired outer $\kw{find}$ and simplifying the game:
CryptoVerif will remove the useless branches of $\kw{find}$.

The conditions $c_{j_0}$ and $c'_{j'}$ are conjunctions of $\kw{defined}$ conditions and a term.
In the current implementation, the transformation is not performed when the term in $c_{j_0}$ is false (the branch of $\FIND$ will be removed by another transformation), and the branches such that the terms in $c'_{j'}$ are false are first removed. Furthermore, the transformation is performed only when one of the following conditions holds: the terms in $c_{j_0}$ and all $c'_{j'}$ are simple, or the term in $c_{j_0}$ is true and all $c'_{j'}$ never abort, or the terms in $c'_{j'}$ are all true and $c_{j_0}$ never aborts. With the usual simplification of ${}\wedge\true$, this guarantees that the transformed game satisfies Property~\ref{prop:expand}. Moreover, this implies the abortion conditions ($c_{j_0}$ and all $c'_{j'}$ never abort).

After this transformation, we advise renaming the variables
$\tilde u_{j_0}$ to distinct names, since they now have multiple
definitions.

\item We reorganize a $\kw{find}[\kw{unique}]$ that occurs in a condition
  of a $\kw{find}$: we transform
\[\begin{split}
\kw{find}\uniqueopt\ (\bigorfind_{j=1}^k &\tilde u_j = \tilde i_j \leq \tilde n_j\ \kw{suchthat}\ c_j\ \kw{then}\ P_j)\ \kw{else}\ P
\end{split}\]
where
\[\begin{split}
&c_{j_0} =  \defined(\tilde M') \wedge \kw{find}[\kw{unique}] (\bigorfind_{j'=1}^{k'} \tilde u'_{j'} = \tilde i'_{j'} \leq \tilde n'_{j'}\ \kw{suchthat}\ c'_{j'}\ \kw{then}\ M'_{j'})\ \kw{else}\ \kwf{false}
\end{split}\]
for all $j'\leq k'$, $M'_{j'}$ never aborts
\bbnote{To take an extreme counter-example, suppose $c'_{j'}$ always evaluates to false and $M'_{j'}$ is a $\FIND\unique{e}$ that always has several solutions. In the game before transformation, $M'_{j'}$ is not reachable, so the $\FIND\unique{e}$ in $M'_{j'}$ is trivially proved unique: $\Pr[C[G]: e] = 0$. However, after transformation, $c'_{j'} \wedge M'_{j'}$ evaluates $M'_{j'}$ so executes event $e$. If no other condition of the find aborts, $\Pr[C[G']: e]$ is the probability of reaching the find with $\tilde M'$ defined. That breaks $\Pr[C[G]: D \wedge\neg e] = \Pr[C[G]: D] \leq \Pr[C[G']: D \wedge\neg e]$.}%
and either for all $j\leq k$, $c_j$ does not contain $\kw{event\_abort}$ nor unproved $\FIND\unique{e}$ or for all $j'\leq k'$, $c'_{j'}$ never aborts, into 
\bbnote{If for all $j\leq k$, $c_j$ does not contain $\kw{event\_abort}$ nor unproved $\FIND\unique{e}$, the only abortion event that appears in the conditions of $\FIND$ is the proved non-unique event. Before transformation, this event is executed when there exists $j \neq j_0$ and $\tup{i}_j$ such that $c_j$ executes that event, or there exist $\tup{i}_{j_0}$, $\tilde i'_{j'}$ such that $\tilde M'$ is defined and $c'_{j'}$ executes that event, or there exists $\tup{i}_{j_0}$, such that $\tilde M'$ is defined and the $\FIND\unique{}$ at the top of $c_{j_0}$ has several solutions. 
After transformation, it is executed in the first 2 conditions, and the last one makes the top $\FIND$ non-unique. That respects the probability conditions.

If for all $j'\leq k'$, $c'_{j'}$ never aborts, then in case $c_j$ aborts for $j \neq j_0$ and the $\FIND\unique{}$ at the top of $c_{j_0}$ has a single solution, the same abortion happens after transformation. In case $c_j$ aborts for $j \neq j_0$ and the $\FIND\unique{}$ at the top of $c_{j_0}$ has several solutions, before transformation, the abortion is chosen randomly among those of $c_j$ and the non-unique $\FIND$ of $c_{j_0}$, and after transformation the abortion is chosen randomly among those of $c_j$ only. Since the non-unique $\FIND$ of $c_{j_0}$ is not counted in the probabilities, the transformation increases the final probability as desired.

However, in case for some $j'$, $c'_{j'}$ contains a proved $\FIND\unique{}$ and for some $j \neq j_0$, $c_j$ contains $\keventabort{e}$, suppose
that $\keventabort{e}$ is executed for $N_1 > 0$ values of $\tup{i}_j$, there are $N_2>0$ values of $\tup{i}_{j_0}$ such that $\tup{M}'$ is defined
and for each of these values $N >0$ values of $\tup{j}'$ such that the defined condition in $c'_{j'}$ holds and the proved $\FIND\unique{}$ in $c'_{j'}$ actually has several solutions (for simplicity, we suppose that the last condition does not depend on the value of $\tup{i}_{j_0}$). Then, before transformation, there are $N_1$ solutions of the top $\FIND$ that yield $\keventabort{e}$ and $N_2$ that yield abortion with the non-unique event,
so $\Pr[C[G]:e] = N_1/(N_1+N_2)$ times the probability of reaching the top $\FIND$.
After transformation, there are $N_1$ solutions of the top $\FIND$ that yield $\keventabort{e}$ and $N_2 N$ that yield abortion with the non-unique event, so $\Pr[C[G']:e] = N_1/(N_1+N_2N)$ times the probability of reaching the top $\FIND$. That breaks $\Pr[C[G]:e] \leq \Pr[C[G']:e]$.
So we forbid this situation.
}%
\[\begin{split}
&\kw{find}\uniqueopt\ (\bigorfind_{j=1,\dots, k; j \neq j_0} \tilde u_j = \tilde i_j \leq \tilde n_j\ \kw{suchthat}\ c_j\ \kw{then}\ P_j) \\
&\!\oplus (\bigorfind_{j'=1}^{k'} \tilde u_{j_0} = \tilde i_{j_0} \leq \tilde n_{j_0}, \tilde u'_{j'} = \tilde i'_{j'} \leq \tilde n'_{j'}\ \kw{suchthat}\\[-3.5mm]
&\phantom{\!\oplus (\bigorfind_{j'=1}^{k'} }\qquad \defined(\tilde M') \wedge c'_{j'} \wedge M'_{j'}\{\tilde i'_{j'}/\tilde u'_{j'}\}\ \kw{then}\ P_{j_0})\\[-3mm]
&\kw{else}\ P
\end{split}\]
The indication $\uniqueopt$ corresponds to either $[\kw{unique}]$
or empty.
The $\kw{find}$ is marked $[\kw{unique}]$ after transformation if 
the outer $\kw{find}$ was $[\kw{unique}]$ before transformation.

The variables $\tilde u'_{j'}$ are defined inside the condition of a $\FIND$ so by Invariant~\ref{invfc}, they have no array accesses. The transformation performed by function $\brfun$ above is therefore not needed here.

The conditions $c'_{j'}$ are conjunctions of $\kw{defined}$ conditions and a term.
The current implementation first removes the branches such that one of the following two conditions holds: the terms in $c'_{j'}$ are false or $M'_{j'}$ is false and $c'_{j'}$ never aborts. Furthermore, the transformation is performed only when for all $j'$, one of the following conditions holds: the term in $c'_{j'}$ and $M_{j'}$ are simple, or the term in $c'_{j'}$ is true and $M'_{j'}$ never aborts, or $M'_{j'}$ is true or false and the term in $c'_{j'}$ never aborts. With the usual simplification of ${}\wedge\true$ and ${}\wedge\false$, this guarantees that the transformed game satisfies Property~\ref{prop:expand}. Moreover, this implies the abortion conditions ($c'_{j'}$ and $M'_{j'}$ never abort).

\bbnote{In fact, I think only one $[\kw{unique}]$ is necessary:
- if the outer find is not unique, the initial set of solutions is mapped to a set of solutions of the same cardinality,
with an additional value of $\tilde u'_{j'}$ (and a single such value because the inner find is unique) when the branch is $j_0$,
the probability of taking each branch remains unchanged.
- if the inner find is not unique, for branches $\neq j_0$, there is still a single solution (because the outer find is unique); if branch $j_0$ is taken,
there may be several solutions instead of a single one, but all these solutions execute the same $P_{j_0}$.
The obtained find is not unique when the initial finds are not both unique.
However, in practice, if the outer find is unique and the inner find is not, losing the information that the outer find is unique may be problematic.
So I take the option of performing the transformation only if the inner find is unique; then the outer find is unique after transformation iff it was unique before transformation.
}%

\end{enumerate}
The simplification is iterated at most \rn{maxIterSimplif} times. The iteration stops
earlier in case a fixpoint is reached.

\begin{lemma}
  The transformation \rn{simplify} requires and preserves
  Properties~\ref{prop:nointervaltypes}, \ref{prop:noreschan}, \ref{prop:channelindices}, \ref{prop:notables}, and~\ref{prop:autosarename}. It preserves~\ref{prop:expand}.
  If transformation \rn{simplify} transforms $G$ into $G'$, then
  $\dset,\dsetsnu : G, \ab D, \ab \usedevents \ab \indistev{V}{p} G', D, \usedevents$, where $p$ is an upper bound on the probability that required equalities do not hold.
\end{lemma}

%\extend{When the game is not expanded, \rn{simplify} is actually a weaker simplification transformation. We plan to improve that.}

%For all these transformations, the correctness proof shows that, when the initial game does not execute event $\nonunique$, the transformed game behaves in the same way as the initial game. We can then apply point~\ref{item:secr:NonUnique} of Lemmas~\ref{lem:improveShoup} and~\ref{lem:improveSecrecy} to bound the probability of attack.

\subsubsection{\rn{all\_simplify}}

\rn{all\_simplify} perform several simplifications on the game, as if 
\begin{itemize}
\item \rn{simplify}, 
\item \rn{move all} if \rn{autoMove = true}, 
\item \rn{remove\_assign useless} if \rn{autoRemoveAssignFindCond = false}, \\
  \rn{remove\_assign findcond} if \rn{autoRemoveAssignFindCond = true},
\item \rn{SArename random} if \rn{autoSARename = true},
\item and \rn{merge\_branches} if \rn{autoMergeBranches = true}
\end{itemize}
had been called.

\subsubsection{\rn{success simplify}}\label{sec:success_simplify}

The transformation \rn{success simplify} is a combination of \rn{success} (Section~\ref{sec:success}) and \rn{simplify} (Section~\ref{sec:simplify}), with the following addition. First, in the \rn{success} step, the command \rn{success simplify} collects information that is known to be true when the adversary manages to break at least one of the desired properties. Then, the first iteration of the \rn{simplify} step removes parts of the game that contradict this information and replaces them with $\keventabort{\kwf{adv\_loses}}$.

\bb{TODO when $\lppf(P,\tup{M}) \in \fset$, we execute $\fset \leftarrow \fset \cup \convertelsefind(\ef_P\{\tup{M}/\tup{i}\}, \fset)$ (all variables known to be defined in $\fset$ are defined at $\ppf(P,\tup{M})$) In the implementation, instead of using $\lppf$, we directly add the $\elsefind$ facts at $P$ and later call $\convertelsefind$.}%

In more detail, \rn{success simplify} collects a set $\lset^-$ of $(\vset, \fset)$ and a set of formulas $\lset^+$. If the adversary breaks a desired security property, then either there exists a set of facts $\fset$ in $\lset^-$ that holds or there exists a formula in $\lset^+$ that holds. The sets of facts $\lset^-$ correspond to cases in which the proof of the security property failed; their probability may be high. The associated variables $\vset$ are replication indices and non-process variables that occur in $\fset$. The formulas in $\lset^+$ correspond to cases in which the security property was proved (up to a certain probability); these formulas are negations of the formulas that prove the security property in the considered case; the probability that they hold is bounded by the equational prover of CryptoVerif. The contents of $\lset^+$ does not influence the game obtained after the transformation. It is useful to compute the probability difference coming from the transformation. The sets $\lset^-$ and $\lset^+$ are computed as follows:
\begin{itemize}
\item In case there is an indistinguishability query, or a (one-session or bit) secrecy query on a variable $x$ not defined only by $\NEW$ or by assignments of variables defined by $\NEW$, or a correspondence query $\forall \tup{x}:\tup{T}; \psi \Rightarrow \exists \tup{y}:\tup{T}'; \phi$ with some event $e$ in $\psi$ such that the game contains $\kevent{e(\tup{M})}$ and some term in $\tup{M}$ is not simple, no information is collected at all and \rn{simplify} is not performed.

\item For each correspondence query $\fevent{e} \Rightarrow \false$ where $e$ is a non-unique event, for every $\pp$ that executes $\fevent{e}$,
  for every $\case$, $(\theta'\replidx{\pp}, \ab \theta'\fset^0_{\fevent{e},\pp,\case})$ is added to $\lset^-$, for some $\theta'$ renaming of $\replidx{\pp}$ to fresh replication indices.
  (Indeed, in order to break $\fevent{e} \Rightarrow \false$, event $e$ must be executed, so the facts $\fset^0_{\fevent{e},\pp,\case}$ at some execution of event $e$ hold.)
  
\item For each other correspondence query $\forall \tup{x}:\tup{T}; F_1 \wedge \dots \wedge F_m \Rightarrow \exists \tup{y}:\tup{T}'; \phi$,
  let $\corresp = \sem{\forall \tup{x}:\tup{T}; F_1 \wedge \dots \wedge F_m \Rightarrow \exists \tup{y}:\tup{T}'; \phi}$ and $\sset_0 = \{ (\pp_1,\case_1,\dots, \pp_m,\case_m) \mid \forall j \leq m, \pp_j$ executes $F_j$ and $\case_j$ is a case for $\fset_{\pp_j, \case_j}\}$;
  by trying to prove the correspondence, we build a subset $\sset_1$ of $\sset_0$, a family of substitutions $\theta$, and a pseudo-formula $\coll$ such that $\prove{\corresp}(\coll, \theta, \sset_1)$. ($\theta$ and $\coll$ are computed incrementally on the successful cases in the proof of the correspondence.)
    
  For all $(\pp_1,\case_1,\dots, \pp_m,\case_m) \in \sset_0 \setminus \sset_1$ (\rn{success} fails to prove the correspondence for those cases):
\begin{itemize}
\item 
  If the query is $\forall \tup{x}:\tup{T}; \fevent{e(\tup{N})} \Rightarrow \false$ and the process or term at $\pp_1$ is $\kevent{e(\tup{M})}; \ab \dots$, then 
  $(\tup{x} \cup \theta'\replidx{\pp_1}, \ab \theta'\fset_{\pp_1, \case_1} \cup \{ \lppf(\pp_1, \theta'\replidx{\pp_1}), \ab \tup{N} = \theta'\tup{M} \})$ is added to $\lset^-$, where $\theta'$ is a renaming of $\replidx{\pp_1}$ to fresh replication indices. (We can stop the trace just after event $e$ without changing the truth of the query, and that is more precise because we can use the $\elsefind$ facts at $e$.)

\item Otherwise,
  $(\tup{x} \cup \bigcup_{j=1}^m \theta_j\replidx{\pp_j}, \bigcup_{j=1}^m \theta_j\fset_{F_j,\pp_j,\case_j})$ is added to $\lset^-$, where for $j \leq m$, $\theta_j$ is a renaming of $\replidx{\pp_j}$ to fresh replication indices.
  (Indeed, in order to break the correspondence $\forall \tup{x}:\tup{T}; F_1 \wedge \dots \wedge F_m \Rightarrow \exists \tup{y}:\tup{T}'; \phi$, the events $F_1$, \dots, $F_m$ must be executed, so the facts $\fset_{F_j,\pp_j,\case_j}$ for $j \leq m$ that hold when $F_1$, \dots, $F_m$ are executed certainly hold when the correspondence is broken. In principle, we could add $\neg \exists \tup{y}:\tup{T}'; \phi$ to the facts $\bigcup_{j=1}^m \theta_j\fset_{F_j,\pp_j,\case_j}$ added to $\lset^-$. However, we have no way to express universal quantification in general in known facts, so when $\tup{y}$ is not empty, we could not add $\forall \tup{y}:\tup{T}'; \neg \phi$ but would end up adding $\neg \phi$ which in fact means $\exists \tup{y}:\tup{T}'; \neg \phi$. That would remain sound assuming the types in $\tup{T}'$ are not empty, but would be weaker. Moreover, in practice, we end up having to distinguish precisely the case in which $\exists \tup{y}:\tup{T}'; \phi$ can be proved from the case in which it cannot, which can typically be done by inserting an appropriate $\FIND$. We generally insert a Shoup event $e$ in the $\kw{else}$ branch of that $\FIND$, triggered when the correspondence cannot be proved, a case for which we want to bound the probability. After that, it remains to prove $\fevent{e} \Rightarrow \false$: we apply \rn{success simplify} to that correspondence. Adding $\neg\phi = \true$ would not change anything for that correspondence, and we exploit that the condition of the inserted $\FIND$ is false at event $e$, which gives us more precise information than having added $\neg\phi$ for the initial correspondence.)

\end{itemize}
  Moreover, $\neg \yctran{\prove{\corresp}(\coll, \theta, \sset_1)}$ is added to $\lset^+$. (\rn{success} proves the correspondence for the cases in $\sset_1$.)

\item For each secrecy, one-session secrecy, or bit secrecy query on a variable $x$ defined only by $\NEW$ or by assignments of variables defined by $\NEW$, let $\sset_0 = \{ \pp \mid \pp$ follows a definition of $x \}$
and $\sset_1 = \{ \pp \in \sset_0 \mid \prove{\secrone(x)}(\pp)\}$.

For each $\pp \in \sset_0 \setminus \sset_1$ (\rn{success} fails to prove one-session secrecy for those cases), $(\theta\replidx{\pp}, \theta\fset_{\pp})$ is added to $\lset^-$, where $\theta$ is a renaming of $\replidx{\pp}$ to fresh replication indices. (Indeed, if secrecy, one-session secrecy, or bit secrecy of $x$ is broken, a definition of $x$ must have been executed, so the facts $\fset_{\pp}$ at that definition hold.)

Moreover, $\neg \yctran{\prove{\secrone(x)}(\sset_1)}$ is added to $\lset^+$. (\rn{success} proves one-session secrecy for the cases in $\sset_1$.)
  %No elsefind facts are added to $\lset$: the formula is correct, \fset_{\pp} does not contain elsefind facts.

Additionally, if the considered query is a secrecy query, then for each $\pp_1, \pp_2 \in \sset_0$, 
let $z_1[\tup{M_1}] = \defrestr_{\pp_1}(x)$ and $z_2[\tup{M_2}] = \defrestr_{\pp_2}(x)$.
If $z_1 \neq z_2$, then the definitions at $\pp_1$ and $\pp_2$ are proved to be independent, and nothing is added to $\lset^-$ nor $\lset^+$.
  If $z_1 = z_2$, let $\tup{i}$ be the current replication indices at the definition of $x$, let $\theta_1$ and $\theta_2$ be two distinct renamings of $\tup{i}$ to fresh replication indices, 
let $\tup{i}_1 = \theta_1\tup{i}$ and $\tup{i}_2 = \theta_2\tup{i}$,
let $\fset = \theta_1\fset_{\pp_1} \cup \theta_2\fset_{\pp_2} \cup \{ \theta_1 \tup{M_1} = \theta_2 \tup{M_2},\ab \tup{i}_1 \neq \tup{i}_2\}$. If $\fset$ yields a contradiction, then the definitions at $\pp_1$ and $\pp_2$ are proved to be independent up to a small probability, $\exists \tup{i}_1, \exists \tup{i}_2, \bigwedge \fset$ is added to $\lset^+$. Otherwise, $(\tup{i}_1 \cup \tup{i}_2, \fset)$ is added to $\lset^-$.
  %We infer facts from defined variables and elsefind facts at their definition point, via Facts_of_elsefind.get_facts_of_elsefind_facts; that reasoning can be included in the computation of $\fset_P$ and "yields a contradiction". No elsefind facts are added to $\lset$: the formula is correct, \fset_P does not contain elsefind facts.

\end{itemize}
%\fset_P is valid at any time in the future = it does not contain elsefind facts
%\fset_P at least at end of block = \fset_P \cup \futfset_P
In the \rn{simplify} step, the set $\lset^-$ is used as follows: for each program point $\pp$ not in a condition of $\FIND$\bbnote{``$\pp$ not in a condition of $\FIND$'' is present in the implementation, but not required for soundness},
if for all $(\vset, \fset) \in \lset^-$, $\fset_{\pp} \cup \fset$ yields a contradiction,
then the code at $\pp$ is replaced with $\keventabort{\advloses}$.
%\ppf used for inferring facts via incompatibility
(The reason why $\vset$ is needed in the implementation is for the optimization of probabilities of collisions:
we determine using which indices in $\vset$ we get the smaller bound for the number of
collisions.)
%in the dependency analysis that depends on the order of definition in Facts_of_elsefind.get_facts_of_elsefind_facts -> Depanal.add_term_collision

The probability that a security property is broken before the transformation
and not after is then bounded by the probability that a modified program
point $\pp$ is reached and the adversary breaks the property.
If that breach corresponds to a case in $\lset^+$, the probability of the breach itself is bounded by construction of $\lset^+$.
If that breach corresponds to a case in $\lset^-$, the probability of the breach and reaching $\pp$ is bounded because
for all $(\vset, \fset) \in \lset^-$, $\fset_{\pp} \cup \fset$ yields a contradiction, which bounds the probability that the facts $\fset_{\pp} \cup \fset$ hold for some $(\vset, \fset) \in \lset$, and $\fset_{\pp}$ holds when $\pp$ is reached while some $\fset$ in $\lset$ holds when the adversary breaks the property.
This is formalized by the following lemma.

\begin{lemma}
  The transformation \rn{success\ simplify} requires and preserves
  Properties~\ref{prop:nointervaltypes}, \ref{prop:noreschan}, \ref{prop:channelindices}, \ref{prop:notables}, \ref{prop:autosarename}, and~\ref{prop:expand}.

  If transformation \rn{success\ simplify} transforms $G$ into $G'$,
the distinguisher $D$ is a disjunction of Shoup and non-unique events,
the property $\prop$ and the disjuncts in $D$ correspond to active queries,
$\lset^- = \{ (\vset_j, \fset_j) \mid j \in J\}$,
the modified program points are $\pp_k$ for $k \in K$,
$\fsetmod_k = \fset_{\pp_k}$,
for all evaluation contexts $C$ acceptable for $G$,
$\Prss{C[G]}{ \left(\bigvee_{j \in J, k \in K} \exists \replidx{\pp_k}, \exists \vset_j, \bigwedge \fset_j \wedge \bigwedge \fsetmod_k\right) \vee \left(\bigvee \lset^+\right)} \leq p(C)$,
%$p$ is an upper bound on the probability that required equalities do not hold.
and $\bound{G'}{V}{\prop}{D}{p'}$, then
$\bound{G}{V}{\prop}{D}{p+p'}$.
\end{lemma}

Unfortunately, we cannot prove
$\dset,\dsetsnu : G, \ab D, \ab \usedevents \indistev{V}{p} G', \ab D,
\ab \usedevents \cup \{ \advloses \}$
for the transformation \rn{success\ simplify}, because, in case of
(one-session or bit) secrecy of a variable $x$, the inequality needed for this property
may not hold: we need to take into account that, when $x$ is not defined, traces that execute $\sevent$
and those that execute $\sbarevent$ compensate in the computation of
$\Advtev{G}{\prop}{C}{D}$ in order to prove the soundness of this
transformation. We cannot prove this soundness independently for
$\Pr[C[G]:\sevent]$ and for $\Pr[C[G]:\neg\sbarevent]$.

Moreover, in the implementation, for (one-session or bit) secrecy properties, a probability
$2p(C)$ is added instead of just $p(C)$ as shown by the lemma above. 
The factor 2 is difficult to avoid because other simplifications are
performed at the same time as described in the transformation \rn{simplify}
(Section~\ref{sec:simplify}), and the factor 2 is needed for these 
transformations.

\begin{proof}
  \paragraph{Fact~1.}
  Let $\corresp$ be a correspondence not of the form $\forall \tup{x}\in\tup{T}, \fevent{e(\tup{N})} \Rightarrow \false$, or the correspondence $\fevent{e} \Rightarrow \false$ for some Shoup event $e$.
Let $C$ be any evaluation context acceptable for $G$ with public variables $V$
that does not contain events used by $\corresp$.
Let $\trace$ be any full trace of $C[G]$ that does not execute any non-unique event of $G$ and such that $\trace \vdash \neg\corresp$.
Then $\trace \vdash \left(\bigvee_{j \in J} \exists \vset_j, \bigwedge\fset_j\right) \vee \left(\bigvee \lset^+\right)$.

  \paragraph{Proof of Fact~1.}
By Lemma~\ref{lem:icorresp},
for any substitutions $\theta(\pp_1,\case_1,\ab \dots,\ab \pp_m,\case_m)$ equal to the 
identity on $\tup{x}$,
for any pseudo-formula $\coll$,
\[\textstyle\trace \vdash \neg \yctran{\prove{\corresp}(\coll, \theta, \sset_0)}\]
So
\[
\textstyle\trace \vdash
(\neg \yctran{\prove{\corresp}(\coll, \theta, \sset_1)}) \vee
\left(
  \bigvee_{(\pp_1, \case_1, \dots, \pp_m, \case_m) \in \sset_0 \setminus \sset_1} 
  \neg\yctran{\prove{\corresp}(\coll, \theta, \pp_1,\case_1, \ab \dots, \ab \pp_m,\case_m)}
\right)\,.\]
Moreover, $\neg\yctran{\fset \shows_{\theta}^{\sri,\vset,\coll} \phi} \Rightarrow \exists \tup{z} \in \tup{T}'', \bigwedge \fset$ where $\tup{z} = \vset$ and $\tup{T}''$ are the types of these variables, by an easy induction on $\phi$,
so $\neg\yctran{\prove{\corresp}(\coll, \ab \theta, \ab \pp_1,\case_1, \ab \dots, \ab \pp_m,\ab \case_m)} \Rightarrow \exists \theta_1\replidx{\pp_1}, \ab\dots, \ab\exists \theta_m\replidx{\pp_m}, \ab \exists \tup{x} \in \tup{T}, \ab \bigwedge\theta_1\fset_{F_1,\pp_1,\case_1} \cup \dots \cup \theta_m\fset_{F_m,\pp_m,\case_m}$ where, for $j\leq m$, $\theta_j$ is a renaming of $\replidx{\pp_j}$ to fresh replication indices.
So
\[\begin{split}
\trace \vdash
&\textstyle(\neg\yctran{\prove{\corresp}(\coll, \theta, \sset_1)}) \vee {}\\
&\textstyle\left(
  \bigvee_{(\pp_1, \case_1, \dots, \pp_m, \case_m) \in \sset_0 \setminus \sset_1} 
  \exists \theta_1\replidx{\pp_1}, \dots, \exists \theta_m\replidx{\pp_m}, \exists \tup{x} \in \tup{T}, \bigwedge \theta_1\fset_{F_1,\pp_1,\case_1} \cup \dots \cup \theta_m\fset_{F_m,\pp_m,\case_m}
\right)\,.
\end{split}\]
The formula $\neg\yctran{\prove{\corresp}(\coll, \theta, \sset_1)}$ is added to $\lset^+$
and $(\tup{x} \cup \bigcup_{j=1}^m \theta_j\replidx{\pp_j}, \bigcup_{j=1}^m \theta_j\fset_{F_j,\pp_j,\case_j})$ is added to $\lset^-$ when $(\pp_1,\case_1,\dots, \pp_m,\case_m) \in \sset_0 \setminus \sset_1$ 
(recall that, when $e$ is Shoup event, $e$ is always executed by $\keventabort{e}$)
so 
$\trace \vdash \left(\bigvee_{j \in J} \exists \vset_j, \bigwedge \fset_j\right) \vee \left(\bigvee \lset^+\right)$.

\paragraph{Fact~2.}
Let $e$ be a non-unique event.
Let $C$ be any evaluation context acceptable for $G$ with public variables $V$
that does not contain $e$.
Let $\trace$ be any full trace of $C[G]$ such that $\trace \vdash e$.
Then $\trace \vdash \left(\bigvee_{j \in J} \exists \vset_j, \bigwedge \fset_j\right) \vee \left(\bigvee \lset^+\right)$.

\paragraph{Proof of Fact~2.}
By Lemma~\ref{lem:fsetFP}, there exist a program point $\pp$ (in $G$) and a case $\case$ such that, for any $\theta'$ renaming of $\replidx{\pp}$ to fresh replication indices, there exists a mapping $\sigma$ with domain $\theta'\replidx{\pp}$ such that
$\trace, \sigma \vdash \theta'\fset^0_{\fevent{e}, \pp, \case}$.
Let $(\vset, \fset) = (\theta'\replidx{\pp}, \theta'\fset^0_{\fevent{e}, \pp, \case})$ be the element of $\lset^-$ for $\pp$ and $\case$, in the treatment of correspondence $\fevent{e} \Rightarrow \false$.
We have $\trace \vdash \exists \vset, \bigwedge\fset$.
Therefore, $\trace \vdash \bigvee_{j \in J} \exists \vset_j, \bigwedge\fset_j$.

\paragraph{Fact~3.}
Let $C$ be any evaluation context acceptable for $G$ with public variables $V$.
Let $\trace$ be any trace of $C[G]$.
If $\trace \vdash \neg\bigvee_{k \in K} \exists \replidx{\pp_k}, \bigwedge \fsetmod_k$, 
then there is no configuration in $\trace$ at a modified program point $\pp_k$.

\paragraph{Proof of Fact~3.}
By contraposition, if there is a configuration $\conf = E, \sigma, {}^{\pp_k} M, \tblcts, \evseq$ or
$\conf = E, (\sigma, {}^{\pp_k} P), \pset, \cset, \tblcts, \evseq$
at a modified program point $\pp_k$ in trace $\trace$,
then
let $\theta$ be a renaming of $\replidx{\pp_k}$ to fresh indices
and $\venv = \{ \theta\replidx{\pp_k} \mapsto \sigma\replidx{\pp_k} \}$;
by Corollary~\ref{lem:factsP}, $\trace, \venv \vdash \theta\facts_{\pp_k}$.
So $\trace, \venv \vdash \theta \fsetmod_k$,
so $\trace \vdash \exists \replidx{\pp_k}, \bigwedge \fsetmod_k$.

\bigskip\noindent
We perform the proof for each query separately.

\paragraph{Case 1: $\prop$ is some correspondence $\corresp$ different from $\forall \tup{x}\in \tup{T}, \fevent{e(\tup{N})} \Rightarrow \false$ (including $\prop$ is $\true$).}
Let $C$ be any evaluation context acceptable for $G$ with public variables $V$
that does not contain events used by $\corresp$, $D$, nor non-unique events of $G$.

Consider any full trace $\trace$ of $C[G]$ such that $\trace \vdash (\neg\corresp \vee D) \wedge\neg\nonunique{G,D}$.
Let us show that $\trace \vdash \left(\bigvee_{j \in J} \exists \vset_j, \bigwedge\fset_j\right) \vee \left(\bigvee \lset^+\right)$.

Case 1.1: $\trace$ does not execute any non-unique event of $G$. Then $\trace \vdash \neg\corresp \vee D_s$ where $D_s$ is the disjunction of Shoup events in $D$.

Case 1.1.1: $\trace \vdash \neg\corresp$
We conclude by Fact~1.

Case 1.1.2: $\trace \vdash e$ for some Shoup event $e$ in $D_s$.
We conclude by Fact~1 for the correspondence $\fevent{e} \Rightarrow \false$.

Case 1.2: $\trace$ executes a non-unique event of $G$. Then $\trace \vdash e$ for some non-unique event $e$ in $G$ and in $D$.
We conclude by Fact~2.

If $\trace \vdash \bigvee_{k \in K} \exists \replidx{\pp_k}, \bigwedge \fsetmod_k$, 
then $\trace \vdash \left(\bigvee_{j \in J, k \in K} \exists \replidx{\pp_k}, \exists \vset_j, \bigwedge \fset_j \wedge \bigwedge \fsetmod_k\right) \vee \left(\bigvee \lset^+\right)$.

Otherwise, by Fact~3, there is no configuration in $\trace$ at a modified program point $\pp_k$,
so $\trace$ has a matching trace in $C[G']$ that also satisfies $(\neg\corresp \vee D) \wedge\neg\nonunique{G',D}$.

We conclude that
{\allowdisplaybreaks\begin{align*}
\Advtev{G}{\corresp}{C}{D} 
\leq {}&\Pr[C[G]: (\neg\corresp \vee D) \wedge\neg\nonunique{G,D}]\\
\begin{split}
\leq {}&\Pr[C[G]: \left(\bigvee_{j \in J, k \in K} \exists \replidx{\pp_k}, \exists \vset_j, \bigwedge \fset_j \wedge \bigwedge \fsetmod_k\right) \vee \left(\bigvee \lset^+\right)]\\*
&{} + \Pr[C[G']: (\neg\corresp \vee D) \wedge\neg\nonunique{G',D}]
\end{split}\\
\begin{split}
\leq {}&\Prss{C[G]}{ \left(\bigvee_{j \in J, k \in K} \exists \replidx{\pp_k}, \exists \vset_j, \bigwedge \fset_j \wedge \bigwedge \fsetmod_k\right) \vee \left(\bigvee \lset^+\right)}\\*
&{} + \Pr[C[G']: (\neg\corresp \vee D) \wedge\neg\nonunique{G',D}]
\end{split}\tag*{by Lemma~\ref{lem:Pr-Prss}}\\*
\leq {}& p(C) + \Advtev{G'}{\corresp}{C}{D}\\*
\leq {}& p(C) + p'(C) \tag*{since $\bound{G'}{V}{\prop}{D}{p'}$}
\end{align*}}%
Hence, we have $\bound{G}{V}{\prop}{D}{p+p'}$.

\paragraph{Case 2: $\prop$ is some correspondence $\corresp = \forall \tup{x}\in \tup{T},\fevent{e(\tup{N})} \Rightarrow \false$.}
Let $C$ be any evaluation context acceptable for $G$ with public variables $V$
that does not contain events used by $\corresp$, $D$, nor non-unique events of $G$.

Consider a full trace $\trace$ of $C[G]$ such that $\trace \vdash (\neg\corresp \vee D) \wedge\neg\nonunique{G,D}$.
Let us show that either some prefix of $\trace$ satisfies 
$\left(\bigvee_{j \in J, k \in K} \exists \replidx{\pp_k}, \exists \vset_j, \bigwedge \fset_j \wedge \bigwedge \fsetmod_k\right) \vee \left(\bigvee \lset^+\right)$
or there is a matching trace in $C[G']$ that satisfies $(\neg\corresp \vee D) \wedge\neg\nonunique{G',D}$.
\bbnote{I would be happy to factorize this proof a bit more, perhaps by showing ``for some prefix $\trace'$ of $\trace$, $\trace' \vdash \left(\bigvee_{j \in J} \exists \vset_j, \bigwedge_{F \in \fset_j} F\right) \vee \left(\bigvee \lset^+\right)$ and ...'' but finding a simple property for $\trace'$ is not that easy.}%

We have $\trace \vdash \exists \tup{x}\in \tup{T}, \fevent{e(\tup{N})} \vee D$.

Case~2.1: $\trace$ executes a Shoup event $e'$ of $D$.
Then $\trace \vdash e'$
and $\trace$ is actually a full trace that does not execute any non-unique event of $G$.
By Fact~1, $\trace \vdash \left(\bigvee_{j \in J} \exists \vset_j, \bigwedge_{F \in \fset_j} F\right) \vee \left(\bigvee \lset^+\right)$.

Case~2.2: $\trace$ executes a non-unique event of $G$. Then $\trace \vdash e'$ for some non-unique event $e'$ in $G$ and in $D$.
Then $\trace$ is actually a full trace. By Fact~2, $\trace \vdash \left(\bigvee_{j \in J} \exists \vset_j, \bigwedge_{F \in \fset_j} F\right) \vee \left(\bigvee \lset^+\right)$.

In cases~2.1 and~2.2, 
\begin{itemize}
\item If $\trace \vdash \bigvee_{k \in K} \exists \replidx{\pp_k}, \bigwedge \fsetmod_k$, then 
$\trace \vdash \left(\bigvee_{j \in J, k \in K} \exists \replidx{\pp_k}, \exists \vset_j, \bigwedge \fset_j \wedge \bigwedge \fsetmod_k\right) \vee \left(\bigvee \lset^+\right)$.
\item Otherwise, by Fact~3, there is no configuration in $\trace$ at a modified program point $\pp_k$,
so $\trace$ has a matching trace in $C[G']$ that also satisfies $(\neg\corresp \vee D) \wedge\neg\nonunique{G',D}$.
\end{itemize}

Case~2.3: $\trace$ does not execute any non-unique event of $G$ nor any Shoup event of $D$.
Then $\trace \vdash \exists \tup{x}\in \tup{T}, \fevent{e(\tup{N})}$.

Let $\evseq$ be the sequence of events and $E$ be the environment in the last configuration
of $\trace$.
There is a mapping $\venv$ of the variables $\tup{x}$ to their values such that
$\trace,\venv \vdash e(\tup{N})$.
As in the proof of Lemma~\ref{lem:fsetFP}, 
there exist $\tup{a}$ and $\step \in \mathbb{N}$
such that $\venv, \tup{N} \evalterm \tup{a}$ and 
$\evseq(\step) = (\pp, \tup{a}_0): e(\tup{a})$ for some $\pp$ and $\tup{a}_0$.
The rule of the semantics that may have added this element to $\evseq$
is \eqref{sem:event}, \eqref{sem:eventabort}, \eqref{sem:ctxevent},
\eqref{sem:finde} or~\eqref{sem:eventt}.
(It cannot be~\eqref{sem:find3} nor~\eqref{sem:get3} because $e$ is not a non-unique event.
It cannot be~\eqref{sem:gete} because $G$ does not contain $\GET$ by Property~\ref{prop:notables}.)
\begin{itemize}
\item Case~2.3.1:
In case~\eqref{sem:event}, 
we have reductions
\begin{align*}
&\conf = E_0, (\sigma_0, \pptag\kevent{e(\tup{M})};P), \pset_0, \cset_0, \tblcts_0,\evseq_0 \\
&\qquad \red{p_0}{\ix_0} \dots \red{p_1}{\ix_1} E_1, (\sigma_1, \pptag\kevent{e(\tup{a})};P), \pset_0, \cset_0, \tblcts_1,\evseq_1 \\
&\qquad \red{1}{} \conf' = E_1, (\sigma_1, P), \pset_0, \cset_0, \tblcts_1,(\evseq_1, (\pp, \image(\sigma_1)):e(\tup{a}))
\end{align*}
where $\pptag\kevent{e(\tup{M})};P$ is a subprocess of $C[G]$ up to renaming of channels, by any number of applications of~\eqref{sem:ctx} and a final application of~\eqref{sem:event}.
The terms $\tup{M}$ are simple terms (when some term in $\tup{M}$ is not simple, no information is collected at all and \rn{simplify} is not performed), so in fact
\begin{align*}
&\conf = E_0, (\sigma_0, \pptag\kevent{e(\tup{M})};P), \pset_0, \cset_0, \tblcts_0,\evseq_0 \\
&\qquad \red{1}{}^* E_0, (\sigma_0, \pptag\kevent{e(\tup{a})};P), \pset_0, \cset_0, \tblcts_0,\evseq_0 \\
&\qquad \red{1}{} \conf' = E_0, (\sigma_0, P), \pset_0, \cset_0, \tblcts_0,(\evseq_0, (\pp, \image(\sigma_0)):e(\tup{a}))
\end{align*}
  
\item Case~2.3.2: In case~\eqref{sem:eventabort}, we have a reduction
\begin{align*}
  &\conf = E, (\sigma_0, \pptag P), \pset_0, \cset_0, \tblcts_0,\evseq_0\\
  &\qquad \red{p}{\ix} \conf' = E, (\sigma_0, \kw{abort}), \pset_0, \cset_0, \tblcts_0,(\evseq_0, (\pp, \image(\sigma_0)):e)
\end{align*}
where $\pptag P$ is a subprocess of $C[G]$ up to renaming of channels, by~\eqref{sem:eventabort}.

\item Case~2.3.3: In case~\eqref{sem:eventt}, we have reductions
\begin{align*}
&\conf = E_0, \sigma_0, \pptag\kevent{e(\tup{M})};N, \tblcts_0,\evseq_0 \\
&\qquad \red{p_0}{\ix_0} \dots \red{p_1}{\ix_1} E_1, \sigma_1, \pptag\kevent{e(\tup{a})};N, \tblcts_1,\evseq_1 \\
&\qquad \red{1}{} E_1, \sigma_1, N, \tblcts_1, (\evseq_1, (\pp, \image(\sigma_1)):e(\tup{a}))
\end{align*}
where $\pptag\kevent{e(\tup{M})};N$ is a subterm of $C[G]$, by any number of applications of~\eqref{sem:ctxt} and a final application of~\eqref{sem:eventt}.
The terms $\tup{M}$ are simple terms, so in fact
\begin{align*}
&\conf = E_0, \sigma_0, \pptag\kevent{e(\tup{M})};N, \tblcts_0,\evseq_0 \\
&\qquad \red{1}{}^* E_0, \sigma_0, \pptag\kevent{e(\tup{a})};N, \tblcts_0,\evseq_0 \\
&\qquad \red{1}{} E_0, \sigma_0, N, \tblcts_0, (\evseq_0, (\pp, \image(\sigma_0)):e(\tup{a}))
\end{align*}
By Invariant~\ref{invtfc}, $\pp$ is not inside a condition of $\FIND$ or $\GET$,
so by Lemma~\ref{lem:condfind}, $\conf$ is not in the derivation of an hypothesis of a rule for $\FIND$ or $\GET$.
The only rule for processes other than those for $\FIND$ or $\GET$ that evaluates a non-simple term is~\eqref{sem:ctx} and similarly, the only rule for terms other than those for $\FIND$ or $\GET$ that evaluates a term is~\eqref{sem:ctxt}, so we have
\begin{align*}
  &E_0, (\sigma_0, C_0[C_1[\dots C_k[\pptag\kevent{e(\tup{M})};N]\dots]]), \pset_0, \cset_0, \tblcts_0,\evseq_0\\
  &\qquad \red{1}{}^* \conf' = E_0, (\sigma_0, C_0[C_1[\dots C_k[N]\dots]]), \pset_0, \cset_0, \tblcts_0, (\evseq_0, (\pp, \image(\sigma_0)):e(\tup{a}))
\end{align*}
for some $C_0$ context defined in
Figure~\ref{fig:proccontexts}, $k \in \mathbb{N}$, and $C_1$, \dots,
$C_k$ contexts defined in Figure~\ref{fig:termcontexts}, by $k$ applications of~\eqref{sem:ctxt}
and one application of~\eqref{sem:ctx}.
\bbnote{In other words, we are in the second case of Lemma~\ref{lem:start_from_pp}, Property~\ref{start_from_pp5}.}%

\item Case~2.3.4: In case~\eqref{sem:ctxevent}, we have
  \[\conf = E_0, \sigma_0, \pptag \keventabort{e}, \tblcts_0,\evseq_0 \red{p}{\ix} E_0, \sigma_0, \keventabort{(\pp, \image(\sigma_0)):e}, \tblcts_0,\evseq_0\]
  by \eqref{sem:eventabortt}
  where $\pptag\keventabort{e}$ is a subterm of $C[G]$, followed by applications \eqref{sem:findte}, \eqref{sem:ctxt}, \eqref{sem:ctxeventt}, \eqref{sem:finde}, \eqref{sem:ctx}, and~\eqref{sem:ctxevent}.
  We let $\conf'$ be the last configuration of $\trace$.

\end{itemize}
Let $\theta'$ be a renaming of $\replidx{\pp}$ to fresh replication indices.
We have
\begin{align*}
  &\prove{\corresp}(\coll, \theta_0, \pp, \case) = (\theta'\fset_{\fevent{e(\tup{N})}, \pp, \case}\text{ yields a contradiction}),\\
  &\yctran{\prove{\corresp}(\coll, \theta_0, \pp, \case)} = \forall \theta'\replidx{\pp},\ab \forall \tup{x} \in \tup{T}, \ab \neg\bigwedge \theta'\fset_{\fevent{e(\tup{N})}, \pp, \case},\\
  &\coll\text{ is a pseudo-formula with all leaves $\bot$, so }\yctran{\vdash\coll} = \true\,.
\end{align*}
If $(\pp, \case) \in \sset_0 \setminus \sset_1$, then in cases~2.3.1 and~2.3.3, $(\tup{x} \cup \theta'\replidx{\pp}, \ab \theta'\fset_{\pp, \case} \cup \{ \lppf(\pp, \ab \theta'\replidx{\pp}), \ab \tup{N} = \theta'\tup{M} \})$ is added to $\lset^-$ and in cases~2.3.2 and~2.3.4, $(\theta'\replidx{\pp}, \theta'\fset_{\fevent{e}, \pp, \case})$ is added to $\lset^-$.
Moreover, $\neg \yctran{\prove{\corresp}(\coll, \theta_0, \sset_1)} = \bigvee_{(\pp,\case) \in \sset_1}\ab \exists \theta'\replidx{\pp},\exists \tup{x} \in \tup{T}, \bigwedge\theta'\fset_{\fevent{e(\tup{N})}, \pp, \case}$ is added to $\lset^+$.

As in the proof of Lemma~\ref{lem:fsetFP},
we have $\sigma_0 = [\replidx{\pp} \mapsto \tup{a}_0]$.
Let $\sigma = \{ \theta'\replidx{\pp} \mapsto \tup{a}_0 \}$.
As in the proof of Lemma~\ref{lem:fsetFP}, we have
$\trace, \sigma\cup \venv \vdash \theta'\fset_{\fevent{e(\tup{N})}, \pp, \case}$,
since $\trace$ does not execute any non-unique event of $G$.
\begin{itemize}
\item In cases~2.3.1 and~2.3.3 when $(\pp, c) \in \sset_1$ and in cases~2.3.2 and~2.3.4, 
we have $\trace \vdash \left(\bigvee_{j \in J} \exists \vset_j, \bigwedge\fset_j\right) \vee \left(\bigvee \lset^+\right)$.
\begin{itemize}
\item \sloppy If $\trace \vdash \bigvee_{k \in K} \exists \replidx{\pp_k}, \bigwedge \fsetmod_k$, then 
$\trace \vdash \left(\bigvee_{j \in J, k \in K} \exists \replidx{\pp_k}, \exists \vset_j, \bigwedge \fset_j \wedge \bigwedge \fsetmod_k\right) \vee \left(\bigvee \lset^+\right)$.
\item Otherwise, by Fact~3, there is no configuration in $\trace$ at a modified program point $\pp_k$,
so $\trace$ has a matching trace in $C[G']$ that also satisfies $(\neg\corresp \vee D) \wedge\neg\nonunique{G',D}$.
\end{itemize}

\item In cases~2.3.1 and~2.3.3 when $(\pp, c) \in \sset_0 \setminus \sset_1$, 
let $\trace'$ be the prefix of $\trace$ that stops at $\conf'$.
We have $\conf \beforetr{\trace} \conf'$, so by Corollary~\ref{cor:factsP},
$\trace \before \conf', \sigma \vdash \theta'\facts_{\pp,\case}$, that is,
$\trace', \sigma \vdash \theta'\facts_{\pp,\case}$.
We have $E_0, \sigma_0, \tup{M} \evalterm \tup{a}$, so
$E_0, \sigma, \theta'\tup{M} \evalterm \tup{a}$.
The environment $E_{\trace'}$ extends $E_0$, so
$E_{\trace'}, \sigma, \theta'\tup{M} \evalterm \tup{a}$, so
$\trace', \sigma\cup\venv \vdash \theta'\tup{M} = \tup{N}$.
Since $E_{\trace \before \conf'} = E_{\trace \before \conf}$ and $\sigma_{\conf'} = \sigma_{\conf}$,
$\trace', \sigma \vdash \lppf(\pp, \theta'\tup{i})$.
Hence $\trace' \vdash \bigvee_{j \in J} \exists \vset_j, \bigwedge \fset_j$.
\begin{itemize}
\item If $\trace' \vdash \bigvee_{k \in K} \exists \replidx{\pp_k}, \bigwedge \fsetmod_k$, then 
$\trace' \vdash \left(\bigvee_{j \in J, k \in K} \exists \replidx{\pp_k}, \exists \vset_j, \bigwedge \fset_j \wedge \bigwedge \fsetmod_k\right) \vee \left(\bigvee \lset^+\right)$ and $\trace'$ is a prefix of $\trace$.
\item Otherwise, by Fact~3, there is no configuration in $\trace'$ at a modified program point $\pp_k$,
so $\trace$ has a matching trace in $C[G']$ that also satisfies $(\neg\corresp \vee D) \wedge\neg\nonunique{G',D}$.
(Note that $\trace$ may still be modified by \rn{success simplify}: the matching trace in $C[G']$ may execute a modified program point $\pp_k$
after $\conf'$, but still the matching trace executes $e(\tup{M})$, so satisfies $\neg\corresp$, and it does not execute any non-unique event of $G'$.)
\end{itemize}

\end{itemize}
We conclude that
\begin{align*}
\Advtev{G}{\corresp}{C}{D} 
= {}& \Pr[C[G]: (\neg\corresp \vee D) \wedge\neg\nonunique{G,D}]\\
\begin{split}
\leq {}&\Prss{C[G]}{\left(\bigvee_{j \in J, k \in K} \exists \replidx{\pp_k}, \exists \vset_j, \bigwedge \fset_j \wedge \bigwedge \fsetmod_k\right) \vee \left(\bigvee \lset^+\right)}\\
&{} + \Pr[C[G']: (\neg\corresp \vee D) \wedge\neg\nonunique{G',D}]
\end{split}\\*
\leq {}& p(C) + \Advtev{G'}{\corresp}{C}{D}\\*
\leq {}& p(C) + p'(C) \tag*{since $\bound{G'}{V}{\prop}{D}{p'}$}
\end{align*}
Hence, we have $\bound{G}{V}{\prop}{D}{p+p'}$.

\paragraph{Case 3: $\prop$ is $\secrone(x)$, $\secr(x)$, or $\secrbit(x)$ with $C = C'[C_{\prop}[\,]]$ and $x$ is defined only by $\NEW$ or by assignments of variables defined by $\NEW$.}
Let $C'$ be any evaluation context acceptable for $C_{\prop}[G]$ with public variables $V \setminus V_{\prop}$
that does not contain $\sevent$, $\sbarevent$, any event in $D$, nor any non-unique event of $G$.
\begin{align*}
\Advtev{G}{\prop}{C}{D} 
&= \Pr[C[G]: \sevent \vee D] - \Pr[C[G]: \sbarevent \vee \nonunique{G,D}]\\
&= \Pr[C[G]: \sevent] + \Pr[C[G]: D] - \Pr[C[G]: \sbarevent] -  \Pr[C[G]: \nonunique{G,D}]\tag*{since these events are mutually exclusive}\\
&= \Pr[C[G]: \sevent \wedge\neg\prop] + \Pr[C[G]: D]\\
&\quad - \Pr[C[G]: \sbarevent\wedge\neg\prop] -  \Pr[C[G]: \nonunique{G,D}]
\tag*{by Lemma~\ref{lem:sec}}\\
&= \Pr[C[G]: (\sevent \wedge\neg\prop) \vee D] - \Pr[C[G]: \sbarevent\wedge\neg\prop] -  \Pr[C[G]: \nonunique{G,D}]
\end{align*}
We have
\begin{itemize}
\item $\Pr[C[G']: \nonunique{G',D}] \leq \Pr[C[G]: \nonunique{G,D}]$, 
\item $\Pr[C[G']: \sbarevent\wedge\neg\prop] \leq \Pr[C[G]: \sbarevent\wedge\neg\prop]$
\end{itemize}
since a trace that executes a non-unique event or $\sbarevent$ in $G'$ cannot execute event $\advloses$,
so it executed without change in $G$.
Let us show that $\Pr[C[G]: (\sevent \wedge\neg\prop) \vee D] \leq p(C) + \Pr[C[G']: (\sevent \wedge\neg\prop) \vee D]$.
Let $\trace$ be a full trace of $C[G]$ such that $\trace \vdash (\sevent \wedge\neg\prop) \vee D$.
Let us show that $\trace \vdash \left(\bigvee_{j \in J} \exists \vset_j, \bigwedge\fset_j\right) \vee \left(\bigvee \lset^+\right)$.

Case 3.1: $\trace$ executes a non-unique event of $G$. Then $\trace \vdash e$ for some non-unique event $e$ in $G$ and in $D$.
We conclude by Fact~2.

Case 3.2: $\trace \vdash e$ for some Shoup event $e$ in $D$. Then 
$\trace$ does not execute any non-unique event of $G$.
We conclude by Fact~1 for the correspondence $\fevent{e} \Rightarrow \false$.

Case 3.3: $\trace \vdash \sevent \wedge\neg\prop$.
Since $\trace \vdash \neg\prop$, we have $\trace \vdash \neg\yctran{\prove{\prop}(\Stestpp(\trace))}$, so we are in one of the following two cases:
\begin{itemize}
\item There exists $\pp \in \Stestpp(\trace)$ such that $\trace \vdash \neg\yctran{\prove{\secrone(x)}(\pp)}$.
\begin{itemize}
\item If $\pp \in \sset_1$, then $\trace \vdash \bigvee \lset^+$ since $\neg \yctran{\prove{\secrone(x)}(\sset_1)}$ is added to $\lset^+$.
\item If $\pp \in \sset_0 \setminus \sset_1$, then $(\theta\replidx{\pp}, \theta\fset_{\pp})$ is added to $\lset^-$, where $\theta$ is a renaming of $\replidx{\pp}$ to fresh replication indices.
Moreover, we have $\neg\yctran{\noleak(z[\tup{M}'], \sri, \fset)} \Rightarrow \exists \sri, \bigwedge \fset$, 
so $\neg \yctran{\prove{\secrone(x)}(\pp)} \Rightarrow \exists \theta\replidx{\pp}, \bigwedge \theta\fset_{\pp}$
since $\defrestr_{\pp}(x)$ is defined.
Therefore, $\trace \vdash \bigvee_{j \in J} \exists \vset_j, \bigwedge\fset_j$.
\end{itemize}
\item In case $\prop = \secr(x)$, there exist $\pp_1,\pp_2 \in \Stestpp(\trace)$ such that $\trace \vdash \neg\yctran{\prove{\distinct(x)}(\pp_1,\ab \pp_2)}$.
Let $z_1[\tup{M_1}] = \defrestr_{\pp_1}(x)$, $z_2[\tup{M_2}] = \defrestr_{\pp_2}(x)$,
$\tup{i}$ be the current replication indices at the definition of $x$,
$\theta_1$ and $\theta_2$ be two distinct renamings of $\tup{i}$ to fresh replication indices, 
$\tup{i}_1 = \theta_1\tup{i}$, $\tup{i}_2 = \theta_2\tup{i}$, and
$\fset = \theta_1\fset_{\pp_1} \cup \theta_2\fset_{\pp_2} \cup \{ \theta_1 \tup{M_1} = \theta_2 \tup{M_2},\ab \tup{i}_1 \neq \tup{i}_2\}$
Then $z_1 = z_2$ and $\trace \vdash \exists \tup{i}_1, \exists \tup{i}_2, \bigwedge \fset$.
\begin{itemize}
\item If $\fset$ yields a contradiction, then $\exists \tup{i}_1, \exists \tup{i}_2, \bigwedge \fset$ is added to $\lset^+$, so $\trace \vdash \bigvee \lset^+$.
\item Otherwise, $(\tup{i}_1 \cup \tup{i}_2, \fset)$ is added to $\lset^-$, so $\trace \vdash \bigvee_{j \in J} \exists \vset_j, \bigwedge\fset_j$.
\end{itemize}
\end{itemize}

If $\trace \vdash \bigvee_{k \in K} \exists \replidx{\pp_k}, \bigwedge \fsetmod_k$, 
then $\trace \vdash \left(\bigvee_{j \in J, k \in K} \exists \replidx{\pp_k}, \exists \vset_j, \bigwedge \fset_j \wedge \bigwedge \fsetmod_k\right) \vee \left(\bigvee \lset^+\right)$.

Otherwise, by Fact~3, there is no configuration in $\trace$ at a modified program point $\pp_k$,
so $\trace$ has a matching trace in $C[G']$ that also satisfies $(\sevent \wedge\neg\prop) \vee D$.

Therefore, $\Pr[C[G]: (\sevent \wedge\neg\prop) \vee D] \leq p(C) + \Pr[C[G']: (\sevent \wedge\neg\prop) \vee D]$, so
\[\Advtev{G}{\prop}{C}{D} 
\leq p(C) + \Pr[C[G']: (\sevent \wedge\neg\prop) \vee D] - \Pr[C[G']: \sbarevent\wedge\neg\prop] -  \Pr[C[G']: \nonunique{G',D}]\]
Moreover, by applying on $G'$ the same steps as on $G$ at the beginning of this proof, we have
\[\Advtev{G'}{\prop}{C}{D} = \Pr[C[G']: (\sevent \wedge\neg\prop) \vee D] - \Pr[C[G']: \sbarevent\wedge\neg\prop] -  \Pr[C[G']: \nonunique{G',D}]\]
so
\[\Advtev{G}{\prop}{C}{D} \leq p(C) + \Advtev{G'}{\prop}{C}{D} \leq p(C) + p'(C)\]
since $\bound{G'}{V}{\prop}{D}{p'}$. Therefore, we have
$\bound{G}{V}{\prop}{D}{p+p'}$.
\proofcomplete
\end{proof}

\subsection{\rn{crypto}: Applying the Security Assumptions on Primitives}\label{sec:primdef}

The \rn{crypto} transformation applies security assumptions
on primitives. The first version of this transformation was
presented in~\cite[Section~3.2 and Appendix~D]{Blanchet07c}.

\section{Proof Strategy}\label{sec:strategy}

The first version of the automatic proof strategy was presented
in~\cite[Section~5]{Blanchet07c}.

\section{Conclusion}

The tool CryptoVerif produces proofs by sequences of games
like those manually written by cryptographers. It generates the
games, using an automatic proof strategy or guidance from the user,
who specifies the transformations to perform. It supports
a wide variety of cryptographic primitives specified by
indistinguishability axioms. Many of these primitives are
included in a library so that the user does not have to redefine
them. It can prove secrecy,
correspondence, and indistinguishability properties.
It has been applied to substantial case studies, including
Signal~\cite{KobeissiEuroSP17}, TLS~1.3~\cite{BhargavanSP17}, and
WireGuard~\cite{Lipp19}.

CryptoVerif still has limitations. In particular, the size of games
tends to grow too fast, which limits its ability to deal with large examples,
especially because some game transformations require the game to be expanded first by the \rn{expand} transformation, which duplicates the code from each test
until the end of protocol. Planed improvements include allowing more game transformations to work without previous application of \rn{expand}; allowing internal oracle calls in games, in order to share code between different parts of the game; using composition results in order to make proofs more modular.
Moreover, some game transformations could be generalized. For instance, the transformation \rn{merge\_branches} merges branches of a test when they execute the same code; the detection that several branches execute equivalent code could be made more flexible, by allowing reorderings of instructions for instance.
CryptoVerif only considers blackbox adversaries: it does not support proofs that manipulate the code of the adversary, such as the forking lemma~\cite{Pointche96}.

\paragraph{Acknowledgments}

I warmly thank David Pointcheval for his advice and explanations
of the computational proofs of protocols. This project would not have
been possible without him. I also thank Jacques Stern
for initiating this work and Charlie Jacomme for helpful discussions
on this document. 
The design and implementation of CryptoVerif was partly done while 
I was at CNRS and at Ecole Normale Sup\'erieure. 
This work was partly supported by the French National Research Agency (ANR) 
under the projects FormaCrypt (ARA SSIA 2005), ProSe (VERSO 2010, decision
number 2010-VERS-004), TECAP (decision number ANR-17-CE39-0004-03) and
received funding from the France 2030 program managed by the ANR under
the reference ANR-22-PECY-0006 (PEPR Cybersecurity SVP).

%\bibliographystyle{abbrv}
%\bibliography{biblio}  

\end{document}